\documentclass{amsbook}
\usepackage{amssymb}
\usepackage{amsmath}

\newtheorem{theorem}{Theorem}[chapter]
\newtheorem{corollary}[theorem]{Corollary}
\newtheorem{proposition}[theorem]{Proposition}
\newtheorem{lemma}[theorem]{Lemma}
\theoremstyle{definition}
\newtheorem{definition}[theorem]{Definition}

\newtheorem{notation}[theorem]{Notation}
\newtheorem{hypothesis}[theorem]{Hypothesis}
\theoremstyle{remark}
\newtheorem{remark}[theorem]{Remark}
\numberwithin{section}{chapter}
\numberwithin{equation}{chapter}

\input{tcilatex}
\makeindex

\begin{document}
\title{Non--cooperative Equilibria of Fermi Systems With Long Range
Interactions}
\author{J.-B. Bru}
\author{W. de Siqueira Pedra}
\date{\today }

\begin{abstract}
We define a Banach space $\mathcal{M}_{1}$ of models for fermions or quantum
spins in the lattice with long range interactions and explicit the structure
of (generalized) equilibrium states for any $\mathfrak{m}\in \mathcal{M}_{1}$%
. In particular, we give a first answer to an old open problem in
mathematical physics -- first addressed by Ginibre in 1968 within a
different context -- about the validity of the so--called Bogoliubov
approximation on the level of states. Depending on the model $\mathfrak{m}%
\in \mathcal{M}_{1}$, our method provides a systematic way to study all its
correlation functions and can thus be used to analyze the physics of long
range interactions. Furthermore, we show that the thermodynamics of long
range models $\mathfrak{m}\in \mathcal{M}_{1}$ is governed by the
non--cooperative equilibria of a zero--sum game, called here the
thermodynamic game. \bigskip

\end{abstract}

\maketitle

\newpage 

\tableofcontents%

\chapter*{Preface}

\setcounter{equation}{0}%
States are the positive and normalized linear functionals on a $\ast $%
--algebra $\mathcal{U}$ and forms a convex set $E$. This set is weak$^{\ast
} $--compact when $\mathcal{U}$ is a unital $C^{\ast }$--algebra and it is
even metrizable if $\mathcal{U}$ is separable, cf. \cite[Theorem 3.16]{Rudin}%
. The structure of the set $E$ of states is then satisfactorily described by
the Choquet theorem \cite{Alfsen,Phe}: Any state has a unique decomposition
as an integral on extreme states of $E$.

Special subsets $\mathit{\Omega }\subseteq E$ of states on $\mathcal{U}$ are
of particular importance in statistical physics, for instance, if $\mathcal{U%
}$ is the observable ($C^{\ast }$--) algebra of Fermi or quantum spin
systems on a lattice $\mathbb{Z}^{d}$ ($d\geq 1$). In this case, one of the
main issues is to understand the limit $l\rightarrow \infty $ of sequences
of (local) Gibbs equilibrium states 
\begin{equation*}
\rho _{l}\left( \cdot \right) =\frac{\mathrm{Trace}\left( \cdot \ \mathrm{e}%
^{-\beta U_{l}}\right) }{\mathrm{Trace}\left( \mathrm{e}^{-\beta
U_{l}}\right) }
\end{equation*}%
defined, for all $\beta >0$ and $l\in \mathbb{N}$, from self--adjoint
operators $U_{l}\in \mathcal{U}$. In quantum statistical mechanics, $\beta
>0 $ is the inverse temperature, $U_{l}$ represents the energy observable of
particles enclosed in a finite box $\Lambda _{l}\subseteq \mathbb{Z}^{d}$,
and the limit $l\rightarrow \infty $ is such that $\Lambda _{l}\nearrow 
\mathbb{Z}^{d}$ (thermodynamic limit). For instance, $l$ can be the side
length of a cubic box $\Lambda _{l}$. As $E$ is weak$^{\ast }$--compact, any
sequence of states $\rho _{l}\in E$ converges -- along a subsequence --
towards an \emph{equilibrium state} $\omega \in E$ as $l\rightarrow \infty $%
. An explicit characterization of the limit state $\omega \in E$ is a rather
difficult issue in most interesting cases.

Taking $U_{l}$ from \emph{local }(i.e., short range) translation invariant
interactions $\Phi $ and by conveniently choosing boundary conditions, the
limit state $\omega \in E$ is found to be a solution of a variational
problem on the (convex and weak$^{\ast }$--compact) set $E_{1}\subseteq E$
of translation invariant states, i.e., it minimizes a weak$^{\ast }$--lower
semi--continuous functional $f_{\Phi }$ on $E_{1}$. This result is standard
for quantum spin systems, see e.g. \cite[Chapter II]{Israel} or \cite[
Section 6.2]{BrattelliRobinson}. $f_{\Phi }$ and its minimizers are called,
respectively, the free--energy density functional and equilibrium states of
the system under consideration.

Fermion systems on a lattice correspond to choose the $C^{\ast }$--algebra $%
\mathcal{U}$ as the inductive limit of the net of complex Clifford algebras $%
\mathcal{U}_{\Lambda }$, $\Lambda \subseteq \mathbb{Z}^{d}$, $|\Lambda
|<\infty $, generated by the elements\footnote{$a_{x,s}$ and $a_{x,s}^{+}$
are the annihilation and creation operators of a particle at lattice
position $x$.} $a_{x,\mathrm{s}}$ and $a_{x,\mathrm{s}}^{+}$ satisfying the
so--called canonical anti--commutation relations (CAR) for $x\in \Lambda $
and $\mathrm{s}\in \mathrm{S}$. Here, the finite set $\mathrm{S}$
corresponds to the internal degrees of freedom (spin) of particles. Quantum
spin systems on a lattice are described by infinite tensor products of
finite dimensional $C^{\ast }$--algebras attached to each site $x\in \mathbb{%
Z}^{d}$. As a consequence, in contrast to lattice quantum spins, elements $%
A\in \mathcal{U}_{\Lambda }$ and $B\in \mathcal{U}_{\Lambda ^{\prime }}$ in
disjoint regions of the lattice ($\Lambda \cap \Lambda ^{\prime }=\emptyset $%
) do \emph{not} generally commute with each other. A study of equilibrium
states of lattice fermions similar to the one for lattice quantum spins is
hence more involved and was only\footnote{%
There are some results on the level of the pressure \cite%
{RaggioWerner1,RaggioWerner2} by using the quantum spin representation of
fermion systems for a specific class of models} performed in 2004 by Araki
and Moriya \cite{Araki-Moriya}. In particular, the limit state $\omega \in E$
is again a minimizer of a weak$^{\ast }$--lower semi--continuous functional $%
f_{\Phi }$ on $E_{1}$.

All these results \cite{Israel,BrattelliRobinson,Araki-Moriya} use Banach
spaces of \emph{local} interactions. Unfortunately, these Banach spaces are
too small to include all physically interesting systems. Indeed, physically
speaking, local interaction mainly means that the interaction between
particles is \emph{short} range, i.e., it has to decrease sufficiently fast
as the inter--particle distance increases. Nevertheless, \emph{long--range}
interactions are also fundamental as they explain important physical
phenomena like conventional superconductivity.

In this monograph, we construct a Banach space $\mathcal{M}_{1}$ of
translation invariant Fermi models including a class of long--range
interactions on the lattice. We restrict our analysis to translation
invariant Fermi systems, but we emphasize that all our studies can also be
performed for quantum spins\footnote{%
In fact, quantum spin systems are easier to analyze than fermion systems.}
as well as for (not necessarily translation invariant, but only)
periodically invariant systems. Then we generalize\footnote{%
In \cite{Araki-Moriya} the authors use a slightly different Banach space of
local interactions, see Remark \ref{remark general interaction}.} some
previous results of \cite{Israel,BrattelliRobinson,Araki-Moriya} to the
larger space $\mathcal{M}_{1}$. By conveniently choosing boundary
conditions, we show, in particular, that the sequence of Gibbs states $\rho
_{l}\in E$ defined from any $\mathfrak{m}\in \mathcal{M}_{1}$ converges
along a subsequence to a minimizer of the $\Gamma $--regularization $\Gamma
(f_{\mathfrak{m}}^{\sharp })$ (cf. (\ref{gamma+})) of the so--called free
energy density functional $f_{\mathfrak{m}}^{\sharp }$. Note that $f_{%
\mathfrak{m}}^{\sharp }$ is affine, but \emph{possibly not }weak$^{\ast }$%
--lower semi--continuous. Nevertheless, we prove that all weak$^{\ast }$%
--limit points of any sequence $\{\rho _{i}\}_{i=1}^{\infty }$ of its
approximating minimizers\footnote{%
It means that $\underset{i\rightarrow \infty }{\lim }f_{\mathfrak{m}}(\rho
_{i})=\inf \,f_{\mathfrak{m}}(E_{1})$} belong to the closed, convex, and weak%
$^{\ast }$--compact set $\mathit{\Omega }_{\mathfrak{m}}^{\sharp }$ of
minimizers of $\Gamma (f_{\mathfrak{m}}^{\sharp })$. Observe that $\Gamma
(f_{\mathfrak{m}}^{\sharp })$ is the largest convex and weak$^{\ast }$%
--lower semi--continuous minorant of $f_{\mathfrak{m}}^{\sharp }$ and its
minimizers are called \emph{generalized equilibrium states}. Minimizers of $%
f_{\mathfrak{m}}^{\sharp }$ are (usual) equilibrium states and form a subset
of $\mathit{\Omega }_{\mathfrak{m}}^{\sharp }$.

If the long--range component of the interaction is purely attractive then $%
\mathit{\Omega }_{\mathfrak{m}}^{\sharp }$ is always a face of $E_{1}$.
However, in the general case, $\mathit{\Omega }_{\mathfrak{m}}^{\sharp }$ is
only a \emph{subset} of a non--trivial face of $E_{1}$. From the Choquet
theorem (see, e.g., \cite[p. 14]{Phe}), any generalized equilibrium state $%
\omega \in \mathit{\Omega }_{\mathfrak{m}}^{\sharp }$ of an arbitrary
long--range model $\mathfrak{m}\in \mathcal{M}_{1}$ has a decomposition in
terms of extreme states of $\mathit{\Omega }_{\mathfrak{m}}^{\sharp
}\subseteq E_{1}$. As $E_{1}$ is known to be a Choquet simplex, this
decomposition is unique whenever $\mathit{\Omega }_{\mathfrak{m}}^{\sharp }$
is a face. Additionally, extreme states are shown to be minimizers of an
explicitly given weak$^{\ast }$--lower semi--continuous (possibly neither
convex nor concave) functional $g_{\mathfrak{m}}$. We also show that --
exactly as in the case of local interactions -- the set $\mathit{\Omega }_{%
\mathfrak{m}}^{\sharp }$ of generalized equilibrium states can be identified
with the set of all continuous tangent functionals at the point $\mathfrak{m}%
\in \mathcal{M}_{1}$ of a convex and continuous functional $\mathrm{P}_{%
\mathfrak{m}}^{\sharp }$, the so--called pressure, on the Banach space $%
\mathcal{M}_{1}$.

Note that non--uniqueness of generalized equilibrium states corresponds to
the existence of phase transitions for the considered model. This cannot be
seen for finite--volume systems. Indeed, the Gibbs equilibrium state is the 
\emph{unique} minimizer of the free energy at finite volume (Theorem \ref%
{passivity.Gibbs}). As a consequence, there are important differences
between the finite--volume system and its thermodynamic limit:

\begin{itemize}
\item \emph{Non--uniqueness of generalized t.i. equilibrium states.}
Similarly to the Gibbs state at finite volume, a generalized t.i.
equilibrium state $\omega \in \mathit{\Omega }_{\mathfrak{m}}^{\sharp }$
represents an infinite--volume thermal state at equilibrium. However, $%
\omega \in \mathit{\Omega }_{\mathfrak{m}}^{\sharp }$ may not be unique,
see, e.g., \cite[Section 6.2]{BruPedra1}. In fact, physically important
phase transitions are those for which the minimizers of $\Gamma (f_{%
\mathfrak{m}}^{\sharp })$ break initial symmetries of the system. This case
is called \emph{spontaneous symmetry breaking}. For concrete local
interactions, such a phenomenon is usually difficult to prove in the quantum
case, whereas there are many explicit models $\mathfrak{m}\in \mathcal{M}%
_{1} $ where it can easily be seen, see, e.g., \cite{BruPedra1,
BruPedraAniko}.\medskip

\item \emph{Space symmetry of generalized equilibrium states.} The Gibbs
equilibrium state minimizes the finite--volume free--energy density
functional over the set $E$ of all states. However, even if the interaction
is translation invariant, it may possibly not converge to a t.i. state in
the thermodynamic limit. In particular, the weak$^{\ast }$--limit state may
not belong to $\mathit{\Omega }_{\mathfrak{m}}^{\sharp }$. In other words, a
t.i. (physical) system can lead to periodic (or more complicated
non--translation invariant) structures. Such a phenomenon could, for
instance, explain the appearance of periodic superconducting phases recently
observed \cite{LeBoeuf,Pfleiderer}. \medskip
\end{itemize}

Observe that, in general, the solutions of the variational problems given in 
\cite{Israel,BrattelliRobinson,Araki-Moriya} for local interactions cannot
be computed explicitly. The variational problem%
\begin{equation*}
\mathrm{P}_{\mathfrak{m}}^{\sharp }=-\inf f_{\mathfrak{m}}^{\sharp
}(E_{1})=-\inf \;\Gamma (f_{\mathfrak{m}}^{\sharp })(E_{1})
\end{equation*}%
generalizing previous results on local interactions to any $\mathfrak{m}\in 
\mathcal{M}_{1}$ is, a priori, even more difficult. We prove, however, that
this minimization problem can be explicitly analyzed from variational
problems with local interactions. This strong simplification is related to
an old open problem in mathematical physics -- first addressed by Ginibre 
\cite[p. 28]{Ginibre} in 1968 within a different context -- about the
validity of the so--called Bogoliubov approximation on the level of states.
Indeed, we give a first answer to this problem in the special class $%
\mathcal{M}_{1}$ of models $\mathfrak{m}$ by showing that any extreme
generalized equilibrium state $\omega \in \mathit{\Omega }_{\mathfrak{m}%
}^{\sharp }$ is an equilibrium state of an effective local interaction $\Phi
_{\omega }$. Such extreme generalized equilibrium states satisfy
Euler--Lagrange equations called \emph{gap equations} in the Physics
literature. In fact, when the correlation functions of the effective local
interaction $\Phi _{\omega }$ turn out to be accessible, our method provides
a systematic way to analyze, at once, \emph{all} correlation functions of
the given long--range model $\mathfrak{m}\in \mathcal{M}_{1}$. Applications
of our method include: A full analysis (postponed to separated papers) of
equilibrium states of BCS--type models, the explicit description of models
showing qualitatively the same density dependency of the critical
temperature observed in high--$T_{c}$\emph{\ }superconductors \cite%
{BruPedra1, BruPedraAniko}, etc.

One important consequence of the detailed analysis of the set $\mathit{%
\Omega }_{\mathfrak{m}}^{\sharp }$ of generalized equilibrium states is the
fact that the thermodynamics of models $\mathfrak{m}\in \mathcal{M}_{1}$ is
governed by the following \emph{two--person zero--sum game}: For any model $%
\mathfrak{m}$, we define a functional 
\begin{equation*}
\mathfrak{f}_{\mathfrak{m}}^{\mathrm{ext}}:L_{-}^{2}\times \mathrm{C}%
(L_{-}^{2},L_{+}^{2})\rightarrow \mathbb{R}.
\end{equation*}%
Here, $L_{\pm }^{2}$ are two orthogonal sub--spaces of a Hilbert space $%
L^{2}(\mathcal{A},\mathbb{C})$ and $\mathrm{C}(L_{-}^{2},L_{+}^{2})$ is the
set of continuous maps from $L_{-}^{2}$ to $L_{+}^{2}$, respectively endowed
with the weak and norm topologies. The set $L_{-}^{2}$ is seen as the set of
strategies of the \textquotedblleft attractive\textquotedblright\ player
with loss function $\mathfrak{f}_{\mathfrak{m}}^{\mathrm{ext}}$ and $\mathrm{%
C}(L_{-}^{2},L_{+}^{2})$ is the set of strategies of the \textquotedblleft
repulsive\textquotedblright\ player with loss function $-\mathfrak{f}_{%
\mathfrak{m}}^{\mathrm{ext}}$. This game has a non--cooperative equilibrium
and the value of the game is precisely $-\mathrm{P}_{\mathfrak{m}}^{\sharp }$%
. Moreover, for any $\mathfrak{m}\in \mathcal{M}_{1}$, equilibria of this
game classify extreme generalized equilibrium states in $\mathit{\Omega }_{%
\mathfrak{m}}^{\sharp }$ in the following sense: There is a set $\{\mathfrak{%
e}_{a}\}_{a\in \mathcal{A}}$ of observables such that, for any extreme state 
$\hat{\omega}\in \mathit{\Omega }_{\mathfrak{m}}^{\sharp }$, there is a
non--cooperative equilibrium 
\begin{equation*}
(d_{a,-},\mathrm{r}_{+})\in L_{-}^{2}\times \mathrm{C}(L_{-}^{2},L_{+}^{2})
\end{equation*}%
with%
\begin{equation*}
\hat{\omega}\left( \mathfrak{e}_{a}\right) =d_{a,-}+\mathrm{r}_{+}\left(
d_{a,-}\right) \in L^{2}(\mathcal{A},\mathbb{C}).
\end{equation*}%
For a more precise definition of $(d_{a,-},\mathrm{r}_{+})$, see (\ref{eq
conserve strategy}) and (\ref{thermodyn decision rule}). Conversely, for
each non--cooperative equilibrium 
\begin{equation*}
(d_{a,-},\mathrm{r}_{+})\in L_{-}^{2}\times \mathrm{C}(L_{-}^{2},L_{+}^{2}),
\end{equation*}%
there is a -- not necessarily extreme -- $\omega \in \mathit{\Omega }_{%
\mathfrak{m}}^{\sharp }$ satisfying the above equation.

This monograph is organized as follows. In Chapter {\ref{Section definition}%
, we briefly explain the mathematical framework of Fermi systems on a
lattice. }Then {the main results concerning the thermodynamic study of any }$%
\mathfrak{m}\in \mathcal{M}_{1}$ are formulated in Chapter {\ref{equilibrium}%
. Note that a discussion on previous results related to the ones presented
here is given in Section \ref{Concluding remark}. }In order to keep the main
issues as transparent as possible, we reduce the technical aspects to a
minimum in Chapters {\ref{Section definition} and \ref{equilibrium}, which
forms Part \ref{part I}. Our main results are Theorems }\ref{BCS main
theorem 1}, \ref{theorem structure of omega copy(1)}, \ref{theorem saddle
point}, and \ref{theorem structure of omega}. Examples of applications are
given in Section \ref{example of application section}.

Part \ref{part II} collects complementary important results and corresponds
to Chapters {\ref{section pbc}--}\ref{Section appendix}. In particular,
Chapter {\ref{section pbc} is an account on periodic boundary conditions
which ensure the }weak$^{\ast }$--{convergence as }$l\rightarrow \infty ${\
of the Gibbs equilibrium states }$\rho _{l}${\ to a generalized equilibrium
state }$\omega \in \mathit{\Omega }_{\mathfrak{m}}^{\sharp }${. In Chapter %
\ref{section set of invarant states} we analyze in details the} set of
periodic states. Except Sections {\ref{Section properties of delta} and \ref%
{section 6.6}}, this analysis is only an adaptation of known results for
quantum spin systems. Chapter \ref{Stoermer} explains permutation invariant
models in relation with the St{\o }rmer theorem \cite{BruPedra1,Stormer} for
permutation invariant states on the CAR algebra because they {are
technically important for the derivation of the variational problem }$%
\mathrm{P}_{\mathfrak{m}}^{\sharp }=-\inf f_{\mathfrak{m}}^{\sharp }(E_{1})${%
\ for the pressure}. Chapters \ref{section proof of theorem main}--\ref%
{Section generalized eq state-effective theory} give the detailed proofs of
the main theorems about the game theoretical issues and generalized
equilibrium states of long--range models. In particular, we analyze in
details in Chapters \ref{Section theorem saddle point}--\ref{Section
generalized eq state-effective theory} the relation between the
thermodynamics of general long--range models $\mathfrak{m}\in \mathcal{M}%
_{1} $ and effective local interactions $\Phi _{\omega }$. This is related
to the so--called approximating Hamiltonian method used on the level of the
pressure in \cite%
{Bogjunior,approx-hamil-method0,approx-hamil-method,approx-hamil-method2}.
We give in Chapter \ref{Section appendix} a short review on this subject as
well as on Gibbs equilibrium states, compact convex sets, Choquet simplices,
tangent functionals, the $\Gamma $--regularization, the Legendre--Fenchel
transform, and on two--person zero--sum games. All the material in Chapter %
\ref{Section appendix}, up to Lemma \ref{Bauer maximum principle bis} and
Theorems \ref{theorem trivial sympa 1}--\ref{theorem trivial sympa 1 copy(1)}%
, can be found in standard textbooks. These topics are concisely discussed
here to make our results accessible to a wide audience, since various fields
of theoretical physics and mathematics are concerned (non--linear analysis,
game theory, convex analysis, and statistical mechanics).

To conclude, we would like to thank Andr\'{e} Verbeure and Valentin A.
Zagrebnov for relevent references as well as Hans--Peter Heinz for many
interesting discussions and important hints about convex analysis and game
theory. We are also very grateful to Volker Bach and Jakob Yngvason for
their hospitality at the Erwin Schr\"{o}dinger International Institute for
Mathematical Physics, at the Physics University of Vienna, and at the
Institute of Mathematics of the Johannes Gutenberg University. Finally, we
thank the referee for his work and constructive criticisms.\bigskip

\noindent \textbf{Remark on the present postprint:} This manuscript has been
originally published in 2013 in \textit{Memoirs of the AMS} (volume 224, no.
1052). This postprint is a corrected version of this publication, but the
historical part has not been updated and thus runs until 2013. We are also
very grateful to S\'{e}bastien Breteaux for pointing out several mistakes
and suggesting various improvements on the text.

\aufm{Jean-Bernard Bru and Walter de Siqueira Pedra}

\mainmatter

\part{Main Results and Discussions\label{part I}}

\chapter{Fermi Systems on Lattices\label{Section definition}}

In Section \ref{Section fermions algebra} we define fermion (field) algebras 
$\mathcal{U}$. Self--adjoint elements of these $C^{\ast }$--algebras $%
\mathcal{U}$ correspond to observables, i.e., physical quantifies which can
be measured for fermion particles on a lattice $\mathfrak{L}$. Fermion
algebras are also referred to CAR algebras in the literature. For technical
simplicity, we only consider cubic lattices $\mathfrak{L=}\mathbb{Z}^{d}$, $%
d\in \mathbb{N}$.

The study of a given physical system needs the additional concept of state,
which represents the statistical distribution of outcomes of measurements on
this system related to any observable. In mathematics, states are identified
with the positive and normalized maps from $\mathcal{U}$ into $\mathbb{C}$.
In particular, states belong to the dual space $\mathcal{U}^{\ast }$ of the
Banach space $\mathcal{U}$. A class of states important in physics is given
by the sets $\{E_{\vec{\ell}}\}_{\vec{\ell}\in \mathbb{N}^{d}}$ of all $\vec{%
\ell}$--periodic states\ whose structure is described in Section \ref%
{section set of states}. The concept of ergodicity plays a key role in this
description and is strongly related to the ($\vec{\ell}$--) space--averaging
functional $\Delta _{A,\vec{\ell}}$, which is analyzed in details for $\vec{%
\ell}=(1,\cdots ,1)$ in Section \ref{Section space averaging}.

Fixing a physical system among all possible ones corresponds to fix a family
of self--adjoint (even) elements $U_{\Lambda }^{\Phi }$ of $\mathcal{U}$,
i.e., observables, which represents the total energy in the finite box $%
\Lambda \subseteq \mathfrak{L}$. These elements are called in this monograph
internal energies and are also known in physics as Hamiltonians. In fact, we
are interested in infinite systems which result from the thermodynamic limit 
$\Lambda \nearrow \mathfrak{L}$ of finite--volume models defined from local
internal energies. To define such families of internal energies we can, for
instance, use a Banach space $\mathcal{W}_{1}$ of translation invariant
(t.i.) local interactions $\Phi $ which define an internal energy $%
U_{\Lambda }^{\Phi }\subseteq \mathcal{U}$ for any $\Lambda \subseteq 
\mathfrak{L}$. The detailed explanation of this construction is found in
Section \ref{Local interactions}. Observe, however, that this is not the
only reasonable way of defining internal energies. In the next chapter we
will generalize this procedure.

Finally, the state of a physical system in thermal equilibrium is defined by
a variational problem (cf. Section \ref{Section Gibbs equilibrium states}).
Any equilibrium state of a given system with interaction $\Phi \in \mathcal{W%
}_{1}$ minimizes the density of free energy $f_{\Phi }$ corresponding to
this interaction. This functional $f_{\Phi }$ is defined on $E_{\vec{\ell}}$
for any $\vec{\ell}\in \mathbb{N}^{d}$ and is a (weighted) sum of two
density functionals: The energy density functional $\rho \mapsto e_{\Phi
}(\rho )$, which correspond to the mean energy $\rho (U_{\Lambda }^{\Phi
})/|\Lambda |$ per volume when $\Lambda \nearrow \mathfrak{L}$ and $\rho \in
E_{\vec{\ell}}$, and the entropy density functional $\rho \mapsto s(\rho )$,
which measures, in a sense, the amount of randomness (per unit of volume)
carried by a state $\rho \in E_{\vec{\ell}}$ when $\Lambda \nearrow 
\mathfrak{L}$. The free energy, energy, and entropy functionals are
described in Section \ref{section entropie}. Such a variational principle
for equilibrium states implements the second law of thermodynamics because
minimizing the free energy density is equivalent to maximize the entropy
density at constant mean energy per unit of volume.

Note that all our studies can also be performed for quantum spins as well as
for (not necessarily translation invariant, but only) periodically invariant
systems. We concentrate our attention to fermion algebras as they are more
difficult to handle because of the non--commutativity of its elements on
different lattice sites, see Remark \ref{Quantum spin systems}. In fact, up
to Section \ref{Section space averaging}, the results presented in this
chapter are known for quantum spin systems (see, e.g., \cite{Israel}). In
this monograph, we extend them to Fermi systems by using results of Araki
and Moriya \cite{Araki-Moriya} (see also Remark \ref{remark general
interaction}). The material presented in Section \ref{Section space
averaging} is new for both quantum spins and Fermi systems. Note that the
detailed proofs are postponed until Chapter \ref{section set of invarant
states}. Sections \ref{Section Compact convex sets and Choquet simplices}-%
\ref{Section gamma regularization} are prerequisites.

\section{Local fermion algebras\label{Section fermions algebra}}

Let $\mathfrak{L}:=\mathbb{Z}^{d}$ be the $d$--dimensional cubic lattice and 
$\mathcal{H}$ be a finite dimensional Hilbert space with orthonormal basis $%
\{e_{\mathrm{s}}\}_{\mathrm{s}\in \mathrm{S}}$.

\begin{notation}
\label{Notation1}Here, we use the convention that $\mathfrak{L}$ stands for $%
\mathbb{Z}^{d}$ as seen as a set (lattice), whereas with $\mathbb{Z}^{d}$
the abelian group $(\mathbb{Z}^{d},+)$ is meant.
\end{notation}

\begin{remark}[Lattices $\mathfrak{L}$]
\mbox{ }\newline
\index{Lattice}The lattice $\mathfrak{L}$ is taken to be a cubic one because
it is technically easier, but this choice is not necessary for our proofs.
\end{remark}

For any set $M$, we define $\mathcal{P}_{f}(M)$ to be the set of all finite
subsets of $M$. In the special case where $M=\mathfrak{L}$ we use below the
sequence of cubic boxes 
\begin{equation}
\Lambda _{l}:=\{x\in \mathfrak{L}\;:\;|x_{i}|\leq l,\,i=1,\ldots ,d\}\in 
\mathcal{P}_{f}(\mathfrak{L})  \label{cubic box}
\end{equation}%
of the lattice $\mathfrak{L}$ with volume $|\Lambda _{l}|=(2l+1)^{d}$ for $%
l\in \mathbb{N}$.

\begin{remark}[Van Hove nets]
\mbox{ }\newline
\index{Van Hove net}The sequence $\{\Lambda _{l}\}_{l\in \mathbb{N}}$ is
used to define the thermodynamic limit. It is a technically convenient
choice, but it is not necessary in our proofs. The minimal requirement on
any net $\{\Lambda _{i}\}_{i\in I}$ of finite boxes is that the volume $%
|\partial \Lambda _{i}|$ of the boundaries\footnote{%
By fixing $m\geq 1$, the boundary $\partial \Lambda $ of any $\Lambda
\subseteq \mathfrak{L}$ is defined by $\partial \Lambda :=\{x\in \Lambda
:\exists y\in \mathfrak{L}\backslash \Lambda \mathrm{\ with\ }d(x,y)\leq m\}$%
, see (\ref{def.dist}) below for the definition of the metric $d(x,y)$.} $%
\partial \Lambda _{i}\subseteq \Lambda _{i}\in \mathcal{P}_{f}(\mathfrak{L})$
must be negligible with respect to (w.r.t.) the volume $|\Lambda _{i}|$ of $%
\Lambda _{i}$ at \textquotedblleft large\textquotedblright\ $i\in I$, i.e., $%
\lim\limits_{I}\{|\partial \Lambda _{i}|/|\Lambda _{i}|\}=0.$ Such families $%
\{\Lambda _{i}\}_{i\in I}$ of subsets are known as Van Hove nets, see, e.g., 
\cite{Araki-Moriya}. Note that also the condition $\Lambda _{l}\subseteq
\Lambda _{l+1}$ is not necessary and it suffices to impose that, for any $%
\Lambda \in \mathcal{P}_{f}(\mathfrak{L})$, there is $i_{\Lambda }\in I$
such that $\Lambda \subseteq \Lambda _{i}$ for all $i\geq i_{\Lambda }$.
\end{remark}

For any $\Lambda \in \mathcal{P}_{f}(\mathfrak{L})$, let $\mathcal{U}%
_{\Lambda }$ be the complex 
\index{Clifford algebra}Clifford algebra with identity $\mathbf{1}$ and
generators $\{a_{x,\mathrm{s}},a_{x,\mathrm{s}}^{+}\}_{x\in \Lambda ,\mathrm{%
s}\in \mathrm{S}}$ satisfying the so--called canonical anti--commutation
relations 
\index{CAR|textbf}(CAR):%
\begin{equation}
\left\{ 
\begin{array}{lll}
a_{x,\mathrm{s}}a_{x^{\prime },\mathrm{s}^{\prime }}+a_{x^{\prime },\mathrm{s%
}^{\prime }}a_{x,\mathrm{s}} & = & 0, \\[0.15cm] 
a_{x,\mathrm{s}}^{+}a_{x^{\prime },\mathrm{s}^{\prime }}^{+}+a_{x^{\prime },%
\mathrm{s}^{\prime }}^{+}a_{x,\mathrm{s}}^{+} & = & 0, \\[0.15cm] 
a_{x,\mathrm{s}}a_{x^{\prime },\mathrm{s}^{\prime }}^{+}+a_{x^{\prime },%
\mathrm{s}^{\prime }}^{+}a_{x,\mathrm{s}} & = & \delta _{x,x^{\prime
}}\delta _{\mathrm{s},\mathrm{s}^{\prime }}\mathbf{1}.%
\end{array}%
\right.  \label{CAR}
\end{equation}%
The set $\mathcal{U}_{\Lambda }$ is a $C^{\ast }$--algebra because it is
isomorphic to the algebra \QTR{cal}{B}$(\bigwedge \mathcal{H}_{\Lambda })$
of all bounded linear operators on the fermion Fock space $\bigwedge 
\mathcal{H}_{\Lambda }$, where%
\begin{equation*}
\mathcal{H}_{\Lambda }:=\bigoplus\limits_{x\in \Lambda }\mathcal{H}_{x},
\end{equation*}%
$\mathcal{H}_{x}$, $x\in \mathfrak{L}$, being copies of the finite
dimensional Hilbert space $\mathcal{H}$. For any $\Lambda \in \mathcal{P}%
_{f}(\mathfrak{L})$, $\mathcal{U}_{\Lambda }$ is called the \emph{local
fermion (field) algebras }of the lattice $\mathfrak{L}$. Indeed, in quantum
statistical mechanics $a_{x,\mathrm{s}}^{+}=(a_{x,\mathrm{s}})^{\ast }$ and $%
a_{x,\mathrm{s}}$ are interpreted, respectively, as the creation and
annihilation of a fermion with spin $\mathrm{s}\in \mathrm{S}$ at the
position $x\in \mathfrak{L}$ of the lattice, and the CAR (\ref{CAR})
implement the Pauli principle.

For any $\Lambda \subseteq \Lambda ^{\prime }\subseteq \Lambda ^{\prime
\prime }\in \mathcal{P}_{f}(\mathfrak{L})$, there are canonical inclusions $%
j_{\Lambda ,\Lambda ^{\prime }}:\mathcal{U}_{\Lambda }\rightarrow \mathcal{U}%
_{\Lambda ^{\prime }}$ satisfying $j_{\Lambda ^{\prime },\Lambda ^{\prime
\prime }}\circ j_{\Lambda ,\Lambda ^{\prime }}=j_{\Lambda ,\Lambda ^{\prime
\prime }}$ and $j_{\Lambda ,\Lambda ^{\prime }}(a_{x,\mathrm{s}})=a_{x,%
\mathrm{s}}$ for any $x\in \Lambda $ and $\mathrm{s}\in \mathrm{S}$. The
inductive limit of local algebras $\{\mathcal{U}_{\Lambda }\}_{\Lambda \in 
\mathcal{P}_{f}(\mathfrak{L})}$ is the $C^{\ast }$--algebra $\mathcal{U}$,
called the 
\index{Fermion field algebra}\emph{fermion (field) algebra} (also known as
the \emph{CAR algebra}). A dense subset of $\mathcal{U}$ is given by the $%
\ast $--algebra%
\begin{equation}
\mathcal{U}_{0}:=\bigcup\limits_{\Lambda \in \mathcal{P}_{f}(\mathfrak{L})}%
\mathcal{U}_{\Lambda }  \label{local elements}
\end{equation}%
of local elements, which implies the separability of $\mathcal{U}$ as $%
\mathcal{U}_{\Lambda }$ is a finite dimensional space for any $\Lambda \in 
\mathcal{P}_{f}(\mathfrak{L})$.

\begin{remark}[Quantum spin systems]
\label{Quantum spin systems}\mbox{ }\newline
\index{Quantum spin systems}For quantum spin systems, $\mathcal{U}$ would be
the infinite tensor product of finite dimensional $C^{\ast }$--algebras
attached to each site $x\in \mathbb{Z}^{d}$. All results of this monograph
hold in this case, but we concentrate our attention on fermion algebras as
they are more difficult to handle because of the non--commutativity of their
elements on different lattice sites.
\end{remark}

For any fixed $\theta \in \mathbb{R}/(2\pi \mathbb{Z)}$, the condition 
\begin{equation}
\sigma _{\theta }(a_{x,\mathrm{s}})=\mathrm{e}^{-i\theta }a_{x,\mathrm{s}}
\label{definition of gauge}
\end{equation}%
defines a unique automorphism $\sigma _{\theta }$ of the algebra $\mathcal{U}
$. A special role is played by $\sigma _{\pi }$. Elements $A,B\in \mathcal{U}
$ satisfying $\sigma _{\pi }(A)=A$ and $\sigma _{\pi }(B)=-B$ are
respectively called \emph{even} and 
\index{Odd!algebra element}\emph{odd}, whereas elements $A\in \mathcal{U}$
satisfying $\sigma _{\theta }(A)=A$ for any $\theta \in \lbrack 0,2\pi )$
are called 
\index{Gauge invariant!algebra elements}\emph{gauge invariant}. The set 
\begin{equation}
\mathcal{U}^{+}:=\{A\in \mathcal{U}\;:\;A-\sigma _{\pi }(A)=0\}\subseteq 
\mathcal{U}  \label{definition of even operators}
\end{equation}%
of all 
\index{Even!algebra elements}even elements and the set 
\begin{equation}
\mathcal{U}^{\circ }:=\bigcap\limits_{\theta \in \mathbb{R}/(2\pi \mathbb{Z)}%
}\{A\in \mathcal{U}\;:\;A=\sigma _{\theta }(A)\}\subseteq \mathcal{U}^{+}
\label{definition of gauge invariant operators}
\end{equation}%
of all gauge invariant elements are $\ast $--algebras. By continuity of $%
\sigma _{\theta }$, it follows that $\mathcal{U}^{+}$ and $\mathcal{U}%
^{\circ }$ are closed and hence $C^{\ast }$--algebras, respectively called 
\emph{sub--algebra of even elements} and \emph{fermion observable algebra}.

\begin{remark}[Gauge invariant projection]
\label{proj.gauge.inv}\mbox{ }\newline
\index{Gauge invariant!projection}By density of the $\ast $--algebra $%
\mathcal{U}_{0}$ of local elements, for any $A\in \mathcal{U}$, the map $%
\theta \mapsto \sigma _{\theta }(A)$ is continuous. Thus, for any $A\in 
\mathcal{U}$, the Riemann integral 
\begin{equation*}
\sigma ^{\circ }(A):=%
\frac{1}{2\pi }\int_{0}^{2\pi }\sigma _{\theta }(A)\ \mathrm{d}\theta
\end{equation*}%
defines a linear map $\sigma ^{\circ }:\mathcal{U}\rightarrow \mathcal{U}%
^{\circ }$, which is a projection on the fermion observable algebra $%
\mathcal{U}^{\circ }$, i.e., $\sigma ^{\circ }\circ \sigma ^{\circ }=\sigma
^{\circ }$.
\end{remark}

\begin{notation}[Gauge invariant objects]
\label{Notation2}\mbox{ }\newline
\index{Gauge invariant!notation}Any symbol with a circle $\circ $ as a
superscript (for instance, $\sigma ^{\circ }$) is, by definition, an object
related to gauge invariance.
\end{notation}

\section{States of Fermi systems on lattices\label{section set of states}}

As $\mathcal{U}$ is a Banach space, by Corollary \ref{thm locally convex
spacebis}, its dual $\mathcal{U}^{\ast }$ is a locally convex real space%
\footnote{%
We use here Rudin's definition, see Definition \ref{space}.} with respect to
(w.r.t.) the weak$^{\ast }$--topology, which is Hausdorff. Moreover, as $%
\mathcal{U}$ is separable, by Theorem \ref{Metrizability}, the weak$^{\ast }$%
--topology is metrizable on any weak$^{\ast }$--compact subset of $\mathcal{U%
}^{\ast }$ as, for instance, on the weak$^{\ast }$--compact convex set $%
E\subseteq \mathcal{U}^{\ast }$ of all \emph{states} on $\mathcal{U}$.

\index{States}States are linear functionals $\rho \in \mathcal{U}^{\ast }$
which are positive, i.e., for all $A\in \mathcal{U}$, $\rho (A^{\ast }A)\geq
0$, and normalized, i.e., $\rho (\mathbf{1})=1$. Equivalently, $\rho \in 
\mathcal{U}^{\ast }$ is a state iff $\rho (\mathbf{1})=1$ and $\Vert \rho
\Vert =1$ which clearly means that $E$ is a subset of the unit ball of $%
\mathcal{U}^{\ast }$. Note that any $\rho \in E$ is continuous and
Hermitian, i.e., for all $A\in \mathcal{U}$, $\rho (A^{\ast })=%
\overline{\rho (A)}$, and defines by restriction a state on the
sub--algebras $\mathcal{U}^{+}$, $\mathcal{U}^{\circ }$, and $\mathcal{U}%
_{\Lambda }$. For any $\Lambda \in \mathcal{P}_{f}(\mathfrak{L})$, we use $%
\rho _{\Lambda }$ and $E_{\Lambda }$ to denote, respectively, the
restriction of any $\rho \in E$ on the local sub--algebra $\mathcal{U}%
_{\Lambda }$ and the set of all states $\rho _{\Lambda }$ on $\mathcal{U}%
_{\Lambda }$.

\begin{notation}[States]
\label{Notation3}\mbox{ }\newline
\index{States!notation}The letters $\rho $, $\varrho $, and $\omega $ are
exclusively reserved to denote states.
\end{notation}

\index{States!invariant}Invariant states under the action of groups $G$ play
a crucial role in the sequel. In the special case where $G=(\mathbb{Z}%
^{d},+) $, the condition 
\begin{equation}
\alpha _{x}(a_{y,\mathrm{s}})=a_{y+x,\mathrm{s}}\ ,\quad \forall y\in 
\mathbb{Z}^{d},\;\forall \mathrm{s}\in \mathrm{S},  \label{transl}
\end{equation}%
defines a homomorphism $x\mapsto \alpha _{x}$ from $\mathbb{Z}^{d}$ to the
group of $\ast $--automorphisms of $\mathcal{U}$. In other words, the family
of $\ast $--automorphisms $\{\alpha _{x}\}_{x\in \mathfrak{L}}$ represents
here the action of the group of lattice translations on $\mathcal{U}$.
Consider now the sub--groups $G=(\mathbb{Z}_{%
\vec{\ell}}^{d},+)\subseteq (\mathbb{Z}^{d},+)$ with%
\begin{equation*}
\mathbb{Z}_{\vec{\ell}}^{d}:=\ell _{1}\mathbb{Z}\times \cdots \times \ell
_{d}\mathbb{Z},\qquad \vec{\ell}\in \mathbb{N}^{d}.
\end{equation*}%
Any state $\rho \in E$ satisfying $\rho \circ \alpha _{x}=\rho $ for all $%
x\in \mathbb{Z}_{\vec{\ell}}^{d}$ is called 
\index{States!l--periodic}$\mathbb{Z}_{%
\vec{\ell}}^{d}$\emph{--invariant} on $\mathcal{U}$ or $\vec{\ell}$\emph{%
--periodic}. The set of all $\mathbb{Z}_{\vec{\ell}}^{d}$--invariant states
is denoted by 
\begin{equation}
E_{\vec{\ell}}:=\bigcap\limits_{x\in \mathbb{Z}_{\vec{\ell}}^{d},\text{ }%
A\in \mathcal{U}}\{\rho \in \mathcal{U}^{\ast }\;:\;\rho (\mathbf{1}%
)=1,\;\rho (A^{\ast }A)\geq 0\text{\quad }\mathrm{with\ }\rho =\rho \circ
\alpha _{x}\}.  \label{periodic invariant states}
\end{equation}%
Note that $E_{1}:=E_{(1,\cdots ,1)}$ corresponds to the set of all 
\index{States!translation invariant}\emph{translation invariant} (t.i.)
states. The $%
\vec{\ell}$--periodicity of states yields a crucial property, deduced from
Corollary \ref{coro.even copy(1)}:

\begin{lemma}[$%
\vec{\ell}$--periodic states are even]
\label{coro.even}\mbox{ }\newline
Any 
\index{States!l--periodic}$\mathbb{Z}_{%
\vec{\ell}}^{d}$--invariant state $\rho $ is 
\index{Even!states}even, i.e., $\rho =\rho \circ \sigma _{\pi }$ with the
automorphism $\sigma _{\pi }$ defined by (\ref{definition of gauge}) for $%
\theta =\pi $.
\end{lemma}

\noindent In other words, all $\mathbb{Z}_{\vec{\ell}}^{d}$--invariant
states $\rho \in E_{\vec{\ell}}$ must be the zero functional on the
sub--space of odd elements of $\mathcal{U}$. This symmetry property is a
necessary ingredient to study thermodynamics of Fermi systems.

The set $E_{\vec{\ell}}$ is clearly convex and weak$^{\ast }$--compact. So,
the Krein--Milman theorem%
\index{Krein--Milman theorem} (Theorem \ref{theorem Krein--Millman}) tells
us that it is the weak$^{\ast }$--closure of the convex hull of the
(non--empty) set 
\index{States!extreme}$\mathcal{E}_{%
\vec{\ell}}$ of its extreme points. (Here, $\mathcal{E}_{1}:=\mathcal{E}%
_{(1,\cdots ,1)}$ is the set of extreme points of the set $E_{1}$ of t.i.
states.) Since $E_{\vec{\ell}}$ is also metrizable (Theorem \ref%
{Metrizability}), from the Choquet theorem%
\index{Choquet theorem} (Theorem \ref{theorem choquet bis}), each state $%
\rho \in E_{%
\vec{\ell}}$ has a decomposition in terms of extreme states $\hat{\rho}\in 
\mathcal{E}_{\vec{\ell}}$ of $E_{\vec{\ell}}$. This decomposition is unique
and norm preserving by Lemma \ref{choquet unique}.

\begin{theorem}[Ergodic decomposition of states in $E_{\vec{\ell}}$]
\label{theorem choquet}\mbox{ }\newline
\index{States!ergodic decomposition}For any $\rho \in E_{%
\vec{\ell}}$, there is a unique probability measure $\mu _{\rho }$ on $E_{%
\vec{\ell}}$ supported on $\mathcal{E}_{\vec{\ell}}$ and representing $\rho
\in E_{\vec{\ell}}$:%
\begin{equation*}
\mu _{\rho }(\mathcal{E}_{\vec{\ell}})=1\text{\quad and\quad }\rho =\int_{E_{%
\vec{\ell}}}\mathrm{d}\mu _{\rho }(\hat{\rho})\;\hat{\rho}.
\end{equation*}%
Furthermore, the map $\rho \mapsto \mu _{\rho }$ is an isometry in the norm
of linear functionals, i.e., $\Vert \rho -\rho ^{\prime }\Vert =\Vert \mu
_{\rho }-\mu _{\rho ^{\prime }}\Vert $ for any $\rho ,\rho ^{\prime }\in E_{%
\vec{\ell}}$.
\end{theorem}

\begin{remark}[Barycenters]
\label{remark barycenter}\mbox{ }\newline
\index{Choquet theorem!Barycenter}%
\index{Barycenters}The integral written in Theorem \ref{theorem choquet}
only means here that $\rho \in E_{%
\vec{\ell}}$ is the (unique) barycenter of the probability measure, i.e.,
the normalized positive Borel regular measure, $\mu _{\rho }\in M_{1}^{+}(E_{%
\vec{\ell}})$ on $E_{\vec{\ell}}$, see Definition \ref{def barycenter} and
Theorem \ref{thm barycenter}.
\end{remark}

\begin{notation}[Extreme states]
\label{notation extreme states}\mbox{ }\newline
\index{States!extreme!notation}Extreme points of $E_{%
\vec{\ell}}$ are written as $\hat{\rho}\in \mathcal{E}_{\vec{\ell}}$ or
sometime $\hat{\omega}\in \mathcal{E}_{\vec{\ell}}$.
\end{notation}

The uniqueness of the probability measure $\mu _{\rho }$ given in Theorem %
\ref{theorem choquet} implies, by Theorem \ref{theorem choquet bis copy(1)},
that $E_{\vec{\ell}}$ is a (Choquet) \emph{simplex}%
\index{Simplex!Choquet} (see Definition \ref{gamm regularisation copy(1)}),
which is in fact a consequence of Lemma \ref{coro.even} together with the
asymptotic abelianess%
\index{Asymptotic abelianess} (\ref{asymtptic abelian}) of the even
sub--algebra $\mathcal{U}^{+}$ (\ref{definition of even operators}), see 
\cite[Corollary 4.3.11.]{BrattelliRobinsonI}. Observe also that the simplex $%
E_{%
\vec{\ell}}$ has a fairly complicated geometrical structure: For any $\vec{%
\ell}\in \mathbb{N}^{d}$, $\mathcal{E}_{\vec{\ell}}$ is a weak$^{\ast }$--%
\emph{dense} $G_{\delta }$ subset in $E_{\vec{\ell}}$, see Corollary \ref%
{lemma density of extremal points}. In fact, up to an affine homeomorphism
the set $E_{\vec{\ell}}$ is the \emph{Poulsen simplex}, see Theorem \ref%
{theorem Bauer copy(1)}:

\begin{theorem}[$E_{\vec{\ell}}$ and the Poulsen simplex]
\label{Thm Poulsen simplex}\mbox{ }\newline
The Choquet simplices $\{E_{\vec{\ell}}\}_{\vec{\ell}\in \mathbb{N}^{d}}$
are all affinely homeomorphic to the Poulsen simplex%
\index{Simplex!Poulsen|textbf}, i.e., $E_{%
\vec{\ell}}$ is unique up to an affine homeomorphism.
\end{theorem}

Note that the simplex $E_{\vec{\ell}}$ can also be seen as a simplexoid%
\index{Simplexoid}, i.e., a compact convex set in which all closed proper
faces\footnote{%
A face $F$ of a convex set $K$ is defined to be a subset of $K$ with the
property that, if $\rho =\lambda _{1}\rho _{1}+\cdots +\lambda _{n}\rho
_{n}\in F$ with $\rho _{1},\ldots ,\rho _{n}\in K$, $\lambda _{1},\ldots
,\lambda _{n}\in (0,1)$ and $\lambda _{1}+\cdots +\lambda _{n}=1$, then $%
\rho _{1},\ldots ,\rho _{n}\in F$.} are simplices. An example of a closed
face of $E_{%
\vec{\ell}}$, for any $\vec{\ell}\in \mathbb{N}^{d}$, is given by the Bauer
simplex%
\index{Simplex!Bauer} $E_{\Pi }\subseteq E_{%
\vec{\ell}}$ of 
\index{States!permutation invariant}permutation invariant states described
in\ Section \ref{set of permutations invariant states}.

\begin{remark}[Gauge invariant t.i. states]
\label{t.i. + gauge inv states}\mbox{ }\newline
An important subset of $E_{1}$ is the convex and weak$^{\ast }$--compact set 
\begin{equation*}
E_{1}^{\circ }:=\{\rho \in E_{1}\;:\;\rho =\rho \circ \sigma ^{\circ
}\}\subseteq E_{1}
\end{equation*}%
of translation and 
\index{Gauge invariant!states}gauge invariant states, cf. Remark \ref%
{proj.gauge.inv}. States describing physical systems generally belong to $%
E_{1}^{\circ }$ which is again the Poulsen simplex%
\index{Simplex!Poulsen} (up to an affine homeomorphism). This can be proven
by identifying $E_{1}^{\circ }$ with the set of all t.i. states on $\mathcal{%
U}^{\circ }$ (\ref{definition of gauge invariant operators}) which is an
asymptotically abelian $C^{\ast }$--algebra.
\end{remark}

The result of Theorem \ref{Thm Poulsen simplex} is standard in statistical
mechanics, in particular for lattice quantum spin systems \cite[p. 405--406,
464]{BrattelliRobinsonI}. It means that the complicated geometrical
structure of the simplices $E_{%
\vec{\ell}}$ is, in a sense, universal and in fact, physically natural.
Indeed, the set $\mathcal{E}_{\vec{\ell}}$ of extreme points of $E_{\vec{\ell%
}}$ can be characterized through a (physically natural) condition related to
space--averaging as follows.

For any $A\in \mathcal{U}$, $L\in \mathbb{N}$ and $\vec{\ell}\in \mathbb{N}%
^{d}$, let $A_{L,\vec{\ell}}\in \mathcal{U}$ be defined by the space--average%
\index{Space--averaging functional!operator|textbf}%
\begin{equation}
A_{L,%
\vec{\ell}}:=\frac{1}{|\Lambda _{L}\cap \mathbb{Z}_{\vec{\ell}}^{d}|}%
\sum\limits_{x\in \Lambda _{L}\cap \mathbb{Z}_{\vec{\ell}}^{d}}\alpha
_{x}(A).  \label{definition de A L}
\end{equation}%
By definition, $A_{L}:=A_{L,\vec{\ell}}$ for $\vec{\ell}=(1,\cdots ,1)$.
This sequence $\{A_{L,\vec{\ell}}\}_{L\in \mathbb{N}}$ of operators in $%
\mathcal{U}$ defines space--averaging functionals:{}

\begin{definition}[Space--averaging functionals]
\label{definition de deltabis}\mbox{ }\newline
\index{Space--averaging functional|textbf}For any $A\in \mathcal{U}$ and $%
\vec{\ell}\in \mathbb{N}^{d}$, the ($\vec{\ell}$--) space--averaging
functional is the map%
\begin{equation*}
\rho \mapsto \Delta _{A,\vec{\ell}}\left( \rho \right)
:=\lim\limits_{L\rightarrow \infty }\rho (A_{L,\vec{\ell}}^{\ast }A_{L,\vec{%
\ell}})
\end{equation*}%
from $E_{\vec{\ell}}$ to $\mathbb{R}$. Here, $\Delta _{A}:=\Delta
_{A,(1,\cdots ,1)}$.
\end{definition}

\noindent The functional $\Delta _{A,\vec{\ell}}$ is well--defined, for all $%
A\in \mathcal{U}$ and $\vec{\ell}\in \mathbb{N}^{d}$, and we give in Section %
\ref{Section space averaging}\ a complete description of $\Delta _{A}$. This
map is pivotal as it is used to define \emph{ergodic }states in the
following way:

\begin{definition}[Ergodic states]
\label{def:egodic}\mbox{ }\newline
\index{States!ergodic}A $%
\vec{\ell}$--periodic state $\hat{\rho}\in E_{\vec{\ell}}$ is ($\vec{\ell}$%
--) ergodic iff, for all $A\in \mathcal{U}$, 
\begin{equation*}
\Delta _{A,\vec{\ell}}\left( \hat{\rho}\right) =|\hat{\rho}(A)|^{2}.
\end{equation*}
\end{definition}

\noindent The equality in this definition says that space fluctuations of
measures on a system described by a $\mathbb{Z}_{\vec{\ell}}^{d}$--invariant
state $\hat{\rho}$ are small when it is ergodic: For any observable $A$, we
are able to determine $\hat{\rho}(A)$ through space--averaging over the
sub--lattice $\Lambda _{L}\cap \mathbb{Z}_{\vec{\ell}}^{d}$ at large $L$. We
can view this result as a non--commutative version of the law of large
numbers. Note that the term \textquotedblleft ergodic\textquotedblright\
comes from the fact we can replace a space average by the correponding
expectation value for these special states. The latter also holds for
polynomials of the space averages $A_{L,\vec{\ell}}$, see (\ref{ergodicity
eq}). Observe however that the linear case is trivial by periodicity of the
states.

The unique decomposition expressed in Theorem \ref{theorem choquet} of any $%
\rho \in E_{\vec{\ell}}$ in terms of extreme states $\hat{\rho}\in \mathcal{E%
}_{\vec{\ell}}$ of $E_{\vec{\ell}}$ is also called the \emph{ergodic
decomposition}. Indeed, we prove in Section \ref{Section theorem ergodic
extremalbis} that any ergodic state is an extreme state in $E_{\vec{\ell}}$
and vice versa, see Lemmata \ref{ergodic.extremal} and \ref{lemma
extremal.ergodic} together with Corollary \ref{Corollary Extremal states -
strongly clustering}.

\begin{theorem}[Extremality = Ergodicity]
\label{theorem ergodic extremal}\mbox{ }\newline
\index{States!extreme}%
\index{States!ergodic}Any extreme state $%
\hat{\rho}\in \mathcal{E}_{\vec{\ell}}$ of $E_{\vec{\ell}}$ is ergodic and
vice versa. Additionally, any extreme state $\hat{\rho}\in \mathcal{E}_{\vec{%
\ell}}$ is strongly clustering%
\index{States!strongly clustering}, i.e., for all $A,B\in \mathcal{U}$,%
\begin{equation*}
\lim\limits_{L\rightarrow \infty }%
\frac{1}{|\Lambda _{L}\cap \mathbb{Z}_{\vec{\ell}}^{d}|}\sum\limits_{y\in
\Lambda _{L}\cap \mathbb{Z}_{\vec{\ell}}^{d}}\hat{\rho}\left( \alpha
_{x}(A)\alpha _{y}(B)\right) =\hat{\rho}(A)\hat{\rho}(B)
\end{equation*}%
uniformly in $x\in \mathbb{Z}_{\vec{\ell}}^{d}$.
\end{theorem}

Observe that a strongly clustering state $\rho \in E_{\vec{\ell}}$ is not
necessarily \emph{strongly mixing}%
\index{States!strongly mixing} which means that%
\begin{equation}
\lim\limits_{|x|\rightarrow \infty }\rho \left( A\alpha _{x}(B)\right) =\rho
(A)\rho (B)  \label{mixing}
\end{equation}%
for all $A,B\in \mathcal{U}$. The converse is trivial: Any strongly mixing
state satisfies the ergodicity property.

\begin{remark}[Gauge invariant states and ergodicity]
\label{t.i. + gauge inv states ergodicity}\mbox{ }\newline
\index{Gauge invariant!ergodicity}From Remark \ref{t.i. + gauge inv states},
a state $%
\hat{\rho}\in E_{1}^{\circ }$ is extreme in $E_{1}^{\circ }$ iff $\hat{\rho}%
\in E_{1}^{\circ }$ is ergodic w.r.t. the sub--algebra $\mathcal{U}^{\circ
}\subseteq \mathcal{U}$, that is, for all $A\in \mathcal{U}^{\circ }$,%
\begin{equation*}
\lim\limits_{L\rightarrow \infty }\frac{1}{|\Lambda _{L}|^{2}}%
\sum\limits_{x,y\in \Lambda _{L}}\hat{\rho}(\alpha _{x}(A^{\ast })\alpha
_{y}(A))=|\hat{\rho}(A)|^{2}.
\end{equation*}%
Compare with Definition \ref{def:egodic}.
\end{remark}

\section{The space--averaging functional $\Delta _{A}$\label{Section space
averaging}%
\index{Space--averaging functional|textbf}}

The set of translation invariant (t.i.) states $E_{1}:=E_{(1,\cdots ,1)}$
and the space--averaging functional $\Delta _{A}:=\Delta _{A,(1,\cdots ,1)}$
play a central r{o}le below as we concentrate our attention on the
thermodynamics of translation invariant (t.i.) Fermi systems. However, our
analysis can easily be generalized to the ($%
\vec{\ell}$--) space--averaging functional $\Delta _{A,\vec{\ell}}$ for any $%
\vec{\ell}\in \mathbb{N}^{d}$, see Definition \ref{definition de deltabis}.

First, by Lemma \ref{Lemma1.vonN copy(2)}, the space--averaging functional $%
\Delta _{A}$ is well--defined for all $\vec{\ell}$--periodic states $\rho
\in E_{\vec{\ell}}$ at any $\vec{\ell}\in \mathbb{N}^{d}$. In this case,%
\begin{equation}
\rho \mapsto \Delta _{A}\left( \rho \right) :=\lim\limits_{L\rightarrow
\infty }\rho \left( A_{L}^{\ast }A_{L}\right) \in \left[ |\rho (A_{\vec{\ell}%
})|^{2},\Vert A\Vert ^{2}\right] ,  \label{delta bounded}
\end{equation}%
with 
\index{Space--averaging functional!operator|textbf}%
\begin{equation}
A_{%
\vec{\ell}}:=\frac{1}{\ell _{1}\cdots \ell _{d}}\sum\limits_{x=(x_{1},\cdots
,x_{d}),\;x_{i}\in \{0,\cdots ,\ell _{i}-1\}}\alpha _{x}(A)
\label{definition de A l}
\end{equation}%
for any $\vec{\ell}\in \mathbb{N}^{d}$.

As explained in the previous section, extremality of t.i. states can be
characterized by means of the space--averaging functional $\Delta _{A}$.
Indeed, the set of t.i. states $\rho \in E_{1}$ which fulfill $\Delta
_{A}\left( \rho \right) =|\rho (A)|^{2}$ for any $A\in \mathcal{U}$, i.e.,
the set of ergodic states (Definition \ref{def:egodic}), is the set $%
\mathcal{E}_{1}$ of extreme states of $E_{1}$, see Theorem \ref{theorem
ergodic extremal}. Nevertheless, this functional has never gained much
attention before beyond the fact that it can be used to characterize
extremality of states. It turns out that other properties of the
space--averaging functional are also crucial in the analysis of
thermodynamic effects of long--range interactions. Its basic properties --
proven in Lemmata \ref{Lemma1.vonN copy(3)} and \ref{Lemma1.vonN copy(5)} --
are listed in the following theorem:

\begin{theorem}[Properties of the functional $\Delta _{A}$ on $E_{\vec{\ell}%
} $]
\label{Lemma1.vonN copy(4)}\mbox{ }\newline
\emph{(i) }At fixed $A\in \mathcal{U}$, the map $\rho \mapsto \Delta
_{A}(\rho )$ from $E_{\vec{\ell}}$ to $\mathbb{R}_{0}^{+}$ is a weak$^{\ast
} $--upper semi--continuous affine functional. It is also t.i., i.e., for
all $x\in \mathbb{Z}^{d}$ and $\rho \in E_{\vec{\ell}}$, $\Delta _{A}(\rho
\circ \alpha _{x})=\Delta _{A}(\rho )$.\newline
\emph{(ii)} At fixed $\rho \in E_{\vec{\ell}}$ and for all $A,B\in \mathcal{U%
}$, 
\begin{equation*}
|\Delta _{A}\left( \rho \right) -\Delta _{B}\left( \rho \right) |\leq (\Vert
A\Vert +\Vert B\Vert )\Vert A-B\Vert .
\end{equation*}%
In particular, the map $A\mapsto \Delta _{A}\left( \rho \right) $ from $%
\mathcal{U}$ to $\mathbb{R}_{0}^{+}$ is locally Lipschitz continuous.
\end{theorem}

The affinity and the translation invariance of $\Delta _{A}$, as well as
(ii), are immediate consequences of its definition (see Lemmata \ref%
{Lemma1.vonN copy(3)} and \ref{Lemma1.vonN copy(5)}). Its weak$^{\ast }$%
--upper semi--continuity follows from the fact that $\Delta _{A}$ is the
infimum of a family of weak$^{\ast }$--continuous functionals (see Lemmata %
\ref{Lemma1.vonN copy(2)} and \ref{Lemma1.vonN copy(3)}).

Note that $\Delta _{A}$ is not weak$^{\ast }$--continuous for all $A\in 
\mathcal{U}$, even on the set $E_{1}$. Indeed, if $\Delta _{A}$ is weak$%
^{\ast }$--continuous on $E_{1}$ then $\Delta _{A}\left( \rho \right) =|\rho
(A)|^{2}$ for all $\rho \in E_{1}$ because of Theorem \ref{theorem ergodic
extremal} and the weak$^{\ast }$--density of the set $\mathcal{E}_{1}$ in $%
E_{1}$ (Corollary \ref{lemma density of extremal points}). Therefore, there
exists $A\in \mathcal{U}$ such that $\Delta _{A}$ is not weak$^{\ast }$%
--continuous. Otherwise, any state $\rho \in E_{1}$ would be ergodic and
hence, an extreme point of $E_{1}$ by Theorem \ref{theorem ergodic extremal}%
. A more detailed study on the weak$^{\ast }$--continuity of the
space--averaging functional $\Delta _{A}$ on the set $E_{1}$ of t.i. states
is given by the following theorem:

\begin{theorem}[Properties of the map $\protect\rho \mapsto \Delta
_{A}\left( \protect\rho \right) $ on $E_{1}$ at fixed $A\in \mathcal{U}$]
\label{Lemma1.vonN}\mbox{ }\newline
\emph{(i) }$\Delta _{A}$ is weak$^{\ast }$--continuous on $E_{1}$ iff the
affine map $\rho \mapsto |\rho (A)|$ from $E_{1}$ to $\mathbb{C}$ is a
constant map.\newline
\emph{(ii) }$\Delta _{A}$ is weak$^{\ast }$--discontinuous on a weak$^{\ast
} $--dense subset of $E_{1}$ unless $\rho \mapsto |\rho (A)|$ is a constant
map from $E_{1}$ to $\mathbb{C}$.\newline
\emph{(iii)} $\Delta _{A}$ is continuous on the $G_{\delta }$ weak$^{\ast }$%
--dense subset $\mathcal{E}_{1}$ of extreme states of $E_{1}$. In
particular, the set of all states of $E_{1}$ where $\Delta _{A}$ is weak$%
^{\ast }$--discontinuous is weak$^{\ast }$--meager.\newline
\emph{(iv)} $\Delta _{A}$ can be decomposed in terms of an integral on the
set $\mathcal{E}_{1}$, i.e., for all $\rho \in E_{1}$, 
\begin{equation*}
\Delta _{A}\left( \rho \right) =\int_{\mathcal{E}_{1}}\mathrm{d}\mu _{\rho
}\left( \hat{\rho}\right) \;\left\vert \hat{\rho}\left( A\right) \right\vert
^{2}
\end{equation*}%
with the probability measure $\mu _{\rho }$ defined by Theorem \ref{theorem
choquet}. \newline
\emph{(v)} Its $\Gamma $--regularization%
\index{Gamma--regularization!space--averaging functional} $\Gamma
_{E_{1}}\left( \Delta _{A}\right) $ on $E_{1}$\ is the weak$^{\ast }$%
--continuous convex map $\rho \mapsto \left\vert \rho \left( A\right)
\right\vert ^{2}$.
\end{theorem}

\noindent Recall that the $\Gamma $--regularization of functionals is
defined in Definition \ref{gamm regularisation}. For more details, we
recommend Section \ref{Section gamma regularization} as well as Corollary %
\ref{Biconjugate} in\ Section \ref{Section Legendre-Fenchel transform}.

The continuity properties (i)--(iii) result partially from Theorems \ref%
{theorem choquet} and \ref{theorem ergodic extremal}, for more details see
Proposition \ref{Lemma1.vonN copy(1)}. The assertion (iv) is a direct
consequence of Theorem \ref{Lemma1.vonN copy(4)} (i) and Lemma \ref%
{Corollary 4.1.18.} combined with Theorems \ref{theorem choquet} and \ref%
{theorem ergodic extremal}. The last statement (v) is deduced from the
density of the set $\mathcal{E}_{1}$ in $E_{1}$ (Corollary \ref{lemma
density of extremal points}) together with Theorems \ref{theorem ergodic
extremal} and standard arguments from convex analysis, see Lemma \ref{Lemma
convex hull}.

\begin{remark}[$\Delta _{A}$ and Jensen's inequality]
\label{Remark lower bound delta}\mbox{ }\newline
\index{Jensen's inequality}The inequality $\Delta _{A}(\rho )\geq |\rho
(A)|^{2}$ can directly be deduced from Theorem \ref{Lemma1.vonN} (iv) and
Jensen's inequality (Lemma \ref{Jensen inequality}) as $\mu _{\rho }$ is a
probability measure.
\end{remark}

\begin{remark}[Trivial space--averaging functional $\Delta _{A}$ on $E_{1}$]

\label{Remark convention0}\mbox{ }\newline
If the affine map $\rho \mapsto |\rho (A)|$ from $E_{1}$ to $\mathbb{C}$ is
a constant map then from Theorem \ref{Lemma1.vonN} (iv), $\Delta _{A}\left(
\rho \right) =|\rho (A)|^{2}$ for any $\rho \in E_{1}$. An example of such
trivial behavior is given by choosing $A=\lambda \mathbf{1}+B-\alpha
_{x}\left( B\right) $ for any $\lambda \in \mathbb{C}$, $B\in \mathcal{U}$,
and $x\in \mathbb{Z}^{d}$. Recall that the translation $\alpha _{x}$ is the $%
\ast $--automorphism defined by (\ref{transl}).
\end{remark}

\section{Local interactions and internal energies\label{Local interactions}}

An \emph{interaction} is defined via a family of even and self--adjoint
local elements $\Phi _{\Lambda }$ and it is associated with \emph{internal
energies} as follows:

\begin{definition}[Interactions and internal energies]
\label{definition standard interaction}\mbox{ }\newline
\emph{(i)} 
\index{Interaction}An interaction is a family $\Phi =\{\Phi _{\Lambda
}\}_{\Lambda \in \mathcal{P}_{f}(\mathfrak{L})}$ of even and self--adjoint
local elements $\Phi _{\Lambda }=\Phi _{\Lambda }^{\ast }\in \mathcal{U}%
^{+}\cap \mathcal{U}_{\Lambda }$ with $\Phi _{\emptyset }=0$.\newline
\emph{(ii)} 
\index{Interaction!internal energy}For any $\Lambda \in \mathcal{P}_{f}(%
\mathfrak{L})$, its internal energy is the local Hamiltonian%
\begin{equation*}
U_{\Lambda }^{\Phi }:=\sum\limits_{\Lambda ^{\prime }\in \mathcal{P}%
_{f}(\Lambda )}\Phi _{\Lambda ^{\prime }}\in \mathcal{U}^{+}\cap \mathcal{U}%
_{\Lambda }.
\end{equation*}
\end{definition}

\begin{notation}[Interactions]
\label{Notation5}\mbox{ }\newline
\index{Interaction!notation}The letters $\Phi $ and $\Psi $ are exclusively
reserved to denote interactions.
\end{notation}

An interaction $\Phi $ is by definition \emph{translation invariant} (t.i.)%
\index{Interaction!translation invariant|textbf} iff, for all $x\in \mathbb{Z%
}^{d}$ and $\Lambda \in \mathcal{P}_{f}(\mathfrak{L})$, $\Phi _{\Lambda
+x}=\alpha _{x}(\Phi _{\Lambda })$ with%
\begin{equation}
\Lambda +x:=\{x^{\prime }+x\in \mathfrak{L}\;:\;x^{\prime }\in \Lambda \}.
\label{definition de lambda translate}
\end{equation}%
Another important symmetry of Fermi models, which appears together with the
translation invariance in most physically relevant situations, is the gauge
symmetry (or the particle number conservation). An interaction $\Phi $ is
said to be 
\index{Gauge invariant!interactions}\emph{gauge invariant} 
\index{Interaction!gauge invariant}(i.e., $\Phi $ conserves the particle
number) iff $\Phi _{\Lambda }\in \mathcal{U}^{\circ }$ for all $\Lambda \in 
\mathcal{P}_{f}(\mathfrak{L})$, see (\ref{definition of gauge invariant
operators}).

Observe now that an interaction $\Phi $ may have finite range. This property
is defined via the Euclidean metric $d:\mathfrak{L}\times \mathfrak{L}%
\rightarrow \lbrack 0,\infty )$ defined by%
\begin{equation}
d(x,x^{\prime }):=%
\sqrt{|x_{1}-x_{1}^{\prime }|^{2}+\dots +|x_{d}-x_{d}^{\prime }|^{2}}
\label{def.dist}
\end{equation}%
on the lattice $\mathfrak{L}:=\mathbb{Z}^{d}$ together with the function%
\index{Diameter} 
\begin{equation}
{\o }(\Lambda ):=\max\limits_{x,x^{\prime }\in \Lambda }\{d(x,x^{\prime
})\}\quad \mathrm{for\ any\ }\Lambda \in \mathcal{P}_{f}(\mathfrak{L}).
\label{diameter of lambda}
\end{equation}%
Indeed, we say that the interaction $\Phi $ has \emph{finite range} 
\index{Interaction!finite range}iff there is some $R<\infty $ such that ${\o 
}(\Lambda )>R$ implies $\Phi _{\Lambda }=0$.

The set of all interactions can be endowed with a real vector space
structure: 
\begin{equation*}
\left( \lambda _{1}\Phi +\lambda _{2}\Psi \right) _{\Lambda }:=\lambda
_{1}\Phi _{\Lambda }+\lambda _{2}\Psi _{\Lambda }
\end{equation*}%
for any interactions $\Phi $, $\Psi $, and any real numbers $\lambda
_{1},\lambda _{2}$. So, we can define a Banach space $\mathcal{W}_{1}$ of
t.i. interactions by using a specific norm:

\begin{definition}[Banach space $\mathcal{W}_{1}$ of t.i. interactions]
\label{definition banach space interaction}\mbox{ }\newline
\index{Interaction!translation invariant}The real Banach space $\mathcal{W}%
_{1}$ is the set of all t.i.\ interactions $\Phi $ with finite norm 
\begin{equation*}
\Vert \Phi \Vert _{\mathcal{W}_{1}}:=\sum\limits_{\Lambda \in \mathcal{P}%
_{f}(\mathfrak{L}),\;\Lambda \ni 0}|\Lambda |^{-1}\ \Vert \Phi _{\Lambda
}\Vert <\infty .
\end{equation*}
\end{definition}

\noindent The norm $\Vert \,\cdot \,\Vert _{\mathcal{W}_{1}}$ plays here an
important r{o}le because its finiteness implies, among other things, the
existence of the pressure in the thermodynamic limit (cf. Theorem \ref{BCS
main theorem 1}). The set $\mathcal{W}_{1}^{\mathrm{f}}$ of all finite range
t.i. interactions is dense in $\mathcal{W}_{1}$. In particular, the set $%
\mathcal{W}_{1}$ is a separable Banach space because, for all $\Lambda \in 
\mathcal{P}_{f}(\mathfrak{L})$, the local algebras $\mathcal{U}_{\Lambda }$
are finite dimensional.

By Corollary \ref{thm locally convex spacebis}, its dual $\mathcal{W}%
_{1}^{\ast }$ is a locally convex real space\footnote{%
We use here Rudin's definition, see Definition \ref{space}.} w.r.t. the weak$%
^{\ast }$--topology. The weak$^{\ast }$--topology is Hausdorff and, by
Theorem \ref{Metrizability}, it is metrizable on any weak$^{\ast }$--compact
subset of $\mathcal{W}_{1}^{\ast }$ as, for instance, on the weak$^{\ast }$%
--compact convex set $E_{1}$ seen as as a subset of $\mathcal{W}_{1}^{\ast }$%
, see Section \ref{Section state=functional on W} for more details.

\begin{remark}[Invariance property of the norm $\Vert \,\cdot \,\Vert _{%
\mathcal{W}_{1}}$]
\mbox{ }\newline
For any $\Phi \in \mathcal{W}_{1}$, we can define another interaction $\Psi
\in \mathcal{W}_{1}$ by viewing each $\Phi _{\Lambda }\in \mathcal{U}%
_{\Lambda }$ as an element $\Phi _{\Lambda }\in \mathcal{U}_{\Lambda
^{\prime }}$ for some set $\Lambda ^{\prime }\varsupsetneq \Lambda $ much
larger than $\Lambda $. One clearly has $\Psi \neq \Phi $, but the norm
stays invariant, i.e., $\Vert \Phi \Vert _{\mathcal{W}_{1}}=\Vert \Psi \Vert
_{\mathcal{W}_{1}}$, because the factor $|\Lambda ^{\prime }|^{-1}$ is
compensated by the larger number of translates of $\Lambda ^{\prime }$
containing $0$.
\end{remark}

\begin{remark}[Generalizations of the norm $\Vert \,\cdot \,\Vert _{\mathcal{%
W}_{1}}$]
\label{remark general interaction0}\mbox{ }\newline
The norm in Definition \ref{definition banach space interaction} is only a
specific example of the general class of norms for t.i. interactions: 
\begin{equation*}
\Vert \Phi \Vert _{\varkappa }:=\sum\limits_{\Lambda \in \mathcal{P}_{f}(%
\mathfrak{L}),\;\Lambda \ni 0}\varkappa \left( |\Lambda |,\o (\Lambda
)\right) \Vert \Phi _{\Lambda }\Vert \quad \mathrm{with\ }\varkappa \left(
x,y\right) >0.
\end{equation*}
\end{remark}

\begin{remark}[Banach space of standard potentials]
\label{remark general interaction}\mbox{ }\newline
\index{Interaction!standard potentials}In \cite[Definition 5.10]%
{Araki-Moriya} the authors use another kind of norm for t.i. interactions.
Their norm is not equivalent to $\Vert \,\cdot \,\Vert _{\mathcal{W}_{1}}$
and also defines a Banach space of the so--called translation covariant
potentials, see \cite[Proposition 8.8.]{Araki-Moriya}. In fact, in contrast
to the potentials of \cite[Section 5.5]{Araki-Moriya} we cannot associate a
symmetric derivation\footnote{%
A symmetric derivation $A\mapsto \delta (A)$ is a linear map satisfying $%
\delta (AB)=\delta (A)B+A\delta (B)$ and $\delta (A^{\ast })=\delta
(A)^{\ast }$ for any $A,B\in \mathcal{D}_{\delta }$ with its domain $%
\mathcal{D}_{\delta }$ being a dense $\ast $--sub-algebra of $\mathcal{U}$.}%
, as it is done in \cite[Theorem 5.7]{Araki-Moriya}, to all t.i.
interactions of $\mathcal{W}_{1}$. But no dynamical questions -- as, for
instance, the existence and characterization of KMS--states done in \cite%
{Araki-Moriya} -- are addressed in the present monograph. That is why we can
use here (in a sense) weaker norms leading to more general classes of t.i.
local interactions than in \cite{Araki-Moriya}.
\end{remark}

\section{Energy and entropy densities\label{section entropie}}

As far as the thermodynamics of Fermi systems is concerned, there are two
other important functionals associated with any $%
\vec{\ell}$--periodic state $\rho \in E_{\vec{\ell}}$ on $\mathcal{U}$: The 
\emph{entropy density} functional $\rho \mapsto s(\rho )$ and the \emph{%
energy} \emph{density} functional $\rho \mapsto e_{\Phi }(\rho )$ w.r.t. a
local t.i. interaction $\Phi \in \mathcal{W}_{1}$. We start with the entropy
density functional which is defined as follows:

\begin{definition}[Entropy density functional $s$]
\label{entropy.density}\mbox{ }\newline
\index{Entropy density functional|textbf}The entropy density functional $%
s:E_{%
\vec{\ell}}\rightarrow \mathbb{R}_{0}^{+}$ is defined by%
\begin{equation*}
s(\rho ):=-\lim\limits_{L\rightarrow \infty }\left\{ \frac{1}{|\Lambda _{L}|}%
\mathrm{Trace}\,\left( \mathrm{d}_{\rho _{\Lambda _{L}}}\ln \mathrm{d}_{\rho
_{\Lambda _{L}}}\right) \right\} ,
\end{equation*}%
where $\rho _{\Lambda _{L}}$ is the restriction of any $\rho \in E_{\vec{\ell%
}}$ on the sub--algebra $\mathcal{U}_{\Lambda _{L}}$ and $\mathrm{d}_{\rho
_{\Lambda _{L}}}\in \mathcal{U}_{\Lambda _{L}}$ is the (uniquely defined)
density matrix representing the state $\rho _{\Lambda _{L}}$ as a trace: 
\begin{equation*}
\rho _{\Lambda _{L}}(\cdot )=\mathrm{Trace}\,\left( \;\cdot \;\mathrm{d}%
_{\rho _{\Lambda _{L}}}\right) .
\end{equation*}
\end{definition}

\noindent The entropy density is therefore given as the so--called von
Neumann entropy per unit volume in the thermodynamic limit, cf. Section \ref%
{section neuman entropy}. The functional $s$ is well--defined on the set $E_{%
\vec{\ell}}$ of $\mathbb{Z}_{\vec{\ell}}^{d}$--invariant states because of
Lemma \ref{lemma property entropybis}. See also \cite[Section 3]%
{Araki-Moriya}. In fact, it has the following properties:

\begin{lemma}[Properties of the entropy density functional $s$]
\label{lemma property entropy}\mbox{ }\newline
\emph{(i) }The map $\rho \mapsto s(\rho )$ from $E_{\vec{\ell}}$ to $\mathbb{%
R}_{0}^{+}$ is a weak$^{\ast }$--upper semi--continuous affine functional.
It is also t.i., i.e., for all $x\in \mathbb{Z}^{d}$ and $\rho \in E_{\vec{%
\ell}}$, $s(\rho \circ \alpha _{x})=s(\rho )$.\newline
\emph{(ii)} For any t.i. state $\rho \in E_{1}$, there is a sequence $\{\hat{%
\rho}_{n}\}_{n=1}^{\infty }\subseteq \mathcal{E}_{1}$ of ergodic states
converging in the weak$^{\ast }$--topology to $\rho $ and such that 
\begin{equation*}
s(\rho )=\lim\limits_{n\rightarrow \infty }s(\hat{\rho}_{n}).
\end{equation*}%
\emph{(iii)} The map $\rho \mapsto s(\rho )$ from $E_{\vec{\ell}}$ to $%
\mathbb{R}_{0}^{+}$ is Lipschitz continuous in the norm topology of states:
For any $\rho ,\varrho \in E_{\vec{\ell}}$, 
\begin{equation*}
|s(\rho )-s(\varrho )|\leq C_{|\mathrm{S}|}\ \Vert \rho -\varrho \Vert
\qquad \mathrm{with\quad }\Vert \rho \Vert :=\sup\limits_{A\in \mathcal{U}%
,\;A=A^{\ast },\;\Vert A\Vert =1}|\rho (A)|.
\end{equation*}%
Here, $C_{|\mathrm{S}|}$ is a finite constant depending on the size $|%
\mathrm{S}|$ of the spin set $\mathrm{S}$.
\end{lemma}

\noindent The assertions (i) and (iii) are two standard results, see, e.g., 
\cite[Theorem 10.3. and Corollary 10.5.]{Araki-Moriya}. The proof of (i) is
shortly checked in Lemma \ref{lemma property entropybis} but we omit the
proof of (iii) which is only used in Remark \ref{remark boundedness entropy}%
. However, the second one (ii) does not seem to have been observed before
although it is not difficult to prove, see Lemma \ref{lemma property entropy
copy(1)}. This property turns out to be \emph{crucial} because it allows us
to go around the lack of weak$^{\ast }$--continuity of the entropy density
functional $s$. The map $\rho \mapsto s(\rho )$ is, indeed, not weak$^{\ast
} $--continuous but only norm continuous as expressed by (iii), see, e.g., 
\cite{Fannes,Fannesbis}. Note that (ii) uses the fact that the set $\mathcal{%
E}_{1}$ of extreme states is a dense subset of $E_{1}$ as explained after
Notation \ref{notation extreme states}, see also Corollary \ref{lemma
density of extremal points}.

\begin{remark}[Boundedness of the entropy density functional $s$]
\label{remark boundedness entropy}\mbox{ }\newline
The third assertion (iii) of Lemma \ref{lemma property entropy} is given for
information as it is only used in the monograph to see that $s(\rho )\in %
\left[ 0,2C_{|\mathrm{S}|}\right] $ for all $\rho \in E_{1}$ because there
is $\varrho \in E_{1}$ such that $s(\varrho )=0$ and $\Vert \rho -\varrho
\Vert \leq \Vert \rho \Vert +\Vert \varrho \Vert =2$. Similarly, for quantum
spin systems (cf. Remark \ref{Quantum spin systems})%
\index{Quantum spin systems!entropy} the entropy density functional belongs
to $[0,D_{|\mathrm{S}|}]$ with $D_{|\mathrm{S}|}<\infty $. In particular, it
is still bounded from below.
\end{remark}

The energy density is the thermodynamic limit of the internal energy $%
U_{\Lambda }^{\Phi }$ (Definition \ref{definition standard interaction}) per
unit volume associated with any fixed local interaction $\Phi \in \mathcal{W}%
_{1}$:

\begin{definition}[Energy density functional $e_{\Phi }$]
\label{definition energy density}\mbox{ }\newline
\index{Energy density functional|textbf}The energy density of any $%
\vec{\ell}$--periodic state $\rho \in E_{\vec{\ell}}$ w.r.t. a t.i. local
interaction $\Phi \in \mathcal{W}_{1}$ is defined by%
\begin{equation*}
e_{\Phi }(\rho ):=\lim\limits_{L\rightarrow \infty }\frac{\rho \left(
U_{\Lambda _{L}}^{\Phi }\right) }{|\Lambda _{L}|}<\infty .
\end{equation*}
\end{definition}

\noindent The existence of the energy density $e_{\Phi }(\rho )$ can easily
be checked for all $\Phi \in \mathcal{W}_{1}$, see Lemma \ref{lemma energy
density exists}. Actually, $e_{\Phi }(\rho )=\rho (\mathfrak{e}_{\Phi ,\vec{%
\ell}})$ with%
\index{Energy observable}%
\begin{equation}
\mathfrak{e}_{\Phi ,%
\vec{\ell}}:=\frac{1}{\ell _{1}\cdots \ell _{d}}\sum\limits_{x=(x_{1},\ldots
,x_{d}),\;x_{i}\in \{0,\cdots ,\ell _{i}-1\}}\quad \sum\limits_{\Lambda \in 
\mathcal{P}_{f}(\mathfrak{L}),\;\Lambda \ni x}\frac{\Phi _{\Lambda }}{%
|\Lambda |}  \label{eq:enpersite}
\end{equation}%
for any $\Phi \in \mathcal{W}_{1}$. Per definition, $\mathfrak{e}_{\Phi }:=%
\mathfrak{e}_{\Phi ,(1,1,\ldots ,1)}$. The operator $\mathfrak{e}_{\Phi ,%
\vec{\ell}}\in \mathcal{U}^{+}$ is called the \emph{energy observable}
associated with the t.i. local interaction $\Phi \in \mathcal{W}_{1}$ for
the set $E_{\vec{\ell}}$ of $\vec{\ell}$--periodic states. Remark that $%
\mathfrak{e}_{\Phi ,\vec{\ell}}\in \mathcal{U}^{+}$ results from the fact
that, for all $\Lambda \in \mathcal{P}_{f}(\mathfrak{L})$, $\Phi _{\Lambda
}\in \mathcal{U}^{+}\cap \mathcal{U}_{\Lambda }$ and 
\begin{equation}
\Vert \mathfrak{e}_{\Phi ,\vec{\ell}}\Vert \leq \Vert \Phi \Vert _{\mathcal{W%
}_{1}}<\infty .  \label{eq:enpersite bounded}
\end{equation}%
Observe additionally that%
\begin{equation*}
\mathfrak{e}_{\Phi ,\vec{\ell}}=\frac{1}{\ell _{1}\cdots \ell _{d}}%
\sum\limits_{x=(x_{1},\ldots ,x_{d}),\;x_{i}\in \{0,\cdots ,\ell
_{i}-1\}}\alpha _{x}(\mathfrak{e}_{\Phi }).
\end{equation*}

It is then straightforward to prove the following properties of the energy
density functional $e_{\Phi }$ (see also \cite[Theorem 9.5]{Araki-Moriya}):

\begin{lemma}[Properties of the energy density functional $e_{\Phi }$]
\label{Th.en.func}\mbox{ }\newline
\emph{(i)} For any $\Phi \in \mathcal{W}_{1}$, the map $\rho \mapsto e_{\Phi
}(\rho )$ from $E_{\vec{\ell}}$ to $\mathbb{R}$ is a weak$^{\ast }$%
--continuous affine functional. It is also t.i., i.e., for all $x\in \mathbb{%
Z}^{d}$, $e_{\Phi }(\rho \circ \alpha _{x})=e_{\Phi }(\rho )$.\newline
\emph{(ii)} At fixed $\rho \in E_{\vec{\ell}}$ and for all $\Phi ,\Psi \in 
\mathcal{W}_{1}$, 
\begin{equation*}
|e_{\Phi }\left( \rho \right) -e_{\Psi }\left( \rho \right) |=|e_{\Phi -\Psi
}\left( \rho \right) |\leq \Vert \Phi -\Psi \Vert _{\mathcal{W}_{1}}.
\end{equation*}%
In particular, the linear map $\Phi \mapsto e_{\Phi }\left( \rho \right) $
from $\mathcal{W}_{1}$ to $\mathbb{R}$ is Lipschitz continuous.
\end{lemma}

Note that the entropy density functional $s$ and the energy density
functional $e_{\Phi }$ define the so--called \emph{free--energy density}
functional $f_{\Phi }$:

\begin{definition}[Free--energy density functional $f_{\Phi }$]
\label{Remark free energy density}\mbox{ }\newline
\index{Free--energy density functional!local|textbf}For $\beta \in (0,\infty
]$, the free--energy density functional $f_{\Phi }$ w.r.t. the t.i.
interaction $\Phi \in \mathcal{W}_{1}$ is the map 
\begin{equation*}
\rho \mapsto f_{\Phi }(\rho ):=e_{\Phi }(\rho )-\beta ^{-1}s(\rho )
\end{equation*}%
from $E_{%
\vec{\ell}}$ to $\mathbb{R}$.
\end{definition}

\noindent From Lemmata \ref{lemma property entropy} (i) and \ref{Th.en.func}
(i), the functional $f_{\Phi }$ is weak$^{\ast }$--lower semi--continuous,
t.i., and affine. Moreover, by Lemma \ref{lemma property entropy} (ii), for
any $\rho \in E_{1}$, there is a sequence $\{\hat{\rho}_{n}\}_{n=1}^{\infty
}\subseteq \mathcal{E}_{1}$ of ergodic states converging in the weak$^{\ast
} $--topology to $\rho $ and such that 
\begin{equation}
f_{\Phi }(\rho )=\lim\limits_{n\rightarrow \infty }f_{\Phi }(\hat{\rho}_{n}).
\label{quasi-continuity}
\end{equation}

\begin{remark}[Temperature of Fermi systems]
\label{Temperature of Fermi systems}\mbox{ }\newline
\index{Temperature}All assertions in the sequel depend on the fixed positive
parameter $\beta >0$. $\beta $ is often omitted to simplify the notation,
but we keep it in all definitions. $\beta \in (0,\infty ]$ is interpreted in
Physics as being the inverse temperature of the system. $\beta =\infty $
corresponds to the zero--temperature for which the contribution of (thermal)
entropy density to the free energy density disappears. In fact, the
free--energy density corresponds to the maximum energy which can be
extracted from a thermodynamical system at fixed temperature $\beta ^{-1}$.
\end{remark}

\chapter{Fermi Systems with Long--Range Interactions\label{equilibrium}}

As explained in Chapter \ref{Section definition}, a physical system can be
described by an interaction which defines an internal energy for any bounded
set $\Lambda \subseteq \mathfrak{L}$ (box) of the lattice $\mathfrak{L}$. A
typical example of interactions are the elements $\Phi $ of the Banach space 
$\mathcal{W}_{1}$ of local interactions described in\ Section \ref{Local
interactions}. Unfortunately, $\mathcal{W}_{1}$ is too small to include all
physically interesting systems. Indeed, any interaction 
\begin{equation*}
\Phi =\{\Phi _{\Lambda }\}_{\Lambda \in \mathcal{P}_{f}(\mathfrak{L})}\in 
\mathcal{W}_{1}
\end{equation*}%
is short range, or \emph{weakly} long--range, in the sense that the norm $%
\Vert \Phi _{\Lambda }\Vert $ has to decrease sufficiently fast as the
volume $|\Lambda |$ of the bounded set $\Lambda \subseteq \mathfrak{L}$
increases. Note that some authors (see, e.g., \cite{Israel-long-range})
refer to the space $\mathcal{W}_{1}$ as a space of long--range interactions
because, even if 
\begin{equation*}
\sum_{\Lambda \ni 0}|\Lambda |^{-1}\ \Vert \Phi _{\Lambda }\Vert <\infty
\end{equation*}%
has to be finite, the numbers 
\begin{equation*}
\sup \left\{ \Vert \Phi _{\Lambda }\Vert \,:\,%
\text{{\o }}(\Lambda )>D\right\}
\end{equation*}%
can decay arbitrarily slowly as $D\rightarrow \infty $. Here, ${\o }(\Lambda
)$ stands for the diameter of $\Lambda \subseteq \mathfrak{L}$, see (\ref%
{diameter of lambda}). Elements of $\mathcal{W}_{1}$ are called in this
monograph \emph{weakly} long--range because they do not include important
physical models with interactions which are long--range in a stronger sense,
for instance those describing conventional superconductivity. Therefore, in\
Section \ref{definition models} we embed the space $\mathcal{W}_{1}$ in a
Banach space $\mathcal{M}_{1}$ of (strong) long--range interactions which
includes physical models like those of conventional superconductivity (BCS
models). We then analyze in the following sections the thermodynamics of any
model $\mathfrak{m}\in \mathcal{M}_{1}$.

Indeed, note first that, for any $\mathfrak{m}\in \mathcal{M}_{1}$, all the
correlation functions w.r.t. the equilibrium state at inverse temperature $%
\beta >0$ of the corresponding physical system restricted to some bounded
set $\Lambda \subseteq \mathfrak{L}$ are encoded in the partition function $%
Z_{\Lambda ,\mathfrak{m}}$ which defines a finite--volume pressure 
\begin{equation*}
p_{\Lambda ,\mathfrak{m}}:=\beta ^{-1}|\Lambda |^{-1}\ln Z_{\Lambda ,%
\mathfrak{m}},
\end{equation*}%
see Section \ref{Section Gibbs equilibrium states}. A first question is thus
to analyze the thermodynamic limit ($\Lambda \nearrow \mathfrak{L)}$ of $%
p_{\Lambda ,\mathfrak{m}}$, i.e., the infinite--volume pressure $\mathrm{P}_{%
\mathfrak{m}}^{\sharp }$. This study is presented in Section \ref{existence
of thermodynamics} and generalizes some previous results of \cite%
{Israel,BrattelliRobinson,Araki-Moriya} to the larger space $\mathcal{M}_{1}$
(see Remark \ref{remark general interaction}). In particular, we show that $%
\mathrm{P}_{\mathfrak{m}}^{\sharp }$ is given by the minimization of two
different free--energy density functionals $f_{\mathfrak{m}}^{\sharp }$ and $%
g_{\mathfrak{m}}$ on the set $E_{1}$ of translation invariant (t.i.) states: 
\begin{equation}
\mathrm{P}_{\mathfrak{m}}^{\sharp }=-\inf \,f_{\mathfrak{m}}^{\sharp
}(E_{1})=-\inf g_{\mathfrak{m}}\left( E_{1}\right) .  \label{var problemsup}
\end{equation}%
The latter corresponds to Theorem \ref{BCS main theorem 1} which shares some
similarities with results previously obtained for quantum spins or for some
rather particular long--range Fermi systems \cite%
{RaggioWerner2,Petz2008,monsieurremark}. For more details on the results of 
\cite{RaggioWerner2,Petz2008,monsieurremark} see discussions after Theorem %
\ref{BCS main theorem 1}.

The rest of the chapter presents new\footnote{%
But we recommend {Section \ref{Concluding remark} and \ref{Section approx
method} which explains previous results on the pressure only.}} results for
both quantum spins and Fermi systems with long--range interactions. In
particular, an important novelty of this monograph is to give {a precise
picture of the thermodynamic impact of long--range interactions }and, with
this, a first answer to an old open problem in mathematical physics -- first
addressed by Ginibre \cite[p. 28]{Ginibre} in 1968 within a different
context -- about the validity of the so--called Bogoliubov approximation on
the level of states. Observe also that interesting hints about this kind of
question can be found in \cite{Fannne-Pule-Verbeure1,Fannne-Pule-Verbeure2}
for Bose systems.

Indeed, similarly to finite--volume cases (cf. Section \ref{Section Gibbs
equilibrium states}), we define in Section \ref{Section equilibrium states}
the (possibly generalized) t.i. equilibrium states of the infinite--volume
system as the (possibly generalized) minimizers of the free--energy density
functional $f_{\mathfrak{m}}^{\sharp }$ on $E_{1}$. The structure of the set 
\index{States!equilibrium}$\mathit{\Omega }_{\mathfrak{m}}^{\sharp }$ of
generalized t.i. equilibrium states is given in detail by Section \ref%
{Section equilibrium states copy(1)}, whereas in Section \ref{Section Gibbs
versus gen eq states} we discuss the set $\mathit{\Omega }_{\mathfrak{m}%
}^{\sharp }$ w.r.t. weak$^{\ast }$--limit points of Gibbs states. One
important consequence of the detailed analysis of the set $\mathit{\Omega }_{%
\mathfrak{m}}^{\sharp }$ is the fact that the thermodynamics of long--range
models $\mathfrak{m}\in \mathcal{M}_{1}$ is governed by the non--cooperative
equilibria of a zero--sum game called here \emph{thermodynamic game} and
explained in Section \ref{Section thermo game}.

The relative universality of this result -- in the case of models considered
here -- comes from the law of large numbers, whose representative in our
setting is the von Neumann ergodic theorem%
\index{von Neumann ergodic theorem} (cf. Theorem \ref{vonN}). It leads to
approximating models by appropriately replacing operators by a complex
numbers. This procedure is well--known in physics as the so--called
Bogoliubov approximation, see {Section \ref{Concluding remark} for more
details}. In Section \ref{Section effective theories} we analyze this
approximation procedure on the level of generalized t.i. equilibrium states.
This study shows that the set $\mathit{\Omega }_{\mathfrak{m}}^{\sharp }$ of
generalized t.i. equilibrium states for any long--range model $\mathfrak{m}%
\in \mathcal{M}_{1}$ can be analyzed via t.i. equilibrium states of local
interactions. This issue is, however, more involved than it looks like at
first glance and leads us to the definition of \emph{effective theories}.
For more details, we recommend Section \ref{Section effective theories}.

As explained at the beginning, our Banach space $\mathcal{M}_{1}$ includes
important physical models which are long--range in a convenient sense. An
important feature of models whose interactions are elements of $\mathcal{M}%
_{1}$ (see, e.g., \cite{BruPedra1, BruPedraAniko}) is the rather generic
appearance of a so--called \emph{off diagonal long--range order} (ODLRO) for
(generalized) equilibrium states at low enough temperatures, a property
proposed by Yang \cite{ODLRO} to define super--conducting phases. We explain
this behavior in Section \ref{Section ODLRO} and show a surprising (at least
for us) result: As expected, long--range attractions can imply an ODLRO, but
long--range repulsions can also produce a \emph{long--range order} (LRO)%
\index{Long--range order (LRO)}%
\index{Long--range models!repulsions}%
\index{Long--range models!attractions} by breaking the face structure%
\footnote{%
Recall that a face $F$ of a convex set $K$ is defined to be a subset of $K$
with the property that, if $\rho =\lambda _{1}\rho _{1}+\cdots +\lambda
_{n}\rho _{n}\in F$ with $\rho _{1},\ldots ,\rho _{n}\in K$, $\lambda
_{1},\ldots ,\lambda _{n}\in (0,1)$ and $\lambda _{1}+\cdots +\lambda _{n}=1$%
, then $\rho _{1},\ldots ,\rho _{n}\in F$.} of the set $\mathit{\Omega }_{%
\mathfrak{m}}^{\sharp }$, a property absolutely not influenced by
long--range attractions. This feature of long--range repulsions was
previously unknown and its physical implications are completely open to our
knowledge.

Finally, examples of applications are given in Section \ref{example of
application section} and we conclude this chapter with a discussion on the
Bogoliubov approximation and the approximating Hamiltonian method in Section %
\ref{Concluding remark}.

\section{Fermi systems with long--range interactions\label{definition models}%
}

\setcounter{equation}{0}%
Let $(\mathcal{A},\mathfrak{A},\mathfrak{a})$ be a \emph{separable} measure
space%
\index{Separable measure space} with $\mathfrak{A}$ and $\mathfrak{a}:%
\mathfrak{A}\rightarrow \mathbb{R}_{0}^{+}$ being respectively some $\sigma $%
--algebra on $\mathcal{A}$\ and some measure on $\mathfrak{A}$. The
separability of $(\mathcal{A},\mathfrak{A},\mathfrak{a})$ means, by
definition, that the space $L^{2}(\mathcal{A},\mathbb{C}):=L^{2}(\mathcal{A},%
\mathfrak{a},\mathbb{C})$ of square integrable complex valued functions on $%
\mathcal{A}$ is a separable Hilbert space. This property is assumed here
because, by Theorem \ref{Metrizability} together with Banach--Alaoglu
theorem, it yields the metrizability of the weak topology on any
norm--bounded subset $B\subseteq L^{2}(\mathcal{A},\mathbb{C})$, which is a
useful property in the sequel.

Then, as $\mathcal{W}_{1}$ is a Banach space (Definition \ref{definition
banach space interaction}), we can follow the construction done in Section %
\ref{Section Preliminaries} with $\mathcal{X}=\mathcal{W}_{1}$ to define the
space $\mathcal{L}^{2}\left( \mathcal{A},\mathcal{W}_{1}\right) $ of $%
\mathcal{L}^{2}$--interactions which in turn is used to define models with
long interactions as follows:

\begin{definition}[Banach space $\mathcal{M}_{1}$ of long--range models]
\label{definition M1bis}\mbox{ }\newline
\index{Long--range models!Banach space}The set of long--range models is
given by 
\begin{equation*}
\mathcal{M}_{1}:=\mathcal{W}_{1}\times \mathcal{L}^{2}\left( \mathcal{A},%
\mathcal{W}_{1}\right) \times \mathcal{L}^{2}\left( \mathcal{A},\mathcal{W}%
_{1}\right)
\end{equation*}%
and is equipped with the semi--norm 
\begin{equation*}
\Vert \mathfrak{m}\Vert _{\mathcal{M}_{1}}=\Vert \Phi \Vert _{\mathcal{W}%
_{1}}+\Vert \Phi _{a}\Vert _{2}+\Vert \Phi _{a}^{\prime }\Vert _{2}
\end{equation*}%
for\ any $\mathfrak{m}:=(\Phi ,\{\Phi _{a}\}_{a\in \mathcal{A}},\{\Phi
_{a}^{\prime }\}_{a\in \mathcal{A}})\in \mathcal{M}_{1}$. We identify in $%
\mathcal{M}_{1}$ models $\mathfrak{m}_{1}$ and $\mathfrak{m}_{2}$ whenever $%
\Vert \mathfrak{m}_{1}-\mathfrak{m}_{2}\Vert _{\mathcal{M}_{1}}=0$, i.e.,
whenever $\mathfrak{m}_{1}$ and $\mathfrak{m}_{2}$ belong to the same
equivalence class of models. For convenience, we ignore the distinction
between models and their equivalence classes and see $\mathcal{M}_{1}$ as a
Banach space of long--range models with norm $\Vert \cdot \Vert _{\mathcal{M}%
_{1}}$.
\end{definition}

\begin{notation}[Models]
\label{Notation6}\mbox{ }\newline
\index{Long--range models!notation}The symbol $\mathfrak{m}$ is exclusively
reserved to denote elements of $\mathcal{M}_{1}$.
\end{notation}

An important sub--space of $\mathcal{M}_{1}$ is the set $\mathcal{M}_{1}^{%
\mathrm{f}}$ of finite range models defined as follows:%
\begin{equation*}
\mathfrak{m}:=(\Phi ,\{\Phi _{a}\}_{a\in \mathcal{A}},\{\Phi _{a}^{\prime
}\}_{a\in \mathcal{A}})\in \mathcal{M}_{1}
\end{equation*}%
has \emph{finite range}%
\index{Long--range models!finite range} iff $\Phi $ is finite range and $%
\{\Phi _{a}\}_{a\in \mathcal{A}},\{\Phi _{a}^{\prime }\}_{a\in \mathcal{A}}$
are finite range almost everywhere (a.e.). The sub--space $\mathcal{M}_{1}^{%
\mathrm{f}}$ of all finite range models is dense in $\mathcal{M}_{1}$
because of Lebesgue's dominated convergence theorem and the density of set $%
\mathcal{W}_{1}^{\mathrm{f}}$ of all finite range t.i. interactions in $%
\mathcal{W}_{1}$. Another dense\footnote{%
This follows from the density of step functions in $\mathcal{L}^{2}\left( 
\mathcal{A},\mathcal{W}_{1}\right) $.} sub--space of $\mathcal{M}_{1}$ is
given by the set $\mathcal{M}_{1}^{\mathrm{d}}$ of \emph{discrete}%
\index{Long--range models!discrete} elements $\mathfrak{m}$, i.e., elements
for which the set 
\begin{equation*}
\left\{ \Phi _{a}:a\in \mathcal{A}\right\} \cup \left\{ \Phi _{a}^{\prime
}:a\in \mathcal{A}\right\}
\end{equation*}%
has a finite number of interactions. Therefore, the sub--space $\mathcal{M}%
_{1}^{\mathrm{df}}:=\mathcal{M}_{1}^{\mathrm{d}}\cap \mathcal{M}_{1}^{%
\mathrm{f}}$ is also clearly dense in $\mathcal{M}_{1}$. It is an important
dense sub--space used to prove Theorem \ref{BCS main theorem 1} in Chapter %
\ref{section proof of theorem main}.

Like t.i. local interactions $\Phi \in \mathcal{W}_{1}$ (cf. Definition \ref%
{definition standard interaction}), any long--range model $\mathfrak{m}\in 
\mathcal{M}_{1}$ is associated with a family of internal energies as follows:

\begin{definition}[Internal energy with long--range interactions]
\label{definition BCS-type model}\mbox{ }\newline
\index{Long--range models!internal energy}For any $\mathfrak{m}\in \mathcal{M%
}_{1}$ and $l\in \mathbb{N}$, its internal energy in the box $\Lambda _{l}$
is defined by%
\begin{equation*}
U_{l}:=U_{\Lambda _{l}}^{\Phi }+%
\frac{1}{|\Lambda _{l}|}\int_{\mathcal{A}}\gamma _{a}(U_{\Lambda _{l}}^{\Phi
_{a}}+iU_{\Lambda _{l}}^{\Phi _{a}^{\prime }})^{\ast }(U_{\Lambda
_{l}}^{\Phi _{a}}+iU_{\Lambda _{l}}^{\Phi _{a}^{\prime }})\mathrm{d}%
\mathfrak{a}\left( a\right) ,
\end{equation*}%
with $\gamma _{a}\in \{-1,1\}$ being a fixed measurable function.
\end{definition}

\noindent The internal energy $U_{l}$ is well--defined. Indeed, by
continuity of the linear map $\Phi \mapsto U_{\Lambda _{l}}^{\Phi }$, for
any $\mathfrak{m}\in \mathcal{M}_{1}$, the map $a\mapsto \gamma
_{a}U_{\Lambda _{l}}^{\Phi _{a}}$ from $\mathcal{A}$ to $\mathcal{U}%
_{\Lambda _{l}}$ belongs to $\mathcal{L}^{2}\left( \mathcal{A},\mathcal{U}%
_{\Lambda _{l}}\right) $ (see Section \ref{Section Preliminaries} for the
definition of the space $\mathcal{L}^{p}\left( \mathcal{A},\mathcal{U}%
_{\Lambda _{l}}\right) $). Then, as $\mathcal{U}$ is a $C^{\ast }$--algebra,
the map 
\begin{equation*}
a\mapsto \gamma _{a}(U_{\Lambda _{l}}^{\Phi _{a}}+iU_{\Lambda _{l}}^{\Phi
_{a}^{\prime }})^{\ast }(U_{\Lambda _{l}}^{\Phi _{a}}+iU_{\Lambda
_{l}}^{\Phi _{a}^{\prime }})
\end{equation*}%
belongs to the space $\mathcal{L}^{1}\left( \mathcal{A},\mathcal{U}_{\Lambda
_{l}}\right) $ and $\mathfrak{m}\mapsto U_{l}$ is a well--defined functional
from the Banach space $\mathcal{M}_{1}$ to the $C^{\ast }$--algebra $%
\mathcal{U}$. By (\ref{bound BCS norms}), this map is even continuous w.r.t.
the norms of $\mathcal{M}_{1}$ and $\mathcal{U}$.

The long--range character of Fermi models $\mathfrak{m}\in \mathcal{M}_{1}$
with local internal energy $U_{l}$ -- as compared to the usual models
defined from local interactions $\Phi \in \mathcal{W}_{1}$ only -- can be
seen as follows. For each fixed $\epsilon \in (0,1)$, we define the
long--range truncation of the internal energy $U_{\Lambda _{l}}^{\Phi }$
(Definition \ref{definition standard interaction}) associated with the local
part $\Phi $ of $\mathfrak{m}$ by%
\begin{equation*}
U_{l,\epsilon }^{\Phi }:=\sum\limits_{\Lambda \in \mathcal{P}_{f}(\Lambda
_{l}),\;\text{{\o }}(\Lambda )>\epsilon l}\Phi _{\Lambda },
\end{equation*}%
where the function ${\o }(\Lambda )$ is the diameter of $\Lambda \in 
\mathcal{P}_{f}(\mathfrak{L})$, see (\ref{diameter of lambda}). Analogously,
the long--range truncation of the internal energy $(U_{l}-U_{\Lambda
_{l}}^{\Phi })$ associated with the long--range part of $\mathfrak{m}$ is by
definition equal to%
\index{Long--range models!internal energy!long--range}%
\begin{equation*}
U_{l,\epsilon }:=%
\frac{1}{|\Lambda _{l}|}\sum\limits_{\Lambda ,\Lambda ^{\prime }\in \mathcal{%
P}_{f}(\Lambda _{l}),\text{{\ \o }}(\Lambda \cup \Lambda ^{\prime
})>\epsilon l}\int_{\mathcal{A}}\gamma _{a}(\Phi _{a,\Lambda }+i\Phi
_{a,\Lambda }^{\prime })^{\ast }(\Phi _{a,\Lambda ^{\prime }}+i\Phi
_{a,\Lambda ^{\prime }}^{\prime })\mathrm{d}\mathfrak{a}\left( a\right) .
\end{equation*}%
Then, because $\Phi \in \mathcal{W}_{1}$, one can generally check for any $%
\epsilon \in (0,1)$ that%
\begin{equation*}
\lim\limits_{l\rightarrow \infty }\frac{\Vert U_{l,\epsilon }^{\Phi }\Vert }{%
\Vert U_{l,\epsilon }\Vert }=0
\end{equation*}%
provided that $\mathfrak{m}\neq (\Phi ,0,0)$. In other words, the
long--range part $(U_{l}-U_{\Lambda _{l}}^{\Phi })$ of the internal energy $%
U_{l}$ generally dominates the interaction at long distances for large $l\in 
\mathbb{N}$.

The aim of the monograph is the study of the thermodynamic behavior of any
models $\mathfrak{m}\in \mathcal{M}_{1}$ with long--range interactions. In
the thermodynamic limit, long--range interactions act completely differently
depending whether they are positive long--range interactions, i.e., \emph{%
long--range repulsions}, or negative long--range interactions, i.e., \emph{%
long--range attractions}. These two types of long--range interactions are
defined via the negative and positive parts%
\begin{equation}
\gamma _{a,\pm }:=1/2(|\gamma _{a}|\pm \gamma _{a})\in \{0,1\}
\label{remark positive negative part gamma}
\end{equation}%
of the fixed measurable function 
\begin{equation*}
\gamma _{a}=\gamma _{a,+}-\gamma _{a,-}\in \{-1,1\}
\end{equation*}%
as follows:

\begin{definition}[Long--range attractions and repulsions]
\label{long range attraction-repulsion}\mbox{ }\newline
\emph{(}$-$\emph{)} 
\index{Long--range models!attractions}%
\index{Interaction!long--range|textbf}The long--range attractions of any $%
\mathfrak{m}\in \mathcal{M}_{1}$ are the $\mathcal{L}^{2}$--interactions%
\begin{equation*}
\{\Phi _{a,-}:=\gamma _{a,-}\Phi _{a}\}_{a\in \mathcal{A}}\in \mathcal{L}%
^{2}\left( \mathcal{A},\mathcal{W}_{1}\right) \quad 
\text{and}\quad \{\Phi _{a,-}^{\prime }:=\gamma _{a,-}\Phi _{a}^{\prime
}\}_{a\in \mathcal{A}}\in \mathcal{L}^{2}\left( \mathcal{A},\mathcal{W}%
_{1}\right) .
\end{equation*}%
\emph{(}$+$\emph{)} 
\index{Long--range models!repulsions}The long--range repulsions of any $%
\mathfrak{m}\in \mathcal{M}_{1}$ are the $\mathcal{L}^{2}$--interactions%
\begin{equation*}
\{\Phi _{a,+}:=\gamma _{a,+}\Phi _{a}\}_{a\in \mathcal{A}}\in \mathcal{L}%
^{2}\left( \mathcal{A},\mathcal{W}_{1}\right) \quad 
\text{and}\quad \{\Phi _{a,+}^{\prime }:=\gamma _{a,+}\Phi _{a}^{\prime
}\}_{a\in \mathcal{A}}\in \mathcal{L}^{2}\left( \mathcal{A},\mathcal{W}%
_{1}\right) .
\end{equation*}
\end{definition}

It is important to observe that our class of models $\mathfrak{m}\in 
\mathcal{M}_{1}$ includes Fermi systems%
\index{Long--range models} 
\begin{equation*}
(\Phi ,\{\Phi _{a}^{1}\}_{a\in \mathcal{A}},\{\Phi _{a}^{2}\}_{a\in \mathcal{%
A}},\{\Phi _{a}^{3}\}_{a\in \mathcal{A}},\{\Phi _{a}^{4}\}_{a\in \mathcal{A}%
})\in \mathcal{M}_{1}\times \mathcal{L}^{2}\left( \mathcal{A},\mathcal{W}%
_{1}\right) \times \mathcal{L}^{2}\left( \mathcal{A},\mathcal{W}_{1}\right)
\end{equation*}%
with internal energies of the type%
\index{Long--range models!internal energy}%
\begin{equation*}
V_{l}:=U_{\Lambda _{l}}^{\Phi }+%
\frac{1}{|\Lambda _{l}|}\int_{\mathcal{A}}(U_{\Lambda _{l}}^{\Phi
_{a}^{1}}+iU_{\Lambda _{l}}^{\Phi _{a}^{2}})^{\ast }(U_{\Lambda _{l}}^{\Phi
_{a}^{3}}+iU_{\Lambda _{l}}^{\Phi _{a}^{4}})\mathrm{d}\mathfrak{a}\left(
a\right) +\mathrm{h.c.}
\end{equation*}%
because 
\begin{equation}
2\left( A^{\ast }B+B^{\ast }A\right) =\left( A+B\right) ^{\ast }\left(
A+B\right) -\left( A-B\right) ^{\ast }\left( A-B\right) .
\label{observing truc}
\end{equation}%
In other words, such Fermi systems correspond to models $\mathfrak{m}\in 
\mathcal{M}_{1}$ with long--range attractions and repulsions together.

\section{Examples of Applications\label{example of application section}}

Long range models are defined in a rather abstract way within Section \ref%
{definition models}. Therefore, before going further, we give here some
concrete examples of long range models used in theoretical physics as well
as a possible generalization in Section \ref{Inhomogeneous sect}. We also
express the main consequences of our results, which will be formulated in
the general case later on in Sections \ref{existence of thermodynamics}--\ref%
{Concluding remark}.

The most general form of a translation invariant model for fermions in a
cubic box $\Lambda _{l}\subseteq \mathfrak{L}:=\mathbb{Z}^{d}$ with a
quartic (in the creation and annihilation operators) gauge invariant t.i.
interaction and spin set $\mathrm{S}$ is formally equal to 
\begin{eqnarray*}
H:= &&\underset{x,y\in \Lambda _{l},\ \mathrm{s}\in \mathrm{S}}{\sum }%
h\left( x-y\right) a_{x,\mathrm{s}}^{\ast }a_{y,\mathrm{s}} \\
&&+\underset{\mathrm{s}_{1},\mathrm{s}_{2},\mathrm{s}_{3},\mathrm{s}_{4}\in 
\mathrm{S}}{\underset{x,y,z,w\in \Lambda _{l}\ }{\sum }}v_{\mathrm{s}_{1},%
\mathrm{s}_{2},\mathrm{s}_{3},\mathrm{s}_{4}}\left( y-x,z-x,w-x\right) a_{x,%
\mathrm{s}_{1}}^{\ast }a_{y,\mathrm{s}_{2}}^{\ast }a_{z,\mathrm{s}_{3}}a_{w,%
\mathrm{s}_{4}}\ .
\end{eqnarray*}%
As an example, the spin set equals $\mathrm{S}=\left\{ \uparrow ,\downarrow
\right\} $ for electrons. In momentum space, the above Hamiltonian reads 
\begin{eqnarray}
H &=&\underset{k\in \Lambda _{l}^{\ast },\ \mathrm{s}\in \mathrm{S}}{\sum }%
\hat{h}\left( k\right) \hat{a}_{k}^{\ast }\hat{a}_{k}  \notag \\
&&+\frac{1}{\left\vert \Lambda _{l}\right\vert }\underset{\mathrm{s}_{1},%
\mathrm{s}_{2},\mathrm{s}_{3},\mathrm{s}_{4}\in \mathrm{S}}{\underset{%
k,k^{\prime },q\in \Lambda _{l}^{\ast }}{\sum }}\hat{g}_{\mathrm{s}_{1},%
\mathrm{s}_{2},\mathrm{s}_{3},\mathrm{s}_{4}}\left( k,k^{\prime },q\right) 
\hat{a}_{k+q,\mathrm{s}_{1}}^{\ast }\hat{a}_{k^{\prime }-q,\mathrm{s}%
_{2}}^{\ast }\hat{a}_{k^{\prime },\mathrm{s}_{3}}\hat{a}_{k,\mathrm{s}_{4}}\
,  \label{hamil general}
\end{eqnarray}%
see \cite[Eq. (2.1)]{Metzner}. Here, 
\begin{equation*}
\Lambda _{l}^{\ast }:=\frac{2\pi }{(2l+1)}\Lambda _{l}\subseteq \left[ -\pi
,\pi \right] ^{d}
\end{equation*}%
is the reciprocal lattice of quasi--momenta (periodic boundary conditions)
and the operator(s) 
\begin{equation*}
\hat{a}_{k,\mathrm{s}}^{\ast }:=\frac{1}{\left\vert \Lambda _{l}\right\vert
^{1/2}}\underset{x\in \Lambda _{l}}{\sum }\mathrm{e}^{-ik\cdot x}a_{x,%
\mathrm{s}}^{\ast }\ ,\qquad \hat{a}_{k,\mathrm{s}}:=\frac{1}{\left\vert
\Lambda _{l}\right\vert ^{1/2}}\underset{x\in \Lambda _{l}}{\sum }\mathrm{e}%
^{ik\cdot x}a_{x,\mathrm{s}}\ ,
\end{equation*}%
creates (resp. annihilates) a fermion with spin $\mathrm{s}\in \mathrm{S}$
and (quasi--) momentum $k\in \Lambda _{l}^{\ast }$. In the interaction part
of (\ref{hamil general}), $k$ and $k^{\prime }$ are physically interpreted
as being the momenta of two incoming particles which interact and exchange a
(quasi--) momentum $q$.

The thermodynamics of the model $H$ is highly non--trivial, in general. In
theoretical physics, one is forced to perform different kinds of
approximations or Ans\"{a}tze to extract physical properties. Many of them
lead to long--range models in the sense of Definition \ref{definition M1bis}
and our method provides rigorous results on these. As a first example, we
start with the so--called \emph{forward scattering approximation}.

\subsection{The forward scattering approximation}

\index{Forward scattering approximation}In many physical situations, forward
processes, i.e., interactions with a very small momentum exchange $q$, are
dominating, see, e.g., \cite[Section 5]{Metzner}. They are for instance
relevant for the physics of high--T$_{c}$ superconductors, see, e.g., \cite%
{Yamase}.

This case is modeled by considering a coupling function $%
\hat{g}_{\mathrm{s}_{1},\mathrm{s}_{2},\mathrm{s}_{3},\mathrm{s}_{4}}\left(
k,k^{\prime },q\right) $ concentrated around $q=0$. As a consequence, one
can consider the Hamiltonian 
\begin{equation*}
H^{F}:=\underset{k\in \Lambda _{l}^{\ast },\ \mathrm{s}\in \mathrm{S}}{\sum }%
\hat{h}\left( k\right) \hat{a}_{k,\mathrm{s}}^{\ast }\hat{a}_{k,\mathrm{s}}+%
\frac{1}{\left\vert \Lambda _{l}\right\vert }\underset{\mathrm{s}_{1},%
\mathrm{s}_{2},\mathrm{s}_{3},\mathrm{s}_{4}\in \mathrm{S}}{\underset{%
k,k^{\prime }\in \Lambda _{l}^{\ast }}{\sum }}\hat{g}_{\mathrm{s}_{1},%
\mathrm{s}_{2},\mathrm{s}_{3},\mathrm{s}_{4}}^{F}\left( k,k^{\prime }\right) 
\hat{a}_{k,\mathrm{s}_{1}}^{\ast }\hat{a}_{k^{\prime },\mathrm{s}_{2}}^{\ast
}\hat{a}_{k^{\prime },\mathrm{s}_{3}}\hat{a}_{k,\mathrm{s}_{4}}
\end{equation*}%
with 
\begin{equation*}
\hat{g}_{\mathrm{s}_{1},\mathrm{s}_{2},\mathrm{s}_{3},\mathrm{s}%
_{4}}^{F}\left( k,k^{\prime }\right) :=\int_{\left[ -\pi ,\pi \right] ^{d}}%
\hat{g}_{\mathrm{s}_{1},\mathrm{s}_{2},\mathrm{s}_{3},\mathrm{s}_{4}}\left(
k,k^{\prime },q\right) \mathrm{d}^{d}q\ .
\end{equation*}%
For instance, this form exactly corresponds to the interaction term in \cite[%
Eq. (3)]{Yamase}. We assume now that $\hat{g}_{\mathrm{s}_{1},\mathrm{s}_{2},%
\mathrm{s}_{3},\mathrm{s}_{4}}^{F}$ is a real--valued continuous and
symmetric function represented in the form 
\begin{equation}
\hat{g}_{\mathrm{s}_{1},\mathrm{s}_{2},\mathrm{s}_{3},\mathrm{s}%
_{4}}^{F}\left( k,k^{\prime }\right) =\delta _{\mathrm{s}_{1},\mathrm{s}%
_{4}}\delta _{\mathrm{s}_{2},\mathrm{s}_{3}}\left( \int_{\mathbb{R}^{+}}\hat{%
f}_{a,+}\left( k\right) \hat{f}_{a,+}\left( k^{\prime }\right) \mathrm{d}%
a-\int_{\mathbb{R}^{-}}\hat{f}_{a,-}\left( k\right) \hat{f}_{a,-}\left(
k^{\prime }\right) \mathrm{d}a\right) .  \label{product assumptionbis}
\end{equation}%
Here, $\hat{f}_{a,\pm }\left( k\right) $ is a real--valued continuous
function and for each $k\in \left[ -\pi ,\pi \right] ^{D}$, the two
functions 
\begin{equation*}
a\mapsto \hat{f}_{a,\pm }\left( k\right)
\end{equation*}%
belong to $L^{2}\left( \mathbb{R}^{\pm }\right) $. The above explicit
dependency of $\hat{g}_{\mathrm{s}_{1},\mathrm{s}_{2},\mathrm{s}_{3},\mathrm{%
s}_{4}}^{F}$ w.r.t. $\mathrm{s}_{1},\mathrm{s}_{2},\mathrm{s}_{3},\mathrm{s}%
_{4}\in \mathrm{S}$ is only chosen for simplicity. A general spin dependency
can also be treated by observing (\ref{observing truc}). Note also that
continuous and symmetric functions $p(k,k^{\prime })$ on $\left[ -\pi ,\pi %
\right] ^{D}\times \left[ -\pi ,\pi \right] ^{D}$ can be arbitrarily well
approximated by sums of products of the form $t\left( k\right) t\left(
k^{\prime }\right) $. Observe that the choice in \cite[Eq. (4)]{Yamase} is a
special case of (\ref{product assumptionbis}).

With this choice of functions, we write the Hamiltonian $H^{F}$ back in the $%
x$--space and get%
\begin{multline*}
H^{F}:=\underset{x,y\in \Lambda _{l},\ \mathrm{s}\in \mathrm{S}}{\sum }%
h^{\prime }\left( x-y\right) a_{x,\mathrm{s}}^{\ast }a_{y,\mathrm{s}} \\
+\frac{1}{\left\vert \Lambda _{l}\right\vert }\int_{\mathbb{R}^{+}}\left( 
\underset{x,y\in \Lambda _{l},\ \mathrm{s}\in \mathrm{S}}{\sum }%
f_{a,+}\left( x-y\right) a_{x,\mathrm{s}}^{\ast }a_{y,\mathrm{s}}\right)
\left( \underset{x,y\in \Lambda _{l},\ \mathrm{s}\in \mathrm{S}}{\sum }%
f_{a,+}\left( x-y\right) a_{x,\mathrm{s}}^{\ast }a_{y,\mathrm{s}}\right) 
\mathrm{d}a \\
-\frac{1}{\left\vert \Lambda _{l}\right\vert }\int_{\mathbb{R}^{-}}\left( 
\underset{x,y\in \Lambda _{l},\ \mathrm{s}\in \mathrm{S}}{\sum }%
f_{a,-}\left( x-y\right) a_{x,\mathrm{s}}^{\ast }a_{y,\mathrm{s}}\right)
\left( \underset{x,y\in \Lambda _{l},\ \mathrm{s}\in \mathrm{S}}{\sum }%
f_{a,-}\left( x-y\right) a_{x,\mathrm{s}}^{\ast }a_{y,\mathrm{s}}\right) 
\mathrm{d}a\ .
\end{multline*}%
The hopping term $h$ is replaced by $h^{\prime }$ because of commutators
used to rearrange the quartic terms. By self-adjointness of $H^{F}$, $%
h^{\prime }$ must be a symmetric function. This Hamiltonian corresponds to a
model 
\begin{equation*}
\mathfrak{m}^{F}:=(\Phi ^{F},\{\Phi _{a}^{F}\}_{a\in \mathcal{A}},\{\Phi
_{a}^{F\prime }\}_{a\in \mathcal{A}})\ ,
\end{equation*}%
by setting $\mathcal{A}=\mathbb{R}$, $\gamma _{a}=\mathbf{1}\left[ a\in 
\mathbb{R}^{+}\right] -\mathbf{1}\left[ a\in \mathbb{R}^{-}\right] $, $%
\mathrm{d}\mathfrak{a}\left( a\right) =\mathrm{d}a$. More precisely, $%
\mathfrak{m}^{F}$ is defined as follows: For all finite subsets $\Lambda
\subseteq \mathfrak{L}$ (i.e., $\Lambda \in \mathcal{P}_{f}(\mathfrak{L})$), 
\begin{eqnarray*}
\Phi _{\Lambda }^{F}:= &&\frac{1}{1+\delta _{x,y}}\underset{\mathrm{s}\in 
\mathrm{S}}{\sum }\left( h^{\prime }\left( x-y\right) a_{x,\mathrm{s}}^{\ast
}a_{y,\mathrm{s}}+h^{\prime }\left( y-x\right) a_{y,\mathrm{s}}^{\ast }a_{x,%
\mathrm{s}}\right) \\
\Phi _{a,\Lambda }^{F}:= &&\frac{1}{1+\delta _{x,y}}\underset{\mathrm{s}\in 
\mathrm{S}}{\sum }\func{Re}\left( f_{a}\left( x-y\right) a_{x,\mathrm{s}%
}^{\ast }a_{y,\mathrm{s}}+f_{a}\left( y-x\right) a_{y,\mathrm{s}}^{\ast
}a_{x,\mathrm{s}}\right) \\
\Phi _{a,\Lambda }^{F\prime }:= &&\frac{1}{1+\delta _{x,y}}\underset{\mathrm{%
s}\in \mathrm{S}}{\sum }\func{Im}\left( f_{a}\left( x-y\right) a_{x,\mathrm{s%
}}^{\ast }a_{y,\mathrm{s}}+f_{a}\left( y-x\right) a_{y,\mathrm{s}}^{\ast
}a_{x,\mathrm{s}}\right)
\end{eqnarray*}%
whenever $\Lambda =\left\{ x,y\right\} $, and $\Phi _{\Lambda }^{F}=\Phi
_{a,\Lambda }^{F}=\Phi _{a,\Lambda }^{F\prime }=0$ otherwise. Here, 
\begin{equation*}
f_{a}:=f_{a,-}+f_{a,+},\qquad a\in \mathbb{R}\ .
\end{equation*}%
To ensure that $\mathfrak{m}^{F}\in \mathcal{M}_{1}$\ we impose at this
point that 
\begin{equation*}
\left\Vert \Phi ^{F}\right\Vert \leq \left\vert \mathrm{S}\right\vert 
\underset{x\in \mathfrak{L}}{\sum }\left\vert h^{\prime }\left( x\right)
\right\vert <\infty
\end{equation*}%
and 
\begin{equation*}
\left\Vert \Phi _{a}^{F}\right\Vert _{2}^{2}+\left\Vert \Phi _{a}^{F\prime
}\right\Vert _{2}^{2}\leq \left\vert \mathrm{S}\right\vert ^{2}\int_{\mathbb{%
R}^{+}}\left( \underset{x\in \mathfrak{L}}{\sum }\left\vert f_{a}\left(
x\right) \right\vert \right) ^{2}\mathrm{d}a<\infty \ .
\end{equation*}

Therefore, we infer from Theorem \ref{theorem saddle point} ($\sharp $) that
the infinite--volume pressure $\mathrm{P}_{\mathfrak{m}^{F}}^{\sharp }$
equals $-\mathrm{F}_{\mathfrak{m}^{F}}^{\sharp }$, where%
\begin{eqnarray*}
\mathrm{F}_{\mathfrak{m}^{F}}^{\sharp } &=&\underset{c_{a,-}\in L^{2}(%
\mathbb{R}^{-})}{\inf }\ \underset{c_{a,+}\in L^{2}(\mathbb{R}^{+})}{\sup }%
\left\{ -\int_{\mathbb{R}^{+}}\left\vert c_{a,+}\right\vert ^{2}\mathrm{d}%
a\right. \\
&&\left. +\int_{\mathbb{R}^{-}}\left\vert c_{a,-}\right\vert ^{2}\mathrm{d}%
a-P_{\mathfrak{m}^{F}}\left( c_{a,-}+c_{a,+}\right) \right\} .
\end{eqnarray*}%
Here, $P_{\mathfrak{m}^{F}}\left( c_{a}\right) $ is the (explicit) pressure
of a \emph{free Fermi gas} with hopping matrix%
\begin{equation*}
\mathrm{h}_{c}\left( x-y\right) :=h^{\prime }\left( x-y\right) +\frac{1}{2}%
\int_{\mathbb{R}}\gamma _{a}\left( c_{a}f_{a}\left( x-y\right) +\bar{c}_{a}%
\bar{f}_{a}\left( y-x\right) \right) \mathrm{d}a
\end{equation*}%
for any $c_{a}\in L^{2}(\mathbb{R})$. From Theorem \ref{theorem structure of
omega} (ii), the generalized equilibrium states are convex combinations of $%
U\left( 1\right) $--invariant quasi--free states and thus, none of them can
break this gauge symmetry and even show superconducting ODLRO. By Theorem %
\ref{lemma limit gibbs states periodic}, observe that weak$^{\ast }$%
--accumulation points of Gibbs states associated with $H^{F}$ and periodic
boundary conditions are particular cases of generalized equilibrium states
of $\mathfrak{m}^{F}$. We have in particular access to \emph{all correlation
functions} of this model in the thermodynamic limit.

Previous results in theoretical physics on the forward scattering
interaction are based on diagrammatic methods \cite{Metzner}, bosonization 
\cite{Haldane,Kopietz} and others. To our knowledge, there is no rigorous
result on the level of the pressure, even for the Hamiltonian $H^{F}$.
Observe that the rigorous methods of \cite%
{RaggioWerner1,RaggioWerner2,BCSrigorous1} may work for this model, but
would yield a more complicated variational problem for the pressure.
Moreover, these technics do not solve the problem of (generalized)
equilibrium states and thermodynamic limit of Gibbs states.

\subsection{The BCS approximation}

\index{BCS model}A second, but more \textquotedblleft
classical\textquotedblright\ application is the celebrated \emph{BCS model} 
\cite{BCS1,BCS2,BCS3}. Indeed, it is defined by (\ref{hamil general}) with $%
k^{\prime }=-k$ and $\mathrm{S}=\left\{ \uparrow ,\downarrow \right\} $,
that is, 
\begin{equation*}
H^{BCS}:=\underset{k\in \Lambda _{l}^{\ast },\ \mathrm{s}\in \mathrm{S}}{%
\sum }%
\hat{h}\left( k\right) a_{k,\mathrm{s}}^{\ast }a_{k,\mathrm{s}}+\frac{1}{%
\left\vert \Lambda _{l}\right\vert }\underset{k,p\in \Lambda _{l}^{\ast }}{%
\sum }\hat{g}^{BCS}\left( k,p\right) a_{p,\uparrow }^{\ast }a_{-p,\downarrow
}^{\ast }a_{k,\downarrow }a_{-k,\uparrow }\ ,
\end{equation*}%
setting $p:=k+q$. Using similar assumptions as before, this Hamiltonian
corresponds again to a model $\mathfrak{m}^{BCS}\in \mathcal{M}_{1}$, where%
\begin{eqnarray*}
\Phi _{\Lambda }^{BCS} &:=&\frac{1}{1+\delta _{x,y}}\underset{\mathrm{s}\in 
\mathrm{S}}{\sum }\left( h\left( x-y\right) a_{x,\mathrm{s}}^{\ast }a_{y,%
\mathrm{s}}+h\left( y-x\right) a_{y,\mathrm{s}}^{\ast }a_{x,\mathrm{s}%
}\right) \\
\Phi _{a,\Lambda }^{BCS} &:=&\func{Re}\left( \left( \tilde{f}_{a}\left(
x-y\right) -\tilde{f}_{a}\left( y-x\right) \right) a_{x,\downarrow
}a_{y,\uparrow }\right) \\
\Phi _{a,\Lambda }^{BCS\prime } &:=&\func{Im}\left( \left( \tilde{f}%
_{a}\left( x-y\right) -\tilde{f}_{a}\left( y-x\right) \right)
a_{x,\downarrow }a_{y,\uparrow }\right)
\end{eqnarray*}%
whenever $\Lambda =\left\{ x,y\right\} $, and $\Phi _{\Lambda }^{BCS}=\Phi
_{a,\Lambda }^{BCS}=\Phi _{a,\Lambda }^{BCS\prime }=0$ otherwise.

The approximating interactions are quadratic in the annihilation and
creation operators. Therefore, the pressure $P_{\mathfrak{m}^{F}}\left(
c_{a}\right) $ can explicitly be computed for any $c_{a}\in L^{2}(\mathbb{R}%
) $. By Theorem \ref{theorem saddle point} ($\sharp $), we get the
infinite--volume pressure $\mathrm{P}_{\mathfrak{m}^{BCS}}^{\sharp }$ via a
variational problem $\mathrm{F}_{\mathfrak{m}^{BCS}}^{\sharp }$. Note that
the rigorous analysis of the infinite--volume pressure was already
rigorously performed in this special case in the eighties \cite%
{RaggioWerner1,RaggioWerner2,BCSrigorous1}, but the resulting variational
problem is technically more difficult to study than $\mathrm{F}_{\mathfrak{m}%
^{BCS}}^{\sharp }$, in general.

Moreover, in contrast to \cite{RaggioWerner1,RaggioWerner2,BCSrigorous1}, by
Theorem \ref{theorem structure of omega} (ii) we also obtain all generalized
equilibrium states which, by Theorem \ref{lemma limit gibbs states periodic}%
, give access to \emph{all correlation functions} of this model with
periodic boundary conditions in the thermodynamic limit. We can in
particular rigorously verify the existence of ODLRO for such models.

\subsection{The forward scattering--BCS approximation}

\index{BCS model}%
\index{Forward scattering approximation}Note that we can also combine the
BCS and the forward scattering interactions to study the competition between
the Cooper and forward scattering channels. This is exactly what is done for
a special case of coupling functions in \cite{Yamase}. Indeed, the resulting
model $\mathfrak{m}^{F-BCS}$ still belongs to $\mathcal{M}_{1}$ and the
associated approximating interaction is again quadratic in the annihilation
and creation operators. Hence, $P_{\mathfrak{m}^{F-BCS}}\left( c_{a}\right) $
can explicitly be computed for any $c_{a}\in L^{2}(\mathbb{R})$ and we can
have access to all correlation functions as above. In particular, we can
rigorously justify the approach of \cite{Yamase} (mean--field approximation,
gap equations, etc) even on the level of states. Note that the resulting
variational problem $\mathrm{F}_{\mathfrak{m}^{F-BCS}}^{\sharp }$ can then
be treated in a rigorous way by numerical methods, see, e.g., \cite{Yamase}.

\subsection{Inhomogeneous Hubbard--type interactions\label{Inhomogeneous
sect}}

\index{Hubbard--type interactions}To conclude, our results can directly be
extended to more general situations where the range of the two--particle
interaction is macroscopic, but very small as compared to the side--length $%
(2l+1)$ of the cubic box $\Lambda _{l}$. A prototype of such models is given
by a Hamiltonian of Hubbard--type 
\begin{equation*}
H^{HT}:=\underset{x,y\in \Lambda _{l},\ \mathrm{s}\in \mathrm{S}}{\sum }%
h\left( x-y\right) a_{x,\mathrm{s}}^{\ast }a_{y,\mathrm{s}}+%
\frac{1}{\left\vert \Lambda _{l}\right\vert }\underset{x,y\in \Lambda _{l}}{%
\sum }v\left( \frac{x}{2l+1},\frac{y}{2l+1}\right) n_{x}n_{y}
\end{equation*}%
for any symmetric continuous function 
\begin{equation*}
v:[-1/2,1/2]^{d}\times \lbrack -1/2,1/2]^{d}\rightarrow \mathbb{R}\ .
\end{equation*}%
Here,%
\begin{equation*}
n_{x}:=\underset{\mathrm{s}\in \mathrm{S}}{\sum }a_{x,\mathrm{s}}^{\ast
}a_{x,\mathrm{s}}
\end{equation*}%
is the density operator at lattice site $x\in \Lambda _{l}$. The particular
case we have in mind would be 
\begin{equation*}
v\left( x,y\right) =\kappa \left( \left\vert x-y\right\vert \right) ,\qquad
x,y\in \lbrack -1/2,1/2]^{d}\ ,
\end{equation*}
for some continuous function $\kappa :\mathbb{R}\rightarrow \mathbb{R}$
concentrated around $0$, but the result is more general. Note additionally
that \emph{neither} the positivity (or negativity) of $v$ \emph{nor} the one
of its Fourier transform is required.

Choosing 
\begin{equation*}
v\left( x,y\right) =\int_{\mathbb{R}^{+}}f_{a,+}\left( x\right)
f_{a,+}\left( y\right) \mathrm{d}a-\int_{\mathbb{R}^{-}}f_{a,-}\left(
x\right) f_{a,-}\left( y\right) \mathrm{d}a
\end{equation*}%
we arrive at the infinite--volume pressure 
\begin{eqnarray*}
\mathrm{P}^{HT} &=&-\underset{c_{a,-}\in L^{2}(\mathbb{R}^{-})}{\inf }\ 
\underset{c_{a,+}\in L^{2}(\mathbb{R}^{+})}{\sup }\left\{ -\int_{\mathbb{R}%
^{+}}\left\vert c_{a,+}\right\vert ^{2}\mathrm{d}a+\int_{\mathbb{R}%
^{-}}\left\vert c_{a,-}\right\vert ^{2}\mathrm{d}a\right. \\
&&\left. -\int_{[-1/2,1/2]^{d}}\mathfrak{p}\left( \zeta
,c_{a,-}+c_{a,+}\right) \mathrm{d}^{d}\zeta \right\} .
\end{eqnarray*}%
with $\mathfrak{p}\left( \zeta ,c_{a}\right) $ being the thermodynamic limit
of the pressure of the \emph{free Fermi gas} with Hamiltonian 
\begin{equation*}
\underset{x,y\in \Lambda _{l},\ \mathrm{s}\in \mathrm{S}}{\sum }h\left(
x-y\right) a_{x,\mathrm{s}}^{\ast }a_{y,\mathrm{s}}+\underset{x\in \Lambda
_{l}}{\sum }n_{x}\int_{\mathbb{R}}\func{Re}\left( \bar{c}_{a}f_{a}\left(
\zeta \right) \right) \mathrm{d}a\ 
\end{equation*}%
for any $c_{a}\in L^{2}(\mathbb{R})$. Because of the absence of space
symmetries in the above model, it is not clear what kind of object \emph{%
generalized equilibrium states} should be, see below Definition \ref%
{definition equilibirum state}. Though, it is possible to study all
correlations functions of the form 
\begin{equation*}
\tilde{\rho}_{l}\left( \alpha _{\left[ (2l+1)x_{1}\right] }\left(
A_{1}\right) \cdots \alpha _{\left[ (2l+1)x_{p}\right] }\left( A_{p}\right)
\right) \ 
\end{equation*}%
for any $A_{1},\ldots ,A_{p}\in \mathcal{U}^{0}$, $x_{1},\ldots ,x_{p}\in
(-1/2,1/2)^{d}$ with $p\in \mathbb{N}$. Here, $\tilde{\rho}_{l}$ is the
Gibbs state associated with $H^{HT}$. These last results are the subject\ of
papers in preparation \cite{BruPedra-homog,BruPedra-homogbis}.

\section{Free--energy densities and existence of thermodynamics\label%
{existence of thermodynamics}}

We now come back to the general situation of Section \ref{definition models}%
. As in the case of local interactions (see, e.g., \cite[Theorem 11.4.]%
{Araki-Moriya}), the analysis of the thermodynamics of long--range Fermi
systems in the grand--canonical ensemble is related to an important
functional associated with any $\vec{\ell}$--periodic state $\rho \in E_{%
\vec{\ell}}$ on $\mathcal{U}$: the \emph{free--energy density functional} $%
f_{\mathfrak{m}}^{\sharp }$ of the long--range model 
\begin{equation*}
\mathfrak{m}:=(\Phi ,\{\Phi _{a}\}_{a\in \mathcal{A}},\{\Phi _{a}^{\prime
}\}_{a\in \mathcal{A}})\in \mathcal{M}_{1}.
\end{equation*}%
This functional is the sum of the local free--energy density functional $%
f_{\Phi }$ (Definition \ref{Remark free energy density}) and the long--range
energy densities defined from the space--averaging functional $\Delta _{A}$
(Definition \ref{definition de deltabis}) for $A=\mathfrak{e}_{\Phi _{a}}+i%
\mathfrak{e}_{\Phi _{a}^{\prime }}$, see (\ref{eq:enpersite}). Using that%
\index{Space--averaging functional}%
\begin{equation}
\Delta _{a,\pm }\left( \rho \right) :=\gamma _{a,\pm }\Delta _{\mathfrak{e}%
_{\Phi _{a}}+i\mathfrak{e}_{\Phi _{a}^{\prime }}}\left( \rho \right) \in
\lbrack 0,\Vert \Phi _{a}\Vert _{\mathcal{W}_{1}}^{2}+\Vert \Phi
_{a}^{\prime }\Vert _{\mathcal{W}_{1}}^{2}]
\label{definition of delta long range}
\end{equation}%
(cf. (\ref{delta bounded}) and (\ref{eq:enpersite bounded})) with $\gamma
_{a,\pm }\in \{0,1\}$ being the negative and positive parts (\ref{remark
positive negative part gamma}) of the fixed measurable function $\gamma
_{a}\in \{-1,1\}$, we define the free--energy density functional $f_{%
\mathfrak{m}}^{\sharp }$ as follows:

\begin{definition}[Free--energy density functional $f_{\mathfrak{m}}^{\sharp
}$]
\label{Free-energy density long range}\mbox{ }\newline
\index{Free--energy density functional!long--range|textbf}For $\beta \in
(0,\infty ]$, the free--energy density functional $f_{\mathfrak{m}}^{\sharp
} $ w.r.t. any $\mathfrak{m}\in \mathcal{M}_{1}$ is the map from $E_{%
\vec{\ell}}$ to $\mathbb{R}$ defined by%
\begin{equation*}
\rho \mapsto f_{\mathfrak{m}}^{\sharp }\left( \rho \right) :=\left\Vert
\Delta _{a,+}\left( \rho \right) \right\Vert _{1}-\left\Vert \Delta
_{a,-}\left( \rho \right) \right\Vert _{1}+e_{\Phi }(\rho )-\beta
^{-1}s(\rho ).
\end{equation*}
\end{definition}

By Corollary \ref{corollary property free--energy density functional} (i),
this functional is well--defined on $E_{\vec{\ell}}$. It is also t.i. and
affine. Moreover, on the dense set $\mathcal{E}_{1}$ of extreme states of $%
E_{1}$, i.e., on the dense set of ergodic states (see Definition \ref%
{def:egodic}, Theorem \ref{theorem ergodic extremal} and Corollary \ref%
{lemma density of extremal points}), $f_{\mathfrak{m}}^{\sharp }$ equals\
the \emph{reduced free--energy density functional} $g_{\mathfrak{m}}$
defined on $E_{\vec{\ell}}$ as follows:

\begin{definition}[Reduced free--energy density functional $g_{\mathfrak{m}}$%
]
\label{Reduced free energy}\mbox{ }\newline
\index{Free--energy density functional!reduced|textbf}For $\beta \in
(0,\infty ]$, the reduced free--energy density functional $g_{\mathfrak{m}}$
w.r.t. any $\mathfrak{m}\in \mathcal{M}_{1}$ is the map from $E_{%
\vec{\ell}}$ to $\mathbb{R}$ defined by%
\begin{equation*}
\rho \mapsto g_{\mathfrak{m}}\left( \rho \right) :=\Vert \gamma _{a,+}\rho
\left( \mathfrak{e}_{\Phi _{a}}+i\mathfrak{e}_{\Phi _{a}^{\prime }}\right)
\Vert _{2}^{2}-\Vert \gamma _{a,-}\rho \left( \mathfrak{e}_{\Phi _{a}}+i%
\mathfrak{e}_{\Phi _{a}^{\prime }}\right) \Vert _{2}^{2}+e_{\Phi }(\rho
)-\beta ^{-1}s(\rho ).
\end{equation*}
\end{definition}

This functional is an essential ingredient of the monograph. By Corollary %
\ref{corollary property free--energy density functional} (ii), it is
well--defined and by using Lemmata \ref{lemma property entropy} and \ref%
{Th.en.func} (i) as well as the weak$^{\ast }$--continuity of the maps 
\begin{equation}
\rho \mapsto \Vert \gamma _{a,\pm }\rho \left( \mathfrak{e}_{\Phi _{a}}+i%
\mathfrak{e}_{\Phi _{a}^{\prime }}\right) \Vert _{2}^{2}\in \lbrack 0,\Vert
\Phi _{a}\Vert _{2}^{2}+\Vert \Phi _{a}^{\prime }\Vert _{2}^{2}]
\label{map cool0}
\end{equation}%
defined for all $\rho \in E_{\vec{\ell}}$ (cf. (\ref{eq:enpersite bounded}%
)), it has the following properties\footnote{%
The proof of (i) uses the weak$^{\ast }$--continuity of $\rho \mapsto
|\gamma _{a,\pm }\rho \left( \mathfrak{e}_{\Phi _{a}}+i\mathfrak{e}_{\Phi
_{a}^{\prime }}\right) |^{2}$, the inequality $|\gamma _{a,\pm }\rho \left( 
\mathfrak{e}_{\Phi _{a}}+i\mathfrak{e}_{\Phi _{a}^{\prime }}\right) |\leq
\Vert \Phi _{a}\Vert _{\mathcal{W}_{1}}^{2}+\Vert \Phi _{a}^{\prime }\Vert _{%
\mathcal{W}_{1}}^{2}$, and Lebesgue's dominated convergence theorem as $%
\mathfrak{m}\in \mathcal{M}_{1}$.}:

\begin{lemma}[Properties of the reduced free--energy density functional $g_{%
\mathfrak{m}}$]
\label{lemma property reduced free--energy}\emph{(i) }%
\index{Free--energy density functional!reduced}The map $\rho \mapsto g_{%
\mathfrak{m}}\left( \rho \right) $ from $E_{%
\vec{\ell}}$ to $\mathbb{R}$ is a weak$^{\ast }$--lower semi--continuous
functional.\newline
\emph{(ii)} For any t.i. state $\rho \in E_{1}$, there is a sequence $\{\hat{%
\rho}_{n}\}_{n=1}^{\infty }\subseteq \mathcal{E}_{1}$ of ergodic states
converging in the weak$^{\ast }$--topology to $\rho $ and such that 
\begin{equation*}
g_{\mathfrak{m}}(\rho )=\lim\limits_{n\rightarrow \infty }g_{\mathfrak{m}}(%
\hat{\rho}_{n}).
\end{equation*}
\end{lemma}

However, since the maps (\ref{map cool0}) are generally not affine, the
reduced free--energy density functional $g_{\mathfrak{m}}$ has, in general,
a \emph{geometrical} drawback:

\begin{itemize}
\item[($-$)] $g_{\mathfrak{m}}$ is \emph{generally not} convex provided that 
$\Phi _{a,-}\neq 0$ (a.e.) or $\Phi _{a,-}^{\prime }\neq 0$ (a.e.), see
Definition \ref{long range attraction-repulsion}.
\end{itemize}

\noindent This does not occur (w.r.t. the set $E_{1}$ of t.i. states) if the
long--range attractions $\Phi _{a,-}$ and $\Phi _{a,-}^{\prime }$ are \emph{%
trivial} on $E_{1}$, i.e., if%
\begin{equation*}
\rho \mapsto |\gamma _{a,-}\rho (\mathfrak{e}_{\Phi _{a}}+i\mathfrak{e}%
_{\Phi _{a}^{\prime }})|
\end{equation*}%
is (a.e.) a constant map on $E_{1}$, see Remark \ref{Remark convention0}.
The property ($-$) represents a problem for our study because we are
interested in the set of t.i. minimizers of $g_{\mathfrak{m}}$, see Theorem %
\ref{BCS main theorem 1} (i) and Section \ref{Section equilibrium states}.

By contrast, since by Lemmata \ref{lemma property entropy} (i), \ref%
{Th.en.func} (i) and \ref{delta reprentation integral}, the functionals $s$, 
$e_{\Phi }$, and the maps 
\begin{equation}
\rho \mapsto \left\Vert \Delta _{a,\pm }\left( \rho \right) \right\Vert _{1}
\label{map cool}
\end{equation}%
are all affine, the free--energy density functional $f_{\mathfrak{m}%
}^{\sharp }$ is affine. In fact, by using Theorem \ref{theorem choquet} and
Lemma \ref{Corollary 4.1.18.} on each functional $s$, $e_{\Phi }$, and (\ref%
{map cool}), we can decompose, for any t.i. state $\rho \in E_{1}$, the
free--energy density functional $f_{\mathfrak{m}}^{\sharp }$ in terms of an
integral on the set $\mathcal{E}_{1}$:

\begin{lemma}[Properties of the free--energy density functional $f_{%
\mathfrak{m}}^{\sharp }$]
\label{lemma property free--energy density functional}\mbox{ }\newline
\index{Free--energy density functional!long--range}\emph{(i) }The map $\rho
\mapsto f_{\mathfrak{m}}^{\sharp }\left( \rho \right) $ from $E_{%
\vec{\ell}}$ to $\mathbb{R}$ is an affine functional. It is also t.i., i.e.,
for all $x\in \mathbb{Z}^{d}$ and $\rho \in E_{\vec{\ell}}$, $f_{\mathfrak{m}%
}^{\sharp }(\rho \circ \alpha _{x})=f_{\mathfrak{m}}^{\sharp }(\rho )$.%
\newline
\emph{(ii)} The map $\rho \mapsto f_{\mathfrak{m}}^{\sharp }\left( \rho
\right) $ from $E_{1}$ to $\mathbb{R}$ can be decomposed in terms of an
integral on the set $\mathcal{E}_{1}$ of extreme states of $E_{1}$, i.e.,
for all $\rho \in E_{1}$, 
\begin{equation*}
f_{\mathfrak{m}}^{\sharp }\left( \rho \right) =\int_{\mathcal{E}_{1}}\mathrm{%
d}\mu _{\rho }\left( \hat{\rho}\right) \;g_{\mathfrak{m}}\left( \hat{\rho}%
\right) ,
\end{equation*}%
with the probability measure $\mu _{\rho }$ defined by Theorem \ref{theorem
choquet}.
\end{lemma}

However, since the maps (\ref{map cool}) are generally not weak$^{\ast }$%
--continuous (see, e.g., Theorem \ref{Lemma1.vonN}), the free--energy
density functional $f_{\mathfrak{m}}^{\sharp }$ has, in general, a \emph{%
topological} drawback:

\begin{itemize}
\item[($+$)] $f_{\mathfrak{m}}^{\sharp }$ is \emph{generally not} weak$%
^{\ast }$--lower semi--continuous on $E_{1}$ provided that $\Phi _{a,+}\neq
0 $ (a.e.) or $\Phi _{a,+}^{\prime }\neq 0$ (a.e.), see Definition \ref{long
range attraction-repulsion}.
\end{itemize}

\noindent This does not appear (w.r.t. the set $E_{1}$ of t.i. states) if
the long--range repulsions $\Phi _{a,+}$ and $\Phi _{a,+}^{\prime }$ are 
\emph{trivial} on $E_{1}$, i.e., if%
\begin{equation*}
\rho \mapsto |\gamma _{a,+}\rho (\mathfrak{e}_{\Phi _{a}}+i\mathfrak{e}%
_{\Phi _{a}^{\prime }})|
\end{equation*}%
is (a.e.) a constant map on $E_{1}$, see Theorem \ref{Lemma1.vonN} (i) and
Remark \ref{Remark convention0}. The problem ($+$) is serious for our study
because we are interested in t.i. minimizers of $f_{\mathfrak{m}}^{\sharp }$%
, see Theorem \ref{BCS main theorem 1} (i) and Section \ref{Section
equilibrium states}.

Neither the free--energy density functional $f_{\mathfrak{m}}^{\sharp }$ nor
the reduced free--energy density functional $g_{\mathfrak{m}}$ has the usual
good properties to analyze their infimum and minimizers over t.i. states.
However, the corresponding variational problems coincide:

\begin{lemma}[Minimum of the free--energy densities]
\label{lemma property free--energy density functional copy(1)}\mbox{ }%
\newline
For any $\mathfrak{m}\in \mathcal{M}_{1}$,%
\begin{equation*}
\inf\limits_{\rho \in E_{1}}\,f_{\mathfrak{m}}^{\sharp }(\rho )=\inf\limits_{%
\hat{\rho}\in \mathcal{E}_{1}}\,f_{\mathfrak{m}}^{\sharp }\left( \hat{\rho}%
\right) =\inf\limits_{\hat{\rho}\in \mathcal{E}_{1}}\,g_{\mathfrak{m}}\left( 
\hat{\rho}\right) =\inf\limits_{\rho \in E_{1}}g_{\mathfrak{m}}\left( \rho
\right) >-\infty
\end{equation*}%
with $\mathcal{E}_{1}$ being the dense set of extreme states of $E_{1}$.
\end{lemma}

\begin{proof}%
First, as $\mathfrak{m}\in \mathcal{M}_{1}$, note that all infima in this
lemma are finite because of Remark \ref{remark boundedness entropy}, (\ref%
{eq:enpersite bounded}), Lemma \ref{Th.en.func} (ii), (\ref{definition of
delta long range}), and (\ref{map cool0}).

Now, the maps (\ref{map cool}) are both weak$^{\ast }$--upper
semi--continuous affine functionals (Lemma \ref{delta reprentation integral}%
) and the map%
\begin{equation}
\rho \mapsto -\left\Vert \Delta _{a,-}\left( \rho \right) \right\Vert
_{1}+e_{\Phi }(\rho )-\beta ^{-1}s(\rho )  \label{map2}
\end{equation}%
from $E_{1}$ to $\mathbb{R}$ is affine and weak$^{\ast }$--lower
semi--continuous (cf. Lemmata \ref{lemma property entropy} (i) and \ref%
{Th.en.func} (i)). Therefore, $f_{\mathfrak{m}}^{\sharp }$ is the sum of a
concave weak$^{\ast }$--lower semi--continuous functional and a concave weak$%
^{\ast }$--upper semi--continuous functional, whereas $E_{1}$ is weak$^{\ast
}$--compact and convex. Applying Lemma \ref{Bauer maximum principle bis}%
\index{Bauer maximum principle!Extension}, we obtain that%
\begin{equation*}
\inf\limits_{\rho \in E_{1}}\,f_{\mathfrak{m}}^{\sharp }(\rho )=\inf\limits_{%
\hat{\rho}\in \mathcal{E}_{1}}\,f_{\mathfrak{m}}^{\sharp }\left( \hat{\rho}%
\right) .
\end{equation*}%
Since $f_{\mathfrak{m}}^{\sharp }=g_{\mathfrak{m}}$ on $\mathcal{E}_{1}$, it
remains to prove the equality 
\begin{equation}
\inf\limits_{\rho \in E_{1}}g_{\mathfrak{m}}\left( \rho \right)
=\inf\limits_{\hat{\rho}\in \mathcal{E}_{1}}\,g_{\mathfrak{m}}\left( \hat{%
\rho}\right) .  \label{equality idiote encore}
\end{equation}%
In fact, using the weak$^{\ast }$--lower semi--continuity of $g_{\mathfrak{m}%
}$ (Lemma \ref{lemma property reduced free--energy} (i)), the functional $g_{%
\mathfrak{m}}$ has, at least, one minimizer $\omega $ over $E_{1}$ and by
Lemma \ref{lemma property reduced free--energy} (ii) there is a sequence $\{%
\hat{\rho}_{n}\}_{n=1}^{\infty }\subseteq \mathcal{E}_{1}$ of ergodic states
converging in the weak$^{\ast }$--topology to $\omega $ with the property
that $g_{\mathfrak{m}}(\hat{\rho}_{n})$ converges to $g_{\mathfrak{m}%
}(\omega )$ as $n\rightarrow \infty $. The latter yields Equality (\ref%
{equality idiote encore}). 
\end{proof}%

\begin{remark}[Extension of the Bauer maximum principle]
\mbox{ }\newline
\index{Bauer maximum principle!Extension}Lemma \ref{Bauer maximum principle
bis} is an extension of the Bauer maximum principle (Lemma \ref{Bauer
maximum principle}) which does not seem to have been observed before. This
lemma can be useful to do similar studies for more general long--range
interactions as it is defined in \cite{Petz2008,monsieurremark} for quantum
spin systems (see Remark \ref{Quantum spin systems}).
\end{remark}

Lemma \ref{lemma property free--energy density functional copy(1)} might be
surprising as \emph{no inequality} between $g_{\mathfrak{m}}\left( \rho
\right) $ and $f_{\mathfrak{m}}^{\sharp }\left( \rho \right) $ is generally
valid for all t.i. states $\rho \in E_{1}$. In fact, it is a pivotal result
because the variational problems of\ Lemma \ref{lemma property free--energy
density functional copy(1)} are found in the analysis of the thermodynamics
of all models $\mathfrak{m}\in \mathcal{M}_{1}$ at fixed inverse temperature 
$\beta \in (0,\infty )$ in the grand--canonical ensemble.

Indeed, the first task on the thermodynamics of long--range models is the
analysis of the thermodynamic limit $l\rightarrow \infty $ of the
finite--volume pressure%
\index{Pressure!finite volume} 
\begin{equation}
p_{l}=p_{l,\mathfrak{m}}:=%
\frac{1}{\beta |\Lambda _{l}|}\ln \mathrm{Trace}_{\wedge \mathcal{H}%
_{\Lambda _{l}}}(\mathrm{e}^{-\beta U_{l}})  \label{BCS pressure}
\end{equation}%
associated with the internal energy $U_{l}$ for $\beta \in (0,\infty )$ and
any $\mathfrak{m}\in \mathcal{M}_{1}$, see Definition \ref{definition
BCS-type model}. This limit defines a map $\mathfrak{m}\mapsto \mathrm{P}_{%
\mathfrak{m}}^{\sharp }$ from $\mathcal{M}_{1}$ to $\mathbb{R}$:

\begin{definition}[Pressure $\mathrm{P}_{\mathfrak{m}}^{\sharp }$]
\label{Pressure}\mbox{ }\newline
\index{Pressure}For $\beta \in (0,\infty )$, the (infinite--volume) pressure
is the map from $\mathcal{M}_{1}$ to $\mathbb{R}$ defined by%
\begin{equation*}
\mathfrak{m}\mapsto \mathrm{P}_{\mathfrak{m}}^{\sharp }:=\underset{%
l\rightarrow \infty }{\lim }\left\{ p_{l,\mathfrak{m}}\right\} .
\end{equation*}
\end{definition}

\noindent The pressure $\mathrm{P}_{\mathfrak{m}}^{\sharp }$ is
well--defined for any $\mathfrak{m}\in \mathcal{M}_{1}$ and can be written
as an infimum of either the free--energy density functional $f_{\mathfrak{m}%
}^{\sharp }$ or the reduced free--energy density functional $g_{\mathfrak{m}%
} $ over states (see Lemma \ref{lemma property free--energy density
functional copy(1)}):

\begin{theorem}[Pressure $\mathrm{P}_{\mathfrak{m}}^{\sharp }$ as a
variational problem on states]
\label{BCS main theorem 1}\mbox{ }\newline
\emph{(i)} 
\index{Pressure!variational problems}For any $\mathfrak{m}\in \mathcal{M}%
_{1} $,%
\begin{equation*}
\mathrm{P}_{\mathfrak{m}}^{\sharp }=-\inf\limits_{\rho \in E_{1}}\,f_{%
\mathfrak{m}}^{\sharp }(\rho )=-\inf\limits_{\rho \in E_{1}}g_{\mathfrak{m}%
}\left( \rho \right) <\infty .
\end{equation*}%
\emph{(ii)} The map $\mathfrak{m}\mapsto \mathrm{P}_{\mathfrak{m}}^{\sharp }$
from $\mathcal{M}_{1}$ to $\mathbb{R}$ is locally Lipschitz continuous.
\end{theorem}

This theorem is a combination of Theorem \ref{BCS main theorem 1 copy(1)}
with Lemma \ref{lemma property free--energy density functional copy(1)}. Its
proof uses many arguments broken, for the sake of clarity, in several
Lemmata in Chapter \ref{section proof of theorem main}. In fact, some
arguments generalize those of \cite[Theorem 11.4]{Araki-Moriya} to
non--standard potentials $\Phi \in \mathcal{W}_{1}$ but others are new, in
particular, the ones related to the long--range interaction%
\index{Interaction!long--range}%
\begin{equation*}
\frac{1}{|\Lambda _{l}|}\int_{\mathcal{A}}\gamma _{a}(U_{\Lambda _{l}}^{\Phi
_{a}}+iU_{\Lambda _{l}}^{\Phi _{a}^{\prime }})^{\ast }(U_{\Lambda
_{l}}^{\Phi _{a}}+iU_{\Lambda _{l}}^{\Phi _{a}^{\prime }})\mathrm{d}%
\mathfrak{a}\left( a\right) .
\end{equation*}%
Note that one\ argument concerning the long--range interaction uses
permutation invariant states described in Chapter \ref{Stoermer}. This
method turns out to be similar to the one used in \cite[Theorem 3.4]%
{Petz2008} and \cite[Lemma 6.1]{monsieurremark} for quantum spin systems
(see Remark \ref{Quantum spin systems}).

Indeed, for t.i. quantum spin systems%
\index{Quantum spin systems!long--range} with long--range components, the
(infinite--volume) pressure was recently proven to be given by a variational
problem over states in \cite{Petz2008,monsieurremark}. In \cite{Petz2008}
the long--range part of one--dimensional models has the form $|\Lambda _{L}|%
\mathfrak{g}\left( A_{L}\right) $ with $\mathfrak{g}$ being any real
continuous function (and with a stronger norm than $\Vert $ $\cdot $ $\Vert
_{\mathcal{W}_{1}}$), whereas in \cite{monsieurremark} there is no
restriction on the dimension and the long--range part is $|\Lambda _{L}|%
\mathfrak{g}\left( A_{L},B_{L}\right) $ for some \textquotedblleft
non--commutative polynomial\textquotedblright\ $\mathfrak{g}$. Here, $A_{L}$
and $B_{L}$ are space--averages (defined similarly as in (\ref{definition de
A L})) for (not necessarily commuting) self-adjoint operators $A$ and $B$ of
the quantum spin algebra described in Remark \ref{Quantum spin systems}.

However, Theorem \ref{BCS main theorem 1} for t.i. Fermi models with
long--range interactions has not been obtained before. Note that a certain
type of t.i. Fermi models with long--range components (e.g., reduced BCS
models) has been analyzed in \cite{RaggioWerner2} via the quantum spin
representation of fermions%
\index{Quantum spin representation}, \emph{which we never use here} as it
generally breaks the translation invariance of interactions of $\mathcal{W}%
_{1}$. Nevertheless, because of the technical approach used in \cite%
{RaggioWerner2}, the (infinite--volume) pressure is given in \cite[II.2
Theorem]{RaggioWerner2} through two variational problems ($\ast $) and ($%
\ast \ast $) over states on a much larger algebra than the original
observable algebra of the model. By \cite[II.2 Theorem and II.3 Proposition
(1)]{RaggioWerner2}, both variational problems ($\ast $) and ($\ast \ast $)
have non--empty compact sets -- respectively $M_{\ast }$ and $M_{\ast \ast }$
-- of minimizers, but the link between them and Gibbs equilibrium states is
unclear. Moreover, by \cite[II.3 Proposition (1)]{RaggioWerner2}, extreme
states of the convex and compact set $M_{\ast }$ are constructed from
minimizers of the second variational problem ($\ast \ast $) which, as the
authors wrote in \cite[p. 642]{RaggioWerner2}, \textquotedblleft can pose a
formidable task\textquotedblright .

In fact, Theorem \ref{BCS main theorem 1} (i) also gives the pressure as two
variational problems. We prove in Theorem \ref{theorem structure of omega
copy(1)} (ii) that extreme states of the convex and weak$^{\ast }$--compact
set $\mathit{\Omega }_{\mathfrak{m}}^{\sharp }$ of all weak$^{\ast }$--limit
points of approximating minimizers of $f_{\mathfrak{m}}^{\sharp }$ over $%
E_{1}$ (cf. Definition \ref{definition equilibirum state} and Lemma \ref%
{lemma minimum sympa copy(1)}) are likewise minimizers of the second
variational problem, i.e., elements of the weak$^{\ast }$--compact set $%
\mathit{%
\hat{M}}_{\mathfrak{m}}$ defined below by (\ref{definition minimizers of
reduced free energy}), see also Lemma \ref{lemma minimum sympa copy(2)} (i).
Meanwhile, the second variational problem can be analyzed and interpreted as
a two--person zero--sum game, see Section \ref{Section thermo game}. In
particular, in contrast to \cite{RaggioWerner2} and the sets $M_{\ast }$ and 
$M_{\ast \ast }$, $\mathit{\hat{M}}_{\mathfrak{m}}$ can be explicitly
characterized for all $\mathfrak{m}\in \mathcal{M}_{1}$, see Theorem \ref%
{theorem structure of omega}, whereas the set $\mathit{\Omega }_{\mathfrak{m}%
}^{\sharp }$ is related to Gibbs equilibrium states in the sense of Theorem %
\ref{lemma limit averaging gibbs states}, see also Theorem \ref{lemma limit
gibbs states periodic}. Before going into such results, we need first to
discuss the definitions and properties of the sets $\mathit{\Omega }_{%
\mathfrak{m}}^{\sharp }$ and $\mathit{\hat{M}}_{\mathfrak{m}}$ in the next
section.

\section{Generalized t.i. equilibrium states\label{Section equilibrium
states}}

We now discuss a special class of states: The (possibly generalized) 
\index{States!equilibrium}equilibrium states which are supposed to describe
physical systems at thermodynamic equilibrium. These states are always
defined in relation to a given interaction which describes the energy
density for a given state as well as the microscopical dynamics. We define
here (possibly generalized) equilibrium states via a variational principle
(Definitions \ref{definition equilibirum state copy(1)} and \ref{definition
equilibirum state}). However, this is not the only reasonable way of
defining equilibrium states. At fixed interaction they can also be defined
as tangent functionals to the corresponding pressure (Definition \ref%
{definition tangent state BCS}) or other conditions like: The local
stability condition, the Gibbs condition, or the Kubo--Martin--Schwinger
(KMS) condition. These definitions are generally not equivalent to each
other. For more details, see \cite{Araki-Moriya}.

From Theorem \ref{BCS main theorem 1} (i), the pressure $\mathrm{P}_{%
\mathfrak{m}}^{\sharp }$ is given by the infimum of the free--energy density
functional $f_{\mathfrak{m}}^{\sharp }$ over t.i. states $\rho \in E_{1}$.
When $\mathfrak{m}=(\Phi ,0,0)\in \mathcal{M}_{1}$ the map%
\begin{equation}
\rho \mapsto f_{\mathfrak{m}}^{\sharp }\left( \rho \right) =f_{\Phi }\left(
\rho \right) :=e_{\Phi }(\rho )-\beta ^{-1}s(\rho )
\label{free--energy density functional local}
\end{equation}%
is weak$^{\ast }$--lower semi--continuous and affine, see Lemmata \ref{lemma
property entropy} (i), \ref{Th.en.func} (i) and Definition \ref{Remark free
energy density}. In particular, it has minimizers in the set $E_{1}$ of t.i.
states. The corresponding set $\mathit{M}_{\Phi }$ of all t.i. minimizers is
a (non--empty) closed face of the Poulsen simplex $E_{1}$. Then, similarly
to what is done for translation invariant quantum spin systems (see, e.g., 
\cite{BrattelliRobinson,Sewell}), \emph{t.i. equilibrium states} are defined
as follows:

\begin{definition}[Set of t.i. equilibrium states]
\label{definition equilibirum state copy(1)}\mbox{ }\newline
\index{States!equilibrium}%
\index{Minimizers}%
\index{Free--energy density functional!long--range!minimizers}For $\beta \in
(0,\infty )$ and any $\mathfrak{m}\in \mathcal{M}_{1}$, the set $\mathit{M}_{%
\mathfrak{m}}^{\sharp }$ of t.i. equilibrium states is the set 
\begin{equation*}
\mathit{M}_{\mathfrak{m}}^{\sharp }:=\left\{ \omega \in E_{1}:\quad f_{%
\mathfrak{m}}^{\sharp }\left( \omega \right) =\inf\limits_{\rho \in
E_{1}}\,f_{\mathfrak{m}}^{\sharp }(\rho )\right\}
\end{equation*}%
of all minimizers of the free--energy density functional $f_{\mathfrak{m}%
}^{\sharp }$ over the set $E_{1}$.
\end{definition}

\noindent The set $\mathit{M}_{\mathfrak{m}}^{\sharp }$ is convex and in
fact, a face by affinity of the free--energy density functional $f_{%
\mathfrak{m}}^{\sharp }$ (Lemma \ref{lemma property free--energy density
functional} (i)):

\begin{lemma}[Properties of non--empty sets $\mathit{M}_{\mathfrak{m}%
}^{\sharp }$]
\label{lemma minimum sympa copy(3)}\mbox{ }\newline
If $\mathfrak{m}\in \mathcal{M}_{1}$ is such that $\mathit{M}_{\mathfrak{m}%
}^{\sharp }$ is non--empty then $\mathit{M}_{\mathfrak{m}}^{\sharp }$ is a
(possibly not closed) face of $E_{1}$.
\end{lemma}

Nevertheless, $\mathit{M}_{\mathfrak{m}}^{\sharp }$ is not necessarily weak$%
^{\ast }$--compact and depending on the model $\mathfrak{m}\in \mathcal{M}%
_{1}$, it \emph{could even be empty}. Indeed, the situation is more involved
in the case of long--range models of $\mathcal{M}_{1}$ than for t.i.
interactions of $\mathcal{W}_{1}$ as $f_{\mathfrak{m}}^{\sharp }$ is
generally not weak$^{\ast }$--lower semi--continuous: As explained above, if
the long--range repulsions $\Phi _{a,+}\neq 0$ or $\Phi _{a,+}^{\prime }\neq
0$ then the functional $f_{\mathfrak{m}}^{\sharp }$ is a sum of the maps (%
\ref{map cool}) (with $+$) and (\ref{map2}) which are, respectively, weak$%
^{\ast }$--upper and weak$^{\ast }$--lower semi--continuous functionals, see
Lemmata \ref{lemma property entropy} (i), \ref{Th.en.func} (i) and \ref%
{delta reprentation integral}. In particular, the existence of minimizers of 
$f_{\mathfrak{m}}^{\sharp }$ over $E_{1}$ is unclear unless $\Phi
_{a,+}=\Phi _{a,+}^{\prime }=0$ (a.e.).

Therefore, we shall consider any sequence $\{\rho _{n}\}_{n=1}^{\infty }$ of 
\emph{approximating t.i. minimizers}%
\index{Minimizers!approximating}, that is, any sequence $\{\rho
_{n}\}_{n=1}^{\infty }$ in $E_{1}$ such that 
\begin{equation}
\underset{n\rightarrow \infty }{\lim }f_{\mathfrak{m}}^{\sharp }(\rho
_{n})=\inf\limits_{\rho \in E_{1}}\,f_{\mathfrak{m}}^{\sharp }(\rho ).
\label{approximating minimizer}
\end{equation}%
Such sequences clearly exist and since $E_{1}$ is sequentially weak$^{\ast }$%
--compact\footnote{$E_{1}$ is sequentially weak$^{\ast }$--compact because
it is weak$^{\ast }$--compact and metrizable in the weak$^{\ast }$--topology
(Theorem \ref{Metrizability}).}, they converge in the weak$^{\ast }$%
--topology -- along subsequences -- towards t.i. states $\omega \in E_{1}$.
Thus, \emph{generalized t.i. equilibrium states} are naturally defined as
follows:

\begin{definition}[Set of generalized t.i. equilibrium states]
\label{definition equilibirum state}\mbox{ }\newline
\index{States!generalized equilibrium}%
\index{Minimizers!generalized}%
\index{Free--energy density functional!long--range!generalized minimizers}%
For $\beta \in (0,\infty ]$ and any $\mathfrak{m}\in \mathcal{M}_{1}$, the
set $\mathit{\Omega }_{\mathfrak{m}}^{\sharp }$ of generalized t.i.
equilibrium states is the (non--empty) set 
\begin{align*}
\mathit{\Omega }_{\mathfrak{m}}^{\sharp }:=& \left. 
\Big\{%
\omega \in E_{1}:\exists \{\rho _{n}\}_{n=1}^{\infty }\subseteq E_{1}\mathrm{%
\ }%
\text{with\ weak}^{\ast }\text{--limit point}\ \omega \right. \\
& \left. \qquad \qquad \qquad \qquad \qquad \text{such\ that\ }\underset{%
n\rightarrow \infty }{\lim }f_{\mathfrak{m}}^{\sharp }(\rho
_{n})=\inf\limits_{\rho \in E_{1}}\,f_{\mathfrak{m}}^{\sharp }(\rho )%
\Big\}%
\right.
\end{align*}%
of all weak$^{\ast }$--limit points of approximating minimizers of the
free--energy density functional $f_{\mathfrak{m}}^{\sharp }$ over the set $%
E_{1}$.
\end{definition}

In contrast to the convex set $\mathit{M}_{\mathfrak{m}}^{\sharp }$ which
may be either empty or not weak$^{\ast }$--compact, the set $\mathit{\Omega }%
_{\mathfrak{m}}^{\sharp }\subseteq E_{1}$ is always a (non--empty) weak$%
^{\ast }$--compact convex set:

\begin{lemma}[Properties of the set $\mathit{\Omega }_{\mathfrak{m}}^{\sharp
}$ for $\mathfrak{m}\in \mathcal{M}_{1}$]
\label{lemma minimum sympa copy(1)}\mbox{ }\newline
The set $\mathit{\Omega }_{\mathfrak{m}}^{\sharp }$ is a (non--empty) convex
and weak$^{\ast }$--compact subset of $E_{1}$.
\end{lemma}

\begin{proof}%
The convexity of the set $\mathit{\Omega }_{\mathfrak{m}}^{\sharp }$ results
from the affinity of $f_{\mathfrak{m}}^{\sharp }$. Since $E_{1}$ is a weak$%
^{\ast }$--compact subset of $\mathcal{U}^{\ast }$, the weak$^{\ast }$%
--topology is metrizable on $E_{1}$ (Theorem \ref{Metrizability}) and $%
\mathit{\Omega }_{\mathfrak{m}}^{\sharp }\subseteq E_{1}$ is weak$^{\ast }$%
--compact by Lemma \ref{theorem trivial sympa 0}. 
\end{proof}%

\begin{notation}[Generalized t.i. equilibrium states]
\label{Notation eq states}\mbox{ }\newline
\index{States!generalized equilibrium!notation}The letter $\omega $ is
exclusively reserved to denote generalized t.i. equilibrium states. Extreme
points of $\mathit{\Omega }_{\mathfrak{m}}^{\sharp }$ are usually written as 
$%
\hat{\omega}\in \mathcal{E}(\mathit{\Omega }_{\mathfrak{m}}^{\sharp })$ (cf.
Theorem \ref{theorem Krein--Millman}).
\end{notation}

Obviously, $\mathit{M}_{\mathfrak{m}}^{\sharp }\subseteq \mathit{\Omega }_{%
\mathfrak{m}}^{\sharp }$ and for any $\Phi \in \mathcal{W}_{1}$, i.e., $%
\mathfrak{m}=(\Phi ,0,0)\in \mathcal{M}_{1}$, $\mathit{M}_{\Phi }=\mathit{%
\Omega }_{\Phi }$. Conversely, ergodic generalized t.i. equilibrium states $%
\hat{\omega}\in \mathit{\Omega }_{\mathfrak{m}}^{\sharp }\cap \mathcal{E}%
_{1} $ are always contained in $\mathit{M}_{\mathfrak{m}}^{\sharp }$:

\begin{lemma}[Ergodic generalized t.i. equilibrium states are minimizers]
\label{lemma minimum sympa}For any $\mathfrak{m}\in \mathcal{M}_{1}$, $%
\mathit{\Omega }_{\mathfrak{m}}^{\sharp }\cap \mathcal{E}_{1}=\mathit{M}_{%
\mathfrak{m}}^{\sharp }\cap \mathcal{E}_{1}$. In particular, if $\mathit{%
\Omega }_{\mathfrak{m}}^{\sharp }$ is a face then it is the weak$^{\ast }$%
--closure of the non--empty set $\mathit{M}_{\mathfrak{m}}^{\sharp }$ of
minimizers of $f_{\mathfrak{m}}^{\sharp }$ over $E_{1}$.
\end{lemma}

\noindent 
\begin{proof}
Because of Definition \ref{definition equilibirum state}, the proof is a
direct consequence of the continuity of the space--averaging functional $%
\Delta _{A}$ at any ergodic state $\rho \in \mathcal{E}_{1}$ together with
Lebesgue's dominated convergence theorem and the weak$^{\ast }$--lower
semi--continuity of the (local) free--energy density functional $f_{\Phi }$ (%
\ref{free--energy density functional local}), see Theorem \ref{Lemma1.vonN}
(iii), Lemmata \ref{lemma property entropy} (i) and \ref{Th.en.func} (i).

Additionally, if $\mathit{\Omega }_{\mathfrak{m}}^{\sharp }$ is a face then
by affinity of $f_{\mathfrak{m}}^{\sharp }$, any state of the convex hull of 
$\mathit{\Omega }_{\mathfrak{m}}^{\sharp }\cap \mathcal{E}_{1}$ is a
minimizer of $f_{\mathfrak{m}}^{\sharp }$. By the Krein--Milman theorem%
\index{Krein--Milman theorem} (Theorem \ref{theorem Krein--Millman}), $%
\mathit{\Omega }_{\mathfrak{m}}^{\sharp }$ is contained in the weak$^{\ast }$%
--closure of the set of all minimizers of $f_{\mathfrak{m}}^{\sharp }$. On
the other hand, as any minimizer of $f_{\mathfrak{m}}^{\sharp }$ is
contained in the closed set $\mathit{\Omega }_{\mathfrak{m}}^{\sharp }$, the
weak$^{\ast }$--closure of the set of all minimizers is obviously included
in $\mathit{\Omega }_{\mathfrak{m}}^{\sharp }$. 
\end{proof}%

We observe now that Definition \ref{definition equilibirum state} is not the
only natural way of defining generalized t.i. 
\index{States!equilibrium}equilibrium states. Indeed, Theorem \ref{BCS main
theorem 1} (i) says that the pressure $\mathrm{P}_{\mathfrak{m}}^{\sharp }$
is also given (up to a minus sign) by the infimum of the reduced
free--energy density functional $g_{\mathfrak{m}}$ over $E_{1}$. The
functional $g_{\mathfrak{m}}$ from Definition \ref{Reduced free energy} is a
weak$^{\ast }$--lower semi--continuous map (Lemma \ref{lemma property
reduced free--energy} (i)) and has only (usual) minimizers in the set $E_{1}$
as any sequence $\{\rho _{n}\}_{n=1}^{\infty }\subseteq E_{1}$ of
approximating t.i. minimizers of $g_{\mathfrak{m}}$ converges to a minimizer
of $g_{\mathfrak{m}}$ over $E_{1}$. Minimizers of $g_{\mathfrak{m}}$ over $%
E_{1}$ form a non--empty set denoted by%
\index{Minimizers}%
\index{Free--energy density functional!reduced!minimizers} 
\begin{equation}
\mathit{%
\hat{M}}_{\mathfrak{m}}:=\left\{ \omega \in E_{1}:\quad g_{\mathfrak{m}%
}\left( \omega \right) =\inf\limits_{\rho \in E_{1}}\,g_{\mathfrak{m}}(\rho
)\right\} .  \label{definition minimizers of reduced free energy}
\end{equation}%
This set is weak$^{\ast }$--compact and included in the set $\mathit{\Omega }%
_{\mathfrak{m}}^{\sharp }$ of generalized t.i. equilibrium states:

\begin{lemma}[Properties of the set $\mathit{\hat{M}}_{\mathfrak{m}}$ for $%
\mathfrak{m}\in \mathcal{M}_{1}$]
\label{lemma minimum sympa copy(2)}\mbox{ }\newline
\emph{(i)} The set $\mathit{\hat{M}}_{\mathfrak{m}}$ is a (non--empty) weak$%
^{\ast }$--compact subset of $E_{1}$. \newline
\emph{(ii)} The weak$^{\ast }$--closed convex hull of $\mathit{\hat{M}}_{%
\mathfrak{m}}$ is included in $\mathit{\Omega }_{\mathfrak{m}}^{\sharp }$,
i.e., 
\begin{equation*}
\overline{\mathrm{co}(\mathit{\hat{M}}_{\mathfrak{m}})}\subseteq \mathit{%
\Omega }_{\mathfrak{m}}^{\sharp }.
\end{equation*}
\end{lemma}

\begin{proof}%
The assertion (i) is a direct consequence of the weak$^{\ast }$--lower
semi--continuity of the functional $g_{\mathfrak{m}}$ (Lemma \ref{lemma
property reduced free--energy} (i)) together with the weak$^{\ast }$%
--compacticity of $E_{1}$. The second one results from Lemmata \ref{lemma
property reduced free--energy} (ii), \ref{lemma property free--energy
density functional copy(1)} and \ref{lemma minimum sympa copy(1)}. Indeed,
by Lemma \ref{lemma property reduced free--energy} (ii), for any $\omega \in 
\mathit{\hat{M}}_{\mathfrak{m}}$, there is a sequence $\{\hat{\rho}%
_{n}\}_{n=1}^{\infty }\subseteq \mathcal{E}_{1}$ of ergodic states
converging in the weak$^{\ast }$--topology to $\omega $ with the property
that $g_{\mathfrak{m}}(\hat{\rho}_{n})=f_{\mathfrak{m}}^{\sharp }(\hat{\rho}%
_{n})$ converges to $g_{\mathfrak{m}}(\omega )$ as $n\rightarrow \infty $.
Since by Lemma \ref{lemma property free--energy density functional copy(1)}, 
$g_{\mathfrak{m}}(\omega )$ is also the infimum of the functional $f_{%
\mathfrak{m}}^{\sharp }$ over $E_{1}$, we obtain that $\omega \in \mathit{%
\Omega }_{\mathfrak{m}}^{\sharp }$, see Definition \ref{definition
equilibirum state}. As a consequence, the second assertion (ii) holds
because $\mathit{\Omega }_{\mathfrak{m}}^{\sharp }$ is convex and weak$%
^{\ast }$--compact by Lemma \ref{lemma minimum sympa copy(1)}. 
\end{proof}%

Definition \ref{definition equilibirum state} seems to be a more reasonable
way of defining generalized t.i. equilibrium states. Indeed, $\mathit{\hat{M}%
}_{\mathfrak{m}}$ is generally not convex because the functional $g_{%
\mathfrak{m}}$ is generally not convex provided that $\Phi _{a,-}\neq 0$
(a.e.) or $\Phi _{a,-}^{\prime }\neq 0$ (a.e.). Hence, we have, in general,
only one inclusion: $\mathit{\hat{M}}_{\mathfrak{m}}\subseteq \mathit{\Omega 
}_{\mathfrak{m}}^{\sharp }$. In fact, we show in Theorem \ref{theorem
structure of omega copy(1)} (i) that the weak$^{\ast }$--closed convex hull
of $\mathit{\hat{M}}_{\mathfrak{m}}$ equals $\mathit{\Omega }_{\mathfrak{m}%
}^{\sharp }$. The equality $\mathit{\hat{M}}_{\mathfrak{m}}=\mathit{\Omega }%
_{\mathfrak{m}}^{\sharp }$ holds for purely repulsive long--range models for
which $\Phi _{a,-}=\Phi _{a,-}^{\prime }=0$ (a.e.), see Theorem \ref{theorem
purement repulsif sympa} ($+$).

\begin{remark}[Generalized t.i. ground states]
\mbox{ }\newline
\index{States!generalized ground}All results concerning generalized t.i.
equilibrium states are performed at finite temperature, i.e., at fixed $%
\beta \in (0,\infty )$. However, each\ weak$^{\ast }$--limit\ point $\omega $%
\ of\ the\ sequence\ of\ states\ $\omega ^{(n)}\in \mathit{\Omega }_{%
\mathfrak{m}_{i}}^{\sharp }$\ of\ models\ $\{\mathfrak{m}_{n}\}_{n\in 
\mathbb{N}}$ in $\mathcal{M}_{1}$ such that $\beta _{n}\rightarrow \infty $
and $\mathfrak{m}_{n}\rightarrow \mathfrak{m}\in \mathcal{M}_{1}$ can be
seen as a generalized t.i. ground state of $\mathfrak{m}$. An analysis of
generalized t.i. ground states is not performed here, but it essentially
uses the same kind of arguments as for $\mathit{\Omega }_{\mathfrak{m}%
}^{\sharp }$, see, e.g., \cite[Section 6.2]{BruPedra1}.
\end{remark}

\section{Structure of the set $\mathit{\Omega }_{\mathfrak{m}}^{\sharp }$ of
generalized t.i. equilibrium states\label{Section equilibrium states copy(1)}%
}

\index{States!generalized equilibrium}By Lemma \ref{lemma property
free--energy density functional} (i) recall that the free--energy density
functional $f_{\mathfrak{m}}^{\sharp }$ is affine but generally not weak$%
^{\ast }$--lower semi--continuous, even on the set $E_{1}$ of t.i. states as
explained in Sections \ref{existence of thermodynamics} and \ref{Section
equilibrium states}. The variational problem 
\begin{equation*}
\mathrm{P}_{\mathfrak{m}}^{\sharp }=-\inf\limits_{\rho \in E_{1}}\,f_{%
\mathfrak{m}}^{\sharp }(\rho )
\end{equation*}%
given in Theorem \ref{BCS main theorem 1} (i) is, however, not as difficult
as it may look like provided it is attacked in the right way.

Indeed, since we are interested in \emph{global} (possibly approximating)
t.i. minimizers of $f_{\mathfrak{m}}^{\sharp }$ (cf. Definition \ref%
{definition equilibirum state}), it is natural to introduce its $\Gamma $%
--regularization $\Gamma _{E_{1}}(f_{\mathfrak{m}}^{\sharp })$ on $E_{1}$,
that is, for all $\rho \in E_{1}$, 
\begin{equation}
\Gamma _{E_{1}}(f_{\mathfrak{m}}^{\sharp })\left( \rho \right) :=\sup
\left\{ m(\rho ):m\in \mathrm{A}\left( \mathcal{U}^{\ast }\right) \;%
\text{and }m|_{E_{1}}\leq f_{\mathfrak{m}}^{\sharp }|_{E_{1}}\right\}
\label{gamma+}
\end{equation}%
with $\mathrm{A}\left( \mathcal{U}^{\ast }\right) $ being the set of all
affine and weak$^{\ast }$--continuous functions on the dual space $\mathcal{U%
}^{\ast }$ of the $C^{\ast }$--algebra $\mathcal{U}$. See also Definition %
\ref{gamm regularisation} in Section \ref{Section gamma regularization}.
Indeed, for all $\mathfrak{m}\in \mathcal{M}_{1}$,%
\index{Gamma--regularization!free--energy density functional} 
\begin{equation*}
\inf\limits_{\rho \in E_{1}}\,f_{\mathfrak{m}}^{\sharp }(\rho
)=\inf\limits_{\rho \in E_{1}}\,\Gamma _{E_{1}}(f_{\mathfrak{m}}^{\sharp
})(\rho ),
\end{equation*}%
see Theorem \ref{theorem trivial sympa 1} (i). The functional $\Gamma
_{E_{1}}(f_{\mathfrak{m}}^{\sharp })$ has the advantage of being a weak$%
^{\ast }$--lower semi--continuous convex functional, see Section \ref%
{Section gamma regularization}. As a consequence, $\Gamma _{E_{1}}(f_{%
\mathfrak{m}}^{\sharp })$ possesses minimizers and only (usual) minimizers
over the set $E_{1}$ as any sequence of approximating t.i. minimizers of $%
\Gamma _{E_{1}}(f_{\mathfrak{m}}^{\sharp })$ automatically converges to a
minimizer of this functional over $E_{1}$. In fact, the set of minimizers of 
$\,\Gamma _{E_{1}}(f_{\mathfrak{m}}^{\sharp })$ coincides with the set $%
\mathit{\Omega }_{\mathfrak{m}}^{\sharp }$ of generalized minimizers of $f_{%
\mathfrak{m}}^{\sharp }$, see Lemma \ref{lemma minimum sympa copy(1)} and
Theorem \ref{theorem trivial sympa 1} (ii). Hence, we shall describe $\Gamma
_{E_{1}}(f_{\mathfrak{m}}^{\sharp })$ in more details.

The free--energy density functional $f_{\mathfrak{m}}^{\sharp }$ is the sum
of maps (\ref{map cool}) (with $+$) and (\ref{map2}). From Theorem \ref%
{Lemma1.vonN} (v), the $\Gamma $--regularization of $\Delta _{a,+}$ on $%
E_{1} $ is the weak$^{\ast }$--lower semi--continuous convex map 
\begin{equation}
\rho \mapsto |\gamma _{a,+}\rho \left( \mathfrak{e}_{\Phi _{a}}+i\mathfrak{e}%
_{\Phi _{a}^{\prime }}\right) |_{2}^{2},
\label{gamma regularization de delta+}
\end{equation}%
(cf. (\ref{eq:enpersite})), whereas the map (\ref{map2}) on $E_{1}$ equals
its $\Gamma $--regularization on $E_{1}$ because (\ref{map2}) is a weak$%
^{\ast }$--lower semi--continuous convex functional (cf. Corollary \ref%
{Biconjugate}). Therefore, we could try to replace the functional $\Delta
_{a,+}$ in $f_{\mathfrak{m}}^{\sharp }$ by its $\Gamma $--regularization (%
\ref{gamma regularization de delta+}). Doing this we denote by $f_{\mathfrak{%
m}}^{\flat }$ the real functional defined by 
\begin{equation}
f_{\mathfrak{m}}^{\flat }\left( \rho \right) :=\Vert \gamma _{a,+}\rho
\left( \mathfrak{e}_{\Phi _{a}}+i\mathfrak{e}_{\Phi _{a}^{\prime }}\right)
\Vert _{2}^{2}-\Vert \Delta _{a,-}\left( \rho \right) \Vert _{1}+e_{\Phi
}(\rho )-\beta ^{-1}s(\rho )  \label{convex functional g_m}
\end{equation}%
for all $\rho \in E_{1}$. However, we can \emph{not} expect that the
functional $f_{\mathfrak{m}}^{\flat }$ is, \emph{in all cases}\footnote{%
In fact, $f_{\mathfrak{m}}^{\flat }=g_{\mathfrak{m}}=\Gamma _{E_{1}}(f_{%
\mathfrak{m}}^{\sharp })$ when $\Phi _{a,-}=0$ (a.e.), see proof of Theorem %
\ref{theorem structure of omega copy(1)}.}, equal to the $\Gamma $%
--regularization $\Gamma _{E_{1}}(f_{\mathfrak{m}}^{\sharp })$ of $f_{%
\mathfrak{m}}^{\sharp }$ because\ the $\Gamma $--regularization $\Gamma
\left( h_{1}+h_{2}\right) $ of the sum of two functionals $h_{1}$ and $h_{2}$
is generally not equal to the sum $\Gamma \left( h_{1}\right) +\Gamma \left(
h_{2}\right) $.

In fact, the $\Gamma $--regularization $\Gamma _{K}(h)$ of any functional $h$
is its largest lower semi--continuous and convex minorant on $K$ (Corollary %
\ref{Biconjugate}) and as $f_{\mathfrak{m}}^{\flat }$ is a convex weak$%
^{\ast }$--lower semi--continuous functional (cf. Lemmata \ref{lemma
property entropy} (i), \ref{Th.en.func} (i) and \ref{delta reprentation
integral}), we have the inequalities 
\begin{equation}
f_{\mathfrak{m}}^{\flat }\left( \rho \right) \leq \Gamma _{E_{1}}(f_{%
\mathfrak{m}}^{\sharp })\left( \rho \right) \leq f_{\mathfrak{m}}^{\sharp
}\left( \rho \right)  \label{inequality extra}
\end{equation}%
for all $\rho \in E_{1}$. The first inequality is \emph{generally strict}.
This can easily be seen by using, for instance, any model $\mathfrak{m}\in 
\mathcal{M}_{1}$ such that%
\begin{equation*}
\left\Vert \Delta _{a,-}\left( \rho \right) \right\Vert _{1}=\left\Vert
\Delta _{a,+}\left( \rho \right) \right\Vert _{1}
\end{equation*}%
for all $\rho \in E_{1}$. As a consequence, the variational problem 
\begin{equation}
\mathrm{P}_{\mathfrak{m}}^{\flat }:=-\inf\limits_{\rho \in E_{1}}f_{%
\mathfrak{m}}^{\flat }(\rho )  \label{pressure bemol}
\end{equation}%
is only a upper bound of the pressure $\mathrm{P}_{\mathfrak{m}}^{\sharp }$,
i.e., $\mathrm{P}_{\mathfrak{m}}^{\flat }\geq \mathrm{P}_{\mathfrak{m}%
}^{\sharp }$.

Nevertheless, $\mathrm{P}_{\mathfrak{m}}^{\flat }$ is still an interesting
variational problem because it has a direct interpretation in terms of the
max--min variational problem $\mathrm{F}_{\mathfrak{m}}^{\flat }$ of the
thermodynamic game defined in Definition \ref{definition two--person
zero--sum game}, see Theorem \ref{theorem saddle point} ($\flat $).
Moreover, as $\Delta _{A}\left( 
\hat{\rho}\right) =|\hat{\rho}(A)|^{2}$ for any ergodic state $\hat{\rho}\in 
\mathcal{E}_{1}$ and $A\in \mathcal{U}$, we have that 
\begin{equation*}
f_{\mathfrak{m}}^{\flat }(\hat{\rho})=g_{\mathfrak{m}}(\hat{\rho})=f_{%
\mathfrak{m}}^{\sharp }(\hat{\rho})
\end{equation*}%
for all extreme states $\hat{\rho}\in \mathcal{E}_{1}$. By (\ref{inequality
extra}), it follows that $\Gamma _{E_{1}}(f_{\mathfrak{m}}^{\sharp })$
coincides on $\mathcal{E}_{1}$ with the explicit weak$^{\ast }$--lower
semi--continuous functional $g_{\mathfrak{m}}$ defined in Definition \ref%
{Reduced free energy}:%
\index{Gamma--regularization!free--energy density functional}%
\begin{equation}
\Gamma _{E_{1}}(f_{\mathfrak{m}}^{\sharp })(%
\hat{\rho})=g_{\mathfrak{m}}(\hat{\rho})=f_{\mathfrak{m}}^{\flat }(\hat{\rho}%
)=f_{\mathfrak{m}}^{\sharp }(\hat{\rho})  \label{equality idiote}
\end{equation}%
for any $\mathfrak{m}\in \mathcal{M}_{1}$ and all $\hat{\rho}\in \mathcal{E}%
_{1}$.

Since the set $\mathcal{E}_{1}$ of extreme points of $E_{1}$ is dense (cf.
Corollary \ref{lemma density of extremal points}), Equality (\ref{equality
idiote}) is a strong property on the functional $\Gamma _{E_{1}}(f_{%
\mathfrak{m}}^{\sharp })$. Indeed, by combining (\ref{equality idiote}) with
Lemma \ref{lemma property entropy}, Lemma \ref{lemma property free--energy
density functional}, Corollary \ref{Biconjugate}, Lemma \ref{Jensen
inequality}, and Theorem \ref{theorem trivial sympa 1}, we arrive at a
fundamental characterization of the set $\mathit{\Omega }_{\mathfrak{m}%
}^{\sharp }$ of generalized t.i. equilibrium states:

\begin{theorem}[Structure of the set $\mathit{\Omega }_{\mathfrak{m}%
}^{\sharp }$ for any $\mathfrak{m}\in \mathcal{M}_{1}$]
\index{States!generalized equilibrium}\label{theorem structure of omega
copy(1)}\mbox{ }\newline
\emph{(i)} 
\index{States!generalized equilibrium}The weak$^{\ast }$--compact and convex
set $\mathit{\Omega }_{\mathfrak{m}}^{\sharp }$ is the weak$^{\ast }$%
--closed convex hull of the weak$^{\ast }$--compact set $\mathit{%
\hat{M}}_{\mathfrak{m}}$ (\ref{definition minimizers of reduced free energy}%
), i.e., 
\begin{equation*}
\mathit{\Omega }_{\mathfrak{m}}^{\sharp }=\overline{\mathrm{co}(\mathit{\hat{%
M}}_{\mathfrak{m}})}.
\end{equation*}%
\emph{(ii)} The set $\mathcal{E}(\mathit{\Omega }_{\mathfrak{m}}^{\sharp })$
of extreme states of $\mathit{\Omega }_{\mathfrak{m}}^{\sharp }$ is included
in $\mathit{\hat{M}}_{\mathfrak{m}}$, i.e., 
\begin{equation*}
\mathcal{E}(\mathit{\Omega }_{\mathfrak{m}}^{\sharp })\subseteq \mathit{\hat{%
M}}_{\mathfrak{m}}.
\end{equation*}%
\emph{(iii)}\ For any $\omega \in \mathit{\Omega }_{\mathfrak{m}}^{\sharp }$%
, there is a probability measure $v_{\omega }$ on $\mathit{\Omega }_{%
\mathfrak{m}}^{\sharp }$ such that%
\begin{equation*}
v_{\omega }(\mathcal{E}(\mathit{\Omega }_{\mathfrak{m}}^{\sharp }))=1\mathrm{%
\quad }\text{and}\mathrm{\quad }\omega =\int_{\mathcal{E}(\mathit{\Omega }_{%
\mathfrak{m}}^{\sharp })}\mathrm{d}v_{\omega }(\hat{\omega})\;\hat{\omega}.
\end{equation*}
\end{theorem}

\begin{proof}
We first prove that $\Gamma _{E_{1}}(f_{\mathfrak{m}}^{\sharp })=\Gamma
_{E_{1}}(g_{\mathfrak{m}})$ on $E_{1}$. We start by showing that $\Gamma
_{E_{1}}(f_{\mathfrak{m}}^{\sharp })$ is a lower bound for $\Gamma
_{E_{1}}(g_{\mathfrak{m}})$.

For any $\rho \in E_{1}$, there is, by Lemma \ref{lemma property reduced
free--energy} (ii), a sequence $\{\hat{\rho}_{n}\}_{n=1}^{\infty }\subseteq 
\mathcal{E}_{1}$ of ergodic states converging in the weak$^{\ast }$%
--topology to $\rho $ and such that $g_{\mathfrak{m}}(\hat{\rho}_{n})$
converges to $g_{\mathfrak{m}}(\rho )$. By (\ref{equality idiote}), it
follows that $\Gamma _{E_{1}}(f_{\mathfrak{m}}^{\sharp })(\hat{\rho}_{n})$
also converges to $g_{\mathfrak{m}}(\rho )$. Moreover, as $\Gamma
_{E_{1}}(f_{\mathfrak{m}}^{\sharp })$ is weak$^{\ast }$--lower
semi--continuous on $E_{1}$, 
\begin{equation}
\Gamma _{E_{1}}(f_{\mathfrak{m}}^{\sharp })(\rho )\leq \ \underset{%
n\rightarrow \infty }{\lim }\Gamma _{E_{1}}(f_{\mathfrak{m}}^{\sharp })(\hat{%
\rho}_{n})=g_{\mathfrak{m}}(\rho )  \label{equality idiotebis}
\end{equation}%
for any $\rho \in E_{1}$. Applying Corollary \ref{Biconjugate} for $h=g_{%
\mathfrak{m}}$, we deduce from (\ref{equality idiotebis}) that 
\begin{equation}
\Gamma _{E_{1}}(f_{\mathfrak{m}}^{\sharp })(\rho )\leq \Gamma _{E_{1}}(g_{%
\mathfrak{m}})(\rho )  \label{equality idiotebisbis}
\end{equation}%
for all $\rho \in E_{1}$. We show next the converse inequality.

Since the functional $\Gamma _{E_{1}}(g_{\mathfrak{m}})$ is convex, by using
Theorem \ref{theorem choquet} together with Jensen's inequality%
\index{Jensen's inequality} (Lemma \ref{Jensen inequality} with $h=\Gamma
_{E_{1}}(g_{\mathfrak{m}})$) and Lemma \ref{lemma property free--energy
density functional} (ii), we obtain that%
\begin{equation*}
\Gamma _{E_{1}}(g_{\mathfrak{m}})(\rho )\leq \int_{\mathcal{E}_{1}}\mathrm{d}%
\mu _{\rho }(%
\hat{\rho})g_{\mathfrak{m}}(\hat{\rho})=f_{\mathfrak{m}}^{\sharp }(\rho )
\end{equation*}%
for all $\rho \in E_{1}$, which, by Corollary \ref{Biconjugate}, implies the
inequality%
\begin{equation}
\Gamma _{E_{1}}(g_{\mathfrak{m}})(\rho )\leq \Gamma _{E_{1}}(f_{\mathfrak{m}%
}^{\sharp })(\rho )  \label{equality idiote2}
\end{equation}%
for all $\rho \in E_{1}$. Therefore, Inequalities (\ref{equality
idiotebisbis}) and (\ref{equality idiote2}) yield $\Gamma _{E_{1}}(g_{%
\mathfrak{m}})=\Gamma _{E_{1}}(f_{\mathfrak{m}}^{\sharp })$ 
\index{Gamma--regularization!free--energy density functional}on $E_{1}$.

We apply now Theorem \ref{theorem trivial sympa 1} to $K=E_{1}$ and $h=f_{%
\mathfrak{m}}^{\sharp }$ to show that the set of minimizers of $\Gamma
_{E_{1}}(f_{\mathfrak{m}}^{\sharp })$ over $E_{1}$ is the weak$^{\ast }$%
--closed convex hull of $\mathit{\Omega }_{\mathfrak{m}}^{\sharp }$. By
Lemma \ref{lemma minimum sympa copy(1)}, $\mathit{\Omega }_{\mathfrak{m}%
}^{\sharp }$ is a convex and weak$^{\ast }$--compact set. Hence, the set of
minimizers of $\Gamma _{E_{1}}(f_{\mathfrak{m}}^{\sharp })$ over $E_{1}$
equals $\mathit{\Omega }_{\mathfrak{m}}^{\sharp }$. Then, as $\Gamma
_{E_{1}}(g_{\mathfrak{m}})=\Gamma _{E_{1}}(f_{\mathfrak{m}}^{\sharp })$ on $%
E_{1}$, $\mathit{\Omega }_{\mathfrak{m}}^{\sharp }$ is also the set of
minimizers of $\Gamma _{E_{1}}(g_{\mathfrak{m}})$\ over $E_{1}$ and by
applying again Theorem \ref{theorem trivial sympa 1} (i)--(ii) and also
Theorem \ref{theorem trivial sympa 1 copy(1)} (i) to $K=E_{1}$ and $h=g_{%
\mathfrak{m}}^{\sharp }$ we get the assertions (i)--(ii).

The third statement (iii) is a consequence of the Choquet theorem%
\index{Choquet theorem} (see Theorem \ref{theorem choquet bis}) because\ the
set $\mathit{\Omega }_{\mathfrak{m}}^{\sharp }$ is convex, weak$^{\ast }$%
--compact (Lemma \ref{lemma minimum sympa copy(1)}), and metrizable by
Theorem \ref{Metrizability}. In particular, the equality 
\begin{equation*}
\omega =\int_{\mathcal{E}(\mathit{\Omega }_{\mathfrak{m}}^{\sharp })}\mathrm{%
d}v_{\omega }(%
\hat{\omega})\;\hat{\omega}
\end{equation*}%
means, by definition, that $\omega \in \mathit{\Omega }_{\mathfrak{m}%
}^{\sharp }$ is the barycenter of the probability measure, i.e., the
normalized positive Borel regular measure, $v_{\omega }$ on $\mathit{\Omega }%
_{\mathfrak{m}}^{\sharp }$, see Definition \ref{def barycenter} and Theorem %
\ref{thm barycenter}. 
\end{proof}%

\begin{remark}[Minimization of real functionals]
\label{remark Minimization of real functionals}\mbox{ }\newline
\index{Minimization of real functionals}Theorem \ref{theorem structure of
omega copy(1)} (i)--(ii) can be proven without Theorems \ref{theorem trivial
sympa 1}--\ref{theorem trivial sympa 1 copy(1)} by using Lemma \ref{lemma
minimum sympa copy(2)} combined with Lanford III -- Robinson theorem%
\index{Lanford III -- Robinson theorem} \cite[Theorem 1]{LanRob} (Theorem %
\ref{Land.Rob}) and Lemma \ref{lemma property free--energy density
functional copy(1)}. However, Theorems \ref{theorem trivial sympa 1}--\ref%
{theorem trivial sympa 1 copy(1)} -- which do not seem to have been proven
before -- are very useful results to analyze variational problems with
non--convex functionals on a compact convex set $K$. Indeed, the
minimization of any real functional $h$ over $K$ can be done in this case by
analyzing a variational problem related to a convex lower semi--continuous
functional $\Gamma _{K}\left( h\right) $ for which various methods are
available.
\end{remark}

Note that the integral representation (iii) in Theorem \ref{theorem
structure of omega copy(1)} may not be unique, i.e., $\mathit{\Omega }_{%
\mathfrak{m}}^{\sharp }$ may not be a Choquet simplex%
\index{Simplex!Choquet} (Definition \ref{gamm regularisation copy(3)}) in
contrast to all sets $E_{%
\vec{\ell}}$ for all $\vec{\ell}\in \mathbb{N}^{d}$, see Theorems \ref%
{theorem choquet} and \ref{Thm Poulsen simplex}. In Theorem \ref{theorem
omega simplex} we give some special (but yet physically relevant) cases for
which the sets $\mathit{\Omega }_{\mathfrak{m}}^{\sharp }$ are simplices.

\begin{remark}[Pure thermodynamic phases]
\label{remark face}\mbox{ }\newline
\index{Pure phases}From\ Theorem \ref{theorem structure of omega copy(1)},
we have in $\mathit{\Omega }_{\mathfrak{m}}^{\sharp }$ a notion of pure and
mixed thermodynamic phases (equilibrium states) by identifying purity with
extremality. If $\mathit{\Omega }_{\mathfrak{m}}^{\sharp }$ turns out to be
a face in $E_{1}$ (see, e.g., Theorem \ref{theorem purement repulsif sympa} (%
$-$)) then purity corresponds to ergodicity as $\mathcal{E}(\mathit{\Omega }%
_{\mathfrak{m}}^{\sharp })=\mathit{\Omega }_{\mathfrak{m}}^{\sharp }\cap 
\mathcal{E}_{1}$ in this special case.
\end{remark}

\begin{remark}[Gauge invariant t.i. equilibrium states]
\label{remark.eq.inv.gauge}\mbox{ }\newline
\index{States!gauge invariant equilibrium}If the model $\mathfrak{m}\in 
\mathfrak{\mathcal{M}}_{1}$ is 
\index{Gauge invariant!models}gauge invariant, which means that $U_{l}\in 
\mathcal{U}^{\circ }$ (cf. (\ref{definition of gauge invariant operators})),
then the set $\mathit{\Omega }_{\beta }^{\sharp ,\circ }:=\mathit{\Omega }_{%
\mathfrak{m}}^{\sharp }\cap E_{1}^{\circ }$ of 
\index{Gauge invariant!equilibrium states}gauge invariant t.i. equilibrium
states of $\mathfrak{m}$ is the weak$^{\ast }$--closed convex hull of the
(non--empty) set $\mathit{%
\hat{M}}_{\mathfrak{m}}\cap E_{1}^{\circ }$ and its set of extreme points
equals%
\begin{equation*}
\mathcal{E}(\mathit{\Omega }_{\beta }^{\sharp ,\circ })=\mathcal{E}(\mathit{%
\Omega }_{\mathfrak{m}}^{\sharp })\cap E_{1}^{\circ }\subseteq \mathit{\hat{M%
}}_{\mathfrak{m}}\cap E_{1}^{\circ },
\end{equation*}%
cf. Remark \ref{t.i. + gauge inv states}. This follows by using Theorem \ref%
{theorem structure of omega copy(1)} together with elementary arguments. We
omit the details.
\end{remark}

We conclude now this section by analyzing some effects of negative and
repulsive long--range interactions on the thermodynamics of models $%
\mathfrak{m}\in \mathfrak{\mathcal{M}}_{1}$, see Definition \ref{long range
attraction-repulsion}. In particular, we observe that long--range
attractions $\Phi _{a,-}$ and $\Phi _{a,-}^{\prime }$ have no important
effect on the structure of the set $\mathit{\Omega }_{\mathfrak{m}}^{\sharp
} $ of generalized t.i. equilibrium states which is, for all purely local
models $(\Phi ,0,0)\in \mathfrak{\mathcal{M}}_{1}$, a (non--empty) closed
face of $E_{1}$. By contrast, long--range repulsions $\Phi _{a,+}$ and $\Phi
_{a,+}^{\prime }$ have generally a \emph{geometrical} effect by possibly
breaking the face structure of the set $\mathit{\Omega }_{\mathfrak{m}%
}^{\sharp }$ of generalized t.i. equilibrium states. Indeed, we have the
following statements:

\begin{theorem}[$\mathit{\Omega }_{\mathfrak{m}}^{\sharp }$ when $\Phi
_{a,+}=\Phi _{a,+}^{\prime }=0$ or $\Phi _{a,-}=\Phi _{a,-}^{\prime }=0$]
\label{theorem purement repulsif sympa}\mbox{ }\newline
\emph{(}$-$\emph{) }%
\index{States!generalized equilibrium}If $\Phi _{a,+}=\Phi _{a,+}^{\prime
}=0 $ (a.e.) then $\mathrm{P}_{\mathfrak{m}}:=\mathrm{P}_{\mathfrak{m}%
}^{\sharp }=\mathrm{P}_{\mathfrak{m}}^{\flat }$ and $\mathit{\Omega }_{%
\mathfrak{m}}^{\sharp }=\mathit{M}_{\mathfrak{m}}^{\sharp }$ is a closed
face of the Poulsen simplex $E_{1}$.\newline
\emph{(}$+$\emph{) }If $\Phi _{a,-}=\Phi _{a,-}^{\prime }=0$ (a.e.) then $%
\mathrm{P}_{\mathfrak{m}}:=\mathrm{P}_{\mathfrak{m}}^{\sharp }=\mathrm{P}_{%
\mathfrak{m}}^{\flat }$ and $\mathit{\Omega }_{\mathfrak{m}}^{\sharp }=%
\mathit{%
\hat{M}}_{\mathfrak{m}}$ is the set of minimizers of the convex functional $%
g_{\mathfrak{m}}$ over $E_{1}$, cf. (\ref{definition minimizers of reduced
free energy}).
\end{theorem}

\begin{proof}
In any case, $f_{\mathfrak{m}}^{\flat }$ is weak$^{\ast }$--lower
semi--continuous, see (\ref{convex functional g_m}). If $\Phi _{a,+}=\Phi
_{a,+}^{\prime }=0$ (a.e.) then $f_{\mathfrak{m}}^{\flat }=f_{\mathfrak{m}%
}^{\sharp }=\Gamma _{E_{1}}(f_{\mathfrak{m}}^{\sharp })$ is also affine, see
Definition \ref{Free-energy density long range} and (\ref{inequality extra}%
). Then the first assertion ($-$) is obvious.

If $\Phi _{a,-}=\Phi _{a,-}^{\prime }=0$ (a.e.) then $f_{\mathfrak{m}%
}^{\flat }=g_{\mathfrak{m}}$ and, by Theorem \ref{BCS main theorem 1} (i), $%
\mathrm{P}_{\mathfrak{m}}:=\mathrm{P}_{\mathfrak{m}}^{\sharp }=\mathrm{P}_{%
\mathfrak{m}}^{\flat }$. Moreover, the weak$^{\ast }$--lower
semi--continuous functional $g_{\mathfrak{m}}$ becomes convex when $\Phi
_{a,-}=\Phi _{a,-}^{\prime }=0$ (a.e.), see Definition \ref{Reduced free
energy} and (\ref{convex functional g_m}). As a consequence, the set $%
\mathit{\hat{M}}_{\mathfrak{m}}$ of minimizers of $f_{\mathfrak{m}}^{\flat
}=g_{\mathfrak{m}}$ over $E_{1}$ is convex and also weak$^{\ast }$--compact
because of Lemma \ref{lemma minimum sympa copy(2)} (i). Then applying
Theorem \ref{theorem structure of omega copy(1)} (i) we arrive at the second
assertion ($+$). 
\end{proof}%

If $\Phi _{a,-}=\Phi _{a,-}^{\prime }=0$ (a.e.) then $g_{\mathfrak{m}}=f_{%
\mathfrak{m}}^{\flat }$ can be strictly convex. As a consequence, its set $%
\mathit{\hat{M}}_{\mathfrak{m}}$ of minimizers over $E_{1}$ is, in general,
not a face, see Lemma \ref{lemma explosion l du mec copacabana2} in Section %
\ref{section breaking theroy}. This geometrical effect can lead to a \emph{%
long--range order} (LRO) implied by long--range repulsions, see Section \ref%
{Section ODLRO}.

\section{Gibbs states versus generalized equilibrium states\label{Section
Gibbs versus gen eq states}}

\index{States!Gibbs}The Gibbs equilibrium state is defined in Definition \ref%
{Gibbs.statebis} and equals the explicitly given state $\rho _{l}:=\rho
_{\Lambda _{l},U_{l}}$ (\ref{Gibbs.state}) because of Theorem \ref%
{passivity.Gibbs}, see Section \ref{Section Gibbs equilibrium states}. The
physical relevance of such a finite--volume equilibrium state is based --
among other things -- on the minimum free energy principle and the second
law of thermodynamics as explained in Section \ref{Section Gibbs equilibrium
states}: $\rho _{l}$ is a finite--volume thermal state at equilibrium. In
the same way, a generalized t.i. equilibrium state $\omega \in \mathit{%
\Omega }_{\mathfrak{m}}^{\sharp }$ represents an infinite--volume thermal
state at equilibrium. There are, however, important differences between the
finite--volume system and its thermodynamic limit:

\begin{itemize}
\item \emph{Non--uniqueness of generalized t.i. equilibrium states.} The
Gibbs equilibrium state is the unique minimizer in $E_{\Lambda }$ of the
finite--volume free--energy density (Theorem \ref{passivity.Gibbs}) but at
infinite--volume, $\omega \in \mathit{\Omega }_{\mathfrak{m}}^{\sharp }$ may
not be unique, see, e.g., \cite[Section 6.2]{BruPedra1}. Such a phenomenon
is found in symmetry broken quantum phases like the superconducting phase.
Mathematically, it is related to the fact that we leave the Fock space
representation of models to go to a representation--free formulation of
thermodynamic phases. Doing so we take advantage of the non--uniqueness of
the representation of the $C^{\ast }$--algebra $\mathcal{U}$, as stressed
for instance in \cite{Haag62,ThirWeh67,Emch} for the BCS model in
infinite--volume. This property is, indeed, necessary to get non--unique
generalized equilibrium states which imply phase transitions.

\item \emph{Space symmetry of generalized equilibrium states.} The Gibbs
equilibrium state minimizes the finite--volume free--energy density
functional over the set $E$ of all states (Theorem \ref{passivity.Gibbs}).
Observe that the Gibbs equilibrium state may possibly not converge to a t.i.
state in the thermodynamic limit. By contrast, generalized t.i. equilibrium
states $\omega \in \mathit{\Omega }_{\mathfrak{m}}^{\sharp }$ are weak$%
^{\ast }$--limit points of approximating minimizers of the free--energy
density functional $f_{\mathfrak{m}}^{\sharp }$ over the subset $%
E_{1}\subseteq E$ of t.i. states (Theorem \ref{BCS main theorem 1} (i)).
Indeed,\ the functional $f_{\mathfrak{m}}^{\sharp }$ is, a priori, only
well--defined on the set $E_{%
\vec{\ell}}$ (cf. Definition \ref{Free-energy density long range}).
Therefore, it only makes sense to speak about generalized $\mathbb{Z}_{\vec{%
\ell}}^{d}$--invariant equilibrium states. The translation invariance
property of interactions in every model $\mathfrak{m}\in \mathcal{M}_{1}$
ensures the existence of generalized t.i. equilibrium states ($\mathit{%
\Omega }_{\mathfrak{m}}^{\sharp }\neq \emptyset $), but it does not exclude
the existence of generalized $\mathbb{Z}_{\vec{\ell}}^{d}$--invariant
equilibrium states for $\vec{\ell}\neq (1,\cdots ,1)$. In other words, a
t.i. (physical) system can lead to periodic (non--translation invariant)
structures. This phenomenon can be an explanation of the appearance of
periodic superconducting phases as observed recently, see, e.g., \cite%
{LeBoeuf,Pfleiderer}. No comprehensive theory is available to explain such a
phenomenon and we will investigate this question in another paper by using
the present formalism, in particular the decomposition of generalized t.i.
equilibrium states w.r.t. generalized $\mathbb{Z}_{\vec{\ell}}^{d}$%
--invariant equilibrium states. Observe further that, by Theorem \ref{BCS
main theorem 1 copy(1)}, there is a natural extension%
\index{Free--energy density functional!long--range!extension} $\mathfrak{F}_{%
\mathfrak{m}}^{\sharp }$ (\ref{extension of fdiese}) of $f_{\mathfrak{m}%
}^{\sharp }$ on $E$ such that%
\index{Pressure!variational problems} 
\begin{equation*}
\mathrm{P}_{\mathfrak{m}}^{\sharp }=-\inf\limits_{\rho \in E}\,\mathfrak{F}_{%
\mathfrak{m}}^{\sharp }\left( \rho \right) =-\inf\limits_{\rho \in E_{%
\vec{\ell}}}\,f_{\mathfrak{m}}^{\sharp }(\rho )=-\inf\limits_{\rho \in
E_{1}}\,f_{\mathfrak{m}}^{\sharp }(\rho ).
\end{equation*}

So, the first equality could be used to define non--periodic generalized
equilibrium states for long--range systems.
\end{itemize}

\begin{remark}[Generalized $\mathbb{Z}_{\vec{\ell}}^{d}$--invariant
equilibrium states]
\mbox{ }\newline
\index{States!generalized equilibrium}Using periodically invariant
interactions, the set of generalized $\mathbb{Z}_{%
\vec{\ell}}^{d}$--invariant equilibrium states can be analyzed in the same
way we study $\mathit{\Omega }_{\mathfrak{m}}^{\sharp }$. In fact, we
restrict our analysis on t.i. Fermi systems, but all our studies can also be
done for models constructed from periodically invariant interactions.
\end{remark}

The Gibbs equilibrium state $\rho _{l}$, seen as a state either on the local
algebra $\mathcal{U}_{\Lambda _{l}}$ or on the whole algebra $\mathcal{U}$
by periodically extending\footnote{%
By the definition of interactions, $\rho _{l}$ is an even state and hence,
products of translates of $\rho _{l}$ are well--defined, see \cite[Theorem
11.2.]{Araki-Moriya}.} it (with period $(2l+1)$ in each direction of the
lattice $\mathfrak{L}$), should converge (possibly only along a subsequence)
to a minimum of the functional $\mathfrak{F}_{\mathfrak{m}}^{\sharp }$ (\ref%
{extension of fdiese}) over $E$. However, $\rho _{l}$ may not converge to a
generalized t.i. equilibrium state $\omega \in \mathit{\Omega }_{\mathfrak{m}%
}^{\sharp }$. By contrast, the \emph{space--averaged} t.i. Gibbs state%
\index{States!Gibbs!space--averaged t.i.}%
\begin{equation}
\hat{\rho}_{l}:=\frac{1}{|\Lambda _{l}|}\sum\limits_{x\in \Lambda _{l}}\rho
_{l}\circ \alpha _{x}\in E_{1}  \label{t.i. state rho l}
\end{equation}%
constructed from $\rho _{l}:=\rho _{\Lambda _{l},U_{l}}$ (\ref{Gibbs.state})
and the $\ast $--automorphisms $\{\alpha _{x}\}_{x\in \mathbb{Z}^{d}}$
defined on $\mathcal{U}$ by (\ref{transl}) always converges in the weak$%
^{\ast }$--topology to a generalized t.i. equilibrium state, see Theorem \ref%
{lemma limit averaging gibbs states}.

This can be seen by using a characterization of generalized t.i. equilibrium
states as \emph{tangent functionals}%
\index{Tangent functionals} to the pressure $\mathrm{P}_{\mathfrak{m}%
}^{\sharp }$. Indeed, by Definition \ref{Pressure}, the pressure $\mathrm{P}%
_{\mathfrak{m}}^{\sharp }$ is a map from $\mathcal{M}_{1}$ to $\mathbb{R}$
and, as a consequence, it defines by restriction a map 
\begin{equation}
\Phi \mapsto \mathrm{P}_{\mathfrak{m}}^{\sharp }\left( \Phi \right) :=%
\mathrm{P}_{\mathfrak{m}+(\Phi ,0,0)}^{\sharp }
\label{map pour def tangeant}
\end{equation}%
from the real Banach space $\mathcal{W}_{1}$ of t.i. interactions to $%
\mathbb{R}$ at any fixed $\mathfrak{m}\in \mathcal{M}_{1}$. By Theorem \ref%
{BCS main theorem 1} (ii), the map $\Phi \mapsto \mathrm{P}_{\mathfrak{m}%
}^{\sharp }\left( \Phi \right) $ is (norm) continuous and also convex
because it is the supremum over the family $\{\mathbb{A}(\rho )\}_{\rho \in
E_{1}}$ of affine maps 
\begin{equation*}
\Phi \mapsto \mathbb{A}(\rho )\,(\Phi ):=-\Vert \Delta _{a,+}\left( \rho
\right) \Vert _{1}+\Vert \Delta _{a,-}\left( \rho \right) \Vert _{1}-e_{\Phi
}(\rho )+\beta ^{-1}s(\rho )
\end{equation*}%
from $\mathcal{W}_{1}$ to $\mathbb{R}$. Therefore, by applying Theorem \ref%
{theorem trivial sympa 2} we observe that the pressure $\mathrm{P}_{%
\mathfrak{m}}^{\sharp }$ has on each point $\Phi \in \mathcal{W}_{1}$, at
least, one continuous \emph{tangent linear functional} in $\mathcal{W}%
_{1}^{\ast }$, see Definition \ref{tangent functional} in Section \ref%
{Section Legendre-Fenchel transform}.

By a slight abuse of notation, note that the set $E_{1}\subseteq \mathcal{U}%
^{\ast }$ of t.i. states can be seen as included in $\mathcal{W}_{1}^{\ast }$%
. Indeed, the energy density functional $e_{\Phi }$ defines an affine weak$%
^{\ast }$--homeomorphism $\rho \mapsto \mathbb{T}(\rho )$ from $E_{1}$ to $%
\mathcal{W}_{1}^{\ast }$ which is a norm--isometry defined for any $\rho \in
E_{1}$ by the linear continuous map 
\begin{equation*}
\Phi \mapsto \mathbb{T}(\rho )\,(\Phi ):=-e_{\Phi }(\rho )
\end{equation*}%
from $\mathcal{W}_{1}$ to $\mathbb{R}$. For more details, we recommend
Section \ref{Section state=functional on W}, in particular Lemma \ref%
{lemma.T}. For convenience, we ignore the distinction between $%
E_{1}\subseteq \mathcal{U}^{\ast }$ and $\mathbb{T}\left( E_{1}\right)
\subseteq \mathcal{W}_{1}^{\ast }$.

Using this view point, Theorem \ref{BCS main theorem 1} (i) says that the
map $\Phi \mapsto \mathrm{P}_{\mathfrak{m}}^{\sharp }\left( \Phi \right) $
is the \emph{Legendre--Fenchel transform}%
\index{Legendre--Fenchel transform} of the free--energy density functional $%
f_{\mathfrak{m}}^{\sharp }$ extended over the whole space $\mathcal{W}%
_{1}^{\ast }$, i.e., 
\begin{equation}
\mathrm{P}_{\mathfrak{m}}^{\sharp }\left( \Phi \right) :=\mathrm{P}_{%
\mathfrak{m}+(\Phi ,0,0)}^{\sharp }=(f_{\mathfrak{m}}^{\sharp })^{\ast
}(\Phi ),  \label{legendre transform free energy}
\end{equation}%
see Definitions \ref{extension of functional} and \ref{Legendre--Fenchel
transform}. Of course, the free--energy density functional $f_{\mathfrak{m}%
}^{\sharp }$ is seen here as a map from $E_{1}\subseteq \mathcal{W}%
_{1}^{\ast }$ to $\mathbb{R}$. As a consequence, the pressure $\mathrm{P}_{%
\mathfrak{m}}^{\sharp }$ is the Legendre--Fenchel transform $(f_{\mathfrak{m}%
}^{\sharp })^{\ast }(0)$ of $f_{\mathfrak{m}}^{\sharp }$ at $\Phi =0$ and it
is thus natural to identify the set of all continuous tangent functionals to 
$\Phi \mapsto \mathrm{P}_{\mathfrak{m}}^{\sharp }\left( \Phi \right) $ at $0$
with a set of t.i. states:

\begin{definition}[Set of tangent states to the pressure]
\label{definition tangent state BCS}\mbox{ }\newline
\index{States!tangent}%
\index{Tangent functionals}For $\beta \in (0,\infty )$ and any $\mathfrak{m}%
\in \mathcal{M}_{1}$, we define $\mathit{T}_{\mathfrak{m}}^{\sharp
}\subseteq E_{1}$ to be the set of t.i. states which are continuous tangent
functionals\footnote{%
Recall that we identify $\rho \in E_{1}$ with $\mathbb{T}\left( \rho \right)
\in \mathcal{W}_{1}^{\ast }$, cf. Lemma \ref{lemma.T}.} to the map $\Phi
\mapsto \mathrm{P}_{\mathfrak{m}}^{\sharp }\left( \Phi \right) $ at the
point $0\in \mathcal{W}_{1}$.
\end{definition}

Definitions \ref{definition equilibirum state} and \ref{definition tangent
state BCS} are, a priori, not equivalent to each other. In the special case
of purely local interactions $\Phi $, i.e., when $\mathfrak{m}=(\Phi ,0,0)$,
it is already known that%
\begin{equation}
\mathit{M}_{\Phi }:=\mathit{M}_{\mathfrak{m}}^{\sharp }=\mathit{\Omega }_{%
\mathfrak{m}}^{\sharp }=\mathit{T}_{\mathfrak{m}}^{\sharp }=:\mathit{T}%
_{\Phi }  \label{equivalence def equilibrium states}
\end{equation}%
for translation covariant potentials $\Phi $, see Remark \ref{remark general
interaction} and \cite[Theorem 12.10.]{Araki-Moriya}.

In fact, upon choosing $h=f_{\mathfrak{m}}^{\sharp }$ and $K=E_{1}$ for
which $\mathit{\Omega }(f_{\mathfrak{m}}^{\sharp },E_{1})=\mathit{\Omega }_{%
\mathfrak{m}}^{\sharp }$ is convex and weak$^{\ast }$--compact (Lemma \ref%
{lemma minimum sympa copy(1)}), Corollary \ref{theorem trivial sympa 3} says
that the set $\mathit{T}_{\mathfrak{m}}^{\sharp }$ of all continuous tangent
functionals equals the set $\mathit{\Omega }_{\mathfrak{m}}^{\sharp }$ of
generalized t.i. equilibrium states. In other words, Definitions \ref%
{definition equilibirum state} and \ref{definition tangent state BCS} turn
out to be equivalent:

\begin{theorem}[Generalized t.i. equilibrium states as tangent states]
\label{eq.tang.bcs.type}\mbox{ }\newline
For all $\mathfrak{m}\in \mathcal{M}_{1}$, $\mathit{T}_{\mathfrak{m}%
}^{\sharp }=\mathit{\Omega }_{\mathfrak{m}}^{\sharp }$.%
\index{States!generalized equilibrium}%
\index{Minimizers!tangent}
\end{theorem}

The equivalence of Definitions \ref{definition equilibirum state} and \ref%
{definition tangent state BCS} -- in the special case of local models $%
\mathfrak{m}=(\Phi ,0,0)$ -- has been proven, for instance, in \cite[Theorem
12.10.]{Araki-Moriya} or in \cite[Proof of Theorem 6.2.42.]%
{BrattelliRobinson} for quantum spin systems by using two results of convex
analysis: Mazur theorem \cite{Mazur}%
\index{Mazur theorem} and Lanford III -- Robinson theorem%
\index{Lanford III -- Robinson theorem} \cite[Theorem 1]{LanRob}, see
Theorems \ref{Mazur} and \ref{Land.Rob}. This method is standard, but highly
non trivial. In fact, as observed in \cite[Theorem I.6.6]{Simon}, the
approach of Theorem \ref{theorem trivial sympa 2}, which uses the
Legendre--Fenchel transform, is much easier.

Mazur theorem \cite{Mazur} (Theorem \ref{Mazur}) has an interesting
consequence on the instability of coexisting thermodynamic phases. Indeed,
thermodynamic phases are identified here with generalized t.i. equilibrium
states. From Theorem \ref{Mazur} and Remark \ref{Mazur remark} combined with
Theorem \ref{eq.tang.bcs.type}, the set of t.i. interactions in $\mathcal{W}%
_{1}$ having exactly one generalized t.i. equilibrium state is dense. Hence,
coexistence of thermodynamic phases is unstable in the sense that they can
be destroyed by arbitrarily small (w.r.t. the norm $\Vert \,\cdot \,\Vert _{%
\mathcal{W}_{1}}$) perturbations of the local interaction $\Phi $ of $%
\mathfrak{m}\in \mathcal{M}_{1}$. This phenomenon is well--known within the
case of purely local models, see, e.g., \cite[Observation 2, p. 303]%
{BrattelliRobinson} for the case of quantum spin systems.

We are now in position to prove that the space--averaged t.i. Gibbs state $%
\hat{\rho}_{l}$ defined by (\ref{t.i. state rho l}) always converges in the
weak$^{\ast }$--topology to a generalized t.i. equilibrium state:

\begin{theorem}[Weak$^{\ast }$--limit of space--averaged t.i. Gibbs states]
\label{lemma limit averaging gibbs states}\mbox{ }\newline
For any $\mathfrak{m}\in \mathcal{M}_{1}$, the weak$^{\ast }$--accumulation
points of the sequence $\{\hat{\rho}_{l}\}_{l\in \mathbb{N}}$ of ergodic
states $\hat{\rho}_{l}\in \mathcal{E}_{1}$ belong to the set $\mathit{\Omega 
}_{\mathfrak{m}}^{\sharp }$ of generalized t.i. equilibrium states.%
\index{States!generalized equilibrium}%
\index{States!Gibbs!space--averaged t.i.}
\end{theorem}

\begin{proof}%
Note that $%
\hat{\rho}_{l}\in \mathcal{E}_{1}$ is an ergodic state, see the proof of
Corollary \ref{lemma density of extremal points}. Because $E_{1}$ is weak$%
^{\ast }$--compact and metrizable, the t.i. state $\hat{\rho}_{l}$ converges
in the weak$^{\ast }$--topology\ -- along a subsequence -- towards $\omega
\in E_{1}$. Therefore, since by Theorem \ref{eq.tang.bcs.type} $\mathit{T}_{%
\mathfrak{m}}^{\sharp }=\mathit{\Omega }_{\mathfrak{m}}^{\sharp } $, we need
to prove that $\omega \in \mathit{T}_{\mathfrak{m}}^{\sharp }$ is a
continuous tangent functionals to the map $\Phi \mapsto \mathrm{P}_{%
\mathfrak{m}}^{\sharp }\left( \Phi \right) $ (\ref{map pour def tangeant})
at the point $0\in \mathcal{W}_{1}$.

For any t.i. interaction $\Phi \in \mathcal{W}_{1}$, we use Theorem \ref%
{passivity.Gibbs} (passivity of Gibbs states)%
\index{Passivity of Gibbs states} to obtain the inequality%
\begin{equation}
p_{l,\mathfrak{m}+(\Phi ,0,0)}-p_{l,\mathfrak{m}}\geq -%
\frac{1}{|\Lambda _{l}|}\rho _{l}\left( U_{\Lambda _{l}}^{\Phi }\right) .
\label{petite inequality averaged state}
\end{equation}%
If $\Phi \in \mathcal{W}_{1}^{\mathrm{f}}$ is a finite range interaction
then Lemma \ref{lemma mean energy trivial-1} tells us that the mean internal
energy per volume $\rho _{l}(U_{\Lambda _{l}}^{\Phi })/|\Lambda _{l}|$ and
the energy density $e_{\Phi }(\hat{\rho}_{l})$ converge as $l\rightarrow
\infty $ to the same limit which is $e_{\Phi }\left( \omega \right) $
because of the weak$^{\ast }$--continuity of $e_{\Phi }$ (Lemma \ref%
{Th.en.func} (i)). Therefore, by combining (\ref{petite inequality averaged
state}) with Definition \ref{Pressure} and Lemma \ref{lemma mean energy
trivial-1} one gets that for all $\Phi \in \mathcal{W}_{1}^{\mathrm{f}}$ and 
$\mathfrak{m}\in \mathcal{M}_{1}$, 
\begin{equation}
\mathrm{P}_{\mathfrak{m}+(\Phi ,0,0)}^{\sharp }-\mathrm{P}_{\mathfrak{m}%
}^{\sharp }\geq -e_{\Phi }\left( \omega \right) .
\label{petite inequality averaged statebis}
\end{equation}%
By density of the space $\mathcal{W}_{1}^{\mathrm{f}}$ in $\mathcal{W}_{1}$
together with the continuity of the maps $\Phi \mapsto e_{\Phi }\left( \rho
\right) $ (Lemma \ref{Th.en.func} (ii)) and $\Phi \mapsto \mathrm{P}_{%
\mathfrak{m}}^{\sharp }\left( \Phi \right) $ (cf. (\ref{map pour def
tangeant}) and Theorem \ref{BCS main theorem 1} (ii)), we extend the
inequality (\ref{petite inequality averaged statebis}) to all t.i.
interactions $\Phi \in \mathcal{W}_{1}$, which means that $\omega \in 
\mathit{T}_{\mathfrak{m}}^{\sharp }=\mathit{\Omega }_{\mathfrak{m}}^{\sharp
} $ (Theorem \ref{eq.tang.bcs.type}). 
\end{proof}%

A sufficient condition to obtain the weak$^{\ast }$--convergence of the
Gibbs equilibrium state $\rho _{l}$ is to have a permutation invariant
model, see Chapter \ref{Stoermer}, in particular Definition \ref{Definition
permutation inv models} and Corollary \ref{limit Gibbs states}. In fact, the
convergence or non--convergence of the Gibbs equilibrium state $\rho _{l}$
drastically depends on the boundary conditions on the box $\Lambda _{l}$
which can break the translation invariance of the infinite--volume system.
If periodic boundary conditions (see Chapter \ref{section pbc}) are imposed,
i.e., the internal energy $\tilde{U}_{l}$ (Definition \ref{definition
BCS-type model periodized}) is defined to be translation invariant on the
torus $\Lambda _{l}$, then the Gibbs equilibrium state $\tilde{\rho}%
_{l}:=\rho _{\Lambda _{l},\tilde{U}_{l}}$ (\ref{Gibbs.state}) with periodic
boundary conditions and its space--average $\hat{\rho}_{l}$ have the same
weak$^{\ast }$--limit point and $\tilde{\rho}_{l}$ converges in the weak$%
^{\ast }$--topology to a generalized t.i. equilibrium state $\omega \in 
\mathit{\Omega }_{\mathfrak{m}}^{\sharp }$, see Theorem \ref{lemma limit
gibbs states periodic}.

We conclude now by another interesting consequence -- already observed by
Israel \cite[Theorem V.2.2.]{Israel} for quantum spin systems with purely
local interactions -- of Theorem \ref{eq.tang.bcs.type}. Indeed, we deduce
from Theorem \ref{eq.tang.bcs.type} that any finite set of extreme t.i.
states can be seen as a subset of $\mathit{\Omega }_{\mathfrak{m}}^{\sharp }$
for some model $\mathfrak{m}$:

\begin{corollary}[Generalized t.i. equilibrium ergodic states]
\label{corolaire Bishop phelps1}\mbox{ }\newline
Let $\mathfrak{m}\in \mathcal{M}_{1}$ such that $\mathit{\Omega }_{\mathfrak{%
m}+(\Phi ,0,0)}^{\sharp }$ is a face for all $\Phi \in \mathcal{W}_{1}$.
Then, for any subset $\left\{ \hat{\omega}_{1},\ldots ,\hat{\omega}%
_{n}\right\} $ of $\mathcal{E}_{1}$, there is $\Phi \in \mathcal{W}_{1}$
such that $\left\{ \hat{\omega}_{1},\ldots ,\hat{\omega}_{n}\right\}
\subseteq \mathit{\Omega }_{\mathfrak{m}+(\Phi ,0,0)}^{\sharp }$.
\end{corollary}

\begin{proof}
The corollary follows from Bishop--Phelps' theorem together with Theorem \ref%
{theorem choquet}. The arguments are exactly those of Israel. Therefore, for
more details, we recommend \cite[Theorem V.2.2.]{Israel}. 
\end{proof}%

Note that the assumption of Corollary \ref{corolaire Bishop phelps1} is
satisfied, for instance, if the long--range part of the model $\mathfrak{m}%
\in \mathcal{M}_{1}$ is purely attractive, i.e., $\Phi _{a,+}=0$ (a.e.), see
Theorem \ref{theorem purement repulsif sympa} ($-$).

\section{Thermodynamics and game theory\label{Section thermo game}}

Effects of the long--range attractions $\Phi _{a,-},\Phi _{a,-}^{\prime }$
and repulsions $\Phi _{a,+},\Phi _{a,+}^{\prime }$ defined in Definition \ref%
{long range attraction-repulsion} \emph{are not symmetric} w.r.t.
thermodynamics as everything depends on variational problems given by
infima, see Theorem \ref{BCS main theorem 1} (i)%
\index{Long--range models!repulsions}%
\index{Long--range models!attractions}. For instance, the long--range
attractions $\Phi _{a,-}$ and $\Phi _{a,-}^{\prime }$ only reinforce the weak%
$^{\ast }$--lower semi--continuity of the free--energy density functional $%
f_{\mathfrak{m}}^{\sharp }$. In particular, if $\Phi _{a,+}=\Phi
_{a,+}^{\prime }=0$ (a.e.) then $\mathit{\Omega }_{\mathfrak{m}}^{\sharp }$
is, as for models $(\Phi ,0,0)\in \mathcal{M}_{1}$, a (non--empty) closed
face of $E_{1}$, see Theorem \ref{theorem purement repulsif sympa} ($-$). By
contrast, the long--range range repulsions $\Phi _{a,+}$ and $\Phi
_{a,+}^{\prime }$ have a stronger effect. Indeed, $\Phi _{a,+}$ and $\Phi
_{a,+}^{\prime }$ generally break the weak$^{\ast }$--lower semi--continuity
of the functional $f_{\mathfrak{m}}^{\sharp }$ on $E_{1}$ which, by
elementary arguments, yields, in general, to a non--affine functional $%
\Gamma _{E_{1}}(f_{\mathfrak{m}}^{\sharp })$. As a consequence, $\mathit{%
\Omega }_{\mathfrak{m}}^{\sharp }$ is generally not anymore a closed face of 
$E_{1}$, see Theorem \ref{theorem purement repulsif sympa} ($+$) and Lemma %
\ref{lemma explosion l du mec copacabana2} in Section \ref{section breaking
theroy}.

To understand this in more details, we use the view point of game theory and
interpret in Definition \ref{definition two--person zero--sum game} the
long--range attractions $\Phi _{a,-},\Phi _{a,-}^{\prime }$ and repulsions $%
\Phi _{a,+},\Phi _{a,+}^{\prime }$ of any model $\mathfrak{m}\in \mathcal{M}%
_{1}$ as attractive and repulsive players, respectively. This approach is
strongly related with the validity of the so--called Bogoliubov
approximation. In the context of the analysis of the thermodynamic pressure
of models $\mathfrak{m}\in \mathcal{M}_{1}^{\mathrm{df}}\subseteq \mathcal{M}%
_{1}$ (cf. Section \ref{definition models}) with discrete long--range part,
it is known as the \emph{approximating Hamiltonian method} \cite%
{Bogjunior,approx-hamil-method0,approx-hamil-method,approx-hamil-method2}%
\index{Approximating Hamiltonian method}, see Sections \ref{Section
historical overview} and \ref{Section approx method}. Beside our
interpretation of thermodynamics in terms of game theory, this method gives
a natural way to compute, from local interactions, the variational problems
given in Theorem \ref{BCS main theorem 1} (i) for the pressure $\mathrm{P}_{%
\mathfrak{m}}^{\sharp }$.

We show below that the pressure $\mathrm{P}_{\mathfrak{m}}^{\sharp }$ can be
studied for any models 
\begin{equation*}
\mathfrak{m}:=(\Phi ,\{\Phi _{a}\}_{a\in \mathcal{A}},\{\Phi _{a}^{\prime
}\}_{a\in \mathcal{A}})\in \mathcal{M}_{1}
\end{equation*}%
via a (Bogoliubov) min--max variational problem on the Hilbert space $L^{2}(%
\mathcal{A},\mathbb{C})$ of square integrable functions, which is
interpreted as the result of a two--person zero--sum game. Our proof
establishes, moreover, a clear link between the Bogoliubov min--max
principle for the pressure of long--range models and von Neumann min--max
theorem%
\index{von Neumann min--max theorem}. Functions $c_{a}\in L^{2}(\mathcal{A},%
\mathbb{C})$ are related to \emph{approximating interactions} defined as
follows:

\begin{definition}[Approximating interactions]
\label{definition BCS-type model approximated}\mbox{ }\newline
\index{Bogoliubov approximation}%
\index{Interaction!approximating}Approximating interactions of any model $%
\mathfrak{m}\in \mathcal{M}_{1}$ are t.i. interactions defined, for each $%
c_{a}\in L^{2}(\mathcal{A},\mathbb{C})$, by 
\begin{equation*}
\Phi (c_{a})=\Phi _{\mathfrak{m}}(c_{a}):=\Phi +2%
\func{Re}\left\{ \left\langle \Phi _{a}+i\Phi _{a}^{\prime },\gamma
_{a}c_{a}\right\rangle \right\} \in \mathcal{W}_{1}
\end{equation*}%
with $\left\langle \cdot ,\cdot \right\rangle $ being the scalar product
constructed in Section \ref{Section Preliminaries} for $\mathcal{X}=\mathcal{%
W}_{1}$ and $\gamma _{a}\in \{-1,1\}$ a fixed measurable function.
\end{definition}

Then, by Definition \ref{definition standard interaction}, the internal
energy $U_{\Lambda _{l}}^{\Phi (c_{a})}$ associated with the t.i.
interaction $\Phi (c_{a})$ equals%
\index{Interaction!approximating!internal energy}%
\begin{equation}
U_{l}(c_{a}):=U_{\Lambda _{l}}^{\Phi }+\int_{\mathcal{A}}\gamma _{a}\left\{
c_{a}(U_{\Lambda }^{\Phi _{a}}+iU_{\Lambda }^{\Phi _{a}^{\prime }})^{\ast }+%
\bar{c}_{a}(U_{\Lambda }^{\Phi _{a}}+iU_{\Lambda }^{\Phi _{a}^{\prime
}})\right\} \mathrm{d}\mathfrak{a}\left( a\right) .
\label{internal and surface energies approximated}
\end{equation}%
In particular, for any generalized t.i. equilibrium state $\omega \in 
\mathit{\Omega }_{\mathfrak{m}}^{\sharp }$ and any $c_{a}\in L^{2}(\mathcal{A%
},\mathbb{C})$, 
\begin{align}
& |\Lambda _{l}|^{-1}\omega (U_{l}-U_{l}(c_{a}))+\Vert c_{a,+}\Vert
_{2}^{2}-\Vert c_{a,-}\Vert _{2}^{2}  \label{finite volume gap eq0} \\
& \approx \int_{\mathcal{A}}\gamma _{a}\left\vert \left( |\Lambda
_{l}|^{-1}\omega (U_{\Lambda _{l}}^{\Phi _{a}}+iU_{\Lambda _{l}}^{\Phi
_{a}^{\prime }})-c_{a}\right) \right\vert ^{2}\mathrm{d}\mathfrak{a}\left(
a\right)  \notag
\end{align}%
with $c_{a,\pm }:=\gamma _{a,\pm }c_{a}$, where $\gamma _{a,\pm }\in \{0,1\}$
are the negative and positive parts (\ref{remark positive negative part
gamma}) of the fixed measurable function $\gamma _{a}$. The heuristic
(uncontrolled) approximation done in (\ref{finite volume gap eq0}) refers to
the ergodicity condition (A4) in the approximating Hamiltonian method
described in Section \ref{Section approx method}. See also \cite%
{approx-hamil-method}. Upon choosing 
\begin{equation}
c_{a}=d_{a}:=|\Lambda _{l}|^{-1}\omega (U_{\Lambda _{l}}^{\Phi
_{a}}+iU_{\Lambda _{l}}^{\Phi _{a}^{\prime }})+o(1)\mathrm{\quad (a.e.)}
\label{finite volume gap eq}
\end{equation}%
we observe that the energy densities 
\begin{equation*}
|\Lambda _{l}|^{-1}\omega (U_{l})\qquad \text{and}\qquad |\Lambda
_{l}|^{-1}\omega (U_{l}(d_{a}))
\end{equation*}%
only differ in the thermodynamic limit $l\rightarrow \infty $ by the
explicit constant 
\begin{equation*}
(\Vert d_{a,-}\Vert _{2}^{2}-\Vert d_{a,+}\Vert _{2}^{2}).
\end{equation*}%
In particular, by using the Bogoliubov (convexity) inequality \cite[%
Corollary D.4]{BruZagrebnov8}, we can expect that the approximating
interaction $\Phi (d_{a})\in \mathcal{W}_{1}$ highlights the thermodynamic
properties of models $\mathfrak{m}\in \mathcal{M}_{1}$.

\begin{remark}
Even if the order parameter $d_{a}\in L^{2}(\mathcal{A},\mathbb{C})$ is
shown to be generally not unique, these heuristic arguments are confirmed by
Theorem \ref{theorem saddle point} on the level of pressure, and by Theorems %
\ref{theorem structure of omega} on the level of states.
\end{remark}

Therefore, in order to understand the variational problems on the set $E_{1}$
given by Theorem \ref{BCS main theorem 1} (i) and more particularly the set $%
\mathit{\Omega }_{\mathfrak{m}}^{\sharp }$ of generalized t.i. equilibrium
states (Definition \ref{definition equilibirum state}), we introduce the
concept of \emph{approximating free--energy density functionals} whose
definition needs some preliminaries.

First, for any $c_{a}\in L^{2}(\mathcal{A},\mathbb{C})$, the finite--volume
pressure%
\begin{equation}
p_{l}\left( c_{a}\right) :=\frac{1}{\beta |\Lambda _{l}|}\ln \mathrm{Trace}%
_{\wedge \mathcal{H}_{\Lambda }}(\mathrm{e}^{-\beta U_{l}\left( c_{a}\right)
})  \label{pression approximated}
\end{equation}%
associated with the internal energy $U_{l}\left( c_{a}\right) $ (\ref%
{internal and surface energies approximated}) converges as $l\rightarrow
\infty $ to a well--defined (infinite--volume) pressure%
\index{Interaction!approximating!pressure} 
\begin{equation}
P_{\mathfrak{m}}\left( c_{a}\right) =-\inf\limits_{\rho \in E_{1}}\,f_{%
\mathfrak{m}}\left( \rho ,c_{a}\right)  \label{variational problem approx}
\end{equation}%
given by a variational problem over t.i. states, see Theorem \ref{BCS main
theorem 1} (i) or Proposition \ref{corrolaire sympa} in Section \ref{Section
Preliminaries copy(1)}. In comparison with the pressure $\mathrm{P}_{%
\mathfrak{m}}^{\sharp }$ for all $\mathfrak{m}\in \mathcal{M}_{1}$, $P_{%
\mathfrak{m}}\left( c_{a}\right) $ is, in practice, easier to compute
because it is associated with the (purely local) approximating interaction $%
\Phi (c_{a})$ (Definition \ref{definition BCS-type model approximated}).
Indeed, $P_{\mathfrak{m}}\left( c_{a}\right) $ is the pressure $\mathrm{P}%
_{\left( \Phi (c_{a}),0,0\right) }$ and the free--energy density functional $%
f_{\Phi (c_{a})}$ (see Definition \ref{Remark free energy density}) is equal
in this case to%
\index{Interaction!approximating!free--energy density} 
\begin{equation}
f_{\mathfrak{m}}\left( \rho ,c_{a}\right) :=2%
\func{Re}\left\{ \left\langle e_{\Phi _{a}}(\rho )+ie_{\Phi _{a}^{\prime
}}(\rho ),\gamma _{a}c_{a}\right\rangle \right\} +e_{\Phi }(\rho )-\beta
^{-1}s(\rho )  \label{free--energy density approximated 1}
\end{equation}%
for all $c_{a}\in L^{2}(\mathcal{A},\mathbb{C})$ and $\rho \in E_{1}$.

From Lemmata \ref{lemma property entropy} (i) and \ref{Th.en.func} (i), the
map $\rho \mapsto f_{\mathfrak{m}}\left( \rho ,c_{a}\right) $ from $E_{1}$
to $\mathbb{R}$ is weak$^{\ast }$--lower semi--continuous and affine. This
implies that the variational problem (\ref{variational problem approx})
leading to the pressure $P_{\mathfrak{m}}(c_{a})$ has a closed face of
minimizers (cf. Definition \ref{definition equilibirum state copy(1)}):

\begin{lemma}[Equilibrium states of approximating interactions]
\label{remark equilibrium state approches}\mbox{ }\newline
\index{Interaction!approximating!equilibrium states}For any $c_{a}\in L^{2}(%
\mathcal{A},\mathbb{C})$, the set $\mathit{M}_{\Phi \left( c_{a}\right) }=%
\mathit{\Omega }_{\Phi \left( c_{a}\right) }$ of t.i. equilibrium states of
the approximating interaction $\Phi (c_{a})$ is a (non--empty) closed face
of the Poulsen simplex $E_{1}$.
\end{lemma}

\noindent For more details concerning the map $\left( \rho ,c_{a}\right)
\mapsto f_{\mathfrak{m}}\left( \rho ,c_{a}\right) $, see Proposition \ref%
{corrolaire sympa} in Section \ref{Section Preliminaries copy(1)}.

Second, we recall again that the thermodynamics of any model $\mathfrak{m}%
\in \mathcal{M}_{1}$ drastically depends on the sign of the coupling
constant 
\begin{equation*}
\gamma _{a}=\gamma _{a,+}-\gamma _{a,-}\in \{-1,1\},\quad \mathrm{where\ }%
\gamma _{a,\pm }:=1/2(|\gamma _{a}|\pm \gamma _{a}),
\end{equation*}%
see also (\ref{remark positive negative part gamma}). Thus, we define two
Hilbert spaces corresponding respectively to the long--range repulsions $%
\Phi _{a,+},\Phi _{a,+}^{\prime }$ and attractions $\Phi _{a,-},\Phi
_{a,-}^{\prime }$ of any model $\mathfrak{m}\in \mathcal{M}_{1}$: 
\begin{equation}
L_{\pm }^{2}(\mathcal{A},\mathbb{C}):=\left\{ c_{a,\pm }\in L^{2}(\mathcal{A}%
,\mathbb{C}):c_{a,\pm }=\gamma _{a,\pm }c_{a,\pm }\right\} .
\label{definition of positive-negative L2 space}
\end{equation}%
Note that we obviously have the equality%
\begin{equation*}
L^{2}(\mathcal{A},\mathbb{C})=L_{+}^{2}(\mathcal{A},\mathbb{C})\oplus
L_{-}^{2}(\mathcal{A},\mathbb{C}).
\end{equation*}%
Then we define the approximating free--energy density functional $\mathfrak{f%
}_{\mathfrak{m}}$ as follows:

\begin{definition}[Approximating free--energy density functional]
\label{definition approximating free--energy}\mbox{ }\newline
\index{Free--energy density functional!approximating}The approximating
free--energy density functional is the map 
\begin{equation*}
\mathfrak{f}_{\mathfrak{m}}:L_{-}^{2}(\mathcal{A},\mathbb{C})\times
L_{+}^{2}(\mathcal{A},\mathbb{C})\rightarrow \mathbb{R}
\end{equation*}%
defined for any $c_{a,\pm }\in L_{\pm }^{2}(\mathcal{A},\mathbb{C})$ by 
\begin{equation*}
\mathfrak{f}_{\mathfrak{m}}\left( c_{a,-},c_{a,+}\right) :=-\left\Vert
c_{a,+}\right\Vert _{2}^{2}+\left\Vert c_{a,-}\right\Vert _{2}^{2}-P_{%
\mathfrak{m}}\left( c_{a,-}+c_{a,+}\right) .
\end{equation*}
\end{definition}

This functional is analyzed in Lemma \ref{lemma chiant property map de base}
and is used to define the (two--person zero--sum)\emph{\ thermodynamic game}
with the so--called \emph{conservative values }$\mathrm{F}_{\mathfrak{m}%
}^{\flat }$ and $\mathrm{F}_{\mathfrak{m}}^{\sharp }$:

\begin{definition}[Thermodynamic game]
\label{definition two--person zero--sum game}\mbox{ }\newline
\index{Thermodynamic game|textbf}%
\index{Thermodynamic game!conservative values}%
\index{Thermodynamic game!worst loss functional}%
\index{Thermodynamic game!least gain functional}%
\index{Zero--sum games!two--person }The thermodynamic game is the
two--person zero--sum game defined from the functional $\mathfrak{f}_{%
\mathfrak{m}}$ with conservative values%
\begin{equation*}
\mathrm{F}_{\mathfrak{m}}^{\flat }:=\underset{c_{a,+}\in L_{+}^{2}(\mathcal{A%
},\mathbb{C})}{\sup }\mathfrak{f}_{\mathfrak{m}}^{\flat }\left(
c_{a,+}\right) \quad 
\text{and}\quad \mathrm{F}_{\mathfrak{m}}^{\sharp }:=\underset{c_{a,-}\in
L_{-}^{2}(\mathcal{A},\mathbb{C})}{\inf }\mathfrak{f}_{\mathfrak{m}}^{\sharp
}\left( c_{a,-}\right) ,
\end{equation*}%
where 
\begin{equation*}
\mathfrak{f}_{\mathfrak{m}}^{\flat }\left( c_{a,+}\right) :=\underset{%
c_{a,-}\in L_{-}^{2}(\mathcal{A},\mathbb{C})}{\inf }\mathfrak{f}_{\mathfrak{m%
}}\left( c_{a,-},c_{a,+}\right) ,\quad \mathfrak{f}_{\mathfrak{m}}^{\sharp
}\left( c_{a,-}\right) :=\underset{c_{a,+}\in L_{+}^{2}(\mathcal{A},\mathbb{C%
})}{\sup }\mathfrak{f}_{\mathfrak{m}}\left( c_{a,-},c_{a,+}\right) .
\end{equation*}
\end{definition}

\noindent Any function $c_{a,+}\in L_{+}^{2}(\mathcal{A},\mathbb{C})$ (resp. 
$c_{a,-}\in L_{-}^{2}(\mathcal{A},\mathbb{C})$) is interpreted as a \emph{%
strategy} of the \emph{repulsive} (resp. \emph{attractive}) player. $%
\mathfrak{f}_{\mathfrak{m}}^{\flat }$ is the \emph{least gain }functional of
the attractive player, whereas $\mathfrak{f}_{\mathfrak{m}}^{\sharp }$ is
called the \emph{worst loss} functional of the repulsive player. Minimizers
(resp. maximizers), if there are any, of $\mathfrak{f}_{\mathfrak{m}%
}^{\sharp }$ (resp. $\mathfrak{f}_{\mathfrak{m}}^{\flat }$) are the \emph{%
conservative strategies} of the attractive (resp. repulsive) player. For
more details concerning two--person zero--sum games, see Section \ref%
{Section two--person zero--sum games}.

In Section \ref{Section optimizations problems}, we prove that both
optimization problems $\mathrm{F}_{\mathfrak{m}}^{\flat }$ and $\mathrm{F}_{%
\mathfrak{m}}^{\sharp }$ are finite and the two optimizations of $\mathfrak{f%
}_{\mathfrak{m}}\left( c_{a,-},c_{a,+}\right) $ can be restricted to balls
in $L_{\pm }^{2}(\mathcal{A},\mathbb{C})$ of radius $R<\infty $, see Lemma %
\ref{lemma idiot interaction approx 2}. Moreover, the $\sup $ and $\inf $,
both in $\mathrm{F}_{\mathfrak{m}}^{\flat }$ and $\mathrm{F}_{\mathfrak{m}%
}^{\sharp }$, are attained, i.e., they are respectively a $\max $ and a $%
\min $ and the sets%
\index{Thermodynamic game!conservative strategies|textbf} 
\begin{equation}
\begin{array}{l}
\mathcal{C}_{\mathfrak{m}}^{\flat }:=\left\{ d_{a,+}\in L_{+}^{2}(\mathcal{A}%
,\mathbb{C}):\mathrm{F}_{\mathfrak{m}}^{\flat }=\mathfrak{f}_{\mathfrak{m}%
}^{\flat }\left( d_{a,+}\right) \right\} , \\ 
\mathcal{C}_{\mathfrak{m}}^{\sharp }:=\left\{ d_{a,-}\in L_{-}^{2}(\mathcal{A%
},\mathbb{C}):\mathrm{F}_{\mathfrak{m}}^{\sharp }=\mathfrak{f}_{\mathfrak{m}%
}^{\sharp }\left( d_{a,-}\right) \right\}%
\end{array}
\label{eq conserve strategy}
\end{equation}%
of conservative strategies of the repulsive and attractive players,
respectively, are non--empty. In fact, by Lemma \ref{lemma idiot interaction
approx 2}, the set $\mathcal{C}_{\mathfrak{m}}^{\flat }$ has exactly one
element $d_{a,+}$ if $\gamma _{a,+}\neq 0$ (a.e.), whereas $\mathcal{C}_{%
\mathfrak{m}}^{\sharp }$ is non--empty, norm--bounded, and weakly compact.

The conservative values $\mathrm{F}_{\mathfrak{m}}^{\flat }$ and $\mathrm{F}%
_{\mathfrak{m}}^{\sharp }$ of the thermodynamic game turn out to be
extremely useful to understand the thermodynamics of models $\mathfrak{m}\in 
\mathcal{M}_{1}$ as they have a direct interpretation in terms of
variational problems over the set $E_{1}$. Indeed, we prove in\ Section \ref%
{Section theorem saddle point bis} (cf. Lemmata \ref{lemma super} (i) and %
\ref{lemma super copy(1)}) the following theorem:

\begin{theorem}[Thermodynamics as\ a two--person zero--sum game]
\label{theorem saddle point}\mbox{ }\newline
\emph{(}$\flat $\emph{)} 
\index{Zero--sum games!two--person }%
\index{Thermodynamic game!pressure}$\mathrm{P}_{\mathfrak{m}}^{\flat }=-%
\mathrm{F}_{\mathfrak{m}}^{\flat }$ with the pressure $\mathrm{P}_{\mathfrak{%
m}}^{\flat }$ defined, for $\mathfrak{m}\in \mathcal{M}_{1}$, by the
minimization of the functional $f_{\mathfrak{m}}^{\flat }$ over $E_{1}$, see
(\ref{pressure bemol}).\newline
\emph{(}$\sharp $\emph{)} 
\index{Pressure!variational problems!zero--sum game}$\mathrm{P}_{\mathfrak{m}%
}^{\sharp }=-\mathrm{F}_{\mathfrak{m}}^{\sharp }$ with the pressure $\mathrm{%
P}_{\mathfrak{m}}^{\sharp }$ given, for $\mathfrak{m}\in \mathcal{M}_{1}$,
by the minimization of the functional $f_{\mathfrak{m}}^{\sharp }$ over $%
E_{1}$, see Definition \ref{Pressure} and Theorem \ref{BCS main theorem 1}
(i).
\end{theorem}

\noindent The proof of this theorem uses neither Ginibre inequalities \cite[%
Eq. (2.10)]{Ginibre} nor the Bogoliubov (convexity) inequality \cite[%
Corollary D.4]{BruZagrebnov8} w.r.t. $U_{l}$ and $U_{l}(c_{a})$ (\ref%
{internal and surface energies approximated}). In particular, we \emph{never
use} Equality (\ref{finite volume gap eq0}). Consequently, the proof given
in this monograph is essentially different from those of \cite%
{Bogjunior,approx-hamil-method0,approx-hamil-method,approx-hamil-method2}.
Additionally, the equality $\mathrm{P}_{\mathfrak{m}}^{\flat }=-\mathrm{F}_{%
\mathfrak{m}}^{\flat }$ is a new result and we do not need additional
assumptions as in \cite%
{Bogjunior,approx-hamil-method0,approx-hamil-method,approx-hamil-method2}
when $\gamma _{a,+}\neq 0$ (a.e.), see Condition (A4) and Theorem \ref%
{Theorem AHM} in Section \ref{Section approx method}. Our proof uses,
instead, Theorem \ref{BCS main theorem 1} (i) together with a fine analysis
of the corresponding variational problems over the set $E_{1}$.

It follows from\ Theorem \ref{theorem saddle point} that $\mathrm{P}_{%
\mathfrak{m}}:=\mathrm{P}_{\mathfrak{m}}^{\sharp }=\mathrm{P}_{\mathfrak{m}%
}^{\flat }$ whenever either $\Phi _{a,-}=0$ (a.e.) or $\Phi _{a,+}=0$
(a.e.), as explained in Theorem \ref{theorem purement repulsif sympa}.
However, in the general case, one only has $\mathrm{F}_{\mathfrak{m}}^{\flat
}\leq \mathrm{F}_{\mathfrak{m}}^{\sharp }$, i.e., $\mathrm{P}_{\mathfrak{m}%
}^{\flat }\geq \mathrm{P}_{\mathfrak{m}}^{\sharp }$, see, e.g., (\ref%
{inequality extra}). In fact, generally, $\mathrm{P}_{\mathfrak{m}}^{\flat }>%
\mathrm{P}_{\mathfrak{m}}^{\sharp }$, i.e., $\mathrm{F}_{\mathfrak{m}%
}^{\flat }<\mathrm{F}_{\mathfrak{m}}^{\sharp }$. This fact is, indeed, not
surprising as a $\sup $ and a $\inf $ do not generally commute.

As an example, take $A=A^{\ast }\in \mathcal{U}_{0}$ and two ergodic states $%
\omega _{1},\omega _{2}\in \mathcal{E}_{1}$ such that $\omega _{1}(A)\neq
\omega _{2}(A)$. From Corollary \ref{corolaire Bishop phelps1}, there is $%
\Phi \in \mathcal{W}_{1}$ such that the t.i. states $\omega _{1}$ and $%
\omega _{2}$ belong to the closed face $\mathit{\Omega }_{\Phi }=\mathit{M}%
_{\Phi }$ of t.i. equilibrium states of the (local) model $\left( \Phi
,0,0\right) \in \mathcal{M}_{1}$. In other words, for any $\lambda \in
\lbrack 0,1]$, the convex sum%
\begin{equation*}
\lambda \omega _{1}+(1-\lambda )\omega _{2}
\end{equation*}%
is a minimizer of the free--energy density functional $f_{\Phi }$ defined in
Definition \ref{Remark free energy density}. Consequently, by using (\ref%
{inegality utile for example}) (see Section \ref{Section properties of delta}%
) we obtain that 
\begin{eqnarray*}
\inf\limits_{\rho \in E_{1}}f_{\Phi }\left( \rho \right)
&=&\inf\limits_{\rho \in E_{1}}\left\{ \Delta _{A}\left( \rho \right)
-\Delta _{A}\left( \rho \right) +f_{\Phi }\left( \rho \right) \right\} \\
&>&\inf\limits_{\rho \in E_{1}}\left\{ \left\vert \rho \left( A\right)
\right\vert _{2}^{2}-\Delta _{A}\left( \rho \right) +f_{\Phi }\left( \rho
\right) \right\} .
\end{eqnarray*}%
Combined with Theorem \ref{theorem saddle point} this strict inequality
gives a trivial example where $\mathrm{P}_{\mathfrak{m}}^{\flat }>\mathrm{P}%
_{\mathfrak{m}}^{\sharp }$, i.e., $\mathrm{F}_{\mathfrak{m}}^{\flat }<%
\mathrm{F}_{\mathfrak{m}}^{\sharp }$, because, for any $A=A^{\ast }\in 
\mathcal{U}_{0}$, there exists a finite range interaction $\Phi ^{A}\in 
\mathcal{W}_{1}$ satisfying $\Vert \Phi ^{A}\Vert _{\mathcal{W}_{1}}=\Vert
A\Vert $ and $e_{\Phi ^{A}}(\rho )=\rho \left( \mathfrak{e}_{\Phi
^{A}}\right) =\rho (A)$. Other less trivial examples can also be found by
directly showing that $\mathrm{F}_{\mathfrak{m}}^{\flat }<\mathrm{F}_{%
\mathfrak{m}}^{\sharp }$. Use, for instance, the strong coupling
BCS--Hubbard Hamiltonian described in \cite{BruPedra1}; See also \cite[Chap.
1, Section 2, 2$%
{{}^\circ}%
$]{approx-hamil-method0}. Therefore, in general, there is no \emph{saddle
points} (Definition \ref{definition saddle points}) in the thermodynamic
game defined in Definition \ref{definition two--person zero--sum game}.

The non--existence of saddle points in the thermodynamic game is an \emph{%
important} observation. It reflects the fact that repulsive and attractive
long--range forces $\Phi _{a,\pm },\Phi _{a,\pm }^{\prime }$ (Definition \ref%
{long range attraction-repulsion}) are not in \textquotedblleft
duality\textquotedblright\ in which concerns thermodynamics properties of a
given long--range model $\mathfrak{m}\in \mathcal{W}_{1}$. Indeed, the
long--range attractions $\Phi _{a,-},\Phi _{a,-}^{\prime }$ and repulsions $%
\Phi _{a,+},\Phi _{a,+}^{\prime }$ act on the thermodynamics of $\mathfrak{m}%
\in \mathcal{M}_{1}$ as the attractive and repulsive players, respectively.
Since the result of the thermodynamic game is the conservative value $%
\mathrm{F}_{\mathfrak{m}}^{\sharp }=-\mathrm{P}_{\mathfrak{m}}^{\sharp }$,
the attractive player minimizes the functional $\mathfrak{f}_{\mathfrak{m}%
}^{\sharp }\left( c_{a,-}\right) $, i.e., he optimizes his worse loss $%
\mathfrak{f}_{\mathfrak{m}}\left( c_{a,-},c_{a,+}\right) $ without knowing
the choice $d_{a,+}\in L_{+}^{2}(\mathcal{A},\mathbb{C})$ of the repulsive
player. By contrast, the repulsive player determines his strategy after
having full information on the choice of the attractive player. In other
words, as in general $\mathrm{F}_{\mathfrak{m}}^{\flat }<\mathrm{F}_{%
\mathfrak{m}}^{\sharp }$, there is a strong asymmetry between both players,
i.e., between the role of the two kinds of long--range interactions $\Phi
_{a,-},\Phi _{a,-}^{\prime }$ and $\Phi _{a,+},\Phi _{a,+}^{\prime }$.

The thermodynamic game of any given long--range model $\mathfrak{m}$ can be
extended \cite[Ch. 7, Section 7.2]{Aubinbis} to another two--person
zero--sum game with \emph{exchange of information} which has the advantage
to have, at least, one \emph{non--cooperative equilibrium}%
\index{Zero--sum games!non--cooperative equilibrium}, also called saddle
point%
\index{Saddle point} in this context. This can be seen as follows.

First, it is instructive to analyze the variational problems respectively
given by $\mathfrak{f}_{\mathfrak{m}}^{\flat }\left( c_{a,+}\right) $ and $%
\mathfrak{f}_{\mathfrak{m}}^{\sharp }\left( c_{a,-}\right) $ at fixed $%
c_{a,\pm }\in L_{\pm }^{2}(\mathcal{A},\mathbb{C})$. So, we introduced their
sets%
\begin{equation}
\begin{array}{l}
\mathcal{C}_{\mathfrak{m}}^{\flat }\left( c_{a,+}\right) :=\Big\{d_{a,-}\in
L_{-}^{2}(\mathcal{A},\mathbb{C}):\mathfrak{f}_{\mathfrak{m}}^{\flat }\left(
c_{a,+}\right) =\mathfrak{f}_{\mathfrak{m}}\left( d_{a,-},c_{a,+}\right) %
\Big\}, \\[1.2ex] 
\mathcal{C}_{\mathfrak{m}}^{\sharp }\left( c_{a,-}\right) :=\Big\{d_{a,+}\in
L_{+}^{2}(\mathcal{A},\mathbb{C}):\mathfrak{f}_{\mathfrak{m}}^{\sharp
}\left( c_{a,-}\right) =\mathfrak{f}_{\mathfrak{m}}\left(
c_{a,-},d_{a,+}\right) \Big\}%
\end{array}
\label{eq conserve strategybis}
\end{equation}%
of, respectively, minimizers and maximizers for any $c_{a,\pm }\in L_{\pm
}^{2}(\mathcal{A},\mathbb{C})$. We prove in Lemma \ref{lemma idiot
interaction approx 2 copy(2)} that, for all $c_{a,+}\in L_{+}^{2}(\mathcal{A}%
,\mathbb{C})$, the set $\mathcal{C}_{\mathfrak{m}}^{\flat }\left(
c_{a,+}\right) $ is non--empty, norm--bounded, and weakly compact, whereas,
for all $c_{a,-}\in L_{-}^{2}(\mathcal{A},\mathbb{C})$, the set $\mathcal{C}%
_{\mathfrak{m}}^{\sharp }\left( c_{a,-}\right) $ has exactly one element $%
\mathrm{r}_{+}(c_{a,-})$ provided that $\gamma _{a,\pm }\neq 0$ (a.e.).
Therefore, we would like to use Theorem \ref{theorem extension games} to
extend the strategy set $L_{+}^{2}(\mathcal{A},\mathbb{C})$ of the
thermodynamic game to the set $\mathrm{C}\left( L_{-}^{2},L_{+}^{2}\right) $
of continuous mappings from $L_{-}^{2}(\mathcal{A},\mathbb{C})$ to $%
L_{+}^{2}(\mathcal{A},\mathbb{C})$ with $L_{-}^{2}(\mathcal{A},\mathbb{C})$
and $L_{+}^{2}(\mathcal{A},\mathbb{C})$ equipped with the weak and norm
topologies, respectively.

In this context $\mathrm{C}\left( L_{-}^{2},L_{+}^{2}\right) $ is called the
set of \emph{continuous decision rules} of the repulsive player. If $\gamma
_{a,\pm }\neq 0$ (a.e.) then an important continuous decision rule is given
by the unique solution $\mathrm{r}_{+}(c_{a,-})$ of the variational problem $%
\mathfrak{f}_{\mathfrak{m}}^{\sharp }\left( c_{a,-}\right) $, see Lemma \ref%
{lemma idiot interaction approx 2 copy(2)} ($\sharp $). Indeed, the map $%
\mathrm{r}_{+}$ from $L_{-}^{2}(\mathcal{A},\mathbb{C})$ to $L_{+}^{2}(%
\mathcal{A},\mathbb{C})$ defined by%
\begin{equation}
\mathrm{r}_{+}:c_{a,-}\mapsto \mathrm{r}_{+}\left( c_{a,-}\right) \in 
\mathcal{C}_{\mathfrak{m}}^{\sharp }\left( c_{a,-}\right)
\label{thermodyn decision rule}
\end{equation}%
belongs to $\mathrm{C}\left( L_{-}^{2},L_{+}^{2}\right) $ because of Lemma %
\ref{lemma idiot interaction approx 2 copy(3)}. The functional $\mathrm{r}%
_{+}$ is called the \emph{thermodynamic decision rule} of the model $%
\mathfrak{m}\in \mathcal{M}_{1}$.%
\index{Thermodynamic game!decision rule|textbf}

We define now, for any long--range model $\mathfrak{m}\in \mathcal{M}_{1}$,
a map $\mathfrak{f}_{\mathfrak{m}}^{\mathrm{ext}}$ from $L_{-}^{2}(\mathcal{A%
},\mathbb{C})$ to $\mathrm{C}(L_{-}^{2},L_{+}^{2})$ by 
\begin{equation}
\mathfrak{f}_{\mathfrak{m}}^{\mathrm{ext}}(c_{a,-},%
\tilde{r}_{+}):=\mathfrak{f}_{\mathfrak{m}}(c_{a,-},\tilde{r}_{+}(c_{a,-}))
\label{extended functional}
\end{equation}%
for all $\tilde{r}_{+}\in \mathrm{C}(L_{-}^{2},L_{+}^{2})$. 
\index{Thermodynamic game!extended}This functional is called the \emph{%
loss--gain} function of the \emph{extended thermodynamic game} of the model $%
\mathfrak{m}$. In contrast to the thermodynamic game defined in Definition %
\ref{definition two--person zero--sum game}, this extended game has the main
advantage to have, at least, one \emph{non--cooperative equilibrium}:

\begin{theorem}[Non--cooperative equilibrium of the extended game]
\label{lemma extension game}\mbox{ }\newline
\index{Thermodynamic game!extended!non--cooperative equilibrium}Let $\gamma
_{a,+}\neq 0$ (a.e.). Then any $d_{a,-}\in \mathcal{C}_{\mathfrak{m}%
}^{\sharp }$ and the map $\mathrm{r}_{+}\in \mathrm{C}\left(
L_{-}^{2},L_{+}^{2}\right) $ defined by (\ref{thermodyn decision rule}) form
a saddle point%
\index{Saddle point}%
\index{Zero--sum games!non--cooperative equilibrium} of the extended
thermodynamic game defined by 
\begin{eqnarray*}
\mathrm{F}_{\mathfrak{m}}^{\sharp } &=&\underset{%
\tilde{r}_{+}\in \mathrm{C}\left( L_{-}^{2},L_{+}^{2}\right) }{\sup }\left\{ 
\underset{c_{a,-}\in L_{-}^{2}(\mathcal{A},\mathbb{C})}{\inf }\mathfrak{f}_{%
\mathfrak{m}}^{\mathrm{ext}}\left( c_{a,-},\tilde{r}_{+}\right) \right\} \\
&=&\underset{c_{a,-}\in L_{-}^{2}(\mathcal{A},\mathbb{C})}{\inf }\left\{ 
\underset{\tilde{r}_{+}\in \mathrm{C}\left( L_{-}^{2},L_{+}^{2}\right) }{%
\sup }\mathfrak{f}_{\mathfrak{m}}^{\mathrm{ext}}\left( c_{a,-},\tilde{r}%
_{+}\right) \right\} .
\end{eqnarray*}
\end{theorem}

\begin{proof}%
The map $\mathrm{r}_{+}$ is well--defined because of Lemma \ref{lemma idiot
interaction approx 2 copy(2)} ($\sharp $) and, by Lemma \ref{lemma idiot
interaction approx 2 copy(3)}, it is continuous w.r.t. the weak topology in $%
L_{-}^{2}(\mathcal{A},\mathbb{C})$ and the norm topology in $L_{+}^{2}(%
\mathcal{A},\mathbb{C})$, i.e., $\mathrm{r}_{+}\in \mathrm{C}\left(
L_{-}^{2},L_{+}^{2}\right) $.

By Lemma \ref{lemma idiot interaction approx 2} ($\sharp $), the non--empty
set $\mathcal{C}_{\mathfrak{m}}^{\sharp }\subseteq L_{-}^{2}(\mathcal{A},%
\mathbb{C})$ of conservative strategies of the attractive player (cf. (\ref%
{eq conserve strategy})) is norm--bounded and weakly compact, whereas, by
Lemma \ref{lemma idiot interaction approx 2 copy(2)} ($\sharp $), the set $%
\mathcal{C}_{\mathfrak{m}}^{\sharp }(c_{a,-})$ (cf. (\ref{eq conserve
strategybis})) has exactly one element $\mathrm{r}_{+}(c_{a,-})$ at any
fixed $c_{a,-}\in L_{-}^{2}(\mathcal{A},\mathbb{C})$. As a consequence, the
infimum and supremum of $\mathrm{F}_{\mathfrak{m}}^{\sharp }<\infty $ can be
restricted to balls $\mathcal{B}_{R}\left( 0\right) $ in $L_{\pm }^{2}(%
\mathcal{A},\mathbb{C})$ of radius $R<\infty $. Therefore, by using Lemma %
\ref{lemma chiant property map de base}, we can apply Theorem \ref{theorem
extension games}%
\index{Lasry theorem} to get%
\begin{equation}
\mathrm{F}_{\mathfrak{m}}^{\sharp }=\underset{%
\tilde{r}_{+}\in \mathrm{C}\left( L_{-}^{2},L_{+}^{2}\right) }{\sup }\left\{ 
\underset{c_{a,-}\in L_{-}^{2}(\mathcal{A},\mathbb{C})}{\inf }\mathfrak{f}_{%
\mathfrak{m}}^{\mathrm{ext}}\left( c_{a,-},\tilde{r}_{+}\right) \right\} .
\label{eq lasry}
\end{equation}%
The $\inf $ and $\sup $ in the r.h.s. of the last equality trivially
commute, i.e., 
\begin{equation}
\mathrm{F}_{\mathfrak{m}}^{\sharp }=\underset{c_{a,-}\in L_{-}^{2}(\mathcal{A%
},\mathbb{C})}{\inf }\left\{ \underset{\tilde{r}_{+}\in \mathrm{C}\left(
L_{-}^{2},L_{+}^{2}\right) }{\sup }\mathfrak{f}_{\mathfrak{m}}^{\mathrm{ext}%
}\left( c_{a,-},\tilde{r}_{+}\right) \right\} ,  \label{eq lasrybis}
\end{equation}%
because%
\begin{equation*}
\underset{\tilde{r}_{+}\in \mathrm{C}\left( L_{-}^{2},L_{+}^{2}\right) }{%
\sup }\mathfrak{f}_{\mathfrak{m}}^{\mathrm{ext}}(c_{a,-},\tilde{r}_{+})=%
\mathfrak{f}_{\mathfrak{m}}^{\mathrm{ext}}\left( c_{a,-},\mathrm{r}%
_{+}\right) =\underset{c_{a,+}\in L_{+}^{2}(\mathcal{A},\mathbb{C})}{\sup }%
\mathfrak{f}_{\mathfrak{m}}\left( c_{a,-},c_{a,+}\right) .
\end{equation*}%
In particular, for any $d_{a,-}\in \mathcal{C}_{\mathfrak{m}}^{\sharp }$, 
\begin{equation*}
\mathrm{F}_{\mathfrak{m}}^{\sharp }=\mathfrak{f}_{\mathfrak{m}}^{\mathrm{ext}%
}\left( d_{a,-},\mathrm{r}_{+}\right) =\underset{c_{a,-}\in L_{-}^{2}(%
\mathcal{A},\mathbb{C})}{\inf }\mathfrak{f}_{\mathfrak{m}}^{\mathrm{ext}%
}\left( c_{a,-},\mathrm{r}_{+}\right) =\underset{\tilde{r}_{+}\in \mathrm{C}%
\left( L_{-}^{2},L_{+}^{2}\right) }{\sup }\mathfrak{f}_{\mathfrak{m}}^{%
\mathrm{ext}}(d_{a,-},\tilde{r}_{+})
\end{equation*}%
which combined with (\ref{eq lasry})--(\ref{eq lasrybis}) implies that $%
\left( d_{a,-},\mathrm{r}_{+}\right) $ is a saddle point of $\mathfrak{f}_{%
\mathfrak{m}}^{\mathrm{ext}}$. 
\end{proof}%

\begin{remark}[Thermodynamics as a three--person zero--sum game]
\mbox{ }\newline
\index{Zero--sum games!three--person}Since the pressure $P_{\mathfrak{m}%
}\left( c_{a}\right) $ in Definition \ref{definition approximating
free--energy} of the approximating free--energy density $\mathfrak{f}_{%
\mathfrak{m}}$ equals the variational problem (\ref{variational problem
approx}) over t.i. states, we could also see the equality $\mathrm{F}_{%
\mathfrak{m}}^{\sharp }=-\mathrm{P}_{\mathfrak{m}}^{\sharp }$ of Theorem \ref%
{theorem saddle point} as the result of a three--person zero--sum game. By (%
\ref{gap eq 1}) and (\ref{gap eq 2}), note that the infimum over t.i. states
and the supremum over $L_{+}^{2}(\mathcal{A},\mathbb{C})$ commute with each
other, see the proof of Lemma \ref{lemma super} for more details.
\end{remark}

\section{Gap equations and effective theories\label{Section effective
theories}}

The structure of the set $\mathit{\Omega }_{\mathfrak{m}}^{\sharp }$ of
generalized t.i. 
\index{States!generalized equilibrium}equilibrium states (Definition \ref%
{definition equilibirum state}) w.r.t. the thermodynamic game can be now
discussed in details. It is based on Section \ref{equilibirum.paragraph}
which gives a rigorous justification, on the level of generalized t.i.
equilibrium states, of the heuristics discussed in the beginning of Section %
\ref{Section thermo game}. In particular, we prove that Equality (\ref%
{finite volume gap eq}) must be satisfied in the thermodynamic limit for any
extreme point of $\mathit{\Omega }_{\mathfrak{m}}^{\sharp }$.

More precisely, for all functions $c_{a}\in L^{2}(\mathcal{A},\mathbb{C})$,
we define the (possibly empty) set 
\begin{equation}
\mathit{\Omega }_{\mathfrak{m}}^{\sharp }\left( c_{a}\right) :=\left\{
\omega \in \mathit{M}_{\Phi (c_{a})}:e_{\Phi _{a}}(\omega )+ie_{\Phi
_{a}^{\prime }}(\omega )=c_{a}\mathrm{\ (a.e.)}\right\}
\label{subset of a face}
\end{equation}%
with $\mathit{M}_{\Phi \left( c_{a}\right) }$ being the closed face
described in Lemma \ref{remark equilibrium state approches}, see also (\ref%
{equivalence def equilibrium states}). Then we obtain \emph{Euler--Lagrange
equations}%
\index{Euler--Lagrange equations} for the approximating interactions (cf.
Remark \ref{remark Euler--Lagrange equations}) -- also called \emph{gap
equations}%
\index{Gap equations} in the Physics literature (cf. Remark \ref{remark gap
eq}) -- which say that any extreme point of $\mathit{\Omega }_{\mathfrak{m}%
}^{\sharp }$ must belong to a set 
\begin{equation}
\mathit{\Omega }_{\mathfrak{m}}^{\sharp }\left( d_{a,-}+\mathrm{r}%
_{+}(d_{a,-})\right)  \label{set gap}
\end{equation}%
with $d_{a,-}\in \mathcal{C}_{\mathfrak{m}}^{\sharp }$, $\mathrm{r}_{+}\in 
\mathrm{C}\left( L_{-}^{2},L_{+}^{2}\right) $ defined by (\ref{thermodyn
decision rule}), and where $\mathcal{C}_{\mathfrak{m}}^{\sharp }$ is the
non--empty, norm--bounded, and weakly compact set defined by (\ref{eq
conserve strategy}), see Lemma \ref{lemma idiot interaction approx 2} ($%
\sharp $). Indeed, we obtain the following statements:

\begin{theorem}[Gap equations for $\mathfrak{m}\in \mathcal{M}_{1}$-- I]
\label{theorem structure of omega}\mbox{ }\newline
\emph{(i)} 
\index{States!generalized equilibrium!gap equations}The set $\mathit{%
\hat{M}}_{\mathfrak{m}}$ (\ref{definition minimizers of reduced free energy}%
) of minimizers of the functional $g_{\mathfrak{m}}$ over $E_{1}$ equals 
\begin{equation*}
\mathit{\hat{M}}_{\mathfrak{m}}=\underset{d_{a,-}\in \mathcal{C}_{\mathfrak{m%
}}^{\sharp }}{\cup }\mathit{\Omega }_{\mathfrak{m}}^{\sharp }\left( d_{a,-}+%
\mathrm{r}_{+}(d_{a,-})\right) .
\end{equation*}%
\emph{(ii)} The set $\mathcal{E}(\mathit{\Omega }_{\mathfrak{m}}^{\sharp })$
of extreme points of $\mathit{\Omega }_{\mathfrak{m}}^{\sharp }$ is included
in the union for all $d_{a,-}\in \mathcal{C}_{\mathfrak{m}}^{\sharp }$ of
the sets of all extreme points of the non--empty, disjoint, convex and weak$%
^{\ast }$--compact sets (\ref{set gap}), i.e., 
\begin{equation*}
\mathcal{E}(\mathit{\Omega }_{\mathfrak{m}}^{\sharp })\subseteq \underset{%
d_{a,-}\in \mathcal{C}_{\mathfrak{m}}^{\sharp }}{\cup }\mathcal{E}\left( 
\mathit{\Omega }_{\mathfrak{m}}^{\sharp }\left( d_{a,-}+\mathrm{r}%
_{+}(d_{a,-})\right) \right)
\end{equation*}
\end{theorem}

\begin{proof}%
\ The first assertion (i) corresponds to Theorem \ref{theorem explosion l du
mec copacabana final}, see Section \ref{equilibirum.paragraph}. By Corollary %
\ref{corollary explosion l du mec copacabana final copy(1)}, we also observe
that 
\begin{equation*}
\left\{ \mathit{\Omega }_{\mathfrak{m}}^{\sharp }\left( d_{a,-}+\mathrm{r}%
_{+}(d_{a,-})\right) \right\} _{d_{a,-}\in \mathcal{C}_{\mathfrak{m}%
}^{\sharp }}
\end{equation*}%
is a family of disjoint subsets of $E_{1}$ which are all non--empty, convex,
and weak$^{\ast }$--compact. Using (i) and Theorem \ref{theorem structure of
omega copy(1)} (ii) we arrive at the second assertion (ii) with the set 
\begin{equation*}
\mathcal{E}\left( \mathit{\Omega }_{\mathfrak{m}}^{\sharp }\left( d_{a,-}+%
\mathrm{r}_{+}(d_{a,-})\right) \right) \neq \emptyset
\end{equation*}%
of all extreme points of (\ref{set gap}) being non--empty for any $%
d_{a,-}\in \mathcal{C}_{\mathfrak{m}}^{\sharp }$\ because of Theorem \ref%
{theorem Krein--Millman} (i). 
\end{proof}%

\begin{remark}[The set $\mathit{\hat{M}}_{\mathfrak{m}}$ for purely
repulsive/attractive models]
\label{theorem structure of omega-remark}\mbox{ }\newline
If $\Phi _{a,-}=0$ (a.e.) and $\Phi _{a,+}\neq 0$ (a.e.) then Theorem \ref%
{theorem structure of omega} reads as follows: If $\Phi _{a,-}=0$ (a.e.)
then $\mathit{\hat{M}}_{\mathfrak{m}}=\mathit{\Omega }_{\mathfrak{m}%
}^{\sharp }\left( d_{a,+}\right) $ with $d_{a,+}\in \mathcal{C}_{\mathfrak{m}%
}^{\flat }$ defined by (\ref{eq conserve strategy}), see Lemma \ref{lemma
explosion l du mec copacabana cas repulsif}. In particular, $\mathcal{E}(%
\mathit{\Omega }_{\mathfrak{m}}^{\sharp })=\mathcal{E}(\mathit{\Omega }_{%
\mathfrak{m}}^{\sharp }\left( d_{a,+}\right) )$. If$\ \Phi _{a,+}=0$ (a.e.)
then $\mathit{\Omega }_{\mathfrak{m}}^{\sharp }=\mathit{M}_{\mathfrak{m}}=%
\overline{\mathrm{co}(\mathit{\hat{M}}_{\mathfrak{m}})}$ is a closed face.
In particular,%
\begin{equation*}
\mathcal{E}(\mathit{\Omega }_{\mathfrak{m}}^{\sharp })=\underset{d_{a,-}\in 
\mathcal{C}_{\mathfrak{m}}^{\sharp }}{\cup }\mathcal{E}\left( \mathit{\Omega 
}_{\mathfrak{m}}^{\sharp }\left( d_{a,-}+\mathrm{r}_{+}(d_{a,-})\right)
\right) .
\end{equation*}%
Theorem \ref{theorem structure of omega} is less useful in this last
situation.
\end{remark}

Theorem \ref{theorem structure of omega} (ii) implies that, for any $\hat{%
\omega}\in \mathcal{E}(\mathit{\Omega }_{\mathfrak{m}}^{\sharp })$, there is 
$d_{a,-}\in \mathcal{C}_{\mathfrak{m}}^{\sharp }$ satisfying the
Euler--Lagrange equations (cf. Remark \ref{remark Euler--Lagrange equations}%
) -- or gap equations in Physics (cf. Remark \ref{remark gap eq}) -- 
\begin{equation}
d_{a}:=d_{a,-}+\mathrm{r}_{+}(d_{a,-})=e_{\Phi _{a}}(\hat{\omega})+ie_{\Phi
_{a}^{\prime }}(\hat{\omega})\mathrm{\ (a.e.)}.  \label{gap equations}
\end{equation}%
Conversely, for any $d_{a,-}\in \mathcal{C}_{\mathfrak{m}}^{\sharp }$, there
is some $\omega \in \mathit{\hat{M}}_{\mathfrak{m}}$ satisfying the
Euler--Lagrange equations but $\omega $ is \emph{not necessarily} an extreme
point of $\mathit{\Omega }_{\mathfrak{m}}^{\sharp }$. Observe, however, that
if $\omega \notin \mathcal{E}(\mathit{\Omega }_{\mathfrak{m}}^{\sharp })$
then we have a strong constraint on the set $\mathcal{C}_{\mathfrak{m}%
}^{\sharp }$:

\begin{theorem}[Gap equations for $\mathfrak{m}\in \mathcal{M}_{1}$-- II]
\label{theorem structure of omega copy(2)}\mbox{ }\newline
\index{States!generalized equilibrium!gap equations}For any $d_{a,-}\in 
\mathcal{C}_{\mathfrak{m}}^{\sharp }$ such that there exists $\omega \in 
\mathcal{E}(\mathit{\Omega }_{\mathfrak{m}}^{\sharp }(d_{a,-}+\mathrm{r}%
_{+}(d_{a,-})))$ satisfying $\omega \notin \mathcal{E}(\mathit{\Omega }_{%
\mathfrak{m}}^{\sharp })$, there is a probability measure $\mathrm{\nu }%
_{d_{a,-}}$ on $\mathcal{C}_{\mathfrak{m}}^{\sharp }$ not concentrated on $%
d_{a,-}$ such that (a.e.) 
\begin{equation*}
d_{a,-}=\int_{\mathcal{C}_{\mathfrak{m}}^{\sharp }}%
\hat{d}_{a,-}\mathrm{\ d\nu }_{d_{a,-}}(\hat{d}_{a,-})\quad \text{and}%
\mathrm{\quad r}_{+}(d_{a,-})=\int_{\mathcal{C}_{\mathfrak{m}}^{\sharp }}%
\mathrm{r}_{+}(\hat{d}_{a,-})\mathrm{\ d\nu }_{d_{a,-}}(\hat{d}_{a,-}).
\end{equation*}
\end{theorem}

\begin{proof}%
If 
\begin{equation*}
\omega \in \mathcal{E}\left( \mathit{\Omega }_{\mathfrak{m}}^{\sharp }\left(
d_{a,-}+\mathrm{r}_{+}(d_{a,-})\right) \right) \subseteq \mathit{\hat{M}}_{%
\mathfrak{m}}\subseteq \mathit{\Omega }_{\mathfrak{m}}^{\sharp }
\end{equation*}%
and $\omega \notin \mathcal{E}(\mathit{\Omega }_{\mathfrak{m}}^{\sharp })$
then, by Theorem \ref{theorem structure of omega copy(1)} (iii), there is a
probability measure $v_{\omega }$ on $\mathit{\Omega }_{\mathfrak{m}%
}^{\sharp }$ not concentrated on the convex weak$^{\ast }$--compact set $%
\mathit{\Omega }_{\mathfrak{m}}^{\sharp }\left( d_{a,-}+\mathrm{r}%
_{+}(d_{a,-})\right) $ such that%
\begin{equation}
v_{\omega }(\mathcal{E}(\mathit{\Omega }_{\mathfrak{m}}^{\sharp }))=1\mathrm{%
\quad }\text{and}\mathrm{\quad }\omega =\int_{\mathcal{E}(\mathit{\Omega }_{%
\mathfrak{m}}^{\sharp })}\mathrm{d}v_{\omega }(\hat{\omega})\;\hat{\omega}.
\label{this}
\end{equation}%
Recall that $e_{\Phi }$ is affine and weak$^{\ast }$--continuous (Lemma \ref%
{Th.en.func} (i)) and applying (\ref{this}) on the energy observable $%
\mathfrak{e}_{\Phi _{a}}+i\mathfrak{e}_{\Phi _{a}^{\prime }}$(cf. (\ref%
{eq:enpersite})) we obtain that%
\begin{equation}
d_{a,-}+\mathrm{r}_{+}(d_{a,-})=\int_{\mathcal{E}(\mathit{\Omega }_{%
\mathfrak{m}}^{\sharp })}\mathrm{d}v_{\omega }(\hat{\omega})\;\gamma
_{a}(e_{\Phi _{a}}(\hat{\omega})+ie_{\Phi _{a}^{\prime }}(\hat{\omega}))%
\mathrm{\ (a.e.)}  \label{gap equation bis}
\end{equation}%
because of Lemma \ref{Corollary 4.1.18.}. Hence, the theorem results from (%
\ref{gap equations}) and (\ref{gap equation bis}). 
\end{proof}%

Because of this last theorem we expect the equality%
\begin{equation}
\mathcal{E}(\mathit{\Omega }_{\mathfrak{m}}^{\sharp })=\underset{d_{a,-}\in 
\mathcal{C}_{\mathfrak{m}}^{\sharp }}{\cup }\mathcal{E}\left( \mathit{\Omega 
}_{\mathfrak{m}}^{\sharp }\left( d_{a,-}+\mathrm{r}_{+}(d_{a,-})\right)
\right)  \label{chouette egality}
\end{equation}%
to hold not only for purely repulsive or purely attractive models (see
Remark \ \ref{theorem structure of omega-remark}), but in a much larger
class of long--range models. In fact, for most relevant models coming from
Physics, like BCS--type models, Equality (\ref{chouette egality}) clearly
holds.

\begin{remark}[Euler--Lagrange equations]
\label{remark Euler--Lagrange equations}\mbox{ }\newline
\index{Euler--Lagrange equations|textbf}Equations (\ref{gap equations}) are
the Euler--Lagrange equations of the min--max variational problem $\mathrm{F}%
_{\mathfrak{m}}^{\sharp }$ defined in\ Definition \ref{definition
two--person zero--sum game}. We observe, however, that the pressure $P_{%
\mathfrak{m}}\left( c_{a,-}+c_{a,+}\right) $ in Definition \ref{definition
approximating free--energy} is generally not G\^{a}teau differentiable
w.r.t. either $c_{a,-}$ or $c_{a,+}$ as the variational problem (\ref%
{variational problem approx}) can have several t.i. equilibrium states (cf.
Lemma \ref{corolaire Bishop phelps1}). In fact, Theorem \ref{Mazur}%
\index{Mazur theorem} and Remark \ref{Mazur remark} only ensure the G\^{a}%
teau differentiability of the convex and continuous map $c_{a}\mapsto P_{%
\mathfrak{m}}\left( c_{a}\right) $ from $L^{2}(\mathcal{A},\mathbb{C})$ to $%
\mathbb{R}$ on a dense subset.
\end{remark}

\begin{remark}[Gap equations in Physics]
\label{remark gap eq}\mbox{ }\newline
\index{Gap equations|textbf}Equations (\ref{gap equations}) are also called
gap equations by analogy with the Bardeen--Cooper--Schrieffer (BCS) theory
for conventional superconductors \cite{BCS1,BCS2,BCS3}. Indeed, within this
theory, the existence of a non--zero solution $d_{a,-}\in \mathcal{C}_{%
\mathfrak{m}}^{\sharp }$ implies a superconducting state as well as a gap in
the spectrum of the effective (approximating) BCS Hamiltonian. The equations
satisfied by $d_{a,-}$ are called gap equations in the Physics literature
because of this property.
\end{remark}

Recall now that the integral representation (iii) in Theorem \ref{theorem
structure of omega copy(1)} may not be unique, i.e., $\mathit{\Omega }_{%
\mathfrak{m}}^{\sharp }$ may not be a Choquet simplex (Definition \ref{gamm
regularisation copy(3)}) as one may conjecture from Theorem \ref{theorem
structure of omega} (ii). For models with purely attractive long--range
interactions for which $\Phi _{a,+}=\Phi _{a,+}^{\prime }=0$ (a.e.), observe
that $\mathit{\Omega }_{\mathfrak{m}}^{\sharp }$ cannot generally be
homeomorphic to the Poulsen simplex in contrast to all sets $\{E_{%
\vec{\ell}}\}_{\vec{\ell}\in \mathbb{N}^{d}}$, see Theorem \ref{Thm Poulsen
simplex}. Indeed, the Poulsen simplex%
\index{Simplex!Poulsen} has a dense set of extreme points whereas we have
the following assertion (cf. Theorems \ref{theorem structure of omega
copy(1)} (i), \ref{theorem trivial sympa 1} (ii) and \ref{theorem trivial
sympa 1 copy(1)} (ii)):

\begin{theorem}[Density of $\mathcal{E}(\mathit{\Omega }_{\mathfrak{m}%
}^{\sharp })$ yields convexity of $\mathit{\hat{M}}_{\mathfrak{m}}$]
\label{theorem omega simplex copy(1)}\mbox{ }\newline
If the compact set $\mathit{%
\hat{M}}_{\mathfrak{m}}$ is not convex then $\mathcal{E}(\mathit{\Omega }_{%
\mathfrak{m}}^{\sharp })$ is not dense in $\mathit{\Omega }_{\mathfrak{m}%
}^{\sharp }$.
\end{theorem}

\noindent Note that the convexity of $\mathit{\hat{M}}_{\mathfrak{m}}$ is 
\emph{only} a necessary condition to obtain a dense set $\mathcal{E}(\mathit{%
\Omega }_{\mathfrak{m}}^{\sharp })$ of extreme points of $\mathit{\Omega }_{%
\mathfrak{m}}^{\sharp }$ in $\mathit{\Omega }_{\mathfrak{m}}^{\sharp }$.

The convexity of the set $\mathit{\hat{M}}_{\mathfrak{m}}$ can only be
broken by the long--range attractions $\Phi _{a,-}$ and $\Phi _{a,-}^{\prime
}$, see discussions following Lemma \ref{lemma property reduced free--energy}%
. Note further that sets of generalized t.i. equilibrium states are
simplices for purely attractive long--range models ($\Phi _{a,+}=\Phi
_{a,+}^{\prime }=0$ (a.e.)) as $\mathit{\Omega }_{\mathfrak{m}}^{\sharp }$
is a closed face of $E_{1}$ in this case, see Theorem \ref{theorem purement
repulsif sympa} ($-$). Additionally, by using Theorem \ref{theorem structure
of omega} $\mathit{\Omega }_{\mathfrak{m}}^{\sharp }$ is even a Bauer
simplex (Definition \ref{gamm regularisation copy(2)}) if the following
assumption holds:

\begin{hypothesis}
\label{Uniqueness}\mbox{ }\newline
For any $d_{a,-}\in \mathcal{C}_{\mathfrak{m}}^{\sharp }$, the set $\mathit{M%
}_{\Phi (d_{a,-}+\mathrm{r}_{+}(d_{a,-}))}$ of t.i. equilibrium states of
the approximating interaction $\Phi (d_{a,-}+\mathrm{r}_{+}(d_{a,-}))$
contains exactly one state.
\end{hypothesis}

\begin{theorem}[The set $\mathit{\Omega }_{\mathfrak{m}}^{\sharp }$ for $%
\mathfrak{m}\in \mathcal{M}_{1}$ as a simplex]
\label{theorem omega simplex}\mbox{ }\newline
\emph{(}$-$\emph{) }%
\index{States!generalized equilibrium!simplex}%
\index{Long--range models!purely attractive}If $\Phi _{a,+}=0$ (a.e.) then
the face $\mathit{\Omega }_{\mathfrak{m}}^{\sharp }$ is a Choquet simplex%
\index{Simplex!Choquet}. \newline
\emph{(}$\exists !$\emph{) }Under Hypothesis \ref{Uniqueness} $\mathit{%
\Omega }_{\mathfrak{m}}^{\sharp }$ is a face and a Bauer simplex%
\index{Simplex!Bauer}.
\end{theorem}

\begin{proof}
The first assertion is trivial. Indeed, by Theorem \ref{theorem choquet},
the set $E_{1}$ is a Choquet simplex and, by Theorem \ref{theorem choquet
bis copy(1)}, its closed faces are Choquet simplices%
\index{Simplex!Choquet}. Then the assertion ($-$) results from Theorem \ref%
{theorem purement repulsif sympa} ($-$).

Assume now that Hypothesis \ref{Uniqueness} holds. Then, as 
\begin{equation*}
\mathit{\Omega }_{\Phi (d_{a,-}+\mathrm{r}_{+}(d_{a,-}))}^{\sharp }=\mathit{M%
}_{\Phi (d_{a,-}+\mathrm{r}_{+}(d_{a,-}))}
\end{equation*}%
is a face of $E_{1}$ (Lemma \ref{remark equilibrium state approches}), its
unique element has to be ergodic and thus extreme in $\mathit{\Omega }_{%
\mathfrak{m}}^{\sharp }$. Hence, using Theorem \ref{theorem structure of
omega}, $\mathit{%
\hat{M}}_{\mathfrak{m}}\subseteq \mathcal{E}(\mathit{\Omega }_{\mathfrak{m}%
}^{\sharp })$. By Theorem \ref{theorem structure of omega copy(1)} (ii), $%
\mathcal{E}(\mathit{\Omega }_{\mathfrak{m}}^{\sharp })\subseteq \mathit{\hat{%
M}}_{\mathfrak{m}}$ and hence, $\mathcal{E}(\mathit{\Omega }_{\mathfrak{m}%
}^{\sharp })=\mathit{\hat{M}}_{\mathfrak{m}}$ is a closed set as $\mathit{%
\hat{M}}_{\mathfrak{m}}$ is weak$^{\ast }$--compact (cf. Lemma \ref{lemma
minimum sympa copy(2)} (i)). In particular, because $\mathit{\hat{M}}_{%
\mathfrak{m}}\subseteq \mathcal{E}_{1}$, $\mathit{\Omega }_{\mathfrak{m}%
}^{\sharp }$ is a closed face of $E_{1}$ and it is thus a Bauer simplex. 
\end{proof}%

If $\mathit{\Omega }_{\mathfrak{m}}^{\sharp }$ is a Bauer simplex (for
instance if Hypothesis \ref{Uniqueness} holds) then, by Theorem \ref{theorem
Bauer}, the generalized t.i. equilibrium states of $\mathfrak{m}$ can be --
affinely and homeomorphicaly -- identified with states on the commutative $%
C^{\ast }$--algebra $C(\mathcal{E}(\mathit{\Omega }_{\mathfrak{m}}^{\sharp
}))$. For instance, Hypothesis \ref{Uniqueness} is satisfied if, for any $%
d_{a,-}\in \mathcal{C}_{\mathfrak{m}}^{\sharp }$, the approximating
interaction%
\index{Interaction!approximating!temperature--dependant} 
\begin{equation*}
\Phi (d_{a,-}+\mathrm{r}_{+}(d_{a,-}))\in \mathcal{W}_{1}
\end{equation*}%
is either quadratic in the annihilation and creation operators $a_{x,\mathrm{%
s}}$, $a_{x^{\prime },\mathrm{s}^{\prime }}^{+}$ in any dimension ($d\geq 1$%
) or corresponds to a finite range one--dimensional ($d=1$) Fermi system.
These conditions hold for many relevant models coming from Physics, like
BCS--type models.

This case has also a specific interpretation in terms of game theory as $%
\mathcal{C}_{\mathfrak{m}}^{\sharp }$ (\ref{eq conserve strategy}) is the
set of conservative strategies of the attractive player of the corresponding
thermodynamic game defined by Definition \ref{definition two--person
zero--sum game}:

\begin{theorem}[Mixed conservative strategies of the attractive player]
\label{th 3.36}%
\index{Thermodynamic game!mixed conservative strategies}For any $\mathfrak{m}%
\in \mathcal{M}_{1}$ satisfying Hypothesis \ref{Uniqueness}, there is an
affine homeomorphism between $\mathit{\Omega }_{\mathfrak{m}}^{\sharp }$ and
the set of states of the commutative $C^{\ast }$--algebra $C(\mathcal{C}_{%
\mathfrak{m}}^{\sharp })$ of continuous functions on the (weakly compact)
set $\mathcal{C}_{\mathfrak{m}}^{\sharp }$. Here, the homeomorphism concerns
the weak$^{\ast }$--topologies in the sets $\mathit{\Omega }_{\mathfrak{m}%
}^{\sharp }$ and $C(\mathcal{C}_{\mathfrak{m}}^{\sharp })$.
\end{theorem}

\begin{proof}
This results is a direct consequence of Theorems \ref{theorem omega simplex}
and \ref{theorem Bauer} combined with Corollary \ref{coro 3.34 copy(1)}.%
\end{proof}%

This last result can be interpreted from the point of view of game theory as
follows. By the Riesz--Markov theorem, the set of states on $C(\mathcal{C}_{%
\mathfrak{m}}^{\sharp })$ is the same as the set of probability measures on
the set $\mathcal{C}_{\mathfrak{m}}^{\sharp }$ of conservative strategies of
the attractive player. As discussed above, the best the attractive player
can do -- as she/he has no access to the choice of strategy of the repulsive
one -- is to choose some conservative strategy in order to minimize her/his
loss in the game. She/he could also do this in a non--deterministic way.
I.e., she/he determines with which probability distribution the different
conservative strategies have to be chosen. This kind of procedure is called 
\emph{mixed strategy}%
\index{Zero--sum games!mixed strategy} in game theory. Hence, the set of all
generalized t.i. equilibrium states is -- in the situation of Theorem \ref%
{th 3.36} above -- (even affinely) the same as the set of all mixed
conservative strategies of the attractive player of the thermodynamic game.

Now, we observe that Theorem \ref{theorem saddle point} ($\sharp $) tell us
that the conservative value $\mathrm{F}_{\mathfrak{m}}^{\sharp }$ for the
thermodynamic game defined in Definition \ref{definition two--person
zero--sum game} leads to the pressure $\mathrm{P}_{\mathfrak{m}}^{\sharp }$
(up to a minus sign) for any model $\mathfrak{m}\in \mathcal{M}_{1}$. In
other words, the approximating Hamiltonian method \cite%
{Bogjunior,approx-hamil-method0,approx-hamil-method,approx-hamil-method2}
(see Section \ref{Section approx method}) extended to all $\mathfrak{m}\in 
\mathcal{M}_{1}$ is still an efficient technique to obtain the pressure. On
the other hand, the min--max variational problem $\mathrm{F}_{\mathfrak{m}%
}^{\sharp }$ is related via (\ref{pression approximated}) and (\ref%
{variational problem approx}) to the family%
\index{Interaction!approximating!temperature--dependant} 
\begin{equation*}
\{\Phi (d_{a,-}+\mathrm{r}_{+}(d_{a,-}))\}_{d_{a,-}\in \mathcal{C}_{%
\mathfrak{m}}^{\sharp }}
\end{equation*}%
of approximating interactions (Definition \ref{definition BCS-type model
approximated}) with $\mathrm{r}_{+}\in \mathrm{C}\left(
L_{-}^{2},L_{+}^{2}\right) $ defined by (\ref{thermodyn decision rule}).
Therefore, for any model $\mathfrak{m}\in \mathcal{M}_{1}$, one could, a
priori, think that the weak$^{\ast }$--closed convex hull of the union of
the family 
\begin{equation*}
\{\mathit{M}_{\Phi \left( d_{a,-}+\mathrm{r}_{+}(d_{a,-})\right)
}\}_{d_{a,-}\in \mathcal{C}_{\mathfrak{m}}^{\sharp }}
\end{equation*}%
of sets of t.i. equilibrium states (cf. Lemma \ref{remark equilibrium state
approches}) equals the set $\mathit{\Omega }_{\mathfrak{m}}^{\sharp }$ of
generalized t.i. equilibrium states. This fact is generally wrong, i.e., the
approximating Hamiltonian method does not generally lead to an \emph{%
effective local theory}.

To explain this, we define more precisely the notion of \emph{theory} as
follows:

\begin{definition}[Theory for $\mathfrak{m}\in \mathcal{M}_{1}$]
\label{definition theory}\mbox{ }\newline
\index{Theory}A theory for $\mathfrak{m}\in \mathcal{M}_{1}$ is any subset $%
\mathfrak{T}_{\mathfrak{m}}\subseteq \mathcal{M}_{1}$.
\end{definition}

\noindent Of course, a good theory $\mathfrak{T}_{\mathfrak{m}}$ for $%
\mathfrak{m}\in \mathcal{M}_{1}$ means that elements of $\mathfrak{T}_{%
\mathfrak{m}}$ are simplified models in comparison with $\mathfrak{m}\in 
\mathcal{M}_{1}$ and that it allows the complete description of the set $%
\mathit{\Omega }_{\mathfrak{m}}^{\sharp }$ of generalized t.i. equilibrium
states. This last property corresponds to have an \emph{effective} theory in
the following sense:

\begin{definition}[Effective theory]
\label{definition effective theories}\mbox{ }\newline
\index{Theory!effective}A theory $\mathfrak{T}_{\mathfrak{m}}$ for $%
\mathfrak{m}\in \mathcal{M}_{1}$ is said to be effective at $\beta \in
(0,\infty )$ iff 
\begin{equation*}
\overline{\mathrm{co}%
\big(%
\underset{\mathfrak{\hat{m}}\in \mathfrak{T}_{\mathfrak{m}}}{\cup }\mathit{%
\Omega }_{\mathfrak{\hat{m}}}^{\sharp }%
\big)%
}=\mathit{\Omega }_{\mathfrak{m}}^{\sharp }\quad \text{and}\quad \mathcal{E}(%
\mathit{\Omega }_{\mathfrak{m}}^{\sharp })\subseteq \underset{\mathfrak{\hat{%
m}}\in \mathfrak{T}_{\mathfrak{m}}}{\cup }\mathit{\Omega }_{\mathfrak{\hat{m}%
}}^{\sharp }.
\end{equation*}%
The closure is taken in the weak$^{\ast }$--topology and $\mathrm{co}(M)$
denotes as usual the convex hull of a set $M\subseteq \mathcal{U}^{\ast }$.
\end{definition}

\noindent The second\ condition in the above definition means that any pure
generalized equilibrium state of $\mathfrak{m}$ should be a generalized t.i.
equilibrium state of $\mathfrak{\hat{m}}$ for some $\mathfrak{\hat{m}}\in 
\mathfrak{T}_{\mathfrak{m}}$ in the theory $\mathfrak{T}_{\mathfrak{m}}$. By
Theorem \ref{lemma Milman} (ii) (Milman theorem), this holds if the union $%
\underset{\mathfrak{\hat{m}}\in \mathfrak{T}_{\mathfrak{m}}}{\cup }\mathit{%
\Omega }_{\mathfrak{\hat{m}}}^{\sharp }$ is closed w.r.t. the weak$^{\ast }$%
--topology. This is the case in the examples of effective theories discussed
here. Two general classes of theories are of particular importance w.r.t.
models $\mathfrak{m}\in \mathcal{M}_{1}$: The \emph{repulsive} and \emph{%
local} theories defined below.

\begin{definition}[Repulsive theory]
\label{definition repulsive theory}\mbox{ }\newline
\index{Theory!repulsive}For $\mathfrak{m}\in \mathcal{M}_{1}$, a theory $%
\mathfrak{T}_{\mathfrak{m}}$ is said to be repulsive iff the subset $%
\mathfrak{T}_{\mathfrak{m}}\subseteq \mathcal{M}_{1}$ has only models with
purely repulsive long--range interactions, i.e., models for which $\Phi
_{a,-}=\Phi _{a,-}^{\prime }=0$ (a.e.), see Definition \ref{long range
attraction-repulsion}.
\end{definition}

\noindent An example of repulsive theory is given by using partially the
approximating Hamiltonian method: For any model $\mathfrak{m}\in \mathcal{M}%
_{1}$ and all $c_{a,-}\in L_{-}^{2}(\mathcal{A},\mathbb{C})$, we define the
approximating repulsive model 
\begin{equation}
\mathfrak{m}\left( c_{a,-}\right) :=(\Phi \left( c_{a,-}\right) ,\{\Phi
_{a,+}\}_{a\in \mathcal{A}},\{\Phi _{a,+}^{\prime }\}_{a\in \mathcal{A}})\in 
\mathcal{M}_{1}.  \label{model purement repulsif}
\end{equation}%
Here, $\Phi _{a,+}:=\gamma _{a,+}\Phi _{a}$ and $\Phi _{a,+}^{\prime
}:=\gamma _{a,+}\Phi _{a}^{\prime }$ (cf. Definition \ref{long range
attraction-repulsion}), whereas $\Phi \left( c_{a,-}\right) $ is defined in
Definition \ref{definition BCS-type model approximated}. Since $\mathfrak{m}%
\left( c_{a,-}\right) $ is a model with purely repulsive long--range
interactions for all $c_{a,-}\in L_{-}^{2}(\mathcal{A},\mathbb{C})$, it can
be used to define a repulsive theory as follows:

\begin{definition}[The min repulsive theory]
\label{definition min repulsive theory}\mbox{ }\newline
\index{Theory!repulsive!min}At $\beta \in (0,\infty )$, the min repulsive
theory for $\mathfrak{m}\in \mathcal{M}_{1}$ is the subset%
\begin{equation*}
\mathfrak{T}_{\mathfrak{m}}^{+}:=\underset{d_{a,-}\in \mathcal{C}_{\mathfrak{%
m}}^{\sharp }}{\cup }\mathfrak{m}\left( d_{a,-}\right) \subseteq \mathcal{M}%
_{1}
\end{equation*}%
with the set $\mathcal{C}_{\mathfrak{m}}^{\sharp }$ of conservative
strategies of the attractive player defined by (\ref{eq conserve strategy}).
\end{definition}

\noindent Observe that $\mathfrak{m}\left( d_{a,-}\right) $ has a local
(effective) interaction $\Phi (d_{a,-})$ non--trivially depending on the
inverse temperature $\beta >0$ of the system (cf. Remark \ref{Temperature of
Fermi systems}). In other words, the min repulsive theory $\mathfrak{T}_{%
\mathfrak{m}}^{+}$ is \emph{temperature--dependent}.%
\index{Interaction!approximating!temperature--dependant}

\emph{Local} theories are made of subsets of the real Banach space $\mathcal{%
W}_{1}$ of t.i. interactions $\Phi $, see Definition \ref{definition banach
space interaction}.

\begin{definition}[Local theories]
\label{definition local theory}\mbox{ }\newline
\index{Theory!local}A theory $\mathfrak{T}_{\mathfrak{m}}$ for $\mathfrak{m}%
\in \mathcal{M}_{1}$ is said to be local iff $\mathfrak{T}_{\mathfrak{m}%
}\subseteq \mathcal{W}_{1}$, where $\mathcal{W}_{1}$ is seen as a sub--space
of $\mathcal{M}_{1}$.
\end{definition}

\noindent The min--max variational problem $\mathrm{F}_{\mathfrak{m}%
}^{\sharp }$ of the thermodynamic game defined by Definition \ref{definition
two--person zero--sum game} leads to an important example of local theories:
The \emph{min--max} local theory, which is also a \emph{%
temperature--dependent} theory.

\begin{definition}[The min--max local theory]
\label{definition effective theories bogo}\mbox{ }\newline
\index{Theory!local!min--max}At $\beta \in (0,\infty )$, the min--max local
theory for $\mathfrak{m}\in \mathcal{M}_{1}$ is the subset%
\begin{equation*}
\mathfrak{T}_{\mathfrak{m}}^{\sharp }:=\underset{d_{a,-}\in \mathcal{C}_{%
\mathfrak{m}}^{\sharp }}{\cup }\Phi (d_{a,-}+\mathrm{r}_{+}(d_{a,-}))%
\subseteq \mathcal{W}_{1},
\end{equation*}%
where the set $\mathcal{C}_{\mathfrak{m}}^{\sharp }$ is defined by (\ref{eq
conserve strategy}) and the map $\mathrm{r}_{+}$ by (\ref{thermodyn decision
rule}).
\end{definition}

To get an effective local theory $\mathfrak{T}_{\mathfrak{m}}$ for a model $%
\mathfrak{m}\in \mathcal{M}_{1}$, the set $\mathit{\Omega }_{\mathfrak{m}%
}^{\sharp }$ of generalized t.i. equilibrium states must be a face. It is a
necessary condition as the weak$^{\ast }$--closed convex hull of the union 
\begin{equation*}
\underset{\Phi \in \mathfrak{T}_{\mathfrak{m}}}{\cup }\mathit{\Omega }_{\Phi
}=\underset{\Phi \in \mathfrak{T}_{\mathfrak{m}}}{\cup }\mathit{M}_{\Phi }
\end{equation*}%
of faces in $E_{1}$ is again a face in $E_{1}$ if 
\begin{equation*}
\overline{\mathrm{co}%
\big(%
\underset{\Phi \in \mathfrak{T}_{\mathfrak{m}}}{\cup }\mathit{\Omega }_{\Phi
}^{\sharp }%
\big)%
}=\mathit{\Omega }_{\mathfrak{m}}^{\sharp }\quad \text{and}\quad \mathcal{E}(%
\mathit{\Omega }_{\mathfrak{m}}^{\sharp })\subseteq \underset{\Phi \in 
\mathfrak{T}_{\mathfrak{m}}}{\cup }\mathit{\Omega }_{\Phi }^{\sharp }.
\end{equation*}%
Indeed, for all $\Phi \in \mathcal{W}_{1}$, the set $\mathit{M}_{\Phi }=%
\mathit{\Omega }_{\Phi }$ is a face by weak$^{\ast }$--lower
semi--continuity and affinity of the functional $f_{\Phi }$, see Lemmata \ref%
{lemma property entropy} (i), \ref{Th.en.func} (i) and Definition \ref%
{Remark free energy density}. Lemma \ref{lemma explosion l du mec
copacabana2} says that $\mathit{\Omega }_{\mathfrak{m}}^{\sharp }$ is
generally not a face in $E_{1}$. As a consequence, we obtain the following
result:

\begin{theorem}[Breakdown of effective local theories]
\label{lemma explosion l homme capacabana1}\mbox{ }\newline
\index{Theory!local!breakdown}At fixed $\beta \in (0,\infty )$, there are
uncountably many $\mathfrak{m}\in \mathcal{M}_{1}$ with no effective local
theory.
\end{theorem}

\noindent In particular, the equality $\mathrm{P}_{\mathfrak{m}}^{\sharp }=-%
\mathrm{F}_{\mathfrak{m}}^{\sharp }$ of Theorem \ref{theorem saddle point} ($%
\sharp $) does not necessarily imply that the min--max local theory $%
\mathfrak{T}_{\mathfrak{m}}^{\sharp }$ (Definition \ref{definition effective
theories bogo}) is an effective theory, see Definition \ref{definition
effective theories}. By contrast, the min repulsive theory (Definition \ref%
{definition min repulsive theory}) is always an effective theory:

\begin{theorem}[Effectiveness of the min repulsive theory $\mathfrak{T}_{%
\mathfrak{m}}^{+}$]
\label{lemma explosion l homme capacabana1 copy(1)}\mbox{ }\newline
\index{Theory!repulsive!min}$\mathfrak{T}_{\mathfrak{m}}^{+}$ is an
effective repulsive theory for any $\mathfrak{m}\in \mathcal{M}_{1}$, i.e., 
\begin{equation*}
\overline{\mathrm{co}%
\big(%
\underset{d_{a,-}\in \mathcal{C}_{\mathfrak{m}}^{\sharp }}{\cup }\mathit{%
\Omega }_{\mathfrak{m}\left( d_{a,-}\right) }^{\sharp }%
\big)%
}=\mathit{\Omega }_{\mathfrak{m}}^{\sharp }\quad \text{and}\quad \mathcal{E}(%
\mathit{\Omega }_{\mathfrak{m}}^{\sharp })\subseteq \underset{d_{a,-}\in 
\mathcal{C}_{\mathfrak{m}}^{\sharp }}{\cup }\mathit{\Omega }_{\mathfrak{m}%
\left( d_{a,-}\right) }^{\sharp }.
\end{equation*}
\end{theorem}

\begin{proof}
This follows from Lemmata \ref{lemma explosion l du mec copacabana1 copy(2)}
and \ref{lemma explosion l du mec copacabana cas repulsif} which yield in
particular the equality 
\begin{equation*}
\mathit{\Omega }_{\mathfrak{m}\left( d_{a,-}\right) }^{\sharp }=\mathit{%
\Omega }_{\mathfrak{m}}^{\sharp }\left( d_{a,-}+\mathrm{r}%
_{+}(d_{a,-})\right)
\end{equation*}%
for all $d_{a,-}\in \mathcal{C}_{\mathfrak{m}}^{\sharp }$. See also Theorems %
\ref{theorem structure of omega copy(1)} (i) and \ref{theorem structure of
omega} (i).%
\end{proof}%

Therefore, the breakdown of effective local theories results from
long--range repulsions $\Phi _{a,+},\Phi _{a,+}^{\prime }$ and not from
long--range attractions $\Phi _{a,-},\Phi _{a,-}^{\prime }$, see Definition %
\ref{long range attraction-repulsion}. This is another strong asymmetry
between both long--range interactions. To illustrate this, observe that for
models $\mathfrak{m}$ with $\Phi _{a,+}=\Phi _{a,+}^{\prime }=0$ (a.e.), the
min repulsive and the min--max local theories are the same, i.e., $\mathfrak{%
T}_{\mathfrak{m}}^{+}=\mathfrak{T}_{\mathfrak{m}}^{\sharp }$, see
Definitions \ref{definition min repulsive theory} and \ref{definition
effective theories bogo} together with Definition \ref{definition BCS-type
model approximated} and (\ref{model purement repulsif}). In this purely
attractive case, for all $d_{a,-}\in \mathcal{C}_{\mathfrak{m}}^{\sharp }$, $%
\mathit{\Omega }_{\mathfrak{m}\left( d_{a,-}\right) }^{\sharp }=\mathit{M}%
_{\Phi \left( d_{a,-}\right) }$ is always a face in $E_{1}$ and so is the
set $\mathit{\Omega }_{\mathfrak{m}}^{\sharp }$ by Theorem \ref{lemma
explosion l homme capacabana1 copy(1)}. In other words, if the long--range
repulsions $\Phi _{a,+}$ and $\Phi _{a,+}^{\prime }$ are switched off, there
is always an effective local theory.

In the general case, the min--max local theory $\mathfrak{T}_{\mathfrak{m}%
}^{\sharp }$ (Definition \ref{definition effective theories bogo}) is not
accurate enough. It means that the set $\mathit{\Omega }_{\mathfrak{m}%
}^{\sharp }$ of generalized t.i. equilibrium states is only included in (but
generally not equal to) the weak$^{\ast }$--closed convex hull of the set 
\begin{equation}
\mathit{M}(\mathfrak{T}_{\mathfrak{m}}^{\sharp }):=\underset{d_{a,-}\in 
\mathcal{C}_{\mathfrak{m}}^{\sharp }}{\cup }\mathit{M}_{\Phi (d_{a,-}+%
\mathrm{r}_{+}(d_{a,-}))}^{\sharp }.
\label{union equilibrium states min max theory}
\end{equation}%
This result is a simple corollary of Theorems \ref{theorem structure of
omega copy(1)} (i) and \ref{theorem structure of omega} (i):

\begin{corollary}[Accuracy of the min--max local theory $\mathfrak{T}_{%
\mathfrak{m}}^{\sharp }$]
\label{lemma explosion l homme capacabana1 copy(2)}\mbox{ }\newline
For any $\mathfrak{m}\in \mathcal{M}_{1}$,%
\index{Theory!local!min--max} 
\begin{equation*}
\mathit{\Omega }_{\mathfrak{m}}^{\sharp }\subseteq 
\overline{\mathrm{co}%
\big(%
\mathit{M}(\mathfrak{T}_{\mathfrak{m}}^{\sharp })%
\big)%
}.
\end{equation*}
\end{corollary}

\noindent Indeed, by Theorems \ref{theorem structure of omega copy(1)} (i)
and \ref{theorem structure of omega} (i), $\mathit{\Omega }_{\mathfrak{m}%
}^{\sharp }$ is the weak$^{\ast }$--closed convex hull of the set of states
in $\mathit{M}(\mathfrak{T}_{\mathfrak{m}}^{\sharp })$ satisfying the
Euler--Lagrange equations (\ref{gap equations}).

\begin{remark}
If $\mathit{\Omega }_{\mathfrak{m}}^{\sharp }$ is not a face then there is,
at least, one ergodic state\footnote{%
Note that $\mathit{M}(\mathfrak{T}_{\mathfrak{m}}^{\sharp })\cap \mathcal{E}%
_{1}\neq \emptyset $ because $\mathit{M}(\mathfrak{T}_{\mathfrak{m}}^{\sharp
})$ is a union of non--empty closed faces by (\ref{union equilibrium states
min max theory}), see also Lemma \ref{remark equilibrium state approches}.} $%
\hat{\omega}\in \mathit{M}(\mathfrak{T}_{\mathfrak{m}}^{\sharp })\cap 
\mathcal{E}_{1}$ which does not satisfy the Euler--Lagrange equations (\ref%
{gap equations}).
\end{remark}

\begin{remark}[Max attractive theory $\mathfrak{T}_{\mathfrak{m}}^{-}$ and
max--min local theory $\mathfrak{T}_{\mathfrak{m}}^{\flat }$]
In the same way 
\index{Theory!local!max--min}%
\index{Theory!max attractive}we define the min repulsive theory $\mathfrak{T}%
_{\mathfrak{m}}^{+}$ (Definition \ref{definition min repulsive theory}) and
the min--max local theory $\mathfrak{T}_{\mathfrak{m}}^{\sharp }$
(Definition \ref{definition effective theories bogo}), one could define the
max attractive theory $\mathfrak{T}_{\mathfrak{m}}^{-}$ and the max--min
local theory $\mathfrak{T}_{\mathfrak{m}}^{\flat }$ for any $\mathfrak{m}\in 
\mathcal{M}_{1}$. In the same way we have Theorems \ref{theorem structure of
omega copy(1)}, \ref{theorem structure of omega} and \ref{lemma explosion l
homme capacabana1 copy(1)}, such theories $\mathfrak{T}_{\mathfrak{m}}^{-}$
and $\mathfrak{T}_{\mathfrak{m}}^{\flat }$ shall be related to the
(non--empty) set $\mathit{M}_{\mathfrak{m}}^{\flat }$ of minimizers of the
functional $f_{\mathfrak{m}}^{\flat }$ over $E_{1}$, see (\ref{convex
functional g_m}), (\ref{pressure bemol}) and Theorem \ref{theorem saddle
point} ($\flat $).
\end{remark}

\section{Long--range interactions and long--range order (LRO)\label{Section
ODLRO}%
\index{Long--range order|textbf}}

The solution $d_{a}\in L^{2}(\mathcal{A},\mathbb{C})$ defined by (\ref{gap
equations}) has a direct interpretation as the mean energy density of
long--range interactions $\Phi _{a}$ and $\Phi _{a}^{\prime }$. Moreover, it
is related to the so--called \emph{long--range order} (LRO) property. In
particular, models with non--zero $d_{a,-}\in \mathcal{C}_{\mathfrak{m}%
}^{\sharp }$ show an \emph{off diagonal long--range order} (ODLRO), a
property proposed by Yang \cite{ODLRO} to define super--conducting phases.
The latter can be seen as a consequence of the following theorem:

\begin{theorem}[Off diagonal long--range order]
\label{corollary ODLRO}\mbox{ }\newline
\index{Long--range order (LRO)!off diagonal (ODLRO)|textbf}For any $c_{a}\in
L^{2}(\mathcal{A},\mathbb{C})$, let $B_{c_{a}}:=\left\langle \mathfrak{e}%
_{\Phi _{a}}+i\mathfrak{e}_{\Phi _{a}^{\prime }},\gamma
_{a}c_{a}\right\rangle $. Then, for any $\omega \in \mathit{\Omega }_{%
\mathfrak{m}}^{\sharp }$,%
\begin{equation*}
\Delta _{B_{c_{a}}}\left( \omega \right) :=\lim\limits_{L\rightarrow \infty }%
\frac{1}{|\Lambda _{L}|^{2}}\sum\limits_{x,y\in \Lambda _{L}}\omega \left(
\alpha _{x}(B_{c_{a}}^{\ast })\alpha _{y}(B_{c_{a}})\right)
\end{equation*}%
satisfies the inequality 
\begin{equation*}
\Delta _{B_{c_{a}}}\left( \omega \right) \geq \underset{d_{a,-}\in \mathcal{C%
}_{\mathfrak{m}}^{\sharp }}{\min }\{\left\vert \left\langle d_{a,-}+\mathrm{r%
}_{+}(d_{a,-}),\gamma _{a}c_{a}\right\rangle \right\vert ^{2}\}.
\end{equation*}
\end{theorem}

\begin{proof}
By Definition \ref{definition de deltabis}, Remark \ref{Remark lower bound
delta}, and Theorem \ref{theorem structure of omega} (ii), for any extreme
state $\hat{\omega}\in \mathcal{E}(\mathit{\Omega }_{\mathfrak{m}}^{\sharp
}) $, there is $d_{a,-}\in \mathcal{C}_{\mathfrak{m}}^{\sharp }$ such that 
\begin{equation*}
\Delta _{B_{c_{a}}}\left( \hat{\omega}\right) \geq \left\vert \hat{\omega}%
\left( B_{c_{a}}\right) \right\vert ^{2}=\left\vert \left\langle
d_{a},\gamma _{a}c_{a}\right\rangle \right\vert ^{2}
\end{equation*}%
with $d_{a}:=d_{a,-}+\mathrm{r}_{+}(d_{a,-})$. Then via Theorem \ref{theorem
structure of omega copy(1)} (iii) combined with Lemma \ref{Corollary 4.1.18.}
one gets the assertion. 
\end{proof}%

\begin{remark}
By using similar arguments as above, if all extreme generalized t.i.
equilibrium states $\hat{\omega}\in \mathcal{E}(\mathit{\Omega }_{\mathfrak{m%
}}^{\sharp })$ are strongly mixing (see (\ref{mixing})) then%
\begin{equation*}
\Delta _{B_{c_{a}}}\left( \omega \right) =\lim\limits_{|y-x|\rightarrow
\infty }\omega \left( \alpha _{x}(B_{c_{a}}^{\ast })\alpha
_{y}(B_{c_{a}})\right) \geq \underset{d_{a,-}\in \mathcal{C}_{\mathfrak{m}%
}^{\sharp }}{\min }\{\left\vert \left\langle d_{a,-}+\mathrm{r}%
_{+}(d_{a,-}),\gamma _{a}c_{a}\right\rangle \right\vert ^{2}\}
\end{equation*}%
for all $\omega \in \mathit{\Omega }_{\mathfrak{m}}^{\sharp }$.
\end{remark}

Theorem \ref{corollary ODLRO} implies ODLRO in the following sense. Take any 
\index{Gauge invariant!models}gauge invariant model $\mathfrak{m}\in 
\mathcal{M}_{1}$ -- which means that $U_{l}\in \mathcal{U}^{\circ }$ (cf. (%
\ref{definition of gauge invariant operators})) -- such that its set $%
\mathit{\Omega }_{\mathfrak{m}}^{\sharp }$ of generalized t.i. equilibrium
states contains at least one state from $\mathcal{E}(E_{1}^{\circ })$, i.e., 
$\mathit{\Omega }_{\mathfrak{m}}^{\sharp }\cap \mathcal{E}(E_{1}^{\circ
})\not=\emptyset $. This is the case, for instance, if the long--range
interactions of $\mathfrak{m}$ are purely attractive%
\index{Long--range models!attractions} (i.e., $\Phi _{a,+}=\Phi
_{a,+}^{\prime }=0$ (a.e.)) as, in this situation, $\mathit{\Omega }_{%
\mathfrak{m}}^{\sharp ,\circ }:=\mathit{\Omega }_{\mathfrak{m}}^{\sharp
}\cap E_{1}^{\circ }$ is a face in $E_{1}^{\circ }$, see also Remark \ref%
{remark.eq.inv.gauge}. Suppose that $c_{a}$ is chosen such that\footnote{%
Both assumption can easily be verified in various long--range gauge
invariant models, see, e.g., \cite{BruPedra1}.}%
\begin{equation*}
\sigma ^{\circ }(B_{c_{a}})=0\quad \mathrm{and}\quad \underset{d_{a,-}\in 
\mathcal{C}_{\mathfrak{m}}^{\sharp }}{\min }\{\left\vert \left\langle
d_{a,-}+\mathrm{r}_{+}(d_{a,-}),\gamma _{a}c_{a}\right\rangle \right\vert
^{2}\}>0,
\end{equation*}%
see Remark \ref{proj.gauge.inv}. Choose now any 
\index{Gauge invariant!equilibrium states}gauge invariant t.i. equilibrium
state $%
\hat{\omega}\in \mathit{\Omega }_{\mathfrak{m}}^{\sharp ,\circ }$, which is
extreme in the set $E_{1}^{\circ }$ of t.i. and gauge invariant states (cf.
Remarks \ref{t.i. + gauge inv states}, \ref{t.i. + gauge inv states
ergodicity}, and \ref{remark.eq.inv.gauge}). Then, by the assumptions, for
all $A^{\circ }\in \mathcal{U}^{\circ }$ such that $\hat{\omega}(A^{\circ
})=0$, 
\begin{equation*}
\lim\limits_{L\rightarrow \infty }\frac{1}{|\Lambda _{L}|^{2}}%
\sum\limits_{x,y\in \Lambda _{L}}\hat{\omega}(\alpha _{x}(A^{\circ })^{\ast
}\alpha _{y}(A^{\circ }))=\left\vert \hat{\omega}\left( A^{\circ }\right)
\right\vert ^{2}=0.
\end{equation*}%
However, 
\begin{equation*}
\lim\limits_{L\rightarrow \infty }\frac{1}{|\Lambda _{L}|^{2}}%
\sum\limits_{x,y\in \Lambda _{L}}\hat{\omega}(\alpha _{x}(B_{c_{a}})^{\ast
}\alpha _{y}(B_{c_{a}}))>0
\end{equation*}%
in spite of the fact that $\hat{\omega}(B_{c_{a}})=\hat{\omega}\circ \sigma
^{\circ }(B_{c_{a}})=0$. Indeed, any quadratic element $A^{\circ
}=A_{1}A_{2}\in \mathcal{U}^{\circ }$ with $A_{1},A_{2}\in \mathcal{U}%
^{\circ }$ is called \textquotedblleft diagonal\textquotedblright ,\ whereas
elements of the form $A^{\circ }=A_{1}A_{2}\in \mathcal{U}^{\circ }$ with $%
A_{1},A_{2}\in \mathcal{U}\backslash \mathcal{U}^{\circ }$ -- as, for
instance, the elements $\alpha _{x}(B_{c_{a}})^{\ast }\alpha _{y}(B_{c_{a}})$
considered above -- are called \textquotedblleft
off--diagonal\textquotedblright\ w.r.t. the algebra $\mathcal{U}^{\circ }$,
see, e.g., \cite[Section 5.2]{Sewell}.%
\index{Long--range order (LRO)!off diagonal (ODLRO)}

In the general case, the order parameter $d_{a,-}$ is, a priori, not unique
since the non--empty set $\mathcal{C}_{\mathfrak{m}}^{\sharp }$ (\ref{eq
conserve strategy}) of conservative strategies of the attractive player is
only weakly compact, see Lemma \ref{lemma idiot interaction approx 2} ($%
\sharp $). Non--uniqueness of solutions of the min--max variational problem $%
\mathrm{F}_{\mathfrak{m}}^{\sharp }$ of the thermodynamic game defined in
Definition \ref{definition two--person zero--sum game} ensures the existence
of a non--zero $d_{a,-}\in \mathcal{C}_{\mathfrak{m}}^{\sharp }$ which
should be related to ODLRO as explained above. In particular, ODLRO w.r.t.
elements of the form $B_{c_{a}}$ as defined above is usually related to
long--range attractions $\Phi _{a,-},\Phi _{a,-}^{\prime }$ (Definition \ref%
{long range attraction-repulsion}). As an example, we recommend to have a
look on the strong coupling BCS--Hubbard model analyzed in \cite{BruPedra1}.

By contrast, the solution $\mathrm{r}_{+}(d_{a,-})\in \mathcal{C}_{\mathfrak{%
m}}^{\sharp }\left( d_{a,-}\right) $ of the variational problem $\mathfrak{f}%
_{\mathfrak{m}}^{\sharp }\left( d_{a,-}\right) $ defined in Definition \ref%
{definition two--person zero--sum game} is always unique, see Lemma \ref%
{lemma idiot interaction approx 2 copy(2)} ($\sharp $). In particular, if
the model $\mathfrak{m}\in \mathcal{M}_{1}$ has purely repulsive long--range
interactions%
\index{Long--range models!repulsions}, i.e., $\Phi _{a,-}=\Phi
_{a,-}^{\prime }=0$ (a.e.), then no first order phase transition (related to
observables of the form $B_{c_{a}}$) can appear. If, additionally, $%
\mathfrak{m}$ is also gauge invariant -- which means that $U_{l}\in \mathcal{%
U}^{\circ }$ -- then%
\begin{equation*}
d_{a,+}=\omega (\mathfrak{e}_{\Phi _{a}}+i\mathfrak{e}_{\Phi _{a}^{\prime
}})=\omega \left( \sigma ^{\circ }(\mathfrak{e}_{\Phi _{a}}+i\mathfrak{e}%
_{\Phi _{a}^{\prime }})\right)
\end{equation*}%
for all $\omega \in \mathit{\Omega }_{\mathfrak{m}}^{\sharp }$, see again (%
\ref{definition of gauge invariant operators}) and Remark \ref%
{proj.gauge.inv} for definitions of the set $\mathcal{U}^{\circ }$ of all
gauge invariant elements and the gauge invariant projection $\sigma ^{\circ
} $ respectively. In particular, for all $a\in \mathcal{A}$ such that $%
\sigma ^{\circ }(\mathfrak{e}_{\Phi _{a}}+i\mathfrak{e}_{\Phi _{a}^{\prime
}})=0$, the unique $d_{a,+}$ must be zero.

However, the existence of a non--zero order parameter $d_{a,-}$ is, a
priori, \emph{not necessary} to get LRO:

\begin{theorem}[Long--Range Order]
\label{thm LRO}%
\index{Long--range order (LRO)}\mbox{ }\newline
Let $\mathfrak{m}\in \mathcal{M}_{1}$ such that $\Phi _{a,-}=\Phi
_{a,-}^{\prime }=0$ (a.e.). For any $c_{a}\in L^{2}(\mathcal{A},\mathbb{C})$%
, let $B_{c_{a}}:=\left\langle \mathfrak{e}_{\Phi _{a}}+i\mathfrak{e}_{\Phi
_{a}^{\prime }},\gamma _{a}c_{a}\right\rangle $. \newline
\emph{(i)} Assume that $\mathit{\Omega }_{\mathfrak{m}}$ is a face in $E_{1}$%
. If $\sigma ^{\circ }(B_{c_{a}})=0$ and $\mathfrak{m}$ is a gauge invariant
model, i.e., $U_{l}\in \mathcal{U}^{\circ }$ for all $l\in \mathbb{N}$, then 
$\Delta _{B_{c_{a}}}\left( \omega \right) =0$ for all $\omega \in \mathit{%
\Omega }_{\mathfrak{m}}^{\sharp ,\circ }:=\mathit{\Omega }_{\mathfrak{m}%
}^{\sharp }\cap E_{1}^{\circ }$. \newline
\emph{(ii)} Assume that $\mathit{\Omega }_{\mathfrak{m}}$ is not a face in $%
E_{1}$. Then, there is $c_{a}\in L^{2}(\mathcal{A},\mathbb{C})$ and $\omega
_{0}\in \mathcal{E}(\mathit{\Omega }_{\mathfrak{m}}^{\sharp })$ such that $%
\omega _{0}\notin \mathcal{E}_{1}$ and $\Delta _{B_{c_{a}}}\left( \omega
_{0}\right) >0$.
\end{theorem}

\begin{proof}
Fix all parameters of the theorem. Assume that $\mathit{\Omega }_{\mathfrak{m%
}}$ is a face in $E_{1}$, i.e., $\mathcal{E}(\mathit{\Omega }_{\mathfrak{m}%
}^{\sharp })=\mathit{\Omega }_{\mathfrak{m}}^{\sharp }\cap \mathcal{E}_{1}$.
Then, by the uniqueness of solution $d_{a,+}$ of the variational problem $%
\mathfrak{f}_{\mathfrak{m}}^{\sharp }\left( c_{a,-}\right) $ (Definition \ref%
{definition two--person zero--sum game}) combined with Theorem \ref%
{Lemma1.vonN} (iv), and Theorem \ref{theorem structure of omega} (ii), we
obtain that%
\begin{equation}
\Delta _{B_{c_{a}}}\left( \omega \right) :=\lim\limits_{L\rightarrow \infty }%
\frac{1}{|\Lambda _{L}|^{2}}\sum\limits_{x,y\in \Lambda _{L}}\omega (\alpha
_{x}(B_{c_{a}})^{\ast }\alpha _{y}(B_{c_{a}}))=|\omega (B_{c_{a}})|^{2}
\label{ODLRO false}
\end{equation}%
for all $\omega \in \mathit{\Omega }_{\mathfrak{m}}^{\sharp }$ and with $%
B_{c_{a}}\in \mathcal{U}$ defined as above. In particular, the condition $%
\sigma ^{\circ }(B_{c_{a}})=0$ implies for $\omega \in \mathit{\Omega }_{%
\mathfrak{m}}^{\sharp ,\circ }$ (cf. Remark \ref{remark.eq.inv.gauge}) that $%
\omega (B_{c_{a}})=\omega \circ \sigma ^{\circ }(B_{c_{a}})=0$ and thus, $%
\Delta _{B_{c_{a}}}\left( \omega \right) =0$. Note that if $\mathfrak{m}$ is
gauge invariant then $\mathit{\Omega }_{\mathfrak{m}}^{\sharp ,\circ }$ is
non-empty.

Assume now that $\mathit{\Omega }_{\mathfrak{m}}$ is not a face in $E_{1}$.
Then there is an extreme generalized t.i. equilibrium states $\hat{\omega}%
_{0}\in \mathcal{E}(\mathit{\Omega }_{\mathfrak{m}}^{\sharp })$ which is not
ergodic. Since $\mathit{M}(\mathfrak{T}_{\mathfrak{m}}^{\sharp })$ is a face
of $E_{1}$ (cf. Lemma \ref{remark equilibrium state approches} and (\ref%
{union equilibrium states min max theory})), by Theorem \ref{Lemma1.vonN}
(iv) and Corollary \ref{lemma explosion l homme capacabana1 copy(2)}, $\hat{%
\omega}_{0}\in \mathit{M}(\mathfrak{T}_{\mathfrak{m}}^{\sharp })$ and%
\begin{equation}
\Delta _{B_{c_{a}}}\left( \hat{\omega}_{0}\right) =\int_{\mathit{M}(%
\mathfrak{T}_{\mathfrak{m}}^{\sharp })\cap \mathcal{E}_{1}}\mathrm{d}\mu _{%
\hat{\omega}_{0}}(\hat{\rho})\;|\langle \gamma _{a}c_{a},e_{\Phi _{a}}(\hat{%
\rho})+ie_{\Phi _{a}^{\prime }}(\hat{\rho})\rangle |^{2}.
\label{equality sup non face}
\end{equation}%
In particular, there is $c_{a}\in L^{2}(\mathcal{A},\mathbb{C})$ and a
non--ergodic state $\hat{\omega}_{0}\in \mathcal{E}(\mathit{\Omega }_{%
\mathfrak{m}}^{\sharp })$ such that 
\begin{equation*}
\Delta _{B_{c_{a}}}\left( \hat{\omega}_{0}\right) >0.
\end{equation*}%
The latter holds even if $d_{a,+}=0$, which implies that%
\begin{equation*}
|\hat{\omega}(B_{c_{a}})|^{2}=\left\vert \left\langle d_{a,+},\gamma
_{a}c_{a}\right\rangle \right\vert ^{2}=0
\end{equation*}%
for all $\hat{\omega}\in \mathcal{E}(\mathit{\Omega }_{\mathfrak{m}}^{\sharp
})$ and $c_{a}\in L^{2}(\mathcal{A},\mathbb{C})$, see (\ref{gap equations})
and Lemma \ref{lemma idiot interaction approx 2 copy(2)} ($\sharp $).
Indeed, assume that $d_{a,+}=0$ and $\Delta _{B_{c_{a}}}\left( \hat{\omega}%
_{0}\right) =0$ for all $c_{a}\in L^{2}(\mathcal{A},\mathbb{C})$. By (\ref%
{equality sup non face}), this would imply for $\hat{\rho}$ $\mu _{\hat{%
\omega}_{0}}$--a.e. that 
\begin{equation*}
e_{\Phi _{a}}(\hat{\rho})+ie_{\Phi _{a}^{\prime }}(\hat{\rho})=0\mathrm{\
(a.e.),}
\end{equation*}%
i.e., $\hat{\rho}\in \mathit{M}(\mathfrak{T}_{\mathfrak{m}}^{\sharp })\cap 
\mathcal{E}_{1}$ solves the Euler--Lagrange equations (\ref{gap equations})
and thus $\hat{\rho}\in \mathit{\Omega }_{\mathfrak{m}}^{\sharp }$. Since
the measure $\mu _{\hat{\omega}_{0}}$ is not concentrated on $\hat{\omega}%
_{0}\notin \mathcal{E}_{1}$, this would imply that $\hat{\omega}_{0}$ is
decomposable within $\mathit{\Omega }_{\mathfrak{m}}^{\sharp }$
contradicting the fact that $\hat{\omega}_{0}\in \mathcal{E}(\mathit{\Omega }%
_{\mathfrak{m}}^{\sharp })$.%
\end{proof}%

Theorem \ref{thm LRO} (i) means that no ODLRO w.r.t. elements of the form $%
B_{c_{a}}$ can be observed under the assumptions of (i). Note meanwhile that
there are uncountably many $\mathfrak{m}\in \mathcal{M}$ for which $\mathit{%
\Omega }_{\mathfrak{m}}^{\sharp }$ is not a face of $E_{1}$, see Lemma \ref%
{lemma explosion l du mec copacabana2} in\ Section \ref{section breaking
theroy}. For instance, the existence of a model $\mathfrak{m}$ such that $%
d_{a,+}=0$ and $\mathit{\Omega }_{\mathfrak{m}}^{\sharp }$ is not a face in $%
E_{1}$ follows easily from the construction done in Lemmata \ref{corolaire
Bishop phelps1 copy(1)} and \ref{lemma explosion l du mec copacabana2}.
Theorem \ref{thm LRO} (ii) shows the existence of LRO in that situation.

In conclusion, both long--range interactions $\Phi _{a,-},\Phi
_{a,-}^{\prime }$ and $\Phi _{a,+},\Phi _{a,+}^{\prime }$ (Definition \ref%
{long range attraction-repulsion}) can produce a LRO, usually at high enough
inverse temperatures $\beta >0$. Nevertheless, long--range attractions $\Phi
_{a,-},\Phi _{a,-}^{\prime }$ and repulsions $\Phi _{a,+},\Phi
_{a,+}^{\prime }$ act in a completely different way. 
\index{Long--range models!repulsions}%
\index{Long--range models!attractions}Long--range attractions $\Phi
_{a,-},\Phi _{a,-}^{\prime }$ imply ODLRO by producing non--uniqueness of
conservative strategies of the attractive player (i.e. $|\mathcal{C}_{%
\mathfrak{m}}^{\sharp }|>1$), whereas long--range repulsions $\Phi
_{a,+},\Phi _{a,+}^{\prime }$ produce LRO by breaking the face structure of
the set $\mathit{\Omega }_{\mathfrak{m}}^{\sharp }$.

\section{Concluding Remarks\label{Concluding remark}}

\setcounter{equation}{0}%
In this section, we explain our achievements in the light of previous
results. We review -- on a formal level -- in Section \ref{Section bog
approx} the original idea of the Bogoliubov approximation, which was so
successfully used in theoretical physics. 
Section \ref{Section historical overview} compares our results with the
approximating Hamiltonian method defined by Bogoliubov Jr., Brankov,
Kurbatov, Tonchev, and Zagrebnov. In order to be as short as possible we
reduce the technical aspects to an absolute minimum in all this section,
hoping that it is still understandable.

\subsection{The Bogoliubov approximation\label{Section bog approx}%
\index{Bogoliubov approximation|textbf}}

Roughly speaking, the Bogoliubov approximation consists in replacing
specific operators appearing in the Hamiltonian of a given physical system
by constants which are determined as solutions of some self-consistency
equation or some associated variational problem. One important issue is the
way such substitutions should be performed. To be successful, it depends
much on the system under consideration. In order to highlight this aspect,
we discuss bellow three different situations were Bogoliubov's method is
usually applied.

Within his celebrated microscopic theory of superfluidity \cite{Bogoliubov1}
of Helium 4, Bogoliubov proposed in 1947 his famous \textquotedblleft
trick\textquotedblright , the so--called Bogoliubov approximation, by
observing the following:

\begin{itemize}
\item[(i)] For the considered Hamiltonian modelling a Bose gas in weak
interaction inside a finite box $\Lambda $, the annihilation and creation
operators\footnote{%
In Bogoliubov's theory, $b_{0}$ and $b_{0}^{\ast }$ are the
annihilation/creation operators w.r.t. the constant function $|\Lambda
|^{-1/2}$ acting on the boson Fock space.} $b_{0}$ and $b_{0}^{\ast }$ of
bosons only appear in the form $b_{0}|\Lambda |^{-1/2}$ and $b_{0}^{\ast
}|\Lambda |^{-1/2}$.

\item[(ii)] Because of the Canonical Commutation Relations (CCR)%
\index{CCR}, $b_{0}|\Lambda |^{-1/2}$ and $b_{0}^{\ast }|\Lambda |^{-1/2}$
almost commute at large volume $|\Lambda |$.

\item[(iii)] The operators $b_{0}$ and $b_{0}^{\ast }$ are \emph{unbounded}.
\end{itemize}

Based on (i)--(iii) Bogoliubov suggested that $b_{0}$ (resp. $b_{0}^{\ast }$%
) can be replaced by a complex number $c_{\Lambda }=\mathcal{O}\left(
|\Lambda |^{1/2}\right) $ (resp. $%
\bar{c}_{\Lambda }$) to be determined \emph{self--consistently}. For a
detailed description of the Bogoliubov theory of superfluidity, we recommend
the review \cite{BruZagrebnov8}.

The Bogoliubov approximation in this precise situation was rigorously
justified in 1968 by Ginibre \cite{Ginibre} on the level of the
grand--canonical pressure in the thermodynamic limit. See also \cite%
{Fannne-Pule-Verbeure1,Fannne-Pule-Verbeure2,LiebSeiringerYngvason3,Suto}.
Actually, the (infinite--volume) pressure is given through a supremum over
complex numbers and the constant $\mathrm{c}:=c_{\Lambda }|\Lambda |^{-1/2}$
in the substitution must be a solution of this variational problem. Up to
additional technical arguments this proof \cite%
{Ginibre,LiebSeiringerYngvason3} is based on Laplace's method together with
the completeness of the family of coherent vectors $\{|c\rangle \}_{c\in 
\mathbb{C}}$ whose elements satisfy $b_{0}|c\rangle =c|c\rangle $. In fact,
in which concerns the (infinite--volume) pressure, the Bogoliubov
approximation is exact for the (stable) Bose gas even if the number $%
n_{\Lambda }$ of boson operators $\{b_{j}\}_{j=1}^{n_{\Lambda }}$ replaced
by a constant is large, provided that $n_{\Lambda }=o(|\Lambda |)$, see \cite%
{LiebSeiringerYngvason3}. Observe that the validity of the Bogoliubov
approximation on the level of the pressure has nothing to do with the
existence, or not, of a Bose condensation. However, this approximation
becomes useful when the expectation value of either $b_{0}$ (resp. $%
b_{0}^{\ast }$) or $b_{0}^{\ast }b_{0}$ becomes macroscopic, i.e., in the
case of a Bose condensation.

\begin{remark}
In the case considered above, the validity of the replacement of operators
by (possibly non zero) complex numbers depends on the unboundedness of boson
operators, whose corresponding expectation value can possibly become
macroscopic (which means $\mathrm{c}\neq 0$). Observe that the same kind of
argument cannot work for Fermi systems since the corresponding annihilation
and creation operators $a_{j}$ and $a_{j}^{+}$ are bounded in norm.
\end{remark}

Another kind of Bogoliubov approximation can be applied on a large class of
(superstable) Bose gases having the long--range interaction $\lambda
N_{\Lambda }^{2}/|\Lambda |$ with $\lambda >0$, see \cite{bru1,bru2}. Here, $%
N_{\Lambda }$ is the particle number operator inside a finite box $\Lambda $
acting on the boson Fock space. Its expectation value per unit volume is
always a finite number, i.e., the particle density, since it is a \emph{%
space--average}. This observation is not depending on the fact that $%
N_{\Lambda }$ is unbounded. It is therefore natural to replace, in the
long--range interaction $\lambda N_{\Lambda }^{2}/|\Lambda |$, the term $%
N_{\Lambda }/|\Lambda |$ by a positive real number $\rho >0$ in order to get
an effective approximating model in the thermodynamic limit. This
approximation is proven in \cite{bru1} to be exact on the level of the
pressure provided that it is done in an appropriated manner. Indeed, the
(infinite--volume) pressure, in this case, is the infimum over strictly
positive real parameters $\rho $ of pressures of approximating models, use $%
\rho =(\mu -\alpha )/2\lambda $ in \cite[Eq. (3.4)]{bru1}. Observe that the
constants replacing operators in the corresponding Bogoliubov approximations
must be a solution of that variational problem, see \cite[Theorem 4.1]{bru1}%
. 
However, the approximating model leading to this variational problem is
derived by replacing $\lambda N_{\Lambda }^{2}/|\Lambda |$ with $\lambda
(2\rho N_{\Lambda }-|\Lambda |\rho ^{2})$, i.e., one term $\lambda \rho
N_{\Lambda }$ for each choice of $N_{\Lambda }$ in $\lambda N_{\Lambda
}^{2}/|\Lambda |$. See again \cite[Eq. (3.4)]{bru1} with the choice $\rho
=(\mu -\alpha )/2\lambda >0$ because of \cite[Theorem 4.1]{bru1}. This kind
of Bogoliubov approximation could also be called a \emph{Bogoliubov
linearization}.%
\index{Bogoliubov linearization}

A similar observation holds of course for our class of Fermi models, see
Definition \ref{definition BCS-type model approximated} and Theorem \ref%
{theorem saddle point} ($\sharp $). Indeed, our long--range interaction
(Definition \ref{definition BCS-type model}) is a sum of products 
\begin{equation*}
(U_{\Lambda _{l}}^{\Phi _{a}}+iU_{\Lambda _{l}}^{\Phi _{a}^{\prime }})^{\ast
}(U_{\Lambda _{l}}^{\Phi _{a}}+iU_{\Lambda _{l}}^{\Phi _{a}^{\prime }}),
\end{equation*}%
where the expectation value of $(U_{\Lambda _{l}}^{\Phi _{a}}+iU_{\Lambda
_{l}}^{\Phi _{a}^{\prime }})$ per unit volume is always a finite number (a
mean energy density) as it is also a \emph{space--average}. Similar to \cite%
{bru1,bru2} for the real case, from our results the following replacement
has to be done: 
\begin{eqnarray*}
&&%
\frac{1}{|\Lambda _{l}|}(U_{\Lambda _{l}}^{\Phi _{a}}+iU_{\Lambda
_{l}}^{\Phi _{a}^{\prime }})^{\ast }(U_{\Lambda _{l}}^{\Phi
_{a}}+iU_{\Lambda _{l}}^{\Phi _{a}^{\prime }}) \\
&\longrightarrow &\quad \bar{c}_{a}(U_{\Lambda _{l}}^{\Phi _{a}}+iU_{\Lambda
_{l}}^{\Phi _{a}^{\prime }})+(U_{\Lambda _{l}}^{\Phi _{a}}+iU_{\Lambda
_{l}}^{\Phi _{a}^{\prime }})^{\ast }c_{a}-|\Lambda _{l}|\ |c_{a}|^{2},\quad
c_{a}\in \mathbb{C}.
\end{eqnarray*}%
The relative universality of this phenomenon comes -- in the case of models
considered here -- from the law of large numbers, whose representative in
our setting is the von Neumann ergodic theorem%
\index{von Neumann ergodic theorem} (Theorem \ref{vonN}). It leads again to
an approximating model by appropriately replacing an operator by a complex
number.

All mathematical results on Bogoliubov approximations are only performed on
the level of the pressure and possibly quasi--means provided the pressure is
known to be differentiable w.r.t. suitable parameters. Some conjectures have
been done on the level of states (see, e.g., \cite[Definition 3.2]%
{BruZagrebnov6}). Concerning Bose systems, we also recommend \cite%
{Fannne-Pule-Verbeure1,Fannne-Pule-Verbeure2} which prove the convex
decomposition of any translation and 
\index{Gauge invariant!equilibrium states}gauge invariant (analytic)
equilibrium state via non--gauge invariant equilibrium states provided the
existence of a Bose condensation. However, as far as we know, this monograph
is a first result describing the validity of the Bogoliubov approximation on
the level of (generalized) equilibrium states. See, e.g., Theorems \ref%
{theorem structure of omega copy(1)} and \ref{theorem structure of omega}.

\label{open problem}Indeed, Ginibre \cite[p. 28]{Ginibre} addressed as an
important open problem the question of the validity of the Bogoliubov
approximation (or Bogoliubov linearization) in the thermodynamic limit on
the level of (generalized) equilibrium states. Theorems \ref{theorem
structure of omega copy(1)} and \ref{theorem structure of omega} give a
first answer to this question, at least for the class of models treated
here. We prove that the Bogoliubov approximation is in general not exact on
the level of equilibrium states in the presence of \emph{non--trivial }%
long--range repulsions $\Phi _{a,+},\Phi _{a,+}^{\prime }\neq 0$ (a.e.), see
Definition \ref{long range attraction-repulsion}, Theorem \ref{lemma
explosion l homme capacabana1} and Corollary \ref{lemma explosion l homme
capacabana1 copy(2)}. This is so in spite of the fact that the Bogoliubov
approximation is exact for \emph{any} long--range model on the level of the
pressure. In the situation where the long--range component of the
interaction is \emph{purely attractive}, i.e., when $\Phi _{a,+}=\Phi
_{a,+}^{\prime }=0$ (a.e.)%
\index{Long--range models!purely attractive}, the Bogoliubov approximation
turns out to be always exact also on the level of generalized t.i.
equilibrium states as the min repulsive and the min--max local theories are
the same, i.e., $\mathfrak{T}_{\mathfrak{m}}^{+}=\mathfrak{T}_{\mathfrak{m}%
}^{\sharp }$, see Definitions \ref{definition min repulsive theory} and \ref%
{definition effective theories bogo} together with Theorem \ref{lemma
explosion l homme capacabana1 copy(1)}.

\subsection{Comparison with the approximating Hamiltonian method\label%
{Section historical overview}%
\index{Approximating Hamiltonian method}}

The Bogoliubov approximation was already used for Fermi systems on lattices
in 1957 to derive the celebrated Bardeen--Cooper--Schrieffer (BCS) theory
for conventional type I superconductors \cite{BCS1,BCS2,BCS3}. The authors
were of course inspired by Bogoliubov and his revolutionary paper \cite%
{Bogoliubov1}. A rigorous justification of this theory was given on the
level of ground states by Bogoliubov in 1960 \cite{Bogoliubov2-1960}. Then a
method for analyzing the Bogoliubov approximation in a systematic way -- on
the level of the pressure -- was introduced by Bogoliubov Jr. in 1966 \cite%
{Bogjunior,Bogjunior-livre} and by Brankov, Kurbatov, Tonchev, Zagrebnov
during the seventies and eighties \cite%
{approx-hamil-method0,approx-hamil-method,approx-hamil-method2}. This method
is known in the literature as the \emph{approximating Hamiltonian method}
and leads -- on the class of Hamiltonians it applies -- to a rigorous proof
of the exactness of the Bogoliubov approximation on the level of the
pressure, provided it is done in an appropriated manner, see discussions in
Section \ref{Section bog approx} about Bogoliubov linearization. For more
details, we recommend \cite{approx-hamil-method} as well as Section \ref%
{Section approx method}.

The class of lattice models on which the approximating Hamiltonian method is
applied belongs to the sub--space $\mathcal{M}_{1}^{\mathrm{d}}\subseteq 
\mathcal{M}_{1}$ of Fermi (or quantum spin) systems with discrete
long--range part%
\index{Long--range models!discrete}, see Section \ref{Section approx method}%
. Within our framework, it means that there is a finite family of
interactions $\{\Phi \}\cup \{\Phi _{k},\Phi _{k}^{\prime }\}_{k=1}^{N}$
defining $\mathfrak{m}$ (cf. Section \ref{definition models}). Observe that
in \cite{approx-hamil-method} the Hamiltonian $\mathrm{H}_{\Lambda }$ (see (%
\ref{Hamiltonian AHM})) can describe particles on lattices or on $\mathbb{R}%
^{d}$ as its local part $\mathrm{T}_{\Lambda }$ could be unbounded. However,
restricted to models of $\mathcal{M}_{1}$, our result is more general --
even on the level of the pressure -- in many aspects: We prove that the
ergodicity condition (A4) formulated in Section \ref{Section approx method}
and needed in \cite{approx-hamil-method} is, by far, unnecessary (cf. Remark %
\ref{important remark}). Moreover, by inspection of explicit examples and
using the triangle inequality of the operator norm, the commutator
inequalities (A3) are very unlikely to hold -- in general -- for all models
of $\mathcal{M}_{1}$ (cf. Remark \ref{important remarkbis}). Technically and
conceptually speaking, our study is performed in a different framework not
included in \cite{approx-hamil-method} and allows any Fermi systems $%
\mathfrak{m}\in \mathcal{M}_{1}$.

Additionally, the method discussed here gives new and \emph{deeper} results
on the level of states. It leads to a natural notion of (generalized)
equilibrium and ground states and, depending on the model $\mathfrak{m}\in 
\mathcal{M}_{1}$, it allows the direct analysis of \emph{all correlation
functions}, in contrast to the approximating Hamiltonian method which can be
applied for the pressure and possibly quasi-averages only. This is the main
and crucial difference between the approximating Hamiltonian method and our
approach using the structure of sets of states.

\part{Proofs and Complementary Results\label{part II}}

\chapter{Periodic Boundary Conditions and Gibbs Equilibrium States\label%
{section pbc}%
\index{Periodic boundary conditions}}

\setcounter{equation}{0}%
\index{States!Gibbs}We have shown in Theorem \ref{BCS main theorem 1} (i)
that the pressure of Fermi systems with long--range interactions is given in
the thermodynamic limit by two different variational problems on the set $%
E_{1}$ of t.i. states. We also present in Sections \ref{Section equilibrium
states copy(1)} and \ref{Section effective theories} a detailed study of
generalized t.i. equilibrium states. The weak$^{\ast }$--convergence of
Gibbs equilibrium states (cf. Section \ref{Section Gibbs equilibrium states}%
) to generalized t.i. equilibrium states is, a priori, not clear. In fact,
Gibbs equilibrium states do not generally converge to a generalized t.i.
equilibrium state, see Section \ref{Section Gibbs versus gen eq states}.
This depends on boundary conditions.

We introduce periodic boundary conditions and show in this particular case
that the Gibbs equilibrium state does converge in the weak$^{\ast }$%
--topology towards a generalized t.i. equilibrium state, see Section \ref%
{Section fermi pbc gibbs} (Theorem \ref{lemma limit gibbs states periodic}).
On the level of the pressure, periodic boundary conditions are
\textquotedblleft universal\textquotedblright\ in the sense that, for any $%
\mathfrak{m}\in \mathcal{M}_{1}$, the thermodynamic limit of the pressure (%
\ref{BCS pressure}) can be studied via models with periodic boundary
conditions, see Section \ref{Section fermi pbc pressure} (Theorem \ref{lemma
reduction p.b.c.}). Note that it is convenient to use \emph{interaction
kernels} to use internal energies with periodic boundary conditions as
defined in Section \ref{Section interaction kernel}. Fermi systems with
periodic boundary conditions are then defined in Section \ref{Section fermi
pbc} by means of such interaction kernels.

\begin{notation}[Periodic boundary conditions]
\label{Notation pbc}\mbox{ }\newline
\index{Periodic boundary conditions!notation}Any symbol with a tilde on the
top (for instance, $%
\tilde{p}$) is, by definition, an object related to periodic boundary
conditions.
\end{notation}

\section{Interaction kernels\label{Section interaction kernel}%
\index{Interaction!kernels}}

It is useful to describe interactions in terms of interaction kernels. This
requires some preliminary definitions.

Let $\mathbb{X}_{\mathfrak{L}}=\{+,-\}\times \mathrm{S}\times \mathfrak{L}$,
where we recall that $\mathfrak{L}:=\mathbb{Z}^{d}$ and $\mathrm{S}$ is a
finite set defining a finite dimensional Hilbert space $\mathcal{H}$ of
spins with orthonormal basis $\{e_{\mathrm{s}}\}_{\mathrm{s}\in \mathrm{S}}$%
. Elements of $\mathbb{X}_{\mathfrak{L}}$ are written as $X=(\nu ,\mathrm{s}%
,x)$ and we define $%
\bar{X}:=(\bar{\nu},\mathrm{s},x)$ with the convention $\bar{+}:=-$ and $%
\bar{-}:=+$. Then \emph{interaction kernels} are defined as follows:

\begin{definition}[Interaction kernels]
\label{definition interaction kernel}\mbox{ }\newline
\index{Interaction!kernels}An interaction kernel is a family $\varphi
=\{\varphi _{n}\}_{n\in \mathbb{N}_{0}}$ of anti--symmetric functions $%
\varphi _{n}:\mathbb{X}_{\mathfrak{L}}^{n}\rightarrow \mathbb{C}$ satisfying 
$\varphi _{n}=0$ for $n\notin 2\mathbb{N}_{0}$ as well as the
self--adjointness property: For any $X_{1},\dots ,X_{n}\in \mathbb{X}_{%
\mathfrak{L}}$, 
\begin{equation*}
\varphi _{n}(X_{1},\dots ,X_{n})=%
\overline{\varphi _{n}(\bar{X}_{n},\dots ,\bar{X}_{1})}.
\end{equation*}%
The set of all interaction kernels is denoted by $\mathcal{K}$.
\end{definition}

\begin{notation}[Interaction kernel]
\label{Notation7}\mbox{ }\newline
\index{Interaction!kernels!notation}The letter $\varphi $ is exclusively
reserved to denote interaction kernels.
\end{notation}

\noindent Note that any $\varphi \in \mathcal{K}$\ can be associated with an
interaction $\Phi \left( \varphi \right) $ (Definition \ref{definition
standard interaction}) with%
\begin{eqnarray}
\Phi _{\Lambda }\left( \varphi \right) &=&\sum\limits_{\{X_{i}=(\nu _{i},%
\mathrm{s}_{i},x_{i})\in \mathbb{X}_{\mathfrak{L}}\}_{i=1}^{n},\{x_{1},\dots
,x_{n}\}=\Lambda }\varphi _{n}(X_{1},\dots ,X_{n}):a_{x_{1},\mathrm{s}%
_{1}}^{\nu _{1}}\dots a_{x_{n},\mathrm{s}_{n}}^{\nu _{n}}:  \notag \\
&=&\sum\limits_{\{X_{i}=(\nu _{i},\mathrm{s}_{i},x_{i})\in \mathbb{X}_{%
\mathfrak{L}}\}_{i=1}^{n},\{x_{1},\dots ,x_{n}\}=\Lambda }\varphi
_{n}(X_{1},\dots ,X_{n}):a(X_{1})\dots a(X_{n}):  \label{action.func}
\end{eqnarray}%
Here, 
\begin{equation*}
a_{x,\mathrm{s}}^{-}:=a_{x,\mathrm{s}}\quad \mathrm{and}\quad a(X):=a_{x,%
\mathrm{s}}^{\nu }
\end{equation*}%
for $X=(\nu ,\mathrm{s},x)$. The notation%
\index{Normal ordered product} 
\begin{equation}
:a_{x_{1},\mathrm{s}_{1}}^{\nu _{1}}\dots a_{x_{n},\mathrm{s}_{n}}^{\nu
_{n}}:\quad :=(-1)^{\varsigma }a_{x_{\varsigma (1)},\mathrm{s}_{\varsigma
(1)}}^{\nu _{\varsigma (1)}}\dots a_{x_{\varsigma (n)},\mathrm{s}_{\varsigma
(n)}}^{\nu _{\varsigma (n)}}  \label{normal}
\end{equation}%
stands for the normal ordered product defined via any permutation $\varsigma 
$ of the set $\{1,\dots ,n\}$ moving all creation operators in the product $%
a_{x_{\varsigma (1)},\mathrm{s}_{\varsigma (1)}}^{\nu _{\varsigma (1)}}\dots
a_{x_{\varsigma (n)},\mathrm{s}_{\varsigma (n)}}^{\nu _{\varsigma (n)}}$ to
the left of all annihilation operators. This permutation is of course not
unique. The operator defined by the normal ordering is nevertheless uniquely
defined because of the factor $(-1)^{\varsigma }$ in (\ref{normal}) and
because of the CAR (\ref{CAR})%
\index{CAR}.

We use below the following convention: For any interaction kernel $\varphi $%
, $\Phi =\Phi (\varphi )$ is always an interaction as an operator valued map
on $\mathcal{P}_{f}(\mathfrak{L})$ which is formally written as%
\begin{equation}
\Phi (\varphi )=:\sum\limits_{X_{1},\dots ,X_{n}\in \mathbb{X}_{\mathfrak{L}%
}}\varphi _{n}(X_{1},\dots ,X_{n}):a(X_{1})\dots a(X_{n}):\ .
\label{eq:formal}
\end{equation}%
The map $\varphi \mapsto \Phi (\varphi )$ is not injective and hence, the
choice of kernels $\{\varphi _{n}\}$ for a given interaction $\Phi $ is not
unique. Note that (\ref{eq:formal}) is only a formal notation since infinite
sums over all $\mathfrak{L}$ do not appear in the definition of
interactions, see (\ref{action.func}). We can now transpose all properties
of interactions $\Phi $ in terms of interaction kernels $\varphi \in 
\mathcal{K}$.

First, we say that the interaction kernel $\varphi $ has \emph{finite range}%
\index{Interaction!kernels!finite range} iff there is a positive real number 
$d_{\max }$ such that $d(x,x^{\prime })>d_{\max }$ (cf. (\ref{def.dist}))
implies 
\begin{equation*}
\varphi _{n}\left( \left( \nu _{1},\mathrm{s}_{1},x\right) ,\left( \nu _{2},%
\mathrm{s}_{2},x^{\prime }\right) ,X_{3},\dots ,X_{n}\right) =0
\end{equation*}%
for any integer $n\geq 2$, any $(\nu _{1},\mathrm{s}_{1}),(\nu _{2},\mathrm{s%
}_{2})\in \{+,-\}\times \mathrm{S}$, and all $X_{3},\dots ,X_{n}\in \mathbb{X%
}_{\mathfrak{L}}$. Because of the CAR (\ref{CAR})%
\index{CAR} we can assume without loss of generality that, for any finite
range interaction $\varphi $, there is $N\in \mathbb{N}$ such that $\varphi
_{n}=0$ for all $n\geq N$. Clearly, if the interaction kernel $\varphi $ is
finite range then the corresponding interaction $\Phi =\Phi (\varphi )$ is
also finite range.

An interaction kernel $\varphi \in \mathcal{K}$ is \emph{translation
invariant}%
\index{Interaction!kernels!translation invariant} (t.i.) iff $\alpha
_{x}(\varphi )=\varphi $ for any $x\in \mathbb{Z}^{d}$. Here, $\alpha _{x}$
is action of the group of lattice translations on the set $\mathcal{K}$
defined by%
\begin{equation*}
\alpha _{x}(\varphi )_{n}((\nu _{1},\mathrm{s}_{1},x_{1}),\dots ,(\nu _{n},%
\mathrm{s}_{n},x_{n})):=\varphi _{n}((\nu _{1},\mathrm{s}_{1},x_{1}-x),\dots
,(\nu _{n},\mathrm{s}_{n},x_{n}-x)).
\end{equation*}%
Note that the notation $\alpha _{x}$ is also used to define via (\ref{transl}%
) the action of the group of lattice translations on $\mathcal{U}$. If the
interaction kernel $\varphi $ is t.i. then the interaction $\Phi =\Phi
(\varphi )$ is obviously translation invariant.

Additionally, the 
\index{Gauge invariant!kernels}\emph{gauge invariance}%
\index{Interaction!kernels!gauge invariant} of interactions $\Phi $ is
translated in terms of interaction kernels $\varphi \in \mathcal{K}$ via the
following property: For any $n\in 2\mathbb{N}$ and $X_{1}=(\nu _{1},\mathrm{s%
}_{1},x_{1}),\ldots ,X_{n}=(\nu _{n},\mathrm{s}_{n},x_{n})\in \mathbb{X}_{%
\mathfrak{L}}$,%
\begin{equation*}
\left\vert \{k\;:\;\nu _{k}=+\}\right\vert \not=\left\vert \{k\;:\;\nu
_{k}=-\}\right\vert \quad \mathrm{implies}\quad \varphi _{n}(X_{1},\ldots
,X_{n})=0.
\end{equation*}%
Here, $|\mathcal{X}|$ denotes the size (or cardinality) of a finite set $%
\mathcal{X}$.

To conclude, we introduce $\ell _{1}$--type norms in the case of t.i.
interaction kernels. Observe that usual $\ell _{1}$--norms would have no
meaning for t.i. functions as it would be either infinite or zero. Indeed,
we define the norm $\Vert 
\hspace{0.01in}\cdot \hspace{0.01in}\Vert _{1,\infty }$ on the space of t.i.
anti--symmetric functions $f_{n}$ on $\mathbb{X}_{\mathfrak{L}}^{n}$ to be 
\begin{equation*}
\Vert f_{n}\Vert _{1,\infty }:=\max\limits_{X_{1}\in \mathbb{X}_{\mathfrak{L}%
}}\sum\limits_{X_{2},\dots ,X_{n}\in \mathbb{X}_{\mathfrak{L}}}\left\vert
f_{n}\left( X_{1},\dots ,X_{n}\right) \right\vert .
\end{equation*}%
Then via this norm we can mimic on interaction kernels $\varphi $ norms of
the form $\Vert \,\cdot \,\Vert _{\varkappa }$ introduced for t.i.
interactions in Remark \ref{remark general interaction0}.

\begin{definition}[The Banach space $\mathcal{K}_{1}$ of t.i. interaction
kernels]
\label{def4.3}\mbox{ }\newline
\index{Interaction!kernels!Banach space}The real Banach space $\mathcal{K}%
_{1}$ is the set of all t.i. interaction kernels $\varphi $ with finite norm%
\begin{equation*}
\Vert \varphi \Vert _{\mathcal{K}_{1}}:=|\varphi
_{0}|+\sum\limits_{n=1}^{\infty }n\Vert \varphi _{n}\Vert _{1,\infty
}<\infty .
\end{equation*}
\end{definition}

\noindent Note that the set of finite range interaction kernels is dense in $%
\mathcal{K}_{1}$. In particular, $\mathcal{K}_{1}$ is separable. One can
also verify the following relations between the norms $\Vert 
\hspace{0.01in}\cdot \hspace{0.01in}\Vert _{\mathcal{W}_{1}}$ and $\Vert 
\hspace{0.01in}\cdot \hspace{0.01in}\Vert _{\mathcal{K}_{1}}$:

\begin{lemma}[Relationship between $\mathcal{K}_{1}$ and $\mathcal{W}_{1}$]
\label{lemma K1 dense dans W1}\mbox{ }\newline
\emph{(i)} For all $\varphi \in \mathcal{K}_{1}$, $\Vert \Phi (\varphi
)\Vert _{\mathcal{W}_{1}}\leq 2|\mathrm{S}|\ \Vert \varphi \Vert _{\mathcal{K%
}_{1}}$ with the size $|\mathrm{S}|\in \mathbb{N}$ of the finite set $%
\mathrm{S}$ being the dimension of the Hilbert space $\mathcal{H}$ of spins.%
\newline
\emph{(ii)} The set $\{\Phi \left( \varphi \right) :\varphi \in \mathcal{K}%
_{1}\}$ of t.i. interactions formally defined by (\ref{eq:formal}) is dense
in $\mathcal{W}_{1}$.
\end{lemma}

A typical example of an interaction $\Phi (\varphi )\in \mathcal{W}_{1}$
defined via an interaction kernel $\varphi \in \mathcal{K}_{1}$ which is 
\index{Gauge invariant!models}gauge invariant is the Hubbard model%
\index{Hubbard model} $\Phi _{\mathrm{Hubb}}$ defined as follows: $\mathrm{S}%
=\{\uparrow ,\downarrow \}$ (because electrons have spin $1/2$) and 
\begin{eqnarray*}
\Phi _{\mathrm{Hubb}} &:&=t\sum\limits_{x,y\in \mathfrak{L},d(x,y)=1,\mathrm{%
s}\in \mathrm{S}}a_{x,\mathrm{s}}^{+}a_{y,\mathrm{s}}+t^{\prime
}\sum\limits_{x,y\in \mathfrak{L},d(x,y)=%
\sqrt{2},\mathrm{s}\in \mathrm{S}}a_{x,\mathrm{s}}^{+}a_{y,\mathrm{s}} \\
&&-\mu \sum\limits_{(x,\mathrm{s})\in \mathfrak{L}\times \mathrm{S}}a_{x,%
\mathrm{s}}^{+}a_{x,\mathrm{s}}+\lambda \sum\limits_{x\in \mathfrak{L}%
}a_{x,\uparrow }^{+}a_{x,\downarrow }^{+}a_{x,\downarrow }a_{x,\uparrow }.
\end{eqnarray*}%
Here, $d(x,y)$ is the metric defined by (\ref{def.dist}) and so, the real
parameters $t$, $t^{\prime }$, $\mu $ and $\lambda $ are respectively the
nearest neighbor hopping amplitude, the next--to--nearest neighbor hopping
amplitude, the chemical potential and the interaction between pairs of
particles of different spins at the same site.

\section{Periodic boundary conditions\label{Section fermi pbc}}

We are now in position to introduce for any t.i. interaction kernel $\varphi
\in \mathcal{K}_{1}$ an interaction $\tilde{\Phi}_{l}=\tilde{\Phi}%
_{l}(\varphi )$ with periodic boundary conditions:

\begin{definition}[Periodic interactions]
\label{loc.per.int}\mbox{ }\newline
\index{Interaction!periodic boundary conditions}For any t.i. interaction
kernel $\varphi \in \mathcal{K}_{1}$ and each $l\in \mathbb{N}$, we define
the interaction $%
\tilde{\Phi}_{l}=\tilde{\Phi}_{l}(\varphi )$ with periodic boundary
conditions as follows: 
\begin{eqnarray*}
\tilde{\Phi}_{l,\Lambda } &:&=\ \ \mathbf{1}_{\left\{ \Lambda \subseteq
\Lambda _{l}\right\} }\sum\limits_{\{X_{i}=(\nu _{i},\mathrm{s}%
_{i},x_{i})\in \mathbb{X}_{\mathfrak{L}}\}_{i=1}^{n},\{x_{1},\dots
,x_{n}\}=\Lambda } \\
&&\ \ \ \left( \sum\limits_{x_{2}^{\prime },\ldots ,x_{n}^{\prime }\in 
\mathfrak{L},\;\xi _{l}(x_{i}^{\prime })=x_{i}}\varphi
_{n}(X_{1},X_{2}^{\prime }\ldots ,X_{n}^{\prime }):a(X_{1})a(X_{2})\cdots
a(X_{n}):\right)
\end{eqnarray*}%
with $X_{i}^{\prime }:=(\nu _{i},\mathrm{s}_{i},x_{i}^{\prime })$, the
normal ordered product $:a(X_{1})\cdots a(X_{n}):$ defined by (\ref{normal}%
), and $\mathbb{X}_{\mathfrak{L}}:=\{+,-\}\times \mathrm{S}\times \mathfrak{L%
}$. Here, the map $\xi _{l}:\mathfrak{L}\rightarrow \Lambda _{l}$ (cf. (\ref%
{cubic box})) is defined, for the $j^{th}$ coordinate, by $\xi
_{l}(x)_{j}=x_{j}\;\func{mod}\;2l+1$ with $j=1,\ldots ,d$.
\end{definition}

\noindent Since $\varphi \in \mathcal{K}_{1}$, observe that the operator $%
\tilde{\Phi}_{l,\Lambda }$ is clearly bounded, i.e., $\Vert \tilde{\Phi}%
_{l,\Lambda }\Vert <\infty $ for all $l\in \mathbb{N}$ and all $\Lambda \in 
\mathcal{P}_{f}(\mathfrak{L})$. The subset $\Lambda _{l}\subseteq \mathfrak{L%
}$ can be seen within this context as the torus $\mathbb{Z}^{d}/((2l+1)%
\mathbb{Z})^{d}$. Therefore, we say that the interaction $\tilde{\Phi}%
_{l,\Lambda }$ fulfills \emph{periodic boundary conditions} because it is
invariant w.r.t. translations in its corresponding torus: For all $x\in 
\mathbb{Z}^{d}$ and all $\Lambda \subseteq \Lambda _{l}$, 
\begin{equation*}
\tilde{\Phi}_{l,\xi _{l}(\Lambda +x)}=\tilde{\alpha}_{l,x}(\tilde{\Phi}%
_{l,\Lambda }).
\end{equation*}%
Here, the torus translation automorphisms $\tilde{\alpha}_{l,x}:\mathcal{U}%
_{\Lambda _{l}}\rightarrow \mathcal{U}_{\Lambda _{l}}$, $l\in \mathbb{N}$, $%
x\in \mathbb{Z}^{d}$ are defined -- uniquely -- by the condition 
\begin{equation*}
\tilde{\alpha}_{l,x}(a_{y})=\alpha _{\xi _{l}(x+y)}
\end{equation*}%
for all $y\in \Lambda _{l}$.

Then we construct from the Banach space $\mathcal{K}_{1}$ of interaction
kernels the space%
\index{Long--range models!Banach space!kernels} 
\begin{equation}
\mathcal{N}_{1}:=\mathcal{K}_{1}\times \mathcal{L}^{2}\left( \mathcal{A},%
\mathcal{K}_{1}\right) \times \mathcal{L}^{2}\left( \mathcal{A},\mathcal{K}%
_{1}\right)  \label{definition model noyaux}
\end{equation}%
of (kernel) models as explained in Section \ref{Section Preliminaries} and
define internal energies with periodic boundary conditions as follows:

\begin{definition}[Internal energy with periodic boundary conditions]
\label{definition BCS-type model periodized}\mbox{ }\newline
\index{Long--range models!internal energy!periodic boundary conditions}For
any $\mathfrak{n}:=(\varphi ,\{\varphi _{a}\}_{a\in \mathcal{A}},\{\varphi
_{a}^{\prime }\}_{a\in \mathcal{A}})\in \mathcal{N}_{1}$ and any $l\in 
\mathbb{N}$, the internal energy $%
\tilde{U}_{l}$ in the box $\Lambda _{l}$ with periodic boundary conditions
is defined to be%
\begin{equation*}
\tilde{U}_{l}:=U_{\Lambda _{l}}^{\tilde{\Phi}_{l}}+\frac{1}{|\Lambda _{l}|}%
\int_{\mathcal{A}}\gamma _{a}(U_{\Lambda _{l}}^{\tilde{\Phi}%
_{l,a}}+iU_{\Lambda _{l}}^{\tilde{\Phi}_{l,a}^{\prime }})^{\ast }(U_{\Lambda
_{l}}^{\tilde{\Phi}_{l,a}}+iU_{\Lambda _{l}}^{\tilde{\Phi}_{l,a}^{\prime }})%
\mathrm{d}\mathfrak{a}\left( a\right) ,
\end{equation*}%
where $\gamma _{a}\in \{-1,1\}$ is a measurable function and with $\tilde{%
\Phi}_{l}=\tilde{\Phi}_{l}\left( \varphi \right) $, $\tilde{\Phi}_{l,a}=%
\tilde{\Phi}_{l}(\varphi _{a})$, and $\tilde{\Phi}_{l,a}^{\prime }=\tilde{%
\Phi}_{l}(\varphi _{a}^{\prime })\ $for any $a\in \mathcal{A}$.
\end{definition}

\begin{notation}[Model kernels]
\label{Notation8}\mbox{ }\newline
\index{Long--range models!Banach space!kernels}The symbol $\mathfrak{n}$ is
exclusively reserved to denote elements of $\mathcal{N}_{1}$.
\end{notation}

\noindent Re-expressing objects in terms of interactions with periodic
boundary conditions has the advantage that the notion of translation
invariance is locally preserved. This implies, among other things, the
translation invariance of the thermodynamic limit of Gibbs equilibrium
states (Definition \ref{Gibbs.statebis}). It is an essential property to
obtain a generalized t.i. equilibrium state in the thermodynamic limit.

\begin{remark}
Any $\mathfrak{n}=(\varphi ,\{\varphi _{a}\}_{a\in \mathcal{A}},\{\varphi
_{a}^{\prime }\}_{a\in \mathcal{A}})\in \mathcal{N}_{1}$ is identified with
the long--range model $(\Phi (\varphi ),\{\Phi (\varphi _{a})\}_{a\in 
\mathcal{A}},\{\Phi (\varphi _{a}^{\prime })\}_{a\in \mathcal{A}})\in 
\mathcal{M}_{1}$ for a given $\gamma _{a}$.
\end{remark}

\section{Pressure and periodic boundary conditions \label{Section fermi pbc
pressure}}

Periodic boundary conditions are very particular and idealized in which
concerns the represented physical situations. Dirichlet--like or von
Neumann--like boundary conditions are -- physically speaking -- more
natural. In spite of that, they are extensively used in theoretical or
mathematical physics because they allow for the use\ of Fourier analysis,
making computations much easier. In fact, we show the \textquotedblleft
universality\textquotedblright\ of periodic boundary conditions on the level
of the pressure. This means that, for any $\mathfrak{m}\in \mathcal{M}_{1}$,
the thermodynamic limit of the pressure (\ref{BCS pressure}) can be studied
via models with periodic boundary conditions, see Definition \ref{definition
BCS-type model periodized}.

Indeed, observe first that periodic boundary conditions do not change the
internal energy per volume associated with any t.i. interaction kernel $%
\varphi \in \mathcal{K}_{1}$:

\begin{lemma}[Internal energy and periodic boundary conditions]
\label{P.period}\mbox{ }\newline
For any $\varphi \in \mathcal{K}_{1}$,%
\begin{equation*}
\lim\limits_{l\rightarrow \infty }%
\frac{1}{|\Lambda _{l}|}\Vert U_{\Lambda _{l}}^{\Phi \left( \varphi \right)
}-U_{\Lambda _{l}}^{\tilde{\Phi}_{l}}\Vert =0
\end{equation*}%
with $U_{\Lambda }^{\Phi }$, $\Phi \left( \varphi \right) $, and $\tilde{\Phi%
}_{l}$ respectively defined by Definition \ref{definition standard
interaction}, (\ref{action.func}) (see also (\ref{eq:formal})) and
Definition \ref{loc.per.int}.
\end{lemma}

\begin{proof}%
For any $\Lambda \in \mathcal{P}_{f}(\mathfrak{L})$, let $\Lambda ^{c}:=%
\mathfrak{L}\backslash \Lambda $ be its complement and we denote by 
\begin{equation*}
\hat{d}(x,\Lambda ):=\underset{x^{\prime }\in \Lambda }{\min }\left\{
d(x,x^{\prime })\right\}
\end{equation*}%
the distance between any point $x\in \mathbb{Z}^{d}$ and the set $\Lambda
\in \mathcal{P}_{f}(\mathfrak{L})$. The latter is constructed via the metric 
$d(x,x^{\prime })$ defined by (\ref{def.dist}). It follows from Definitions %
\ref{definition standard interaction} and \ref{loc.per.int} together with
Equality (\ref{action.func}) that 
\begin{eqnarray}
\Vert U_{\Lambda _{l}}^{\Phi (\varphi )}-U_{\Lambda _{l}}^{\tilde{\Phi}%
_{l}}\Vert &\leq &\sum\limits_{\{X_{i}\in \mathbb{X}_{\mathfrak{L}%
}\}_{i=1}^{n},\;x_{1}\in \Lambda _{l},\;\{x_{2},\ldots ,x_{n}\}\cap \Lambda
_{l}^{c}\not=\emptyset }n|\varphi _{n}(X_{1},\ldots ,X_{n})|  \notag \\
&\leq &\sum\limits_{\{X_{i}\in \mathbb{X}_{\mathfrak{L}}\}_{i=1}^{n},\;x_{1}%
\in \Lambda _{l},\;\{x_{2},\ldots ,x_{n}\}\cap \Lambda
_{l}^{c}\not=\emptyset }\left( \mathbf{1}_{\{\hat{d}(x_{1},\Lambda
_{l}^{c})\leq \sqrt{l}\}}n|\varphi _{n}(X_{1},\ldots ,X_{n})|\right.  \notag
\\
&&\left. +\mathbf{1}_{\{\hat{d}(x_{1},\Lambda _{l}^{c})>\sqrt{l}\}}n|\varphi
_{n}(X_{1},\ldots ,X_{n})|\right) .  \label{inequality periodic lemma}
\end{eqnarray}%
We observe that%
\begin{equation}
\lim\limits_{l\rightarrow \infty }\frac{1}{|\Lambda _{l}|}%
\sum\limits_{\{X_{i}\in \mathbb{X}_{\mathfrak{L}}\}_{i=1}^{n},\;x_{1}\in
\Lambda _{l},\;\{x_{2},\ldots ,x_{n}\}\cap \Lambda _{l}^{c}\not=\emptyset }%
\mathbf{1}_{\{\hat{d}(x_{1},\Lambda _{l}^{c})\leq \sqrt{l}\}}n|\varphi
_{n}(X_{1},\ldots ,X_{n})|=0  \label{inequality periodic lemma2}
\end{equation}%
as $\Vert \varphi \Vert _{\mathcal{K}_{1}}<\infty $. Moreover, since, by
translation invariance of the interaction kernel $\varphi $,%
\begin{eqnarray*}
&&\frac{1}{|\Lambda _{l}|}\sum\limits_{\{X_{i}\in \mathbb{X}_{\mathfrak{L}%
}\}_{i=1}^{n},\;x_{1}\in \Lambda _{l},\;\{x_{2},\ldots ,x_{n}\}\cap \Lambda
_{l}^{c}\not=\emptyset }\mathbf{1}_{\{\hat{d}(x_{1},\Lambda _{l}^{c})>\sqrt{l%
}\}}n|\varphi _{n}(X_{1},\ldots ,X_{n})| \\
&\leq &\sum\limits_{\{X_{i}\in \mathbb{X}_{\mathfrak{L}}\}_{i=2}^{n}}\
\sum\limits_{(\nu ,\mathrm{s})\in \{+,-\}\times \mathrm{S}}\mathbf{1}%
_{\{\min \{|x_{2}|,\ldots ,|x_{n}|\}>\sqrt{l}\}}n|\varphi _{n}(\tilde{X}%
,X_{2},\ldots ,X_{n})|
\end{eqnarray*}%
with $\tilde{X}:=(\nu ,\mathrm{s},0)$, we use again $\Vert \varphi \Vert _{%
\mathcal{K}_{1}}<\infty $ and Lebesgue's dominated convergence theorem to
obtain that 
\begin{equation}
\lim\limits_{l\rightarrow \infty }\frac{1}{|\Lambda _{l}|}%
\sum\limits_{\{X_{i}\in \mathbb{X}_{\mathfrak{L}}\}_{i=1}^{n},\;x_{1}\in
\Lambda _{l},\;\{x_{2},\ldots ,x_{n}\}\cap \Lambda _{l}^{c}\not=\emptyset }%
\mathbf{1}_{\{\hat{d}(x_{1},\Lambda _{l}^{c})>\sqrt{l}\}}n|\varphi
_{n}(X_{1},\ldots ,X_{n})|=0.  \label{inequality periodic lemma3}
\end{equation}%
Therefore, the lemma follows from Inequality (\ref{inequality periodic lemma}%
) together with the limits (\ref{inequality periodic lemma2}) and (\ref%
{inequality periodic lemma3}). 
\end{proof}%

To show now that the pressure (\ref{BCS pressure}) can be studied via models
with periodic boundary conditions, we need some preliminary definitions.
First, for any $\mathfrak{n}\in \mathcal{N}_{1}$ and $l\in \mathbb{N}$, let%
\begin{equation}
\tilde{p}_{l}=\tilde{p}_{l,\mathfrak{n}}:=\frac{1}{\beta |\Lambda _{l}|}\ln 
\mathrm{Trace}_{\wedge \mathcal{H}_{\Lambda }}(\mathrm{e}^{-\beta \tilde{U}%
_{l}})  \label{pressure pbc}
\end{equation}%
be the pressure associated with the internal energy $\tilde{U}_{l}$
(Definition \ref{definition BCS-type model periodized}). Then we extend the
map $\varphi \mapsto \Phi (\varphi )$ (cf. (\ref{eq:formal})) to a map $%
\mathfrak{n}\mapsto \mathfrak{m}(\mathfrak{n})$ from $\mathcal{N}_{1}$ to $%
\mathcal{M}_{1}$. To simplify the notation let%
\index{Free--energy density functional!long--range!kernels|textbf} 
\begin{equation}
\mathfrak{n}\mapsto \mathfrak{m}(\mathfrak{n})\mapsto f_{\mathfrak{m}(%
\mathfrak{n})}^{\sharp }=:f_{\mathfrak{n}}^{\sharp }.  \label{map a la con}
\end{equation}%
$f_{\mathfrak{m}}^{\sharp }$ is seen below as map from $\mathcal{M}_{1}$ to
the set $\mathcal{F}_{E_{1}}$ of affine functionals on $E_{1}$ (see
Definition \ref{Free-energy density long range} and Lemma \ref{lemma
property free--energy density functional} (i)), whereas $f_{\mathfrak{n}%
}^{\sharp }$ is seen as a map from $\mathcal{N}_{1}$ to $\mathcal{F}_{E_{1}}$
via (\ref{map a la con}). In the same way we have introduced the dense
sub--spaces $\mathcal{M}_{1}^{\mathrm{f}}$, $\mathcal{M}_{1}^{\mathrm{d}}$,
and $\mathcal{M}_{1}^{\mathrm{df}}$ in Section \ref{definition models}, we
finally define the dense sub--spaces $\mathcal{N}_{1}^{\mathrm{f}}$ and $%
\mathcal{N}_{1}^{\mathrm{d}}$ of $\mathcal{N}_{1}$ to be, respectively, the
sets of finite range $\mathfrak{n}$ and discrete elements $\mathfrak{n}$,
see Section \ref{Section interaction kernel} and (\ref{definition model
noyaux})%
\index{Long--range models!Banach space!kernels}. So, $\mathcal{N}_{1}^{%
\mathrm{df}}:=\mathcal{N}_{1}^{\mathrm{d}}\cap \mathcal{N}_{1}^{\mathrm{f}}$
is the (dense) sub--space of finite range discrete elements $\mathfrak{n}$.

We are now in position to give the main theorem of this section about the
\textquotedblleft universality\textquotedblright\ of periodic boundary
conditions w.r.t. the pressure of long--range Fermi systems.

\begin{theorem}[Reduction to periodic boundary conditions]
\label{lemma reduction p.b.c.}\mbox{ }\newline
For any $\mathfrak{m}\in \mathcal{M}_{1}^{\mathrm{df}}$, there exists $%
\mathfrak{n}\in \mathcal{N}_{1}^{\mathrm{df}}$ such that:%
\begin{equation*}
\mathrm{(i)\ }\underset{l\rightarrow \infty }{\lim }\left\{ 
\tilde{p}_{l,\mathfrak{n}}-p_{l,\mathfrak{m}}\right\} =0;\quad \mathrm{(ii)\ 
}f_{\mathfrak{m}}^{\sharp }=f_{\mathfrak{n}}^{\sharp }.
\end{equation*}
\end{theorem}

\begin{proof}%
For any \emph{finite range} interaction $\Phi \in \mathcal{W}_{1}$, the
energy observable $\mathfrak{e}_{\Phi }\in \mathcal{U}^{+}$ defined by (\ref%
{eq:enpersite}) belongs to the set $\mathcal{U}_{0}$ of local elements and
thus, there is a finite range interaction kernel $\varphi (\Phi )$ such that 
\begin{equation}
\mathfrak{e}_{\Phi }=\mathfrak{e}_{\Phi (\varphi \left( \Phi \right) )}\quad 
\mathrm{and}\quad \Vert U_{\Lambda _{l}}^{\Phi }-U_{\Lambda _{l}}^{\Phi
(\varphi \left( \Phi \right) )}\Vert \leq \mathcal{O}(|\partial \Lambda
_{l}|)=\mathcal{O}(l^{d-1})  \label{fin.range.kernel}
\end{equation}%
with $\partial \Lambda _{l}$ being the boundary\footnote{%
By fixing $m\geq 1$ the boundary $\partial \Lambda $ of any $\Lambda \subset
\Gamma $ is defined by $\partial \Lambda :=\{x\in \Lambda \;:\;\exists y\in
\Gamma \backslash \Lambda \mathrm{\ with\ }d(x,y)\leq m\},$ see (\ref%
{def.dist}) for the definition of the metric $d(x,y)$.\label{footnote20}} of
the cubic box $\Lambda _{l}$. Therefore, for any finite range discrete model%
\begin{equation*}
\mathfrak{m}:=\{\Phi \}\cup \{\Phi _{k},\Phi _{k}^{\prime }\}_{k=1}^{N}\in 
\mathcal{M}_{1}^{\mathrm{df}},
\end{equation*}%
there exists 
\begin{equation*}
\mathfrak{n}:=\{\varphi (\Phi )\}\cup \{\varphi (\Phi _{k}),\varphi (\Phi
_{k}^{\prime })\}_{k=1}^{N}\in \mathcal{N}_{1}^{\mathrm{df}}
\end{equation*}%
satisfying (\ref{fin.range.kernel}) for each interaction $\Phi $, $\Phi _{k}$%
, and $\Phi _{k}^{\prime }$. Any $\mathfrak{n}\in \mathcal{N}_{1}$ defines
an internal energy $\tilde{U}_{l}$ with periodic boundary conditions. So,
the first statement (i) of the lemma is a consequence of the bound 
\begin{equation}
|\ln (\mathrm{Trace}_{\wedge \mathcal{H}_{\Lambda }}(\mathrm{e}^{A}))-\ln (%
\mathrm{Trace}_{\wedge \mathcal{H}_{\Lambda }}(\mathrm{e}^{B}))|\leq \Vert
A-B\Vert  \label{petite inequality}
\end{equation}%
combined with Lemma \ref{P.period} for any t.i. interaction kernel $\varphi
\in \mathcal{K}_{1}$. The second statement (ii) is a direct consequence of (%
\ref{fin.range.kernel}). 
\end{proof}%

\begin{remark}
Note that the restriction $\mathfrak{m}\in \mathcal{M}_{1}^{\mathrm{df}}$ in
this last theorem is unimportant, see Corollary \ref{lemma reduction
dfbisbis}.
\end{remark}

\section{Gibbs and generalized t.i. equilibrium states\label{Section fermi
pbc gibbs}}

\index{States!Gibbs}Periodic boundary conditions are, on the level of the
pressure, universal in the sense described by Theorem \ref{lemma reduction
p.b.c.}. However, it is important to note that periodic boundary conditions
do \emph{not} yield a complete thermodynamic description of long--range
Fermi systems on the level of equilibrium states. As shown below (Theorem %
\ref{lemma limit gibbs states periodic}), any weak$^{\ast }$--convergent
sequence of Gibbs equilibrium states (Definition \ref{Gibbs.statebis}) of
long--range Fermi systems with periodic boundary conditions converges to a
generalized t.i. equilibrium state. The convergence of arbitrary convergent
sequences $\rho _{l}$ of (local) Gibbs equilibrium states of t.i.
long--range models $\mathfrak{m}\in \mathcal{M}_{1}$ (defined by $\rho
_{l}:=\rho _{\Lambda _{l},U_{l}}$ (\ref{Gibbs.state})) towards a
(infinite--volume) generalized t.i. equilibrium state is, a priori, not
clear and could in fact be even wrong in some cases (depending on boundary
conditions). Together with Theorem \ref{lemma reduction p.b.c.}, this means
that the infimum over the set $E$ of all states given in Theorem \ref{BCS
main theorem 1 copy(1)} (i) could also be attained by a sequence of
approximating minimizers (cf. (\ref{approximating minimizer})) with weak$%
^{\ast }$--limit points not in $E_{1}$ as explained in Section \ref{Section
Gibbs versus gen eq states}.

Therefore, we study now the convergence of (local) Gibbs equilibrium states
only for the particular case of periodic boundary conditions, i.e., the
convergence of the states $%
\tilde{\rho}_{l}:=\rho _{\Lambda _{l},\tilde{U}_{l}}$ (\ref{Gibbs.state}).
Note that this state $\tilde{\rho}_{l}$ is as usual seen as defined either
on the local algebra $\mathcal{U}_{\Lambda _{l}}$ or on the whole algebra $%
\mathcal{U}$ by periodically extending it (with period $(2l+1)$ in each
direction of the lattice $\mathfrak{L}$). Observe here that, by the
definition of interaction kernels, $\tilde{\rho}_{l}$ is an even state and
hence products of translates of $\tilde{\rho}_{l}$ are well--defined (cf. 
\cite[Theorem 11.2.]{Araki-Moriya}). The Gibbs equilibrium state $\tilde{\rho%
}_{l}$ is \emph{generally not }translation invariant. We construct the
space--averaged t.i. Gibbs state $\hat{\rho}_{l}\in E_{1}$ from $\tilde{\rho}%
_{l}$ as it is done in (\ref{t.i. state rho l}), that is,%
\index{States!Gibbs!space--averaged t.i.}%
\begin{equation}
\hat{\rho}_{l}:=\frac{1}{|\Lambda _{l}|}\sum\limits_{x\in \Lambda _{l}}%
\tilde{\rho}_{l}\circ \alpha _{x},  \label{space average state p.b.c.}
\end{equation}%
where we recall that the $\ast $--automorphisms $\{\alpha _{x}\}_{x\in 
\mathbb{Z}^{d}}$ defined by (\ref{transl}) are the action of the group of
lattice translations on $\mathcal{U}$. Then, from Theorems \ref{BCS main
theorem 1} (i) and \ref{eq.tang.bcs.type}, we prove the convergence of local
states $\tilde{\rho}_{l}$ and $\hat{\rho}_{l}$ towards the same generalized
t.i. equilibrium state:

\begin{theorem}[Weak$^{\ast }$--limit of Gibbs equilibrium states]
\label{lemma limit gibbs states periodic}\mbox{ }\newline
For any $\mathfrak{n}\in \mathcal{N}_{1}$, the states $\tilde{\rho}_{l}$ and 
$\hat{\rho}_{l}$ converge in the weak$^{\ast }$--topology along any
convergent subsequence towards the same generalized t.i. equilibrium state $%
\omega \in \mathit{\Omega }_{\mathfrak{n}}^{\sharp }$.%
\index{States!Gibbs!space--averaged t.i.}%
\index{States!Gibbs}
\end{theorem}

\begin{proof}%
By weak$^{\ast }$--compactness of $E_{1}$, the space--averaged t.i. Gibbs
state $%
\hat{\rho}_{l}$ converges in the weak$^{\ast }$--topology along a
subsequence towards $\omega \in E_{1}$. By translation invariance of $\tilde{%
\rho}_{l}$ in the torus $\Lambda _{l}$, it is also easy to see that the
sequences of states $\tilde{\rho}_{l}$ and $\hat{\rho}_{l}$ have the same
weak$^{\ast }$--limit points. Then, since Theorem \ref{eq.tang.bcs.type}
says that $\mathit{T}_{\mathfrak{m}}^{\sharp }=\mathit{\Omega }_{\mathfrak{m}%
}^{\sharp }$ for all $\mathfrak{m}\in \mathcal{M}_{1}$, we show that $\omega
\in \mathit{T}_{\mathfrak{n}}^{\sharp }$ in the same way we prove Theorem %
\ref{lemma limit averaging gibbs states} because of Lemma \ref{P.period},
Theorem \ref{lemma reduction p.b.c.}, and the density of the sets $\mathcal{N%
}_{1}^{\mathrm{df}}$ and $\{\Phi \left( \varphi \right) \}_{\varphi \in 
\mathcal{K}_{1}}$ respectively in $\mathcal{N}_{1}$ and $\mathcal{W}_{1}$.
We omit the details. 
\end{proof}%

\chapter{The Set $E_{\vec{\ell}}$ of $\mathbb{Z}_{\vec{\ell}}^{d}$%
--Invariant States\label{section set of invarant states}}

\setcounter{equation}{0}%
\index{States!l--invariant}In this chapter, we study in details the
structure of the convex and weak$^{\ast }$--compact sets $E_{%
\vec{\ell}}$ of $\mathbb{Z}_{\vec{\ell}}^{d}$--invariant states defined by (%
\ref{periodic invariant states}) for any $\vec{\ell}\in \mathbb{N}^{d}$. The
set $\mathcal{E}_{\vec{\ell}}$ of extreme points of $E_{\vec{\ell}}$ is
intimately related with a property of ergodicity (Definition \ref{def:egodic}%
). For $\vec{\ell}=(1,\cdots ,1)$, the ergodicity of states is characterized
via the space--averaging functional $\Delta _{A}$ defined for any $A\in 
\mathcal{U}$ in Definition \ref{definition de deltabis}.

We discuss in Section \ref{Section theorem ergodic extremalbis} the main
structural properties of the set $\mathcal{E}_{\vec{\ell}}$ and analyze the
map $\Delta _{A}$ in Section \ref{Section properties of delta}. The
properties of the entropy density functional $s$ defined in Definition \ref%
{entropy.density} are discussed in Section \ref{section neuman entropy}. In
Section \ref{Section state=functional on W} we analyze the energy density
functional $e_{\Phi }$\ defined, for any t.i. interaction $\Phi \in \mathcal{%
W}_{1}$, in\ Definition \ref{definition energy density}. By means of the
energy density $e_{\Phi }$, each $\vec{\ell}$--periodic state $\rho \in E_{%
\vec{\ell}}$ defines a continuous linear functional $\mathbb{T}\left( \rho
\right) \in \mathcal{W}_{1}^{\ast }$ on the Banach space $\mathcal{W}_{1}$
(Definition \ref{definition banach space interaction}). The map $\rho
\mapsto \mathbb{T}\left( \rho \right) $ restricted to the set $E_{1}$ of
t.i. states is injective. This allows the identification of states of $E_{1}$
with functionals of $\mathcal{W}_{1}^{\ast }$.

Note that some important statements presented here are standard (see, e.g.,
Theorems \ref{GNS} and \ref{vonN}). They are given in Section \ref{Section
GNS Neuman} for completeness. We start with a preliminary discussion about
the Gelfand--Naimark--Segal (GNS) representation of $G$--invariant states 
\cite[Corollary 2.3.17]{BrattelliRobinsonI} and then about the von Neumann
ergodic theorem \cite[Proposition 4.3.4]{BrattelliRobinsonI}.

\section{GNS representation and the von Neumann ergodic theorem\label%
{Section GNS Neuman}}

\index{States!GNS representation|textbf}Any state $\rho \in E$ has a GNS
representation \cite[Theorem 2.3.16]{BrattelliRobinsonI}: For any $\rho \in
E $, there exist a Hilbert space $\mathcal{H}_{\rho }$, a representation $%
\pi _{\rho }:\mathcal{U}\rightarrow \mathcal{B}(\mathcal{H}_{\rho })$ from $%
\mathcal{U}$ to the set $\mathcal{B}(\mathcal{H}_{\rho })$ of bounded
operators on $\mathcal{H}_{\rho }$, and a cyclic vector $\Omega _{\rho }\in 
\mathcal{H}_{\rho }$ w.r.t. $\pi _{\rho }(\mathcal{U})$ such that, for all $%
A\in \mathcal{U}$, 
\begin{equation*}
\rho (A)=\langle \Omega _{\rho },\pi _{\rho }(A)\Omega _{\rho }\rangle .
\end{equation*}%
The representation $\pi _{\rho }$ is faithful if $\rho $ is faithful, that
is, if $\rho (A^{\ast }A)=0$ implies $A=0$. The triple $(\mathcal{H}_{\rho
},\pi _{\rho },\Omega _{\rho })$ is unique up to unitary equivalence.

Assume now the existence of a group homomorphism $g\mapsto \alpha _{g}$ from 
$G$ to the group of $\ast $--automorphisms of $\mathcal{U}$. The state $\rho 
$ is $G$--invariant iff $\rho \circ \alpha _{g}=\rho $ for any $g\in G$. The
GNS representation of such a $G$--invariant state $\rho $ carries this
symmetry through a uniquely defined family of unitary operators, see \cite[%
Corollary 2.3.17]{BrattelliRobinsonI}:

\begin{theorem}[GNS representation of $G$--invariant states]
\label{GNS}\mbox{ }\newline
\index{States!GNS representation!G-invariant|textbf}Let $\rho $ be a $G$%
--invariant state with GNS\ representation $(\mathcal{H}_{\rho },\pi _{\rho
},\Omega _{\rho })$. Then there is a uniquely defined family $%
\{U_{g}\}_{g\in G}$ of unitary operators in $\mathcal{B}(\mathcal{H}_{\rho
}) $ with invariant vector $\Omega _{\rho }$, i.e., $\Omega _{\rho
}=U_{g}\Omega _{\rho }$ for any $g\in G$, and such that $\pi _{\rho }(\alpha
_{g}(A))=U_{g}\pi _{\rho }(A)U_{g}^{\ast }$ for any $g\in G$ and $A\in 
\mathcal{U}$. In particular, $U_{g_{1}+g_{2}}=U_{g_{1}}U_{g_{2}}$ for any $%
g_{1},g_{2}\in G$.
\end{theorem}

\begin{proof}
See \cite[Corollary 2.3.17]{BrattelliRobinsonI}. In particular, for any $%
g_{1},g_{2}\in G$ and $A\in \mathcal{U}$,%
\begin{eqnarray*}
U_{g_{1}+g_{2}}\pi _{\rho }(A)U_{g_{1}+g_{2}}^{\ast } &=&\pi _{\rho }(\alpha
_{g_{1}+g_{2}}(A))=\pi _{\rho }(\alpha _{g_{1}}\circ \alpha _{g_{2}}(A)) \\
&=&U_{g_{1}}\pi _{\rho }(\alpha _{g_{2}}(A))U_{g_{1}}^{\ast
}=U_{g_{1}}U_{g_{2}}\pi _{\rho }(\alpha _{g_{2}}(A))U_{g_{2}}^{\ast
}U_{g_{1}}^{\ast }.
\end{eqnarray*}%
By uniqueness of the family $\{U_{g}\}_{g\in G}$, one gets $%
U_{g_{1}+g_{2}}=U_{g_{1}}U_{g_{2}}$ for any $g_{1},g_{2}\in G$.%
\end{proof}%

Since we study the set $E_{%
\vec{\ell}}$ (\ref{periodic invariant states}) of $\vec{\ell}$\emph{%
--periodic} states, the special cases we are interested in are $G=(\mathbb{Z}%
_{\vec{\ell}}^{d},+)$ for all $\vec{\ell}\in \mathbb{N}^{d}$. The group
homomorphism $g\mapsto \alpha _{g}$ from $G$ to the group of $\ast $%
--automorphisms of $\mathcal{U}$ corresponds, in this case, to the group $%
\{\alpha _{x}\}_{x\in \mathbb{Z}^{d}}$ (\ref{transl}) of lattice
translations on $\mathcal{U}$. Within this framework, an essential
ingredient of our analysis is the von Neumann ergodic theorem \cite[%
Proposition 4.3.4]{BrattelliRobinsonI} which is a representative of the law
of large numbers:

\begin{theorem}[von Neumann ergodic theorem]
\label{vonN}\mbox{ }\newline
\index{von Neumann ergodic theorem|textbf}Let $x\mapsto U_{x}$ be a
representation of the abelian group $(\mathbb{Z}_{%
\vec{\ell}}^{d},+)$ by unitary operators on a Hilbert space $\mathcal{H}$
and the set 
\begin{equation*}
I:=\bigcap_{x\in \mathbb{Z}_{\vec{\ell}}^{d}}\{\psi \in \mathcal{H}\,:\,\psi
=U_{x}(\psi )\}
\end{equation*}%
be the closed sub--space of all invariant vectors. For any $L\in \mathbb{N}$%
, define the contraction%
\begin{equation*}
P^{(L)}:=\frac{1}{|\Lambda _{L}\cap \mathbb{Z}_{\vec{\ell}}^{d}|}%
\sum\limits_{x\in \Lambda _{L}\cap \mathbb{Z}_{\vec{\ell}}^{d}}U_{x}\in 
\mathcal{B}(\mathcal{H})
\end{equation*}%
and denote the orthogonal projection on $I$ by $P$. Then, for all $L\in 
\mathbb{N}$, $PP^{(L)}=P^{(L)}P=P$ and the operator $P^{(L)}$ converges
strongly to $P$ as $L\rightarrow \infty $.
\end{theorem}

\begin{proof}%
The proof of this statement is standard, see, e.g., \cite[Theorem IV.2.2 ]%
{Israel}. It is given here for completeness. Note that the property $%
PP^{(L)}=P^{(L)}P=P$ is, in general, not explicitly given in the versions of
the von Neumann ergodic theorem found in textbooks.

Without loss of generality, assume that $\vec{\ell}=(1,\cdots ,1)$. For any $%
i\in \{1,\ldots ,d\}$, let us consider the unitary operators $%
U_{i}:=U_{(\delta _{i,1},\dots ,\delta _{i,d})}$ with $\delta _{i,j}=0$ for
any $i\neq j$ and $\delta _{i,i}=1$. Since $\mathbb{Z}^{d}$ is abelian, the
normal operators $U_{i}$ for $i\in \{1,\ldots ,d\}$ commute with each other.
Their joint spectrum is contained in the $d$--dimensional torus%
\begin{equation*}
\mathrm{T}_{d}:=\{(z_{1},\ldots ,z_{d})\in \mathbb{C}^{d}\,:\,|z_{i}|=1,%
\,i=1,\ldots ,d\}
\end{equation*}%
and the spectral theorem \cite[Chap. 6, Sect. 5]{spectral thm} ensures the
existence of a projection--valued measure $\mathrm{d}P$ on the torus $%
\mathrm{T}_{d}$ such that%
\begin{equation}
P^{(L)}=\int_{\mathrm{T}_{d}}f_{L}(z_{1},\ldots ,z_{d})\mathrm{d}%
P(z_{1},\ldots ,z_{d})  \label{lim.proj}
\end{equation}%
for any $L\in \mathbb{N}$, where%
\begin{equation*}
f_{L}(z_{1},\ldots ,z_{d}):=\frac{1}{|\Lambda _{L}|}\sum\limits_{(x_{1},%
\ldots ,x_{d})\in \Lambda _{L}}z_{1}^{x_{1}}\cdots z_{d}^{x_{d}}.
\end{equation*}%
Observe that $f_{L}$ converges point--wise as $L\rightarrow \infty $ to the
characteristic function of the set $\{(1,\ldots ,1)\}\subseteq \mathrm{T}%
_{d} $, i.e.,%
\begin{equation}
f_{\infty }(z_{1},\ldots ,z_{d}):=\lim\limits_{L\rightarrow \infty
}f_{L}(z_{1},\ldots ,z_{d})=\left\{ 
\begin{array}{l}
1\quad \mathrm{if\ }(z_{1},\ldots ,z_{d})=(1,\ldots ,1). \\ 
0\quad \mathrm{else.}%
\end{array}%
\right.  \label{lim.proj1}
\end{equation}%
Hence, from (\ref{lim.proj}), the operator $P^{(L)}$ converges strongly to%
\begin{equation}
P^{(\infty )}:=\int_{\mathrm{T}_{d}}f_{\infty }(z_{1},\ldots ,z_{d})\mathrm{d%
}P(z_{1},\ldots ,z_{d}).  \label{lim.proj2}
\end{equation}%
Note that the operator $P^{(\infty )}$ is an orthogonal projection because
of (\ref{lim.proj1})--(\ref{lim.proj2}). Additionally, $I\subseteq
P^{(\infty )}(\mathcal{H})$ by definition and $P^{(\infty )}(\mathcal{H}%
)\subseteq I$ by using (\ref{lim.proj})--(\ref{lim.proj2}) combined with $%
U_{x+y}=U_{y}U_{x}$ for any $x,y\in \mathbb{Z}^{d}$. Therefore, $P^{(\infty
)}=P$ and from (\ref{lim.proj2}) together with $f_{L}f_{\infty }=f_{\infty }$
we deduce that $PP^{(L)}=P^{(L)}P=P$ for any $L\in \mathbb{N}$. 
\end{proof}%

For any $\rho \in E_{\vec{\ell}}$ with GNS\ representation $(\mathcal{H}%
_{\rho },\pi _{\rho },\Omega _{\rho })$, we define $P_{\rho }$ to be the
strong limit of contractions $P^{(L)}$ defined in Theorem \ref{vonN} w.r.t.
the unitary operators $\{U_{x}\}_{x\in \vec{\ell}\cdot \mathbb{Z}^{d}}$ on
the Hilbert space $\mathcal{H}_{\rho }$ of Theorem \ref{GNS} for $G=(\mathbb{%
Z}_{\vec{\ell}}^{d},+)$. By using the projection $P_{\rho }$ and the
equality $\Omega _{\rho }=P_{\rho }\Omega _{\rho }$, it is then easy to
check that all $\vec{\ell}$--periodic states $\rho \in E_{\vec{\ell}}$ are
even (see, e.g., \cite[Example 5.2.21]{BrattelliRobinson}):

\begin{corollary}[$\vec{\ell}$--periodic states are even]
\label{coro.even copy(1)}\mbox{ }\newline
\index{Even!states}Let $\rho \in E_{%
\vec{\ell}}$ with GNS\ representation $(\mathcal{H}_{\rho },\pi _{\rho
},\Omega _{\rho })$. Then, for all odd elements $A\in \mathcal{U}$, $P_{\rho
}\pi _{\rho }(A)P_{\rho }=0$.
\end{corollary}

\begin{proof}%
Since $\rho $ is $\vec{\ell}$--periodic, by Theorem \ref{GNS}, there are
unitary operators $\{U_{x}\}_{x\in \mathbb{Z}_{\vec{\ell}}^{d}}$ acting on $%
\mathcal{H}_{\rho }$ and defining a representation of $(\mathbb{Z}_{\vec{\ell%
}}^{d},+)$ such that $U_{x}\Omega _{\rho }=\Omega _{\rho }$ and $\pi _{\rho
}(\alpha _{x}(A))=U_{x}\pi _{\rho }(A)U_{x}^{\ast }$ for all $x\in \mathbb{Z}%
_{\vec{\ell}}^{d}$. The $\ast $--automorphism $\alpha _{x}$ is defined by (%
\ref{transl}). If $A\in \mathcal{U}$ is odd, i.e., $\sigma _{\pi }(A)=-A$
(cf. (\ref{definition of gauge})), then%
\begin{equation*}
\lim\limits_{|x|\rightarrow \infty }(A^{\ast }\alpha _{x}(A)+\alpha
_{x}(A)A^{\ast })=0.
\end{equation*}%
Consequently, by using Theorem \ref{vonN}%
\index{von Neumann ergodic theorem} and observing that $U_{x}P_{\rho
}=P_{\rho }U_{x}=P_{\rho }$, for any $x\in 
\vec{\ell}\cdot \mathbb{Z}^{d}$, 
\begin{equation*}
(P_{\rho }\pi _{\rho }(A)^{\ast }P_{\rho })(P_{\rho }\pi _{\rho }(A)P_{\rho
})+(P_{\rho }\pi _{\rho }(A)P_{\rho })(P_{\rho }\pi _{\rho }(A)^{\ast
}P_{\rho })=0.
\end{equation*}%
Both terms on the l.h.s. of the last equality are positive. Therefore, if $%
A\in \mathcal{U}$ is odd then $P_{\rho }\pi _{\rho }(A)P_{\rho }=0$.%
\end{proof}%

The set $E_{\vec{\ell}}$ is clearly convex, weak$^{\ast }$--compact, and
also metrizable, by Theorem \ref{Metrizability}. By using the Choquet theorem%
\index{Choquet theorem} (Theorem \ref{theorem choquet bis}), each state $%
\rho \in E_{%
\vec{\ell}}$ has a decomposition in terms of states in the (non--empty) set $%
\mathcal{E}_{\vec{\ell}}$ of extreme points of $E_{\vec{\ell}}$. The Choquet
decomposition is, generally, not unique. However, in the particular case of
the convex set $E_{\vec{\ell}}$ the uniqueness of this decomposition follows
from the von Neumann ergodic theorem (Theorem \ref{vonN}):

\begin{lemma}[Uniqueness of the Choquet decomposition in $E_{\vec{\ell}}$]
\label{choquet unique}\mbox{ }\newline
For any $\rho \in E_{\vec{\ell}}$, the probability measure $\mu _{\rho }$
given by Theorem \ref{theorem choquet bis} is unique and norm preserving in
the sense that $\Vert \rho -\rho ^{\prime }\Vert =\Vert \mu _{\rho }-\mu
_{\rho ^{\prime }}\Vert $ for any $\rho ,\rho ^{\prime }\in E_{\vec{\ell}}$.
Here, $\Vert \rho -\rho ^{\prime }\Vert $ and $\Vert \mu _{\rho }-\mu _{\rho
^{\prime }}\Vert $ stand for the norms of $(\rho -\rho ^{\prime })$ and $%
(\mu _{\rho }-\mu _{\rho ^{\prime }})$ seen as linear functionals.
\end{lemma}

\begin{proof}
Observe that the map $\rho \mapsto \mu _{\rho }$ is norm preserving, by \cite%
[Theorem IV.4.1]{Israel}. See also \cite[Corollary IV.4.2]{Israel} for the
special case of spin systems. To prove the uniqueness of $\mu _{\rho }$, we
adapt here the proof given in \cite[Theorem IV.3.3]{Israel} for quantum spin
systems to our case of Fermi systems. For all $A\in \mathcal{U}$, let the
(affine) weak$^{\ast }$--continuous map%
\begin{equation*}
\rho \mapsto \hat{A}(\rho ):=\rho (A)
\end{equation*}%
from the set $E_{\vec{\ell}}$ to $\mathbb{C}$. The family $\{\hat{A}\}_{A\in 
\mathcal{U}}$ of continuous functionals separates states, i.e., for all $%
\rho ,\rho ^{\prime }\in E_{\vec{\ell}}$ with $\rho \not=\rho ^{\prime }$,
there is $A\in \mathcal{U}$ such that $\hat{A}(\rho )\not=\hat{A}(\rho
^{\prime })$. Thus, by the Stone--Weierstrass theorem, the uniqueness of the
probability measure $\mu _{\rho }$ of Theorem \ref{theorem choquet bis} is
equivalent to the uniqueness of the complex numbers 
\begin{equation}
\mu _{\rho }(\hat{A}_{1}\cdots \hat{A}_{n})=\int_{E_{\vec{\ell}}}\mathrm{d}%
\mu _{\rho }(\hat{\rho})\ \hat{\rho}(A_{1})\cdots \hat{\rho}(A_{n}),\quad
A_{1},\ldots ,A_{n}\in \mathcal{U},\;n\in \mathbb{N}.
\label{int choquet unique}
\end{equation}%
By the von Neumann ergodic theorem%
\index{von Neumann ergodic theorem} (Theorem \ref{vonN}), for any $\rho \in
E_{%
\vec{\ell}}$, $A_{1},\ldots ,A_{n}\in \mathcal{U}$ and $n\in \mathbb{N}$, 
\begin{equation*}
\lim\limits_{L\rightarrow \infty }\rho \left( (A_{1})_{L,\vec{\ell}}\cdots
(A_{n})_{L,\vec{\ell}}\right) =\langle \Omega _{\rho },\pi _{\rho
}(A_{1})P_{\rho }\pi _{\rho }(A_{2})P_{\rho }\cdots P_{\rho }\pi _{\rho
}(A_{n})\Omega _{\rho }\rangle .
\end{equation*}%
Recall that $A_{L,\vec{\ell}}$ is defined by (\ref{definition de A L}) for
any $A\in \mathcal{U}$, $L\in \mathbb{N}$, and any $\vec{\ell}\in \mathbb{N}%
^{d}$. By Lemma \ref{lemma extremal.ergodic} below, the projection $P_{\rho
} $ is one--dimensional with $\mathrm{ran\;}P_{\rho }=%
\mathbb{C}
\Omega _{\rho }$ whenever $\rho \in \mathcal{E}_{\vec{\ell}}$ is extreme in $%
E_{\vec{\ell}}$. In particular, for all extreme states $\hat{\rho}\in 
\mathcal{E}_{\vec{\ell}}$ and all $A_{1},\ldots ,A_{n}\in \mathcal{U}$, $%
n\in \mathbb{N}$, 
\begin{equation}
\hat{\rho}(A_{1})\cdots \hat{\rho}(A_{n})=\lim\limits_{L\rightarrow \infty }%
\hat{\rho}\left( (A_{1})_{L,\vec{\ell}}\cdots (A_{n})_{L,\vec{\ell}}\right) .
\label{ergodicity eq}
\end{equation}%
Hence, as $\mu _{\rho }(E_{\vec{\ell}}\backslash \mathcal{E}_{\vec{\ell}})=0$
(Theorem \ref{theorem choquet bis}), by using (\ref{int choquet unique})
together with Lebesgue's dominated convergence, it follows that, for any $%
A_{1},\ldots ,A_{n}\in \mathcal{U}$ with $n\in \mathbb{N}$, the complex
number%
\begin{eqnarray*}
\mu _{\rho }(\hat{A}_{1}\cdots \hat{A}_{n}) &=&\lim\limits_{L\rightarrow
\infty }\int_{E_{\vec{\ell}}}\mathrm{d}\mu _{\rho }\left( \hat{\rho}\right)
\ \hat{\rho}\left( (A_{1})_{L,\vec{\ell}}\cdots (A_{n})_{L,\vec{\ell}}\right)
\\
&=&\lim\limits_{L\rightarrow \infty }\rho \left( (A_{1})_{L,\vec{\ell}%
}\cdots (A_{n})_{L,\vec{\ell}}\right)
\end{eqnarray*}%
is uniquely determined. 
\end{proof}%

\noindent As a consequence, the set $E_{\vec{\ell}}$ is a (Choquet) simplex,
see Definition \ref{gamm regularisation copy(1)} and Theorem \ref{theorem
choquet bis copy(1)}.

\section{The set $\mathcal{E}_{\vec{\ell}}$ of extreme states of $E_{\vec{%
\ell}}$\label{Section theorem ergodic extremalbis}%
\index{States!extreme}%
\index{States!ergodic}}

We want to prove next that all extreme states are ergodic w.r.t. the
space--average (\ref{definition de A L}) (Definition \ref{def:egodic}) and
conversely. The fact that all ergodic states are extreme is not difficult to
verify:

\begin{lemma}[Ergodicity implies extremality]
\label{ergodic.extremal}\mbox{ }\newline
Any ergodic state $\rho \in E_{%
\vec{\ell}}$ is extreme in $E_{\vec{\ell}}$, i.e., $\rho \in \mathcal{E}_{%
\vec{\ell}}$.
\end{lemma}

\begin{proof}
If $\rho \notin \mathcal{E}_{\vec{\ell}}$ is not extreme, there are two
states $\rho _{1},\rho _{2}\in E_{\vec{\ell}}$ with $\rho =\frac{1}{2}\rho
_{1}+\frac{1}{2}\rho _{2}$ and $\rho _{1}(A)\not=\rho _{2}(A)$ for some $%
A=A^{\ast }\in \mathcal{U}$. Then%
\begin{equation}
|\rho (A)|^{2}<\frac{1}{2}|\rho _{1}(A)|^{2}+\frac{1}{2}|\rho _{2}(A)|^{2}.
\label{inequality debile0}
\end{equation}%
For all $\vec{\ell}\in \mathbb{N}^{d}$ and any state $\rho \in E_{\vec{\ell}%
} $ with GNS\ representation $(\mathcal{H}_{\rho },\pi _{\rho },\Omega
_{\rho })$, by Theorem \ref{GNS} for $G=(\mathbb{Z}_{\vec{\ell}}^{d},+)$ and
Theorem \ref{vonN}%
\index{von Neumann ergodic theorem}, we get%
\begin{equation}
\Delta _{A,%
\vec{\ell}}\left( \rho \right) :=\lim\limits_{L\rightarrow \infty }\rho
(A_{L,\vec{\ell}}^{\ast }A_{L,\vec{\ell}})=\underset{L\rightarrow \infty }{%
\lim }\Vert P_{\rho }^{(L)}\pi _{\rho }(A)\Omega _{\rho }\Vert ^{2}=\Vert
P_{\rho }\pi _{\rho }(A)\Omega _{\rho }\Vert ^{2}.  \label{ergodic eqbis}
\end{equation}%
Using Cauchy--Schwarz inequality together with $P_{\rho }\Omega _{\rho
}=\Omega _{\rho }$ (Theorem \ref{vonN}),%
\begin{equation*}
|\rho (A)|^{2}=|\langle \Omega _{\rho },P_{\rho }\pi _{\rho }(A)\Omega
_{\rho }\rangle |^{2}\leq \Vert P_{\rho }\pi _{\rho }(A)\Omega _{\rho }\Vert
^{2}=\Delta _{A,\vec{\ell}}\left( \rho \right)
\end{equation*}%
for any state $\rho \in E_{\vec{\ell}}$. Applying the last inequality to
states $\rho _{1}$ and $\rho _{2}$ we conclude from (\ref{inequality debile0}%
) that%
\begin{equation*}
|\rho (A)|^{2}<\frac{1}{2}\Delta _{A,\vec{\ell}}\left( \rho _{1}\right) +%
\frac{1}{2}\Delta _{A,\vec{\ell}}\left( \rho _{2}\right) =\Delta _{A,\vec{%
\ell}}\left( \rho \right) .
\end{equation*}%
It follows that $\rho \notin \mathcal{E}_{\vec{\ell}}$ is not ergodic. 
\end{proof}%

The last lemma is elementary, but it implies an essential topological
property of the set $\mathcal{E}_{\vec{\ell}}$ of extreme points of the
convex and weak$^{\ast }$--compact set $E_{\vec{\ell}}$:

\begin{corollary}[Density of the set $\mathcal{E}_{\vec{\ell}}$ of extreme
points of $E_{\vec{\ell}}$]
\label{lemma density of extremal points}\mbox{ }\newline
For any $\vec{\ell}\in \mathbb{N}^{d}$, the set $\mathcal{E}_{\vec{\ell}}$
is a $G_{\delta }$ weak$^{\ast }$--dense subset of $E_{\vec{\ell}}$.
\end{corollary}

\begin{proof}
The proof of this lemma is a slight adaptation of the proof of \cite[Lemma
IV.3.2.]{Israel} for quantum spin systems to the case of even states over
the fermion algebra $\mathcal{U}$. It is a \emph{pivotal} proof in the
sequel.

The set $\mathcal{E}_{\vec{\ell}}$ of extreme points of $E_{\vec{\ell}}$ is
a $G_{\delta }$ set, by Theorem \ref{lemma Milman} (i), as $E_{\vec{\ell}}$
is metrizable. Thus, it suffices to prove that $\mathcal{E}_{\vec{\ell}}$ is
dense in $E_{\vec{\ell}}$. For any $\rho \in E_{\vec{\ell}}$, we define the
state $\tilde{\rho}_{n}$ to be the restriction $\rho _{\Lambda _{n}}\in
E_{\Lambda _{n}}$ on the box%
\begin{equation}
\Lambda _{n,\vec{\ell}}:=\left\{ x=(x_{1},\cdots ,x_{d})\in \mathbb{Z}%
^{d}:\left\vert x_{i}\right\vert \leq n\ell _{i}\right\}
\label{equation toto0}
\end{equation}%
seen as a $(2n+1)\vec{\ell}$--periodic state. This is possible, by \cite[%
Theorem 11.2.]{Araki-Moriya}, because any $\vec{\ell}$--periodic state is
even, by Corollary \ref{coro.even copy(1)}. From the state $\tilde{\rho}%
_{n}\in E_{(2n+1)\vec{\ell}}$ we define next the $\vec{\ell}$--periodic
state 
\begin{equation}
\hat{\rho}_{n}:=\frac{1}{|\Lambda _{n,\vec{\ell}}\cap \mathbb{Z}_{\vec{\ell}%
}^{d}|}\sum\limits_{x\in \Lambda _{n,\vec{\ell}}\cap \mathbb{Z}_{\vec{\ell}%
}^{d}}\tilde{\rho}_{n}\circ \alpha _{x}\in E_{\vec{\ell}}.
\label{equation toto}
\end{equation}%
Clearly, the space--averaged state $\hat{\rho}_{n}$ converges towards $\rho
\in E_{\vec{\ell}}$ w.r.t. the weak$^{\ast }$--topology and we prove below
that $\hat{\rho}_{n}\in \mathcal{E}_{\vec{\ell}}$ by using Lemma \ref%
{ergodic.extremal}.

Indeed, for any $A\in \mathcal{U}_{0}$, there is a positive constant $C>0$
such that%
\begin{equation*}
\tilde{\rho}_{n}\left( \alpha _{x}(A^{\ast })\alpha _{y}(A)\right) =\tilde{%
\rho}_{n}\left( \alpha _{x}(A^{\ast })\right) \tilde{\rho}_{n}\left( \alpha
_{y}(A)\right)
\end{equation*}%
whenever $d(x,y)\geq C$. Here, $d:\mathfrak{L}\times \mathfrak{L}\rightarrow
\lbrack 0,\infty )$ is the Euclidean metric defined on the lattice $%
\mathfrak{L}:=\mathbb{Z}^{d}$ by (\ref{def.dist}). Using the space--average $%
A_{L,\vec{\ell}}$ defined by (\ref{definition de A L}) we then deduce that%
\begin{eqnarray}
\tilde{\rho}_{n}(A_{L,\vec{\ell}}^{\ast }A_{L,\vec{\ell}}) &=&\frac{1}{%
|\Lambda _{L}\cap \mathbb{Z}_{\vec{\ell}}^{d}|^{2}}\sum\limits_{x,y\in
\Lambda _{L}\cap \mathbb{Z}_{\vec{\ell}}^{d}}\tilde{\rho}_{n}\left( \alpha
_{x}(A^{\ast })\right) \tilde{\rho}_{n}\left( \alpha _{y}(A)\right)
\label{asymptotics1} \\
&&+\mathcal{O}(L^{-d}).  \notag
\end{eqnarray}%
Since $\hat{\rho}_{n}\in E_{\vec{\ell}}$ is a $\vec{\ell}$--periodic state,
for any $A\in \mathcal{U}_{0}$, one has that%
\begin{equation*}
\frac{1}{|\Lambda _{L}\cap \mathbb{Z}_{\vec{\ell}}^{d}|}\sum\limits_{x\in
\Lambda _{L}\cap \mathbb{Z}_{\vec{\ell}}^{d}}\tilde{\rho}_{n}\left( \alpha
_{x}(A)\right) =\hat{\rho}_{n}\left( A\right) +\mathcal{O}(L^{-1})
\end{equation*}%
which combined with the asymptotics (\ref{asymptotics1}) implies that%
\begin{equation*}
\lim\limits_{L\rightarrow \infty }\tilde{\rho}_{n}(A_{L,\vec{\ell}}^{\ast
}A_{L,\vec{\ell}})=\left\vert \hat{\rho}_{n}\left( A\right) \right\vert ^{2}.
\end{equation*}%
Using this last equality we then obtain from (\ref{equation toto}) that, for
any $A\in \mathcal{U}_{0}$,%
\begin{equation}
\lim\limits_{L\rightarrow \infty }\hat{\rho}_{n}(A_{L,\vec{\ell}}^{\ast
}A_{L,\vec{\ell}})=\left\vert \hat{\rho}_{n}\left( A\right) \right\vert ^{2}
\label{asymptotics2}
\end{equation}%
because $\hat{\rho}_{n}\in E_{\vec{\ell}}$ and 
\begin{equation}
\alpha _{x}(A_{L,\vec{\ell}}^{\ast }A_{L,\vec{\ell}})=\left( \alpha
_{x}\left( A\right) \right) _{L,\vec{\ell}}^{\ast }\left( \alpha _{x}\left(
A\right) \right) _{L,\vec{\ell}}\   \label{asymptotics3}
\end{equation}%
for all $x\in \mathbb{Z}^{d}$. Since the set $\mathcal{U}_{0}$ is dense in
the fermion algebra $\mathcal{U}$, we can extend (\ref{asymptotics2}) to any 
$A\in \mathcal{U}$ which shows that the state $\hat{\rho}_{n}\in E_{\vec{\ell%
}}$ is ergodic and thus extreme in $E_{\vec{\ell}}$, by Lemma \ref%
{ergodic.extremal}. 
\end{proof}%

We show now the converse of Lemma \ref{ergodic.extremal} which is not as
obvious as the proof of Lemma \ref{ergodic.extremal}. Take, for instance,
the trivial action of the group $(\mathbb{Z}_{\vec{\ell}}^{d},+)$ on the $%
C^{\ast }$-algebra $\mathcal{U}$ given by $\tilde{\alpha}_{x}:A\mapsto A$
for all $x\in \mathbb{Z}_{\vec{\ell}}^{d}$. Observe that w.r.t. this choice,
the set of invariant states is simply the set $E$ of all states. Then, by
the proof of Lemma \ref{ergodic.extremal}, any ergodic state w.r.t. this
action is again an extreme point of the set of all states. But, generally,
extreme states are not ergodic w.r.t. the trivial action of $\mathbb{Z}_{%
\vec{\ell}}^{d}$: Consider for simplicity the case of quantum spin systems
(cf. Remark \ref{Quantum spin systems}). For a given element $A\in \mathcal{U%
}$ such that $A^{\ast }A\not=A$, we can always find a state $\rho $
satisfying $\rho (A^{\ast }A)\not=|\rho (A)|^{2}$ and thus, because the set $%
\mathcal{E}$ of extreme states of $E$ is weak$^{\ast }$--dense in $E$ (see 
\cite[Example 4.1.31.]{BrattelliRobinsonI}), there is an extreme state with
this property .

In order to get the equivalence between ergodicity and extremality of
states, the \emph{asymptotic abelianess}%
\index{Asymptotic abelianess|textbf} of the even sub--algebra $\mathcal{U}%
^{+}$ (\ref{definition of even operators}), i.e., the fact that 
\begin{equation}
\lim\limits_{|x|\rightarrow \infty }[A,\alpha _{x}(B)]=0\mathrm{\quad for\
any}\ A,B\in \mathcal{U}^{+},  \label{asymtptic abelian}
\end{equation}%
is crucial.

Indeed, for any state $\rho \in E_{%
\vec{\ell}}$ with GNS\ representation $(\mathcal{H}_{\rho },\pi _{\rho
},\Omega _{\rho })$%
\index{States!GNS representation}, let us first consider the von Neumann
algebra%
\begin{equation*}
\mathfrak{R}_{\rho }:=\left[ \pi _{\rho }(\mathcal{U})\cup \{U_{x}\}_{x\in 
\mathbb{Z}_{%
\vec{\ell}}^{d}}\right] ^{\prime \prime }\subseteq {B(\mathcal{H}_{\rho })}.
\end{equation*}%
Here, $\{U_{x}\}_{x\in \mathbb{Z}_{\vec{\ell}}^{d}}$ are the unitary
operators of Theorem \ref{GNS} with $G=(\mathbb{Z}_{\vec{\ell}}^{d},+)$.
This von Neumann algebra is related to the projection $P_{\rho }:=P$ defined
via Theorem \ref{vonN} for $\mathcal{H}=\mathcal{H}_{\rho }$:

\begin{lemma}[Properties of the von Neumann algebra $\mathfrak{R}_{\protect%
\rho }$]
\label{vonN.abelian}\mbox{ }\newline
For any $\vec{\ell}$--periodic state $\rho \in E_{\vec{\ell}}$, $P_{\rho
}\in \mathfrak{R}_{\rho }$ and $P_{\rho }\mathfrak{R}_{\rho }P_{\rho }$ is
an abelian von Neumann algebra on $P_{\rho }\mathcal{H}_{\rho }$.
\end{lemma}

\begin{proof}%
On the one hand, by Theorem \ref{vonN}, the projection $P_{\rho }$ is the
strong limit of linear combinations of unitary operators $U_{x}$ for $x\in 
\mathbb{Z}_{\vec{\ell}}^{d}$ and so, $P_{\rho }\in \mathfrak{R}_{\rho }$. On
the other hand, if $\mathfrak{M}$ is a von Neumann algebra on a Hilbert
space $\mathcal{H}$ and $P$ is any projection from $\mathfrak{M}$, the set $P%
\mathfrak{M}P$ is a von Neumann algebra on $P\mathcal{H}$. See, e.g., \cite[%
Lemma IV.2.5]{Israel}. Therefore, it remains to show that $P_{\rho }%
\mathfrak{R}_{\rho }P_{\rho }$ is abelian. To prove it we adapt now the
proof of \cite[Lemma IV.2.6]{Israel} -- performed for quantum spin systems
-- to the case where $\mathcal{U}$ is a fermion algebra. In particular, we
show first that $P_{\rho }\mathfrak{R}_{\rho }P_{\rho }=[P_{\rho }\pi _{\rho
}(\mathcal{U})P_{\rho }]^{\prime \prime }$ and then the abelianess of $%
[P_{\rho }\pi _{\rho }(\mathcal{U})P_{\rho }]^{\prime \prime }$.

Since $U_{x}\pi _{\rho }(A)=\pi _{\rho }(\alpha _{x}(A))U_{x}$, it follows
that each element $B\in \mathfrak{R}_{\rho }$ is the strong limit as $%
n\rightarrow \infty $ of a sequence of elements of the form 
\begin{equation*}
B_{n}:=\sum\limits_{j}U_{x_{j}}\pi _{\rho }(A_{j})
\end{equation*}%
with $x_{j}\in \mathbb{Z}_{\vec{\ell}}^{d}$ and $A_{j}\in \mathcal{U}$. In
particular, by using Theorem \ref{vonN}, each element of $P_{\rho }\mathfrak{%
R}_{\rho }P_{\rho }$ is the strong limit as $n\rightarrow \infty $ of
elements of the form%
\begin{equation*}
P_{\rho }B_{n}P_{\rho }=P_{\rho }\pi _{\rho }(\sum\limits_{j}A_{j})P_{\rho }.
\end{equation*}%
In other words, since $P_{\rho }\pi _{\rho }(\mathcal{U})P_{\rho }\subseteq
P_{\rho }\mathfrak{R}_{\rho }P_{\rho }$ is clear, we deduce from the last
equality that $P_{\rho }\mathfrak{R}_{\rho }P_{\rho }=[P_{\rho }\pi _{\rho }(%
\mathcal{U})P_{\rho }]^{\prime \prime }$ as $[P_{\rho }\pi _{\rho }(\mathcal{%
U})P_{\rho }]^{\prime \prime }$ is the strong closure of $P_{\rho }\pi
_{\rho }(\mathcal{U})P_{\rho }$.

Take now two local even elements $A,B\in \mathcal{U}_{\Lambda }\cap \mathcal{%
U}^{+}$ with $\Lambda \in \mathcal{P}_{f}(\mathfrak{L})$. Then via Theorem %
\ref{vonN}, for all $\xi \in \mathcal{H}_{\rho }$,%
\begin{eqnarray}
\lefteqn{\left[ \left( P_{\rho }\pi _{\rho }(A)P_{\rho }\right) \left(
P_{\rho }\pi _{\rho }(B)P_{\rho }\right) -\left( P_{\rho }\pi _{\rho
}(B)P_{\rho }\right) \left( P_{\rho }\pi _{\rho }(A)P_{\rho }\right) \right]
\xi }  \label{vonN.commut.2bis} \\
&=&\lim\limits_{L\rightarrow \infty }\left[ \left( P_{\rho }\pi _{\rho
}(A)P_{\rho }^{(L)}\pi _{\rho }(B)P_{\rho }\right) -\left( P_{\rho }\pi
_{\rho }(B)P_{\rho }^{(L)}\pi _{\rho }(A)P_{\rho }\right) \right] \xi  \notag
\\
&=&\lim\limits_{L\rightarrow \infty }\frac{1}{|\Lambda _{L}\cap \mathbb{Z}_{%
\vec{\ell}}^{d}|}\sum\limits_{x\in \Lambda _{L}\cap \mathbb{Z}_{\vec{\ell}%
}^{d}}P_{\rho }[A,\alpha _{x}(B)]P_{\rho }=0\   \notag
\end{eqnarray}%
because $[A,\alpha _{x}(B)]=0$ for any $x\in \mathbb{Z}^{d}$ such that $%
d(x,0)\geq 2|\Lambda |$, see (\ref{def.dist}) for the definition of the
metric $d$. From Corollary \ref{coro.even copy(1)}, recall that $P_{\rho
}\pi _{\rho }(A)P_{\rho }=0$ for any odd element $A\in \mathcal{U}$.
Therefore, by combining this with the density of the $\ast $--algebra $%
\mathcal{U}_{0}\subseteq \mathcal{U}$ of local elements, we can extend the
equality (\ref{vonN.commut.2bis}) to any $A,B\in \mathcal{U}$, i.e., for all 
$A,B\in \mathcal{U}$, 
\begin{equation*}
\lbrack P_{\rho }\pi _{\rho }(A)P_{\rho },P_{\rho }\pi _{\rho }(B)P_{\rho
}]=0.
\end{equation*}%
In other words, $P_{\rho }\pi _{\rho }(\mathcal{U})P_{\rho }$ is abelian.
Since $P_{\rho }\pi _{\rho }(\mathcal{U})P_{\rho }$ is strongly dense in $%
[P_{\rho }\pi _{\rho }(\mathcal{U})P_{\rho }]^{\prime \prime }=P_{\rho }%
\mathfrak{R}_{\rho }P_{\rho }$, the von Neumann algebra $P_{\rho }\mathfrak{R%
}_{\rho }P_{\rho }$ is itself abelian. 
\end{proof}%

We are now in position to show that all extreme points $\rho \in \mathcal{E}%
_{\vec{\ell}}$ of $E_{\vec{\ell}}$ are ergodic.

\begin{lemma}[Extremality implies ergodicity]
\label{lemma extremal.ergodic}\mbox{ }\newline
For any extreme state $\hat{\rho}\in \mathcal{E}_{\vec{\ell}}$ of $E_{\vec{%
\ell}}$, $P_{\hat{\rho}}$ is the orthogonal projection on the
one--dimensional sub--space generated by $\Omega _{\hat{\rho}}$. In
particular, any state $\hat{\rho}\in \mathcal{E}_{\vec{\ell}}$ is ergodic.
\end{lemma}

\begin{proof}
For any $\hat{\rho}\in \mathcal{E}_{\vec{\ell}}$, observe that the von
Neumann algebra $\mathfrak{R}_{\hat{\rho}}$ is irreducible, i.e., $\mathfrak{%
R}_{\hat{\rho}}^{\prime }=\mathbb{C}\,\mathbf{1}$. Indeed, by contradiction,
assume that $\mathfrak{R}_{\hat{\rho}}^{\prime }$ is strictly larger than
its sub--algebra $\mathbb{C}\,\mathbf{1}$. Then there is at least one
non--trivial (orthogonal) projection $P\in \mathfrak{R}_{\hat{\rho}}^{\prime
}$. By cyclicity of $\Omega _{\hat{\rho}}$ w.r.t. $\mathfrak{R}_{\hat{\rho}%
}^{\prime \prime }$, $P\Omega _{\hat{\rho}}\not=0$ and thus $\langle \Omega
_{\hat{\rho}},P\Omega _{\hat{\rho}}\rangle =\Vert P\Omega _{\hat{\rho}}\Vert
_{2}^{2}>0$. Similarly, $\langle \Omega _{\hat{\rho}},(\mathbf{1}-P)\Omega _{%
\hat{\rho}}\rangle >0$. Define the following continuous linear functionals
on $\mathcal{U}$: 
\begin{eqnarray*}
\rho _{1}(A):= &&\langle \Omega _{\hat{\rho}},P\Omega _{\hat{\rho}}\rangle
^{-1}\langle \Omega _{\hat{\rho}},P\pi _{\hat{\rho}}(A)\Omega _{\hat{\rho}%
}\rangle , \\
\rho _{2}(A):= &&\langle \Omega _{\hat{\rho}},(\mathbf{1}-P)\Omega _{\hat{%
\rho}}\rangle ^{-1}\langle \Omega _{\hat{\rho}},(\mathbf{1}-P)\pi _{\hat{\rho%
}}(A)\Omega _{\hat{\rho}}\rangle .
\end{eqnarray*}%
Observe that, by cyclicity of $\Omega _{\hat{\rho}}$ w.r.t. $\pi _{\hat{\rho}%
}(\mathcal{U})$, $\rho _{1}\neq \rho _{2}$. Since $U_{x}\Omega _{\hat{\rho}%
}=\Omega _{\hat{\rho}}$ and $P$ commutes by definition with $\pi _{\hat{\rho}%
}(A)\ $and $U_{x}$ for all $A\in \mathcal{U}$ and $x\in \mathbb{Z}_{\vec{\ell%
}}^{d}$, the functionals $\rho _{1}$ and $\rho _{2}\ $belong to $E_{\vec{\ell%
}}$, whereas%
\begin{equation*}
\hat{\rho}=\langle \Omega _{\hat{\rho}},P\Omega _{\hat{\rho}}\rangle \rho
_{1}+\langle \Omega _{\hat{\rho}},(\mathbf{1}-P)\Omega _{\hat{\rho}}\rangle
\rho _{2}\ .
\end{equation*}%
Since $\langle \Omega _{\hat{\rho}},(\mathbf{1}-P)\Omega _{\hat{\rho}%
}\rangle >0$ and $\langle \Omega _{\hat{\rho}},P\Omega _{\hat{\rho}}\rangle
>0$, this last equality contradicts the fact that $\hat{\rho}\in \mathcal{E}%
_{\vec{\ell}}$. Therefore, $\mathfrak{R}_{\hat{\rho}}^{\prime }=\mathbb{C}\,%
\mathbf{1}$ whenever $\hat{\rho}\in \mathcal{E}_{\vec{\ell}}$.

Observe now that%
\begin{equation}
\lbrack P_{\hat{\rho}}\mathfrak{R}_{\hat{\rho}}P_{\hat{\rho}}]^{\prime }=P_{%
\hat{\rho}}\mathfrak{R}_{\hat{\rho}}^{\prime }P_{\hat{\rho}}=\mathbb{C}\,P_{%
\hat{\rho}}\ .  \label{PRP=C1}
\end{equation}%
Here we use that, for any von Neumann algebra $\mathfrak{M}$ and any
orthogonal projection $P\in \mathfrak{M}$, $[P\,\mathfrak{M}\,P]^{\prime
}=P\,\mathfrak{M}^{\prime }P$, see, e.g., \cite[Lemma IV.2.5]{Israel}. By
Lemma \ref{vonN.abelian}, the von Neumann algebra $P_{\hat{\rho}}\mathfrak{R}%
_{\hat{\rho}}P_{\hat{\rho}}$ is abelian. In particular, from (\ref{PRP=C1}), 
\begin{equation*}
P_{\hat{\rho}}\mathfrak{R}_{\hat{\rho}}P_{\hat{\rho}}\subseteq P_{\hat{\rho}}%
\mathfrak{R}_{\hat{\rho}}^{\prime }P_{\hat{\rho}}=\mathbb{C}\,P_{\hat{\rho}}
\end{equation*}%
which implies that $P_{\hat{\rho}}\mathfrak{R}_{\hat{\rho}}P_{\hat{\rho}}=%
\mathbb{C}\,P_{\hat{\rho}}$. This yields%
\begin{equation*}
P_{\hat{\rho}}\pi _{\hat{\rho}}(A)\Omega _{\hat{\rho}}=P_{\hat{\rho}}\pi _{%
\hat{\rho}}(A)P_{\hat{\rho}}\Omega _{\hat{\rho}}\in \mathbb{C}\,P_{\hat{\rho}%
}\Omega _{\hat{\rho}}=\mathbb{C}\,\Omega _{\hat{\rho}}
\end{equation*}%
for any $A\in \mathcal{U}$. In other words, by cyclicity of $\Omega _{\hat{%
\rho}}$, $P_{\hat{\rho}}\mathcal{H}_{\hat{\rho}}=\mathbb{C}\,\Omega _{\hat{%
\rho}}$ and thus 
\begin{eqnarray*}
\Vert P_{\hat{\rho}}\pi _{\hat{\rho}}(A)\Omega _{\hat{\rho}}\Vert ^{2}
&=&\langle P_{\hat{\rho}}\pi _{\hat{\rho}}(A)\Omega _{\hat{\rho}},P_{\hat{%
\rho}}\pi _{\hat{\rho}}(A)\Omega _{\hat{\rho}}\rangle \\
&=&\langle P_{\hat{\rho}}\pi _{\hat{\rho}}(A)\Omega _{\hat{\rho}},\Omega _{%
\hat{\rho}}\rangle \langle \Omega _{\hat{\rho}},P_{\hat{\rho}}\pi _{\hat{\rho%
}}(A)\Omega _{\hat{\rho}}\rangle \\
&=&\langle \pi _{\hat{\rho}}(A)\Omega _{\hat{\rho}},\Omega _{\hat{\rho}%
}\rangle \langle \Omega _{\hat{\rho}},\pi _{\hat{\rho}}(A)\Omega _{\hat{\rho}%
}\rangle
\end{eqnarray*}%
implying, by (\ref{ergodic eqbis}), that any state $\hat{\rho}\in \mathcal{E}%
_{\vec{\ell}}$ is ergodic. 
\end{proof}%

As we can relate the ergodicity with the so--called \emph{strongly clustering%
} property \cite[Section 4.3.2]{BrattelliRobinsonI}, we deduce from Lemma %
\ref{lemma extremal.ergodic} that any extreme state $\hat{\rho}\in \mathcal{E%
}_{\vec{\ell}}$ is strongly clustering:

\begin{corollary}[Extreme states are strongly clustering]
\label{Corollary Extremal states - strongly clustering}\mbox{ }\newline
\index{States!strongly clustering}Any extreme state $%
\hat{\rho}\in \mathcal{E}_{\vec{\ell}}$ is strongly clustering, i.e., for
all $A,B\in \mathcal{U}$, 
\begin{equation}
\lim\limits_{L\rightarrow \infty }\frac{1}{|\Lambda _{L}\cap \mathbb{Z}_{%
\vec{\ell}}^{d}|}\sum\limits_{y\in \Lambda _{L}\cap \mathbb{Z}_{\vec{\ell}%
}^{d}}\hat{\rho}\left( \alpha _{x}(A)\alpha _{y}(B)\right) =\hat{\rho}(A)%
\hat{\rho}(B)  \label{strongly.clustering}
\end{equation}%
uniformly in $x\in \mathbb{Z}_{\vec{\ell}}^{d}$.
\end{corollary}

\begin{proof}
This corollary can directly be seen from Lemma \ref{lemma extremal.ergodic}
combined with Theorem \ref{vonN} because%
\begin{eqnarray*}
\lim\limits_{L\rightarrow \infty }\frac{1}{|\Lambda _{L}\cap \mathbb{Z}_{%
\vec{\ell}}^{d}|}\sum\limits_{y\in \Lambda _{L}\cap \mathbb{Z}_{\vec{\ell}%
}^{d}}\hat{\rho}\left( \alpha _{x}(A)\alpha _{y}(B)\right)
&=&\lim\limits_{L\rightarrow \infty }\langle U_{x}\pi _{\hat{\rho}}(A^{\ast
})\Omega _{\hat{\rho}},P_{\hat{\rho}}^{(L)}\pi _{\hat{\rho}}(B)\Omega _{\hat{%
\rho}}\rangle \\
&=&\langle U_{x}\pi _{\hat{\rho}}(A^{\ast })\Omega _{\hat{\rho}},P_{\hat{\rho%
}}\pi _{\hat{\rho}}(B)\Omega _{\hat{\rho}}\rangle \\
&=&\langle \Omega _{\hat{\rho}},\pi _{\hat{\rho}}(A)\Omega _{\hat{\rho}%
}\rangle \langle \Omega _{\hat{\rho}},\pi _{\hat{\rho}}(B)\Omega _{\hat{\rho}%
}\rangle
\end{eqnarray*}%
for any $A,B\in \mathcal{U}$ and $x\in \mathbb{Z}_{\vec{\ell}}^{d}$. By
using Cauchy--Schwarz inequality, note that the limit $L\rightarrow \infty $
is uniform in $x\in \mathbb{Z}_{\vec{\ell}}^{d}$ because $P_{\hat{\rho}%
}^{(L)}$ converges strongly to the projection $P_{\hat{\rho}}$. See Theorem %
\ref{vonN}.%
\end{proof}%

Therefore, Theorem \ref{theorem ergodic extremal} is a consequence of
Lemmata \ref{ergodic.extremal} and \ref{lemma extremal.ergodic} together
with Corollary \ref{Corollary Extremal states - strongly clustering}.

\section{Properties of the space--averaging functional $\Delta _{A}$\label%
{Section properties of delta}%
\index{Space--averaging functional}}

We characterize now the properties of the space--averaging functional $%
\Delta _{A}$ defined in Definition \ref{definition de deltabis} for any $%
A\in \mathcal{U}$ because it is intimately related with the structure of the
set $E_{1}$ of t.i. states. We start by proving that this functional is
well--defined, even for $%
\vec{\ell}$--periodic states $\rho \in E_{\vec{\ell}}$:

\begin{lemma}[Well--definiteness of the map $\protect\rho \mapsto \Delta
_{A}\left( \protect\rho \right) $]
\label{Lemma1.vonN copy(2)}\mbox{ }\newline
For any $A\in \mathcal{U}$, the space--averaging functional $\Delta _{A}$ is
well--defined on the set $E_{\vec{\ell}}$ of $\vec{\ell}$--periodic states
for any $\vec{\ell}\in \mathbb{N}^{d}$ and it satisfies%
\begin{equation*}
\Delta _{A}\left( \rho \right) =\inf\limits_{(L,\cdots ,L)\in \vec{\ell}.%
\mathbb{N}^{d}}\left\{ \rho (A_{L}^{\ast }A_{L})\right\} \in \left[ |\rho
(A_{\vec{\ell}})|^{2},\Vert A\Vert ^{2}\right] .
\end{equation*}
\end{lemma}

\begin{proof}
Assume that $(L,\cdots ,L)\in \vec{\ell}.\mathbb{N}^{d}$. In the same way we
prove (\ref{ergodic eqbis}), for any state $\rho \in E_{\vec{\ell}}$ with
GNS\ representation $(\mathcal{H}_{\rho },\pi _{\rho },\Omega _{\rho })$, we
obtain, by using Theorem \ref{GNS}%
\index{States!GNS representation} for $G=(\mathbb{Z}_{%
\vec{\ell}}^{d},+)$ and Theorem \ref{vonN}%
\index{von Neumann ergodic theorem} for $\mathcal{H}=\mathcal{H}_{\rho }$,
that%
\begin{equation}
\lim\limits_{L\rightarrow \infty }\rho (A_{L}^{\ast }A_{L})=\underset{%
L\rightarrow \infty }{\lim }\Vert P_{\rho }^{(L)}\pi _{\rho }(A_{%
\vec{\ell}})\Omega _{\rho }\Vert ^{2}=\Vert P_{\rho }\pi _{\rho }(A_{\vec{%
\ell}})\Omega _{\rho }\Vert ^{2}\leq \Vert A\Vert ^{2}.  \label{ergodic eq}
\end{equation}%
The inequality $\Delta _{A}\left( \rho \right) \geq |\rho (A_{\vec{\ell}%
})|^{2}$ then follows by using the Cauchy--Schwarz inequality and $P_{\rho
}\Omega _{\rho }=\Omega _{\rho }$. Additionally, by using again Theorem \ref%
{vonN} we see that, for all $(L,\cdots ,L)\in \vec{\ell}.\mathbb{N}^{d}$,%
\begin{equation*}
\Vert P_{\rho }^{(L)}\pi _{\rho }(A_{\vec{\ell}})\Omega _{\rho }\Vert
^{2}\geq \Vert P_{\rho }P_{\rho }^{(L)}\pi _{\rho }(A_{\vec{\ell}})\Omega
_{\rho }\Vert ^{2}=\Vert P_{\rho }\pi _{\rho }(A_{\vec{\ell}})\Omega _{\rho
}\Vert ^{2}.
\end{equation*}%
Therefore, the functional $\Delta _{A}$ is an infimum over $(L,\cdots ,L)\in 
\vec{\ell}.\mathbb{N}^{d}$ as claimed in the lemma.

Now, there is a constant $C<\infty $ such that, for all $L^{\prime }\in 
\mathbb{N}$, \ there is $L\in \mathbb{N}$ such that $|L-L^{\prime }|\leq C$
and $(L,\cdots ,L)\in \vec{\ell}.\mathbb{N}^{d}$. It follows that 
\begin{equation*}
\rho (A_{L^{\prime }}^{\ast }A_{L^{\prime }})=\rho (A_{L}^{\ast }A_{L})+%
\mathcal{O}\left( L^{-1}\right) ,
\end{equation*}%
which implies, for any diverging sequence $\{L_{n}\}_{n=1}^{\infty }$ of
natural numbers, that 
\begin{equation*}
\lim\limits_{n\rightarrow \infty }\rho (A_{L_{n}}^{\ast }A_{L_{n}})=\Vert
P_{\rho }\pi _{\rho }(A_{\vec{\ell}})\Omega _{\rho }\Vert ^{2}\in \left[
|\rho (A_{\vec{\ell}})|^{2},\Vert A\Vert ^{2}\right]
\end{equation*}%
because of (\ref{ergodic eq}). 
\end{proof}%

From\ Lemma\ \ref{Lemma1.vonN copy(2)}\ we deduce now the main properties of
the functional $\Delta _{A}$:

\begin{lemma}[Weak$^{\ast }$--upper semi--continuity, t.i., and affinity of $%
\Delta _{A}$]
\label{Lemma1.vonN copy(3)}\mbox{ }\newline
For any $A\in \mathcal{U}$, the space--averaging functional $\Delta _{A}$ on
the set $E_{\vec{\ell}}$ of $\vec{\ell}$--periodic states is affine, t.i.,
and weak$^{\ast }$--upper semi--continuous.
\end{lemma}

\begin{proof}
Because the map $\rho \mapsto \rho (A)$ is affine, $\Delta _{A}$ is also
affine. Moreover, by using (\ref{definition de A l}), (\ref{asymptotics3}),
and (\ref{ergodic eq}) we obtain, for all $x\in \mathbb{Z}^{d}$, that 
\begin{equation*}
\Delta _{A}(\rho \circ \alpha _{x})=\Delta _{\alpha _{x}\left( A\right)
}(\rho )=\Vert P_{\rho }\pi _{\rho }(\alpha _{x}\left( A_{\vec{\ell}}\right)
)\Omega _{\rho }\Vert ^{2}=\Vert P_{\rho }\pi _{\rho }(A_{\vec{\ell}})\Omega
_{\rho }\Vert ^{2}=\Delta _{A}(\rho )
\end{equation*}%
because $\rho \in E_{\vec{\ell}}$. In other words, the map $\rho \mapsto
\Delta _{A}\left( \rho \right) $ is t.i. on $E_{\vec{\ell}}$. Finally, by
Lemma \ref{Lemma1.vonN copy(2)}, $\Delta _{A}$ is an infimum over weak$%
^{\ast }$--continuous functionals and is therefore weak$^{\ast }$--upper
semi--continuous. The latter is completely standard to verify. Indeed, by
Lemma \ref{Lemma1.vonN copy(2)}, 
\begin{equation*}
M_{r}:=\{\rho \in E_{\vec{\ell}}\;:\;\Delta _{A}\left( \rho \right)
<r\}=\bigcup\limits_{(L,\cdots ,L)\in \vec{\ell}.\mathbb{N}^{d}}\{\rho \in
E_{\vec{\ell}}\;:\;\rho (A_{L}^{\ast }A_{L})<r\}
\end{equation*}%
for any constant $r\in \mathbb{R}_{0}^{+}$. Since, for any $A\in \mathcal{U}$%
, the map $\rho \mapsto \rho (A)$ is weak$^{\ast }$--continuous, $M_{r}$ is
the union of open sets which implies the weak$^{\ast }$--upper
semi--continuity of $\Delta _{A}$. 
\end{proof}%

\begin{lemma}[Locally Lipschitz continuity of the map $A\mapsto \Delta
_{A}\left( \protect\rho \right) $]
\label{Lemma1.vonN copy(5)}\mbox{ }\newline
For all $\rho \in E_{\vec{\ell}}$ and all $A,B\in \mathcal{U}$, 
\begin{equation*}
|\Delta _{A}\left( \rho \right) -\Delta _{B}\left( \rho \right) |\leq (\Vert
A\Vert +\Vert B\Vert )\Vert A-B\Vert .
\end{equation*}
\end{lemma}

\begin{proof}
This proof is straightforward. Indeed, observe that 
\begin{equation*}
\left\Vert \rho \left( A_{L}^{\ast }A_{L}\right) -\rho \left( B_{L}^{\ast
}B_{L}\right) \right\Vert \leq \left\Vert A^{\ast }-B^{\ast }\right\Vert
\left\Vert A\right\Vert +\left\Vert B^{\ast }\right\Vert \left\Vert
A-B\right\Vert
\end{equation*}%
from which we deduce the lemma. 
\end{proof}%

We analyze now the space--averaging functional $\rho \mapsto \Delta
_{A}\left( \rho \right) $ seen as a map from the set $E_{1}$ of t.i. states
to $\mathbb{R}$.

\begin{proposition}[Continuity/Discontinuity of $\Delta _{A}$ on $E_{1}$]
\label{Lemma1.vonN copy(1)}\mbox{ }\newline
\emph{(i)} $\Delta _{A}$ is continuous on $E_{1}$ iff the affine map $\rho
\mapsto |\rho (A)|$ from $E_{1}$ to $\mathbb{C}$ is a constant map.\newline
\emph{(ii)} For all $A\in \mathcal{U}$ such that $\rho \mapsto |\rho (A)|$
is not constant, $\Delta _{A}$ is discontinuous on a weak$^{\ast }$--dense
subset of $E_{1}$.\newline
\emph{(iii)} $\Delta _{A}$ is weak$^{\ast }$--continuous on the $G_{\delta }$
weak$^{\ast }$--dense subset $\mathcal{E}_{1}$ of ergodic states in $E_{1}$.
In particular, the set of all points in $E_{1}$ where this functional is
discontinuous is meager.
\end{proposition}

\begin{proof}
We start by proving the statements (i)--(ii). From Lemmata \ref{lemma
extremal.ergodic}, \ref{Lemma1.vonN copy(3)} and \ref{Corollary 4.1.18.}
combined with Theorem \ref{theorem choquet}, $\Delta _{A}$ can be decomposed
in terms of an integral on the set $\mathcal{E}_{1}$, see Theorem \ref%
{Lemma1.vonN} (iv). As a consequence, if $\rho \mapsto |\rho (A)|$ is a
constant map on $E_{1}$ then the functional $\Delta _{A}$ is clearly
constant on $E_{1}$ and hence continuous. Take now $A\in \mathcal{U}$ such
that the map $\rho \mapsto |\rho (A)|$ is not constant. Then, for any $\rho
\in E_{1}$, there is at least one state $\beth _{\rho }\in E_{1}$ such that $%
|\rho (A)|\neq |\beth _{\rho }(A)|$. For all $\rho \in E_{1}$, we define the
subset $\mathrm{I}(\rho )\subseteq E_{1}$ by%
\begin{equation*}
\mathrm{I}(\rho ):=\{\lambda \rho +(1-\lambda )\beth _{\rho }\ \mathrm{for\
any}\;\lambda \in (0,1)\}.
\end{equation*}%
Finally, let us consider the subset 
\begin{equation*}
\mathrm{D}:=\bigcup\limits_{\rho \in E_{1}}\mathrm{I}(\rho )\subseteq
E_{1}\backslash \mathcal{E}_{1}.
\end{equation*}%
By continuity of the map $\lambda \mapsto \lambda \rho $ for $\lambda \in 
\mathbb{C}$ and $\rho \in E_{1}$, the set $\mathrm{D}$ is dense in $E_{1}$
w.r.t. the weak$^{\ast }$--topology. Moreover, the map $\rho \mapsto \Delta
_{A}(\rho )$ is discontinuous at any $\rho \in \mathrm{D}$. This can be seen
as follows.

Recall that any $\rho \in \mathrm{D}$ is of the form 
\begin{equation*}
\rho =\lambda \rho _{1}+(1-\lambda )\rho _{2}
\end{equation*}%
for some $\lambda \in (0,1)$ and states $\rho _{1},\rho _{2}\in E_{1}$ with $%
|\rho _{1}(A)|\not=|\rho _{2}(A)|$. From Corollary \ref{lemma density of
extremal points}, the set $\mathcal{E}_{1}$ of extreme states is weak$^{\ast
}$--dense in $E_{1}$. So, for any $\rho \in \mathrm{D}$, there is a sequence 
$\{\hat{\rho}_{n}\}_{n=1}^{\infty }\subseteq \mathcal{E}_{1}$ of extreme
states converging w.r.t. the weak$^{\ast }$--topology to $\rho $. Then, by
Lemma \ref{lemma extremal.ergodic}, it follows that%
\begin{eqnarray}
\lim\limits_{n\rightarrow \infty }\Delta _{A}(\hat{\rho}_{n})
&=&\lim\limits_{n\rightarrow \infty }|\hat{\rho}_{n}(A)|^{2}=|\lambda \rho
_{1}(A)+(1-\lambda )\rho _{2}(A)|^{2}  \notag \\
&<&\lambda |\rho _{1}(A)|^{2}+(1-\lambda )|\rho _{2}(A)|^{2}\leq \Delta
_{A}(\rho )  \label{inegality utile for example}
\end{eqnarray}%
because $\rho \mapsto |\rho (A)|^{2}$ is weak$^{\ast }$--continuous, $%
\lambda \in (0,1)$, $\Delta _{A}(\rho )$ is affine, and $\Delta _{A}(\rho
)\geq |\rho (A)|^{2}$ for any $\rho \in E_{1}$.

We conclude this proof by showing that $\Delta _{A}$ is weak$^{\ast }$%
--continuous for any $\hat{\rho}\in \mathcal{E}_{1}$ which yields (iii), by
Corollary \ref{lemma density of extremal points}. Take $\hat{\rho}\in 
\mathcal{E}_{1}$ and consider any sequence $\{\rho _{n}\}_{n=1}^{\infty }$
of states of $E_{1}$ converging w.r.t. the weak$^{\ast }$--topology to $\hat{%
\rho}$. The functional $\Delta _{A}$ is weak$^{\ast }$--upper
semi--continuous, whereas, for all $\rho \in E_{1}$, $\Delta _{A}(\rho )\geq
|\rho (A)|^{2}$ with equality whenever $\rho \in \mathcal{E}_{1}$ (see Lemma %
\ref{lemma extremal.ergodic}). Therefore, 
\begin{equation*}
|\hat{\rho}(A)|^{2}=\Delta _{A}\left( \hat{\rho}\right) \geq
\limsup\limits_{n\rightarrow \infty }\Delta _{A}(\rho _{n})\geq
\liminf_{n\rightarrow \infty }\Delta _{A}(\rho _{n})\geq
\lim\limits_{n\rightarrow \infty }|\rho _{n}(A)|^{2}=|\hat{\rho}(A)|^{2}.
\end{equation*}%
In other words, the functional $\Delta _{A}$ is weak$^{\ast }$--continuous
on $\mathcal{E}_{1}$. 
\end{proof}%

Note that the map $\rho \mapsto |\rho (A)|^{2}$ is a weak$^{\ast }$%
--continuous convex minorant of the space--averaging functional $\Delta _{A}$%
, see Lemma \ref{Lemma1.vonN copy(2)} for $\vec{\ell}=(1,\cdots ,1)$. From
Definitions \ref{definition de deltabis}, \ref{def:egodic}, and Theorem \ref%
{theorem ergodic extremal} (or Lemma \ref{lemma extremal.ergodic}), $\Delta
_{A}\left( \hat{\rho}\right) =|\hat{\rho}(A)|^{2}$ for any extreme state $%
\hat{\rho}\in \mathcal{E}_{1}$. Since, by Corollary \ref{lemma density of
extremal points}, the set $\mathcal{E}_{1}$ of extreme states is weak$^{\ast
}$--dense in $E_{1}$, these last properties suggest that the map $\rho
\mapsto |\rho (A)|^{2}$ is the largest weak$^{\ast }$--lower
semi--continuous convex minorant of $\Delta _{A}$. This is proven in our
last lemma on the functional $\Delta _{A}$.

\begin{lemma}[$\Gamma $--regularization of $\Delta _{A}$]
\label{Lemma convex hull}\mbox{ }\newline
\index{Gamma--regularization!space--averaging functional}The $\Gamma $%
--regularization on $E_{1}$ of the functional $\Delta _{A}$ is the weak$%
^{\ast }$--continuous convex functional $\rho \mapsto |\rho (A)|^{2}$. In
particular, $\rho \mapsto |\rho (A)|^{2}$ is the largest weak$^{\ast }$%
--lower semi--continuous convex minorant of $\Delta _{A}$ on $E_{1}$.
\end{lemma}

\begin{proof}
Recall that the $\Gamma $--regularization of functionals are defined by
Definition \ref{gamm regularisation}. By Lemmata \ref{lemma extremal.ergodic}
and \ref{Lemma1.vonN copy(2)} for $%
\vec{\ell}=(1,\cdots ,1)$, $\Delta _{A}\left( \hat{\rho}\right) =|\hat{\rho}%
(A)|^{2}$ for any $\hat{\rho}\in \mathcal{E}_{1}$, whereas, for all $\rho
\in E_{1}$, $\Delta _{A}\left( \rho \right) \geq |\rho (A)|^{2}$. Since the
map $\rho \mapsto |\rho (A)|^{2}$ from $E_{1}$ to $\mathbb{R}$ is a weak$%
^{\ast }$--continuous convex functional, by Corollary \ref{Biconjugate}, the 
$\Gamma $--regularization $\Gamma _{E_{1}}\left( \Delta _{A}\right) $ of $%
\Delta _{A}$ is bounded from below on $E_{1}$ by the map $\rho \mapsto |\rho
(A)|^{2}$, whereas, for any extreme state $\hat{\rho}\in \mathcal{E}_{1}$, $%
\Gamma _{E_{1}}\left( \Delta _{A}\right) \left( \hat{\rho}\right) =|\hat{\rho%
}(A)|^{2}$. Because of the weak$^{\ast }$--density of $\mathcal{E}_{1}$ in $%
E_{1}$ (Corollary \ref{lemma density of extremal points}), we deduce by
using the weak$^{\ast }$--lower semi--continuity of the functional $\Gamma
_{E_{1}}(\Delta _{A})$ that $\Gamma _{E_{1}}\left( \Delta _{A}\right) \left(
\rho \right) =|\rho (A)|^{2}$ for all $\rho \in E_{1}$. 
\end{proof}%

\section{Von Neumann entropy and entropy density of $\vec{\ell}$--periodic
states\label{section neuman entropy}}

For any local state $\rho _{\Lambda }\in E_{\Lambda }$, there exists a
unique density matrix $\mathrm{d}_{\rho _{\Lambda }}\in \mathcal{U}^{+}\cap 
\mathcal{U}_{\Lambda }$ satisfying $\rho _{\Lambda }(A)=\text{Trace}(\mathrm{%
d}_{\rho _{\Lambda }}A)$ for all $A\in \mathcal{U}_{\Lambda }$. The von
Neumann entropy is then defined, for any local state $\rho _{\Lambda }$ with
density matrix $\mathrm{d}_{\rho _{\Lambda }}$, by%
\index{Entropy density functional!von Neumann|textbf}%
\begin{equation}
S(\rho _{\Lambda }):=\mathrm{Trace}\,\left( \eta (\mathrm{d}_{\rho _{\Lambda
}})\right) \geq 0.  \label{neuman entropy}
\end{equation}%
Here, $\eta (x):=-x\log (x)$. Observe that $\mathcal{U}_{\Lambda }$ is
isomorphic to some (finite dimensional) matrix algebra $B(%
\mathbb{C}
^{N_{\Lambda }})$. The linear functional $\mathrm{Trace}:\mathcal{U}%
_{\Lambda }\rightarrow 
\mathbb{C}
$ is defined by $\mathrm{Trace}:=\mathrm{Tr}\circ \varphi $ with $\varphi $
being an arbitrary $\ast $--isomorphism $\mathcal{U}_{\Lambda }\rightarrow B(%
\mathbb{C}
^{N_{\Lambda }})$ and $\mathrm{Tr}$ being the usual trace for linear
operators on $%
\mathbb{C}
^{N_{\Lambda }}$. Note further that $\mathrm{Trace}$ does not depend on the
choice of the isomorphism $\varphi $. The von Neumann entropy has the
following well--known properties:

\begin{enumerate}
\item[\textbf{S1}] It is $%
\vec{\ell}$--periodic in the sense that, for any $\rho \in E_{\vec{\ell}}$, $%
\Lambda \in \mathcal{P}_{f}(\mathfrak{L})$, and $x\in \mathbb{Z}_{\vec{\ell}%
}^{d}$, 
\begin{equation*}
S(\rho _{\Lambda })=S(\rho _{\Lambda +x})
\end{equation*}%
with the local state $\rho _{\Lambda }$ being the restriction of the $\vec{%
\ell}$--periodic state $\rho $ on the sub--algebra $\mathcal{U}_{\Lambda
}\subseteq \mathcal{U}$ and with $\Lambda +x$ defined by (\ref{definition de
lambda translate}).

\item[\textbf{S2}] It is strongly sub--additive, i.e., for any $\Lambda
_{1},\Lambda _{2}\in \mathcal{P}_{f}(\mathfrak{L})$ and any local state $%
\rho _{\Lambda _{1}\cup \Lambda _{2}}$ on $\mathcal{U}_{\Lambda _{1}\cup
\Lambda _{2}}$, 
\begin{equation*}
S(\rho _{\Lambda _{1}\cup \Lambda _{2}})-S(\rho _{\Lambda _{1}})-S(\rho
_{\Lambda _{2}})+S(\rho _{\Lambda _{1}\cap \Lambda _{2}})\leq 0,
\end{equation*}%
see \cite[Theorems 3.7 and 10.1]{Araki-Moriya}.

\item[\textbf{S3}] It is concave, i.e., for any $\Lambda \in \mathcal{P}_{f}(%
\mathfrak{L})$, any states $\rho _{\Lambda ,1},\rho _{\Lambda ,2}$ on $%
\mathcal{U}_{\Lambda }$, and $\lambda \in \lbrack 0,1]$, 
\begin{equation*}
S(\lambda \rho _{\Lambda ,1}+(1-\lambda )\rho _{\Lambda ,2})\geq \lambda
S(\rho _{\Lambda ,1})+(1-\lambda )S(\rho _{\Lambda ,2}),
\end{equation*}%
see \cite[Proposition 6.2.28]{BrattelliRobinson}.

\item[\textbf{S4}] It is approximately convex, i.e., for any $\Lambda \in 
\mathcal{P}_{f}(\mathfrak{L})$, any states $\rho _{\Lambda ,1},\rho
_{\Lambda ,2}$ on $\mathcal{U}_{\Lambda }$, and $\lambda \in \lbrack 0,1]$, 
\begin{equation*}
S(\lambda \rho _{\Lambda ,1}+(1-\lambda )\rho _{\Lambda ,2})\leq \lambda
S(\rho _{\Lambda ,1})+(1-\lambda )S(\rho _{\Lambda ,2})+\eta (\lambda )+\eta
(1-\lambda ),
\end{equation*}%
see \cite[Proposition 6.2.28]{BrattelliRobinson}.
\end{enumerate}

\noindent \textbf{S1}--\textbf{S4} ensure the existence as well as some
basic properties of the entropy density $s:E_{\vec{\ell}}\rightarrow \mathbb{%
R}_{0}^{+}$ defined in Definition \ref{entropy.density}:

\begin{lemma}[Existence and properties of the entropy density]
\label{lemma property entropybis}\mbox{ }\newline
\index{Entropy density functional!existence}The map $\rho \mapsto s(\rho )$
from $E_{%
\vec{\ell}}$ to $\mathbb{R}$ equals 
\begin{equation*}
s(\rho ):=\lim\limits_{L\rightarrow \infty }\frac{1}{|\Lambda _{L}|}S(\rho
_{\Lambda _{L}})=\inf\limits_{(L,\cdots ,L)\in \vec{\ell}.\mathbb{N}^{d}}%
\frac{1}{|\Lambda _{L}|}S(\rho _{\Lambda _{L}}).
\end{equation*}%
It is an affine, t.i., and weak$^{\ast }$--upper semi--continuous functional.
\end{lemma}

\begin{proof}%
This lemma is standard, see, e.g., \cite[Section 3]{Araki-Moriya}. Indeed,
the existence of the entropy density is a direct consequence of properties{\
S1--S2} because one deduces from {these properties that}%
\begin{equation*}
s(\rho )=\inf\limits_{(L,\cdots ,L)\in \vec{\ell}.\mathbb{N}^{d}}\frac{1}{%
|\Lambda _{L}|}S(\rho _{\Lambda _{L}}).
\end{equation*}%
This equation implies the weak$^{\ast }$--upper semi--continuity of the
entropy density functional $s$ as the map $\rho \mapsto S(\rho _{\Lambda
_{L}})$ is weak$^{\ast }$--continuous for any $L\in \mathbb{N}$, see similar
arguments performed in the proof of Lemma \ref{Lemma1.vonN copy(3)}. By
using the property S3, the functional $s$ is concave, whereas from S4 one
deduces that it is also convex. Therefore, $\rho \mapsto s(\rho )$ defines a
weak$^{\ast }$--upper semi--continuous affine functional on $E_{\vec{\ell}}$%
. The translation invariance of $s$ follows from the strong sub--additivity
S2 together with standard estimates.%
\end{proof}%

Observe that the entropy density functional $s$ is not weak$^{\ast }$%
--continuous but only norm continuous. These properties are well known, see,
e.g., \cite{Fannes,Fannesbis}. Nevertheless, the entropy density functional $%
s$ has still an interesting weak$^{\ast }$--\textquotedblleft
pseudo--continuity\textquotedblright\ property w.r.t. specific sequences of
ergodic states. This property is important in the following and reads as
follows:

\begin{lemma}[Weak$^{\ast }$--pseudo--continuity of the entropy density]
\label{lemma property entropy copy(1)}\mbox{ }\newline
\index{Entropy density functional!pseudo--continuity}For any t.i. state $%
\rho \in E_{1}$, there is a sequence $\{%
\hat{\rho}_{n}\}_{n=1}^{\infty }$ of ergodic states converging in the weak$%
^{\ast }$--topology to $\rho $ and such that 
\begin{equation*}
s(\rho )=\lim\limits_{n\rightarrow \infty }s(\hat{\rho}_{n}).
\end{equation*}
\end{lemma}

\begin{proof}%
The states $\hat{\rho}_{n}$ defined by (\ref{equation toto}) with $\vec{\ell}%
=(1,\cdots ,1)$ for any $\rho \in E_{1}$ and all $n\in \mathbb{N}$ are
ergodic, i.e., $\hat{\rho}_{n}\in \mathcal{E}_{1}$, see (\ref{asymptotics2})
(extended by density of $\mathcal{U}_{0}$ to all $A\in \mathcal{U}$).
Moreover, the sequence $\{\hat{\rho}_{n}\}_{n=1}^{\infty }$ converges in the
weak$^{\ast }$--topology towards $\rho $. On the other hand, by translation
invariance and affinity of the entropy (Lemma \ref{lemma property entropybis}%
), 
\begin{equation*}
s(\hat{\rho}_{n})=s(\tilde{\rho}_{n})=\frac{1}{|\Lambda _{n}|}S(\rho
_{\Lambda _{n}})
\end{equation*}%
with $S$ being the von Neumann entropy (\ref{neuman entropy}) and $\tilde{%
\rho}_{n}$ the $(2n+1)(1,\ldots ,1)$--periodic continuation of the
restriction $\rho _{\Lambda _{n}}\in E_{\Lambda _{n}}$ of the state $\rho
\in E_{1}$ on the box $\Lambda _{n}=\Lambda _{n,(1,\ldots ,1)}$ (defined by (%
\ref{equation toto0})). Therefore the entropy density $s(\hat{\rho}_{n})$
converges to $s(\rho )$ as $n\rightarrow \infty $, see Definition \ref%
{entropy.density}.%
\end{proof}%

\section{The set $E_{1}$ as a subset of the dual space $\mathcal{W}%
_{1}^{\ast }$\label{Section state=functional on W}}

Another important thermodynamic quantity associated with any $\vec{\ell}$%
--periodic state $\rho \in E_{\vec{\ell}}$ on $\mathcal{U}$ is the energy
density $\rho \mapsto e_{\Phi }(\rho )$\ defined for any t.i. interaction $%
\Phi \in \mathcal{W}_{1}$. It is the thermodynamic limit of the internal
energy $\rho (U_{\Lambda }^{\Phi })$ (Definition \ref{definition standard
interaction} (ii)) per unit volume associated with any fixed local
interaction $\Phi $, see Definition \ref{definition energy density}. This
last definition makes sense as soon as $\Phi \in \mathcal{W}_{1}$. Indeed,
this basically follows from Lebesgue's dominated convergence theorem:

\begin{lemma}[Well--definiteness of the energy density]
\label{lemma energy density exists}\mbox{ }\newline
\index{Energy density functional}The energy density $e_{\Phi }(\rho )$ of
any $%
\vec{\ell}$--periodic state $\rho \in E_{\vec{\ell}}$ w.r.t. $\Phi \in 
\mathcal{W}_{1}$ equals $e_{\Phi }(\rho )=\rho (\mathfrak{e}_{\Phi ,\vec{\ell%
}})$ with $\mathfrak{e}_{\Phi ,\vec{\ell}}$ being defined by (\ref%
{eq:enpersite}) for any $\vec{\ell}\in \mathbb{N}^{d}$.
\end{lemma}

\begin{proof}
For any t.i. interaction $\Phi \in \mathcal{W}_{1}$, its internal energy
equals 
\begin{eqnarray}
U_{\Lambda }^{\Phi } &:&=\sum\limits_{\Lambda ^{\prime }\in \mathcal{P}_{f}(%
\mathfrak{L})}\mathbf{1}_{\left\{ \Lambda ^{\prime }\subseteq \Lambda
\right\} }\Phi _{\Lambda ^{\prime }}=\sum\limits_{x=(x_{1},\cdots
,x_{d}),\;x_{i}\in \{0,\cdots ,\ell _{i}-1\}}  \notag \\
&&\sum\limits_{y\in \Lambda \cap \mathbb{Z}_{\vec{\ell}}^{d},x+y\in \Lambda
}\ \sum\limits_{\Lambda ^{\prime }\in \mathcal{P}_{f}(\mathfrak{L}),\Lambda
^{\prime }\ni 0}\mathbf{1}_{\left\{ \Lambda ^{\prime }\subseteq (\Lambda
-x-y)\right\} }\frac{\Phi _{(x+y)+\Lambda ^{\prime }}}{|\Lambda ^{\prime }|}.
\label{eq:enpersite -1}
\end{eqnarray}%
Then, for any $\vec{\ell}$--periodic state $\rho \in E_{\vec{\ell}}$ and any 
$L\in \mathbb{R}$,%
\begin{eqnarray*}
\frac{\rho \left( U_{\Lambda _{L}}^{\Phi }\right) }{|\Lambda _{L}|} &=&\frac{%
|\Lambda _{L}\cap \mathbb{Z}_{\vec{\ell}}^{d}|}{|\Lambda _{L}|}%
\sum\limits_{x=(x_{1},\cdots ,x_{d}),\;x_{i}\in \{0,\cdots ,\ell
_{i}-1\}}\;\sum\limits_{\Lambda ^{\prime }\in \mathcal{P}_{f}(\mathfrak{L}%
),\Lambda ^{\prime }\ni 0}\rho \left( \frac{\Phi _{x+\Lambda ^{\prime }}}{%
|\Lambda ^{\prime }|}\right) \\
&&\times \frac{1}{|\Lambda _{L}\cap \mathbb{Z}_{\vec{\ell}}^{d}|}\left(
\sum\limits_{y\in \Lambda _{L}\cap \mathbb{Z}_{\vec{\ell}}^{d},x+y\in
\Lambda _{L}}\mathbf{1}_{\left\{ \Lambda ^{\prime }\subseteq (\Lambda
_{L}-x-y)\right\} }\right) .
\end{eqnarray*}%
As $\Vert \Phi \Vert _{\mathcal{W}_{1}}<\infty $, we can perform the limit $%
L\rightarrow \infty $ in this last equality by using Lebesgue's dominated
convergence theorem in order to show that 
\begin{equation*}
e_{\Phi }(\rho ):=\underset{L\rightarrow \infty }{\lim }\frac{\rho \left(
U_{\Lambda _{L}}^{\Phi }\right) }{|\Lambda _{L}|}=\rho (\mathfrak{e}_{\Phi ,%
\vec{\ell}}).
\end{equation*}%
\end{proof}%

The functional $e_{\Phi }$ can be seen either as the affine map $\rho
\mapsto e_{\Phi }(\rho )$ at fixed $\Phi \in \mathcal{W}_{1}$ or as the
linear functional $\Phi \mapsto e_{\Phi }(\rho )$ at fixed $\rho \in E_{\vec{%
\ell}}$. In this section we use the second point of view to identify the set 
$E_{1}$ of all t.i. states on $\mathcal{U}$ with a weak$^{\ast }$--compact
set of norm one functionals on the Banach space $\mathcal{W}_{1}$
(Definition \ref{definition banach space interaction}). Indeed, we define
the map $\rho \mapsto \mathbb{T}(\rho )$ from $E_{\vec{\ell}}$ to the dual
space $\mathcal{W}_{1}^{\ast }$ which associates to any $\vec{\ell}$%
--periodic state $\rho \in E_{\vec{\ell}}$ on $\mathcal{U}$ the affine
continuous functional $\mathbb{T}(\rho )\in \mathcal{W}_{1}^{\ast }$ defined
on the Banach space $\mathcal{W}_{1}$ by%
\begin{equation}
\Phi \mapsto \mathbb{T}(\rho )\,(\Phi ):=-e_{\Phi }(\rho ).  \label{eq.func}
\end{equation}%
The functional $\mathbb{T}(\rho )$ is clearly continuous and linear for any $%
\rho \in E_{\vec{\ell}}$ since 
\begin{equation*}
|e_{\Phi }(\rho )|\leq \Vert \Phi \Vert _{\mathcal{W}_{1}}\quad \mathrm{and}%
\quad e_{\left( \lambda _{1}\Phi +\lambda _{2}\Psi \right) }(\rho )=\lambda
_{1}e_{\Phi }(\rho )+\lambda _{2}e_{\Psi }(\rho )
\end{equation*}%
for any $\lambda _{1},\lambda _{2}\in \mathbb{R}$ and any $\Phi ,\Psi \in 
\mathcal{W}_{1}$, see Lemma \ref{lemma energy density exists}. Observe that
the minus sign in the definition (\ref{eq.func}) is arbitrary. It is used
only for convenience when we have to deal with tangent functionals
(Definition \ref{tangent functional}) of the pressure (\ref{map pour def
tangeant}), see Section \ref{Section Gibbs versus gen eq states}. The map $%
\mathbb{T}$ restricted on the set $E_{1}$ has some interesting topological
properties:

\begin{lemma}[Properties of $\mathbb{T}$ on $E_{1}$]
\label{lemma.T}\mbox{ }\newline
The affine map $\mathbb{T}:E_{1}\rightarrow \mathbb{T}\left( E_{1}\right)
\subseteq \mathcal{W}_{1}^{\ast }$ is a homeomorphism in the weak$^{\ast }$%
--topology and an isometry in the norm topology, i.e., $\Vert \mathbb{T(\rho
)-T(\rho }^{\prime }\mathbb{)}\Vert =\Vert \mathbb{\rho -\rho }^{\prime
}\Vert $ for all $\rho ,\rho ^{\prime }\in E_{1}$.
\end{lemma}

\begin{proof}%
The functional $\mathbb{T}$ is weak$^{\ast }$--continuous because the map $%
\rho \rightarrow e_{\Phi }(\rho )$ is weak$^{\ast }$--continuous, by Lemma %
\ref{Th.en.func} (i). As $E_{1}$ is compact w.r.t. the weak$^{\ast }$%
--topology and the dual space $\mathcal{W}_{1}^{\ast }$ is Hausdorff w.r.t.
the weak$^{\ast }$--topology (cf. Corollary \ref{thm locally convex spacebis}%
), it is a homeomorphism from $E_{1}\ $to $\mathbb{T}\left( E_{1}\right) $
if it is an injection from $E_{1}$ to $\mathcal{W}_{1}^{\ast }$.

In fact, for any $\rho ,\rho ^{\prime }\in E_{\vec{\ell}}$, observe that%
\begin{equation}
\Vert \mathbb{T}(\rho )-\mathbb{T}(\rho ^{\prime })\Vert :=\underset{\Phi
\in \mathcal{W}_{1},\;\Vert \Phi \Vert _{\mathcal{W}_{1}}=1}{\sup }|\rho (%
\mathfrak{e}_{\Phi })-\rho ^{\prime }(\mathfrak{e}_{\Phi })|\leq \Vert \rho
-\rho ^{\prime }\Vert .  \label{lemma.Teq0}
\end{equation}%
Therefore, in order to show that the functional $\mathbb{T}$ on $E_{1}$ is
an isometry, which yields its injectivity, it suffices to prove the opposite
inequality.

For any $A=A^{\ast }\in \mathcal{U}_{0}$, there exists a finite range
interaction $\Phi ^{A}\in \mathcal{W}_{1}$ with $\Vert \Phi ^{A}\Vert _{%
\mathcal{W}_{1}}=\Vert A\Vert $ such that, for any $\rho \in E_{1}$, 
\begin{equation*}
e_{\Phi ^{A}}(\rho )=\rho (A).
\end{equation*}%
For $A=A^{\ast }\in \mathcal{U}_{\Lambda }$, choose, for instance, $\Phi
^{A}(\Lambda ^{\prime })=$ $\alpha _{x}(A)$ if $\Lambda ^{\prime }=\Lambda
+x $ and $\Phi ^{A}(\Lambda ^{\prime })=0$ else. It follows that, for any $%
A=A^{\ast }\in \mathcal{U}_{0}$, 
\begin{equation}
|\rho (A)-\rho ^{\prime }(A)|\leq \Vert \mathbb{T}(\rho )-\mathbb{T}(\rho
^{\prime })\Vert \;\Vert A\Vert .  \label{lemma.Teq00}
\end{equation}%
The difference $(\rho -\rho ^{\prime })$ of states $\rho ,\rho ^{\prime }\in
E_{1}$ is a Hermitian functional on a $C^{\ast }$--algebra which implies
that 
\begin{equation*}
\Vert \rho -\rho ^{\prime }\Vert =\sup\limits_{A\in \mathcal{U},\;A=A^{\ast
},\;\Vert A\Vert =1}|\rho (A)-\rho ^{\prime }(A)|.
\end{equation*}%
Since the algebra $\mathcal{U}_{0}$ of local elements is dense in $\mathcal{U%
}$, this last equality together with (\ref{lemma.Teq0}) and (\ref%
{lemma.Teq00}) implies that, for all $\rho ,\rho ^{\prime }\in E_{1}$, 
\begin{equation*}
\Vert \mathbb{T}(\rho )-\mathbb{T}(\rho ^{\prime })\Vert =\Vert \rho -\rho
^{\prime }\Vert .
\end{equation*}%
\end{proof}%

\noindent As a consequence, we can identify any t.i. state $\rho \in E_{1}$
with the continuous linear functional $\mathbb{T}\left( \rho \right) \in 
\mathcal{W}_{1}^{\ast }$.

\section{Well--definiteness of the free--energy densities on $E_{\vec{\ell}}$%
\label{section 6.6}}

Two crucial functionals related to the thermodynamics of long--range models%
\begin{equation*}
\mathfrak{m}:=(\Phi ,\{\Phi _{a}\}_{a\in \mathcal{A}},\{\Phi _{a}^{\prime
}\}_{a\in \mathcal{A}})\in \mathcal{M}_{1}
\end{equation*}%
are the free--energy density functional $f_{\mathfrak{m}}^{\sharp }$ defined
on the set $E_{\vec{\ell}}$ of $\vec{\ell}$--periodic states by%
\begin{equation*}
f_{\mathfrak{m}}^{\sharp }\left( \rho \right) :=\left\Vert \Delta
_{a,+}\left( \rho \right) \right\Vert _{1}-\left\Vert \Delta _{a,-}\left(
\rho \right) \right\Vert _{1}+e_{\Phi }(\rho )-\beta ^{-1}s(\rho )
\end{equation*}%
and the reduced free--energy density functional $g_{\mathfrak{m}}$ defined
on $E_{\vec{\ell}}$ by 
\begin{equation*}
g_{\mathfrak{m}}\left( \rho \right) :=\Vert \gamma _{a,+}\rho \left( 
\mathfrak{e}_{\Phi _{a}}+i\mathfrak{e}_{\Phi _{a}^{\prime }}\right) \Vert
_{2}^{2}-\Vert \gamma _{a,-}\rho \left( \mathfrak{e}_{\Phi _{a}}+i\mathfrak{e%
}_{\Phi _{a}^{\prime }}\right) \Vert _{2}^{2}+e_{\Phi }(\rho )-\beta
^{-1}s(\rho ),
\end{equation*}%
see Definitions \ref{Free-energy density long range} and \ref{Reduced free
energy}. Here, $\Delta _{a,\pm }\left( \rho \right) $ is defined by (\ref%
{definition of delta long range}), that is, 
\begin{equation*}
\Delta _{a,\pm }\left( \rho \right) :=\gamma _{a,\pm }\Delta _{\mathfrak{e}%
_{\Phi _{a}}+i\mathfrak{e}_{\Phi _{a}^{\prime }}}\left( \rho \right) \in
\lbrack 0,\Vert \Phi _{a}\Vert _{\mathcal{W}_{1}}^{2}+\Vert \Phi
_{a}^{\prime }\Vert _{\mathcal{W}_{1}}^{2}]
\end{equation*}%
(cf. (\ref{eq:enpersite bounded}) and Lemma \ref{Lemma1.vonN copy(2)}) with 
\begin{equation*}
\gamma _{a,\pm }:=1/2(|\gamma _{a}|\pm \gamma _{a})\in \{0,1\}
\end{equation*}%
being the negative and positive parts (\ref{remark positive negative part
gamma}) of the fixed measurable function $\gamma _{a}\in \{-1,1\}$.

Both functionals $f_{\mathfrak{m}}^{\sharp }$ and $g_{\mathfrak{m}}$ are
well--defined. Indeed, the entropy density functional $s$ as well as the
energy density functional $e_{\Phi }$ are both well--defined, see Lemmata %
\ref{lemma property entropybis} and \ref{lemma energy density exists}.
Moreover, for any $\rho \in E_{\vec{\ell}}$ and any $\mathfrak{m}\in 
\mathcal{M}_{1}$, the maps $a\mapsto \Delta _{a,\pm }(\rho )$ are measurable
and $\Vert \Delta _{a,\pm }\left( \rho \right) \Vert _{1}<\infty $:

\begin{lemma}[Long--range energy densities for $\mathfrak{m}\in \mathcal{M}%
_{1}$]
\label{delta reprentation integral}\mbox{ }\newline
\index{Space--averaging functional}The maps $\rho \mapsto \Vert \Delta
_{a,\pm }\left( \rho \right) \Vert _{1}$ from $E_{%
\vec{\ell}}$ to $\mathbb{R}_{0}^{+}$ are well--defined affine, t.i., and weak%
$^{\ast }$--upper semi--continuous functionals which equal 
\begin{equation}
\Vert \Delta _{a,\pm }(\rho )\Vert _{1}=\inf\limits_{(L,\cdots ,L)\in \vec{%
\ell}.\mathbb{N}^{d}}\left\{ \int_{\mathcal{A}}\gamma _{a,\pm }\rho (%
\mathfrak{u}_{L,a}^{\ast }\mathfrak{u}_{L,a})\mathrm{d}\mathfrak{a}\left(
a\right) \right\} \leq \left\Vert \Phi _{a}\right\Vert _{2}^{2}+\left\Vert
\Phi _{a}^{\prime }\right\Vert _{2}^{2}  \label{equality delta sympa}
\end{equation}%
for any $\rho \in E_{\vec{\ell}}$, where%
\begin{equation*}
\mathfrak{u}_{L,a}:=\frac{1}{|\Lambda _{L}|}\sum\limits_{x\in \Lambda
_{L}}\alpha _{x}(\mathfrak{e}_{\Phi _{a}}+i\mathfrak{e}_{\Phi _{a}^{\prime
}})\in \mathcal{U}.
\end{equation*}
\end{lemma}

\begin{proof}
The maps $a\mapsto \Delta _{a,\pm }(\rho )$ are measurable and 
\begin{equation*}
\Vert \Delta _{a,\pm }\left( \rho \right) \Vert _{1}\leq \left\Vert \Phi
_{a}\right\Vert _{2}^{2}+\left\Vert \Phi _{a}^{\prime }\right\Vert
_{2}^{2}<\infty
\end{equation*}%
for any $\mathfrak{m}\in \mathcal{M}_{1}$ and $\rho \in E_{\vec{\ell}}$. It
is a consequence of (\ref{eq:enpersite bounded}) and Lemma \ref{Lemma1.vonN
copy(2)} which also implies that 
\begin{equation*}
\Vert \Delta _{a,\pm }(\rho )\Vert _{1}=\int_{\mathcal{A}}\gamma _{a,\pm
}\Delta _{a,\pm }(\rho )\mathrm{d}\mathfrak{a}\left( a\right) =\int \gamma
_{a,\pm }%
\Big\{%
\inf\limits_{(L,\cdots ,L)\in \vec{\ell}.\mathbb{N}^{d}}\rho (\mathfrak{u}%
_{L,a}^{\ast }\mathfrak{u}_{L,a})%
\Big\}%
\mathrm{d}\mathfrak{a}\left( a\right)
\end{equation*}%
for any $\rho \in E_{\vec{\ell}}$. Thus, by the monotonicity of integrals,%
\begin{equation}
\Vert \Delta _{a,\pm }(\rho )\Vert _{1}\leq \inf\limits_{(L,\cdots ,L)\in 
\vec{\ell}.\mathbb{N}^{d}}\left\{ \int_{\mathcal{A}}\gamma _{a,\pm }\rho (%
\mathfrak{u}_{L,a}^{\ast }\mathfrak{u}_{L,a})\mathrm{d}\mathfrak{a}\left(
a\right) \right\} .  \label{upper bound delta L1 norm}
\end{equation}%
By (\ref{eq:enpersite bounded}), note that, for all $L\in \mathbb{N}$, 
\begin{equation*}
\rho (\mathfrak{u}_{L,a}^{\ast }\mathfrak{u}_{L,a})\leq 2\Vert \Phi
_{a}\Vert _{\mathcal{W}_{1}}^{2}+2\Vert \Phi _{a}^{\prime }\Vert _{\mathcal{W%
}_{1}}^{2}.
\end{equation*}%
Therefore, using that 
\begin{equation*}
\Delta _{a,\pm }(\rho )=\inf\limits_{(L,\cdots ,L)\in \vec{\ell}.\mathbb{N}%
^{d}}\rho (\mathfrak{u}_{L,a}^{\ast }\mathfrak{u}_{L,a})=\underset{%
L\rightarrow \infty }{\lim }\rho (\mathfrak{u}_{L,a}^{\ast }\mathfrak{u}%
_{L,a})
\end{equation*}%
and Lebesgue's dominated convergence we obtain that 
\begin{equation*}
\Vert \Delta _{a,\pm }(\rho )\Vert _{1}=\underset{L\rightarrow \infty }{\lim 
}\int_{\mathcal{A}}\gamma _{a,\pm }\rho (\mathfrak{u}_{L,a}^{\ast }\mathfrak{%
u}_{L,a})\mathrm{d}\mathfrak{a}\left( a\right) =\liminf\limits_{L\rightarrow
\infty }\int_{\mathcal{A}}\gamma _{a,\pm }\rho (\mathfrak{u}_{L,a}^{\ast }%
\mathfrak{u}_{L,a})\mathrm{d}\mathfrak{a}\left( a\right) .
\end{equation*}%
In particular, we have that 
\begin{equation*}
\Vert \Delta _{a,\pm }(\rho )\Vert _{1}\geq \inf\limits_{(L,\cdots ,L)\in 
\vec{\ell}.\mathbb{N}^{d}}\left\{ \int_{\mathcal{A}}\gamma _{a,\pm }\rho (%
\mathfrak{u}_{L,a}^{\ast }\mathfrak{u}_{L,a})\mathrm{d}\mathfrak{a}\left(
a\right) \right\}
\end{equation*}%
which combined with (\ref{upper bound delta L1 norm}) implies Equality (\ref%
{equality delta sympa}).

By Lebesgue's dominated convergence theorem, the map 
\begin{equation*}
\rho \mapsto \int_{\mathcal{A}}\gamma _{a,\pm }\rho (\mathfrak{u}%
_{L,a}^{\ast }\mathfrak{u}_{L,a})\mathrm{d}\mathfrak{a}\left( a\right)
\end{equation*}%
is weak$^{\ast }$--continuous for any $L\in \mathbb{N}$. So, the weak$^{\ast
}$--upper semi--continuity of the maps $\rho \mapsto \Vert \Delta _{a,\pm
}\left( \rho \right) \Vert _{1}$ results from (\ref{equality delta sympa}),
see similar arguments in the proof of Lemma \ref{Lemma1.vonN copy(3)}.
Additionally, the maps $\rho \mapsto \Vert \Delta _{a,\pm }\left( \rho
\right) \Vert _{1}$ inherit the t.i. and affinity of the space--averaging
functionals $\Delta _{a,\pm }$, see again Lemma \ref{Lemma1.vonN copy(3)}.
\hfill 
\end{proof}%

Therefore, combining Lemmata \ref{lemma property entropybis} and \ref{lemma
energy density exists} with Lemma \ref{delta reprentation integral}, we
obtain the well--definiteness of the functionals $f_{\mathfrak{m}}^{\sharp }$
and $g_{\mathfrak{m}}$:

\begin{corollary}[Well--definiteness of the functionals $f_{\mathfrak{m}%
}^{\sharp }$ and $g_{\mathfrak{m}}$]
\label{corollary property free--energy density functional}\mbox{ }\newline
\emph{(i)} 
\index{Free--energy density functional!long--range}$\rho \mapsto f_{%
\mathfrak{m}}^{\sharp }\left( \rho \right) $ is a well--defined map from $E_{%
\vec{\ell}}$ to $\mathbb{R}$.\newline
\emph{(ii)} 
\index{Free--energy density functional!reduced}$\rho \mapsto g_{\mathfrak{m}%
}\left( \rho \right) $ is a well--defined map from $E_{%
\vec{\ell}}$ to $\mathbb{R}$.
\end{corollary}

\chapter{Permutation Invariant Fermi Systems\label{Stoermer}}

\setcounter{equation}{0}%
By using the so--called passivity of Gibbs states (Theorem \ref%
{passivity.Gibbs})%
\index{Passivity of Gibbs states} the pressure $p_{l}=p_{l,\mathfrak{m}}$
defined by (\ref{BCS pressure}) for $l\in \mathbb{N}$ and any \emph{discrete}
model 
\begin{equation*}
\mathfrak{m}=\{\Phi \}\cup \{\Phi _{k},\Phi _{k}^{\prime }\}_{k=1}^{N}\in 
\mathcal{M}_{1}^{\mathrm{d}}\subseteq \mathcal{M}_{1}
\end{equation*}%
(see Section \ref{definition models}) can easily be bounded from below, for
all states $\rho \in E$, by%
\begin{eqnarray}
p_{l} &\geq &-\sum\limits_{k=1}^{N}%
\frac{\gamma _{k}}{|\Lambda _{l}|^{2}}\rho \left( (U_{\Lambda _{l}}^{\Phi
_{k}}+iU_{\Lambda _{l}}^{\Phi _{k}^{\prime }})^{\ast }(U_{\Lambda
_{l}}^{\Phi _{k}}+iU_{\Lambda _{l}}^{\Phi _{k}^{\prime }})\right)  \notag \\
&&-\frac{1}{|\Lambda _{l}|}\rho \left( U_{\Lambda _{l}}^{\Phi }\right) +%
\frac{1}{\beta |\Lambda _{l}|}S(\rho _{\Lambda _{l}})  \label{BCS equation 5}
\end{eqnarray}%
with $S$ being the von Neumann entropy defined by (\ref{neuman entropy}).
Furthermore, Theorem \ref{passivity.Gibbs} tells us that the equality in (%
\ref{BCS equation 5}) is only satisfied for the Gibbs equilibrium state $%
\rho _{l}=\rho _{\Lambda _{l},U_{l}}$ (\ref{Gibbs.state}), i.e., 
\begin{eqnarray}
p_{l} &=&-\sum\limits_{k=1}^{N}\frac{\gamma _{k}}{|\Lambda _{l}|^{2}}\rho
_{l}\left( (U_{\Lambda _{l}}^{\Phi _{k}}+iU_{\Lambda _{l}}^{\Phi
_{k}^{\prime }})^{\ast }(U_{\Lambda _{l}}^{\Phi _{k}}+iU_{\Lambda
_{l}}^{\Phi _{k}^{\prime }})\right)  \notag \\
&&-\frac{1}{|\Lambda _{l}|}\rho _{l}\left( U_{\Lambda _{l}}^{\Phi }\right) +%
\frac{1}{\beta |\Lambda _{l}|}S(\rho _{l}).  \label{BCS equation 4}
\end{eqnarray}%
Therefore, in order to prove Theorem \ref{BCS main theorem 1} for any
discrete models, one has to control each term in (\ref{BCS equation 5}) and (%
\ref{BCS equation 4}) as $l\rightarrow \infty $. Unfortunately, it is not
clear how to perform this program directly, even if we concentrate on
discrete long--range models. In fact, as it is originally done in \cite%
{Petz2008} and subsequently in \cite{monsieurremark} for quantum spin
systems (Remark \ref{Quantum spin systems}), we first need to understand 
\emph{permutation invariant} models $\mathfrak{m}\in \mathcal{M}_{1}$ to be
able to prove Theorem \ref{BCS main theorem 1}.

This specific class of models is defined and analyzed in Section \ref%
{Section Permutation invariant Fermi systems}. Indeed, such a study requires
a preliminary analysis, done in Section \ref{set of permutations invariant
states}, of the set $E_{\Pi }\subseteq E_{1}$ of permutation invariant
states. This corresponds to a direct extension of our results \cite%
{BruPedra1} on the strong coupling BCS--Hubbard model to general permutation
invariant systems and is given for completeness as well as a kind of
\textquotedblleft warm up\textquotedblright\ for the non--expert reader.
Among other things, we shortly establish St{\o }rmer theorem, a
non--commutative version of the celebrated de Finetti theorem for
permutation invariant states on the fermion algebra $\mathcal{U}$ as it is
proven in \cite{BruPedra1}.

\begin{remark}[Energy--entropy balance conditions]
\mbox{ }\newline
\index{Energy--entropy balance conditions}Our study of equilibrium states is
reminiscent of the work of Fannes, Spohn, and Verbeure \cite{S}, performed,
however, within a different framework. For instance, equilibrium states are
defined in \cite{S} via the energy--entropy balance conditions, also called
the correlation inequalities for quantum states (see, e.g., \cite[Appendix E]%
{BruZagrebnov8}).
\end{remark}

\section{The set $E_{\Pi }$ of permutation invariant states\label{set of
permutations invariant states}%
\index{States!permutation invariant|textbf}}

Let $\Pi $ be the set of all bijective maps from $\mathfrak{L}$ to $%
\mathfrak{L}$ which leaves all but finitely many elements invariant. It is a
group w.r.t. the composition of maps. The condition 
\begin{equation}
\alpha _{\pi }:a_{x,\mathrm{s}}\mapsto a_{\pi (x),\mathrm{s}},\quad \mathrm{s%
}\in \mathrm{S},\;x\in \mathfrak{L},  \label{definition perm automorphism}
\end{equation}%
defines a group homomorphism $\pi \mapsto \alpha _{\pi }$ from $\Pi $ to the
group of $\ast $--automorphisms of $\mathcal{U}$. The set of all permutation
invariant states is then defined by%
\begin{equation}
E_{\Pi }:=\bigcap\limits_{\pi \in \Pi ,%
\text{ }A\in \mathcal{U}}\{\rho \in \mathcal{U}^{\ast }\;:\;\rho (\mathbf{1}%
)=1,\;\rho (A^{\ast }A)\geq 0\text{\quad }\mathrm{with\ }\rho =\rho \circ
\alpha _{\pi }\}.  \label{permutation inv states}
\end{equation}%
Since obviously 
\begin{equation*}
E_{\Pi }\subseteq E_{1}\subseteq \bigcap\limits_{\vec{\ell}\in \mathbb{N}%
^{d}}E_{\vec{\ell}}\ ,
\end{equation*}%
every permutation invariant state $\rho \in E_{\Pi }$ is even, by Lemma \ref%
{coro.even}. Furthermore, $E_{\Pi }$ is clearly convex and weak$^{\ast }$%
--compact and, by the Krein--Milman theorem%
\index{Krein--Milman theorem} (Theorem \ref{theorem Krein--Millman}), it is
the weak$^{\ast }$--closure of the convex hull of the (non--empty) set $%
\mathcal{E}_{\Pi }$ of its extreme points.

The set $\mathcal{E}_{%
\vec{\ell}}$ of extreme states of $E_{\vec{\ell}}$ is characterized by
Theorem \ref{theorem ergodic extremal} and $\mathcal{E}_{\Pi }$ can likewise
be precisely characterized by St{\o }rmer theorem for permutation invariant
states on the fermion algebra $\mathcal{U}$. This theorem is a
non--commutative version of the celebrated de Finetti theorem from
(classical) probability theory and it is proven in the case of even states
on the fermion algebra $\mathcal{U}$ in \cite{BruPedra1}. Indeed, extreme
permutation invariant states $\rho \in \mathcal{E}_{\Pi }$ are product
states defined as follows.

Let $\rho _{\{0\}}\in E_{\mathcal{U}_{\{0\}}}$ be any \emph{even} state on
the one--site $C^{\ast }$--algebra $\mathcal{U}_{\{0\}}$, i.e., $\rho
_{\{0\}}=\rho _{\{0\}}\circ \sigma _{\pi }$ with $\sigma _{\pi }$ defined by
(\ref{definition of gauge}) for $\theta =\pi $. Then, from \cite[Theorem
11.2.]{Araki-Moriya}, there is a unique even state $\hat{\rho}\in E_{\Pi }$
satisfying%
\begin{equation*}
\hat{\rho}(\alpha _{x_{1}}(A_{1})\cdots \alpha _{x_{n}}(A_{n}))=\rho
_{\{0\}}(A_{1})\cdots \rho _{\{0\}}(A_{n})
\end{equation*}%
for all $A_{1}\ldots A_{n}\in \mathcal{U}_{\{0\}}$ and all $x_{1},\ldots
x_{n}\in \mathbb{Z}^{d}$ such that $x_{i}\not=x_{j}$ for $i\not=j$. The set
of all states $\hat{\rho}$ of this form, called product states, is denoted
by $E_{\otimes }$ which is nothing else but the set $\mathcal{E}_{\Pi }$ of
extreme points of $E_{\Pi }$:

\begin{theorem}[St{\o }rmer theorem, lattice CAR--algebra version]
\label{Stoermer CAR}\mbox{ }\newline
\index{St\o rmer theorem! lattice CAR--algebra version}Extreme permutation
invariant states $%
\hat{\rho}\in \mathcal{E}_{\Pi }$ are product states and conversely, i.e., $%
\mathcal{E}_{\Pi }=E_{\otimes }$.
\end{theorem}

\noindent This theorem was proven by St{\o }rmer \cite{Stormer} for the case
of lattice quantum spin systems (cf. Remark \ref{Quantum spin systems}). Its
corresponding version for permutation invariant states on the fermion
algebra $\mathcal{U}$ follows from \cite[Lemmata 6.6--6.8]{BruPedra1}.
Observe that the proof of Theorem \ref{Stoermer CAR} is performed in \cite%
{BruPedra1} for a spin set $\mathrm{S}=\{\uparrow ,\downarrow \}$. It can
easily be extended to the general case of Theorem \ref{Stoermer CAR}.

It follows from Theorem \ref{Stoermer CAR} that all permutation invariant
states $\hat{\rho}\in \mathcal{E}_{\Pi }$ are strongly mixing which means (%
\ref{mixing}). They are, in particular, strongly clustering and thus ergodic
w.r.t. any sub--group $\mathbb{Z}_{\vec{\ell}}^{d}$ of $\mathbb{Z}^{d}$,
where $\vec{\ell}\in \mathbb{N}^{d}$. In other words, for all $\vec{\ell}\in 
\mathbb{N}^{d}$, $\mathcal{E}_{\Pi }=E_{\otimes }\subseteq \mathcal{E}_{\vec{%
\ell}}$ and the set $E_{\Pi }\subseteq E_{\vec{\ell}}$ is hence a closed
metrizable face of $E_{\vec{\ell}}$. Therefore, by using Theorem \ref%
{theorem choquet} and Theorem \ref{Stoermer CAR}, we obtain the existence of
a unique decomposition of states $\rho \in E_{\Pi }$ in terms of product
states:

\begin{theorem}[Unique decomposition of permutation invariant states]
\label{theorem choquet copy(1)}%
\index{States!permutation invariant!extremal decomposition}For any $\rho \in
E_{\Pi }$, there is a unique probability measure $\mu _{\rho }$ on $E_{\Pi }$
such that%
\begin{equation*}
\mu _{\rho }(E_{\otimes })=1%
\text{\quad and\quad }\rho =\int_{E_{\Pi }}\mathrm{d}\mu _{\rho }(\hat{\rho}%
)\;\hat{\rho}.
\end{equation*}%
Furthermore, the map $\rho \mapsto \mu _{\rho }$ is an isometry in the norm
of linear functionals, i.e., $\Vert \rho -\rho ^{\prime }\Vert =\Vert \mu
_{\rho }-\mu _{\rho }^{\prime }\Vert $ for any $\rho ,\rho ^{\prime }\in
E_{\Pi }$.
\end{theorem}

From Theorem \ref{Thm Poulsen simplex}, for all $\vec{\ell}\in \mathbb{N}%
^{d} $, the sets $E_{\vec{\ell}}$ are affinely homeomorphic to the Poulsen
simplex, but the set $E_{\Pi }$ of all permutation invariant states do not
share this property. Indeed, $E_{\Pi }$ is a Bauer simplex (Definition \ref%
{gamm regularisation copy(2)}), i.e., a simplex whose set of extreme points
is closed:

\begin{theorem}[$E_{\Pi }$ is a Bauer simplex]
\label{Thm Poulsen simplex copy(1)}\mbox{ }\newline
\index{States!permutation invariant!Bauer simplex}The set $E_{\Pi }$ is a
Bauer simplex%
\index{Simplex!Bauer}. In particular, the map $\rho \mapsto \mu _{\rho }$ of
Theorem \ref{theorem choquet copy(1)} from $E_{\Pi }$ to the set $M_{1}^{+}(%
\mathcal{E}_{\Pi })=M_{1}^{+}(E_{\otimes })$ of probability measures on $%
\mathcal{E}_{\Pi }=E_{\otimes }$ is an affine homeomorphism w.r.t. the weak$%
^{\ast }$--topologies on $E_{\Pi }$ and $M_{1}^{+}(\mathcal{E}_{\Pi })$.
\end{theorem}

\begin{proof}
As explained above, for all $%
\vec{\ell}\in \mathbb{N}^{d}$, $E_{\Pi }$ is a closed face of $E_{\vec{\ell}%
} $ (and thus a closed simplex) with set $\mathcal{E}_{\Pi }$ of extreme
points being the set $E_{\otimes }$ of product states, i.e., $\mathcal{E}%
_{\Pi }=E_{\otimes }\subseteq \mathcal{E}_{\vec{\ell}}$, see Theorem \ref%
{Stoermer CAR}. Since the set $E_{\otimes }$ is obviously closed in the weak$%
^{\ast }$--topology, it is a Bauer simplex which, combined with Theorem \ref%
{theorem Bauer}, implies the statement. 
\end{proof}%

Therefore, the simplex $E_{\Pi }$ has a much simpler geometrical structure
than all simplices $\{E_{\vec{\ell}}\}_{\vec{\ell}\in \mathbb{N}^{d}}$ and
it is easier to use in practice, see, e.g., \cite{BruPedra1}. For instance,
for any fixed element $A$ of the one--site $C^{\ast }$--algebra $\mathcal{U}%
_{\{0\}}$, the space--averaging functional $\Delta _{A}$ described in
Sections \ref{Section space averaging} and \ref{Section properties of delta}
has a very explicit representation on the Bauer simplex $E_{\Pi }$:

\begin{lemma}[The space--averaging functional $\Delta _{A}$ on $E_{\Pi }$]
\label{space--averaging functional perm inv}\mbox{ }\newline
\index{Space--averaging functional!permutation invariant}At fixed $A\in 
\mathcal{U}_{\{0\}}$, the restriction on $E_{\Pi }$ of the functional $%
\Delta _{A}$ equals, for any $x\in \mathbb{Z}^{d}\backslash \{0\}$, the weak$%
^{\ast }$--continuous affine map $\rho \mapsto \rho (A^{\ast }\alpha
_{x}(A)) $ from $E_{\Pi }$ to $\mathbb{R}_{0}^{+}$.
\end{lemma}

\begin{proof}
This lemma follows from elementary combinatorics, see, e.g., \cite[Lemma 6.2]%
{BruPedra1}.%
\end{proof}%

Permutation invariance is, however,\ a too restrictive condition in general.
Indeed, most of models coming from Physics are only translation invariant.
In particular, the general set of states to be considered in these cases is
the Poulsen simplex (up to an affine homeomorphism), which is in a sense 
\emph{complementary} to the Bauer simplices%
\index{Simplex!Bauer}%
\index{Simplex!Poulsen}, see \cite[p. 164]{Alfsen} or \cite[Section 5]%
{Lindenstrauss-etal}.

\section{Thermodynamics of permutation invariant Fermi systems\label{Section
Permutation invariant Fermi systems}}

\emph{Permutation invariant} interactions form a subset of the real Banach
space $\mathcal{W}_{1}$ of all t.i. interactions $\Phi $, see Definition \ref%
{definition banach space interaction}. They are naturally defined as follows:

\begin{definition}[Permutation invariant interactions]
\mbox{ }\newline
\index{Interaction!permutation invariant}A t.i. interaction $\Phi \in 
\mathcal{W}_{1}$ is permutation invariant if $\Phi _{\Lambda }=0$ whenever $%
|\Lambda |\not=1$.
\end{definition}

\noindent \emph{Permutation invariant} Fermi systems $\mathfrak{m}\in 
\mathcal{M}_{1}$ with long--range interactions (see Definition \ref%
{definition M1bis}) are then defined from permutation invariant interactions
as follows:

\begin{definition}[Permutation invariant models]
\label{Definition permutation inv models}\mbox{ }\newline
\index{Long--range models!permutation invariant}A long--range model $%
\mathfrak{m}:=(\Phi ,\{\Phi _{a}\}_{a\in \mathcal{A}},\{\Phi _{a}^{\prime
}\}_{a\in \mathcal{A}})\in \mathcal{M}_{1}$ is permutation invariant
whenever the interactions $\Phi $, $\Phi _{a}$ and $\Phi _{a}^{\prime }$ are
permutation invariant for all (a.e.) $a\in \mathcal{A}$.
\end{definition}

If the model $\mathfrak{m}\in \mathcal{M}_{1}$ is permutation invariant then
the corresponding internal energies $U_{l}$ defined for $l\in \mathbb{N}$ in
Definition \ref{definition BCS-type model} are invariant w.r.t. permutations
of lattice sites inside the boxes $\Lambda _{l}$. More precisely: For all $%
l\in \mathbb{N}$ and all $\pi \in \Pi $ such that $\pi |_{\mathfrak{L}%
\backslash \Lambda _{l}}=\mathrm{id}|_{\mathfrak{L}\backslash \Lambda _{l}}$%
, $\alpha _{\pi }(U_{l})=U_{l}$. Here, $\mathrm{id}\in \Pi $ is the neutral
element of the group $\Pi $, i.e., the identity map $\mathfrak{L}\rightarrow 
\mathfrak{L}$. As a consequence, for any permutation invariant $\mathfrak{m}%
\in \mathcal{M}_{1}$, the thermodynamic limit 
\begin{equation*}
\mathrm{P}_{\mathfrak{m}}^{\sharp }:=\underset{l\rightarrow \infty }{\lim }%
\left\{ p_{l}\right\}
\end{equation*}%
of the pressure $p_{l}=p_{l,\mathfrak{m}}$ (\ref{BCS pressure}) associated
with the internal energy $U_{l}$ can be computed via the minimization of the
affine free--energy functional $f_{\mathfrak{m}}^{\sharp }$ on the subset $%
E_{\Pi }\subseteq E_{1}$ of permutation invariant states, see Definitions %
\ref{Free-energy density long range}, \ref{Pressure} and Lemma \ref{lemma
property free--energy density functional} (i).

\begin{theorem}[Thermodynamics as a variational problem on $E_{\Pi }$]
\label{theorem pressure perm inv}\mbox{ }\newline
For any permutation invariant $\mathfrak{m}\in \mathcal{M}_{1}$,%
\index{Pressure!variational problems!permutation invariant}%
\begin{equation*}
\mathrm{P}_{\mathfrak{m}}^{\sharp }=-\inf\limits_{\rho \in E_{\Pi }}\,f_{%
\mathfrak{m}}^{\sharp }(\rho )=-\inf\limits_{\rho \in E_{\otimes }}\,f_{%
\mathfrak{m}}^{\sharp }\left( \rho \right) .
\end{equation*}%
Here, the restriction of $f_{\mathfrak{m}}^{\sharp }$ on the weak$^{\ast }$%
--compact convex set $E_{\Pi }$ equals, for any $x\in \mathbb{Z}%
^{d}\backslash \{0\}$, the weak$^{\ast }$--lower semi--continuous affine map%
\index{Free--energy density functional!long--range!permutation invariant}%
\begin{equation}
\rho \mapsto \int_{\mathcal{A}}\gamma _{a}\rho \left( (\mathfrak{e}_{\Phi
_{a}}-i\mathfrak{e}_{\Phi _{a}^{\prime }})\alpha _{x}(\mathfrak{e}_{\Phi
_{a}}+i\mathfrak{e}_{\Phi _{a}^{\prime }})\right) \mathrm{d}\mathfrak{a}%
\left( a\right) +e_{\Phi }(\rho )-\beta ^{-1}s(\rho )
\label{free-energy functional perm inv}
\end{equation}%
from $E_{\Pi }$ to $\mathbb{R}$, see (\ref{eq:enpersite}) for the definition
of $\mathfrak{e}_{\Phi }$.
\end{theorem}

\begin{proof}
Observe first that the equality between $f_{\mathfrak{m}}^{\sharp }$ and the
weak$^{\ast }$--lower semi--continuous affine map (\ref{free-energy
functional perm inv}) (cf. Lemmata \ref{lemma property entropy} (i) and \ref%
{Th.en.func} (i)) is a direct consequence of Lemma \ref{space--averaging
functional perm inv} because $\mathfrak{m}$ is permutation invariant. By the
Bauer maximum principle (Lemma \ref{Bauer maximum principle}), it follows
that the minimization of $f_{\mathfrak{m}}^{\sharp }$ on the weak$^{\ast }$%
--compact convex set $E_{\Pi }$ can be restricted to the subset $\mathcal{E}%
_{\Pi }$ of extreme points which by Theorem \ref{Stoermer CAR} equals the
set $E_{\otimes }$ of product states.

We analyze now the thermodynamic limit $l\rightarrow \infty $ of the
pressure $p_{l}=p_{l,\mathfrak{m}}$. We concentrate our study on discrete
and finite range permutation invariant models 
\begin{equation*}
\mathfrak{m}=\{\Phi \}\cup \{\Phi _{k},\Phi _{k}^{\prime }\}_{k=1}^{N}\in 
\mathcal{M}_{1}^{\mathrm{df}}\subseteq \mathcal{M}_{1}^{\mathrm{d}}\subseteq 
\mathcal{M}_{1}
\end{equation*}%
only. The extension of this proof to any permutation invariant models $%
\mathfrak{m}\in \mathcal{M}_{1}$ is performed by using the density of the
set of discrete permutation invariant models in the set of permutation
invariant models, see similar arguments performed in Section \ref{definition
models} as well as in Section \ref{section reduction of the problem}.

The lower bound on the pressure $p_{l}=p_{l,\mathfrak{m}}$ for discrete
models $\mathfrak{m}\in \mathcal{M}_{1}^{\mathrm{df}}$ follows from the
passivity of Gibbs states (Theorem \ref{passivity.Gibbs}). Indeed, note that 
$\mathfrak{e}_{\Phi }\in \mathcal{U}_{\{0\}}$ for any permutation invariant
interaction $\Phi \in \mathcal{W}_{1}$. Therefore, as $\mathfrak{m}$ is
permutation invariant, straightforward estimates show, for all $\rho \in
E_{\Pi }$ and any $x\in \mathbb{Z}^{d}\backslash \{0\}$, that 
\begin{equation}
\underset{l\rightarrow \infty }{\lim }\left\{ 
\frac{1}{|\Lambda _{l}|^{2}}\rho \left( (U_{\Lambda _{l}}^{\Phi
_{k}}+iU_{\Lambda _{l}}^{\Phi _{k}^{\prime }})^{\ast }(U_{\Lambda
_{l}}^{\Phi _{k}}+iU_{\Lambda _{l}}^{\Phi _{k}^{\prime }})\right) \right\}
=\rho (\mathfrak{e}_{\Phi _{k}}^{\ast }\alpha _{x}(\mathfrak{e}_{\Phi
_{k}})).  \label{eq sup mean enerfybisbis}
\end{equation}%
Therefore, from (\ref{BCS equation 5}) and (\ref{eq sup mean enerfybisbis})
combined with Definitions \ref{entropy.density} and \ref{definition energy
density}, we deduce that 
\begin{equation}
\liminf\limits_{l\rightarrow \infty }p_{l}\geq -\inf\limits_{\rho \in E_{\Pi
}}\,f_{\mathfrak{m}}^{\sharp }(\rho ).  \label{perm inv lower bound}
\end{equation}%
So, we concentrate now our analysis on the upper bound.

Let $\rho _{l}\in E_{\Lambda _{l}}$ be the Gibbs equilibrium state (\ref%
{Gibbs.state}) w.r.t. the internal energy $U_{l}\in \mathcal{U}_{\Lambda
_{l}}$. We define as usual a space--averaged t.i. Gibbs state $\hat{\rho}%
_{l}\in E_{1}$ by using (\ref{t.i. state rho l}) with the even state $\rho
_{l}$ seen as a periodic state on the whole $C^{\ast }$--algebra $\mathcal{U}
$. Observe that the sequences $\{\rho _{l}\}_{l\in \mathbb{N}}$ and $\{\hat{%
\rho}_{l}\}_{l\in \mathbb{N}}$ have the same weak$^{\ast }$--accumulation
points. Since $\mathfrak{m}$ is permutation invariant, the internal energy $%
U_{l}$ is invariant w.r.t. permutations of lattice sites inside the boxes $%
\Lambda _{l}$ which in turn implies the invariance of the state $\rho
_{l}\in E$ under permutations $\pi \in \Pi $ such that $\pi |_{\mathfrak{L}%
\backslash \Lambda _{l}}=\mathrm{id}|_{\mathfrak{L}\backslash \Lambda _{l}}$%
. This invariance property of $\rho _{l}$ yields that the weak$^{\ast }$%
--accumulation points of sequences $\{\rho _{l}\}_{l\in \mathbb{N}}$ and $\{%
\hat{\rho}_{l}\}_{l\in \mathbb{N}}$ belong to $E_{\Pi }$. As a consequence,
there is $\rho _{\infty }\in E_{\Pi }$ and a diverging subsequence $%
\{l_{n}\}_{n\in \mathbb{N}}$ such that both $\rho _{l_{n}}$ and $\hat{\rho}%
_{l_{n}}$ converge in the weak$^{\ast }$--topology to the permutation
invariant state $\rho _{\infty }$.

As $\mathfrak{e}_{\Phi }\in \mathcal{U}_{\{0\}}$ for any permutation
invariant model $\mathfrak{m}\in \mathcal{M}_{1}$, observe, by Lemma \ref%
{Th.en.func} (i), that 
\begin{equation}
\lim\limits_{n\rightarrow \infty }\frac{1}{|\Lambda _{l_{n}}|}\rho
_{l_{n}}(U_{l_{n}}^{\Phi })=\lim\limits_{n\rightarrow \infty }\rho _{l_{n}}(%
\widehat{\mathfrak{e}}_{\Phi ,l_{n}})=\lim\limits_{n\rightarrow \infty
}e_{\Phi }(\rho _{l_{n}})=e_{\Phi }(\rho _{\infty }),  \label{perm inv 1}
\end{equation}%
where 
\begin{equation}
\widehat{\mathfrak{e}}_{\Phi ,L}:=\frac{1}{|\Lambda _{L}|}\sum\limits_{x\in
\Lambda _{L}}\alpha _{x}\left( \mathfrak{e}_{\Phi }\right) =\widehat{%
\mathfrak{e}}_{\Phi ,L}^{\ast }.  \label{perm inv 1bis}
\end{equation}%
By combining the symmetry of the state $\rho _{l}\in E$ under permutations
of lattice sites inside the boxes $\Lambda _{l}$ with elementary
combinatorics, 
\begin{align}
& \lim\limits_{n\rightarrow \infty }\left\{ \frac{1}{|\Lambda _{l_{n}}|^{2}}%
\rho _{l_{n}}\left( (U_{\Lambda _{l_{n}}}^{\Phi _{k}}+iU_{\Lambda
_{l_{n}}}^{\Phi _{k}^{\prime }})^{\ast }(U_{\Lambda _{l_{n}}}^{\Phi
_{k}}+iU_{\Lambda _{l_{n}}}^{\Phi _{k}^{\prime }})\right) \right\}
\label{perm inv 2} \\
& =\lim\limits_{n\rightarrow \infty }\rho _{l_{n}}((\mathfrak{e}_{\Phi
_{k}}-i\mathfrak{e}_{\Phi _{k}^{\prime }})\alpha _{x}(\mathfrak{e}_{\Phi
_{k}}+i\mathfrak{e}_{\Phi _{k}^{\prime }}))=\rho _{\infty }((\mathfrak{e}%
_{\Phi _{k}}-i\mathfrak{e}_{\Phi _{k}^{\prime }})\alpha _{x}(\mathfrak{e}%
_{\Phi _{k}}+i\mathfrak{e}_{\Phi _{k}^{\prime }}))  \notag
\end{align}%
for any $x\in \mathbb{Z}^{d}\backslash \{0\}$. Furthermore, by using Lemma %
\ref{lemma property entropy} (i), the periodicity of $\rho _{l}$ and the
additivity of the von Neumann entropy for product states, 
\begin{equation}
s(\hat{\rho}_{l})=\frac{1}{|\Lambda _{l}|}\sum\limits_{x\in \Lambda
_{l}}s(\rho _{l}\circ \alpha _{x})=s(\rho _{l})=\lim\limits_{n\rightarrow
\infty }\frac{1}{|\Lambda _{l}^{(n)}|}S(\rho _{l}|_{\mathcal{U}_{\Lambda
_{l}^{(n)}}})=\frac{1}{|\Lambda _{l}|}S(\rho _{l})
\label{mean entropy per volume}
\end{equation}%
with the definition 
\begin{equation}
\Lambda _{l}^{(n)}:=\underset{x\in \Lambda _{n}}{\cup }\{\Lambda
_{l}+(2l+1)x\}.  \label{lambda_n_l}
\end{equation}%
Therefore, by using (\ref{BCS equation 4}) combined with (\ref{perm inv 1}),
(\ref{perm inv 2}), (\ref{mean entropy per volume}), and Lemma \ref%
{space--averaging functional perm inv}, 
\begin{equation}
\underset{l\rightarrow \infty }{\lim \sup }\ p_{l}\leq
-\lim\limits_{n\rightarrow \infty }f_{\mathfrak{m}}^{\sharp }(\rho
_{l_{n}})\leq -f_{\mathfrak{m}}^{\sharp }(\rho _{\infty })
\label{perm inv upper bound}
\end{equation}%
because the entropy density functional $s$ is a weak$^{\ast }$--upper
semi--continuous functional on $E_{1}$ (Lemma \ref{lemma property entropy}
(i)).

Since $\rho _{\infty }\in E_{\Pi }$, the theorem follows from (\ref{perm inv
lower bound}) and (\ref{perm inv upper bound}) combined with the density of
the set of discrete permutation invariant models in the set of permutation
invariant models, see, e.g., Corollary \ref{lemma reduction dfbisbis}. 
\end{proof}%

As a consequence, the thermodynamics of any permutation invariant model $%
\mathfrak{m}\in \mathcal{M}_{1}$ can be related to a weak$^{\ast }$%
--continuous free--energy density functional over one--site states:

\begin{corollary}[Variational problem on one--site states]
\label{pression_inv_perm}\mbox{ }\newline
For any permutation invariant $\mathfrak{m}\in \mathcal{M}_{1}$, the
(infinite--volume) pressure equals%
\begin{equation*}
\mathrm{P}_{\mathfrak{m}}^{\sharp }=-\inf_{\rho _{\left\{ 0\right\} }\in
E_{\left\{ 0\right\} }}\left\{ \int_{\mathcal{A}}\gamma _{a}|\rho _{\left\{
0\right\} }(\mathfrak{e}_{\Phi _{a}}+i\mathfrak{e}_{\Phi _{a}^{\prime
}})|^{2}\mathrm{d}\mathfrak{a}(a)+\rho _{\left\{ 0\right\} }(\mathfrak{e}%
_{\Phi })-\beta ^{-1}S(\rho _{\left\{ 0\right\} })\right\}
\end{equation*}%
with the weak$^{\ast }$--continuous functional $S$ being the von Neumann
entropy defined by (\ref{neuman entropy}).%
\index{Entropy density functional!von Neumann}
\end{corollary}

\begin{proof}
By Lemma \ref{space--averaging functional perm inv}, for any permutation
invariant model $\mathfrak{m}\in \mathcal{M}_{1}$, $x\in \mathbb{Z}%
^{d}\backslash \{0\}$ and all product states $\rho \in E_{\otimes }$, 
\begin{equation*}
\Delta _{\mathfrak{e}_{\Phi _{a}}+i\mathfrak{e}_{\Phi _{a}^{\prime }}}\left(
\rho \right) =\rho \left( (\mathfrak{e}_{\Phi _{a}}-i\mathfrak{e}_{\Phi
_{a}^{\prime }})\alpha _{x}(\mathfrak{e}_{\Phi _{a}}+i\mathfrak{e}_{\Phi
_{a}^{\prime }})\right) =|\rho _{\left\{ 0\right\} }(\mathfrak{e}_{\Phi
_{a}}+i\mathfrak{e}_{\Phi _{a}^{\prime }})|^{2}
\end{equation*}%
with the state $\rho _{\left\{ 0\right\} }\in E_{\left\{ 0\right\} }$ being
the restriction of $\rho \in E_{\otimes }$ on the local sub--algebra $%
\mathcal{U}_{\left\{ 0\right\} }$. Furthermore, observe that, for any
product state $\rho \in E_{\otimes }$, $s(\rho )=S(\rho _{\left\{ 0\right\}
})$. Therefore, Corollary \ref{pression_inv_perm} is a direct consequence of
Theorem \ref{theorem pressure perm inv}.%
\end{proof}%

The map (\ref{free-energy functional perm inv}) is a weak$^{\ast }$--lower
semi--continuous affine map from $E_{\Pi }$ to $\mathbb{R}$. So, from
Theorem \ref{theorem pressure perm inv}, all generalized permutation
invariant equilibrium states are (usual) equilibrium states as 
\begin{equation*}
\mathit{\Omega }_{\mathfrak{m}}^{\sharp }\cap E_{\Pi }=\mathit{M}_{\mathfrak{%
m}}^{\sharp }\cap E_{\Pi }\neq \emptyset .
\end{equation*}%
Moreover, $\mathit{\Omega }_{\mathfrak{m}}^{\sharp }\cap E_{\Pi }$\ is a 
\emph{face} of $E_{\Pi }$ (cf. Definitions \ref{definition equilibirum state
copy(1)} and \ref{definition equilibirum state}). Since $E_{\Pi }$ is a
Bauer simplex (Theorem \ref{Thm Poulsen simplex copy(1)}) with its set $%
\mathcal{E}_{\Pi }$ of extreme points being the set $E_{\otimes }$ of
product states (Theorem \ref{Stoermer CAR}), $\mathit{M}_{\mathfrak{m}%
}^{\sharp }\cap E_{\Pi }$ is also a simplex%
\index{Simplex} and, by using the Choquet theorem%
\index{Choquet theorem} (cf. Theorems \ref{theorem choquet bis} and \ref%
{theorem choquet bis copy(1)}), each permutation invariant equilibrium state 
$\omega \in \mathit{M}_{\mathfrak{m}}^{\sharp }\cap E_{\Pi }$ has a \emph{%
unique} decomposition in terms of states of the set 
\begin{equation*}
\mathcal{E}(\mathit{M}_{\mathfrak{m}}^{\sharp }\cap E_{\Pi })=\mathcal{E}(%
\mathit{M}_{\mathfrak{m}}^{\sharp }\cap E_{\Pi })\cap E_{\otimes }
\end{equation*}%
of extreme states of $\mathit{M}_{\mathfrak{m}}^{\sharp }\cap E_{\Pi }$. In
fact, Theorem \ref{theorem pressure perm inv} and Corollary \ref%
{pression_inv_perm} make a detailed analysis of the set $\mathit{M}_{%
\mathfrak{m}}^{\sharp }\cap E_{\Pi }$ of permutation invariant equilibrium
states possible. As an example we recommend \cite{BruPedra1}, where a
complete description of permutation invariant equilibrium states for a class
of physically relevant models is performed.

Note that $\mathit{\Omega }_{\mathfrak{m}}^{\sharp }\backslash E_{\Pi }$ may
not be empty, i.e., the existence of a generalized t.i. equilibrium state
which is not permutation invariant, is, a priori, not excluded. However, for
permutation invariant models $\mathfrak{m}$, this set $\mathit{\Omega }_{%
\mathfrak{m}}^{\sharp }\backslash E_{\Pi }$ is not relevant as soon as the
weak$^{\ast }$--limit of Gibbs states is concerned:

\begin{corollary}[Weak$^{\ast }$--limit of Gibbs equilibrium states]
\label{limit Gibbs states}\mbox{ }\newline
\index{States!Gibbs}For any permutation invariant $\mathfrak{m}\in \mathcal{M%
}_{1}$, the weak$^{\ast }$--accumulation points of Gibbs equilibrium states $%
\{\rho _{l}\}_{l\in \mathbb{N}}$ belong to the set $\mathit{M}_{\mathfrak{m}%
}^{\sharp }\cap E_{\Pi }$ of permutation invariant equilibrium states.
\end{corollary}

\begin{proof}
As explained in the proof of Theorem \ref{theorem pressure perm inv}, the
state $\rho _{l}\in E_{\Lambda _{l}}$ (\ref{Gibbs.state}) associated with $%
U_{\Lambda _{l}}$ allows us to define a space--averaged t.i. Gibbs state $%
\hat{\rho}_{l}\in E_{1}$. The sequences $\{\rho _{l}\}_{l\in \mathbb{N}}$
and $\{\hat{\rho}_{l}\}_{l\in \mathbb{N}}$ have the same weak$^{\ast }$%
--accumulation points which all belong to $E_{\Pi }$ because $\rho _{l}$ is
invariant under permutations $\pi \in \Pi $ such that $\pi |_{\mathfrak{L}%
\backslash \Lambda _{l}}=\mathrm{id}|_{\mathfrak{L}\backslash \Lambda _{l}}$%
. Therefore, the corollary is a direct consequence of Theorem \ref{theorem
pressure perm inv} combined with Equation (\ref{perm inv upper bound})
extended to any permutation invariant model $\mathfrak{m}\in \mathcal{M}_{1}$
(instead of discrete models only). 
\end{proof}%

\chapter{Analysis of the Pressure via t.i. States\label{section proof of
theorem main}}

\setcounter{equation}{0}%
The aim of this chapter is to prove Theorem \ref{BCS main theorem 1}. This
proof is broken in several lemmata. We first show in Section \ref{section
reduction of the problem} that one can reduce the computation of the
thermodynamic limit of (\ref{BCS pressure}), for any $\mathfrak{m}\in 
\mathcal{M}_{1}$, to discrete finite range models 
\begin{equation*}
\{\Phi \}\cup \{\Phi _{k},\Phi _{k}^{\prime }\}_{k=1}^{N}\in \mathcal{M}%
_{1}^{\mathrm{df}}:=\mathcal{M}_{1}^{\mathrm{d}}\cap \mathcal{M}_{1}^{%
\mathrm{f}}\subseteq \mathcal{M}_{1},
\end{equation*}%
see Corollary \ref{lemma reduction dfbisbis}. Then in Section \ref{Section
passivity of gibbs states} we use the so--called passivity of Gibbs states
(Theorem \ref{passivity.Gibbs}) to find the thermodynamic limit of (\ref{BCS
pressure}), for any $\mathfrak{m}\in \mathcal{M}_{1}^{\mathrm{df}}$, from
which we deduce Theorem \ref{BCS main theorem 1}, see Theorem \ref{BCS main
theorem 1 copy(1)}.

\section{Reduction to discrete finite range models\label{section reduction
of the problem}}

From the density of the set of finite range interactions in $\mathcal{W}_{1}$%
, recall that the sub--space $\mathcal{M}_{1}^{\mathrm{df}}:=\mathcal{M}%
_{1}^{\mathrm{d}}\cap \mathcal{M}_{1}^{\mathrm{f}}$ of discrete finite range
models is dense in $\mathcal{M}_{1}$. As a consequence, the thermodynamic
limit 
\begin{equation*}
\lim_{l\rightarrow \infty }p_{l,\mathfrak{m}}=\lim_{l\rightarrow \infty
}\left\{ \frac{1}{\beta |\Lambda _{l}|}\ln \mathrm{Trace}_{\wedge \mathcal{H}%
_{\Lambda _{l}}}(\mathrm{e}^{-\beta U_{l}})\right\}
\end{equation*}%
of (\ref{BCS pressure}), for any $\mathfrak{m}\in \mathcal{M}_{1}$, can be
found by using a sequence $\{\mathfrak{m}_{n}\}_{n\in \mathbb{N}}\subseteq 
\mathcal{M}_{1}^{\mathrm{df}}$ of discrete finite range models converging to 
$\mathfrak{m}$. This result follows from the next two lemmata:

\begin{lemma}[Equicontinuity of the map $\mathfrak{m}\mapsto p_{\ell ,%
\mathfrak{m}}$]
\label{lemma reduction df}\mbox{ }\newline
The family of maps $\mathfrak{m}\mapsto p_{l,\mathfrak{m}}$ is equicontinuous%
\footnote{%
For each sequence $\{\mathfrak{m}_{n}\}_{n\in \mathbb{N}}\subset \mathcal{M}%
_{1}$ converging to $\mathfrak{m}$, $p_{\ell ,\mathfrak{m}_{n}}$ converges
uniformly in $\ell \in \mathbb{N}$ to $p_{\ell ,\mathfrak{m}}.$} for $l\in 
\mathbb{N}$. Then, $\mathfrak{m}\mapsto \mathrm{P}_{\mathfrak{m}}^{\sharp }$
(Definition \ref{Pressure}) is a locally Lipschitz continuous map from $%
\mathcal{M}_{1}$ to $\mathbb{R}$.
\end{lemma}

\begin{proof}
For any $\mathfrak{m}_{1},\mathfrak{m}_{2}\in \mathcal{M}_{1}$ observe that
the corresponding internal energies $U_{l,1}$ and $U_{l,2}$\ (Definition \ref%
{definition BCS-type model}) satisfy the bound%
\begin{equation}
\left\Vert U_{l,1}-U_{l,2}\right\Vert \leq \left\vert \Lambda
_{l}\right\vert \left\Vert \mathfrak{m}_{1}-\mathfrak{m}_{2}\right\Vert _{%
\mathcal{M}_{1}}\left( 1+\left\Vert \mathfrak{m}_{1}\right\Vert _{\mathcal{M}%
_{1}}+\left\Vert \mathfrak{m}_{2}\right\Vert _{\mathcal{M}_{1}}\right) .
\label{bound BCS norms}
\end{equation}%
In particular, the map $\mathfrak{m}\longmapsto U_{l}$ is continuous at
fixed $l\in \mathbb{N}$. For each sequence $\{\mathfrak{m}_{n}\}_{n\in 
\mathbb{N}}\subseteq \mathcal{M}_{1}$ converging to $\mathfrak{m}$, from (%
\ref{bound BCS norms}) and the bound (\ref{petite inequality}), that is, in
this special case, 
\begin{eqnarray}
\left\vert p_{l,\mathfrak{m}_{1}}-p_{l,\mathfrak{m}_{2}}\right\vert &=&\frac{%
1}{\beta |\Lambda _{l}|}\left\vert \ln \mathrm{Trace}_{\wedge \mathcal{H}%
_{\Lambda _{l}}}(\mathrm{e}^{-\beta U_{l,1}})-\ln \mathrm{Trace}_{\wedge 
\mathcal{H}_{\Lambda _{l}}}(\mathrm{e}^{-\beta U_{l,2}})\right\vert  \notag
\\
&\leq &\frac{1}{|\Lambda _{l}|}\left\Vert U_{l,1}-U_{l,2}\right\Vert ,
\label{petite inequality2}
\end{eqnarray}%
we obtain the upper bound 
\begin{equation*}
\left\vert p_{l,\mathfrak{m}_{n}}-p_{l,\mathfrak{m}}\right\vert \leq
\left\Vert \mathfrak{m}_{n}-\mathfrak{m}\right\Vert _{\mathcal{M}_{1}}\left(
1+\left\Vert \mathfrak{m}_{n}\right\Vert _{\mathcal{M}_{1}}+\left\Vert 
\mathfrak{m}\right\Vert _{\mathcal{M}_{1}}\right) .
\end{equation*}%
This bound leads to the equicontinuity of the family of maps $\mathfrak{m}%
\mapsto p_{l}$ for $l\in \mathbb{N}$ and the locally Lipschitz continuity of
the map $\mathfrak{m}\mapsto \mathrm{P}_{\mathfrak{m}}^{\sharp }$. 
\end{proof}%

\begin{lemma}[Equicontinuity of the map $\mathfrak{m}\mapsto f_{\mathfrak{m}%
}^{\sharp }(\protect\rho )$]
\label{lemma reduction dfbis}\mbox{ }\newline
The family of maps $\mathfrak{m}\mapsto f_{\mathfrak{m}}^{\sharp }(\rho )$
is equicontinuous for $\rho \in E_{1}$. Then, for any sequence $\{\mathfrak{m%
}_{n}\}_{n\in \mathbb{N}}\subseteq \mathcal{M}_{1}$ converging to $\mathfrak{%
m}\in \mathcal{M}_{1}$,%
\begin{equation*}
\inf\limits_{\rho \in E_{1}}\,f_{\mathfrak{m}}^{\sharp }(\rho )=\underset{%
n\rightarrow \infty }{\lim }\inf\limits_{\rho \in E_{1}}\,f_{\mathfrak{m}%
_{n}}^{\sharp }(\rho ).
\end{equation*}
\end{lemma}

\begin{proof}%
This lemma is a consequence of the norm equicontinuity of the family of maps 
$\Phi \mapsto e_{\Phi }(\rho )$ and 
\begin{equation*}
\mathfrak{m}\mapsto \left\Vert \Delta _{a,+}\left( \rho \right) \right\Vert
_{1}-\left\Vert \Delta _{a,-}\left( \rho \right) \right\Vert _{1}
\end{equation*}%
for $\rho \in E_{1}$. Indeed, for all $\mathfrak{m}_{1},\mathfrak{m}_{2}\in 
\mathcal{M}_{1}$ and $\rho \in E_{1}$, the corresponding functionals $\Delta
_{a,\pm }^{(1)}$ and $\Delta _{a,\pm }^{(2)}$ satisfy the inequality%
\begin{eqnarray*}
&&\left\vert \int_{\mathcal{A}}\Delta _{a,\pm }^{(1)}\left( \rho \right) 
\mathrm{d}\mathfrak{a}\left( a\right) -\int_{\mathcal{A}}\Delta _{a,\pm
}^{(2)}\left( \rho \right) \mathrm{d}\mathfrak{a}\left( a\right) \right\vert
\\
&\leq &\left\Vert \mathfrak{m}_{1}-\mathfrak{m}_{2}\right\Vert _{\mathcal{M}%
_{1}}\left( 1+\left\Vert \mathfrak{m}_{1}\right\Vert _{\mathcal{M}%
_{1}}+\left\Vert \mathfrak{m}_{2}\right\Vert _{\mathcal{M}_{1}}\right)
\end{eqnarray*}%
for all $\rho \in E_{1}$. 
\end{proof}%

Therefore, by using Lemmata \ref{lemma reduction df}--\ref{lemma reduction
dfbis}, we can assume, without loss of generality, that 
\begin{equation*}
\mathfrak{m}:=\{\Phi \}\cup \{\Phi _{k},\Phi _{k}^{\prime }\}_{k=1}^{N}\in 
\mathcal{M}_{1}^{\mathrm{df}}:=\mathcal{M}_{1}^{\mathrm{d}}\cap \mathcal{M}%
_{1}^{\mathrm{f}}\subseteq \mathcal{M}_{1}
\end{equation*}%
in order to prove Theorem \ref{BCS main theorem 1}. Indeed, using the
density of the set $\mathcal{M}_{1}^{\mathrm{df}}$ in $\mathcal{M}_{1}$, we
deduce from Lemmata \ref{lemma reduction df} and \ref{lemma reduction dfbis}
the following corollary:

\begin{corollary}[Reduction to discrete finite range models]
\label{lemma reduction dfbisbis}\mbox{ }\newline
For any $\mathfrak{m}\in \mathcal{M}_{1}$, there exists a sequence $\{%
\mathfrak{m}_{n}\}_{n\in \mathbb{N}}\subseteq \mathcal{M}_{1}^{\mathrm{df}}$
converging to $\mathfrak{m}\in \mathcal{M}_{1}$ such that 
\begin{equation*}
\mathrm{P}_{\mathfrak{m}}^{\sharp }=\underset{n\rightarrow \infty }{\lim }%
\mathrm{P}_{\mathfrak{m}_{n}}^{\sharp }\text{\quad and\quad }%
\inf\limits_{\rho \in E_{1}}\,f_{\mathfrak{m}}^{\sharp }(\rho )=\underset{%
n\rightarrow \infty }{\lim }\inf\limits_{\rho \in E_{1}}\,f_{\mathfrak{m}%
_{n}}^{\sharp }(\rho ).
\end{equation*}
\end{corollary}

\section{Passivity of Gibbs states and thermodynamics\label{Section
passivity of gibbs states}}

\index{States!Gibbs}%
\index{Passivity of Gibbs states}From Theorem \ref{passivity.Gibbs}, the
pressure $p_{l}=p_{l,\mathfrak{m}}$ (\ref{BCS pressure}) of any finite range
discrete model $\mathfrak{m}\in \mathcal{M}_{1}^{\mathrm{df}}$ is bounded
from below, for all states $\rho \in E$, by (\ref{BCS equation 5}) with
Equality (\ref{BCS equation 4}) for $\rho =\rho _{l}$. Recall that $\rho
_{l}:=\rho _{\Lambda _{l},U_{l}}$ is the Gibbs equilibrium state (\ref%
{Gibbs.state}) with internal energy $U_{l}$ defined in Definition \ref%
{definition BCS-type model} for any $\mathfrak{m}\in \mathcal{M}_{1}$ and $%
l\in \mathbb{N}$. This even state $\rho _{l}$ is seen as defined either on
the local algebra $\mathcal{U}_{\Lambda _{l}}$ or on the whole algebra $%
\mathcal{U}$ by periodically extending it (with period $(2l+1)$ in each
direction of the lattice $\mathfrak{L}$).

Thus, for any $\mathfrak{m}\in \mathcal{M}_{1}^{\mathrm{df}}$, the lower
bound on the pressure $p_{l}$ in the thermodynamic limit is found by
studying the r.h.s. of (\ref{BCS equation 5}) as $l\rightarrow \infty $:

\begin{lemma}[Thermodynamic limit of the pressure $p_{l}$ -- lower bound]
\label{Lemma lower bound pressure copy(1)}\mbox{ }\newline
For any $\mathfrak{m}\in \mathcal{M}_{1}^{\mathrm{df}}$, 
\begin{equation*}
\liminf\limits_{l\rightarrow \infty }p_{l}\geq -\inf\limits_{\rho \in
E_{1}}\,f_{\mathfrak{m}}^{\sharp }(\rho ),
\end{equation*}%
with the free--energy density functional $f_{\mathfrak{m}}^{\sharp }$
defined in Definition \ref{Free-energy density long range}.
\end{lemma}

\begin{proof}
The first term in the r.h.s. of (\ref{BCS equation 5}) is the only one we
really need to control. To this purpose, observe that, for any $\Phi \in 
\mathcal{W}_{1}$ and $l\in \mathbb{N}$, the space--average $%
\widehat{\mathfrak{e}}_{\Phi ,l}$ (\ref{perm inv 1bis}) of the energy
observable $\mathfrak{e}_{\Phi }$ (\ref{eq:enpersite}) is obviously a
bounded operator. Hence, by using 
\begin{equation*}
\widehat{\mathfrak{e}}_{\Phi ,l}-|\Lambda _{l}|^{-1}U_{\Lambda _{l}}^{\Phi
}=\sum\limits_{\Lambda \in \mathcal{P}_{f}(\mathfrak{L}),\Lambda \ni 0}\frac{%
1}{|\Lambda _{l}|}\sum\limits_{x\in \Lambda _{l}}\mathbf{1}_{\left\{ \Lambda
\nsubseteq (\Lambda _{l}-x)\right\} }\frac{\Phi _{\Lambda +x}}{|\Lambda |},
\end{equation*}%
$\Vert \Phi _{\Lambda +x}\Vert =\Vert \Phi _{\Lambda }\Vert $, $\Vert \Phi
\Vert _{\mathcal{W}_{1}}<\infty $, and Lebesgue's dominated convergence
theorem, we have that 
\begin{equation}
\underset{l\rightarrow \infty }{\lim }\left\Vert \widehat{\mathfrak{e}}%
_{\Phi ,l}-|\Lambda _{l}|^{-1}U_{\Lambda _{l}}^{\Phi }\right\Vert =0.
\label{eq sup mean enerfybis0}
\end{equation}%
Therefore, by using the definition $\mathfrak{C}_{l}:=U_{\Lambda _{l}}^{\Phi
}+iU_{\Lambda _{l}}^{\Phi ^{\prime }}$ for any $l\in \mathbb{N}$ and any
finite range interaction $\Phi \in \mathcal{W}_{1}^{\mathrm{f}}$, we obtain%
\begin{equation}
\underset{l\rightarrow \infty }{\lim }\left\{ \frac{1}{|\Lambda _{l}|^{2}}%
\rho (\mathfrak{C}_{l}^{\ast }\mathfrak{C}_{l})-\rho ((\widehat{\mathfrak{e}}%
_{\Phi ,l}+i\widehat{\mathfrak{e}}_{\Phi ^{\prime },l})^{\ast }(\widehat{%
\mathfrak{e}}_{\Phi ,l}+i\widehat{\mathfrak{e}}_{\Phi ^{\prime
},l}))\right\} =0  \label{eq sup mean enerfybis}
\end{equation}%
uniformly in $\rho \in E$. Consequently, the lower bound on the pressure $%
p_{l}$ as $l\rightarrow \infty $ follows from (\ref{BCS equation 5})
combined with Definitions \ref{definition de deltabis}, \ref{entropy.density}%
, \ref{definition energy density}, and (\ref{eq sup mean enerfybis}).%
\end{proof}%

In order to obtain the upper bound on the $\limsup $ of the pressure $p_{l}$%
, as in the proof of Theorem \ref{theorem pressure perm inv}, one needs to
control each term in (\ref{BCS equation 4}) when $l\rightarrow \infty $.
Observe that $\rho _{l}$\ is \emph{generally not t.i.} even if $\mathfrak{m}%
\in \mathcal{M}_{1}^{\mathrm{df}}$ is t.i., by definition. But, we can
canonically construct a space--averaged t.i. Gibbs state $\hat{\rho}_{l}$
from $\rho _{l}$, see (\ref{t.i. state rho l}). If we restrict ourselves to
the case of models with purely repulsive long--range interactions (i.e. $%
\Phi _{a,-}=\Phi _{a,-}^{\prime }=0$ (a.e.)), we can analyze each term in (%
\ref{BCS equation 4}) as a function of $\hat{\rho}_{l}\in E_{1}$ in the
limit $l\rightarrow \infty $. The mean entropy per volume as a function of
the t.i. state $\hat{\rho}_{l}$ (\ref{t.i. state rho l}) is already given in
the proof of Theorem \ref{theorem pressure perm inv} by Equality (\ref{mean
entropy per volume}). The analysis of the other terms is, however, more
involved than for permutation invariant models (Definition \ref{Definition
permutation inv models}). The first term of the r.h.s. of (\ref{BCS equation
4}) being the most problematic one if we tries to use the space--averaged
t.i. Gibbs state $\hat{\rho}_{l}$ as test states.

We now prove that, at large $l$, the internal energy computed from a large
box $\Lambda _{l}^{(n)}$ (\ref{lambda_n_l}) is the same as the one for $%
|\Lambda _{n}|$ copies of boxes of volume $|\Lambda _{l}|$. This is a
standard method often used in statistical mechanics to prove the existence
of the thermodynamic limit.

\begin{lemma}[Internal energy]
\label{lemma internal energy}\mbox{ }\newline
For any finite range t.i. interaction $\Phi \in \mathcal{W}_{1}^{\mathrm{f}}$%
,%
\begin{equation*}
\underset{n\in \mathbb{N}}{\sup }\left\{ \frac{1}{|\Lambda _{l}^{(n)}|}\Vert
U_{\Lambda _{l}^{(n)}}^{\Phi }-\sum\limits_{x\in \Lambda _{n}}U_{\Lambda
_{l}+(2l+1)x}^{\Phi }\Vert \right\} =\mathcal{O}(l^{-1}).
\end{equation*}
\end{lemma}

\begin{proof}
From Definition \ref{definition standard interaction} (ii) of $U_{\Lambda
}^{\Phi }$, it is straightforward to check, for any t.i. finite range
interaction $\Phi $, that 
\begin{eqnarray*}
\frac{1}{|\Lambda _{l}^{(n)}|}\Vert U_{\Lambda _{l}^{(n)}}^{\Phi
}-\sum\limits_{x\in \Lambda _{n}}U_{\Lambda _{l}+(2l+1)x}^{\Phi }\Vert &\leq
&\frac{|\Lambda _{n}|}{|\Lambda _{l}^{(n)}|}\sum\limits_{\Lambda \subseteq
\partial \Lambda _{l}}\Vert \Phi _{\Lambda }\Vert \\
&\leq &\frac{|\partial \Lambda _{l}|}{|\Lambda _{l}|}\Vert \Phi \Vert _{%
\mathcal{W}_{1}}=\mathcal{O}(l^{-1})
\end{eqnarray*}%
with $\partial \Lambda _{l}$ being the boundary\footnote{%
By fixing $m\geq 1$ the boundary $\partial \Lambda $ of any $\Lambda \subset
\Gamma $ is defined by $\partial \Lambda :=\{x\in \Lambda \;:\;\exists y\in
\Gamma \backslash \Lambda \mathrm{\ with\ }d(x,y)\leq m\},$ see (\ref%
{def.dist}) for the definition of the metric $d(x,y)$.\label{footnote20
copy(2)}}$\ $of the cubic box $\Lambda _{l}$ defined for large enough $m\geq
1$.%
\end{proof}%

As a consequence, as far as the limit $l\rightarrow \infty $ is concerned
one can use, for all $\Phi \in \mathcal{W}_{1}$, the energy density $e_{\Phi
}(\hat{\rho}_{l})$ instead of the mean internal energy per volume $\rho
_{l}\left( U_{\Lambda _{l}}^{\Phi }\right) /|\Lambda _{l}|$. (Recall that $%
\hat{\rho}_{l}\in E_{1}$ is the t.i. state (\ref{t.i. state rho l}).)
Indeed, one deduces from Lemma \ref{lemma internal energy} the following
result:

\begin{lemma}[Mean internal energy per volume as $\mathbb{\ell }\rightarrow
\infty $]
\label{lemma mean energy trivial-1}\mbox{ }\newline
For any $\mathfrak{m}\in \mathcal{M}_{1}$ and all finite range interactions $%
\Phi \in \mathcal{W}_{1}^{\mathrm{f}}$, 
\begin{equation*}
\left\vert e_{\Phi }(\hat{\rho}_{l})-\frac{\rho _{l}(U_{\Lambda _{l}}^{\Phi
})}{|\Lambda _{l}|}\right\vert =\mathcal{O}(l^{-1})
\end{equation*}%
with the energy density $e_{\Phi }(\rho )$ defined by Definition \ref%
{definition energy density}.
\end{lemma}

\begin{proof}
By $(2l+1)\mathbb{Z}^{d}$--invariance of Gibbs equilibrium states $\rho _{l}$%
, it follows that%
\begin{equation*}
\sum\limits_{x\in \Lambda _{n}}\rho _{l}(U_{\Lambda _{l}+(2l+1)x}^{\Phi
})=|\Lambda _{n}|\rho _{l}(U_{\Lambda _{l}}^{\Phi }).
\end{equation*}%
Consequently, by using Lemma \ref{lemma internal energy} and the limit $%
n\rightarrow \infty $, one obtains that%
\begin{equation}
\left\vert e_{\Phi }(\rho _{l})-\frac{\rho _{l}(U_{\Lambda _{l}}^{\Phi })}{%
|\Lambda _{l}|}\right\vert =\mathcal{O}(l^{-1}).  \label{eq sup mean enerfy}
\end{equation}%
The functional $\rho \mapsto e_{\Phi }(\rho )$ is affine and t.i., see Lemma %
\ref{Th.en.func} (i). Therefore $e_{\Phi }(\hat{\rho}_{l})=e_{\Phi }(\rho
_{l})$ which combined with (\ref{eq sup mean enerfy}) implies the lemma.%
\end{proof}%

The next step to find the upper bound on the $\limsup $ of the pressure $%
p_{l}$ is now to study the first term in the r.h.s of (\ref{BCS equation 4})
because the others terms can be controlled by using (\ref{mean entropy per
volume}) and Lemma \ref{lemma mean energy trivial-1}.\ The relationship of
this term with $\Delta _{\mathfrak{e}_{\Phi }+i\mathfrak{e}_{\Phi ^{\prime
}}}(\hat{\rho}_{l})$ at large $l$ is problematic (recall that $\mathfrak{e}%
_{\Phi }:=\mathfrak{e}_{\Phi ,(1,\cdots ,1)}$ and $\Delta _{A}$ are
respectively defined by (\ref{eq:enpersite}) and Definition \ref{definition
de deltabis}): On the one hand, we cannot expect the limit 
\begin{equation*}
\underset{l\rightarrow \infty }{\lim }\left( \frac{1}{|\Lambda _{l}|^{2}}%
\rho _{l}((U_{\Lambda _{l}}^{\Phi }+iU_{\Lambda _{l}}^{\Phi ^{\prime
}})^{\ast }(U_{\Lambda _{l}}^{\Phi }+iU_{\Lambda _{l}}^{\Phi ^{\prime
}}))-|\rho _{l}(\mathfrak{e}_{\Phi }+i\mathfrak{e}_{\Phi ^{\prime
}})|^{2}\right) =0
\end{equation*}%
to hold in general. Otherwise it would follow -- at least w.r.t. the
observables $\mathfrak{e}_{\Phi }$ and $\mathfrak{e}_{\Phi ^{\prime }}$ --
the absence of long--range order (LRO). On the other hand, we know -- as $%
\hat{\rho}_{l}$ are ergodic states -- that:%
\begin{equation*}
\Delta _{\mathfrak{e}_{\Phi }+i\mathfrak{e}_{\Phi ^{\prime }}}(\hat{\rho}%
_{l})=\left\vert \hat{\rho}_{l}\left( \mathfrak{e}_{\Phi }+i\mathfrak{e}%
_{\Phi ^{\prime }}\right) \right\vert ^{2}.
\end{equation*}

In the case of purely repulsive long--range coupling constants where $\Phi
_{a,-}=\Phi _{a,-}^{\prime }=0$ (a.e.) (cf. Definition \ref{long range
attraction-repulsion}), the arguments become easier because from the GNS
representation of $\rho _{l}$ combined with (\ref{eq sup mean enerfybis})
for $\rho =\rho _{l}$ we obtain that, for any $\mathfrak{m}\in \mathcal{M}%
_{1}^{\mathrm{df}}$, 
\begin{equation*}
\liminf\limits_{l\rightarrow \infty }\left\{ \frac{1}{|\Lambda _{l}|^{2}}%
\rho _{l}\left( (U_{\Lambda _{l}}^{\Phi }+iU_{\Lambda _{l}}^{\Phi ^{\prime
}})^{\ast }(U_{\Lambda _{l}}^{\Phi }+iU_{\Lambda _{l}}^{\Phi ^{\prime
}})\right) -\Delta _{\mathfrak{e}_{\Phi }+i\mathfrak{e}_{\Phi ^{\prime
}}}\left( \hat{\rho}_{l}\right) \right\} \geq 0.
\end{equation*}%
This last limit combined with (\ref{BCS equation 4}), (\ref{mean entropy per
volume}), and Lemma \ref{lemma mean energy trivial-1}, yields the desired
upper bound when $\Phi _{a,-}=0$ (a.e.), i.e., for purely repulsive
long--range models.

However, as soon as we have long--range attractions $\Phi _{a,-},\Phi
_{a,-}^{\prime }\neq 0$ (a.e.), the proof of the upper bound on the pressure
requires Corollary \ref{pression_inv_perm} as a key ingredient to obtain a
more convenient sequence of test states $\hat{\varrho}_{l}\in \mathcal{E}%
_{1} $. ($\rho _{l}$, $\hat{\rho}_{l}$, and $\hat{\varrho}_{l}$ have not
necessarily the same weak$^{\ast }$--accumulation points.) In fact, similar
arguments was first used in \cite{Petz2008} and subsequently in \cite%
{monsieurremark} for translation invariant quantum spin systems (Remark \ref%
{Quantum spin systems}). Following their strategy \cite%
{Petz2008,monsieurremark} combined with Corollary \ref{pression_inv_perm},
we obtain the desired upper bound for any $\mathfrak{m}\in \mathcal{M}_{1}^{%
\mathrm{df}}$:

\begin{lemma}[Thermodynamic limit of the pressure ${p}_{l}$ -- upper bound]
\label{Lemma lower bound pressure}\mbox{ }\newline
For any $\mathfrak{m}\in \mathcal{M}_{1}^{\mathrm{df}}$, there is a sequence 
$\{\hat{\varrho}_{l}\}_{l\in \mathbb{N}}\subseteq \mathcal{E}_{1}$ of
ergodic states such that%
\begin{equation*}
\limsup\limits_{l\rightarrow \infty }p_{l,\mathfrak{m}}=-\lim\limits_{l%
\rightarrow \infty }g_{\mathfrak{m}}\left( \hat{\varrho}_{l}\right)
=-\lim\limits_{l\rightarrow \infty }f_{\mathfrak{m}}^{\sharp }(\hat{\varrho}%
_{l})\leq -\inf\limits_{\rho \in E_{1}}\,f_{\mathfrak{m}}^{\sharp }(\rho )
\end{equation*}%
with the functional $g_{\mathfrak{m}}$ defined by Definition \ref{Reduced
free energy}.
\end{lemma}

\begin{proof}
For any $l\in \mathbb{N}$, $\Phi \in \mathcal{W}_{1}$ and $n\in \mathbb{N}%
_{0}$, define the self-adjoint elements 
\begin{equation*}
U_{l,n}^{\Phi }:=\sum\limits_{x\in \Lambda _{n}}\alpha _{(2l+1)x}(U_{\Lambda
_{l}}^{\Phi }).
\end{equation*}%
Then, for any $l,n\in \mathbb{N}$ and any discrete finite range model%
\begin{equation*}
\mathfrak{m}:=\{\Phi \}\cup \{\Phi _{k},\Phi _{k}^{\prime }\}_{k=1}^{N}\in 
\mathcal{M}_{1}^{\mathrm{df}},
\end{equation*}%
we define the internal energy $U_{l,n}$ by%
\begin{equation*}
U_{l,n}:=U_{l,n}^{\Phi }+\sum\limits_{k=1}^{N}\frac{\gamma _{k}}{|\Lambda
_{l}^{(n)}|}(U_{l,n}^{\Phi _{k}}+iU_{l,n}^{\Phi _{k}^{\prime }})^{\ast
}(U_{l,n}^{\Phi _{k}}+iU_{l,n}^{\Phi _{k}^{\prime }})
\end{equation*}%
with $\Lambda _{l}^{(n)}$ defined by (\ref{lambda_n_l}). The pressure
associated with $U_{l,n}$ is as usual defined, for $\beta \in (0,\infty )$,
by%
\begin{equation*}
p_{l,\mathfrak{m}}\left( n,\beta \right) :=\frac{1}{\beta |\Lambda
_{l}^{(n)}|}\ln \mathrm{Trace}_{\wedge \mathcal{H}_{\Lambda }}(\mathrm{e}%
^{-\beta U_{l,n}}).
\end{equation*}

Now, by using Lemma \ref{lemma internal energy} together with (\ref{petite
inequality2}), observe that%
\begin{equation}
\lim\limits_{l\rightarrow \infty }\left\{ \limsup\limits_{n\rightarrow
\infty }\left\vert p_{l,\mathfrak{m}}\left( n,\beta \right) -p_{2ln+n+l,%
\mathfrak{m}}\right\vert \right\} =0  \label{important equation}
\end{equation}%
for any $\mathfrak{m}\in \mathcal{M}_{1}^{\mathrm{df}}$. The pressure $p_{l,%
\mathfrak{m}}\left( n,\beta \right) $ can be seen as a finite-volume
pressure of a \emph{permutation invariant} model $\mathfrak{m}_{l}$ defined
as follows. Recall that the $C^{\ast }$--algebra $\mathcal{U}$ is the
fermion algebra defined in\ Section \ref{Section fermions algebra} with a
spin set $\mathrm{S}$. Then the space $\mathcal{M}_{1}=\mathcal{M}_{1}(%
\mathcal{U})$ defined by Definition \ref{definition M1bis} is the Banach
space of long--range models constructed from $\mathcal{U}$. Now, for each $%
l\in \mathbb{N}$, we define the $C^{\ast }$--algebra $\mathcal{U}_{l}$ to be
the fermion algebra with spin set $\mathrm{S}\times \Lambda _{l}$, and in
the same way $\mathcal{M}_{1}$ is defined, we construct from $\mathcal{U}%
_{l} $ the Banach space $\mathcal{M}_{1}(\mathcal{U}_{l})$ of long--range
models. For any $x\in \Lambda _{n}$, note that the sub--algebra $(\mathcal{U}%
_{l})_{\{x\}}$ of $\mathcal{U}_{l}$ can be canonically identified with the
sub--algebra $\mathcal{U}_{\Lambda _{l}+(2l+1)x}$ of $\mathcal{U}$. At $l\in 
\mathbb{N}$ and for any $\mathfrak{m}\in \mathcal{M}_{1}^{\mathrm{df}}$, the
permutation invariant discrete long--range model $\mathfrak{m}_{l}$ is the
element 
\begin{equation*}
\mathfrak{m}_{l}:=\{\Phi ^{(l)}\}\cup \{\Phi _{k}^{(l)},(\Phi
^{(l)})_{k}^{\prime }\}_{k=1}^{N}\in \mathcal{M}_{1}(\mathcal{U}_{l})
\end{equation*}%
uniquely defined by the conditions 
\begin{equation*}
\Phi _{\{0\}}^{(l)}:=|\Lambda _{l}|^{-1}U_{\Lambda _{l}}^{\Phi },\quad (\Phi
_{k}^{(l)})_{\{0\}}:=|\Lambda _{l}|^{-1}U_{\Lambda _{l}}^{\Phi _{k}},\quad
((\Phi ^{(l)})_{k}^{\prime })_{\{0\}}:=|\Lambda _{l}|^{-1}U_{\Lambda
_{l}}^{\Phi _{k}^{\prime }}
\end{equation*}%
with $\Phi _{\Lambda }^{(l)}=(\Phi _{k}^{(l)})_{\Lambda }=((\Phi
^{(l)})_{k}^{\prime })_{\Lambda }=0$ whenever $|\Lambda |\not=1$.

By using these definitions, we have 
\begin{equation*}
p_{l,\mathfrak{m}}\left( n,\beta \right) =p_{n,\mathfrak{m}_{l}}\left(
0,\beta _{l}\right)
\end{equation*}%
with $\beta _{l}:=|\Lambda _{l}|\beta $. Therefore, we are in position to
use Corollary \ref{pression_inv_perm} in order to compute the thermodynamic
limit $n\rightarrow \infty $ of the permutation invariant discrete model $%
\mathfrak{m}_{l}\in \mathcal{M}_{1}(\mathcal{U}_{l})$:%
\begin{eqnarray*}
\underset{n\rightarrow \infty }{\lim }p_{l,\mathfrak{m}}\left( n,\beta
\right) &=&\underset{n\rightarrow \infty }{\lim }p_{n,\mathfrak{m}%
_{l}}\left( 0,\beta _{l}\right) =-\inf_{\rho _{\Lambda _{l}}\in E_{\Lambda
_{l}}}\left\{ \sum\limits_{k=1}^{N}\gamma _{k}|\Lambda _{l}|^{-2}|\rho
_{\Lambda _{l}}(U_{\Lambda _{l}}^{\Phi _{k}}+iU_{\Lambda _{l}}^{\Phi
_{k}^{\prime }})|^{2}\right. \\
&&\left. \overset{}{\underset{}{}}+|\Lambda _{l}|^{-1}\rho _{\Lambda
_{l}}(U_{\Lambda _{l}}^{\Phi })-(\beta |\Lambda _{l}|)^{-1}S(\rho _{\Lambda
_{l}})\right\}
\end{eqnarray*}%
with the weak$^{\ast }$--continuous functional $S$ being the von Neumann
entropy defined by (\ref{neuman entropy}). This variational problem is a
minimization of a weak$^{\ast }$--continuous functional over the set $%
E_{\Lambda _{l}}$ of all (local) states on the finite dimensional algebra $%
\mathcal{U}_{\Lambda _{l}}$. Therefore, for each $l\in \mathbb{N}$, it has a
minimizer $\varrho _{l}\in E_{\Lambda _{l}}$ which can also be seen as a
state on the whole algebra $\mathcal{U}$ by periodically extending it (with
period $(2l+1)$ in each direction of the lattice $\mathfrak{L}$). We define
from $\varrho _{l}\in E$ the t.i. space--averaged state 
\begin{equation*}
\hat{\varrho}_{l}:=\frac{1}{|\Lambda _{l}|}\sum\limits_{x\in \Lambda
_{l}}\varrho _{l}\circ \alpha _{x}\in \mathcal{E}_{1}
\end{equation*}%
(compare this definition with (\ref{equation toto}) for $\vec{\ell}%
=(1,\cdots ,1)$). Recall that $\hat{\varrho}_{l}$ is ergodic (and thus
extremal), as shown in the proof of Lemma \ref{lemma density of extremal
points}. Then, by using $\Delta _{A}\left( \hat{\varrho}_{l}\right) =|\hat{%
\varrho}_{l}(A)|^{2}$ (see Theorem \ref{Lemma1.vonN} (iv)), Equality (\ref%
{mean entropy per volume}) and Lemma \ref{lemma mean energy trivial-1}
applied to states $\varrho _{l}\in E$ and $\hat{\varrho}_{l}\in \mathcal{E}%
_{1}$, we obtain that 
\begin{equation}
\lim\limits_{l\rightarrow \infty }\underset{n\rightarrow \infty }{\lim }p_{l,%
\mathfrak{m}}\left( n,\beta \right) =-\lim\limits_{l\rightarrow \infty }g_{%
\mathfrak{m}}\left( \hat{\varrho}_{l}\right) =-\lim\limits_{l\rightarrow
\infty }f_{\mathfrak{m}}^{\sharp }(\hat{\varrho}_{l}),
\label{important limit}
\end{equation}%
see also Lemma \ref{lemma property free--energy density functional} (ii).
Therefore, the limits (\ref{important equation}) and (\ref{important limit})
yield the lemma. 
\end{proof}%

Consequently, Theorem \ref{BCS main theorem 1} is a direct consequence of
Lemmata \ref{lemma property free--energy density functional copy(1)}, \ref%
{lemma reduction df}, \ref{Lemma lower bound pressure copy(1)}, and \ref%
{Lemma lower bound pressure} together with Corollary \ref{lemma reduction
dfbisbis}. In fact, we obtain a bit more than Theorem \ref{BCS main theorem
1}. Indeed, by combining Theorem \ref{BCS main theorem 1} with Theorem \ref%
{passivity.Gibbs}, (\ref{eq sup mean enerfybis}) and the fact that the
space--average $\widehat{\mathfrak{e}}_{\Phi ,l}$ (\ref{perm inv 1bis}) is
uniformly bounded by $\Vert \Phi \Vert _{\mathcal{W}_{1}}$ for $l\in \mathbb{%
N}$, we show that the map%
\index{Free--energy density functional!long--range!extension|textbf} 
\begin{eqnarray}
\rho \mapsto \mathfrak{F}_{\mathfrak{m}}^{\sharp }\left( \rho \right) : &=&%
\underset{l\rightarrow \infty }{\limsup }\left\{ \int_{\mathcal{A}}\gamma
_{a}\rho ((%
\widehat{\mathfrak{e}}_{\Phi _{a},l}+i\widehat{\mathfrak{e}}_{\Phi
_{a}^{\prime },l})^{\ast }(\widehat{\mathfrak{e}}_{\Phi _{a},l}+i\widehat{%
\mathfrak{e}}_{\Phi _{a}^{\prime },l}))\mathrm{d}\mathfrak{a}\left( a\right)
\right.  \notag \\
&&\left. +\frac{1}{|\Lambda _{l}|}\rho \left( U_{\Lambda _{l}}^{\Phi
}\right) -\frac{1}{\beta |\Lambda _{l}|}S(\rho _{\Lambda _{l}})\right\}
\label{extension of fdiese}
\end{eqnarray}%
from $E$ to $\mathbb{R}$ makes sense, as the quantity in the $\limsup $
above is uniformly bounded in $l\in \mathbb{N}$. Furthermore, for any $\rho
\in E_{\vec{\ell}}$, $\mathfrak{F}_{\mathfrak{m}}^{\sharp }\left( \rho
\right) =f_{\mathfrak{m}}^{\sharp }\left( \rho \right) $ because the $%
\limsup $ in the definition of $\mathfrak{F}_{\mathfrak{m}}^{\sharp }$ above
can be changed into a $\lim $ on the set $E_{\vec{\ell}}$ of $\mathbb{Z}_{%
\vec{\ell}}^{d}$--invariant states for any $\vec{\ell}\in \mathbb{N}^{d}$.
See also Corollary \ref{corollary property free--energy density functional}
(i). By deriving upper and lower bounds for the pressure w.r.t. $\mathfrak{F}%
_{\mathfrak{m}}^{\sharp }\left( \rho \right) $, exactly in the same way we
did for $f_{\mathfrak{m}}^{\sharp }\left( \rho \right) $, we get the
following theorem:

\begin{theorem}[Pressure $\mathrm{P}_{\mathfrak{m}}^{\sharp }$ as
variational problems on states]
\label{BCS main theorem 1 copy(1)}\mbox{ }\newline
\emph{(i)} 
\index{Pressure!variational problems}For $%
\vec{\ell}\in \mathbb{N}^{d}$ and any $\mathfrak{m}\in \mathcal{M}_{1}$,%
\begin{equation*}
\mathrm{P}_{\mathfrak{m}}^{\sharp }:=\underset{l\rightarrow \infty }{\lim }%
\left\{ p_{l}\right\} =-\inf\limits_{\rho \in E}\,\mathfrak{F}_{\mathfrak{m}%
}^{\sharp }\left( \rho \right) =-\inf\limits_{\rho \in E_{\vec{\ell}}}\,f_{%
\mathfrak{m}}^{\sharp }(\rho )=-\inf\limits_{\rho \in E_{1}}\,f_{\mathfrak{m}%
}^{\sharp }(\rho )<\infty .
\end{equation*}%
\emph{(ii)} The map $\mathfrak{m}\mapsto \mathrm{P}_{\mathfrak{m}}^{\sharp }$
from $\mathcal{M}_{1}$ to $\mathbb{R}$ is locally Lipschitz continuous.
\end{theorem}

The two infima, respectively over the set $E$ and $E_{\vec{\ell}}$ of
Theorem \ref{BCS main theorem 1 copy(1)} (i), are not really used in the
sequel as we concentrate our attention on t.i. states. These results are
only discussed in Section \ref{Section Gibbs versus gen eq states}.

\begin{remark}[Convexity of the functional $\mathfrak{F}_{\mathfrak{m}%
}^{\sharp }$]
\mbox{ }\newline
As $\mathfrak{F}_{\mathfrak{m}}^{\sharp }\left( \rho \right) $ is defined by
a $\limsup $, by using the property \textbf{S4} of the von Neumann entropy,
it is easy to check that the map $\rho \mapsto \mathfrak{F}_{\mathfrak{m}%
}^{\sharp }\left( \rho \right) $ from $E$ to $\mathbb{R}$ is a convex
functional.
\end{remark}

\chapter{Purely Attractive Long--Range Fermi Systems%
\index{Long--range models!purely attractive}}

\setcounter{equation}{0}%
Recall that generalized t.i. equilibrium states are defined to be weak$%
^{\ast }$--limit points of approximating minimizers of the free--energy
density functional $f_{\mathfrak{m}}^{\sharp }$, see Definition \ref%
{definition equilibirum state}. It is, a priori, not clear that the first
variational problem%
\begin{equation*}
\mathrm{P}_{\mathfrak{m}}^{\sharp }=-\inf\limits_{\rho \in E_{1}}\,f_{%
\mathfrak{m}}^{\sharp }(\rho )
\end{equation*}%
given in Theorem \ref{BCS main theorem 1} (i) has any minimizer. The problem
comes from the fact that $f_{\mathfrak{m}}^{\sharp }$ is generally not weak$%
^{\ast }$--lower semi--continuous because of the long--range repulsions, see
discussions after Lemma \ref{lemma property free--energy density functional}%
. As a consequence, models without long--range repulsions (Definition \ref%
{long range attraction-repulsion} ($+$)), i.e., with $\Phi _{a,+}=\Phi
_{a,+}^{\prime }=0$ (a.e.), are the easiest case to handle. This specific
case is analyzed in this chapter also because it is necessary to understand
the variational problem $\mathrm{F}_{\mathfrak{m}}^{\flat }$ of the
thermodynamic game defined in Definition \ref{definition two--person
zero--sum game} and studied in Section \ref{Section optimizations problems}.

Thermodynamics of models without long--range repulsions is then discussed in
Section \ref{equilibirum.paragraph negative}. We start, indeed, in Section %
\ref{Section Preliminaries copy(1)} with some preliminary results about the
thermodynamics of approximating interactions of long--range models, see
Definition \ref{definition BCS-type model approximated}.

\section{Thermodynamics of approximating interactions\label{Section
Preliminaries copy(1)}%
\index{Interaction!approximating!thermodynamics}}

As a preliminary step, we describe the thermodynamic limit 
\begin{equation*}
P_{\mathfrak{m}}\left( c_{a}\right) :=\underset{l\rightarrow \infty }{\lim }%
\left\{ p_{l}\left( c_{a}\right) \right\}
\end{equation*}%
of the pressure $p_{l}\left( c_{a}\right) $ (\ref{pression approximated})
associated with the internal energy $U_{l}(c_{a}):=U_{\Lambda _{l}}^{\Phi
(c_{a})}$ (\ref{internal and surface energies approximated}) for any $%
c_{a}\in L^{2}(\mathcal{A},\mathbb{C})$. This question is already solved by
Theorem \ref{BCS main theorem 1} for all $\mathfrak{m}\in \mathfrak{\mathcal{%
M}}_{1}$, and so, in particular for\ $(\Phi (c_{a}),0,0)\in \mathfrak{%
\mathcal{M}}_{1}$, see Definition \ref{definition BCS-type model
approximated}. We give this result together with additional properties as a
proposition:

\begin{proposition}[Pressure of approximating interactions of $\mathfrak{m}%
\in \mathcal{M}_{1}$]
\label{corrolaire sympa}\emph{(i)} 
\index{Interaction!approximating!pressure}For any $c_{a}\in L^{2}(\mathcal{A}%
,\mathbb{C})$,%
\begin{equation*}
P_{\mathfrak{m}}\left( c_{a}\right) =-\inf\limits_{\rho \in E_{1}}\,f_{%
\mathfrak{m}}\left( \rho ,c_{a}\right) =-\underset{%
\hat{\rho}\in \mathcal{E}_{1}}{\inf }f_{\mathfrak{m}}\left( \hat{\rho}%
,c_{a}\right)
\end{equation*}%
with the map $\left( \rho ,c_{a}\right) \mapsto f_{\mathfrak{m}}\left( \rho
,c_{a}\right) $ defined by (\ref{free--energy density approximated 1}), see
also (\ref{free--energy density approximated 2}) just below. \newline
\emph{(ii) }The map $c_{a}\mapsto P_{\mathfrak{m}}\left( c_{a}\right) $ from 
$L^{2}(\mathcal{A},\mathbb{C})$ to $\mathbb{R}$ is convex and Lipschitz norm
continuous as, for all $c_{a},c_{a}^{\prime }\in L^{2}(\mathcal{A},\mathbb{C}%
)$,%
\begin{equation*}
|P_{\mathfrak{m}}(c_{a})-P_{\mathfrak{m}}(c_{a}^{\prime })|\leq 2(\Vert \Phi
_{a}\Vert _{2}+\Vert \Phi _{a}^{\prime }\Vert _{2})\Vert c_{a}-c_{a}^{\prime
}\Vert _{2}.
\end{equation*}%
It is also continuous w.r.t. the weak topology on any ball $\mathcal{B}%
_{R}\left( 0\right) \subseteq L^{2}(\mathcal{A},\mathbb{C})$ of arbitrary
radius $R>0$ centered at $0$.
\end{proposition}

\begin{proof}%
The first assertion (i) is just Lemma \ref{lemma property free--energy
density functional copy(1)} and Theorem \ref{BCS main theorem 1} (i) applied
to the (local) model $(\Phi (c_{a}),0,0)\in \mathfrak{\mathcal{M}}_{1}$
because, for all $c_{a}\in L^{2}(\mathcal{A},\mathbb{C})$ and $\rho \in
E_{1} $,%
\index{Interaction!approximating!free--energy density} 
\begin{equation}
f_{\Phi (c_{a})}=f_{\mathfrak{m}}\left( \rho ,c_{a}\right) :=2%
\func{Re}\left\{ \left\langle e_{\Phi _{a}}(\rho )+ie_{\Phi _{a}^{\prime
}}(\rho ),\gamma _{a}c_{a}\right\rangle \right\} +e_{\Phi }(\rho )-\beta
^{-1}s(\rho ),  \label{free--energy density approximated 2}
\end{equation}%
see Definition \ref{Remark free energy density}. The definition of $%
\left\langle \cdot ,\cdot \right\rangle $ is given in Section \ref{Section
Preliminaries}. Thus, the Lipschitz norm continuity of the map $c_{a}\mapsto
P_{\mathfrak{m}}\left( c_{a}\right) $ is a direct consequence of (i)
together with the Cauchy--Schwarz inequality and the uniform upper bound of
Lemma \ref{Th.en.func} (ii). Knowing (i), the convexity of $c_{a}\mapsto P_{%
\mathfrak{m}}\left( c_{a}\right) $ is also easy to deduce because the map $%
c_{a}\mapsto f_{\mathfrak{m}}\left( \rho ,c_{a}\right) $ is obviously real
linear for any $\rho \in E_{1}$. The proof of the continuity of $%
c_{a}\mapsto P_{\mathfrak{m}}\left( c_{a}\right) $ w.r.t. the weak topology
on any ball $\mathcal{B}_{R}\left( 0\right) $ results from the weak
equicontinuity of the family%
\begin{equation}
\{c_{a}\mapsto f_{\mathfrak{m}}\left( \rho ,c_{a}\right) \}_{\rho \in E_{1}}
\label{family equicont1}
\end{equation}%
of real linear functionals on $\mathcal{B}_{R}\left( 0\right) $. The latter
is proven as follows.

If $\mathfrak{m}=\{\Phi \}\cup \{\Phi _{k},\Phi _{k}^{\prime
}\}_{k=1}^{N}\in \mathcal{M}_{1}^{\mathrm{d}}$ is a discrete model then the
family (\ref{family equicont1}) of maps is weakly equicontinuous on $L^{2}(%
\mathcal{A},\mathbb{C})$. This follows from the (uniform) upper bound 
\begin{equation*}
\left\vert \ \left\langle e_{\Phi _{a}}(\rho )+ie_{\Phi _{a}^{\prime }}(\rho
),\gamma _{a}c_{a}^{\prime }\right\rangle \ \right\vert \leq
\sum\limits_{k=1}^{N}\left( \Vert \Phi _{k}\Vert _{\mathcal{W}_{1}}+\Vert
\Phi _{k}^{\prime }\Vert _{\mathcal{W}_{1}}\right) \left\vert \left\langle
c_{a}^{\prime },\mathbf{1}_{I_{k}}\right\rangle \right\vert ,
\end{equation*}%
satisfied for all $\rho \in E_{1}$, where $I_{k}\in \mathfrak{A}$ are
conveniently chosen subsets of $\mathcal{A}$ such that $\mathfrak{a}\left(
I_{k}\right) <\infty $ for $k\in \{1,\ldots ,N\}$. Let $\varepsilon ,R>0$
and $\mathfrak{m}\in \mathcal{M}_{1}$. From the density of $\mathcal{M}_{1}^{%
\mathrm{d}}$ in $\mathcal{M}_{1}$ and the uniform upper bound of Lemma \ref%
{Th.en.func} (ii) combined with the Cauchy--Schwarz inequality, there is $%
\mathfrak{m}^{\prime }\in \mathcal{M}_{1}^{\mathrm{d}}$ such that, for all $%
c_{a}\in \mathcal{B}_{R}\left( 0\right) $ and $\rho \in E_{1}$,%
\begin{equation*}
\left\vert f_{\mathfrak{m}}\left( \rho ,c_{a}\right) -f_{\mathfrak{m}%
^{\prime }}\left( \rho ,c_{a}\right) \right\vert \leq \frac{\varepsilon }{3}.
\end{equation*}%
By the equicontinuity on $L^{2}(\mathcal{A},\mathbb{C})$ of the family (\ref%
{family equicont1}) of maps for any discrete models, for all $c_{a}\in 
\mathcal{B}_{R}\left( 0\right) $ there is a weak neighborhood $\mathcal{V}%
_{\epsilon }$ of $c_{a}$ such that, for all $c_{a}^{\prime }\in \mathcal{V}%
_{\epsilon }$ and all $\rho \in E_{1}$, 
\begin{equation*}
\left\vert f_{\mathfrak{m}^{\prime }}\left( \rho ,c_{a}\right) -f_{\mathfrak{%
m}^{\prime }}\left( \rho ,c_{a}^{\prime }\right) \right\vert \leq \frac{%
\varepsilon }{3}.
\end{equation*}%
Therefore, for all $c_{a}\in \mathcal{B}_{R}\left( 0\right) $, there is a
weak neighborhood $\mathcal{V}_{\epsilon }$ of $c_{a}$ such that, for all $%
c_{a}^{\prime }\in \mathcal{V}_{\epsilon }$ and all $\rho \in E_{1}$, 
\begin{equation*}
\left\vert f_{\mathfrak{m}}\left( \rho ,c_{a}^{\prime }\right) -f_{\mathfrak{%
m}}\left( \rho ,c_{a}\right) \right\vert \leq \varepsilon .
\end{equation*}%
In other words, for any $\mathfrak{m}\in \mathcal{M}_{1}$, the family (\ref%
{family equicont1}) of maps is weakly equicontinuous on $\mathcal{B}%
_{R}\left( 0\right) $ which yields the continuity of the map $c_{a}\mapsto
P_{\mathfrak{m}}\left( c_{a}\right) $ in the weak topology on $\mathcal{B}%
_{R}\left( 0\right) $. 
\end{proof}%

\section{Structure of the set $\mathit{M}_{\mathfrak{m}}^{\sharp }=\mathit{%
\Omega }_{\mathfrak{m}}^{\sharp }$ of t.i. equilibrium states\label%
{equilibirum.paragraph negative}}

\index{States!equilibrium}We analyze models without long--range repulsions
(Definition \ref{long range attraction-repulsion} ($+$)), i.e., $\mathfrak{m}%
\in \mathcal{M}_{1}$ satisfying $\Phi _{a,+}=\Phi _{a,+}^{\prime }=0$
(a.e.). Their (infinite--volume) pressure 
\begin{equation*}
\mathrm{P}_{\mathfrak{m}}:=\mathrm{P}_{\mathfrak{m}}^{\sharp }=\mathrm{P}_{%
\mathfrak{m}}^{\flat }
\end{equation*}%
defined in Definition \ref{Pressure} is already given by Theorem \ref{BCS
main theorem 1} (see also Theorem \ref{theorem purement repulsif sympa}) and
we first prove Theorem \ref{theorem saddle point}. In fact, by using the
simple inequality%
\index{Bogoliubov approximation} 
\begin{equation}
|\rho \left( A-c\right) |^{2}=|\rho \left( A\right) |^{2}-2%
\func{Re}\left\{ \rho \left( A\right) \bar{c}\right\} +|c|^{2}\geq 0
\label{eq idiot sympabis}
\end{equation}%
for any $c\in \mathbb{C}$ and $A\in \mathcal{U}$, Theorem \ref{theorem
saddle point} for models without long--range repulsions is not difficult to
show. Indeed, (\ref{eq idiot sympabis}) yields the following lemma:

\begin{lemma}[$c_{a,\pm }$--approximation of $\Vert \protect\gamma _{a,\pm }%
\protect\rho (\mathfrak{e}_{\Phi _{a}}+i\mathfrak{e}_{\Phi _{a}^{\prime
}})\Vert _{2}^{2}$]
\label{eq idiot sympa}\mbox{ }\newline
For any $\mathfrak{m}\in \mathcal{M}_{1}$ and all $\rho \in E_{1}$,%
\begin{equation*}
\underset{c_{a,\pm }\in L_{\pm }^{2}(\mathcal{A},\mathbb{C})}{\sup }\left\{
-\Vert c_{a,\pm }\Vert _{2}^{2}+2\func{Re}\{\langle e_{\Phi _{a}}(\rho
)+ie_{\Phi _{a}^{\prime }}(\rho ),c_{a,\pm }\rangle \}\right\} =\Vert \gamma
_{a,\pm }\rho (\mathfrak{e}_{\Phi _{a}}+i\mathfrak{e}_{\Phi _{a}^{\prime
}})\Vert _{2}^{2}
\end{equation*}%
with unique maximizer $d_{a,\pm }\left( \rho \right) =\gamma _{a,\pm
}(e_{\Phi _{a}}(\rho )+ie_{\Phi _{a}^{\prime }}(\rho ))$ (a.e.).
\end{lemma}

\begin{proof}%
This lemma is a direct consequence of (\ref{eq idiot sympabis}). In
particular, the solution $d_{a,\pm }\left( \rho \right) \in L_{\pm }^{2}(%
\mathcal{A},\mathbb{C})$ of the variational problem satisfies, for all $%
c_{a,-}\in L_{-}^{2}(\mathcal{A},\mathbb{C})$, the Euler--Lagrange equations%
\begin{equation*}
\func{Re}\left\{ \left\langle d_{a,\pm }\left( \rho \right) ,c_{a,\pm
}\right\rangle \right\} =\func{Re}\left\{ \left\langle e_{\Phi _{a}}(\rho
)+ie_{\Phi _{a}^{\prime }}(\rho ),c_{a,\pm }\right\rangle \right\} .
\end{equation*}%
\end{proof}%

Then Theorem \ref{theorem saddle point} for models without long--range
repulsions is a direct consequence of Theorem \ref{BCS main theorem 1} (i)
together with Proposition \ref{corrolaire sympa} and Lemma \ref{eq idiot
sympa}.

\begin{proposition}[Pressure of models without long--range repulsions]
\label{Bogo interaction negative}\mbox{ }\newline
\index{Pressure!variational problems!purely attractive models}For any $%
\mathfrak{m}\in \mathcal{M}_{1}$ satisfying $\Phi _{a,+}=\Phi _{a,+}^{\prime
}=0$ (a.e.), 
\begin{equation*}
\mathrm{P}_{\mathfrak{m}}=-\mathrm{F}_{\mathfrak{m}}^{\sharp }=-\mathrm{F}_{%
\mathfrak{m}}^{\flat }=-\underset{c_{a,-}\in \mathcal{B}_{R,-}}{\inf }%
\mathfrak{f}_{\mathfrak{m}}\left( c_{a,-},0\right) =:-\mathrm{F}_{\mathfrak{m%
}}
\end{equation*}%
with $\mathfrak{f}_{\mathfrak{m}}\left( c_{a,-},0\right) $ defined by
Definition \ref{definition approximating free--energy} and $\mathcal{B}%
_{R,-}\subseteq L_{-}^{2}(\mathcal{A},\mathbb{C})$ (\ref{definition of
positive-negative L2 space}) being a closed ball of sufficiently large
radius $R>0$ centered at $0$.
\end{proposition}

\begin{proof}
If $\Phi _{a,+}=\Phi _{a,+}^{\prime }=0$ (a.e.) then, for all extreme states 
$%
\hat{\rho}\in \mathcal{E}_{1}$, 
\begin{equation*}
f_{\mathfrak{m}}^{\sharp }\left( \hat{\rho}\right) =g_{\mathfrak{m}}\left( 
\hat{\rho}\right) =-\left\Vert \gamma _{a,-}\hat{\rho}\left( \mathfrak{e}%
_{\Phi _{a}}+i\mathfrak{e}_{\Phi _{a}^{\prime }}\right) \right\Vert
_{2}^{2}+e_{\Phi }(\hat{\rho})-\beta ^{-1}s(\hat{\rho}),
\end{equation*}%
see Lemma \ref{lemma property free--energy density functional} (ii). From
Lemma \ref{eq idiot sympa} it follows that 
\begin{equation}
\inf\limits_{\hat{\rho}\in \mathcal{E}_{1}}\,f_{\mathfrak{m}}^{\sharp }(\hat{%
\rho})=\inf\limits_{\hat{\rho}\in \mathcal{E}_{1}}\left\{ \underset{%
c_{a,-}\in L_{-}^{2}(\mathcal{A},\mathbb{C})}{\inf }\left\{ \left\Vert
c_{a,-}\right\Vert _{2}^{2}+f_{\mathfrak{m}}\left( \hat{\rho},c_{a,-}\right)
\right\} \right\}  \label{equallity for gap equation}
\end{equation}%
with $f_{\mathfrak{m}}\left( \rho ,c_{a,-}\right) $ defined by (\ref%
{free--energy density approximated 2}) for $\Phi _{a,+}=\Phi _{a,+}^{\prime
}=0$ (a.e.). The infima in Equality (\ref{equallity for gap equation})
obviously commute with each other and, by doing this, we get via Theorem \ref%
{BCS main theorem 1} (i) and Proposition \ref{corrolaire sympa} (i) that 
\begin{equation*}
\mathrm{P}_{\mathfrak{m}}=\underset{c_{a,-}\in L_{-}^{2}(\mathcal{A},\mathbb{%
C})}{\sup }\left\{ -\left\Vert c_{a,-}\right\Vert _{2}^{2}+P_{\mathfrak{m}%
}\left( c_{a,-}\right) \right\} =-\underset{c_{a,-}\in L_{-}^{2}(\mathcal{A},%
\mathbb{C})}{\inf }\mathfrak{f}_{\mathfrak{m}}\left( c_{a,-},0\right)
<\infty .
\end{equation*}%
Finally, the existence of a radius $R>0$ such that%
\begin{equation*}
\underset{c_{a,-}\in L_{-}^{2}(\mathcal{A},\mathbb{C})}{\inf }\mathfrak{f}_{%
\mathfrak{m}}\left( c_{a,-},0\right) =\underset{c_{a,-}\in \mathcal{B}_{R,-}}%
{\inf }\mathfrak{f}_{\mathfrak{m}}\left( c_{a,-},0\right)
\end{equation*}%
directly follows from the upper bound of Proposition \ref{corrolaire sympa}
(ii). 
\end{proof}%

The description of the set $\mathit{\Omega }_{\mathfrak{m}}^{\sharp }$ of
generalized t.i. equilibrium states (Definition \ref{definition equilibirum
state}) is also easy to perform when there is no long--range repulsions.
Indeed, the free--energy density functional $f_{\mathfrak{m}}^{\sharp }$
becomes weak$^{\ast }$--lower semi--continuous when $\Phi _{a,+}=\Phi
_{a,+}^{\prime }=0$ (a.e.), see discussions after Lemma \ref{lemma property
free--energy density functional}. In particular, the variational problem 
\begin{equation*}
\mathrm{P}_{\mathfrak{m}}=-\inf\limits_{\rho \in E_{1}}\,f_{\mathfrak{m}%
}^{\sharp }(\rho )
\end{equation*}%
has t.i. minimizers, i.e., $\mathit{\Omega }_{\mathfrak{m}}^{\sharp }=%
\mathit{M}_{\mathfrak{m}}^{\sharp }$ (Definition \ref{definition equilibirum
state copy(1)}). Recall that $\mathit{\Omega }_{\mathfrak{m}}^{\sharp }$ is
convex and weak$^{\ast }$--compact, by Lemma \ref{lemma minimum sympa
copy(1)}, and since $\mathit{M}_{\mathfrak{m}}^{\sharp }=\mathit{\Omega }_{%
\mathfrak{m}}^{\sharp }$ in this case, the non--empty set $\mathit{\Omega }_{%
\mathfrak{m}}^{\sharp }$ is a closed face of $E_{1}$ by Lemma \ref{lemma
minimum sympa copy(3)}. Therefore, to extract the structure of the set $%
\mathit{\Omega }_{\mathfrak{m}}^{\sharp }=\mathit{M}_{\mathfrak{m}}^{\sharp
} $, it suffices to describe extreme states $\hat{\omega}\in \mathit{\Omega }%
_{\mathfrak{m}}^{\sharp }\cap \mathcal{E}_{1}$ which are directly related
with the solutions $d_{a,-}\in \mathcal{C}_{\mathfrak{m}}^{\sharp }\subseteq
L_{-}^{2}(\mathcal{A},\mathbb{C})$ of the variational problem given in
Proposition \ref{Bogo interaction negative}:

\begin{proposition}[Gap equations]
\label{theorem structure etat equilibre copy(1)}\mbox{ }\newline
\index{Gap equations!purely attractive models}Let $\mathfrak{m}\in \mathcal{M%
}_{1}$ be a model without long--range repulsions: $\Phi _{a,+}=\Phi
_{a,+}^{\prime }=0$ (a.e.). \newline
\emph{(i)} For all ergodic states $%
\hat{\omega}\in \mathit{\Omega }_{\mathfrak{m}}^{\sharp }\cap \mathcal{E}%
_{1} $, 
\begin{equation*}
d_{a,-}:=e_{\Phi _{a}}(\hat{\omega})+ie_{\Phi _{a}^{\prime }}(\hat{\omega}%
)\in \mathcal{C}_{\mathfrak{m}}^{\sharp }
\end{equation*}%
and $\hat{\omega}\in \mathit{M}_{\Phi \left( d_{a,-}\right) }$ with $\mathit{%
M}_{\Phi \left( d_{a,-}\right) }$ being described in Lemma \ref{remark
equilibrium state approches}.\newline
\emph{(ii)} Conversely, for any fixed $d_{a,-}\in \mathcal{C}_{\mathfrak{m}%
}^{\sharp }$, $\mathit{M}_{\Phi \left( d_{a,-}\right) }\cap \mathcal{E}%
_{1}\subseteq \mathit{\Omega }_{\mathfrak{m}}^{\sharp }\cap \mathcal{E}_{1}$
and all states $\omega \in \mathit{M}_{\Phi \left( d_{a,-}\right) }$ satisfy 
\begin{equation}
d_{a,-}=e_{\Phi _{a}}(\hat{\omega})+ie_{\Phi _{a}^{\prime }}(\hat{\omega})%
\mathrm{\ (a.e.)}.\text{ }  \label{gap equation no repulsion}
\end{equation}
\end{proposition}

\begin{proof}
(i) Recall that $\mathit{\Omega }_{\mathfrak{m}}^{\sharp }=\mathit{M}_{%
\mathfrak{m}}^{\sharp }$. Any $\hat{\omega}\in \mathit{\Omega }_{\mathfrak{m}%
}^{\sharp }\cap \mathcal{E}_{1}$ is a solution of the l.h.s. of (\ref%
{equallity for gap equation}) and the solution $d_{a,-}=d_{a,-}\left( \hat{%
\omega}\right) $ of the variational problem 
\begin{equation*}
\underset{c_{a,-}\in L_{-}^{2}(\mathcal{A},\mathbb{C})}{\inf }\left\{
\left\Vert c_{a,-}\right\Vert _{2}^{2}+f_{\mathfrak{m}}\left( \hat{\omega}%
,c_{a,-}\right) \right\}
\end{equation*}%
satisfies the Euler--Lagrange equations (\ref{gap equation no repulsion}),
by Lemma \ref{eq idiot sympa}. The two infima in (\ref{equallity for gap
equation}) commute with each other. It is what it is done above to prove
Proposition \ref{Bogo interaction negative}. Therefore, $d_{a,-}\left( \hat{%
\omega}\right) \in \mathcal{C}_{\mathfrak{m}}^{\sharp }$ and, by (\ref%
{free--energy density approximated 2}), $\hat{\omega}$ belongs to the set $%
\mathit{M}_{\Phi \left( d_{a,-}\right) }=\mathit{\Omega }_{\Phi \left(
d_{a,-}\right) }$ of t.i. equilibrium states of the approximating
interaction $\Phi \left( d_{a,-}\right) $.

(ii) Any $d_{a,-}\in \mathcal{C}_{\mathfrak{m}}^{\sharp }$ is solution of
the variational problem given in Proposition \ref{Bogo interaction negative}%
, that is, 
\begin{equation}
\underset{c_{a,-}\in L_{-}^{2}(\mathcal{A},\mathbb{C})}{\inf }\left\{
\left\Vert c_{a,-}\right\Vert _{2}^{2}+\inf\limits_{\rho \in E_{1}}f_{%
\mathfrak{m}}\left( \rho ,c_{a,-}\right) \right\} ,
\label{variational problem encore stupide}
\end{equation}%
see Proposition \ref{corrolaire sympa} (i). Since the two infima in (\ref%
{variational problem encore stupide}) commute with each other as before, any
t.i. equilibrium state $\omega \in \mathit{M}_{\Phi \left( d_{a,-}\right) }$
satisfies (\ref{gap equation no repulsion}) because of Lemma \ref{eq idiot
sympa}, and $\mathit{M}_{\Phi \left( d_{a,-}\right) }\cap \mathcal{E}%
_{1}\subseteq \mathit{\Omega }_{\mathfrak{m}}^{\sharp }\cap \mathcal{E}_{1}$
because of (\ref{equallity for gap equation}). 
\end{proof}%

Therefore, since the convex and weak$^{\ast }$--compact set $\mathit{\Omega }%
_{\mathfrak{m}}^{\sharp }=\mathit{M}_{\mathfrak{m}}^{\sharp }$ is a closed
face of $E_{1}$ in this case, Proposition \ref{theorem structure etat
equilibre copy(1)} leads to an exact characterization of the set $\mathit{%
\Omega }_{\mathfrak{m}}^{\sharp }$ of generalized t.i. equilibrium states
via the closed faces $\mathit{M}_{\Phi \left( d_{a,-}\right) }$ for $%
d_{a,-}\in \mathcal{C}_{\mathfrak{m}}^{\sharp }$:

\begin{corollary}[Structure of $\mathit{\Omega }_{\mathfrak{m}}^{\sharp }$
through approximating interactions]
\label{theorem structure etat equilibre}%
\index{States!generalized equilibrium!purely attractive models}For any model 
$\mathfrak{m}\in \mathcal{M}_{1}$ such that $\Phi _{a,+}=\Phi _{a,+}^{\prime
}=0$ \emph{(a.e.)}, the closed face $\mathit{\Omega }_{\mathfrak{m}}^{\sharp
}$ is the weak$^{\ast }$--closed convex hull of 
\begin{equation*}
\underset{d_{a,-}\in \mathcal{C}_{\mathfrak{m}}^{\sharp }}{\cup }\mathit{M}%
_{\Phi \left( d_{a,-}\right) }.
\end{equation*}
\end{corollary}

\chapter{The max--min and min--max Variational Problems\label{Section
theorem saddle point}}

\setcounter{equation}{0}%
The thermodynamics of any model $\mathfrak{m}\in \mathcal{M}_{1}$ is given
on the level of the pressure by Theorem \ref{BCS main theorem 1}. This
result is not satisfactory enough because we also would like to have access
to generalized t.i. equilibrium states from local theories (cf. Definition %
\ref{definition local theory}). The additional information we need for this
purpose is Theorem \ref{theorem saddle point}. In particular, it is
necessary to relate the thermodynamics of models $\mathfrak{m}\in \mathcal{M}%
_{1}$ with their approximating interactions through the thermodynamic games
defined in Definition \ref{definition two--person zero--sum game}.

As a preliminary step of the proof of Theorem \ref{theorem saddle point}, we
need to analyze more precisely the max--min and min--max variational
problems $\mathrm{F}_{\mathfrak{m}}^{\flat }$ and $\mathrm{F}_{\mathfrak{m}%
}^{\sharp }$. This is performed in Section \ref{Section optimizations
problems} and the proof of Theorem \ref{theorem saddle point} is postponed
until Section \ref{Section theorem saddle point bis}, see Lemmata \ref{lemma
super} and \ref{lemma super copy(1)}.

\section{Analysis of the conservative values $\mathrm{F}_{\mathfrak{m}%
}^{\flat }$ and $\mathrm{F}_{\mathfrak{m}}^{\sharp }$\label{Section
optimizations problems}%
\index{Thermodynamic game!conservative values}}

We start by giving important properties of the map%
\begin{equation*}
\left( c_{a,-},c_{a,+}\right) \mapsto \mathfrak{f}_{\mathfrak{m}}\left(
c_{a,-},c_{a,+}\right) :=-\left\Vert c_{a,+}\right\Vert _{2}^{2}+\left\Vert
c_{a,-}\right\Vert _{2}^{2}-P_{\mathfrak{m}}\left( c_{a,-}+c_{a,+}\right)
\end{equation*}%
from $L_{-}^{2}(\mathcal{A},\mathbb{C})\times L_{+}^{2}(\mathcal{A},\mathbb{C%
})$ to $\mathbb{R}$, see Definition \ref{definition approximating
free--energy}.

\begin{lemma}[Approximating free--energy density $\mathfrak{f}_{\mathfrak{m}%
} $\ for $\mathfrak{m}\in \mathcal{M}_{1}$]
\label{lemma chiant property map de base}\mbox{ }\newline
\emph{(}$+$\emph{)} 
\index{Free--energy density functional!approximating}At any fixed $%
c_{a,-}\in L_{-}^{2}(\mathcal{A},\mathbb{C})$, the map $c_{a,+}\mapsto 
\mathfrak{f}_{\mathfrak{m}}\left( c_{a,-},c_{a,+}\right) $ from $L_{+}^{2}(%
\mathcal{A},\mathbb{C})$ to $\mathbb{R}$ is upper semi--continuous in the
weak topology and strictly concave ($\gamma _{a,+}\neq 0$ (a.e.)).\newline
\emph{(}$-$\emph{)} At any fixed $c_{a,+}\in L_{+}^{2}(\mathcal{A},\mathbb{C}%
)$, the map $c_{a,-}\mapsto \mathfrak{f}_{\mathfrak{m}}\left(
c_{a,-},c_{a,+}\right) $ from $L_{-}^{2}(\mathcal{A},\mathbb{C})$ to $%
\mathbb{R}$ is lower semi--continuous in the weak topology.
\end{lemma}

\begin{proof}
The maps $c_{a,\pm }\mapsto \left\Vert c_{a,\pm }\right\Vert _{2}^{2}$ from $%
L_{\pm }^{2}(\mathcal{A},\mathbb{C})$ to $\mathbb{R}$ are lower
semi--continuous in the weak topology and, as soon as $\gamma _{a,\pm }\neq
0 $ (a.e.), strictly convex. By Proposition \ref{corrolaire sympa} (ii), the
map $c_{a}\mapsto P_{\mathfrak{m}}\left( c_{a}\right) $ is weakly continuous
on any ball $\mathcal{B}_{R}\left( 0\right) \subseteq L^{2}(\mathcal{A},%
\mathbb{C})$ of radius $R<\infty $ and convex. Therefore, the map $%
c_{a,+}\mapsto \mathfrak{f}_{\mathfrak{m}}\left( c_{a,-},c_{a,+}\right) $ is
upper semi--continuous and strictly concave if $\gamma _{a,\pm }\neq 0$
(a.e.), whereas $c_{a,-}\mapsto \mathfrak{f}_{\mathfrak{m}}\left(
c_{a,-},c_{a,+}\right) $ is lower semi--continuous.%
\end{proof}%

We continue our analysis of the conservative values $\mathrm{F}_{\mathfrak{m}%
}^{\flat }$ and $\mathrm{F}_{\mathfrak{m}}^{\sharp }$ by studying the
functionals $\mathfrak{f}_{\mathfrak{m}}^{\flat }$ and $\mathfrak{f}_{%
\mathfrak{m}}^{\sharp }$ of the thermodynamic game defined in Definition \ref%
{definition two--person zero--sum game}.

\begin{lemma}[Properties of functionals $\mathfrak{f}_{\mathfrak{m}}^{\flat
} $ and $\mathfrak{f}_{\mathfrak{m}}^{\sharp }$\ for $\mathfrak{m}\in 
\mathcal{M}_{1}$]
\label{lemma idiot interaction approx 2 copy(1)}\mbox{ }\newline
\emph{(}$\flat $\emph{)} 
\index{Thermodynamic game!worst loss functional}%
\index{Thermodynamic game!least gain functional}The map $c_{a,+}\mapsto 
\mathfrak{f}_{\mathfrak{m}}^{\flat }\left( c_{a,+}\right) $ from $L_{+}^{2}(%
\mathcal{A},\mathbb{C})$ to $\mathbb{R}$ is upper semi--continuous in the
weak topology and strictly concave ($\gamma _{a,+}\neq 0$ (a.e.)).\newline
\emph{(}$\sharp $\emph{)} The map $c_{a,-}\mapsto \mathfrak{f}_{\mathfrak{m}%
}^{\sharp }\left( c_{a,-}\right) $ from $L_{-}^{2}(\mathcal{A},\mathbb{C})$
to $\mathbb{R}$ is lower semi--continuous in the weak topology.
\end{lemma}

\begin{proof}
By Proposition \ref{corrolaire sympa} (ii), we first observe that there is $%
R>0$ such that 
\begin{equation*}
\mathfrak{f}_{\mathfrak{m}}^{\flat }\left( c_{a,+}\right) =\underset{%
c_{a,-}\in \mathcal{B}_{R,-}}{\inf }\mathfrak{f}_{\mathfrak{m}}\left(
c_{a,-},c_{a,+}\right) \mathrm{\quad and\quad }\mathfrak{f}_{\mathfrak{m}%
}^{\sharp }\left( c_{a,-}\right) =\underset{c_{a,+}\in \mathcal{B}_{R,+}}{%
\sup }\mathfrak{f}_{\mathfrak{m}}\left( c_{a,-},c_{a,+}\right) ,
\end{equation*}%
where $\mathcal{B}_{R,\pm }\subseteq L_{\pm }^{2}(\mathcal{A},\mathbb{C})$
are the closed balls of radius $R$ centered at $0$. In other words, $%
\mathfrak{f}_{\mathfrak{m}}^{\flat }(c_{a,+})$ and $\mathfrak{f}_{\mathfrak{m%
}}^{\sharp }(c_{a,-})$ are well--defined for any $c_{a,\pm }\in L_{\pm }^{2}(%
\mathcal{A},\mathbb{C})$.

($\flat $) From Proposition \ref{Bogo interaction negative}, there exists
also $R<\infty $ such that 
\begin{equation}
\mathrm{P}_{\mathfrak{m}\left( c_{a,+}\right) }=\underset{c_{a,-}\in 
\mathcal{B}_{R,-}}{\sup }\left\{ -\left\Vert c_{a,-}\right\Vert _{2}^{2}+P_{%
\mathfrak{m}}\left( c_{a,-}+c_{a,+}\right) \right\}  \label{pressure p 2}
\end{equation}%
is the pressure of the Fermi system%
\begin{equation}
\mathfrak{m}\left( c_{a,+}\right) :=(\Phi \left( c_{a,+}\right) ,\{\Phi
_{a,-}\}_{a\in \mathcal{A}},\{\Phi _{a,-}^{\prime }\}_{a\in \mathcal{A}})\in 
\mathcal{M}_{1}.  \label{m approche 2}
\end{equation}%
Here, $\Phi _{a,-}:=\gamma _{a,-}\Phi _{a}$ and $\Phi _{a,-}^{\prime
}:=\gamma _{a,-}\Phi _{a}^{\prime }$, whereas $\Phi \left( c_{a,+}\right)
=\Phi _{\mathfrak{m}}\left( c_{a,+}\right) $ is defined in Definition \ref%
{definition BCS-type model approximated}.

By using similar arguments as in the proof of Proposition \ref{corrolaire
sympa} (ii), one obtains that the family 
\begin{equation}
\{c_{a,+}\mapsto f_{\mathfrak{m}}\left( \rho ,c_{a,+}+c_{a,-}\right)
\}_{\rho \in E_{1},%
\text{$c_{a,-}\in $}\mathcal{B}_{R,-}}
\end{equation}%
of real linear functionals is weakly equicontinuous on the ball $\mathcal{B}%
_{R,+}$. It follows from Proposition \ref{corrolaire sympa} (i) and (\ref%
{pressure p 2}) that\ the map $c_{a,+}\mapsto \mathrm{P}_{\mathfrak{m}\left(
c_{a,+}\right) }$ is weakly continuous on the ball $\mathcal{B}_{R,+}$.
Additionally, $c_{a,+}\mapsto \left\Vert c_{a,+}\right\Vert _{2}^{2}$ is
lower semi--continuous in the weak topology. Therefore, the map%
\begin{equation}
c_{a,+}\mapsto \mathfrak{f}_{\mathfrak{m}}^{\flat }\left( c_{a,+}\right)
=-\left\Vert c_{a,+}\right\Vert _{2}^{2}-\mathrm{P}_{\mathfrak{m}\left(
c_{a,+}\right) }  \label{petite equality cool}
\end{equation}%
is upper semi--continuous in the weak topology. As soon as $\gamma
_{a,+}\neq 0$ (a.e.), the functional $\mathfrak{f}_{\mathfrak{m}}^{\flat }$
is also strictly concave: For all $\lambda \in (0,1)$ and $%
c_{a,+}^{(1)},c_{a,+}^{(2)}\in L_{+}^{2}(\mathcal{A},\mathbb{C})$ such that $%
c_{a,+}^{(1)}\neq c_{a,+}^{(2)}$ (a.e.), 
\begin{equation*}
\lambda \mathfrak{f}_{\mathfrak{m}}^{\flat }(c_{a,+}^{(1)})+(1-\lambda )%
\mathfrak{f}_{\mathfrak{m}}^{\flat }(c_{a,+}^{(2)})<\mathfrak{f}_{\mathfrak{m%
}}^{\flat }(\lambda c_{a,+}^{(1)}+(1-\lambda )c_{a,+}^{(2)}).
\end{equation*}

($\sharp $) The functional $\mathfrak{f}_{\mathfrak{m}}^{\sharp }$ is lower
semi--continuous w.r.t. the weak topology because it is the supremum of a
family 
\begin{equation*}
\{c_{a,-}\mapsto \mathfrak{f}_{\mathfrak{m}}\left( c_{a,-},c_{a,+}\right)
\}_{c_{a,+}\text{$\in L_{+}^{2}(\mathcal{A},\mathbb{C})$}}
\end{equation*}%
of lower semi--continuous functionals, see Lemma \ref{lemma chiant property
map de base} ($-$).%
\end{proof}%

For all $c_{a,\pm }\in L_{\pm }^{2}(\mathcal{A},\mathbb{C})$, we study now
the sets $\mathcal{C}_{\mathfrak{m}}^{\flat }\left( c_{a,+}\right) $ and $%
\mathcal{C}_{\mathfrak{m}}^{\sharp }\left( c_{a,-}\right) $ related to the
solutions of the variational problems $\mathfrak{f}_{\mathfrak{m}}^{\flat }$
and $\mathfrak{f}_{\mathfrak{m}}^{\sharp }$ and defined by (\ref{eq conserve
strategybis}).

\begin{lemma}[Solutions of variational problems $\mathfrak{f}_{\mathfrak{m}%
}^{\flat }$ and $\mathfrak{f}_{\mathfrak{m}}^{\sharp }$]
\label{lemma idiot interaction approx 2 copy(2)}\mbox{ }\newline
\emph{(}$\flat $\emph{)} For all $c_{a,+}\in L_{+}^{2}(\mathcal{A},\mathbb{C}%
)$, the set $\mathcal{C}_{\mathfrak{m}}^{\flat }\left( c_{a,+}\right) $ is
non--empty, norm--bounded and weakly compact.\newline
\emph{(}$\sharp $\emph{)} If $\gamma _{a,+}\neq 0$ (a.e.) then, for all $%
c_{a,-}\in L_{-}^{2}(\mathcal{A},\mathbb{C})$, the set $\mathcal{C}_{%
\mathfrak{m}}^{\sharp }\left( c_{a,-}\right) $ has exactly one element $%
\mathrm{r}_{+}(c_{a,-})$.%
\index{Thermodynamic game!decision rule}
\end{lemma}

\begin{proof}
Fix $c_{a,\pm }\in L_{\pm }^{2}(\mathcal{A},\mathbb{C})$. From Proposition %
\ref{corrolaire sympa} (ii), there is $R<\infty $ such that $\mathcal{C}_{%
\mathfrak{m}}^{\flat }\left( c_{a,+}\right) \subseteq \mathcal{B}_{R,-}$ and 
$\mathcal{C}_{\mathfrak{m}}^{\sharp }\left( c_{a,-}\right) \subseteq 
\mathcal{B}_{R,+}$ with $\mathcal{B}_{R,\pm }\subseteq L_{\pm }^{2}(\mathcal{%
A},\mathbb{C})$ being the closed balls of radius $R$ centered at $0$.

($\flat $) We first observe that, by the separability assumption on the
measure space $(\mathcal{A},\mathfrak{a})$, the weak topology of any weakly
compact set is metrizable, by Theorem \ref{Metrizability}. Therefore,\
since, by Banach--Alaoglu theorem, balls $\mathcal{B}_{R,-}$ are weakly
compact, they are metrizable and we can restrict ourself on sequences
instead of more general nets. Take now any sequence $\{c_{a,-}^{(n)}%
\}_{n=1}^{\infty }$ of approximating minimizers in $\mathcal{B}_{R,-}$ such
that%
\begin{equation*}
\mathfrak{f}_{\mathfrak{m}}^{\flat }(c_{a,+})=\underset{n\rightarrow \infty }%
{\lim }\mathfrak{f}_{\mathfrak{m}}(c_{a,-}^{(n)},c_{a,+}).
\end{equation*}%
By compactness and metrizability of balls $\mathcal{B}_{R,-}$ in the weak
topology, we can assume without loss of generality that $\{c_{a,-}^{(n)}%
\}_{n=1}^{\infty }$ converges weakly towards $d_{a,-}\in \mathcal{B}_{R,-}$.

The map $c_{a,-}\mapsto \mathfrak{f}_{\mathfrak{m}}\left(
c_{a,-},c_{a,+}\right) $ is lower semi--continuous in the weak topology, see
Lemma \ref{lemma chiant property map de base} ($-$). It follows that 
\begin{equation*}
\mathfrak{f}_{\mathfrak{m}}^{\flat }\left( c_{a,+}\right) =\mathfrak{f}_{%
\mathfrak{m}}\left( d_{a,-},c_{a,+}\right) .
\end{equation*}%
In other words, for all $c_{a,+}\in L_{+}^{2}(\mathcal{A},\mathbb{C})$, the
set $\mathcal{C}_{\mathfrak{m}}^{\flat }\left( c_{a,+}\right) \subseteq 
\mathcal{B}_{R,-}$ is non--empty and norm--bounded. Again by weakly lower
semi--continuity of the map $c_{a,-}\mapsto \mathfrak{f}_{\mathfrak{m}%
}\left( c_{a,-},c_{a,+}\right) $, for any sequence $\{d_{a,-}^{(n)}%
\}_{n=1}^{\infty }$ in $\mathcal{C}_{\mathfrak{m}}^{\flat }\left(
c_{a,+}\right) $ converging weakly towards $d_{a}^{(\infty )}\in L_{-}^{2}(%
\mathcal{A},\mathbb{C})$ as $n\rightarrow \infty $, it is also clear that $%
d_{a}^{(\infty )}\in \mathcal{C}_{\mathfrak{m}}^{\flat }\left(
c_{a,+}\right) $ is weakly compact. Thus $\mathcal{C}_{\mathfrak{m}}^{\flat
}\left( c_{a,+}\right) $ is weakly compact because it is a weakly closed
subset of a weakly compact set.

($\sharp $) Similarly as in ($\flat $), the set $\mathcal{C}_{\mathfrak{m}%
}^{\sharp }\left( c_{a,-}\right) \subseteq \mathcal{B}_{R,+}$ is non--empty
because the map $c_{a,+}\mapsto \mathfrak{f}_{\mathfrak{m}}\left(
c_{a,-},c_{a,+}\right) $ is upper semi--continuous in the weak topology, by
Lemma \ref{lemma chiant property map de base} ($+$). The uniqueness of $%
\mathrm{r}_{+}(c_{a,-})$ in the $L_{+}^{2}(\mathcal{A},\mathbb{C})$--sense
for any fixed $c_{a,-}\in L_{-}^{2}(\mathcal{A},\mathbb{C})$ follows from
the strict concavity of the functional $c_{a,+}\mapsto \mathfrak{f}_{%
\mathfrak{m}}\left( c_{a,-},c_{a,+}\right) $, see again Lemma \ref{lemma
chiant property map de base} ($+$). 
\end{proof}%

Then we conclude the analysis of the two optimization problems $\mathrm{F}_{%
\mathfrak{m}}^{\flat }$ and $\mathrm{F}_{\mathfrak{m}}^{\sharp }$ of the
thermodynamic game defined in Definition \ref{definition two--person
zero--sum game} with a study of their sets $\mathcal{C}_{\mathfrak{m}%
}^{\flat }$ and $\mathcal{C}_{\mathfrak{m}}^{\sharp }$ of conservative
strategies, see (\ref{eq conserve strategy}).

\begin{lemma}[The set of optimizers for $\mathfrak{m}\in \mathcal{M}_{1}$]
\label{lemma idiot interaction approx 2}\mbox{ }\newline
\emph{(}$\flat $\emph{)} 
\index{Thermodynamic game!conservative strategies}If $\gamma _{a,+}\neq 0$
(a.e.), the set $\mathcal{C}_{\mathfrak{m}}^{\flat }\subseteq L_{+}^{2}(%
\mathcal{A},\mathbb{C})$ has exactly one element $d_{a,+}$.\newline
\emph{(}$\sharp $\emph{)} The set $\mathcal{C}_{\mathfrak{m}}^{\sharp
}\subseteq L_{-}^{2}(\mathcal{A},\mathbb{C})$ is non--empty, norm--bounded,
and weakly compact.
\end{lemma}

\begin{proof}
From Proposition \ref{corrolaire sympa} (ii), there is $R<\infty $ such that 
$\mathcal{C}_{\mathfrak{m}}^{\flat }\subseteq \mathcal{B}_{R,+}$ and $%
\mathcal{C}_{\mathfrak{m}}^{\sharp }\subseteq \mathcal{B}_{R,-}$ with $%
\mathcal{B}_{R,\pm }\subseteq L_{\pm }^{2}(\mathcal{A},\mathbb{C})$ being
the closed balls of radius $R$ centered at $0$. In particular, $-\infty <%
\mathrm{F}_{\mathfrak{m}}^{\flat }\leq \mathrm{F}_{\mathfrak{m}}^{\sharp
}<\infty $.

($\flat $) From\ Lemma \ref{lemma idiot interaction approx 2 copy(1)} ($%
\flat $), $\mathrm{F}_{\mathfrak{m}}^{\flat }$ is a supremum of a weakly
upper semi--continuous functional $\mathfrak{f}_{\mathfrak{m}}^{\flat }$ and 
$\mathcal{C}_{\mathfrak{m}}^{\flat }$ is the set of its maximizers.
Therefore, in the same way we prove ($\flat $) in Lemma \ref{lemma idiot
interaction approx 2 copy(2)}, $\mathcal{C}_{\mathfrak{m}}^{\flat }\subseteq 
\mathcal{B}_{R,+}$ is non--empty and weakly compact. Moreover, Lemma \ref%
{lemma idiot interaction approx 2 copy(1)} ($\flat $) also tells us that $%
\mathfrak{f}_{\mathfrak{m}}^{\flat }$ is strictly concave as soon as $\gamma
_{a,+}\neq 0$ (a.e.). Therefore, there is actually a unique solution $%
d_{a,+}\in L_{+}^{2}(\mathcal{A},\mathbb{C})$ of the variational problem 
\begin{equation*}
\mathrm{F}_{\mathfrak{m}}^{\flat }:=\underset{c_{a,+}\in L_{+}^{2}(\mathcal{A%
},\mathbb{C})}{\sup }\mathfrak{f}_{\mathfrak{m}}^{\flat }\left(
c_{a,+}\right) .
\end{equation*}

($\sharp $) To prove the second statement, we use similar arguments as in ($%
\flat $). Indeed, one uses Lemma \ref{lemma idiot interaction approx 2
copy(1)} ($\sharp $). Observe, however, that $\mathfrak{f}_{\mathfrak{m}%
}^{\sharp }$ is not strictly convex and so, the solution of the variational
problem%
\begin{equation*}
\mathrm{F}_{\mathfrak{m}}^{\sharp }:=\underset{c_{a,-}\in L_{-}^{2}(\mathcal{%
A},\mathbb{C})}{\inf }\mathfrak{f}_{\mathfrak{m}}^{\sharp }\left(
c_{a,-}\right)
\end{equation*}%
may not be unique. 
\end{proof}%

\section{$\mathrm{F}_{\mathfrak{m}}^{\flat }$ and $\mathrm{F}_{\mathfrak{m}%
}^{\sharp }$ as variational problems over states\label{Section theorem
saddle point bis}%
\index{Thermodynamic game!conservative values}}

Theorem \ref{theorem saddle point} ($\flat $), i.e., 
\begin{equation*}
\mathrm{P}_{\mathfrak{m}}^{\flat }:=-\inf\limits_{\rho \in E_{1}}f_{%
\mathfrak{m}}^{\flat }(\rho )=-\mathrm{F}_{\mathfrak{m}}^{\flat },
\end{equation*}%
follows from Lemma \ref{eq idiot sympa} together with von Neumann min--max
theorem (Theorem \ref{theorem minmax von Neumann}) which also give us
additional information about the non--empty set%
\index{Minimizers} 
\begin{equation}
\mathit{M}_{\mathfrak{m}}^{\flat }:=\left\{ \varrho \in E_{1}:\quad f_{%
\mathfrak{m}}^{\flat }\left( \varrho \right) =\inf\limits_{\rho \in
E_{1}}\,f_{\mathfrak{m}}^{\flat }(\rho )\right\}
\label{set of minimizers de f bemol}
\end{equation}%
of t.i. minimizers of the weak$^{\ast }$--lower semi--continuous convex
functional $f_{\mathfrak{m}}^{\flat }$ (\ref{convex functional g_m}). This
is proven in the next lemma.

\begin{lemma}[$\mathrm{F}_{\mathfrak{m}}^{\flat }$ and gap equations]
\label{lemma super}\mbox{ }\newline
\index{Gap equations}%
\index{Thermodynamic game!conservative strategies}For any $\mathfrak{m}\in 
\mathcal{M}_{1}$, $\mathrm{P}_{\mathfrak{m}}^{\flat }=-\mathrm{F}_{\mathfrak{%
m}}^{\flat }$ and there is $\omega \in \mathit{\Omega }_{\mathfrak{m}%
(d_{a,+})}^{\sharp }\cap \mathit{M}_{\mathfrak{m}}^{\flat }$ satisfying 
\begin{equation}
d_{a,+}=\gamma _{a,+}(e_{\Phi _{a}}(\omega )+ie_{\Phi _{a}^{\prime }}(\omega
))\mathrm{\ (a.e.)}  \label{gap equation1}
\end{equation}%
with $d_{a,+}\in \mathcal{C}_{\mathfrak{m}}^{\flat }$ and $\mathit{\Omega }_{%
\mathfrak{m}(d_{a,+})}^{\sharp }=\mathit{M}_{\mathfrak{m}(d_{a,+})}^{\sharp
} $ being the set of generalized t.i. equilibrium states of the model $%
\mathfrak{m}(d_{a,+})\in \mathcal{M}_{1}$ with purely attractive long--range
interactions defined by (\ref{m approche 2}). Compare with Proposition \ref%
{theorem structure etat equilibre copy(1)} and Corollary \ref{theorem
structure etat equilibre}.
\end{lemma}

\begin{proof}
On the one hand, by using Lemma \ref{eq idiot sympa}, observe that%
\begin{eqnarray}
\inf\limits_{\rho \in E_{1}}f_{\mathfrak{m}}^{\flat }\left( \rho \right)
&=&\inf\limits_{\rho \in E_{1}}\left\{ \underset{c_{a,+}\in \mathcal{B}_{R,+}%
}{\sup }\left\{ -\left\Vert c_{a,+}\right\Vert _{2}^{2}+2%
\func{Re}\left\{ \left\langle e_{\Phi _{a}}\left( \rho \right) +ie_{\Phi
_{a}^{\prime }}\left( \rho \right) ,c_{a,+}\right\rangle \right\} \right.
\right.  \notag \\
&&\left. \underset{\underset{}{}}{}\left. -\left\Vert \Delta _{a,-}\left(
\rho \right) \right\Vert _{1}+e_{\Phi }(\rho )-\beta ^{-1}s(\rho )\right\}
\right\}  \label{gap eq 1}
\end{eqnarray}%
with $\mathcal{B}_{R,+}\subseteq L_{+}^{2}(\mathcal{A},\mathbb{C})$ being a
closed ball of sufficiently large radius $R>0$ centered at $0$.

On the other hand, the set $\mathcal{C}_{\mathfrak{m}}^{\flat }\subseteq 
\mathcal{B}_{R,+}$ of conservative strategies of $\mathrm{F}_{\mathfrak{m}%
}^{\flat }$ defined by (\ref{eq conserve strategy}) has a unique element
(Lemma \ref{lemma idiot interaction approx 2} ($\flat $)) and, by using
Proposition \ref{Bogo interaction negative} and (\ref{petite equality cool}%
), 
\begin{eqnarray}
\mathrm{F}_{\mathfrak{m}}^{\flat } &=&\underset{c_{a,+}\in \mathcal{B}_{R,+}}%
{\sup }\left\{ \inf\limits_{\rho \in E_{1}}\left\{ -\left\Vert
c_{a,+}\right\Vert _{2}^{2}+2\func{Re}\left\{ \left\langle e_{\Phi
_{a}}\left( \rho \right) +ie_{\Phi _{a}^{\prime }}\left( \rho \right)
,c_{a,+}\right\rangle \right\} \right. \right.  \notag \\
&&\left. \underset{}{}\left. -\left\Vert \Delta _{a,-}\left( \rho \right)
\right\Vert _{1}+e_{\Phi }(\rho )-\beta ^{-1}s(\rho )\right\} \right\}
\label{gap eq 2}
\end{eqnarray}%
provided the radius $R>0$ is taken sufficiently large.

Now, the real functional%
\begin{eqnarray*}
(\rho ,c_{a,+}) &\mapsto &-\Vert c_{a,+}\Vert _{2}^{2}+2\func{Re}\{\langle
e_{\Phi _{a}}(\rho )+ie_{\Phi _{a}^{\prime }}(\rho ),c_{a,+}\rangle \} \\
&&-\Vert \Delta _{a,-}(\rho )\Vert _{1}+e_{\Phi }(\rho )-\beta ^{-1}s(\rho )
\end{eqnarray*}%
is convex and weak$^{\ast }$--lower semi--continuous w.r.t. $\rho \in E_{1}$%
, but concave and weakly upper semi--continuous w.r.t. $c_{a,+}\in L_{+}^{2}(%
\mathcal{A},\mathbb{C})$. Additionally, the sets $E_{1}$ and $\mathcal{B}%
_{R,+}$ are clearly convex and compact, in the weak$^{\ast }$ and weak
topologies respectively. Therefore, from von Neumann min--max theorem%
\index{von Neumann min--max theorem} (Theorem \ref{theorem minmax von
Neumann}), there is a saddle point $\left( \omega ,d_{a,+}\right) \in
E_{1}\times L_{+}^{2}(\mathcal{A},\mathbb{C})$ which yields $\mathrm{P}_{%
\mathfrak{m}}^{\flat }=-\mathrm{F}_{\mathfrak{m}}^{\flat }$, see Definition %
\ref{definition saddle points}. In particular, by Lemma \ref{eq idiot sympa}%
, there are $\omega \in \mathit{\Omega }_{\mathfrak{m}(d_{a,+})}^{\sharp
}\cap \mathit{M}_{\mathfrak{m}}^{\flat }$ and $d_{a,+}\in \mathcal{C}_{%
\mathfrak{m}}^{\flat }$ satisfying the Euler--Lagrange equations (\ref{gap
equation1}), which are also called gap equations in Physics (Remark \ref%
{remark gap eq}). 
\end{proof}%

Note that (\ref{gap eq 2}) can be interpreted as a two--person zero--sum
game with a non--cooperative equilibrium defined by the saddle point $\left(
\omega ,d_{a,+}\right) $. Observe also that Lemma \ref{lemma super} combined
with\ Theorem \ref{theorem purement repulsif sympa} ($+$) directly yields
Theorem \ref{theorem saddle point} ($\sharp $) for purely repulsive
long--range interactions:

\begin{corollary}[Thermodynamics game and pressure -- I]
\label{corrolaire purement repulsif sympa bis}\mbox{ }\newline
\index{Thermodynamic game!pressure}%
\index{Pressure!variational problems!purely repulsive models}For any $%
\mathfrak{m}\in \mathcal{M}_{1}$ and under the condition that $\Phi
_{a,-}=\Phi _{a,-}^{\prime }=0$ (a.e.), 
\begin{equation*}
\mathrm{P}_{\mathfrak{m}}:=\mathrm{P}_{\mathfrak{m}}^{\sharp }=\mathrm{P}_{%
\mathfrak{m}}^{\flat }=-\mathrm{F}_{\mathfrak{m}}
\end{equation*}%
with\textit{\ }$\mathrm{F}_{\mathfrak{m}}:=\mathrm{F}_{\mathfrak{m}}^{\sharp
}=\mathrm{F}_{\mathfrak{m}}^{\flat }$, see Definition \ref{definition
two--person zero--sum game}.
\end{corollary}

We are now in position to prove Theorem \ref{theorem saddle point} ($\sharp $%
) in the general case.

\begin{lemma}[Thermodynamics game and pressure -- II]
\label{lemma super copy(1)}\mbox{ }\newline
\index{Thermodynamic game!pressure}%
\index{Pressure!variational problems}For any $\mathfrak{m}\in \mathcal{M}%
_{1} $, $\mathrm{P}_{\mathfrak{m}}^{\sharp }=-\mathrm{F}_{\mathfrak{m}%
}^{\sharp }$ with the pressure $\mathrm{P}_{\mathfrak{m}}^{\sharp }$ given
for $\mathfrak{m}\in \mathcal{M}_{1}$ by the minimization of the
free--energy density functional $f_{\mathfrak{m}}^{\sharp }$ over $E_{1}$,
see Definition \ref{Pressure} and Theorem \ref{BCS main theorem 1} (i).
\end{lemma}

\begin{proof}
From Theorem \ref{BCS main theorem 1} (i) combined with Lemmata \ref{lemma
property free--energy density functional copy(1)} and \ref{eq idiot sympa}, 
\begin{eqnarray}
-\mathrm{P}_{\mathfrak{m}}^{\sharp } &=&\inf\limits_{%
\hat{\rho}\in \mathcal{E}_{1}}\left\{ \underset{c_{a,-}\in L_{-}^{2}(%
\mathcal{A},\mathbb{C})}{\inf }\left\{ \left\Vert c_{a,-}\right\Vert
_{2}^{2}+f_{\mathfrak{m}(c_{a,-})}^{\sharp }\left( \hat{\rho}\right)
\right\} \right\}  \notag \\
&=&\underset{c_{a,-}\in L_{-}^{2}(\mathcal{A},\mathbb{C})}{\inf }\left\{
\inf\limits_{\hat{\rho}\in \mathcal{E}_{1}}\left\{ \left\Vert
c_{a,-}\right\Vert _{2}^{2}+f_{\mathfrak{m}(c_{a,-})}^{\sharp }\left( \hat{%
\rho}\right) \right\} \right\}  \label{equality sympa general}
\end{eqnarray}%
with%
\begin{equation*}
f_{\mathfrak{m}(c_{a,-})}^{\sharp }(\rho ):=-2\func{Re}\{\langle e_{\Phi
_{a}}(\rho )+ie_{\Phi _{a}^{\prime }}(\rho ),c_{a,-}\rangle \}+\Vert \Delta
_{a,+}(\rho )\Vert _{1}+e_{\Phi }(\rho )-\beta ^{-1}s(\rho )
\end{equation*}%
for all $\rho \in E_{1}$. By using again Lemma \ref{lemma property
free--energy density functional copy(1)} and Theorem \ref{BCS main theorem 1}
(i), for all $c_{a,-}\in L_{-}^{2}(\mathcal{A},\mathbb{C})$,%
\begin{equation*}
\mathrm{P}_{\mathfrak{m}(c_{a,-})}^{\sharp }=-\inf\limits_{\hat{\rho}\in 
\mathcal{E}_{1}}f_{\mathfrak{m}(c_{a,-})}^{\sharp }\left( \hat{\rho}\right)
=-\inf\limits_{\rho \in E_{1}}f_{\mathfrak{m}(c_{a,-})}^{\sharp }\left( \rho
\right)
\end{equation*}%
is the pressure associated with the purely repulsive long--range model 
\begin{equation}
\mathfrak{m}\left( c_{a,-}\right) :=(\Phi \left( c_{a,-}\right) ,\{\Phi
_{a,+}\}_{a\in \mathcal{A}},\{\Phi _{a,+}^{\prime }\}_{a\in \mathcal{A}})\in 
\mathcal{M}_{1},  \label{model purement repulsifbis}
\end{equation}%
where $\Phi _{a,+}:=\gamma _{a,+}\Phi _{a}$ and $\Phi _{a,+}^{\prime
}:=\gamma _{a,+}\Phi _{a}^{\prime }$, see (\ref{model purement repulsif}).
In particular, 
\begin{equation}
-\mathrm{P}_{\mathfrak{m}}^{\sharp }=\underset{c_{a,-}\in L_{-}^{2}(\mathcal{%
A},\mathbb{C})}{\inf }\left\{ \left\Vert c_{a,-}\right\Vert _{2}^{2}-\mathrm{%
P}_{\mathfrak{m}(c_{a,-})}^{\sharp }\right\} .
\label{equality sympa general2}
\end{equation}%
Therefore, applying Corollary \ref{corrolaire purement repulsif sympa bis}
on the model $\mathfrak{m}\left( c_{a,-}\right) $ with purely repulsive
long--range interactions, one gets from (\ref{equality sympa general2}) that%
\begin{equation*}
-\mathrm{P}_{\mathfrak{m}}^{\sharp }=\underset{c_{a,-}\in L_{-}^{2}(\mathcal{%
A},\mathbb{C})}{\inf }\left\{ \underset{c_{a,+}\in L_{+}^{2}(\mathcal{A},%
\mathbb{C})}{\sup }\mathfrak{f}_{\mathfrak{m}}\left( c_{a,-},c_{a,+}\right)
\right\} =\mathrm{F}_{\mathfrak{m}}^{\sharp }
\end{equation*}%
for any $\mathfrak{m}\in \mathcal{M}_{1}$. 
\end{proof}%

\noindent Observe that treating first the positive part of the model $%
\mathfrak{m}\in \mathcal{M}_{1}$ in $\mathrm{P}_{\mathfrak{m}}^{\sharp }$ by
using Lemma \ref{eq idiot sympa} is not necessarily useful in the general
case unless $\mathrm{F}_{\mathfrak{m}}^{\sharp }=\mathrm{F}_{\mathfrak{m}%
}^{\flat }$. Indeed, we approximate first the long--range attractions $\Phi
_{a,-}$ and $\Phi _{a,-}^{\prime }$ because we can then commute in (\ref%
{equality sympa general}) two infima. If we would have first approximated
the long--range repulsions $\Phi _{a,+}$ and $\Phi _{a,+}^{\prime }$, by
using Lemma \ref{eq idiot sympa}, we would have to commute a $\sup $ and a $%
\inf $, which is generally not possible because we would have obtained $%
\mathrm{P}_{\mathfrak{m}}^{\flat }$ and not $\mathrm{P}_{\mathfrak{m}%
}^{\sharp }\geq \mathrm{P}_{\mathfrak{m}}^{\flat }$, see Lemma \ref{lemma
super}.

Finally, we conclude by giving an interesting lemma about the continuity of
the thermodynamic decision rule%
\index{Thermodynamic game!decision rule} 
\begin{equation*}
\mathrm{r}_{+}:c_{a,-}\mapsto \mathrm{r}_{+}\left( c_{a,-}\right) \in 
\mathcal{C}_{\mathfrak{m}}^{\sharp }\left( c_{a,-}\right)
\end{equation*}%
(cf. (\ref{thermodyn decision rule})) with $\mathrm{r}_{+}\left(
c_{a,-}\right) $ being the unique element of the set $\mathcal{C}_{\mathfrak{%
m}}^{\sharp }\left( c_{a,-}\right) $\ defined by (\ref{eq conserve
strategybis}) for all $c_{a,-}\in L_{-}^{2}(\mathcal{A},\mathbb{C})$, cf.
Lemma \ref{lemma idiot interaction approx 2 copy(2)} ($\sharp $). This lemma
follows from Lemma \ref{lemma super copy(1)}.

\begin{lemma}[Weak--norm continuity of the map $\mathrm{r}_{+}$]
\label{lemma idiot interaction approx 2 copy(3)}\mbox{ }\newline
If $\gamma _{a,+}\neq 0$ (a.e.) then the map 
\begin{equation*}
\mathrm{r}_{+}:c_{a,-}\mapsto \mathrm{r}_{+}\left( c_{a,-}\right) \in 
\mathcal{C}_{\mathfrak{m}}^{\sharp }\left( c_{a,-}\right)
\end{equation*}%
from $L_{-}^{2}(\mathcal{A},\mathbb{C})$ to $L_{+}^{2}(\mathcal{A},\mathbb{C}%
)$ is continuous w.r.t. the weak topology in $L_{-}^{2}(\mathcal{A},\mathbb{C%
})$ and the norm topology in $L_{+}^{2}(\mathcal{A},\mathbb{C})$.
\end{lemma}

\begin{proof}
First, recall that $\mathfrak{m}(c_{a,-})\in \mathcal{M}_{1}$ is the model
with purely repulsive long--range interactions defined by (\ref{model
purement repulsifbis}) for any $c_{a,-}\in L_{-}^{2}(\mathcal{A},\mathbb{C})$%
. From Lemma \ref{lemma super copy(1)}, its pressure equals%
\begin{equation}
\mathrm{P}_{\mathfrak{m}(c_{a,-})}=\underset{c_{a,+}\in L_{+}^{2}(\mathcal{A}%
,\mathbb{C})}{\inf }\left\{ \left\Vert c_{a,+}\right\Vert _{2}^{2}+P_{%
\mathfrak{m}}\left( c_{a,-}+c_{a,+}\right) \right\} =\left\Vert
c_{a,-}\right\Vert _{2}^{2}-\mathfrak{f}_{\mathfrak{m}}^{\sharp }\left(
c_{a,-}\right) .  \label{inequality debile}
\end{equation}

Take any sequence $\{c_{a,-}^{(n)}\}_{n=1}^{\infty }$ converging to $%
c_{a,-}\in L_{-}^{2}(\mathcal{A},\mathbb{C})$ in the weak topology. From the
uniform boundedness principle (Banach--Steinhaus theorem), it follows that
any weakly convergent sequence in $L_{-}^{2}(\mathcal{A},\mathbb{C})$ is
norm--bounded. In particular, the sequence $\{c_{a,-}^{(n)}\}_{n=1}^{\infty
} $ belongs to a ball $\mathcal{B}_{R,-}\subseteq L_{-}^{2}(\mathcal{A},%
\mathbb{C})$ of sufficiently large radius $R$ centered at $0$. By
Proposition \ref{corrolaire sympa} (ii), the family%
\begin{equation*}
\left\{ c_{a,-}\mapsto P_{\mathfrak{m}}\left( c_{a,-}+c_{a,+}\right)
\right\} _{c_{a,+}\in L_{+}^{2}(\mathcal{A},\mathbb{C})}
\end{equation*}%
of functionals is weakly equicontinuous on the ball $\mathcal{B}%
_{R,-}\subseteq L_{-}^{2}(\mathcal{A},\mathbb{C})$. It follows that 
\begin{equation}
\underset{n\rightarrow \infty }{\lim }\mathrm{P}_{\mathfrak{m}%
(c_{a,-}^{(n)})}=\mathrm{P}_{\mathfrak{m}(c_{a,-})}.
\label{inequality debile2}
\end{equation}

For all $n\in \mathbb{N}$, the unique $\mathrm{r}_{+}(c_{a,-}^{(n)})\in 
\mathcal{C}_{\mathfrak{m}}^{\sharp }(c_{a,-}^{(n)})$ satisfies 
\begin{equation}
\mathrm{P}_{\mathfrak{m}(c_{a,-}^{(n)})}=\Vert \mathrm{r}_{+}(c_{a,-}^{(n)})%
\Vert _{2}^{2}+P_{\mathfrak{m}}(c_{a,-}^{(n)}+\mathrm{r}_{+}(c_{a,-}^{(n)})).
\label{inequality debile2bis}
\end{equation}%
By (\ref{inequality debile}), we obtain that, for all $n\in \mathbb{N}$, 
\begin{equation}
\Vert \mathrm{r}_{+}(c_{a,-}^{(n)})\Vert _{2}^{2}\leq P_{\mathfrak{m}%
}(c_{a,-}^{(n)})-P_{\mathfrak{m}}(c_{a,-}^{(n)}+\mathrm{r}%
_{+}(c_{a,-}^{(n)})).  \label{inequality debile3}
\end{equation}%
Using Proposition \ref{corrolaire sympa} (ii), one also gets that, for all $%
n\in \mathbb{N}$, 
\begin{equation*}
P_{\mathfrak{m}}(c_{a,-}^{(n)})-P_{\mathfrak{m}}(c_{a,-}^{(n)}+\mathrm{r}%
_{+}(c_{a,-}^{(n)}))\leq 2\left( \Vert \Phi _{a}\Vert _{2}+\Vert \Phi
_{a}^{\prime }\Vert _{2}\right) \Vert \mathrm{r}_{+}(c_{a,-}^{(n)})\Vert
_{2}.
\end{equation*}%
Combined with (\ref{inequality debile3}), the previous inequality yields the
existence of a closed ball $\mathcal{B}_{R,+}\subseteq L_{+}^{2}(\mathcal{A},%
\mathbb{C})$ of radius $R$ centered at $0$ such that 
\begin{equation*}
\{\mathrm{r}_{+}(c_{a,-}^{(n)})\}_{n=1}^{\infty }\in \mathcal{B}_{R,+}.
\end{equation*}%
By compactness and metrizability of $\mathcal{B}_{R,+}$ in the weak topology
(cf. Banach--Alaoglu theorem and Theorem \ref{Metrizability}), we can then
assume that $\mathrm{r}_{+}(c_{a,-}^{(n)})$ weakly converges to $%
d_{a,+}^{\infty }\in L_{+}^{2}(\mathcal{A},\mathbb{C})$ as $n\rightarrow
\infty $.

The map $c_{a,+}\mapsto \left\Vert c_{a,+}\right\Vert _{2}^{2}$ from $%
L_{+}^{2}(\mathcal{A},\mathbb{C})$ to $\mathbb{R}$ is weakly lower
semi--continuity and, by Proposition \ref{corrolaire sympa} (ii), 
\begin{equation*}
c_{a,+}\mapsto P_{\mathfrak{m}}\left( c_{a,-}+c_{a,+}\right)
\end{equation*}%
is weakly continuous on $\mathcal{B}_{R,+}$. It follows that 
\begin{equation*}
\underset{n\rightarrow \infty }{\lim }\left\{ \Vert \mathrm{r}%
_{+}(c_{a,-}^{(n)})\Vert _{2}^{2}+P_{\mathfrak{m}}(c_{a,-}^{(n)}+\mathrm{r}%
_{+}(c_{a,-}^{(n)}))\right\} \geq \left\Vert d_{a,+}^{\infty }\right\Vert
_{2}^{2}+P_{\mathfrak{m}}\left( c_{a,-}+d_{a,+}^{\infty }\right) .
\end{equation*}%
Combined with (\ref{inequality debile}), (\ref{inequality debile2}), and (%
\ref{inequality debile2bis}), the previous inequality implies that $%
d_{a,+}^{\infty }\in \mathcal{C}_{\mathfrak{m}}^{\sharp }(c_{a,-})$ and 
\begin{equation}
\underset{n\rightarrow \infty }{\lim }\Vert \mathrm{r}_{+}(c_{a,-}^{(n)})%
\Vert _{2}^{2}=\Vert d_{a,+}^{\infty }\Vert _{2}^{2}  \label{limite debile4}
\end{equation}%
because of Proposition \ref{corrolaire sympa} (ii). As a consequence, 
\begin{equation*}
d_{a,+}^{\infty }=\mathrm{r}_{+}(c_{a,-})\in \mathcal{C}_{\mathfrak{m}%
}^{\sharp }(c_{a,-}),
\end{equation*}%
cf. Lemma \ref{lemma idiot interaction approx 2 copy(2)} ($\sharp $).
Moreover, since%
\begin{equation*}
\Vert \mathrm{r}_{+}(c_{a,-}^{(n)})-d_{a,+}^{\infty }\Vert _{2}^{2}=\Vert 
\mathrm{r}_{+}(c_{a,-}^{(n)})\Vert _{2}^{2}+\Vert d_{a,+}^{\infty }\Vert
_{2}^{2}-2%
\func{Re}\{\langle \mathrm{r}_{+}(c_{a,-}^{(n)}),d_{a,+}^{\infty }\rangle \},
\end{equation*}%
the limit (\ref{limite debile4}) and the weak convergence of the sequence $\{%
\mathrm{r}_{+}(c_{a,-}^{(n)})\}_{n=1}^{\infty }$ to $d_{a,+}^{\infty }$
imply that $\mathrm{r}_{+}(c_{a,-}^{(n)})$ converges in the norm topology to 
$d_{a,+}^{\infty }\in L_{+}^{2}(\mathcal{A},\mathbb{C})$ as $n\rightarrow
\infty $. 
\end{proof}%

\chapter{Bogoliubov Approximation\ and Effective Theories\label{Section
generalized eq state-effective theory}}

\setcounter{equation}{0}%
The precise characterization of the set $\mathit{\Omega }_{\mathfrak{m}%
}^{\sharp }$ of generalized t.i. equilibrium states defined in Definition %
\ref{definition equilibirum state} is performed in Theorem \ref{theorem
structure of omega copy(1)}. It is the weak$^{\ast }$--closed convex hull of
the set%
\index{Minimizers}%
\index{Free--energy density functional!reduced!minimizers} 
\begin{equation*}
\mathit{%
\hat{M}}_{\mathfrak{m}}:=\left\{ \omega \in E_{1}:\quad g_{\mathfrak{m}%
}\left( \omega \right) =\inf\limits_{\rho \in E_{1}}\,g_{\mathfrak{m}}(\rho
)\right\}
\end{equation*}%
of t.i. minimizers of the reduced free--energy density functional defined by 
\begin{equation}
g_{\mathfrak{m}}\left( \rho \right) :=\Vert \gamma _{a,+}\rho \left( 
\mathfrak{e}_{\Phi _{a}}+i\mathfrak{e}_{\Phi _{a}^{\prime }}\right) \Vert
_{2}^{2}-\Vert \gamma _{a,-}\rho \left( \mathfrak{e}_{\Phi _{a}}+i\mathfrak{e%
}_{\Phi _{a}^{\prime }}\right) \Vert _{2}^{2}+e_{\Phi }(\rho )-\beta
^{-1}s(\rho )  \label{Reduced free energybis}
\end{equation}%
for all $\rho \in E_{1}$, see Definition \ref{Reduced free energy} and (\ref%
{definition minimizers of reduced free energy}). Thus the first aim of the
present chapter is to characterize the weak$^{\ast }$--compact set $\mathit{%
\hat{M}}_{\mathfrak{m}}$ (see Lemma \ref{lemma minimum sympa copy(2)} (i)).

A key information to analyze the set $\mathit{\hat{M}}_{\mathfrak{m}}$ is
given by Theorem \ref{theorem saddle point}. It establishes a relation
between the thermodynamics of models $\mathfrak{m}\in \mathcal{M}_{1}$ and
the thermodynamics of their approximating interactions through thermodynamic
games. Combining this with some additional arguments we prove that $\mathit{%
\hat{M}}_{\mathfrak{m}}$ is a subset of the set\textit{\ }$\overline{\mathrm{%
co}%
\big(%
\mathit{M}(\mathfrak{T}_{\mathfrak{m}}^{\sharp })%
\big)%
}$ (\ref{union equilibrium states min max theory}) of convex combinations of
t.i. equilibrium states coming from the min--max local theory $\mathfrak{T}_{%
\mathfrak{m}}^{\sharp }$ (Definition \ref{definition effective theories bogo}%
). This last result is proven in Section \ref{equilibirum.paragraph} and
gives a first answer to an old open problem in mathematical physics -- first
addressed by Ginibre \cite[p. 28]{Ginibre} in 1968 within a different
context -- about the validity of the so--called Bogoliubov approximation
(see Section \ref{Section bog approx}) on the level of states. Then in\
Section \ref{section breaking theroy} we show that the set $\mathit{\Omega }%
_{\mathfrak{m}}^{\sharp }$ of generalized t.i. equilibrium states is not a
face for an uncountable set of models of $\mathcal{M}_{1}$. This last fact
implies that $\mathit{\Omega }_{\mathfrak{m}}^{\sharp }$ is strictly smaller
than $\overline{\mathrm{co}%
\big(%
\mathit{M}(\mathfrak{T}_{\mathfrak{m}}^{\sharp })%
\big)%
}$, i.e., $\mathit{\Omega }_{\mathfrak{m}}^{\sharp }\varsubsetneq \overline{%
\mathrm{co}%
\big(%
\mathit{M}(\mathfrak{T}_{\mathfrak{m}}^{\sharp })%
\big)%
}$,\ preventing such models to have effective local theories, see
Definitions \ref{definition effective theories} and \ref{definition local
theory}.

\section{Gap equations\label{equilibirum.paragraph}%
\index{Gap equations}}

From Lemma \ref{eq idiot sympa}, we have that 
\begin{eqnarray}
\inf\limits_{\rho \in E_{1}}g_{\mathfrak{m}}\left( \rho \right)
&=&\inf\limits_{\rho \in E_{1}}\left\{ \underset{c_{a,-}\in L_{-}^{2}(%
\mathcal{A},\mathbb{C})}{\inf }\left\{ \left\Vert c_{a,-}\right\Vert
_{2}^{2}+f_{\mathfrak{m}(c_{a,-})}^{\flat }\left( \rho \right) \right\}
\right\}  \label{equality sympa generalbis} \\
&=&\underset{c_{a,-}\in L_{-}^{2}(\mathcal{A},\mathbb{C})}{\inf }\left\{
\inf\limits_{\rho \in E_{1}}\left\{ \left\Vert c_{a,-}\right\Vert
_{2}^{2}+f_{\mathfrak{m}(c_{a,-})}^{\flat }\left( \rho \right) \right\}
\right\}  \label{equality sympa generalbisbis}
\end{eqnarray}%
for any $\mathfrak{m}\in \mathcal{M}_{1}$, where the model $\mathfrak{m}%
(c_{a,-})$ with \emph{purely repulsive} long--range interactions $\Phi
_{a,+}:=\gamma _{a,+}\Phi _{a}$ and $\Phi _{a,+}^{\prime }:=\gamma
_{a,+}\Phi _{a}^{\prime }$ is defined by (\ref{model purement repulsifbis})
in Section \ref{Section theorem saddle point bis} or by (\ref{model purement
repulsif}) in Section \ref{Section effective theories}.

It is thus natural to relate the set $\mathit{%
\hat{M}}_{\mathfrak{m}}$ of t.i. minimizers of the functional $g_{\mathfrak{m%
}}$ with the sets $\mathit{\Omega }_{\mathfrak{m}(d_{a,-})}^{\sharp }$ of
generalized t.i. equilibrium states of models $\mathfrak{m}\left(
d_{a,-}\right) $ for all $d_{a,-}\in \mathcal{C}_{\mathfrak{m}}^{\sharp }$ (%
\ref{eq conserve strategy}). In fact, we verify below that the set $\mathit{%
\hat{M}}_{\mathfrak{m}}$ is the union of the sets $\mathit{\Omega }_{%
\mathfrak{m}(d_{a,-})}^{\sharp }$ for all $d_{a,-}\in \mathcal{C}_{\mathfrak{%
m}}^{\sharp }$:

\begin{lemma}[$\mathit{\hat{M}}_{\mathfrak{m}}$ and generalized t.i.
equilibrium states of $\mathfrak{m}(d_{a,-})$]
\label{lemma explosion l du mec copacabana1 copy(2)}\mbox{ }\newline
\emph{(i)} 
\index{Free--energy density functional!reduced!minimizers}For any $\mathfrak{%
m}\in \mathcal{M}_{1}$,%
\begin{equation*}
\mathit{%
\hat{M}}_{\mathfrak{m}}=\underset{d_{a,-}\in \mathcal{C}_{\mathfrak{m}%
}^{\sharp }}{\cup }\mathit{\Omega }_{\mathfrak{m}(d_{a,-})}^{\sharp }.
\end{equation*}%
\emph{(ii)} For any state $\omega \in \mathit{\hat{M}}_{\mathfrak{m}}$,
there is $d_{a,-}\in \mathcal{C}_{\mathfrak{m}}^{\sharp }$ such that $\omega
\in \mathit{\Omega }_{\mathfrak{m}(d_{a,-})}^{\sharp }$ and%
\begin{equation}
d_{a,-}=\gamma _{a,-}(e_{\Phi _{a}}(\omega )+ie_{\Phi _{a}^{\prime }}(\omega
))\mathrm{\ \ (a.e.).}  \label{gap equation -}
\end{equation}%
\emph{(iii)} Conversely, for any $d_{a,-}\in \mathcal{C}_{\mathfrak{m}%
}^{\sharp }$, all states $\omega \in \mathit{\Omega }_{\mathfrak{m}%
(d_{a,-})}^{\sharp }\subseteq \mathit{\hat{M}}_{\mathfrak{m}}$ satisfy (\ref%
{gap equation -}).
\end{lemma}

\begin{proof}
By using Lemma \ref{eq idiot sympa}, any minimizer $\omega \in \mathit{\hat{M%
}}_{\mathfrak{m}}$ is solution of the variational problem (\ref{equality
sympa generalbis}) with $d_{a,-}\in L_{-}^{2}(\mathcal{A},\mathbb{C})$
satisfying the Euler--Lagrange equations (\ref{gap equation -}). Since the
two infima commute in (\ref{equality sympa generalbis}), $(\omega ,d_{a,-})$%
\ is also solution of the variational problem (\ref{equality sympa
generalbisbis}), i.e., $\omega \in \mathit{\Omega }_{\mathfrak{m}%
(d_{a,-})}^{\sharp }$ and $d_{a,-}\in \mathcal{C}_{\mathfrak{m}}^{\sharp }$.

Conversely, for any $d_{a,-}\in \mathcal{C}_{\mathfrak{m}}^{\sharp }$ and
all $\omega \in \mathit{\Omega }_{\mathfrak{m}(d_{a,-})}^{\sharp }$, $%
(\omega ,d_{a,-})$\ is solution of the variational problem (\ref{equality
sympa generalbisbis}). The latter implies that $(\omega ,d_{a,-})$\ is a
minimum of (\ref{equality sympa generalbis}), i.e., by Lemma \ref{eq idiot
sympa}, $\omega \in \mathit{\hat{M}}_{\mathfrak{m}}$ and $d_{a,-}\in 
\mathcal{C}_{\mathfrak{m}}^{\sharp }$ satisfies the Euler--Lagrange
equations (\ref{gap equation -}). 
\end{proof}%

It now remains to characterize the set $\mathit{\Omega }_{\mathfrak{m}%
(d_{a,-})}^{\sharp }$\ of generalized t.i. equilibrium states for the model $%
\mathfrak{m}\left( d_{a,-}\right) $ (\ref{model purement repulsifbis}) with
purely repulsive long--range interactions. So, the next step is to analyze
the set $\mathit{\Omega }_{\mathfrak{m}}^{\sharp }$ for any arbitrary model
without long--range attractions, that is, $\mathfrak{m}\in \mathcal{M}_{1}$
such that $\Phi _{a,-}=\Phi _{a,-}^{\prime }=0$ (a.e.), see Definition \ref%
{long range attraction-repulsion}. In this case we can relate $\mathit{%
\Omega }_{\mathfrak{m}}^{\sharp }$ to the set 
\begin{equation*}
\mathit{\Omega }_{\mathfrak{m}}^{\sharp }\left( d_{a,+}\right) :=\left\{
\omega \in \mathit{M}_{\Phi (d_{a,+})}:\gamma _{a,+}(e_{\Phi _{a}}(\omega
)+ie_{\Phi _{a}^{\prime }}(\omega ))=d_{a,+}\mathrm{\ (a.e.)}\right\}
\end{equation*}%
defined by (\ref{subset of a face}) for the unique element $c_{a}=d_{a,+}\in 
\mathcal{C}_{\mathfrak{m}}^{\flat }$ (see (\ref{eq conserve strategy}) and
Lemma \ref{lemma idiot interaction approx 2} ($\flat $)), where $\mathit{M}%
_{\Phi (d_{a,+})}$ is the closed face described in Lemma \ref{remark
equilibrium state approches}. In fact we show below that the sets $\mathit{%
\Omega }_{\mathfrak{m}}^{\sharp }$ and $\mathit{\Omega }_{\mathfrak{m}%
}^{\sharp }\left( d_{a,+}\right) $ coincide for any model with purely
repulsive long--range interactions:

\begin{lemma}[$\mathit{\Omega }_{\mathfrak{m}}^{\sharp }$ for models without
long--range attractions]
\label{lemma explosion l du mec copacabana cas repulsif}\mbox{ }\newline
\index{States!generalized equilibrium!purely repulsive models}For any $%
\mathfrak{m}\in \mathcal{M}_{1}$ such that $\Phi _{a,-}=\Phi _{a,-}^{\prime
}=0$ (a.e.) and $\gamma _{a,+}\neq 0$, 
\begin{equation*}
\mathit{\Omega }_{\mathfrak{m}}^{\sharp }=\mathit{%
\hat{M}}_{\mathfrak{m}}=\mathit{\Omega }_{\mathfrak{m}}^{\sharp }\left(
d_{a,+}\right)
\end{equation*}%
with $d_{a,+}\in \mathcal{C}_{\mathfrak{m}}^{\flat }$ being unique.
\end{lemma}

\begin{proof}
If $\Phi _{a,-}=\Phi _{a,-}^{\prime }=0$ (a.e.) then, by (\ref{convex
functional g_m}) and (\ref{Reduced free energybis}), $f_{\mathfrak{m}%
}^{\flat }=g_{\mathfrak{m}}$ on $E_{1}$ and, by\ Theorem \ref{theorem
purement repulsif sympa} ($+$), 
\begin{equation*}
\mathit{\Omega }_{\mathfrak{m}}^{\sharp }=\mathit{\hat{M}}_{\mathfrak{m}}=%
\mathit{M}_{\mathfrak{m}}^{\flat },
\end{equation*}%
where $\mathit{M}_{\mathfrak{m}}^{\flat }$ is the non--empty set of t.i.
minimizers of $f_{\mathfrak{m}}^{\flat }$, see (\ref{set of minimizers de f
bemol}). Therefore, since $\mathfrak{m}(d_{a,+})=\Phi (d_{a,+})$ when $\Phi
_{a,-}=\Phi _{a,-}^{\prime }=0$ (a.e.) (cf. (\ref{m approche 2})), applying\
Lemma \ref{lemma super} we have a t.i. equilibrium state $\omega \in \mathit{%
M}_{\Phi (d_{a,+})}\cap \mathit{\Omega }_{\mathfrak{m}}^{\sharp }$
satisfying the Euler--Lagrange equations 
\begin{equation}
\gamma _{a,+}(e_{\Phi _{a}}(\omega )+ie_{\Phi _{a}^{\prime }}(\omega
))=d_{a,+}\mathrm{\ (a.e.)},  \label{gap equation repulsif}
\end{equation}%
where $d_{a,+}\in \mathcal{C}_{\mathfrak{m}}^{\flat }$ is the unique element
of the set $\mathcal{C}_{\mathfrak{m}}^{\flat }$, see Lemma \ref{lemma idiot
interaction approx 2} ($\flat $).

We now observe that 
\begin{eqnarray}
&&2\func{Re}\{\langle e_{\Phi _{a}}(\rho )+ie_{\Phi _{a}^{\prime }}(\rho
),d_{a,+}\rangle \}  \label{equality super bisbis} \\
&=&\Vert \gamma _{a,+}\rho (\mathfrak{e}_{\Phi _{a}}+i\mathfrak{e}_{\Phi
_{a}^{\prime }})\Vert _{2}^{2}+\Vert d_{a,+}\Vert _{2}^{2}-\Vert \gamma
_{a,+}\rho (\mathfrak{e}_{\Phi _{a}}+i\mathfrak{e}_{\Phi _{a}^{\prime
}})-d_{a,+}\Vert _{2}^{2}  \notag
\end{eqnarray}%
and since $\omega \in \mathit{M}_{\Phi (d_{a,+})}\cap \mathit{\Omega }_{%
\mathfrak{m}}^{\sharp }$ satisfies (\ref{gap equation repulsif}), we obtain
that 
\begin{eqnarray}
&&\inf\limits_{\rho \in E_{1}}\left\{ 2\func{Re}\{\langle e_{\Phi _{a}}(\rho
)+ie_{\Phi _{a}^{\prime }}(\rho ),d_{a,+}\rangle +e_{\Phi }(\rho )-\beta
^{-1}s(\rho )\right\}  \label{equality super bis} \\
&=&\Vert \gamma _{a,+}\omega (\mathfrak{e}_{\Phi _{a}}+i\mathfrak{e}_{\Phi
_{a}^{\prime }})\Vert _{2}^{2}+e_{\Phi }(\omega )-\beta ^{-1}s(\omega
)+\Vert d_{a,+}\Vert _{2}^{2}  \notag \\
&=&g_{\mathfrak{m}}(\omega )+\Vert d_{a,+}\Vert _{2}^{2}  \notag \\
&=&\inf\limits_{\rho \in E_{1}}g_{\mathfrak{m}}(\rho )+\Vert d_{a,+}\Vert
_{2}^{2}.  \label{equality super bisequality super bis}
\end{eqnarray}%
Going backwards from (\ref{equality super bisequality super bis}) to (\ref%
{equality super bis}) and using then (\ref{equality super bisbis}), we
obtain, for any generalized t.i. equilibrium state $\omega \in \mathit{%
\Omega }_{\mathfrak{m}}^{\sharp }=\mathit{\hat{M}}_{\mathfrak{m}}$, the
inequality 
\begin{equation*}
g_{\mathfrak{m}}(\omega )+\Vert d_{a,+}\Vert _{2}^{2}\leq g_{\mathfrak{m}%
}(\omega )-\Vert \gamma _{a,+}\omega \left( \mathfrak{e}_{\Phi _{a}}+i%
\mathfrak{e}_{\Phi _{a}^{\prime }}\right) -d_{a,+}\Vert _{2}^{2}+\Vert
d_{a,+}\Vert _{2}^{2},
\end{equation*}%
i.e., 
\begin{equation*}
\Vert \gamma _{a,+}\omega \left( \mathfrak{e}_{\Phi _{a}}+i\mathfrak{e}%
_{\Phi _{a}^{\prime }}\right) -d_{a,+}\Vert _{2}^{2}\leq 0.
\end{equation*}%
As a consequence, any generalized t.i. equilibrium state $\omega \in \mathit{%
\Omega }_{\mathfrak{m}}^{\sharp }=\mathit{\hat{M}}_{\mathfrak{m}}$ satisfies
the Euler--Lagrange equations (\ref{gap equation repulsif}) with $d_{a,+}\in 
\mathcal{C}_{\mathfrak{m}}^{\flat }$. Combining this with (\ref{equality
super bisbis}) it follows that $\mathit{\Omega }_{\mathfrak{m}}^{\sharp
}\subseteq \mathit{\Omega }_{\mathfrak{m}}^{\sharp }(d_{a,+})$.

Conversely, take any $\omega \in \mathit{M}_{\Phi (d_{a,+})}$ satisfying the
Euler--Lagrange equations (\ref{gap equation repulsif}) with $d_{a,+}\in 
\mathcal{C}_{\mathfrak{m}}^{\flat }$. Such a state $\omega \in \mathit{M}%
_{\Phi (d_{a,+})}$ is a solution of the variational problem (\ref{equality
super bis}) and we easily deduce that $\omega \in \mathit{\Omega }_{%
\mathfrak{m}}^{\sharp }=\mathit{\hat{M}}_{\mathfrak{m}}$. 
\end{proof}%

Applying Lemma \ref{lemma explosion l du mec copacabana cas repulsif} to the
model $\mathfrak{m}\left( d_{a,-}\right) $ (\ref{model purement repulsifbis}%
) with purely repulsive long--range interactions, we obtain the following
corollary:

\begin{corollary}[Generalized t.i. equilibrium states of $\mathfrak{m}%
(d_{a,-})$]
\label{corollary explosion l du mec copacabana final copy(1)}\mbox{ }\newline
\index{States!generalized equilibrium!purely attractive models}For any $%
\mathfrak{m}\in \mathcal{M}_{1}$ and all $d_{a,-}\in \mathcal{C}_{\mathfrak{m%
}}^{\sharp }$, 
\begin{equation*}
\mathit{\Omega }_{\mathfrak{m}(d_{a,-})}^{\sharp }=\mathit{\Omega }_{%
\mathfrak{m}(d_{a,-})}^{\sharp }\left( \mathrm{r}_{+}(d_{a,-})\right) =%
\mathit{\Omega }_{\mathfrak{m}}^{\sharp }\left( d_{a,-}+\mathrm{r}%
_{+}(d_{a,-})\right)
\end{equation*}%
are (non--empty) convex and weak$^{\ast }$--compact subsets of $E_{1}$
satisfying 
\begin{equation*}
\mathit{\Omega }_{\mathfrak{m}(d_{a,-})}^{\sharp }\cap \mathit{\Omega }_{%
\mathfrak{m}(d_{a,-}^{\prime })}^{\sharp }=\emptyset
\end{equation*}%
whenever $d_{a,-}\neq d_{a,-}^{\prime }$ with $d_{a,-},d_{a,-}^{\prime }\in 
\mathcal{C}_{\mathfrak{m}}^{\sharp }$. Here, $\mathrm{r}_{+}$ is the
thermodynamic decision rule defined by (\ref{thermodyn decision rule}) and $%
\mathit{\Omega }_{\mathfrak{m}}^{\sharp }\left( d_{a,-}+\mathrm{r}%
_{+}(d_{a,-})\right) $ is defined by (\ref{subset of a face}).
\end{corollary}

\begin{proof}%
First, $\mathit{\Omega }_{\mathfrak{m}(d_{a,-})}^{\sharp }$ is a
(non--empty) convex and weak$^{\ast }$--compact subset of $E_{1}$ for any $%
d_{a,-}\in \mathcal{C}_{\mathfrak{m}}^{\sharp }$, by Lemma \ref{lemma
minimum sympa copy(1)}. By Lemma \ref{lemma explosion l du mec copacabana1
copy(2)} (iii), all states $\omega \in \mathit{\Omega }_{\mathfrak{m}%
(d_{a,-})}^{\sharp }$ must satisfy (\ref{gap equation -}). On the other
hand, by Lemma \ref{lemma explosion l du mec copacabana cas repulsif}
applied to the model $\mathfrak{m}\left( d_{a,-}\right) $ (\ref{model
purement repulsifbis}) without long--range attractions, we have 
\begin{equation*}
\mathit{\Omega }_{\mathfrak{m}(d_{a,-})}^{\sharp }=\mathit{\Omega }_{%
\mathfrak{m}(d_{a,-})}^{\sharp }\left( \mathrm{r}_{+}(d_{a,-})\right) ,
\end{equation*}%
see (\ref{thermodyn decision rule}). Therefore, by combining (\ref{gap
equation -}) with the last equality, we deduce that 
\begin{equation*}
\mathit{\Omega }_{\mathfrak{m}(d_{a,-})}^{\sharp }=\mathit{\Omega }_{%
\mathfrak{m}}^{\sharp }\left( d_{a,-}+\mathrm{r}_{+}(d_{a,-})\right) ,
\end{equation*}%
which in turn implies that%
\begin{equation*}
\mathit{\Omega }_{\mathfrak{m}(d_{a,-})}^{\sharp }\cap \mathit{\Omega }_{%
\mathfrak{m}(d_{a,-}^{\prime })}^{\sharp }=\emptyset \ 
\end{equation*}%
when $d_{a,-}\neq d_{a,-}^{\prime }$ with $d_{a,-},d_{a,-}^{\prime }\in 
\mathcal{C}_{\mathfrak{m}}^{\sharp }$. 
\end{proof}%

As a consequence, by combining Lemma \ref{lemma explosion l du mec
copacabana1 copy(2)} (i) with Corollary \ref{corollary explosion l du mec
copacabana final copy(1)}, we finally obtain the following theorem:

\begin{theorem}[Characterization of the set $\mathit{%
\hat{M}}_{\mathfrak{m}}$]
\label{theorem explosion l du mec copacabana final}\mbox{ }\newline
\index{Free--energy density functional!reduced!minimizers}For any $\mathfrak{%
m}\in \mathcal{M}_{1}$, 
\begin{equation*}
\mathit{%
\hat{M}}_{\mathfrak{m}}=\underset{d_{a,-}\in \mathcal{C}_{\mathfrak{m}%
}^{\sharp }}{\bigcup }\mathit{\Omega }_{\mathfrak{m}}^{\sharp }\left(
d_{a,-}+\mathrm{r}_{+}(d_{a,-})\right) ,
\end{equation*}%
where $\mathrm{r}_{+}$ is the thermodynamic decision rule defined by (\ref%
{thermodyn decision rule}).%
\index{Thermodynamic game!decision rule}
\end{theorem}

For many relevant models coming from Physics, like, for instance, BCS type
models, the set $\mathit{M}_{\Phi (c_{a})}$ contains exactly one state.
(Actually it is enough to have $|\mathit{M}_{\Phi (c_{a})}|=1$ for $%
c_{a}=d_{a,-}+\mathrm{r}_{+}(d_{a,-})$ with $d_{a,-}\in \mathcal{C}_{%
\mathfrak{m}}^{\sharp }$.) This special case has an interesting
interpretation in terms of game theory as explained in Section \ref{Section
effective theories} after Theorem \ref{th 3.36}. We conclude this section by
proving Theorem \ref{th 3.36}.

First, observe that, in this case, there is an injective and continuous map $%
d_{a,-}\mapsto \omega _{d_{a,-}}$ from $\mathcal{C}_{\mathfrak{m}}^{\sharp }$
to $\mathcal{E}_{1}$:

\begin{lemma}[Properties of the map $d_{a,-}\mapsto \hat{\protect\omega}%
_{d_{a,-}}$]
\label{coro 3.34}\mbox{ }\newline
For any $\mathfrak{m}\in \mathcal{M}_{1}$ and all $d_{a,-}\in \mathcal{C}_{%
\mathfrak{m}}^{\sharp }$, assume that $\mathit{M}_{\Phi
(d_{a,-}+r(d_{a,-}))} $ contains exactly one state denoted by $\hat{\omega}%
_{d_{a,-}}$. Then the map $d_{a,-}\mapsto \hat{\omega}_{d_{a,-}}$ from $%
\mathcal{C}_{\mathfrak{m}}^{\sharp }$ to $\mathcal{E}_{1}$ is injective and
continuous w.r.t. the weak topology on $\mathcal{C}_{\mathfrak{m}}^{\sharp }$
and the weak$^{\ast } $--topology on the set $\mathcal{E}_{1}$ of ergodic
states.
\end{lemma}

\begin{proof}
By the assumptions, $\hat{\omega}_{d_{a,-}}$ is ergodic as $\mathit{M}_{\Phi
(d_{a,-}+\mathrm{r}_{+}(d_{a,-}))}$ is always a face of $E_{1}$, see Lemma %
\ref{remark equilibrium state approches}. If $d_{a,-}\not=d_{a,-}^{\prime }$
then $\hat{\omega}_{d_{a,-}^{\prime }}\not=\hat{\omega}_{d_{a,-}}$because of
Corollary \ref{corollary explosion l du mec copacabana final copy(1)}. Thus
the map $d_{a,-}\mapsto \hat{\omega}_{d_{a,-}}$ is injective. The Hilbert
space $L^{2}(\mathcal{A},\mathbb{C})$ is separable and $\mathcal{C}_{%
\mathfrak{m}}^{\sharp }$ is weakly compact and, therefore, closed in the
weak topology. By Theorem \ref{Metrizability}, the weak topology in $%
\mathcal{C}_{\mathfrak{m}}^{\sharp }$ is metrizable and we can restrict \
ourself to sequences instead of more general nets.

Take any sequence $\{d_{a,-}^{(n)}\}_{n=0}^{\infty }\subseteq \mathcal{C}_{%
\mathfrak{m}}^{\sharp }$ converging in the weak topology to $d_{a,-}\in 
\mathcal{C}_{\mathfrak{m}}^{\sharp }$ as $n\rightarrow \infty $. The
thermodynamic decision rule $\mathrm{r}_{+}$ is weak--norm continuous, by
Lemma \ref{lemma idiot interaction approx 2 copy(3)}, and, from the
definition of $\Phi (c_{a})$, the map $c_{a}\mapsto \Phi (c_{a})$ from $%
L^{2}(\mathcal{A},\mathbb{C})$ to $\mathcal{W}_{1}$ is continuous w.r.t. the
weak topology of $L^{2}(\mathcal{A},\mathbb{C})$ and the norm topology of $%
\mathcal{W}_{1}$. It follows that the sequence 
\begin{equation*}
\left\{ \Phi (d_{a,-}^{(n)}+\mathrm{r}_{+}(d_{a,-}^{(n)}))\right\}
_{n=0}^{\infty }\subseteq \mathcal{W}_{1}
\end{equation*}%
converges in norm to $\Phi (d_{a,-}+\mathrm{r}_{+}(d_{a,-}))\in \mathcal{W}%
_{1}$. The map $\Phi \mapsto \mathrm{P}_{(\Phi ,0,0)}^{\sharp }$ from $%
\mathcal{W}_{1}$ to $\mathbb{R}$ is (norm) continuous, by Theorem \ref{BCS
main theorem 1} (ii). Therefore, 
\begin{equation}
\mathrm{P}_{(\Phi (d_{a,-}+\mathrm{r}_{+}(d_{a,-})),0,0)}^{\sharp }=\underset%
{n\rightarrow \infty }{\lim }\mathrm{P}_{(\Phi (d_{a,-}^{(n)}+\mathrm{r}%
_{+}(d_{a,-}^{(n)})),0,0)}^{\sharp }.  \label{coro 3.34 eq1}
\end{equation}%
By Theorem \ref{BCS main theorem 1} (i) and Lemma \ref{remark equilibrium
state approches}, 
\begin{equation*}
\mathrm{P}_{(\Phi (d_{a,-}^{(n)}+\mathrm{r}_{+}(d_{a,-}^{(n)})),0,0)}^{%
\sharp }=f_{\mathfrak{m}}^{\sharp }(\hat{\omega}_{d_{a,-}^{(n)}})
\end{equation*}%
with $\hat{\omega}_{d_{a,-}^{(n)}}\in \mathit{M}_{\Phi (d_{a,-}^{(n)}+%
\mathrm{r}_{+}(d_{a,-}^{(n)}))}$. Combined with (\ref{coro 3.34 eq1}) and
Lemma \ref{remark equilibrium state approches} for the t.i. interaction $%
\Phi (d_{a,-}+\mathrm{r}_{+}(d_{a,-}))$, the last equality implies that any
accumulation point of the sequence $\{\hat{\omega}_{d_{a,-}^{(n)}}\}_{n=0}^{%
\infty }$ converges in the weak$^{\ast }$--topology to a t.i. equilibrium
state $\hat{\omega}_{d_{a,-}}\in \mathit{M}_{\Phi (d_{a,-}+\mathrm{r}%
_{+}(d_{a,-}))}$ which is assumed to be unique and is thus ergodic. 
\end{proof}%

Notice that, by Lemma \ref{lemma explosion l du mec copacabana1 copy(2)} (i)
and Corollary \ref{corollary explosion l du mec copacabana final copy(1)},
for all $d_{a,-}\in \mathcal{C}_{\mathfrak{m}}^{\sharp }$, the sets $\mathit{%
\Omega }_{\mathfrak{m}}^{\sharp }(d_{a,-}+\mathrm{r}_{+}(d_{a,-}))$ are
never empty. As a consequence, by Theorem \ref{theorem structure of omega
copy(1)} (ii), the map $d_{a,-}\mapsto \omega _{d_{a,-}}$ of Lemma \ref{coro
3.34} is bijective from $\mathcal{C}_{\mathfrak{m}}^{\sharp }$ to the set $%
\mathcal{E}(\mathit{\Omega }_{\mathfrak{m}}^{\sharp })$ of extreme
generalized t.i. equilibrium states. Since $\mathcal{C}_{\mathfrak{m}%
}^{\sharp }$ is weakly compact, it is a homeomorphism:

\begin{corollary}[The map $d_{a,-}\mapsto \protect\omega _{d_{a,-}}$ from $%
\mathcal{C}_{\mathfrak{m}}^{\sharp }$ to $\mathcal{E}(\mathit{\Omega }_{%
\mathfrak{m}}^{\sharp })$]
\label{coro 3.34 copy(1)}\mbox{ }\newline
For any $\mathfrak{m}\in \mathcal{M}_{1}$ and all $d_{a,-}\in \mathcal{C}_{%
\mathfrak{m}}^{\sharp }$, assume that $\mathit{M}_{\Phi
(d_{a,-}+r(d_{a,-}))} $ contains exactly one state denoted by $\omega
_{d_{a,-}}$. Then the map $d_{a,-}\mapsto \omega _{d_{a,-}}$ from $\mathcal{C%
}_{\mathfrak{m}}^{\sharp }$ to $\mathcal{E}_{1}$ defines a homeomorphism
between $\mathcal{C}_{\mathfrak{m}}^{\sharp }$ and $\mathcal{E}(\mathit{%
\Omega }_{\mathfrak{m}}^{\sharp })$ w.r.t. the weak topology in $\mathcal{C}%
_{\mathfrak{m}}^{\sharp }$ and the weak$^{\ast }$--topology in the set $%
\mathcal{E}(\mathit{\Omega }_{\mathfrak{m}}^{\sharp })$. In particular, $%
\mathcal{E}(\mathit{\Omega }_{\mathfrak{m}}^{\sharp })$ is weak$^{\ast }$%
--compact.
\end{corollary}

Consequently, any continuous function $f\in C(\mathcal{E}(\mathit{\Omega }_{%
\mathfrak{m}}^{\sharp }))$ can be identified with a continuous function $%
g\in C(\mathcal{C}_{\mathfrak{m}}^{\sharp })$ through the prescription $%
g(d_{a,-}):=f(\omega _{d_{a,-}})$. This map $C(\mathcal{E}(\mathit{\Omega }_{%
\mathfrak{m}}^{\sharp }))\rightarrow C(\mathcal{C}_{\mathfrak{m}}^{\sharp })$
clearly defines an isomorphism of $C^{\ast }$--algebras. Therefore, by
combining this with Theorems \ref{theorem Bauer} and \ref{theorem omega
simplex}, we obtain Theorem \ref{th 3.36}.

\section{Breakdown of effective local theories\label{section breaking theroy}%
\index{Theory!local!breakdown}}

The fact that the approximating Hamiltonian method (Section \ref{Section
approx method}) leads to the correct pressure (cf. Theorem \ref{theorem
saddle point} ($\sharp $)) does not mean that the min--max local theory $%
\mathfrak{T}_{\mathfrak{m}}^{\sharp }$ (Definition \ref{definition effective
theories bogo}) is an effective theory for $\mathfrak{m}\in \mathcal{M}_{1}$%
. In fact, we prove the existence of uncountably many models $\mathfrak{m}%
\in \mathcal{M}_{1}$ having no effective local theory.

The construction of such models uses the fact, first observed by Israel \cite%
[Theorem V.2.2.]{Israel} for lattice spin systems with purely local
interactions, that any finite set of extreme t.i. states can be seen as t.i.
equilibrium states of some t.i. interaction\ $\Phi \in \mathcal{W}_{1}$:

\begin{lemma}[Ergodic states as t.i. equilibrium states]
\label{corolaire Bishop phelps1 copy(1)}\mbox{ }\newline
For any finite subset $\left\{ 
\hat{\rho}_{1},\ldots ,\hat{\rho}_{n}\right\} \subseteq \mathcal{E}_{1}$ of
ergodic states, there is $\Phi \in \mathcal{W}_{1}$ such that $\left\{ \hat{%
\rho}_{1},\ldots ,\hat{\rho}_{n}\right\} \subseteq \mathit{M}_{\Phi }$.
\end{lemma}

\begin{proof}%
For any $\Phi \in \mathcal{W}_{1}$, recall that the map%
\begin{equation*}
\rho \mapsto f_{\Phi }\left( \rho \right) :=e_{\Phi }(\rho )-\beta
^{-1}s(\rho )
\end{equation*}%
is weak$^{\ast }$--lower semi--continuous and affine, see Lemmata \ref{lemma
property entropy} (i), \ref{Th.en.func} (i) and Definition \ref{Remark free
energy density}. In particular, $\mathit{\Omega }_{\Phi }=\mathit{M}_{\Phi }$
is the (non--empty) set of all t.i. minimizers which is a closed face of $%
E_{1}$. Therefore, the lemma follows from Bishop--Phelps' theorem \cite[%
Theorem V.1.1.]{Israel} together with the Choquet theorem%
\index{Choquet theorem} (Theorem \ref{theorem choquet}) and Theorem \ref%
{eq.tang.bcs.type} for $\mathfrak{m}=(\Phi ,0,0)$. The arguments are \emph{%
exactly} those of Israel and we recommend \cite[Theorem V.2.2. (a)]{Israel}
for more details. 
\end{proof}%

Using this last lemma, we can then construct uncountably many models $%
\mathfrak{m}\in \mathcal{M}_{1}$ such that its set $\mathit{\Omega }_{%
\mathfrak{m}}^{\sharp }$ of generalized t.i. equilibrium states is not a
face of $E_{1}$.

\begin{lemma}[The set $\mathit{\Omega }_{\mathfrak{m}}^{\sharp }$ is
generally not a face]
\label{lemma explosion l du mec copacabana2}\mbox{ }\newline
There are uncountably many $\mathfrak{m}\in \mathcal{M}_{1}$ for which $%
\mathit{\Omega }_{\mathfrak{m}}^{\sharp }$ is not a face of $E_{1}$.
\end{lemma}

\begin{proof}%
Let 
\begin{equation*}
\mathcal{U}^{-}\subseteq \mathcal{U}_{0}\backslash \{A\in \mathcal{U}%
_{0}\;:\;A=A^{\ast }\}
\end{equation*}%
be the (non--empty) set of non self--adjoint local elements of the $\ast $%
--algebra $\mathcal{U}_{0}$ defined by 
\begin{equation*}
\mathcal{U}^{-}:=\underset{\theta \in \mathbb{R}/(2\pi \mathbb{Z)}}{\bigcup }%
\{A\in \mathcal{U}_{0}\;:\;A=-\sigma _{\theta }(A),\;\rho (A)\neq 0\;\mathrm{%
\;for\;some\;}\rho \in E_{1}\}
\end{equation*}%
with $\sigma _{\theta }$ being the automorphism of the algebra $\mathcal{U}$
defined by (\ref{definition of gauge}). Since, for any $x,y\in \mathfrak{L}$
with $x\neq y$, any $\mathrm{s}\in \mathrm{S}$, and any $\lambda \in \mathbb{%
R}\backslash \{0\}$, we have $\lambda a_{x,\mathrm{s}}a_{y,\mathrm{s}}\in 
\mathcal{U}^{-}$, the set $\mathcal{U}^{-}$ contains uncountably many
elements.

By assumption, for any $A\in \mathcal{U}^{-}$, there is $%
\hat{\rho}_{1}\in E_{1}$ such that $\hat{\rho}_{1}(A)\neq 0$. By density of
the set $\mathcal{E}_{1}$ of extreme points of $E_{1}$ (Corollary \ref{lemma
density of extremal points}), we can assume without loss of generality that $%
\hat{\rho}_{1}\in \mathcal{E}_{1}$. As $A\in \mathcal{U}^{-}$, there is $%
\theta \in \mathbb{R}/(2\pi \mathbb{Z)}$ such that 
\begin{equation}
\hat{\rho}_{1}\left( A\right) =-\hat{\rho}_{2}\left( A\right) \neq 0
\label{mec copa eq1}
\end{equation}%
with $\hat{\rho}_{2}:=\hat{\rho}_{1}\circ \sigma _{\theta }$. Since $\sigma
_{\theta }$ is an automorphism of $\mathcal{U}$, $\hat{\rho}_{2}\neq \hat{%
\rho}_{1}$ is clearly a state. As $\hat{\rho}_{1}\in \mathcal{E}_{1}$, by
using Theorem \ref{theorem ergodic extremal} and $\alpha _{x}\circ \sigma
_{\theta }=\sigma _{\theta }\circ \alpha _{x}$, we have that $\hat{\rho}%
_{2}\in \mathcal{E}_{1}$ . Now, by Lemma \ref{corolaire Bishop phelps1
copy(1)}, there is $\Phi \in \mathcal{W}_{1}$ such that $\left\{ \hat{\rho}%
_{1},\hat{\rho}_{2}\right\} \subseteq \mathit{M}_{\Phi }$.

Any non self--adjoint local element $A\in \mathcal{U}^{-}\subseteq \mathcal{U%
}_{0}$ can be decomposed as $A=A_{R}+iA_{I}$ with $A_{R}=A_{R}^{\ast }\in 
\mathcal{U}_{0}$ and $A_{I}=A_{I}^{\ast }\in \mathcal{U}_{0}$. Thus, as
explained in the proof of Lemma \ref{lemma.T}, there exists two finite range
t.i. interactions $\Phi ^{A_{R}},\Phi ^{A_{I}}\in \mathcal{W}_{1}$ with $%
\Vert \Phi ^{A_{R}}\Vert _{\mathcal{W}_{1}}=\Vert A_{R}\Vert $ and $\Vert
\Phi ^{A_{I}}\Vert _{\mathcal{W}_{1}}=\Vert A_{I}\Vert $ such that%
\begin{equation}
\rho (A)=e_{\Phi ^{A_{R}}}(\rho )+ie_{\Phi ^{A_{I}}}(\rho )
\label{mec copa eq2}
\end{equation}%
for any $\rho \in E_{1}$. For any $A\in \mathcal{U}^{-}$, we define the
discrete model 
\begin{equation*}
\mathfrak{m}_{A}:=\left( \Phi ,\Phi ^{A_{R}},\Phi ^{A_{I}}\right) \in 
\mathcal{M}_{1}
\end{equation*}%
without long--range attractions, i.e., $\Phi _{a,-}=\Phi _{a,-}^{\prime }=0$%
, $\Phi _{a,+}:=\Phi ^{A_{R}}$, and $\Phi _{a,+}^{\prime }:=\Phi ^{A_{I}}$,
see Definition \ref{long range attraction-repulsion}.

As $\left\{ \hat{\rho}_{1},\hat{\rho}_{2}\right\} \subseteq \mathit{M}_{\Phi
}$ and by convexity of the set $\mathit{M}_{\Phi }$, 
\begin{equation}
\omega :=\frac{1}{2}\hat{\rho}_{1}+\frac{1}{2}\hat{\rho}_{2}\in \mathit{M}%
_{\Phi }.  \label{mec copa eq3}
\end{equation}%
It follows from Definition \ref{Reduced free energy} that%
\begin{equation}
g_{\mathfrak{m}_{A}}\left( \omega \right) =f_{\Phi }\left( \omega \right)
<g_{\mathfrak{m}_{A}}\left( \hat{\rho}_{1}\right) =g_{\mathfrak{m}%
_{A}}\left( \hat{\rho}_{2}\right)  \label{mec copa eq4}
\end{equation}%
because of (\ref{mec copa eq1}) and (\ref{mec copa eq2}). Therefore, $\hat{%
\rho}_{1},\hat{\rho}_{2}\notin \mathit{\hat{M}}_{\mathfrak{m}_{A}}$ do not
belong to the set $\mathit{\hat{M}}_{\mathfrak{m}_{A}}$ (\ref{definition
minimizers of reduced free energy}) of minimizers of $g_{\mathfrak{m}_{A}}$
over $E_{1}$. However, since $\mathfrak{m}_{A}$ is a model with purely
repulsive long--range interactions, $f_{\Phi }\leq g_{\mathfrak{m}_{A}}$ on $%
E_{1}$ and, by (\ref{mec copa eq3}) and (\ref{mec copa eq4}), we obtain that 
$\omega \in \mathit{\hat{M}}_{\mathfrak{m}_{A}}$. Since $\mathit{\Omega }_{%
\mathfrak{m}_{A}}^{\sharp }=\mathit{\hat{M}}_{\mathfrak{m}_{A}}$, by Theorem %
\ref{theorem purement repulsif sympa} ($+$), we finally get that $\omega \in 
\mathit{\Omega }_{\mathfrak{m}_{A}}^{\sharp }$, whereas $\hat{\rho}_{1},\hat{%
\rho}_{2}\notin \mathit{\Omega }_{\mathfrak{m}_{A}}^{\sharp }$ in spite of
the decomposition (\ref{mec copa eq3})%
\index{Face}. In other words, for any $A\in \mathcal{U}^{-}$, $\mathit{%
\Omega }_{\mathfrak{m}_{A}}^{\sharp }$ is not a face of $E_{1}$. 
\end{proof}%

As a consequence, the equality $\mathrm{P}_{\mathfrak{m}}^{\sharp }=-\mathrm{%
F}_{\mathfrak{m}}^{\sharp }$ of Theorem \ref{theorem saddle point} ($\sharp $%
) does not necessarily imply that the min--max local theory $\mathfrak{T}_{%
\mathfrak{m}}^{\sharp }$ (Definition \ref{definition effective theories bogo}%
) is an effective theory, see Definition \ref{definition effective theories}%
. In fact, if $\mathit{\Omega }_{\mathfrak{m}}^{\sharp }$ is not a face then
there is no effective local theory and Lemma \ref{lemma explosion l du mec
copacabana2} implies Theorem \ref{lemma explosion l homme capacabana1}.

\chapter{Appendix\label{Section appendix}}

\setcounter{equation}{0}%
For the reader's convenience we give here a short review on the following
subjects:

\begin{itemize}
\item Gibbs equilibrium states (Section \ref{Section Gibbs equilibrium
states}), see, e.g., \cite{BrattelliRobinson};

\item The approximating Hamiltonian method (Section \ref{Section approx
method}), see, e.g., \cite%
{Bogjunior,approx-hamil-method0,approx-hamil-method,approx-hamil-method2};

\item $\mathcal{L}^{p}$--spaces of maps with values in a Banach space
(Section \ref{Section Preliminaries});

\item Compact convex sets and Choquet simplices (Section \ref{Section
Compact convex sets and Choquet simplices}), see, e.g., \cite{Alfsen,Phe};

\item $\Gamma $--regularization of real functionals (Section \ref{Section
gamma regularization}), see, e.g., \cite{Alfsen,Schirotzek,BruPedraconvex};

\item Legendre--Fenchel transform and tangent functionals (Section \ref%
{Section Legendre-Fenchel transform}), see, e.g., \cite{Simon,Zeidler3};

\item Two--person zero--sum games (Section \ref{Section two--person
zero--sum games}), see, e.g., \cite{Aubinbis,Aubin}.
\end{itemize}

\noindent These subjects are rather standard and can be found in many
textbooks. Therefore, we keep the exposition here as short as possible and
only concentrate on results used in this monograph. It is important to note,
however, that we also give two new and useful theorems -- Theorems \ref%
{theorem trivial sympa 1} and \ref{theorem trivial sympa 1 copy(1)} -- which
are general results related to the study of variational problems with
non--convex functionals on compact convex sets. Observe further that Lemma %
\ref{Bauer maximum principle bis} in Section \ref{Section gamma
regularization} does not seem to have been observed before. In fact, Lemma %
\ref{Bauer maximum principle bis} and Theorems \ref{theorem trivial sympa 1}%
--\ref{theorem trivial sympa 1 copy(1)} are given in this appendix -- and
not in the main part of the text -- as they are the subject of a separate
paper \cite{BruPedraconvex} to be published soon.

\section{Gibbs equilibrium states\label{Section Gibbs equilibrium states}%
\index{States!Gibbs|textbf}}

\index{States!Gibbs}In quantum statistical mechanics a physical system of
fermions on a lattice is first characterized by its energy observables $%
U_{\Lambda }$ for particles enclosed in finite boxes $\Lambda \subseteq 
\mathfrak{L}$. Mathematically speaking, $U_{\Lambda }$ are self--adjoint
elements of the local algebras $\mathcal{U}_{\Lambda }$. Given any local
state $\rho _{\Lambda }\in E_{\Lambda }$ on $\mathcal{U}_{\Lambda }$, the
energy observable $U_{\Lambda }$ fixes the so--called \emph{finite--volume}
free--energy density (in the box $\Lambda \subseteq \mathfrak{L}$)%
\index{Free--energy density functional!finite volume}%
\begin{equation*}
f_{\Lambda ,U_{\Lambda }}\left( \rho _{\Lambda }\right) :=|\Lambda
|^{-1}\rho _{\Lambda }(U_{\Lambda })-(\beta |\Lambda |)^{-1}S(\rho _{\Lambda
}),
\end{equation*}%
of the physical system at inverse temperature $\beta >0$. The functional $%
f_{\Lambda ,U_{\Lambda }}$ can be seen either as a map from $E_{\Lambda }$
to $\mathbb{R}$ or from $E$ to $\mathbb{R}$ by taking, for all $\rho \in E$,
the restriction $\rho _{\Lambda }\in E_{\Lambda }$ on $\mathcal{U}_{\Lambda
} $. The first term in $f_{\Lambda ,U_{\Lambda }}$ is obviously the mean
energy per volume of the physical system found in the state $\rho _{\Lambda
} $, whereas $S$ is the von Neumann entropy%
\index{Entropy density functional!von Neumann} defined by (\ref{neuman
entropy}) which measures, in a sense, the amount of randomness carried by
the state. See Section \ref{section neuman entropy} for more details.

The state of a system in thermal equilibrium and at fixed mean energy per
volume maximizes the entropy, by the second law of thermodynamics.
Therefore, it minimizes the free--energy density functional $f_{\Lambda
,U_{\Lambda }}$. Such well--known arguments lead to the study of the
variational problem 
\begin{equation}
\underset{\rho \in E}{\inf }f_{\Lambda ,U_{\Lambda }}\left( \rho \right) =%
\underset{\rho _{\Lambda }\in E_{\Lambda }}{\inf }f_{\Lambda ,U_{\Lambda
}}\left( \rho _{\Lambda }\right) .  \label{finite volume variational problem}
\end{equation}%
As the von Neumann entropy $S$ is weak$^{\ast }$--continuous, the functional 
$f_{\Lambda ,U_{\Lambda }}$ has at least one minimizer on $E_{\Lambda }$
which is the local equilibrium state of the physical system, also called 
\emph{Gibbs equilibrium state}:

\begin{definition}[Gibbs equilibrium state]
\label{Gibbs.statebis}\mbox{ }\newline
A Gibbs equilibrium state is a solution of the variational problem (\ref%
{finite volume variational problem}), i.e., a minimizer of the
finite--volume free--energy density functional $f_{\Lambda ,U_{\Lambda }}$
on $E_{\Lambda }$.
\end{definition}

The set of solutions of the variational problem (\ref{finite volume
variational problem}) is, a priori, not unique. But, for $\beta \in
(0,\infty )$, it is well--known that the maximum of $-f_{\Lambda ,U_{\Lambda
}}$ over $E$ equals the finite--volume pressure%
\index{Pressure!finite volume} 
\begin{equation*}
p_{\Lambda ,U_{\Lambda }}:=%
\frac{1}{\beta |\Lambda _{l}|}\ln \mathrm{Trace}_{\wedge \mathcal{H}%
_{\Lambda }}\left( \mathrm{e}^{-\beta U_{\Lambda }}\right)
\end{equation*}%
(compare with (\ref{BCS pressure})) and is attained for the unique minimizer 
$\rho _{\Lambda ,U_{\Lambda }}\in E_{\Lambda }$ of $f_{\Lambda ,U_{\Lambda
}} $ defined by%
\begin{equation}
\rho _{\Lambda ,U_{\Lambda }}\left( A\right) :=\frac{\mathrm{Trace}_{\wedge 
\mathcal{H}_{\Lambda }}\left( A\,\mathrm{e}^{-\beta U_{\Lambda }}\right) }{%
\mathrm{Trace}_{\wedge \mathcal{H}_{\Lambda }}\left( \mathrm{e}^{-\beta
U_{\Lambda }}\right) },\quad A\in \mathcal{U}_{\Lambda }.
\label{Gibbs.state}
\end{equation}%
This result is a key ingredient in the proof of Theorem \ref{BCS main
theorem 1} (see Chapters \ref{Stoermer} and \ref{section proof of theorem
main}) and is also known in the literature as the \emph{passivity of Gibbs
states}:

\begin{theorem}[Passivity of Gibbs states]
\label{passivity.Gibbs}\mbox{ }\newline
\index{Passivity of Gibbs states|textbf}For $\beta \in (0,\infty )$ and any
self--adjoint $U_{\Lambda }\in \mathcal{U}_{\Lambda }$, 
\begin{equation*}
p_{\Lambda ,U_{\Lambda }}=-\underset{\rho \in E}{\inf }f_{\Lambda
,U_{\Lambda }}\left( \rho \right) =-\underset{\rho _{\Lambda }\in E_{\Lambda
}}{\inf }f_{\Lambda ,U_{\Lambda }}\left( \rho _{\Lambda }\right)
=-f_{\Lambda ,U_{\Lambda }}\left( \rho _{\Lambda ,U_{\Lambda }}\right)
\end{equation*}%
with the Gibbs equilibrium state $\rho _{\Lambda ,U_{\Lambda }}\in
E_{\Lambda }$ being the unique minimizer on $E_{\Lambda }$ of the
finite--volume free--energy density functional $f_{\Lambda ,U_{\Lambda }}$.
\end{theorem}

\noindent The proof of this standard theorem is a (non--trivial) consequence
of\ Jensen's inequality%
\index{Jensen's inequality}, see, e.g., \cite[Lemma 6.3]{BruPedra1} or \cite[%
Proposition 6.2.22]{BrattelliRobinson} (for quantum spin systems).

\section{The approximating Hamiltonian method\label{Section approx method}%
\index{Approximating Hamiltonian method|textbf}}

The approximating Hamiltonian method is presented in \cite%
{Bogjunior,approx-hamil-method0,approx-hamil-method,approx-hamil-method2}.
This rigorous technique for computing the thermodynamic pressure does not
seem to be well--known in the mathematical physics community, unfortunately.
Therefore, we give below a brief account on the approximating Hamiltonian
method and we compare it to our results.

Let 
\begin{equation}
\mathrm{H}_{\Lambda }:=\mathrm{T}_{\Lambda }+%
\frac{1}{|\Lambda |}\underset{k=1}{\overset{N}{\sum }}\gamma _{k}\left(
U_{k,\Lambda }+iU_{k,\Lambda }^{\prime }\right) ^{\ast }\left( U_{k,\Lambda
}+iU_{k,\Lambda }^{\prime }\right)  \label{Hamiltonian AHM}
\end{equation}%
be any self--adjoint operator acting on a Hilbert space $\mathfrak{H}%
_{\Lambda }$ of a box $\Lambda $ with $\gamma _{k}=-1$ for any $k\in
\{1,\cdots ,n\}$ and $\gamma _{k}=1$ for $k\in \{n+1,\cdots ,N\}$ ($n<N$
being fixed). Here, $\mathrm{T}_{\Lambda }=\mathrm{T}_{\Lambda }^{\ast }$
and $\{U_{k,\Lambda },U_{k,\Lambda }^{\prime }\}_{k=1}^{N}$ are operators
acting on $\mathfrak{H}_{\Lambda }$. Then the approximating Hamiltonian
method corresponds to use so--called approximating Hamiltonians to compute
the finite--volume pressure%
\begin{equation*}
p\left[ \mathrm{H}_{\Lambda }\right] :=\dfrac{1}{\beta |\Lambda |}\ln 
\mathrm{Trace}_{\mathfrak{H}_{\Lambda }}(\mathrm{e}^{-\beta \mathrm{H}%
_{\Lambda }})
\end{equation*}%
associated with $\mathrm{H}_{\Lambda }$, for any $\beta \in (0,\infty )$, in
the thermodynamic limit. A minimal requirement on $\mathrm{H}_{\Lambda }$ to
have a thermodynamic behavior is of course to ensure the finiteness of $p%
\left[ \mathrm{H}_{\Lambda }\right] $. The latter is, in fact, fulfilled
because this method is based on operators $\mathrm{T}_{\Lambda }=\mathrm{T}%
_{\Lambda }^{\ast }$ and $\{U_{k,\Lambda },U_{k,\Lambda }^{\prime
}\}_{k=1}^{N}$ satisfying the following conditions:

\begin{itemize}
\item[(A1)] The finite--volume pressure of $\mathrm{T}_{\Lambda }$ exists,
i.e., 
\begin{equation*}
\left\vert \ln \mathrm{Trace}_{\mathfrak{H}_{\Lambda }}(\mathrm{e}^{-\beta 
\mathrm{T}_{\Lambda }})\right\vert \leq \beta |\Lambda |\ C_{0}.
\end{equation*}

\item[(A2)] The operators 
\begin{equation*}
(U_{k,\Lambda }+iU_{k,\Lambda }^{\prime })^{\#}\in \{U_{k,\Lambda
}+iU_{k,\Lambda }^{\prime },(U_{k,\Lambda }+iU_{k,\Lambda }^{\prime })^{\ast
}\}
\end{equation*}%
are bounded in operator norm, for any $k\in \{1,\cdots ,N\}$, by $%
C_{1}|\Lambda |$.

\item[(A3)] The following commutators are also bounded for any $k,q,p\in
\{1,\cdots ,N\}$: 
\begin{equation*}
\begin{array}{l}
\Vert \lbrack U_{k,\Lambda }+iU_{k,\Lambda }^{\prime },(U_{q,\Lambda
}+iU_{q,\Lambda }^{\prime })^{\#}]\Vert \leq |\Lambda |C_{2}. \\ 
\Vert \lbrack (U_{k,\Lambda }+iU_{k,\Lambda }^{\prime })^{\#},[(U_{q,\Lambda
}+iU_{q,\Lambda }^{\prime })^{\#},U_{p,\Lambda }+iU_{p,\Lambda }^{\prime
}]]\Vert \leq |\Lambda |C_{3}. \\ 
\Vert \lbrack (U_{k,\Lambda }+iU_{k,\Lambda }^{\prime })^{\#},[U_{q,\Lambda
}+iU_{q,\Lambda }^{\prime },\mathrm{T}_{\Lambda }]]\Vert \leq |\Lambda
|C_{4}.%
\end{array}%
\end{equation*}
\end{itemize}

\noindent For all $k\in \{1,2,3,4\}$, note that the constants $C_{k}$ are
finite and do not depend on the box $\Lambda $.

Approximating Hamiltonians are then defined from $\mathrm{H}_{\Lambda }$ by%
\begin{eqnarray*}
\mathrm{H}_{\Lambda }\left( \vec{c}_{-},\vec{c}_{+}\right) := &&\mathrm{T}%
_{\Lambda }-\underset{k=1}{\overset{n}{\sum }}\left( \bar{c}_{k,-}\left(
U_{k,\Lambda }+iU_{k,\Lambda }^{\prime }\right) +c_{k,-}\left( U_{k,\Lambda
}+iU_{k,\Lambda }^{\prime }\right) ^{\ast }\right) \\
&&+\underset{k=n+1}{\overset{N}{\sum }}\left( \bar{c}_{k,+}\left(
U_{k,\Lambda }+iU_{k,\Lambda }^{\prime }\right) +c_{k,+}\left( U_{k,\Lambda
}+iU_{k,\Lambda }^{\prime }\right) ^{\ast }\right)
\end{eqnarray*}%
with $\vec{c}_{-}:=(c_{1,-},\cdots ,c_{n,-})\in \mathbb{C}^{n}$, $\vec{c}%
_{+}:=(c_{n+1,+},\cdots ,c_{N,+})\in \mathbb{C}^{N-n}$. Let%
\index{Approximating Hamiltonian method!free--energy density} 
\begin{equation*}
\mathfrak{f}_{\mathrm{H},\Lambda }\left( 
\vec{c}_{-},\vec{c}_{+}\right) :=-\left\vert \vec{c}_{+}\right\vert
^{2}+\left\vert \vec{c}_{-}\right\vert ^{2}-\dfrac{1}{\beta |\Lambda |}\ln 
\mathrm{Trace}_{\mathfrak{H}_{\Lambda }}(\mathrm{e}^{-\beta \mathrm{H}%
_{\Lambda }\left( \vec{c}_{-},\vec{c}_{+}\right) })
\end{equation*}%
be the approximating free--energy density and%
\index{Approximating Hamiltonian method!Gibbs equilibrium state} 
\begin{equation*}
\left\langle -\right\rangle _{%
\vec{c}_{-},\vec{c}_{+}}:=\frac{\mathrm{Trace}_{\mathfrak{H}_{\Lambda
}}\left( -\,\mathrm{e}^{-\beta \mathrm{H}_{\Lambda }\left( \vec{c}_{-},\vec{c%
}_{+}\right) }\right) }{\mathrm{Trace}_{\mathfrak{H}_{\Lambda }}\left( 
\mathrm{e}^{-\beta \mathrm{H}_{\Lambda }\left( \vec{c}_{-},\vec{c}%
_{+}\right) }\right) }
\end{equation*}%
be the (local) Gibbs equilibrium state associated with $\mathrm{H}_{\Lambda
}\left( \vec{c}_{-},\vec{c}_{+}\right) $, see Section \ref{Section Gibbs
equilibrium states}. From (A1)--(A2) it can be proven that, for any $\vec{c}%
_{-}\in \mathbb{C}^{n}$, there is a unique solution $\mathrm{r}_{+}(\vec{c}%
_{-}):=(d_{1,+},\cdots ,d_{N,+})\in \mathbb{C}^{N-n}$ of the
(finite--volume) gap equations%
\index{Approximating Hamiltonian method!gap equations} 
\begin{equation}
|\Lambda |^{-1}\left\langle U_{k,\Lambda }+iU_{k,\Lambda }^{\prime
}\right\rangle _{%
\vec{c}_{-},\mathrm{r}_{+}(\vec{c}_{-})}=d_{k,+}  \label{gap equations AHM}
\end{equation}%
for all $k\in \{n+1,\cdots ,N\}$. Then let us consider two additional
conditions:

\begin{itemize}
\item[(A4)] For any $k\in \{n,\cdots ,N\}$ with fixed $n<N$, the operators $%
U_{k,\Lambda }^{\#}$ satisfy the ergodicity condition%
\index{Approximating Hamiltonian method!ergodicity condition} 
\begin{multline*}
\underset{\left\vert \Lambda \right\vert \rightarrow \infty }{\lim }\left\{
|\Lambda |^{-2}\left( \left\langle \left( U_{k,\Lambda }+iU_{k,\Lambda
}^{\prime }\right) ^{\ast }(U_{k,\Lambda }+iU_{k,\Lambda }^{\prime
})\right\rangle _{%
\vec{c}_{-},\mathrm{r}_{+}(\vec{c}_{-})}\right. \right. \\
\left. \left. -|\left\langle U_{k,\Lambda }+iU_{k,\Lambda }^{\prime
}\right\rangle _{\vec{c}_{-},\mathrm{r}_{+}(\vec{c}_{-})}|^{2}\right)
\right\} =0
\end{multline*}%
for all $\vec{c}_{-}\in \mathbb{C}^{n}$.

\item[(A5)] The free--energy density $\mathfrak{f}_{\mathrm{H},\Lambda
}\left( \vec{c}_{-},\vec{c}_{+}\right) $ converges in the thermodynamic
limit $\left\vert \Lambda \right\vert \rightarrow \infty $ towards 
\begin{equation*}
\underset{\left\vert \Lambda \right\vert \rightarrow \infty }{\lim }%
\mathfrak{f}_{\mathrm{H},\Lambda }\left( \vec{c}_{-},\vec{c}_{+}\right) =:%
\mathfrak{f}_{\mathrm{H}}\left( \vec{c}_{-},\vec{c}_{+}\right)
\end{equation*}%
for any $\vec{c}_{-}\in \mathbb{C}^{n}$ and $\vec{c}_{+}\in \mathbb{C}^{N-n}$%
.
\end{itemize}

\noindent Bogoliubov Jr. \textit{et al.} have shown \cite%
{approx-hamil-method} the following:

\begin{theorem}[Bogoliubov Jr., Brankov, Zagrebnov, Kurbatov, and Tonchev]
\label{Theorem AHM}Under assumptions (A1)--(A4) we obtain:%
\index{Approximating Hamiltonian method!pressure} \newline
\emph{(i)} For any box $\Lambda $ and at fixed $%
\vec{c}_{-}\in \mathbb{C}^{n}$, the solution of the variational problem%
\begin{equation*}
\underset{c_{+}\in \mathbb{C}^{N-n}}{\sup }\mathfrak{f}_{\mathrm{H},\Lambda
}(\vec{c}_{-},\vec{c}_{+})=\mathfrak{f}_{\mathrm{H},\Lambda }(\vec{c}_{-},%
\mathrm{r}_{+}(\vec{c}_{-}))
\end{equation*}%
is unique and solution of (\ref{gap equations AHM}), whereas there is $\vec{d%
}_{-}\in \mathbb{C}^{n}$ such that%
\begin{equation*}
\underset{\vec{c}_{-}\in \mathbb{C}^{n}}{\inf }\left\{ \underset{\vec{c}%
_{+}\in \mathbb{C}^{N-n}}{\sup }\mathfrak{f}_{\mathrm{H},\Lambda }(\vec{c}%
_{-},\vec{c}_{+})\right\} =\mathfrak{f}_{\mathrm{H},\Lambda }(\vec{d}_{-},%
\mathrm{r}_{+}(\vec{d}_{-})).
\end{equation*}%
\emph{(ii)} In the thermodynamic limit 
\begin{equation*}
\underset{\left\vert \Lambda \right\vert \rightarrow \infty }{\lim }\left\{
p[\mathrm{H}_{\Lambda }]+\mathfrak{f}_{\mathrm{H},\Lambda }(\vec{d}_{-},%
\mathrm{r}_{+}(\vec{d}_{-}))\right\} =0
\end{equation*}%
and if (A5) also holds then 
\begin{equation*}
\underset{\left\vert \Lambda \right\vert \rightarrow \infty }{\lim }p[%
\mathrm{H}_{\Lambda }]=-\underset{\vec{c}_{-}\in \mathbb{C}^{n}}{\inf }%
\left\{ \underset{\vec{c}_{+}\in \mathbb{C}^{N-n}}{\sup }\mathfrak{f}_{%
\mathrm{H}}(\vec{c}_{-},\vec{c}_{+})\right\} .
\end{equation*}
\end{theorem}

\noindent The proof of this theorem uses as a key ingredient the Bogoliubov
(convexity) inequality \cite[Corollary D.4]{BruZagrebnov8} which can be
deduced from Theorem \ref{passivity.Gibbs}. Another important technique used
by the authors \cite%
{approx-hamil-method0,approx-hamil-method,approx-hamil-method2} are the
Ginibre inequalities \cite[Eq. (2.10)]{Ginibre}. Their proofs are thus
essentially different from ours.

To conclude we analyze Conditions (A1)--(A3) and (A5) for discrete Fermi
systems $\mathfrak{m}\in \mathcal{M}_{1}^{\mathrm{d}}$ (see Section \ref%
{definition models}).

\begin{lemma}[Conditions (A1)--(A3) and (A5) for $\mathfrak{m}\in \mathcal{M}%
_{1}^{\mathrm{d}}$]
\label{Theorem AHM comparison}\mbox{ }\newline
For any discrete Fermi system $\mathfrak{m}:=\{\Phi \}\cup \{\Phi _{k},\Phi
_{k}^{\prime }\}_{k=1}^{N}\in \mathcal{M}_{1}^{\mathrm{d}}$, the
self--adjoint operators $\mathrm{T}_{\Lambda _{l}}:=U_{\Lambda _{l}}^{\Phi }$%
, $U_{k,\Lambda _{l}}:=U_{\Lambda _{l}}^{\Phi _{k}}$, and $U_{k,\Lambda
_{l}}:=U_{\Lambda _{l}}^{\Phi _{k}}$ satisfy Conditions (A1)--(A2) and (A5).
(A3) holds whenever $\mathfrak{m}\in \mathcal{M}_{1}^{\mathrm{df}}$ is also
finite range.
\end{lemma}

\begin{proof}
Condition (A1)--(A2) and (A5) are clearly satisfied. Condition (A3) requires
direct computations.\ We omit the details. 
\end{proof}%

\begin{remark}[Condition (A4) as a non--necessary assumption]
\label{important remark}\mbox{ }\newline
\index{Long--range models!discrete}Condition (A4) is used in\ Theorem \ref%
{Theorem AHM} to handle the positive part of long--range interactions. It is
generally not satisfied for discrete Fermi systems $\mathfrak{m}\in \mathcal{%
M}_{1}^{\mathrm{d}}$. This condition is shown here to be absolutely not
necessary to handle the thermodynamic limit of the pressure of Fermi systems 
$\mathfrak{m}\in \mathcal{M}_{1}$ (see Theorem \ref{theorem saddle point}).
\end{remark}

\begin{remark}[Condition (A3) as a non--necessary assumption]
\label{important remarkbis}\mbox{ }\newline
Let $\Phi ,\Phi ^{\prime }\in \mathcal{W}_{1}$ such that $\Vert \Phi
_{\Lambda _{l}}\Vert ,\Vert \Phi _{\Lambda _{l}}\Vert =\mathcal{O}%
(l^{-(d+\epsilon )})$ for some small $\epsilon >0$. Such interactions
clearly exist. If this is the only information we have about the
interactions then the only bound we can give for the commutators $%
[U_{\Lambda }^{\Phi },U_{\Lambda }^{\Phi ^{\prime }}]$ is 
\begin{equation*}
\Vert \lbrack U_{\Lambda }^{\Phi },U_{\Lambda }^{\Phi ^{\prime }}]\Vert \leq
\sum\limits_{\Lambda _{1},\Lambda _{1}^{\prime }\subseteq \Lambda ,\;\Lambda
_{1}\cap \Lambda _{1}^{\prime }\not=\emptyset }2\Vert \Phi _{\Lambda
_{1}}\Vert \,\Vert \Phi _{\Lambda _{1}^{\prime }}^{\prime }\Vert .
\end{equation*}%
Depending on $\epsilon >0$, the r.h.s. of the last inequality grows at large 
$|\Lambda |$ much faster than the volume $|\Lambda |$. Hence, the condition
(A3) is very unlikely to hold for all $\Phi ,\Phi ^{\prime }\in \mathcal{W}%
_{1}$.
\end{remark}

\section{$\mathcal{L}^{p}$--spaces of maps with values in a Banach space 
\label{Section Preliminaries}}

Let $(\mathcal{A},\mathfrak{A},\mathfrak{a})$ be a \emph{separable} measure
space%
\index{Separable measure space|textbf} with $\mathfrak{A}$ and $\mathfrak{a}:%
\mathfrak{A}\rightarrow \mathbb{R}_{0}^{+}$ being respectively some $\sigma $%
--algebra on $\mathcal{A}$\ and some measure on $\mathfrak{A}$. Recall that $%
(\mathcal{A},\mathfrak{A},\mathfrak{a})$ being separable means that the
space $L^{2}(\mathcal{A},\mathbb{C}):=L^{2}(\mathcal{A},\mathfrak{a},\mathbb{%
C})$ of square integrable complex valued functions on $\mathcal{A}$ is a
separable Hilbert space. This property implies, in particular, that $(%
\mathcal{A},\mathfrak{A},\mathfrak{a})$ is a $\sigma $--finite measure
space, see \cite[p. 54]{Conway}.

Let $\mathcal{X}$ be any Banach space with norm $\left\Vert \cdot
\right\Vert _{\mathcal{X}}$. We denote by $\mathcal{S}\left( \mathcal{A},%
\mathcal{X}\right) $ the set of measurable step functions with support of
finite measure. For any measurable map $\mathfrak{s}_{a}:\mathcal{A}%
\rightarrow \mathcal{X}$ and any $p\geq 1$, we define the semi--norm%
\begin{equation*}
\left\Vert \mathfrak{s}\right\Vert _{p}:=\int_{\mathcal{A}}\left\Vert 
\mathfrak{s}_{a}\right\Vert _{\mathcal{X}}^{p}\mathrm{d}\mathfrak{a}\left(
a\right) \in \lbrack 0,\infty ].
\end{equation*}%
Let $\mathfrak{s}_{a}^{(n)}$ be any $\mathcal{L}^{p}$--Cauchy sequence of
measurable maps, i.e., $\Vert \mathfrak{s}_{a}^{(n)}\Vert _{p}<\infty $ and 
\begin{equation*}
\underset{N\rightarrow \infty }{\lim }%
\text{ }\underset{n,m>N}{\sup }\Vert \mathfrak{s}_{a}^{(n)}-\mathfrak{s}%
_{a}^{(m)}\Vert _{p}=0.
\end{equation*}%
Then there is a measurable function $\mathfrak{s}_{\infty }$ from $\mathcal{A%
}$ to $\mathcal{X}$ with $\Vert \mathfrak{s}_{a}^{(\infty )}\Vert
_{p}<\infty $ such that 
\begin{equation*}
\underset{n\rightarrow \infty }{\lim }\Vert \mathfrak{s}_{a}^{(n)}-\mathfrak{%
s}_{a}^{(\infty )}\Vert _{p}=0
\end{equation*}%
(Completeness of Banach--valued $\mathcal{L}^{p}$--spaces). Now, define the
sub--space 
\begin{equation*}
\mathcal{L}^{p}\left( \mathcal{A},\mathcal{X}\right) :=\left\{ \mathfrak{s}%
_{a}^{(\infty )}:\text{\textrm{there}}\;\text{\textrm{is}}\;\{\mathfrak{s}%
_{a}^{(n)}\}_{n=1}^{\infty }\;\text{\textrm{in}}\;\mathcal{S}\left( \mathcal{%
A},\mathcal{X}\right) \text{\textrm{\ with}}\;\underset{n\rightarrow \infty }%
{\lim }\Vert \mathfrak{s}_{a}^{(n)}-\mathfrak{s}_{a}^{(\infty )}\Vert
_{p}=0\right\}
\end{equation*}%
of the space of measurable functions $\mathcal{A}\rightarrow \mathcal{X}$.
Observe that the semi--norm $\left\Vert \cdot \right\Vert _{p}$ is finite on 
$\mathcal{L}^{p}\left( \mathcal{A},\mathcal{X}\right) $. In other words, $%
\mathcal{L}^{p}\left( \mathcal{A},\mathcal{X}\right) $ is the closure of $%
\mathcal{S}\left( \mathcal{A},\mathcal{X}\right) $ w.r.t. the semi-norm $%
\left\Vert \cdot \right\Vert _{p}$.

Define the linear map from $\mathcal{S}\left( \mathcal{A},\mathcal{X}\right) 
$ to $\mathcal{X}$ by 
\begin{equation}
\mathfrak{s}_{a}\mapsto \int_{\mathcal{A}}\mathfrak{s}_{a}\mathrm{d}%
\mathfrak{a}\left( a\right) :=\underset{x\in \mathfrak{s}_{a}\left( \mathcal{%
A}\right) }{\sum }x\mathfrak{a}\left( \mathfrak{s}_{a}^{-1}\left( x\right)
\right) .  \label{map1}
\end{equation}%
\ Obviously, for\ all\ $\mathfrak{s}_{a}\in \mathcal{S}\left( \mathcal{A},%
\mathcal{X}\right) $, 
\begin{equation}
\left\Vert \int_{\mathcal{A}}\mathfrak{s}_{a}\mathrm{d}\mathfrak{a}\left(
a\right) \right\Vert _{\mathcal{X}}\leq \left\Vert \mathfrak{s}%
_{a}\right\Vert _{1}.  \label{Hahn banach toto idiot}
\end{equation}%
Now, for each function $c_{a}\in L^{2}\left( \mathcal{A},\mathbb{C}\right) $%
, let us consider the linear map $\mathfrak{s}_{a}\mapsto \left\langle 
\mathfrak{s}_{a},c_{a}\right\rangle $ from $\mathcal{S}\left( \mathcal{A},%
\mathcal{X}\right) $ to $\mathcal{X}$ defined by%
\begin{equation}
\left\langle \mathfrak{s}_{a},c_{a}\right\rangle :=\underset{x\in \mathfrak{s%
}_{a}\left( \mathcal{A}\right) }{\sum }x\int_{\mathfrak{s}_{a}^{-1}\left(
x\right) }\bar{c}_{a}\mathrm{d}\mathfrak{a}\left( a\right) .  \label{map 2}
\end{equation}%
From the (finite dimensional) Cauchy--Schwarz inequality note that, for\ all 
$\mathfrak{s}_{a}\in \mathcal{S}\left( \mathcal{A},\mathcal{X}\right) $, 
\begin{equation}
\left\Vert \left\langle \mathfrak{s}_{a},c_{a}\right\rangle \right\Vert _{%
\mathcal{X}}\leq \left\Vert \mathfrak{s}_{a}\right\Vert _{2}\left\Vert
c_{a}\right\Vert _{2}.  \label{Hahn banach toto}
\end{equation}%
By using Hahn--Banach theorem and the density of $\mathcal{S}\left( \mathcal{%
A},\mathcal{X}\right) $ in $\mathcal{L}^{p}\left( \mathcal{A},\mathcal{X}%
\right) $, we obtain the existence and uniqueness of linear extensions of
the maps (\ref{map1}) and (\ref{map 2}) respectively to the spaces $\mathcal{%
L}^{1}\left( \mathcal{A},\mathcal{X}\right) $ and $\mathcal{L}^{2}\left( 
\mathcal{A},\mathcal{X}\right) $. In particular, the linear extensions of (%
\ref{map1}) and (\ref{map 2}) satisfy (\ref{Hahn banach toto idiot}) and (%
\ref{Hahn banach toto}), respectively.

\section{Compact convex sets and Choquet simplices\label{Section Compact
convex sets and Choquet simplices}}

The theory of compact convex subsets of a locally convex (topological
vector) space $\mathcal{X}$ is standard. For more details, see, e.g., \cite%
{Alfsen,Phe}. Note, however, that the definitions of topological vector
spaces found in the literature differ slightly from each other. Those
differences mostly concern the \emph{Hausdorff} property. Here, we use
Rudin's definition \cite[Section 1.6]{Rudin}:

\begin{definition}[Topological vector spaces]
\label{space}\mbox{ }\newline
\index{Topological vector space}A topological vector space $\mathcal{X}$ is
a vector space equipped with a topology $\tau $ for which the vector space
operations of $\mathcal{X}$ are continuous and every point of $\mathcal{X}$
defines a closed set.
\end{definition}

\noindent The fact that every point of $\mathcal{X}$ is a closed set is
usually not part of the definition of a topological vector space in many
textbooks. It is used here because it is satisfied in most applications --
including those of this monograph -- and, in this case, the space $\mathcal{X%
}$ is automatically Hausdorff by \cite[Theorem 1.12]{Rudin}. Examples of
topological vector spaces used in this monograph are the dual spaces (cf. 
\cite[Theorem 3.10]{Rudin}):

\begin{theorem}[Dual space of a topological vector space]
\label{thm locally convex space}\mbox{ }\newline
\index{Topological vector space!dual}The dual space $\mathcal{X}^{\ast }$ of
a (topological vector) space $\mathcal{X}$ is a locally convex space in the $%
\sigma (\mathcal{X}^{\ast },\mathcal{X})$--topology -- known as the weak$%
^{\ast }$--topology -- and its dual is $\mathcal{X}$.
\end{theorem}

Since any Banach space is a topological vector space in the sense of
Definition \ref{space}, the dual space of a Banach space is a locally convex
space:

\begin{corollary}[Dual space of a Banach space]
\label{thm locally convex spacebis}\mbox{ }\newline
The dual space $\mathcal{X}^{\ast }$ of a Banach space $\mathcal{X}$ is a
locally convex space in the $\sigma (\mathcal{X}^{\ast },\mathcal{X})$%
--topology -- known as the weak$^{\ast }$--topology -- and its dual is $%
\mathcal{X}$.
\end{corollary}

It follows that the dual spaces $\mathcal{U}^{\ast }$ and $\mathcal{W}%
_{1}^{\ast }$ respectively of the Banach spaces $\mathcal{U}$ and $\mathcal{W%
}_{1}$ (cf. Section \ref{Section fermions algebra} and Definition \ref%
{definition banach space interaction}) are both locally convex real spaces
w.r.t. the weak$^{\ast }$--topology. Note that $\mathcal{U}$ and $\mathcal{W}%
_{1}$ are separable. This property yields the metrizability of any weak$%
^{\ast }$--compact subset $K$ of their dual spaces (cf. \cite[Theorem 3.16]%
{Rudin}):

\begin{theorem}[Metrizability of weak$^{\ast }$--compact sets]
\label{Metrizability}\mbox{ }\newline
Let $K\subseteq \mathcal{X}^{\ast }$ be any weak$^{\ast }$--compact subset
of the dual $\mathcal{X}^{\ast }$ of a separable topological vector space $%
\mathcal{X}$. Then $K$ is metrizable in the weak$^{\ast }$--topology.
\end{theorem}

One important observation concerning locally convex spaces $\mathcal{X}$ is
that any compact convex subset $K\subseteq \mathcal{X}$ is the closure of
the convex hull of the (non--empty) set $\mathcal{E}(K)$ of its extreme
points%
\index{Extreme points|textbf}, i.e., of the points which cannot be written
as -- non--trivial -- convex combinations of other elements in $K$. This is
the Krein--Milman theorem (see, e.g., \cite[Theorems 3.4 (b) and 3.23]{Rudin}%
):

\begin{theorem}[Krein--Milman]
\label{theorem Krein--Millman}\mbox{ }\newline
\index{Krein--Milman theorem|textbf}Let $K\subseteq \mathcal{X}$ be any
(non--empty) compact convex subset of a locally convex space $\mathcal{X}$.
Then we have that: \newline
\emph{(i)} The set $\mathcal{E}(K)$ of its extreme points is non--empty.%
\newline
\emph{(ii)} The set $K$ is the closed convex hull of $\mathcal{E}(K)$.
\end{theorem}

\begin{remark}
$\mathcal{X}$ being a topological vector space on which its dual space $%
\mathcal{X}^{\ast }$ separates points is the only condition necessary on $%
\mathcal{X}$ in the Krein--Milman theorem. For more details, see, e.g., \cite%
[Theorem 3.23]{Rudin}.
\end{remark}

\noindent In fact, the set $\mathcal{E}(K)$ of extreme points is even a $%
G_{\delta }$ set if the compact convex set $K\subseteq \mathcal{X}$ is
metrizable. Moreover, among all subsets $Z\subseteq K$ generating $K$, $%
\mathcal{E}(K)$ is -- in a sense -- the smallest one (see, e.g., \cite[%
Proposition 1.5]{Phe}):

\begin{theorem}[Properties of the set $\mathcal{E}(K)$]
\label{lemma Milman}\mbox{ }\newline
\index{Extreme points}Let $K\subseteq \mathcal{X}$ be any (non--empty)
compact convex subset of a locally convex space $\mathcal{X}$. Then we have
that: \newline
\emph{(i)} If $K$ is metrizable then the set $\mathcal{E}(K)$ of extreme
points of $K$ forms a $G_{\delta }$ set.\newline
\emph{(ii)} If $K$ is the closed convex hull of $Z\subseteq K$ then $%
\mathcal{E}(K)$ is included in the closure of $Z$.
\end{theorem}

\noindent Property (i) can be found in \cite[Proposition 1.3]{Phe} and only
needs that $\mathcal{X}$ is a topological vector space, whereas the second
statement (ii) is a classical result obtained by Milman, see \cite[%
Proposition 1.5]{Phe}.

Theorem \ref{theorem Krein--Millman} restricted to finite dimensions is a
classical result of Minkowski%
\index{Minkowski theorem} which, for any $x\in K$ in (non--empty) compact
convex subset $K\subseteq \mathcal{X}$, states the existence of a finite
number of extreme points $%
\hat{x}_{1},\ldots ,\hat{x}_{k}\in \mathcal{E}(K)$ and positive numbers $\mu
_{1},\ldots ,\mu _{k}\geq 0$ with $\Sigma _{j=1}^{k}\mu _{j}=1$ such that 
\begin{equation}
x=\overset{k}{\sum\limits_{j=1}}\mu _{j}\hat{x}_{j}.  \label{barycenter1}
\end{equation}%
To this\ simple decomposition we can associate a probability measure, i.e.,
a \emph{normalized positive Borel regular measure}, $\mu $\ on $K$.

Indeed, the Borel sets of any set $K$ are elements of the $\sigma $--algebra 
$\mathfrak{B}$ generated by closed -- or open -- subsets of $K$. Positive
Borel regular measures are the positive countably additive set functions $%
\mu $ over $\mathfrak{B}$ satisfying 
\begin{equation*}
\mu \left( B\right) =\sup \left\{ \mu \left( C\right) :C\subseteq B,\text{ }C%
\text{ closed}\right\} =\inf \left\{ \mu \left( O\right) :B\subseteq O,\text{
}O\text{ open}\right\}
\end{equation*}%
for any Borel subset $B\in \mathfrak{B}$ of $K$. If $K$ is compact then any
positive Borel regular measure $\mu $ corresponds (one--to--one) to an
element of the set $M^{+}(K)$ of Radon measures with $\mu \left( K\right)
=\left\Vert \mu \right\Vert $ and we write 
\begin{equation}
\mu \left( h\right) =\int_{K}\mathrm{d}\mu (\hat{x})\;h\left( \hat{x}\right)
\label{barycenter1bis}
\end{equation}%
for any continuous function $h$ on $K$. A probability measure%
\index{Probability measure} $\mu \in M_{1}^{+}(K)$ is per definition a
positive Borel regular measure $\mu \in M^{+}(K)$ which is \emph{normalized}%
: $\left\Vert \mu \right\Vert =1$.

\begin{remark}
The set $M_{1}^{+}(K)$ of probability measures on $K$ can also be seen as
the set of states on the commutative $C^{\ast }$--algebra $C(K)$ of
continuous functionals on the compact set $K$, by\ the Riesz--Markov theorem.
\end{remark}

Therefore, using the probability measure $\mu _{x}\in M_{1}^{+}(K)$\ on $K$
defined by 
\begin{equation*}
\mu _{x}=\overset{k}{\sum\limits_{j=1}}\mu _{j}\delta _{%
\hat{x}_{j}}
\end{equation*}%
with $\delta _{y}$ being the Dirac -- or point -- mass\footnote{$\delta _{y}$
is the Borel measure such that for any Borel subset $B\in \mathfrak{B}$ of $%
K $, $\delta _{y}(B)=1$ if $y\in B$ and $\delta _{y}(B)=0$ if $y\notin B$.}
at $y$, Equation (\ref{barycenter1}) can be seen as an integral defined by (%
\ref{barycenter1bis}) for the probability measure $\mu _{x}\in M_{1}^{+}(K)$%
: 
\begin{equation}
x=\int_{K}\mathrm{d}\mu _{x}(\hat{x})\;\hat{x}\ .  \label{barycenter2}
\end{equation}%
The point $x$ is in fact the \emph{barycenter} of the probability measure $%
\mu _{x}$. This notion is defined in the general case as follows (cf. \cite[%
Eq. (2.7) in Chapter I]{Alfsen} or \cite[p. 1]{Phe}):

\begin{definition}[Barycenters of probability measures in convex sets]
\label{def barycenter}%
\index{Barycenters|textbf}Let $K\subseteq \mathcal{X}$ be any (non--empty)
compact convex subset of a locally convex space $\mathcal{X}$ and let $\mu
\in M_{1}^{+}(K)$ be a probability measure on $K$. We say that $x\in K$ is
the barycenter\footnote{%
Other terminology existing in the literature: \textquotedblleft $x$ is
represented by $\mu $\textquotedblright , \textquotedblleft $x$ is the
resultant of $\mu $\textquotedblright .} of $\mu $ if, for all continuous
linear\footnote{%
Barycenters can also be defined in the same way via affine functionals
instead of linear functionals, see \cite[Proposition 4.1.1.]%
{BrattelliRobinsonI}.} functionals $h$ on $\mathcal{X}$, 
\begin{equation*}
h\left( x\right) =\int_{K}\mathrm{d}\mu (%
\hat{x})\;h\left( \hat{x}\right) .
\end{equation*}
\end{definition}

\noindent Barycenters are well--defined for \emph{all} probability measures
in convex compact subsets of locally convex spaces (cf. \cite[Propositions
1.1 and 1.2]{Phe}):

\begin{theorem}[Well-definiteness and uniqueness of barycenters]
\label{thm barycenter}\mbox{ }\newline
Let $K\subseteq \mathcal{X}$ be any (non--empty) compact subset of a locally
convex space $\mathcal{X}$ such that $\overline{\mathrm{co}\left( K\right) }$
is also compact. Then we have that: \newline
\emph{(i)} For any probability measure $\mu \in M_{1}^{+}(K)$ on $K$,$\ $%
there is a unique barycenter $x_{\mu }\in \overline{\mathrm{co}\left(
K\right) }$. In particular, if $K$ is convex then, for any $\mu \in
M_{1}^{+}(K)$, there is a unique barycenter $x_{\mu }\in K$. Moreover, the
map $\mu \mapsto x_{\mu }$ from $M_{1}^{+}(K)$ to $\overline{\mathrm{co}%
\left( K\right) }$ is affine and weak$^{\ast }$--continuous.\newline
\emph{(ii)} Conversely, for any $x\in \overline{\mathrm{co}\left( K\right) }$%
, there is a probability measure $\mu _{x}\in M_{1}^{+}(K)$ on $K$ with
barycenter $x$.
\end{theorem}

Therefore, we write the barycenter $x_{\mu }$ of any probability measure $%
\mu $ in $K$ as%
\begin{equation*}
x_{\mu }=\int_{K}\mathrm{d}\mu (\hat{x})\;\hat{x},
\end{equation*}%
where the integral has to be understood in the weak sense. By Definition \ref%
{def barycenter}, it means that $h\left( x_{\mu }\right) $\ can be
decomposed by the probability measure $\mu \in M_{1}^{+}(K)$ provided $h$ is
a continuous linear functional. In fact, this last property can also be
extended to all affine upper semi--continuous functionals on $K$, see, e.g., 
\cite[Corollary 4.1.18.]{BrattelliRobinsonI} together with \cite[Theorem 1.12%
]{Rudin}:

\begin{lemma}[Barycenters and affine maps]
\label{Corollary 4.1.18.}\mbox{ }\newline
Let $K\subseteq \mathcal{X}$ be any (non--empty) compact convex subset of a
locally convex space $\mathcal{X}$. Then, for any probability measure $\mu
\in M_{1}^{+}(K)$ on $K$ with barycenter $x_{\mu }\in K$ and for any affine
upper semi--continuous functional $h$ on $K$, 
\begin{equation*}
h\left( x_{\mu }\right) =\int_{K}\mathrm{d}\mu \left( \hat{x}\right)
\;h\left( \hat{x}\right) .
\end{equation*}
\end{lemma}

It is natural to ask whether, for any $x\in K$ in a convex set $K$, there is
a (possibly not unique) probability measure $\mu _{x}$ on $K$ supported on $%
\mathcal{E}(K)$ with barycenter $x$. Equation (\ref{barycenter2}) already
gives a first positive answer to that problem in the finite dimensional
case. The general case has been proven by Choquet, whose theorem is a
remarkable refinement of the Krein--Milman theorem (see, e.g., \cite[p. 14]%
{Phe}):

\begin{theorem}[Choquet]
\label{theorem choquet bis}\mbox{ }\newline
\index{Choquet theorem|textbf}Let $K\subseteq \mathcal{X}$ be any
(non--empty) metrizable compact convex subset of a locally convex space $%
\mathcal{X}$. Then, for any $x\in K$, there is a probability measure $\mu
_{x}\in M_{1}^{+}(K)$ on $K$ such that%
\begin{equation*}
\mu _{x}(\mathcal{E}(K))=1%
\text{\quad and\quad }x=\int_{K}\mathrm{d}\mu _{x}(\hat{x})\;\hat{x}.
\end{equation*}%
Recall that the integral above means that $x\in K$ is the barycenter of $\mu
_{x}$.
\end{theorem}

\begin{remark}[Choquet theorem and affine maps]
\label{remark choquet bis copy(1)}\mbox{ }\newline
By Lemma \ref{Corollary 4.1.18.}, the Choquet theorem can be used to
decompose any affine upper semi--continuous functional defined on the
metrizable compact convex subset $K\subseteq \mathcal{X}$ w.r.t. extreme
points of $K$.
\end{remark}

\begin{remark}[Choquet theorem for non--metrizable $K$]
\label{remark choquet bis}\mbox{ }\newline
If the (non--empty) compact convex subset $K\subseteq \mathcal{X}$ is not
metrizable then $\mathcal{E}(K)$ may not form a Borel set. The Choquet
theorem (Theorem \ref{theorem choquet bis}) stays, however, valid under the
modification that $\mu _{x}$ is pseudo--supported by $\mathcal{E}(K)$ which
means that $\mu _{x}(B)=1$ for all Baire sets $B\supseteq \mathcal{E}(K)$.
This result is known as the the Choquet--Bishop--de Leeuw theorem%
\index{Choquet--Bishop--de Leeuw theorem}, see \cite[p. 17]{Phe}.
\end{remark}

\noindent Note that the probability measure $\mu _{x}$ of Theorem \ref%
{theorem choquet bis} is a priori not unique. For instance, in the
2--dimensional plane, simplices (points, segments, and triangles) are
uniquely decomposed in terms of their extreme points, i.e., they are
uniquely represented by a convex combination of extreme points. But this
decomposition is not anymore unique for a square. In fact, uniqueness of the
decomposition given in Theorem \ref{theorem choquet bis} is related to the
theory of simplices.

To define them in the general case, let $\mathfrak{S}$ be a compact convex
set of a locally convex real space $\mathcal{X}$. Without loss of generality
assume that the compact convex set $\mathfrak{S}$ is included in a closed
hyper--plane which does not contain the origin\footnote{%
Otherwise, we embed $\mathcal{X}$ as $\mathcal{X}\times \{1\}$ in $\mathcal{X%
}\times \mathbb{R}$.}. Let 
\begin{equation*}
\mathfrak{K}:=\left\{ \alpha x:\alpha \geq 0,\ x\in \mathfrak{S}\right\}
\end{equation*}%
be the\ cone with base $\mathfrak{S}$. Recall that the cone $\mathfrak{K}$
induces a partial ordering on $\mathcal{X}$ by using the definition $x\geqq
y $ iff $x-y\in \mathfrak{K}$. A \emph{least upper bound} for $x$ and $y$ is
an element $x\vee y\geqq x,y$ satisfying $w\geqq x\vee y$ for all $w$ with $%
w\geqq x,y$. Then a simplex is defined as follows:

\begin{definition}[Simplices]
\label{gamm regularisation copy(1)}\mbox{ }\newline
\index{Simplex|textbf}The (non--empty) compact convex set $\mathfrak{S}$ is
a simplex whenever $\mathfrak{K}$ is a lattice with respect to the partial
ordering $\geqq $. This means that each pair $x,y\in \mathfrak{K}$ has a
least upper bound $x\vee y\in \mathfrak{K}$.
\end{definition}

\noindent Observe that a simplex can also be defined for non--compact convex
sets but we are only interested here in compact simplices. Such simplices
are particular examples of \emph{simplexoids}%
\index{Simplexoid}, i.e., compact convex sets whose closed proper faces are
simplices. Recall that, here, a face $F$ of a convex set $K$ is defined to
be a subset of $K$ with the property that if $\rho =\lambda _{1}\rho
_{1}+\cdots +\lambda _{n}\rho _{n}\in F$ with $\rho _{1},\ldots ,\rho
_{n}\in K$, $\lambda _{1},\ldots ,\lambda _{n}\in (0,1)$ and $\lambda
_{1}+\cdots +\lambda _{n}=1$ then $\rho _{1},\ldots ,\rho _{n}\in F$.

The definition of simplices above agrees with the usual definition in finite
dimensions as the $n$--dimensional simplex $\{(\lambda _{1},\lambda
_{2},\cdots ,\lambda _{n+1}),\Sigma _{j}\lambda _{j}=1\}$ is the base of the 
$(n+1)$--dimensional cone $\{(\lambda _{1},\lambda _{2},\cdots ,\lambda
_{n+1}),\lambda _{j}\geq 0\}$. In fact, for all metrizable simplices, the
probability measure $\mu _{x}$ of Theorem \ref{theorem choquet bis} is
unique and conversely, if $\mu _{x}$\ is always uniquely defined then the
corresponding metrizable compact convex set is a simplex (see, e.g., \cite[%
p. 60]{Phe}):

\begin{theorem}[Choquet]
\label{theorem choquet bis copy(1)}\mbox{ }\newline
\index{Simplex}Let $\mathfrak{S}\subseteq \mathcal{X}$ be any (non--empty)
closed convex metrizable subset of a locally convex space $\mathcal{X}$.
Then $\mathfrak{S}$ is a simplex iff, for any $x\in \mathfrak{S}$, there is
a unique probability measure $\mu _{x}\in M_{1}^{+}(\mathfrak{S})$ on $%
\mathfrak{S}$ such that%
\begin{equation*}
\mu _{x}(\mathcal{E}(\mathfrak{S}))=1%
\text{\quad and\quad }x=\int_{\mathfrak{S}}\mathrm{d}\mu _{x}(\hat{x})\;\hat{%
x}.
\end{equation*}
\end{theorem}

\noindent Compact and metrizable convex sets for which the integral
representation in Theorem \ref{theorem choquet bis copy(1)} is unique are
also called \emph{Choquet simplices}:

\begin{definition}[Choquet simplex]
\label{gamm regularisation copy(3)}\mbox{ }\newline
\index{Simplex!Choquet|textbf}A metrizable simplex $\mathfrak{S}$ is a
Choquet simplex whenever the decomposition of $\mathfrak{S}$ on $\mathcal{E}(%
\mathfrak{S})$ given by Theorem \ref{theorem choquet bis} is unique. A
Choquet simplex can also be defined when $\mathfrak{S}$ is not metrizable,
using Remark \ref{remark choquet bis}.
\end{definition}

\noindent In this monograph we are only interested in metrizable compact
convex set on which Theorem \ref{theorem choquet bis copy(1)} is applied.
Therefore, all our examples of simplices are in fact Choquet simplices.

Two further special types of simplices are of particular importance: The 
\emph{Bauer }and the \emph{Poulsen} simplices. The first one is defined as
follows:

\begin{definition}[Bauer simplex]
\label{gamm regularisation copy(2)}\mbox{ }\newline
\index{Simplex!Bauer|textbf}The simplex $\mathfrak{S}$ is a Bauer simplex
whenever its set $\mathcal{E}(\mathfrak{S})$ of extreme points is closed.
\end{definition}

\noindent A compact Bauer simplex $\mathfrak{S}$ has the interesting
property that it is affinely homeomorphic to the set of states on the
commutative $C^{\ast }$--algebra $C(\mathcal{E}(\mathfrak{S}))$ (see, e.g., 
\cite[Corollary II.4.2]{Alfsen}):

\begin{theorem}[Bauer]
\label{theorem Bauer}\mbox{ }\newline
\index{Simplex!Bauer}Let $\mathfrak{S}\subseteq \mathcal{X}$ be any compact
metrizable Bauer Simplex of a locally convex space $\mathcal{X}$. Then the
map $x\mapsto \mu _{x}$ defined by Theorem \ref{theorem choquet bis copy(1)}
from $\mathfrak{S}$ to the set $M_{1}^{+}(\mathcal{E}(\mathfrak{S}))$ of
probability measures\footnote{%
I.e. the set of states on the commutative $C^{\ast }$--algebra $C(\mathcal{E}%
(S))$ of continuous functionals on the compact set $\mathcal{E}(S)$.} on $%
\mathcal{E}(\mathfrak{S})$ is an affine homeomorphism.
\end{theorem}

Bauer simplices are special simplices as the set of $\mathcal{E}(\mathfrak{S}%
)$ of extreme points of a simplex $\mathfrak{S}$ may not be closed. In fact,
E. T. Poulsen \cite{Poulsen} constructed in 1961 an example of a metrizable
simplex $\mathfrak{S}$ with $\mathcal{E}(\mathfrak{S})$ being dense in $%
\mathfrak{S}$. This simplex is now well--known as the Poulsen simplex
because it is unique \cite[Theorem 2.3.]{Lindenstrauss-etal} up to an affine
homeomorphism:

\begin{theorem}[Lindenstrauss--Olsen--Sternfeld]
\label{theorem Bauer copy(1)}\mbox{ }\newline
\index{Simplex!Poulsen}Every (non--empty) compact metrizable simplex $%
\mathfrak{S}$ with $\mathcal{E}(\mathfrak{S})$ being dense in $\mathfrak{S}$
is affinely homeomorphic to the Poulsen simplex.
\end{theorem}

\noindent The original example given by Poulsen \cite{Poulsen} is not
explained here as we give in Section \ref{section set of states} a prototype
of the Poulsen simplex: The set $E_{%
\vec{\ell}}\subseteq \mathcal{U}^{\ast }$ of all $\mathbb{Z}_{\vec{\ell}%
}^{d} $--invariant states defined by (\ref{periodic invariant states}) for
any $\vec{\ell}\in \mathbb{N}^{d}$, see Theorem \ref{Thm Poulsen simplex}.

For more details on the Poulsen simplex we recommend \cite%
{Lindenstrauss-etal} where its specific properties are described. They also
show that the Poulsen simplex is, in a sense, complementary to the Bauer
simplices, see \cite[p. 164]{Alfsen} or \cite[Section 5]{Lindenstrauss-etal}.

\section{$\Gamma $--regularization of real functionals\label{Section gamma
regularization}%
\index{Gamma--regularization|textbf}}

The $\Gamma $--regularization of real functionals on a subset $K\subseteq 
\mathcal{X}$ is defined from the space $\mathrm{A}\left( \mathcal{X}\right) $
of all affine continuous real valued functionals on $\mathcal{X}$ as follows
(cf. \cite[Eq. (1.3) in Chapter I]{Alfsen} or \cite[Definition 2.1.1]%
{Schirotzek}):

\begin{definition}[$\Gamma $--regularization of real functionals]
\label{gamm regularisation}\mbox{ }\newline
For any real functional $h$ defined from a locally convex space $\mathcal{X}$
to $\left( -\infty ,\infty \right] $, its $\Gamma $--regularization $\Gamma
_{K}\left( h\right) $ on a subset $K\subseteq \mathcal{X}$ is the functional
defined as the supremum over all affine and continuous minorants\ from $%
\mathcal{X}$ to $\mathbb{R}$ of $h|_{K}$, i.e., for all $x\in \mathcal{X}$, 
\begin{equation*}
\Gamma _{K}\left( h\right) \left( x\right) :=\sup \left\{ m(x):m\in \mathrm{A%
}\left( \mathcal{X}\right) \;%
\text{and }m|_{K}\leq h|_{K}\right\} .
\end{equation*}
\end{definition}

\noindent If a functional $h$ is only defined on a subset $K\subseteq 
\mathcal{X}$ of a locally convex space $\mathcal{X}$ then we compute $\Gamma
_{\mathcal{X}}\left( h\right) $ by extending $h$ to the locally convex space 
$\mathcal{X}$ as follows:

\begin{definition}[Extension of functionals on a locally convex space $%
\mathcal{X}$]
\label{extension of functional}Any functional $h:K\subseteq \mathcal{X}%
\rightarrow \left( -\infty ,\infty \right] $ is seen as a map from $\mathcal{%
X}$ to $(-\infty ,\infty ]$ by the definition%
\begin{equation*}
h(x):=\left\{ 
\begin{array}{lcl}
h(x) & , & \mathrm{for\ }x\in K. \\ 
+\infty & , & \mathrm{for\ }x\in \mathcal{X}\backslash K.%
\end{array}%
\right.
\end{equation*}
\end{definition}

\noindent If $h$ is convex and lower semi--continuous on the closed and
convex subset $K\subseteq \mathcal{X}$ then its extension on $\mathcal{X}$
is also convex and lower semi--continuous. Moreover, in this case, $\Gamma _{%
\mathcal{X}}\left( h\right) =\Gamma _{K}\left( h\right) $ on $\mathcal{X}$.

Since the $\Gamma $--regularization $\Gamma _{K}\left( h\right) $ of a real
functional $h$ is a supremum of continuous functionals, $\Gamma _{K}\left(
h\right) $ is a convex and lower semi--continuous functional on $\mathcal{X}$%
. In fact, every convex and lower semi--continuous functional on $K$ equals
its $\Gamma $--regularization on $K$ (see, e.g., \cite[Proposition I.1.2.]%
{Alfsen} or \cite[Proposition 2.1.2]{Schirotzek}):

\begin{proposition}[$\Gamma $--regularization of lower semi--cont. convex
maps]
\label{lemma gamma regularisation}Let $h$ be any functional from a
(non--empty) closed convex subset $K\subseteq \mathcal{X}$ of a locally
convex space $\mathcal{X}$ to $\left( -\infty ,\infty \right] $. Then the
following statements are equivalent:\newline
\emph{(i)} $\Gamma _{K}\left( h\right) =h$ on $K$.\newline
\emph{(ii)} $h$ is a lower semi--continuous convex functional on $K$.
\end{proposition}

This proposition is a standard result which can directly be proven without
using the fact that the $\Gamma $--regularization $\Gamma _{K}\left(
h\right) $ of a functional $h$ on $K$ equals its twofold \emph{%
Legendre--Fenchel transform} -- 
\index{Legendre--Fenchel transform!biconjugate}%
\index{Legendre--Fenchel transform!Gamma--regularization}%
\index{Gamma--regularization!Legendre--Fenchel transform}also called the 
\emph{biconjugate }(functional) of $h$. Indeed, $\Gamma _{K}\left( h\right) $
is the largest lower semi--continuous and convex minorant of $h$:

\begin{corollary}[Largest lower semi--continuous convex minorant of $h$]
\label{Biconjugate}Let $h$ be any functional from a (non--empty) closed
convex subset $K\subseteq \mathcal{X}$ of a locally convex space $\mathcal{X}
$ to $\left( -\infty ,\infty \right] $. Then its $\Gamma $--regularization $%
\Gamma _{K}\left( h\right) $ is its largest lower semi--continuous and
convex minorant on $K$.
\end{corollary}

\begin{proof}%
For any lower semi--continuous convex functional $f$ satisfying $f\leq h$ on 
$K$, we have, by Proposition \ref{lemma gamma regularisation}, that 
\begin{equation*}
f\left( x\right) =\sup \left\{ m(x):m\in \mathrm{A}\left( \mathcal{X}\right)
\;%
\text{and }m|_{K}\leq f|_{K}\leq h|_{K}\right\} \leq \Gamma _{K}\left(
h\right) \left( x\right)
\end{equation*}%
for any $x\in K$. 
\end{proof}%

In particular, if $(\mathcal{X},\mathcal{X}^{\ast })$ is a dual pair and $h$
is any functional from $\mathcal{X}$ to $(-\infty ,\infty ]$ then $\Gamma _{%
\mathcal{X}}\left( h\right) =h^{\ast \ast }$, by using Theorem \ref{theorem
fenchel moreau} together with Corollary \ref{Biconjugate}. See Corollary \ref%
{theorem fenchel moreaubis}.

Proposition \ref{lemma gamma regularisation} has further interesting
consequences. The first one we would like to mention is an extension of the
Bauer maximum principle \cite[Lemma 4.1.12]{BrattelliRobinsonI} (or \cite[%
Theorem I.5.3.]{Alfsen}), that is:

\begin{lemma}[Bauer maximum principle]
\label{Bauer maximum principle}\mbox{ }\newline
\index{Bauer maximum principle|textbf}Let $\mathcal{X}$ be a topological
vector space. An upper semi--continuous convex real functional $h$ over a
(non-empty)\ compact convex subset $K\subseteq \mathcal{X}$ attains its
maximum at an extreme point of $K$, i.e., 
\begin{equation*}
\underset{x\in K}{\sup }h\left( x\right) =\underset{%
\hat{x}\in \mathcal{E}(K)}{\max }h\left( \hat{x}\right) .
\end{equation*}%
Here, $\mathcal{E}(K)$ is the (non--empty) set of extreme points of $K$, cf.
Theorem \ref{theorem Krein--Millman}.
\end{lemma}

\noindent Indeed, by combining Proposition \ref{lemma gamma regularisation}
with Lemma \ref{Bauer maximum principle} it is straightforward to check the
following statement which does not seem to have been observed before:

\begin{lemma}[Extension of the Bauer maximum principle]
\label{Bauer maximum principle bis}\mbox{ }\newline
\index{Bauer maximum principle!Extension|textbf}Let $h_{\pm }$ be two convex
real functionals from a locally convex space $\mathcal{X}$ to $\left(
-\infty ,\infty \right] $ such that $h_{-}$ and $h_{+}$ are, respectively,
lower and upper semi--continuous. Then the supremum of the sum $%
h:=h_{-}+h_{+}$ over a (non-empty) compact convex subset $K\subseteq 
\mathcal{X}$ can be reduced to the (non--empty) set $\mathcal{E}(K)$ of
extreme points of $K$, i.e., 
\begin{equation*}
\underset{x\in K}{\sup }h\left( x\right) =\underset{%
\hat{x}\in \mathcal{E}(K)}{\sup }h\left( \hat{x}\right) .
\end{equation*}
\end{lemma}

\begin{proof}%
We first use Proposition \ref{lemma gamma regularisation} in order to write $%
h_{-}=\Gamma _{K}\left( h_{-}\right) $ as a supremum over affine and
continuous functionals. Then we commute this supremum with the one over $K$
and apply the Bauer maximum principle to obtain that%
\begin{equation*}
\underset{x\in K}{\sup }h\left( x\right) =\sup \left\{ \underset{\hat{x}\in 
\mathcal{E}(K)}{\sup }\left\{ m(\hat{x})+h_{+}\left( \hat{x}\right) \right\}
:m\in \mathrm{A}\left( \mathcal{X}\right) \;\text{and }m|_{K}\leq
h_{-}|_{K}\right\} .
\end{equation*}%
The lemma follows by commuting once again both suprema and by using $%
h_{-}=\Gamma _{K}\left( h_{-}\right) $. 
\end{proof}%

Observe, however, that, under the conditions of the lemma above the supremum
of $h=h_{-}+h_{+}$ is, in general, not attained on $\mathcal{E}(K)$.

Another consequence of Proposition \ref{lemma gamma regularisation} is
Jensen's inequality for convex lower semi-continuous real functionals on a
compact convex sets $K$.

\begin{lemma}[Jensen's inequality on compact convex sets]
\label{Jensen inequality}\mbox{ }\newline
\index{Jensen's inequality|textbf}Let $\mathcal{X}$ be a locally convex
space, $h$ be any lower semi--continuous convex real functional over a
(non-empty) compact convex subset $K\subseteq \mathcal{X}$ and $\mu _{x}\in
M_{1}^{+}(K)$ be any probability measure with barycenter $x\in K$
(Definition \ref{def barycenter}). Assume the existence of some positive and 
$\mu _{x}$--integrable upper bound $\mathfrak{h}$ for $h$, i.e., some
measurable functional $\mathfrak{h}$ from $K$ to $\mathbb{R}_{0}^{+}$
satisfying 
\begin{equation*}
\int_{K}\mathrm{d}\mu _{x}(%
\hat{x})\;\mathfrak{h}(\hat{x})<\infty \quad \text{and}\quad h\leq \mathfrak{%
h}\quad \mu _{x}\text{--a.e. on }K\mathrm{.}
\end{equation*}%
Then 
\begin{equation*}
h\left( x\right) \leq \int_{K}\mathrm{d}\mu _{x}(\hat{x})\;h(\hat{x}).
\end{equation*}
\end{lemma}

\noindent Jensen's inequality is of course a well--known result stated in
various situations including functionals taking value in a topological
vector space. A simple proof of this lemma using Proposition \ref{lemma
gamma regularisation} is given by \cite[Proposition I.2.2.]{Alfsen}. We give
it for completeness as it is rather short.

\begin{proof}
As $h$ is convex and lower semi--continuous, by Proposition \ref{lemma gamma
regularisation}, 
\begin{equation*}
h(x)=\sup \left\{ m(x):m\in \mathrm{A}\left( \mathcal{X}\right) \;\text{and }%
m|_{K}\leq h|_{K}\right\}
\end{equation*}%
for any $x\in K$. We further observe that, for any affine continuous real
functional $m$ and any probability measure $\mu _{x}$ with barycenter $x\in
K $, 
\begin{equation*}
m(x)=\int_{K}\mathrm{d}\mu _{x}(\hat{x})\,m(\hat{x}),
\end{equation*}%
see Lemma \ref{Corollary 4.1.18.}. Thus 
\begin{equation}
h(x)=\sup \left\{ \int_{K}\mathrm{d}\mu _{x}(\hat{x})\,m(\hat{x}):m\in 
\mathrm{A}\left( \mathcal{X}\right) \;\text{and }m|_{K}\leq h|_{K}\right\} .
\label{jensen1}
\end{equation}%
Since there is a positive and $\mu _{x}$--integrable upper bound $\mathfrak{h%
}$ for $h$, we have that 
\begin{equation*}
\int_{K}\mathrm{d}\mu _{x}(\hat{x})\,\max \left\{ h(\hat{x}),0\right\}
<\infty .
\end{equation*}%
Hence, by (\ref{jensen1}) together with the monotonicity of integrals, 
\begin{equation*}
h(x)\leq \int_{K}\mathrm{d}\mu _{x}(\hat{x})\,h(\hat{x})<\infty .
\end{equation*}%
\end{proof}%

We give now an interesting property concerning the $\Gamma $--regularization
of real functionals in relation with compact convex sets (cf. \cite[%
Corollary I.3.6.]{Alfsen}):

\begin{theorem}[$\Gamma $--regularization of continuous maps]
\label{Thm - Corollary I.3.6}\mbox{ }\newline
Let $K\subseteq \mathcal{X}$ be any (non--empty) compact convex subset of a
locally convex space $\mathcal{X}$ and $h:K\rightarrow \left( -\infty
,\infty \right] $ be a continuous real functional. Then, for any $x\in K$,
there is a probability measure $\mu _{x}\in M_{1}^{+}(K)$\ on $K$ with
barycenter $x$ such that 
\begin{equation*}
\Gamma _{K}\left( h\right) \left( x\right) =\int_{K}\mathrm{d}\mu _{x}(\hat{x%
})\;h\left( \hat{x}\right) .
\end{equation*}
\end{theorem}

\noindent This theorem is a useful result to study variational problems --
at least the ones appearing in this monograph. Indeed, if $h$ is a
continuous functional from a compact convex set $K$ to $[\mathrm{k},\infty ]$
with $\mathrm{k}\in \mathbb{R}$ then extreme points of the compact set of
minimizers of $\Gamma _{K}\left( h\right) $ on $K$ are minimizers of $h$.
This can be seen -- in a more general setting -- as follows.

Let $K$ be a compact convex subset of a locally convex space $\mathcal{X}$
and $h:K\rightarrow \left( -\infty ,\infty \right] $ be any real functional.
Then $\{x_{i}\}_{i\in I}\subseteq K$ is -- by definition -- a net of \emph{%
approximating minimizers}%
\index{Minimizers!approximating|textbf} when 
\begin{equation*}
\underset{I}{\lim }\ h(x_{i})=\inf\limits_{x\in K}\,h(x).
\end{equation*}%
Note that nets $\{x_{i}\}_{i\in I}\subseteq K$ converges along a subnet as $%
K $ is compact. Then we define the set of generalized minimizers of $h$ as
follows:

\begin{definition}[Set of generalized minimizers]
\label{gamm regularisation copy(4)}\mbox{ }\newline
\index{Minimizers!generalized|textbf}Let $K$ be a (non--empty) compact
convex subset of a locally convex space $\mathcal{X}$ and $h:K\rightarrow
\left( -\infty ,\infty \right] $ be any real functional. Then the set $%
\mathit{\Omega }\left( h,K\right) \subseteq K$ of generalized minimizers of $%
h$ is the (non--empty) set 
\begin{equation*}
\mathit{\Omega }\left( h,K\right) :=%
\Big\{%
y\in K:\exists \{x_{i}\}_{i\in I}\subseteq K\mathrm{\ converging\ to}\ y%
\mathrm{\ with\ }\underset{I}{\lim }\ h(x_{i})=\inf\limits_{K}\,h%
\Big\}%
\end{equation*}%
of all limit points of approximating minimizers of $h$.
\end{definition}

\noindent Note that the non--empty set $\mathit{\Omega }\left( h,K\right) $
is compact when $K$ is metrizable:

\begin{lemma}[Properties of the set $\mathit{\Omega }\left( h,K\right) $]
\label{theorem trivial sympa 0}\mbox{ }\newline
Let $K$ be a compact, convex, and metrizable subset of a locally convex
space $\mathcal{X}$ and $h:K\rightarrow \left( -\infty ,\infty \right] $ be
any real functional. Then the set $\mathit{\Omega }\left( h,K\right) $ of
generalized minimizers of $h$ over $K$ is compact.
\end{lemma}

\begin{proof}%
Since $K$ is compact, $\mathit{\Omega }\left( h,K\right) \subseteq K$ is
compact if it is a closed set. Because it is metrizable, $K$ is sequentially
compact and we can restrict ourself on sequences instead of more general
nets. Then the lemma can easily be proven by using any metric $d_{K}(x,y)$
on $K$ generating the topology. Indeed, for any sequence $%
\{y_{n}\}_{n=1}^{\infty }\subseteq \mathit{\Omega }\left( h,K\right) $ of
generalized minimizers converging to $y$, there is, by Definition \ref{gamm
regularisation copy(4)}, a sequence $\{x_{n,m}\}_{n,m=1}^{\infty }\subseteq
K $ of approximating minimizers converging, for any $n\in \mathbb{N}$, to $%
y_{n}\in \mathit{\Omega }\left( h,K\right) $ as $m\rightarrow \infty $. In
particular, for all $n\in \mathbb{N}$, there exists $N_{n}>0$ such that, for
all $m>N_{n}$,%
\begin{equation*}
d_{K}(x_{n,m},y)\leq 2^{-n}+d_{K}(y_{n},y)\mathrm{\quad and\quad }%
|h(x_{n,m})-\inf\limits_{K}\,h|\leq 2^{-n}.
\end{equation*}%
By taking any function $p(n)\in \mathbb{N}$ satisfying $p(n)>N_{n}$ and
converging to $\infty $ as $n\rightarrow \infty $ we obtain that $%
\{x_{n,p(n)}\}_{n=1}^{\infty }$ is a sequence of approximating minimizers
converging to $y$ as $n\rightarrow \infty $. In other words, $y\in \mathit{%
\Omega }\left( h,K\right) $. 
\end{proof}%

Now, we are in position to give a useful theorem on the minimization of real
functionals:

\begin{theorem}[Minimization of real functionals -- I]
\label{theorem trivial sympa 1}\mbox{ }\newline
\index{Minimizers}%
\index{Minimization of real functionals|textbf}Let $K$ be any (non--empty)
compact convex subset of a locally convex space $\mathcal{X}$ and $%
h:K\rightarrow \lbrack \mathrm{k},\infty ]$ be any real functional with $%
\mathrm{k}\in \mathbb{R}$. Then we have that:\newline
\emph{(i)} 
\begin{equation*}
\inf \,h\left( K\right) =\inf \,\Gamma _{K}\left( h\right) \left( K\right) .
\end{equation*}%
\emph{(ii) }The set $\mathit{M}$ of minimizers of$\ \Gamma _{K}\left(
h\right) $ over $K$ equals the closed convex hull of the set $\mathit{\Omega 
}\left( h,K\right) $ of generalized minimizers of $h$ over $K$, i.e., 
\begin{equation*}
\mathit{M}=%
\overline{\mathrm{co}\left( \mathit{\Omega }\left( h,K\right) \right) }.
\end{equation*}
\end{theorem}

\begin{proof}%
The assertion (i) is a standard result. Indeed, by Definition \ref{gamm
regularisation}, $\Gamma _{K}\left( h\right) \leq h$ on $K$ and thus%
\begin{equation*}
\inf \,\Gamma _{K}\left( h\right) \left( K\right) \leq \inf \,h\left(
K\right) .
\end{equation*}%
The converse inequality is derived by restricting the supremum in Definition %
\ref{gamm regularisation} to constant maps $m$ from $K$ to $\mathbb{R}$ with 
$\mathrm{k}\leq m\leq h$.

By Definition \ref{gamm regularisation}, we also observe that $\Gamma
_{K}\left( h\right) $ is a lower semi--continuous functional. This implies
that the variational problem $\inf \,\Gamma _{K}\left( h\right) (K)$ has
minimizers and the set $\mathit{M}=\mathit{\Omega }\left( \Gamma _{K}\left(
h\right) ,K\right) $ of all minimizers of $\Gamma _{K}\left( h\right) $ is
compact. Moreover, again by Definition \ref{gamm regularisation}, the
functional $\Gamma _{K}\left( h\right) $ is convex which obviously yields
the convexity of the set $\mathit{M}$.

For any $y\in \mathit{\Omega }\left( h,K\right) $, there is a net $\left\{
x_{i}\right\} _{i\in I}\subseteq K$ of approximating minimizers of $h$ on $K$
converging to $y$. In particular, since the functional $\Gamma _{K}\left(
h\right) $ is lower semi--continuous and $\Gamma _{K}\left( h\right) \leq h$
on $K$, we have that 
\begin{equation*}
\Gamma _{K}\left( h\right) (y)\leq \underset{I}{\liminf }\,\Gamma _{K}\left(
h\right) (x_{i})\leq \underset{I}{\lim }\,h(x_{i})=\inf \,h(K)=\inf \,\Gamma
_{K}\left( h\right) (K),
\end{equation*}%
i.e., $y\in \mathit{M}$. As $\mathit{M}$ is convex and compact we obtain
that 
\begin{equation}
\mathit{M}\supseteq \overline{\mathrm{co}\left( \mathit{\Omega }\left(
h,K\right) \right) }.  \label{inclusion1}
\end{equation}%
So, we prove now the converse inclusion. We can assume without loss of
generality that $\overline{\mathrm{co}\left( \mathit{\Omega }\left(
h,K\right) \right) }\neq K$ since there is otherwise nothing to prove. We
show\ next that, for any $x\in K\backslash \overline{\mathrm{co}\left( 
\mathit{\Omega }\left( h,K\right) \right) }$, we have $x\notin $ $\mathit{M}$%
.

As $\overline{\mathrm{co}\left( \mathit{\Omega }\left( h,K\right) \right) }$
is a closed set of a locally convex space $\mathcal{X}$, for any $x\in
K\backslash \overline{\mathrm{co}\left( \mathit{\Omega }\left( h,K\right)
\right) }$, there is an open and convex neighborhood $\mathcal{V}%
_{x}\subseteq $ $\mathcal{X}$ of $\{0\}\subseteq \mathcal{X}$ which is
symmetric, i.e., $\mathcal{V}_{x}=-\mathcal{V}_{x}$, and which satisfies 
\begin{equation*}
\mathcal{G}_{x}\cap \left[ \{x\}+\mathcal{V}_{x}\right] =\emptyset
\end{equation*}%
with%
\begin{equation*}
\mathcal{G}_{x}:=K\cap \left[ \overline{\mathrm{co}\left( \mathit{\Omega }%
\left( h,K\right) \right) }+\mathcal{V}_{x}\right] .
\end{equation*}%
This follows from \cite[Theorem 1.10]{Rudin} together with the fact that
each neighborhood of $\{0\}\subseteq \mathcal{X}$ contains some open and
convex neighborhood of $\{0\}\subseteq \mathcal{X}$ because $\mathcal{X}$ is
locally convex. Observe also that any one--point set $\{x\}\subseteq $ $%
\mathcal{X}$ is compact.

For any neighborhood $\mathcal{V}_{x}$ of $\{0\}\subseteq \mathcal{X}$ in a
locally convex space, there is another convex, symmetric, and open
neighborhood $\mathcal{V}_{x}^{\prime }$ of $\{0\}\subseteq \mathcal{X}$
such that $[\mathcal{V}_{x}^{\prime }+\mathcal{V}_{x}^{\prime }]\subseteq 
\mathcal{V}_{x}$, see proof of \cite[Theorem 1.10]{Rudin}. Let%
\begin{equation*}
\mathcal{G}_{x}^{\prime }:=K\cap \left[ \overline{\mathrm{co}\left( \mathit{%
\Omega }\left( h,K\right) \right) }+\mathcal{V}_{x}^{\prime }\right] .
\end{equation*}%
Then the following inclusions hold:%
\begin{equation}
\overline{\mathrm{co}\left( \mathit{\Omega }\left( h,K\right) \right) }%
\subseteq \mathcal{G}_{x}^{\prime }\subseteq \overline{\mathcal{G}%
_{x}^{\prime }}\subseteq \mathcal{G}_{x}\subseteq \overline{\mathcal{G}_{x}}%
\subseteq K\backslash \{x\}.  \label{eq inclusions}
\end{equation}%
Since $K$, $\mathcal{V}_{x}$, $\mathcal{V}_{x}^{\prime }$, and $\overline{%
\mathrm{co}\left( \mathit{\Omega }\left( h,K\right) \right) }$ are all
convex sets, $\mathcal{G}_{x}$ and $\mathcal{G}_{x}^{\prime }$ are also
convex. Seen as subsets of $K$ they are open neighborhoods of $\overline{%
\mathrm{co}\left( \mathit{\Omega }\left( h,K\right) \right) }$.

By Definition \ref{space}, the set $\mathcal{X}$ is a Hausdorff space and
thus any compact subset $K$ of $\mathcal{X}$ is a normal space. By Urysohn
lemma, there is a continuous function%
\begin{equation*}
f_{x}:K\rightarrow \lbrack \inf h(K),\inf h(K\backslash \mathcal{G}%
_{x}^{\prime })]
\end{equation*}%
satisfying $f_{x}\leq h$ and 
\begin{equation*}
f_{x}\left( y\right) =\left\{ 
\begin{array}{ll}
\inf h(K) & \mathrm{for\ }y\in \overline{\mathcal{G}_{x}^{\prime }}. \\ 
\inf h(K\backslash \mathcal{G}_{x}^{\prime }) & \mathrm{for\ }y\in
K\backslash \mathcal{G}_{x}.%
\end{array}%
\right.
\end{equation*}%
By compacticity of $K\backslash \mathcal{G}_{x}^{\prime }$ and the inclusion 
$\mathit{\Omega }\left( h,K\right) \subseteq \mathcal{G}_{x}^{\prime }$,
observe that%
\begin{equation*}
\inf h(K\backslash \mathcal{G}_{x}^{\prime })>\inf h(K).
\end{equation*}%
Then we have by construction that 
\begin{equation}
f_{x}(\overline{\mathrm{co}\left( \mathit{\Omega }\left( h,K\right) \right) }%
)=\{\inf h(K)\}  \label{Omega f sympa}
\end{equation}%
and%
\begin{equation}
f_{x}^{-1}(\inf h(K))=\mathit{\Omega }\left( f_{x},K\right) \subseteq 
\mathcal{G}_{x}  \label{Omega f sympabis}
\end{equation}%
for any $x\in K\backslash \overline{\mathrm{co}\left( \mathit{\Omega }\left(
h,K\right) \right) }$.

We use now the $\Gamma $--regularization $\Gamma _{K}\left( f_{x}\right) $
of $f_{x}$ on the set $K$ and denote by $\mathit{M}_{x}=\mathit{\Omega }%
\left( \Gamma _{K}(f_{x}),K\right) $ its non--empty set of minimizers over $%
K $. Applying Theorem \ref{Thm - Corollary I.3.6} for any $y\in \mathit{M}%
_{x}$ we have a probability measure $\mu _{y}\in M_{1}^{+}(K)$\ on $K$ with
barycenter $y$ such that 
\begin{equation}
\Gamma _{K}\left( f_{x}\right) \left( y\right) =\int_{K}\mathrm{d}\mu
_{y}(z)\;f_{x}\left( z\right) .  \label{herve bis}
\end{equation}%
As $y\in \mathit{M}_{x}$, i.e., 
\begin{equation}
\Gamma _{K}\left( f_{x}\right) \left( y\right) =\inf \,\Gamma _{K}\left(
f_{x}\right) (K)=\inf f_{x}(K),  \label{herve 2}
\end{equation}%
we deduce from (\ref{herve bis}) that 
\begin{equation*}
\mu _{y}(\mathit{\Omega }\left( f_{x},K\right) )=1
\end{equation*}%
and it follows that $y\in $ $\overline{\mathrm{co}\left( \mathit{\Omega }%
\left( f_{x},K\right) \right) }$, by Theorem \ref{thm barycenter}. By (\ref%
{Omega f sympabis}) together with the convexity of the open neighborhood $%
\mathcal{G}_{x}$ of $\overline{\mathrm{co}\left( \mathit{\Omega }\left(
h,K\right) \right) \text{,}}$ we thus obtain%
\begin{equation}
\mathit{M}_{x}\subseteq \overline{\mathrm{co}\left( \mathit{\Omega }\left(
f_{x},K\right) \right) }\subseteq \overline{\mathcal{G}_{x}}  \label{herve 3}
\end{equation}%
for any $x\in K\backslash \overline{\mathrm{co}\left( \mathit{\Omega }\left(
h,K\right) \right) }$.

We remark now that the inequality $f_{x}\leq h$ on $K$ yields $\Gamma
_{K}\left( f_{x}\right) \leq \Gamma _{K}\left( h\right) $ on $K$ because of
Corollary \ref{Biconjugate}. As a consequence, it results from (i) and (\ref%
{Omega f sympa}) that the set $\mathit{M}$ of minimizers of $\Gamma
_{K}\left( h\right) $ over $K$ is included in $\mathit{M}_{x}$, i.e., $%
\mathit{M}\subseteq \mathit{M}_{x}$. Hence, by (\ref{eq inclusions}) and (%
\ref{herve 3}), we have the inclusions 
\begin{equation}
\mathit{M}\subseteq \overline{\mathcal{G}_{x}}\subseteq K\backslash \{x\}.
\label{inclusion2bis}
\end{equation}%
Therefore, we combine (\ref{inclusion1}) with (\ref{inclusion2bis}) for all $%
x\in K\backslash \overline{\mathrm{co}\left( \mathit{\Omega }\left(
h,K\right) \right) }$ to obtain the desired equality in the assertion (ii). 
\end{proof}%

This last theorem can be useful to analyze variational problems with
non--convex functionals on compact convex sets $K$. Indeed, the minimization
of a real functional $h$ over $K$ can be done in this case by analyzing a
variational problem related to a lower semi--continuous convex functional $%
\Gamma _{K}\left( h\right) $ for which many different methods of analysis
are available.

To conclude, note that extreme points of the compact convex set $\mathit{M}$
belongs to the set $\mathit{\Omega }\left( h,K\right) $ and the
non--convexity of $\mathit{\Omega }\left( h,K\right) $ prevents the set $%
\mathit{M}$ from being homeomorphic to the Poulsen simplex:

\begin{theorem}[Minimization of real functionals -- II]
\label{theorem trivial sympa 1 copy(1)}\mbox{ }\newline
\index{Minimizers}%
\index{Minimization of real functionals}Let $K$ be any (non--empty) compact
convex subset of a locally convex space $\mathcal{X}$ and $h:K\rightarrow
\lbrack \mathrm{k},\infty ]$ be any real functional with $\mathrm{k}\in 
\mathbb{R}$. Then we have that:\newline
\emph{(i)} Extreme points of the compact convex set $\mathit{M}$ of
minimizers of $\Gamma _{K}\left( h\right) $ over $K$ belong to the closure
of the set of generalized minimizers of $h$, i.e., $\mathcal{E}\left( 
\mathit{M}\right) \subseteq 
\overline{\mathit{\Omega }\left( h,K\right) }$.\newline
\emph{(ii)} If $\mathcal{E}\left( \mathit{M}\right) $ is dense in $\mathit{M}
$ then $\overline{\mathit{\Omega }\left( h,K\right) }=\mathit{M}$ is a
compact and convex set.
\end{theorem}

\begin{proof}%
The first statement (i) results from Theorem \ref{theorem trivial sympa 1}
(ii) together with Theorem \ref{lemma Milman} (ii). The second assertion
(ii) is also straightforward. Indeed, if $\mathcal{E}\left( \mathit{M}%
\right) $ is dense in $\mathit{M}$ then $\mathit{\Omega }\left( h,K\right) $
is also dense in $\mathit{M}$ as $\mathcal{E}\left( \mathit{M}\right)
\subseteq \mathit{\Omega }\left( h,K\right) $, by (i). As a consequence, $%
\mathit{M}=\overline{\mathit{\Omega }\left( h,K\right) }$. 
\end{proof}%

Therefore, if $K$ is metrizable and $\mathcal{E}\left( \mathit{M}\right) $
is dense in $\mathit{M}$ then, by Lemma \ref{theorem trivial sympa 0}
together with Theorem \ref{theorem trivial sympa 1 copy(1)} (ii), $\mathit{%
\Omega }\left( h,K\right) =\mathit{M}$ is a compact and convex set.

\section{The Legendre--Fenchel transform and tangent functionals\label%
{Section Legendre-Fenchel transform}}

In contrast to the $\Gamma $--regularization defined in Section \ref{Section
gamma regularization} the notion of Legendre--Fenchel transform requires the
use of dual pairs defined as follow:

\begin{definition}[Dual pairs]
\label{dual pairs}\mbox{ }\newline
\index{Dual pairs}For any locally convex space $(\mathcal{X},\tau )$, let $%
\mathcal{X}^{\ast }$ be its dual space, i.e., the set of all continuous
linear functionals on $\mathcal{X}$. Let $\tau ^{\ast }$ be any locally
convex topology on $\mathcal{X}^{\ast }$. $(\mathcal{X},\mathcal{X}^{\ast })$
is called a dual pair iff, for all $x\in \mathcal{X}$, the functional $%
y^{\ast }\mapsto y^{\ast }(x)$ on $\mathcal{X}^{\ast }$ is continuous w.r.t. 
$\tau ^{\ast }$, and all linear functionals which are continuous w.r.t. $%
\tau ^{\ast }$ have this form.
\end{definition}

\noindent By Theorem \ref{thm locally convex space}, a typical example of a
dual pair $(\mathcal{X},\mathcal{X}^{\ast })$ is given by any locally convex
space $(\mathcal{X},\tau )$ and $\mathcal{X}^{\ast }$ equipped with the $%
\sigma (X^{\ast },X)$--topology $\tau ^{\ast }$, i.e., the weak$^{\ast }$%
--topology. In particular, as $\mathcal{W}_{1}$ is a Banach space, by
Corollary \ref{thm locally convex spacebis}, $(\mathcal{W}_{1},\mathcal{W}%
_{1}^{\ast })$ is a dual pair w.r.t. the norm and weak$^{\ast }$%
--topologies. We also observe that if $(\mathcal{X},\mathcal{X}^{\ast })$ is
a dual pair w.r.t. $\tau $ and $\tau ^{\ast }$ then $(\mathcal{X}^{\ast },%
\mathcal{X})$ is a dual pair w.r.t. $\tau ^{\ast }$ and $\tau $.

The \emph{Legendre--Fenchel transform} of a functional $h$ on $\mathcal{X}$
-- also called the \emph{conjugate }(functional) of $h$ -- is defined as
follows:

\begin{definition}[The Legendre--Fenchel transform]
\label{Legendre--Fenchel transform}\mbox{ }\newline
\index{Legendre--Fenchel transform|textbf}Let $(\mathcal{X},\mathcal{X}%
^{\ast })$ be a dual pair. For any functional $h:\mathcal{X}\rightarrow
\left( -\infty ,\infty \right] $, its Legendre--Fenchel transform $h^{\ast }$
is the convex lower semi--continuous functional from $\mathcal{X}^{\ast }$
to $\left( -\infty ,\infty \right] $ defined, for any $x^{\ast }\in \mathcal{%
X}^{\ast }$, by 
\begin{equation*}
h^{\ast }\left( x^{\ast }\right) :=\underset{y\in \mathcal{X}}{\sup }\left\{
x^{\ast }\left( y\right) -h\left( y\right) \right\} .
\end{equation*}
\end{definition}

\noindent If a functional $h$ is only defined on a subset $K\subseteq 
\mathcal{X}$ of a locally convex space $\mathcal{X}$ then one uses
Definition \ref{extension of functional} to compute its Legendre--Fenchel
transform $h^{\ast }$.

The Legendre--Fenchel transform and the $\Gamma $--regularization $\Gamma _{%
\mathcal{X}}\left( h\right) $ of $h$ are strongly related to one another.
This can be seen in the next theorem which gives an important property\ --
proven, for instance, in \cite[Proposition 51.6]{Zeidler3}\ -- of the double
Legendre--Fenchel transform $h^{\ast \ast }$, also called the \emph{%
biconjugate }(functional) of $h$:

\begin{theorem}[Property of the biconjugate]
\label{theorem fenchel moreau}\mbox{ }\newline
\index{Legendre--Fenchel transform!biconjugate}Let $(\mathcal{X},\mathcal{X}%
^{\ast })$ be a dual pair and $h:\mathcal{X}\rightarrow (-\infty ,\infty ]$
be any real functional. Then $h^{\ast \ast }\leq h$ and $h^{\ast \ast }\leq 
\hat{h}\leq h$ implies $h^{\ast \ast }=\hat{h}$ whenever $\hat{h}$ is convex
and lower semi--continuous.
\end{theorem}

\noindent By using Theorem \ref{theorem fenchel moreau} together with
Proposition \ref{lemma gamma regularisation}, we observe that $h^{\ast \ast
} $ is thus equal to the $\Gamma $--regularization $\Gamma _{\mathcal{X}%
}\left( h\right) $ of $h$:

\begin{corollary}[Biconjugate and $\Gamma $--regularization of $h$]
\label{theorem fenchel moreaubis}\mbox{ }\newline
\index{Legendre--Fenchel transform!Gamma--regularization}%
\index{Gamma--regularization!Legendre--Fenchel transform}Let a dual pair $(%
\mathcal{X},\mathcal{X}^{\ast })$ and $h:\mathcal{X}\rightarrow (-\infty
,\infty ]$ be any real functional. Then $h^{\ast \ast }=\Gamma _{\mathcal{X}%
}\left( h\right) $ on $\mathcal{X}$.
\end{corollary}

Another important notion related to the Legendre--Fenchel transform is the
concept of tangent functionals on real linear spaces:

\begin{definition}[Tangent functionals]
\label{tangent functional}\mbox{ }\newline
\index{Tangent functionals|textbf}Let $h$ be any real functional on a real
linear space $\mathcal{X}$. A linear functional $\mathrm{d}h:\mathcal{X}%
\rightarrow \mathbb{(-\infty },\infty ]$ is said to be tangent to the
function $h$ at $x\in \mathcal{X}$ iff, for all $x^{\prime }\in \mathcal{X}$%
, $h(x+x^{\prime })\geq h(x)+\mathrm{d}h(x^{\prime })$.
\end{definition}

\noindent If $\mathcal{X}$ is a separable real Banach space and $h$ is
convex and continuous then it is well--known that $h$ has, on each point $%
x\in \mathcal{X}$, at least one continuous tangent functional \textrm{d}$%
h\in \mathcal{X}^{\ast }$. This is a crucial result coming from Mazur
theorem \cite{Mazur} and Lanford III -- Robinson theorem \cite[Theorem 1]%
{LanRob}. Indeed, Mazur theorem describes the set $\mathcal{Y}$ where $h$
has exactly one continuous tangent functional \textrm{d}$h(x)\in \mathcal{X}%
^{\ast }$ at any $x\in \mathcal{Y}$:

\begin{theorem}[Mazur]
\label{Mazur}\mbox{ }\newline
\index{Mazur theorem|textbf}Let $\mathcal{X}$ be a separable real Banach
space and let $h:\mathcal{X}\rightarrow \mathbb{R}$ be a continuous convex
functional. The set $\mathcal{Y\subseteq X}$ of elements where $h$ has
exactly one continuous tangent functional \textrm{d}$h(x)\in \mathcal{X}%
^{\ast }$ at $x\in \mathcal{Y}$ is residual, i.e., a countable intersection
of dense open sets.
\end{theorem}

\begin{remark}
\label{Mazur remark}By Baire category theorem, the set $\mathcal{Y}$ is
dense in $\mathcal{X}$.
\end{remark}

\noindent Lanford III -- Robinson theorem \cite[Theorem 1]{LanRob} completes
Mazur theorem by characterizing the set of continuous tangent functionals 
\textrm{d}$h(x)\in \mathcal{X}^{\ast }$ for any $x\in \mathcal{X}$. In
particular, there is at least one continuous tangent functional \textrm{d}$%
h(x)\in \mathcal{X}^{\ast }$ at any $x\in \mathcal{X}$.

\begin{theorem}[Lanford III -- Robinson]
\label{Land.Rob}\mbox{ }\newline
\index{Lanford III -- Robinson theorem|textbf}Let $\mathcal{X}$ be a
separable real Banach space and let $h:\mathcal{X}\rightarrow \mathbb{R}$ be
a continuous convex functional. Then the set of tangent functionals $\mathrm{%
d}h(x)\in \mathcal{X}^{\ast }$ to $h$, at any $x\in \mathcal{X}$, is the weak%
$^{\ast }$--closed convex hull of the set $\mathcal{Z}_{x}$. Here, at fixed $%
x\in \mathcal{X}$, $\mathcal{Z}_{x}$ is the set of functionals $x^{\ast }\in 
\mathcal{X}^{\ast }$ such that there is a net $\{x_{i}\}_{i\in I}$ in $%
\mathcal{Y}$ converging to $x$ with the property that the unique tangent
functional $\mathrm{d}h(x_{i})\in \mathcal{X}^{\ast }$ to $h$ at $x_{i}$
converges towards $x^{\ast }$ in the weak$^{\ast }$--topology.
\end{theorem}

The Legendre--Fenchel transform and the tangent functionals are also related
to each other via the $\Gamma $--regularization of real functionals. Indeed,
the $\Gamma $--regularization $\Gamma _{\mathcal{X}}\left( h\right) $ of a
real functional $h$ allows to characterize all tangent functionals to $%
h^{\ast }$ at the point $x^{\ast }\in \mathcal{X}^{\ast }$ (see, e.g., \cite[%
Theorem I.6.6]{Simon}):

\begin{theorem}[Tangent functionals as minimizers]
\label{theorem trivial sympa 2}\mbox{ }\newline
\index{Minimizers!tangent|textbf}%
\index{Tangent functionals!minimizers|textbf}Let $(\mathcal{X},\mathcal{X}%
^{\ast })$ be a a dual pair and $h$ be any real functional from a
(non--empty) convex subset $K\subseteq \mathcal{X}$ to $(-\infty ,\infty ]$.
Then the set $\mathit{T}\subseteq \mathcal{X}$ of tangent functionals to $%
h^{\ast }$ at the point $x^{\ast }\in \mathcal{X}^{\ast }$ is the
(non--empty) set $\mathit{M}$ of minimizers over $K$ of the map 
\begin{equation*}
y\mapsto -x^{\ast }\left( y\right) +\Gamma _{K}\left( h\right) \left(
y\right)
\end{equation*}%
from $K\subseteq \mathcal{X}$ to $(-\infty ,\infty ]$.
\end{theorem}

\begin{proof}%
The proof is standard and simple, see, e.g., \cite[Theorem I.6.6]{Simon}.
Indeed, by Definition \ref{extension of functional}, any tangent functional $%
x\in \mathcal{X}$ to $h^{\ast }$ at $x^{\ast }\in \mathcal{X}$ satisfies the
inequality: 
\begin{equation}
x^{\ast }\left( x\right) +h^{\ast }\left( y^{\ast }\right) -y^{\ast }\left(
x\right) \geq h^{\ast }\left( x^{\ast }\right)  \label{landford1landford1}
\end{equation}%
for any $y^{\ast }\in \mathcal{X}^{\ast }$. Since $\Gamma _{K}\left(
h\right) =h^{\ast \ast }$ and $h^{\ast }=h^{\ast \ast \ast }$, we have (\ref%
{landford1landford1}) iff 
\begin{equation*}
x^{\ast }\left( x\right) +\underset{y^{\ast }\in \mathcal{X}^{\ast }}{\inf }%
\left\{ h^{\ast }\left( y^{\ast }\right) -y^{\ast }\left( x\right) \right\}
=x^{\ast }\left( x\right) -\Gamma _{K}\left( h\right) \left( x\right) \geq 
\underset{y\in \mathcal{X}}{\sup }\left\{ x^{\ast }\left( y\right) -\Gamma
_{K}\left( h\right) \left( y\right) \right\} .
\end{equation*}%
\end{proof}%

We combine Theorem \ref{theorem trivial sympa 1} with Theorem \ref{theorem
trivial sympa 2} to characterize the set $\mathit{T}\subseteq \mathcal{X}$
of tangent functionals to $h^{\ast }$ at the point $0\in \mathcal{X}^{\ast }$
as the closed convex hull of the set $\mathit{\Omega }\left( h,K\right) $ of
generalized minimizers of $h$ over a compact convex subset $K$, see
Definition \ref{gamm regularisation copy(4)}.

\begin{corollary}[Tangent functional and generalized minimizers]
\label{theorem trivial sympa 3}\mbox{ }\newline
\index{Tangent functionals!minimizers!generalized}Let $(\mathcal{X},\mathcal{%
X}^{\ast })$ be a dual pair and $h$ be any functional from a (non--empty)
compact convex subset $K\subseteq \mathcal{X}$ to $[\mathrm{k},\infty ]$
with $\mathrm{k}\in \mathbb{R}$. Then the set $\mathit{T}\subseteq \mathcal{X%
}$ of tangent functionals to $h^{\ast }$ at the point $0\in \mathcal{X}%
^{\ast }$ is the set 
\begin{equation*}
\mathit{T}=\mathit{M}=%
\overline{\mathrm{co}\left( \mathit{\Omega }\left( h,K\right) \right) }
\end{equation*}%
of minimizers of $\Gamma _{K}\left( h\right) $ over $K$, see Theorem \ref%
{theorem trivial sympa 1}.
\end{corollary}

\noindent This last result has some similarity with Lanford III -- Robinson
theorem (Theorem \ref{Land.Rob}) which has only been proven for separable
real Banach spaces $\mathcal{X}$ and continuous and convex functionals $h:%
\mathcal{X}\rightarrow \mathbb{R}$.

\section{Two--person zero--sum games%
\index{Zero--sum games!two--person |textbf}\label{Section two--person
zero--sum games}}

A study of two--person zero--sum games belongs to any elementary book on
game theory. These are defined via a map $(x,y)\mapsto f(x,y)$ from the
strategy set $M\times N$ to $\mathbb{R}$. Here, $M\subseteq \mathcal{X}$ and 
$N\subseteq \mathcal{Y}$ are subsets of two topological vector spaces $%
\mathcal{X}$ and $\mathcal{Y}$. The value $f(x,y)$ is the loss of the first
player making the decision $x$ and the gain of the second one making the
decision $y$. Without exchange of information and by minimizing the
functional%
\begin{equation*}
f^{\sharp }\left( x\right) :=\underset{y\in N}{\sup }f\left( x,y\right)
\end{equation*}%
the first player obtains her/his least maximum loss 
\begin{equation*}
\mathrm{F}^{\sharp }:=\underset{x\in M}{\inf }f^{\sharp }\left( x\right) ,
\end{equation*}%
whereas the greatest minimum gain of the second player is 
\begin{equation*}
\mathrm{F}^{\flat }:=\underset{y\in N}{\sup }f^{\flat }\left( y\right) \quad 
\mathrm{with}\quad f^{\flat }\left( y\right) :=\underset{x\in M}{\inf }%
f\left( x,y\right) .
\end{equation*}%
$\mathrm{F}^{\flat }$ and $\mathrm{F}^{\sharp }$ are called the \emph{%
conservative values} 
\index{Zero--sum games!two--person !conservative values}of the game. The
sets 
\begin{equation*}
\mathcal{C}^{\sharp }:=\left\{ x\in M:\mathrm{F}^{\sharp }=f^{\sharp }\left(
x\right) \right\} \quad \mathrm{and}\quad \mathcal{C}^{\flat }:=\left\{ y\in
N:\mathrm{F}^{\flat }=f^{\flat }\left( y\right) \right\}
\end{equation*}%
are the so--called set of \emph{conservatives strategies}%
\index{Zero--sum games!two--person !conservatives strategies} and $[\mathrm{F%
}^{\flat },\mathrm{F}^{\sharp }]$ is the \emph{duality interval}%
\index{Zero--sum games!two--person !duality interval}.

\emph{Non--cooperative equilibria} (or Nash equilibria) \cite[Definition 7.4.%
]{Aubin} of two--person zero--sum games are also called \emph{saddle points}%
. They are defined as follows:

\begin{definition}[Saddle points]
\label{definition saddle points}\mbox{ }\newline
\index{Saddle point|textbf}%
\index{Zero--sum games!two--person !saddle points}Let $M\subseteq \mathcal{X}
$ and $N\subseteq \mathcal{Y}$ be two subsets of topological vector spaces $%
\mathcal{X}$ and $\mathcal{Y}$. Then the element $(x_{0},y_{0})\in M\times N$
is a saddle point of the real functional $f:M\times N\rightarrow \mathbb{R}$
iff $x_{0}\in \mathcal{C}^{\sharp }$, $y_{0}\in \mathcal{C}^{\flat }$, and $%
\mathrm{F}:=\mathrm{F}^{\flat }=\mathrm{F}^{\sharp }$.
\end{definition}

\noindent It follows from this definition that a saddle point $%
(x_{0},y_{0})\in M\times N$ satisfies $\mathrm{F}=f(x_{0},y_{0})$. In this
case $\mathrm{F}:=\mathrm{F}^{\flat }=\mathrm{F}^{\sharp }$ is called the 
\emph{value of the game}. As a $\sup $ and a $\inf $ do not generally
commute we have in general $\mathrm{F}^{\flat }<\mathrm{F}^{\sharp }$ and
so, no saddle point of a two--person zero--sum game. An important criterion
for the existence of saddle points is given by the von Neumann min--max
theorem \cite[Theorem 8.2]{Aubin}:

\begin{theorem}[von Neumann]
\label{theorem minmax von Neumann}\mbox{ }\newline
\index{von Neumann min--max theorem|textbf}Let $M\subseteq \mathcal{X}$ and $%
N\subseteq \mathcal{Y}$ be two (non--empty) compact convex subsets of
topological vector spaces $\mathcal{X}$ and $\mathcal{Y}$. Assume that $%
f:M\times N\rightarrow \mathbb{R}$ is a real functional such that, for all $%
y\in N$, the map $x\mapsto f(x,y)$ is convex and lower semi--continuous,
whereas, for all $x\in M$, the map $y\mapsto f(x,y)$ is concave and upper
semi--continuous. Then there exists a saddle point $(x_{0},y_{0})\in M\times
N$ of $f$.
\end{theorem}

If the game ends up with a maximum loss $\mathrm{F}^{\sharp }$ for the first
player then it means that the second player has full information on the
choice of the first one. Indeed, the second player maximizes his gain $%
f\left( x,y\right) $ knowing always the choice $x$ of the first player.
(Similar interpretations can of course be done if one gets $\mathrm{F}%
^{\flat }$ instead of $\mathrm{F}^{\sharp }$.)

Another way to highlight this phenomenon can be done by introducing the
so--called \emph{decision rule} $r:M\rightarrow N$.%
\index{Zero--sum games!two--person !decision rule} Indeed, from \cite[%
Proposition 8.7]{Aubin} we have 
\begin{equation}
\mathrm{F}^{\sharp }=\underset{r\in N^{M}}{\sup }f^{\flat }\left( r\left(
x\right) \right) =\underset{r\in N^{M}}{\sup }\ \underset{x\in M}{\inf }%
f\left( x,r\left( x\right) \right)  \label{equation lasry}
\end{equation}%
with $N^{M}$ being the set of all decision rules (functions from $M$ to $N$%
). It means that the second player is informed of the choice $x$ of the
first player and uses a decision rule to maximize his gain. Under stronger
assumptions on the sets $M$, $N$ and on the map $(x,y)\mapsto f(x,y)$ (cf. 
\cite[Theorem 8.4]{Aubin}), observe that the second player can restrict
himself to \emph{continuous decision rules} only:

\begin{theorem}[Lasry]
\label{theorem extension games}\mbox{ }\newline
\index{Lasry theorem|textbf}Let $M\subseteq \mathcal{X}$ and $N\subseteq 
\mathcal{Y}$ be two subsets of topological vector spaces $\mathcal{X}$ and $%
\mathcal{Y}$ such that $M$ is compact and $N$ is convex. Assume that $%
f:M\times N\rightarrow \mathbb{R}$ is a real functional such that, for all $%
y\in N$, the map $x\mapsto f(x,y)$ is lower semi--continuous, whereas, for
all $x\in M$, the map $y\mapsto f(x,y)$ is concave. Then%
\begin{equation*}
\underset{x\in M}{\inf }\ \underset{y\in N}{\sup }f\left( x,y\right) =%
\underset{r\in \mathrm{C}\left( M,N\right) }{\sup }\ \underset{x\in M}{\inf }%
f\left( x,r\left( x\right) \right)
\end{equation*}%
with $\mathrm{C}\left( M,N\right) $ being the set of continuous mappings
from $M$ to $N$.
\end{theorem}

Equation (\ref{equation lasry}) or Theorem \ref{theorem extension games} can
be interpreted as an \emph{extension} of the two--person zero--sum game 
\emph{with exchange of information}.%
\index{Zero--sum games!two--person !extension} Extension of games are
defined for instance in \cite[Ch. 7, Section 7.2]{Aubinbis}. In the special
case of two--person zero--sum games, saddle point may not exist, but such a
non--cooperative equilibrium may appear by extending the strategy sets $M$
or $N$ (or both). This is, in fact, what we prove in Theorem \ref{lemma
extension game} for the extended thermodynamic game.

\backmatter

\chapter*{Index of Notation}

\noindent \textbf{Lattice and related matters}\bigskip

\noindent For any set $M$, we define $\mathcal{P}_{f}(M)$ to be the set of
all finite subsets of $M$.

\noindent $\mathfrak{L}=\mathbb{Z}^{d}$ seen as a set (lattice), see
Notation \ref{Notation1}.\smallskip

\noindent $d:\mathfrak{L}\times \mathfrak{L}\rightarrow \lbrack 0,\infty )$
is the Euclidean metric defined by (\ref{def.dist}).\smallskip

\noindent $\mathbb{Z}_{%
\vec{\ell}}^{d}:=\ell _{1}\mathbb{Z}\times \cdots \times \ell _{d}\mathbb{Z}$
for $\vec{\ell}\in \mathbb{N}^{d}$.\smallskip

\noindent $\Lambda _{l}$ is the cubic boxe of volume $|\Lambda
_{l}|=(2l+1)^{d}$ for $l\in \mathbb{N}$ defined by (\ref{cubic box}%
).\smallskip

\noindent $\Lambda +x$ is the translation of the set $\Lambda \in \mathcal{P}%
_{f}(\mathfrak{L})$ defined by (\ref{definition de lambda translate}%
).\smallskip

\noindent ${\o }(\Lambda )$ is the diameter of the set $\Lambda \in \mathcal{%
P}_{f}(\mathfrak{L})$ defined by (\ref{diameter of lambda}).\bigskip

\noindent \textbf{The fermion }$C^{\ast }$\textbf{--algebra and related
matters}\bigskip

\noindent $\mathcal{U}_{\Lambda }$ is the complex Clifford algebra with
identity $\mathbf{1}$ and generators $\{a_{x,\mathrm{s}},a_{x,\mathrm{s}%
}^{+}\}_{x\in \Lambda ,\mathrm{s}\in \mathrm{S}}$ satisfying the so--called
canonical anti--commutation relations (CAR), see (\ref{CAR}).\smallskip

\noindent $\mathcal{U}_{0}$ is the $\ast $--algebra of local elements, see (%
\ref{local elements}).\smallskip

\noindent $\mathcal{U}$ is the fermion (field) $C^{\ast }$--algebra, also
known as the CAR algebra.\smallskip

\noindent $\mathcal{U}^{+}$ is the $\ast $--algebra of of all even elements,
see (\ref{definition of even operators}).\smallskip

\noindent $\mathcal{U}^{\circ }$ is the $\ast $--algebra of of all gauge
invariant elements, see (\ref{definition of gauge invariant operators}) and
Notation \ref{Notation2}.\smallskip \smallskip

\noindent $\sigma _{\theta }$ is the automorphism of the algebra $\mathcal{U}
$ defined by (\ref{definition of gauge}).\smallskip

\noindent $\sigma ^{\circ }$ is the projection on the fermion observable
algebra $\mathcal{U}^{\circ }$, see Remark \ref{proj.gauge.inv}.\smallskip

\noindent $x\mapsto \alpha _{x}$ is the homomorphism from $\mathbb{Z}^{d}$
to the group of $\ast $--automorphisms of $\mathcal{U}$ defined by (\ref%
{transl}).\smallskip

\noindent $\pi \mapsto \alpha _{\pi }$ is the homomorphism from $\Pi $ to
the group of $\ast $--automorphisms of $\mathcal{U}$ defined by (\ref%
{definition perm automorphism}). \bigskip

\noindent \textbf{Sets of states}\bigskip

\noindent $\mathcal{U}^{\ast }$ is the dual space of the Banach space $%
\mathcal{U}$.\smallskip

\noindent $E\subseteq \mathcal{U}^{\ast }$ is the set of all states on $%
\mathcal{U}$.\smallskip

\noindent $E_{\Lambda }\subseteq \mathcal{U}_{\Lambda }^{\ast }$ for $%
\Lambda \in \mathcal{P}_{f}(\mathfrak{L})$ is the set of all states $\rho
_{\Lambda }$ on the local sub--algebra $\mathcal{U}_{\Lambda }$.\smallskip

\noindent $E_{\vec{\ell}}$ for $\vec{\ell}\in \mathbb{N}^{d}$ is the set of
all $\mathbb{Z}_{\vec{\ell}}^{d}$--invariant states defined by (\ref%
{periodic invariant states}).\smallskip

\noindent $E_{1}:=E_{(1,\cdots ,1)}$ is the set of all translation invariant
(t.i.) states.\smallskip

\noindent $E_{1}^{\circ }$ is the set of of translation and gauge invariant
states, see Remark \ref{t.i. + gauge inv states}.\smallskip

\noindent $E_{\Pi }$ is the set of all permutation invariant states defined
by (\ref{permutation inv states}).\smallskip

\noindent $E_{\otimes }$ is the set of product states.\smallskip

\noindent $\mathcal{E}_{\vec{\ell}}$ is the set of extreme points of the set 
$E_{\vec{\ell}}$ for $\vec{\ell}\in \mathbb{N}^{d}$.\smallskip

\noindent $\mathcal{E}_{1}:=\mathcal{E}_{(1,\cdots ,1)}$ is the set of t.i.
extreme states.\smallskip

\noindent $\mathcal{E}_{\Pi }$ is the set of extreme points of $E_{\Pi }$%
.\bigskip

\noindent \textbf{Sets of (generalized) minimizers of variational problems
on states}\bigskip

\noindent $\mathit{M}_{\Phi }$ is the set of t.i. equilibrium states of a
t.i. interaction $\Phi \in \mathcal{W}_{1}\subseteq \mathcal{M}_{1}$, see (%
\ref{equivalence def equilibrium states}).

\noindent $\mathit{M}_{\mathfrak{m}}^{\flat }$ is the set of t.i. minimizers
of $f_{\mathfrak{m}}^{\flat }$, see (\ref{set of minimizers de f bemol}%
).\smallskip

\noindent $\mathit{M}_{\mathfrak{m}}^{\sharp }$ is the set of t.i.
equilibrium states of a model $\mathfrak{m}\in \mathcal{M}_{1}$, see
Definition \ref{definition equilibirum state copy(1)}.\smallskip

\noindent $\mathit{\hat{M}}_{\mathfrak{m}}$ is the set of t.i. minimizers of
the reduced free--energy density functional $g_{\mathfrak{m}}$ defined by (%
\ref{definition minimizers of reduced free energy}). \smallskip

\noindent $\mathit{\Omega }_{\mathfrak{m}}^{\sharp }$ is the set of
generalized t.i. equilibrium states of a model $\mathfrak{m}\in \mathcal{M}%
_{1}$, see Definition \ref{definition equilibirum state}.\smallskip

\noindent $\mathit{\Omega }_{\mathfrak{m}}^{\sharp }\left( c_{a}\right) $ is
the subset (\ref{subset of a face}) of $\mathit{M}_{\Phi (c_{a})}$
satisfying the gap equations. \bigskip

\noindent \textbf{Banach space of all t.i.\ interactions}\bigskip

\noindent $\mathcal{W}_{1}$ is the real Banach space of all t.i.\
interactions, see Definition \ref{definition banach space interaction}.
\smallskip

\noindent $\Vert \,\cdot \,\Vert _{\mathcal{W}_{1}}$ is the norm of $%
\mathcal{W}_{1}$.\smallskip

\noindent $\mathcal{W}_{1}^{\mathrm{f}}\subseteq \mathcal{W}_{1}$ is the set
of all finite range t.i. interactions.\smallskip

\noindent $\mathcal{W}_{1}^{\ast }$ is the dual space of $\mathcal{W}_{1}$%
.\smallskip

\noindent $E_{1}\subseteq \mathcal{W}_{1}^{\ast }$ is also seen as including
in $\mathcal{W}_{1}^{\ast }$, see Section \ref{Section state=functional on W}%
\smallskip

\noindent $\mathcal{K}_{1}$ is the real Banach space of all t.i. interaction
kernels, see Definition \ref{def4.3}. \smallskip

\noindent $\Vert \,\cdot \,\Vert _{\mathcal{K}_{1}}$ is the norm of $%
\mathcal{K}_{1}$.\bigskip

\noindent \textbf{Banach space of long--range models}\bigskip

\noindent $(\mathcal{A},\mathfrak{A},\mathfrak{a})$ is a separable measure
space with $\mathfrak{A}$ and $\mathfrak{a}:\mathfrak{A}\rightarrow \mathbb{R%
}_{0}^{+}$ being respectively some $\sigma $--algebra on $\mathcal{A}$\ and
some measure on $\mathfrak{A}$.\smallskip

\noindent $\gamma _{a}\in \{-1,1\}$ is a fixed measurable function.\smallskip

\noindent $\gamma _{a,\pm }:=1/2(|\gamma _{a}|\pm \gamma _{a})\in \{0,1\}$,
see (\ref{remark positive negative part gamma}).\smallskip

\noindent $\mathcal{M}_{1}$ is the Banach space of long--range models, see
Definition \ref{definition M1bis}.\smallskip

\noindent $\Vert \,\cdot \,\Vert _{\mathcal{M}_{1}}$ is the norm of $%
\mathcal{M}_{1}$.\smallskip

\noindent $\mathcal{M}_{1}^{\mathrm{f}}\subseteq \mathcal{M}_{1}$ is the
sub--space of all finite range models. \smallskip

\noindent $\mathcal{M}_{1}^{\mathrm{d}}\subseteq \mathcal{M}_{1}$ is the
sub--space of discrete elements. \smallskip

\noindent $\mathcal{M}_{1}^{\mathrm{df}}:=\mathcal{M}_{1}^{\mathrm{d}}\cap 
\mathcal{M}_{1}^{\mathrm{f}}$. \smallskip

\noindent $\{\Phi _{a}\}_{a\in \mathcal{A}}$ is the long--range interaction
of any $\mathfrak{m}:=(\Phi ,\{\Phi _{a}\}_{a\in \mathcal{A}},\{\Phi
_{a}^{\prime }\}_{a\in \mathcal{A}})\in \mathcal{M}_{1}$, see Definition \ref%
{long range attraction-repulsion}.\smallskip

\noindent $\{\Phi _{a,-}:=\gamma _{a,-}\Phi _{a}\}_{a\in \mathcal{A}},\{\Phi
_{a,-}^{\prime }:=\gamma _{a,-}\Phi _{a}^{\prime }\}_{a\in \mathcal{A}}\in 
\mathcal{L}^{2}\left( \mathcal{A},\mathcal{W}_{1}\right) $ are the
long--range attractions of any $\mathfrak{m}\in \mathcal{M}_{1}$, see
Definition \ref{long range attraction-repulsion}.\smallskip

\noindent $\{\Phi _{a,+}:=\gamma _{a,+}\Phi _{a}\}_{a\in \mathcal{A}},\{\Phi
_{a,+}^{\prime }:=\gamma _{a,+}\Phi _{a}^{\prime }\}_{a\in \mathcal{A}}\in 
\mathcal{L}^{2}\left( \mathcal{A},\mathcal{W}_{1}\right) $ are the
long--range repulsions of any $\mathfrak{m}\in \mathcal{M}_{1}$, see
Definition \ref{long range attraction-repulsion}.\smallskip

\noindent $\mathcal{N}_{1}$ is the Banach space (\ref{definition model
noyaux}). \smallskip

\noindent $\Vert \,\cdot \,\Vert _{\mathcal{N}_{1}}$ is the norm of $%
\mathcal{N}_{1}$.\smallskip

\noindent $\mathcal{N}_{1}^{\mathrm{f}}\subseteq \mathcal{N}_{1}$ is the
sub--space of all finite range models of $\mathcal{N}_{1}$. \smallskip

\noindent $\mathcal{N}_{1}^{\mathrm{d}}\subseteq \mathcal{N}_{1}$ is the
sub--space of discrete elements of $\mathcal{N}_{1}$. \smallskip

\noindent $\mathcal{N}_{1}^{\mathrm{df}}:=\mathcal{N}_{1}^{\mathrm{d}}\cap 
\mathcal{N}_{1}^{\mathrm{f}}$. \bigskip

\noindent \textbf{Space--averaging functionals}\bigskip

\noindent $A_{L,\vec{\ell}}\in \mathcal{U}$ for $A\in \mathcal{U}$, $L\in 
\mathbb{N}$ and $\vec{\ell}\in \mathbb{N}^{d}$ is the element defined by the
space--average (\ref{definition de A L}).\smallskip

\noindent $A_{L}:=A_{L,\vec{\ell}}$ for $\vec{\ell}=(1,\cdots ,1)$, $A\in 
\mathcal{U}$, $L\in \mathbb{N}$.\smallskip

\noindent $A_{\vec{\ell}}$ is the space--average defined by (\ref{definition
de A l}) for any $\vec{\ell}\in \mathbb{N}^{d}$.\smallskip

\noindent $\Delta _{A,\vec{\ell}}$ for $A\in \mathcal{U}$ and $\vec{\ell}\in 
\mathbb{N}^{d}$ is the ($\vec{\ell}$--) space--averaging functional defined
by Definition \ref{definition de deltabis}.\smallskip

\noindent $\Delta _{A}:=\Delta _{A,(1,\cdots ,1)}$ for $A\in \mathcal{U}$ is
the space--averaging functional defined by (\ref{delta bounded}).\smallskip

\noindent $\Delta _{a,\pm }:E_{\vec{\ell}}\rightarrow \mathbb{R}$ is the
functional defined by (\ref{definition of delta long range}).\bigskip

\noindent \textbf{Internal energies and finite--volume thermodynamic
functionals}\bigskip

\noindent $U_{\Lambda }^{\Phi }\in \mathcal{U}^{+}\cap \mathcal{U}_{\Lambda
} $ is the internal energy of an interaction $\Phi $ for $\Lambda \in 
\mathcal{P}_{f}(\mathfrak{L})$, see Definition \ref{definition standard
interaction}.\smallskip

\noindent $U_{l}\in \mathcal{U}^{+}\cap \mathcal{U}_{\Lambda }$ is the
internal energy in the box $\Lambda _{l}$ of a model $\mathfrak{m}\in 
\mathcal{M}_{1}$ for $l\in \mathbb{N}$, see Definition \ref{definition
BCS-type model}.\smallskip

\noindent $\tilde{U}_{l}\in \mathcal{U}^{+}\cap \mathcal{U}_{\Lambda }$ is
the internal energy with periodic boundary conditions of a model $\mathfrak{m%
}\in \mathcal{M}_{1}$ for $l\in \mathbb{N}$, see Definition \ref{definition
BCS-type model periodized}. \smallskip

\noindent $p_{l}=p_{l,\mathfrak{m}}$ is the finite--volume pressure of $%
\mathfrak{m}\in \mathcal{M}_{1}$ defined by (\ref{BCS pressure}).\smallskip

\noindent $\tilde{p}_{l}=\tilde{p}_{l,\mathfrak{n}}$ is the finite--volume
pressure, with periodic boundary conditions, of $\mathfrak{n}\in \mathcal{N}%
_{1}$ defined by (\ref{pressure pbc}).\smallskip

\noindent $\rho _{l}:=\rho _{\Lambda _{l},U_{l}}$ is the Gibbs state (\ref%
{Gibbs.state}) associated with the internal energy $U_{l}$ in the box $%
\Lambda _{l}$ for $\mathfrak{m}\in \mathcal{M}_{1}$. \smallskip

\noindent $\tilde{\rho}_{l}:=\rho _{\Lambda _{l},\tilde{U}_{l}}$ is the
Gibbs state (\ref{Gibbs.state}) associated with the internal energy $\tilde{U%
}_{l}$ in the box $\Lambda _{l}$ for $\mathfrak{m}\in \mathcal{M}_{1}$.
\smallskip

\noindent $\hat{\rho}_{l}$ is the space--averaged t.i. Gibbs state (\ref%
{t.i. state rho l}) or (\ref{space average state p.b.c.}). \bigskip

\noindent \textbf{Infinite--volume thermodynamic functionals}\bigskip

\noindent $s:E_{\vec{\ell}}\rightarrow \mathbb{R}_{0}^{+}$ is the entropy
density functional, see Definition \ref{entropy.density}.\smallskip

\noindent $e_{\Phi }:E_{\vec{\ell}}\rightarrow \mathbb{R}$ is the energy
density functional, see Definition \ref{definition energy density}.\smallskip

\noindent $f_{\Phi }:E_{\vec{\ell}}\rightarrow \mathbb{R}$ is the
free--energy density functional, see Definition \ref{Remark free energy
density}.\smallskip

\noindent $g_{\mathfrak{m}}:E_{\vec{\ell}}\rightarrow \mathbb{R}$ is the
reduced free--energy density functional w.r.t. any $\mathfrak{m}\in \mathcal{%
M}_{1}$, see Definition \ref{Reduced free energy}.\smallskip

\noindent $f_{\mathfrak{m}}^{\flat }:E_{\vec{\ell}}\rightarrow \mathbb{R}$
is the functional defined by (\ref{convex functional g_m}). \smallskip

\noindent $f_{\mathfrak{m}}^{\sharp }:E_{\vec{\ell}}\rightarrow \mathbb{R}$
is the reduced free--energy density functional w.r.t. any $\mathfrak{m}\in 
\mathcal{M}_{1}$, see Definition \ref{Free-energy density long range}%
.\smallskip

\noindent $\mathrm{P}_{\mathfrak{m}}^{\flat }:\mathcal{M}_{1}\rightarrow 
\mathbb{R}$ is the the variational problem (\ref{pressure bemol}). \smallskip

\noindent $\mathrm{P}_{\mathfrak{m}}^{\sharp }:\mathcal{M}_{1}\rightarrow 
\mathbb{R}$ is the (infinite--volume) pressure, see Definition \ref{Pressure}%
.\bigskip

\noindent \textbf{Approximating interactions and thermodynamic game}\bigskip

\noindent $L_{\pm }^{2}(\mathcal{A},\mathbb{C})\subseteq L^{2}(\mathcal{A},%
\mathbb{C})$ are the Hilbert spaces defined by (\ref{definition of
positive-negative L2 space}).\smallskip

\noindent $\Phi (c_{a})$ for $c_{a}\in L^{2}(\mathcal{A},\mathbb{C})$ is the
approximating interaction of any model $\mathfrak{m}\in \mathcal{M}_{1}$,
see Definition \ref{definition BCS-type model approximated}.\smallskip

\noindent $U_{l}(c_{a})$ is the internal energy of the approximating
interaction of $\Phi (c_{a})$, see (\ref{internal and surface energies
approximated}).\smallskip

\noindent $p_{l}\left( c_{a}\right) $ is the finite--volume pressure
associated with $U_{l}(c_{a})$, see (\ref{pression approximated}).\smallskip

\noindent $P_{\mathfrak{m}}\left( c_{a}\right) $ is the (infinite--volume)
pressure associated with $\Phi (c_{a})$, see (\ref{variational problem
approx}). \smallskip

\noindent $f_{\mathfrak{m}}\left( \rho ,c_{a}\right) $ is the free--energy
density functional associated with $\Phi (c_{a})$, see (\ref{free--energy
density approximated 1}). \smallskip

\noindent $\mathfrak{f}_{\mathfrak{m}}:L_{-}^{2}(\mathcal{A},\mathbb{C}%
)\times L_{+}^{2}(\mathcal{A},\mathbb{C})\rightarrow \mathbb{R}$ is the
approximating free--energy density functional, see Definition \ref%
{definition approximating free--energy}.\smallskip

\noindent $\mathrm{F}_{\mathfrak{m}}^{\flat }$ is the first conservative
value of the thermodynamic game, see Definition \ref{definition two--person
zero--sum game}.\smallskip

\noindent $\mathrm{F}_{\mathfrak{m}}^{\sharp }$ is the second conservative
value of the thermodynamic game, see Definition \ref{definition two--person
zero--sum game}.\smallskip

\noindent $\mathfrak{f}_{\mathfrak{m}}^{\flat }$ is the least gain\emph{\ }%
functional of the attractive player, see Definition \ref{definition
two--person zero--sum game}.\smallskip

\noindent $\mathfrak{f}_{\mathfrak{m}}^{\sharp }$ is the worst loss
functional of the repulsive player, see Definition \ref{definition
two--person zero--sum game}.\smallskip

\noindent $\mathcal{C}_{\mathfrak{m}}^{\flat }$ is the set of conservative
strategies of the repulsive player, i.e., the set of minimizers of $%
\mathfrak{f}_{\mathfrak{m}}^{\flat }$, see (\ref{eq conserve strategy}%
).\smallskip

\noindent $\mathcal{C}_{\mathfrak{m}}^{\sharp }$ is the set of conservative
strategies of the attractive player, i.e., the set of minimizers of $%
\mathfrak{f}_{\mathfrak{m}}^{\sharp }$, see (\ref{eq conserve strategy}%
).\smallskip

\noindent $\mathcal{C}_{\mathfrak{m}}^{\flat }\left( c_{a,+}\right) $ is the
set of minimizers of $\mathfrak{f}_{\mathfrak{m}}\left(
c_{a,-},c_{a,+}\right) $ at fixed $c_{a,+}\in L_{+}^{2}(\mathcal{A},\mathbb{C%
})$, see (\ref{eq conserve strategybis}).\smallskip

\noindent $\mathcal{C}_{\mathfrak{m}}^{\sharp }\left( c_{a,-}\right) $ is
the set of minimizers of $\mathfrak{f}_{\mathfrak{m}}\left(
c_{a,-},c_{a,+}\right) $ at fixed $c_{a,-}\in L_{-}^{2}(\mathcal{A},\mathbb{C%
})$, see (\ref{eq conserve strategybis}).\smallskip

\noindent $\mathrm{C}\left( L_{-}^{2},L_{+}^{2}\right) $ is the set of
continuous decision rules of the repulsive player, that is, the set of
continuous mappings from $L_{-}^{2}(\mathcal{A},\mathbb{C})$ to $L_{+}^{2}(%
\mathcal{A},\mathbb{C})$ with $L_{-}^{2}(\mathcal{A},\mathbb{C})$ and $%
L_{+}^{2}(\mathcal{A},\mathbb{C})$ equipped with the weak and norm
topologies, respectively.\smallskip

\noindent $\mathrm{r}_{+}\in \mathrm{C}\left( L_{-}^{2},L_{+}^{2}\right) $
is the thermodynamic decision rule (\ref{thermodyn decision rule}) of the
model $\mathfrak{m}\in \mathcal{M}_{1}$. \smallskip

\noindent $\mathfrak{f}_{\mathfrak{m}}^{\mathrm{ext}}:L_{-}^{2}(\mathcal{A},%
\mathbb{C})\rightarrow \mathrm{C}(L_{-}^{2},L_{+}^{2})$ is the loss--gain
function (\ref{extended functional}) of the extended thermodynamic game of
the model $\mathfrak{m}$.\bigskip

\noindent \textbf{Theories}\bigskip

\noindent $\mathfrak{T}_{\mathfrak{m}}\subseteq \mathcal{M}_{1}$ is a theory
for $\mathfrak{m}\in \mathcal{M}_{1}$, see Definition \ref{definition theory}%
.\smallskip

\noindent $\mathfrak{T}_{\mathfrak{m}}^{+}\subseteq \mathcal{M}_{1}$ is the
min repulsive theory for $\mathfrak{m}\in \mathcal{M}_{1}$, see Definition %
\ref{definition min repulsive theory}.\smallskip

\noindent $\mathfrak{T}_{\mathfrak{m}}^{\sharp }\subseteq \mathcal{W}_{1}$
is the min--max local theory for $\mathfrak{m}\in \mathcal{M}_{1}$, see
Definition \ref{definition effective theories bogo}.\bigskip

\noindent \textbf{General notation}\bigskip

\noindent $\mathfrak{L}$ stands for $\mathbb{Z}^{d}$ as seen as a set
(lattice), whereas with $\mathbb{Z}^{d}$ the abelian group $(\mathbb{Z}%
^{d},+)$ is meant, cf. Notation \ref{Notation1}.\smallskip

\noindent Any symbol with a circle $\circ $ as a superscript is, by
definition, an object related to gauge invariance, see Notation \ref%
{Notation2}.\smallskip

\noindent The letters $\rho $, $\varrho $, and $\omega $ are exclusively
reserved to denote states, see Notation \ref{Notation3}.\smallskip

\noindent Extreme points of $E_{\vec{\ell}}$ are written as $\hat{\rho}\in 
\mathcal{E}_{\vec{\ell}}$ or sometime $\hat{\omega}\in \mathcal{E}_{\vec{\ell%
}}$, see Notation \ref{notation extreme states}.\smallskip

\noindent The letters $\Phi $ and $\Psi $ are exclusively reserved to denote
interactions, see Notation \ref{Notation5}.\smallskip

\noindent The letter $\omega $ is exclusively reserved to denote generalized
t.i. equilibrium states. Extreme points of $\mathit{\Omega }_{\mathfrak{m}%
}^{\sharp }$ are usually written as $\hat{\omega}\in \mathcal{E}(\mathit{%
\Omega }_{\mathfrak{m}}^{\sharp })$ (cf. Theorem \ref{theorem Krein--Millman}%
), see Notation \ref{Notation eq states}.\smallskip

\noindent The letter $\varphi $ is exclusively reserved to denote
interaction kernels, see Definition \ref{definition interaction kernel} and
Notation \ref{Notation7}.

\noindent The symbol $\mathfrak{m}:=(\Phi ,\{\Phi _{a}\}_{a\in \mathcal{A}%
},\{\Phi _{a}^{\prime }\}_{a\in \mathcal{A}})\in \mathcal{M}_{1}$ is
exclusively reserved to denote elements of $\mathcal{M}_{1}$, see Notation %
\ref{Notation6}.\smallskip

\noindent Any symbol with a tilde on the top (for instance, $\tilde{p}$) is,
by definition, an object related to periodic boundary conditions., see
Notation \ref{Notation pbc}.\smallskip

\noindent The symbol $\mathfrak{n}=(\varphi ,\{\varphi _{a}\}_{a\in \mathcal{%
A}},\{\varphi _{a}^{\prime }\}_{a\in \mathcal{A}})\in \mathcal{N}_{1}$ is
exclusively reserved to denote elements of $\mathcal{N}_{1}$, see Notation %
\ref{Notation8}.

\noindent $\Gamma _{K}\left( h\right) $ is the $\Gamma $--regularization of
a real functional $h$ on a subset $K$, see Definition \ref{Section gamma
regularization}.

\include{Bru_Pedra_memoire_AMS_18-index}%

\end{document}